\documentclass[11pt]{amsart}
\usepackage{amsmath}
\usepackage{graphicx}
\usepackage{enumerate}
\usepackage{natbib}
\usepackage{url} 
\usepackage{amsthm}
\usepackage{float} 
\usepackage{comment}

\usepackage[margin=1in]{geometry} 
\usepackage{setspace}

\usepackage{tikz}
\usepackage{enumitem}
\usepackage[utf8]{inputenc}
\usepackage{amsfonts}
\usepackage{amssymb}
\usepackage{booktabs}
\usepackage{tabularx}
\usepackage{makecell}
\usepackage{array}
\usepackage{longtable}
\usepackage[ruled,lined]{algorithm2e}

\usepackage{amsaddr}
\usepackage{mathtools}
\usepackage{accents}
\usepackage{mathrsfs}
\usepackage{threeparttable}
\usepackage{caption}
\usepackage{subcaption}
\usetikzlibrary{decorations.pathreplacing, calligraphy}
\usetikzlibrary{fadings}

\SetCommentSty{mycommfont}

\RequirePackage[OT1]{fontenc}
\RequirePackage{amsthm,amsmath}
\RequirePackage[colorlinks,citecolor=blue,urlcolor=blue]{hyperref}

\usepackage{bm}

\newcommand{\D}{{\mathop{}\!\mathrm{d}}} 

\newcommand{\R}{\mathbb{R}}

\newcommand{\N}{\mathbb{N}}

\newcommand{\PP}{\mathbb{P}}

\newcommand{\E}{\mathbb{E}}

\newcommand{\T }{\mathcal{T}}

\newcommand{\one}{ 1 \hspace{-3pt} \mathrm{l}} %

\newcommand{\Omloc}{{\Omega_{\operatorname{loc}}}}
\newcommand{\ab}{{\mathbf{a}}}

\numberwithin{equation}{section}  
\newtheorem{defn}{Definition}[section]

\newtheorem{rem}[defn]{Remark}

\newtheorem{prop}[defn]{Proposition}
\newtheorem{cor}[defn]{Corollary}
\newtheorem{lem}[defn]{Lemma}

\usepackage[textwidth=3.2cm]{todonotes}
\usepackage{xcolor}
\hypersetup{
    colorlinks,
    linkcolor={red!50!black},
    citecolor={blue!50!black},
    urlcolor={blue!80!black}
}

\title[Robustness in Sequential Decision Making under Evolving Uncertainty]{Robustness in Sequential Decision Making under Evolving Uncertainty: Evidence from High-Frequency Market Making}

\date{\today}
\author[Y. Chen, J. Sester, H. Tran, Y. Zhang, H. Nguyen]{ Ying Chen$^{1,6,7}$, Hoa Nguyen$^{3}$, Julian Sester$^{1}$, Hoang Hai Tran$^{2}$, Yijiong Zhang$^{4,5}$}

\begin{document}
\maketitle
\begin{center}
\normalsize{\today} \\ \vspace{0.5cm}
\small\textit{
$^{1}$National University of Singapore, Centre for Quantitative Finance, Department of Mathematics,\\ 21 Lower Kent Ridge Road, 119077. \\[2mm]
$^{2}$Grasshopper Asset Management Pte Ltd, \\ 72 Anson Road \#04-01 Anson House, Singapore 079911 \\[2mm]
$^{3}$National University of Singapore, Asian Institute of Digital Finance \\[2mm]
$^{4}$NUS (Chongqing) Research Institute \\[2mm]
$^{5}$Tsinghua University, Department of Statistics and Data Science \\[2mm]
$^{6}$National University of Singapore, SIA-NUS Digital Aviation Corporate Laboratory, Institute of Operations Research \& Analytics,\\[2mm]
$^{7}$National University of Singapore, Risk Management Institute}
\end{center}
\begin{abstract}\noindent 
We study sequential decision making under evolving uncertainty in high-frequency financial markets, where changing market dynamics continually challenge static decision policies. We show that robustness has two economically meaningful dimensions: uncertainty tolerance, which determines how much uncertainty the decision maker allows, and action robustness, which governs how conservatively decisions respond. Robustness is not merely protection against model misspecification, but a state-dependent mechanism that reshapes sequential decision behaviors. Simulation and empirical evidence show that action robustness has a substantially larger impact than uncertainty tolerance. Moreover, excessive robustness may reduce profitability in illiquid markets by limiting execution opportunities.


\end{abstract}

\section{Introduction}

Decision making in finance and economics is inherently conducted under uncertainty. In practice, decision makers rarely observe the true law governing the environment and instead rely on models estimated from limited, noisy, and evolving data. A large body of evidence shows that estimation error and model misspecification can induce substantial economic distortions, suggesting that uncertainty is not merely a statistical nuisance but a central determinant of optimal decision behavior.

This paper studies robust sequential decision making under evolving uncertainty through the lens of high-frequency market making. We develop a robust reinforcement learning framework that explicitly incorporates distributional uncertainty into the dynamic decision process using Sinkhorn-based ambiguity sets. The proposed framework combines a model-based Markov decision process with robust dynamic programming to learn adaptive quoting strategies under uncertain and evolving market conditions. By integrating reinforcement learning with distributionally robust optimization, the framework preserves economic interpretability while allowing flexible adaptation to complex and dynamically changing environments.

A central feature of our framework is that robustness operates through two distinct dimensions. The first dimension, \emph{uncertainty tolerance}, determines how much uncertainty the decision maker allows around future market dynamics. The second dimension, \emph{action robustness}, governs how conservatively decisions respond to adverse scenarios within the admitted uncertainty set. This distinction separates the amount of uncertainty considered plausible from the behavioral response to uncertainty, providing a richer characterization of robust sequential decision making.

Our results reveal that robustness is not merely a safeguard against uncertainty; it fundamentally reshapes sequential decision behavior under evolving market conditions. We obtain three main findings.

\begin{enumerate}
\item \textbf{Robustness changes how decisions are made, not merely how well they perform.}
Robustness systematically changes quoting, inventory-management, and risk-control behavior. Robust policies respond differently to inventory pressure, volatility, order-flow imbalances, liquidity conditions, and fill uncertainty, leading to distinct patterns of inventory adjustment and quote placement over time.

\item \textbf{Not all robustness is created equal.}
We identify two economically distinct dimensions of robustness: uncertainty tolerance and action robustness. Empirically, action robustness exerts a substantially stronger influence on policy behavior than uncertainty tolerance, suggesting that robust decision making is driven primarily by how agents adapt to adverse scenarios rather than by how much uncertainty they admit.

\item \textbf{Robustness is beneficial, but not universally so.}
The effectiveness of robustness depends critically on market liquidity. In high-liquidity environments, robustness improves performance by mitigating uncertainty while preserving execution opportunities. In low-liquidity environments, excessive robustness may reduce profitability by limiting trading opportunities. Robustness should therefore be viewed as a context-dependent design choice rather than a universally desirable feature.
\end{enumerate}
From a managerial perspective, our findings suggest that robustness should be calibrated to market conditions rather than applied uniformly. In particular, the distinction between uncertainty tolerance and action robustness provides practitioners with separate levers for controlling model risk and decision conservatism.

Our work contributes to the literature in three ways. First, we develop a robust reinforcement learning framework for sequential decision making under evolving uncertainty that integrates Sinkhorn-based ambiguity sets with dynamic programming in a realistic high-frequency market-making environment. Second, we identify and quantify two economically distinct dimensions of robustness, providing a new perspective on how robustness influences sequential decisions. Third, we show that robustness affects not only performance outcomes but also the structure of optimal decision policies, providing new insights into when and how robust decision making creates value.

Beyond market making, the proposed framework applies to a broad class of AI-driven sequential decision systems in which uncertainty evolves jointly with the environment and decisions influence future states. More generally, our findings highlight the importance of robustness in high-dimensional and dynamically evolving decision environments.

\subsection{Related literature}

\subsubsection{Sequential Decision Making under Uncertainty}

Uncertainty is a fundamental feature of economic and managerial decision making. In dynamic environments, decision makers rarely observe the true law governing the evolution of the system and instead rely on models estimated from limited, noisy, and evolving data. As decisions influence future states and outcomes, uncertainty can propagate through time, potentially amplifying the consequences of model misspecification and estimation error. A large literature has documented that even small deviations from assumed dynamics can lead to substantial distortions in optimal decisions and significant deterioration in out-of-sample performance \citep{hansen2011robustness}.

These concerns have motivated extensive research on decision making under uncertainty. Early studies demonstrated that portfolio allocation decisions can be highly sensitive to estimation errors in expected returns, often producing unstable allocations and poor out-of-sample performance \citep{best1991sensitivity, chopra1993effect}. Related work showed that model uncertainty can dominate estimation risk in dynamic asset allocation problems, while uncertainty in volatility fundamentally alters pricing and hedging decisions in derivative markets \citep{avramov2002stock, avellaneda1995pricing}. In stochastic control and operations research, robust optimization and robust control frameworks were developed to account for uncertainty in model parameters and system dynamics by optimizing performance under adverse scenarios \citep{Iyengar2005, Nilim2005}. More recently, distributionally robust optimization has provided a flexible framework for modeling ambiguity through uncertainty sets defined around a reference distribution \citep{delage2010distributionally,Wiesemann2013,kuhn2025distributionally}. 

While these approaches have established robustness as an effective mechanism for mitigating the adverse effects of uncertainty, the literature has focused predominantly on performance improvements and worst-case guarantees. Comparatively less attention has been paid to a more fundamental question: \emph{How does robustness change the structure of optimal decisions?} In particular, robustness may alter not only the outcomes of a decision policy but also the manner in which decisions respond to changing environmental conditions. Understanding these behavioral effects is especially important in sequential decision problems, where actions, states, and uncertainty interact dynamically over time.

\subsubsection{Robustness in Reinforcement Learning}

Robustness has become a central concern in reinforcement learning because policies trained under one environment may perform poorly when deployed under different or evolving conditions. This issue is especially important in sequential decision problems, where small errors in transition dynamics can propagate over time and affect both future states and future rewards. To address this challenge, the robust Markov decision process (MDP) framework allows transition dynamics to vary within an uncertainty set and seeks policies that perform well under adverse realizations of the environment. Subsequent developments have extended this framework to incorporate Bayesian and distributional information \citep{Strens2000,Xu2012}, uncertainty in model parameters and dynamics \citep{Goh2018,Goyal2023}, and uncertainty regarding the underlying model itself \citep{neufeld2022markov}.

Building on the robust MDP framework, a growing literature has developed reinforcement learning algorithms capable of learning policies under uncertain environments. Existing approaches include robust policy iteration methods \citep{Xing2014}, kernel-based robust reinforcement learning \citep{Lim2019}, and model-free robust reinforcement learning under misspecified environments \citep{Wang2021}. Collectively, these studies demonstrate that incorporating robustness can improve policy reliability when the deployed environment deviates from the assumptions used during training.

More recently, distributionally robust reinforcement learning has emerged as a powerful framework for data-driven sequential decision making. Rather than protecting against uncertainty in individual model parameters, distributionally robust approaches account for ambiguity in the underlying probability distribution itself. This perspective is especially appealing when the underlying environment is unknown, evolving, and only partially observed through data. Wasserstein and optimal-transport-based ambiguity sets have attracted particular attention because they provide a flexible and theoretically grounded way to characterize uncertainty while preserving tractability in many decision problems \citep{mohajerin2018data,gao2023distributionally}. These approaches establish important links between reinforcement learning and the broader literature on distributionally robust optimization.

While this literature has generated important insights into robustness, performance guarantees, and policy stability, existing studies typically treat robustness as a one-dimensional design choice: the decision maker specifies an uncertainty set, and the robust policy is evaluated against adverse scenarios within that set. However, it does not distinguish between two conceptually different roles of robustness: how much uncertainty the decision maker is willing to tolerate, and how conservatively decisions respond once that uncertainty is admitted. Our framework separates these two roles through \emph{uncertainty tolerance} and \emph{action robustness}. This distinction allows us to examine not only whether robustness improves performance, but also how different forms of robustness reshape sequential decision behavior.

\subsubsection{Market Making under Uncertainty}

Market making is a canonical sequential decision problem in which liquidity providers continuously balance profitability against inventory risk under uncertain market conditions. The seminal framework of \cite{Avellaneda2008} formulates market making as a stochastic control problem in which optimal bid and ask quotes are determined by inventory considerations and market dynamics. Subsequent research has extended this framework to incorporate richer market features, including order-flow dynamics, adverse selection, market impact, latency, and contract design \citep{Cartea2015,Guilbaud2013, gueant2013dealing,Baldacci2021}. Collectively, this literature has established inventory management and liquidity provision as central components of successful market-making strategies.

More recently, reinforcement learning has emerged as a flexible framework for market making because it can learn adaptive policies directly from interactions with the environment without requiring restrictive analytical assumptions. Existing studies have applied deep reinforcement learning to market making under realistic market dynamics and have demonstrated improvements in profitability, inventory control, and execution quality relative to traditional benchmark strategies \citep{Spooner2018,gueant2019deep}. These approaches are particularly attractive in complex and rapidly changing environments where closed-form stochastic control solutions may be difficult to obtain.

Despite these advances, existing market-making studies primarily evaluate policies through performance metrics such as profit-and-loss, inventory risk, execution quality, or Sharpe ratios. Comparatively little attention has been paid to how robustness affects market-making decisions or whether its benefits persist across different market conditions. In particular, the interaction between robustness and market liquidity remains poorly understood. A policy that performs well under uncertainty in highly liquid markets may become overly conservative when trading opportunities are limited. Our framework addresses this gap by examining how robustness reshapes market-making decisions and by identifying the market conditions under which robustness enhances or impairs performance.

\subsubsection{Positioning of the Present Paper}
Taken together, the literature suggests that robustness plays a critical role in sequential decision making, yet three fundamental questions remain unresolved. First, how does robustness alter decision behavior rather than merely performance outcomes? Second, do different dimensions of robustness influence decisions differently? Third, under what market conditions does robustness improve performance, and when can it become overly conservative? This paper addresses these questions through a robust reinforcement learning framework for high-frequency market making.

\subsection{Organization of the paper}
The paper is organized as follows. Section~\ref{sec:setting} presents the robust sequential decision framework, including market-making model, model uncertainty formulation and robust dynamic programming framework. Section~\ref{sec:algorithms} describes the reference transition model and reinforcement learning algorithms. Section~\ref{sec:numericals} reports the numerical and empirical results together with policy interpretation. Section~\ref{sec:Conclusion} concludes with managerial implications, limitations, and future research directions. Additional numerical diagnostics and proofs are collected in the appendix.

\section[Market-Making Model]{Robust Sequential Decision Framework}\label{sec:setting}
\noindent

This section develops a robust sequential decision framework under evolving uncertainty. We first formulate market making as a finite-horizon Markov decision process and then introduce distributional robustness through Sinkhorn ambiguity sets. The resulting framework distinguishes two economically meaningful dimensions of robustness: uncertainty tolerance and action robustness.

\subsection{Market-Making Decision Environment}
\label{subsec:market-making}
We formulate market making as a finite-horizon sequential decision problem. We begin with a baseline market-making model that serves as the reference Markov decision process. The model describes the economic environment, the state and action spaces. 
\subsubsection{Economic Setting}
We consider a single-asset market-making problem over a finite trading horizon $[0,T]$, where the trading period is discretized into decision times
\[
0=t_0<t_1<\cdots<t_T=T.
\]
At each decision time $t$, the market maker (MM)\footnote{We consider a single representative market maker operating over a finite trading horizon. This setting allows us to isolate the impact of robustness on sequential decision behavior without introducing additional strategic interactions among competing liquidity providers. While modern electronic markets typically involve multiple market makers, the representative-agent framework remains the standard benchmark in the market-making literature and provides a tractable environment for studying the behavioral effects of robustness.} observes all relevant market information and continuously provides liquidity by posting bid and ask quotes and earns profits from the bid--ask spread when orders are executed. At the same time, executed trades generate inventory positions that expose the market maker to future price fluctuations. The central trade-off is therefore between profitability and inventory risk.

Let $S_t$ denote the reference market price at time $t$. In practice, the mid-price is commonly used as a proxy for the efficient price because of its observability and high-frequency availability. The MM chooses bid and ask quotes according to
\[
p_t^{\text{bid}} = S_t - \delta_t^{\operatorname{bid}}, 
\qquad 
p_t^{\text{ask}} = S_t + \delta_t^{\operatorname{ask}},
\]
where $\delta_t^{\operatorname{bid}}>0$ and $\delta_t^{\operatorname{ask}}>0$ are the bid- and ask-side half-spreads. These quoting decisions jointly determine the profitability of executed trades and their execution probabilities.
This setup follows the classical inventory-based market-making framework of \cite{Avellaneda2008}, where optimal quoting balances spread revenues against inventory risk.

\subsubsection[Cash, Inventory, and Terminal Objective]{Cash, Inventory, and Terminal Objective}\label{sec:objective}
Let $Q_{t+1}^{\operatorname{bid}}$ and $Q_{t+1}^{\operatorname{ask}}$ denote the executed quantities between $t$ and $t+1$, let $c_{t}$ denote transaction cost per trade at time t and let $W_t$ denote the MM's wealth at time t. The wealth process evolves according to
\[
\begin{aligned}
W_{t+1}
= W_t
&+
\underbrace{
Q_{t+1}^{\operatorname{ask}}
\bigl(S_t+\delta_t^{\operatorname{ask}}\bigr)
}_{\text{cash received from ask-side executions}}
-
\underbrace{
Q_{t+1}^{\operatorname{bid}}
\bigl(S_t-\delta_t^{\operatorname{bid}}\bigr)
}_{\text{cash paid for bid-side executions}}  -
\underbrace{
c_{t+1}\bigl(Q_{t+1}^{\operatorname{bid}}
        +Q_{t+1}^{\operatorname{ask}}\bigr)
}_{\text{transaction costs}},
\end{aligned}
\]
and the inventory dynamics are given by the size of the previous inventory $I_t$ adjusted by the net difference in traded quantities
\[
I_{t+1} 
= I_t 
+ Q_{t+1}^{\operatorname{bid}} 
- Q_{t+1}^{\operatorname{ask}}.
\]
Profits arise from capturing the spread on executed trades, while inventory accumulates according to net order flow. Holding inventory exposes the market maker to fluctuations in the underlying asset price. Under the commonly adopted diffusion model for the reference price process, the variance of inventory-related P\&L is proportional to $\sigma^2\int_0^T I_t^2\,dt$, where $\sigma$ denotes the asset volatility, which therefore serves as a natural measure of cumulative inventory risk. The MM seeks to manage inventory exposure throughout the trading horizon rather than only at terminal time.
The objective is to maximize the expected risk-adjusted terminal mark-to-market value over all admissible policies:
\begin{equation}\label{eq:prob1}
\sup_{\ab \in \mathcal{A}}
\;
\mathbb{E}\!\left[
W_T + S_T I_T
- \frac{\sigma^2}{2} \int_0^T \gamma_t I_t^2 \, dt
\right],
\end{equation}
where $\gamma_t > 0$ is a time-varying risk-aversion and $\sigma$ is daily price volatility. The risk-aversion coefficient may vary over time to reflect the increasing cost of holding inventory as the end of the trading horizon approaches.

\subsubsection[State, Actions and Admissible Policies]{State and Admissible Policies}\label{subsubsec:state-action-detail}
The information relevant for the MM's decisions is summarized by the \emph{state} of the environment. At time $t\in\mathbb{N}$, the state vector $X_t$ is constructed from rolling windows of the most recent $m$ observations of inventory, price dynamics, order flow, liquidity conditions, market microstructure variables, and the remaining trading horizon:
\begin{align}\label{eq:def_X}
X_t
=
\big(
I^t,R^t,
N^{t,\operatorname{bid}},N^{t,\operatorname{ask}},
V^{t,\operatorname{bid}},V^{t,\operatorname{ask}},
\delta^{t,\operatorname{tick}},\delta^{t,\operatorname{rel}},
D^{t,\operatorname{micro}},\nu^t, \rm{OFI}^t,
Z^{t,\operatorname{bid}},Z^{t,\operatorname{ask}},
\tau^t
\big)
\in \Omloc,
\end{align}
where
\[
{
\Omloc := \left(\mathbb{R}^{m}\right)^{13} \times \mathbb{N}_0^m
\equiv \mathbb{R}^{13m}\times\mathbb{N}_0^m
}
\]
{describes the \emph{state space}.} In the empirical implementation, $X_t$ is represented by a rolling-window state vector and the detailed construction is deferred to Section~\ref{sec:algorithms}.

At each decision time, the market maker influences the evolution of inventory and wealth by posting bid and ask limit orders characterized by bid and ask spreads, $\delta_t^{\operatorname{bid}}$ and $\delta_t^{\operatorname{ask}}$, together with associated order quantities, $q_t^{\operatorname{bid}}$ and $q_t^{\operatorname{ask}}$. The spreads determine the profitability and execution likelihood of submitted quotes, while the order quantities control the market maker's exposure to inventory risk. This leads to an \emph{action space} of the form
\begin{equation}\label{eq:setA_0}\begin{aligned}
A :& = \big\{ (\delta^{\operatorname{bid}},\delta^{\operatorname{ask}},q^{\operatorname{bid}},q^{\operatorname{ask}}
) ~\big|~ \delta^{\operatorname{bid}}\in [\underline{C_S},\overline{C_S}], \delta^{\operatorname{ask}} \in [\underline{C_S},\overline{C_S}], q^{\operatorname{bid}}\in [\underline{C_q},\overline{C_q}], q^{\operatorname{ask}} \in [\underline{C_q},\overline{C_q}] \big\} \\
& = [\underline{C_S},\overline{C_S}]\times [\underline{C_S},\overline{C_S}]\times [\underline{C_q},\overline{C_q}] \times [\underline{C_q},\overline{C_q}] \subset \R_{\geq 0}^4,
\end{aligned}
\end{equation}
where $0<\underline{C_S}<\overline{C_S}<\infty$ and $0<\underline{C_q}<\overline{C_q}<1$ restrict the possible spreads and order quantities, respectively. Unlike classical market-making models that optimize only bid and ask spreads, the proposed framework allows the market maker to jointly determine both quote placement and order size.\footnote{In the empirical implementation, order quantities are parameterized as fractions of contemporaneous market trading volume rather than as absolute numbers of shares. This normalization makes the framework comparable across equities with different levels of liquidity and trading activity, while preserving the economic interpretation of order size as the market maker's participation rate in the market.}
We refer to an MM's \emph{policy} as a sequence of state-based executed actions.
For each \(t=0,\ldots,T-1\), the set of admissible continuation policies from time \(t\)
is defined by
\begin{equation}\label{eq:defAt}
\begin{aligned}
\mathcal{A}_t
:=\bigg\{(a_s)_{t \le s \leq T-1}~\bigg|~
&(a_s)_{t \le s \leq T-1}: \Omega \rightarrow A; a_s \text{ is } \sigma(X_s)\text{-measurable}
\text{ for all } t \le s \leq T-1
\bigg\}.
\end{aligned}
\end{equation}
Equivalently,
\begin{equation}\label{eq:defAt_state}
\mathcal{A}_t
=
\left\{
\left(a_s(X_s)\right)_{t\le s\le T-1}
~\middle|~
a_s:\Omloc \rightarrow A
\text{ is Borel measurable for all } t \le s \le T-1
\right\}.
\end{equation}
In particular, the set of admissible policies on the full time horizon is defined by
\begin{equation}\label{eq:defA}
\mathcal A:=\mathcal A_0.
\end{equation}

\subsection{Distributional Uncertainty}
\label{subsec:robustness}

We now introduce robustness into the sequential decision framework. Our approach combines distributionally robust optimization with Sinkhorn ambiguity sets to account for uncertainty in future market dynamics. 
\subsubsection{Distributional Uncertainty and Robust Objective}
\label{subsec:robust-objective}

The objective \eqref{eq:prob1} is formulated under a reference transition model. To account for model misspecification and evolving market conditions, we adopt a distributionally robust optimization framework; see, for example, \cite{kuhn2025distributionally}. Rather than optimizing under a single transition law, we evaluate decisions against a family of plausible transition distributions contained in an ambiguity set $\mathfrak P$. The market maker evaluates decisions against the most adverse probability measure within this set.

Accordingly, the robust counterpart of \eqref{eq:prob1} is given by
\begin{equation}
\label{eq:prob1_robust}
\sup_{\ab \in \mathcal A}
\inf_{\PP \in \mathfrak P}
\mathbb E_{\PP}
\left[
W_T + S_T I_T
-\frac{\sigma^2}{2}
\int_0^T \gamma_t I_t^2\,dt
\right],
\end{equation}
where the inner minimization represents the worst-case transition dynamics within the ambiguity set and provides protection against adverse deviations from the reference model. Economically, the ambiguity set captures uncertainty regarding future order arrivals, execution probabilities, liquidity conditions, and price dynamics. The robust formulation therefore protects against adverse deviations from the reference model while preserving the sequential nature of the decision problem.

\subsubsection{Additive Reward Representation}
\label{subsec:additive-reward}

The objective in \eqref{eq:prob1_robust} is stated in terms of terminal wealth. For dynamic programming, however, it is more convenient to express the problem recursively as the sum of one-period rewards. We therefore reformulate the objective in an equivalent additive form that preserves its economic interpretation while enabling Bellman recursion. At time $t\in [0,T]$ the market maker seeks to maximize the worst-case conditional expectation of terminal marked-to-market wealth net of cumulative inventory risk. By using the approximation
$\frac{\sigma^2}{2} \int_t^T \gamma_s I_s^2
ds\approx
\sum_{u=t}^{T-1}
\frac{\gamma_u}{2}I_{u+1}^2\sigma^2\Delta_u,
$
the time-$t$ objective becomes
\begin{equation}\label{eq:exact_value_function}
\sup_{\ab \in \mathcal A_t}\inf_{\PP\in\mathfrak{P}^{\varepsilon,\delta}_t(x)}
\mathbb E_{\PP}\!\left[
W_T + S_T I_T
-\sum_{u=t}^{T-1}
\frac{\gamma_u}{2}I_{u+1}^2\sigma^2\Delta_u 
\,\middle|\, X_t=x
\right],
\end{equation}
where \(\mathfrak{P}^{\varepsilon,\delta}_t(x)\) denotes the ambiguity set of transition kernels from time \(t\) onward, conditioned on \(X_t = x\),  (formally defined in \eqref{eq_definition_Pepsilon}).

To obtain a recursive sequential decision formulation, we decompose the terminal objective into a sequence of one-period rewards that capture the trade-off between spread revenues, inventory risk, and transaction costs. This transformation preserves the economic structure of the original problem while enabling a dynamic programming characterization of optimal decisions.

Accordingly, we define the one-period reward function

\begin{equation}
\label{eq_reward}
\begin{aligned}
r:\Omloc \times A \times \Omloc
&\rightarrow \mathbb{R},
\\
(X_t,a_t,X_{t+1})
\mapsto\;&
\underbrace{
Q_{t+1}^{\operatorname{ask}}
\delta_t^{\operatorname{ask}}
+
Q_{t+1}^{\operatorname{bid}}
\delta_t^{\operatorname{bid}}
}_{\text{spread capture}}
-
\underbrace{
\frac{\gamma_t}{2}
I_{t+1}^{2}
\sigma^2
\Delta_t
}_{\text{inventory risk penalty}}
-
\underbrace{
c_{t+1}
\big(
Q_{t+1}^{\operatorname{bid}}
+
Q_{t+1}^{\operatorname{ask}}
\big)
}_{\text{transaction costs}},
\end{aligned}
\end{equation}

The following observations motivates a sequential reformulation of the original terminal-value problem \eqref{eq:exact_value_function}.

\begin{lem}\label{lem_reward_approximation}
Fix \(t\in\{0,\ldots,T-1\}\) and an initial state \(x\in\Omega_{\operatorname{loc}}\).  
Consider a policy \(\mathbf a=(a_u)_{u=t}^{T-1}\in\mathcal A_t\) and a probability measure
\(\mathbb P\in\mathfrak P_t^{\varepsilon,\delta}(x)\), such that the following holds

\begin{enumerate}
    \item The mid-price fulfils at all times $u\in \{0, \dots, T-1\}$
\[
\mathbb E_{\mathbb P}
\left[
S_{u+1}-S_u
\,\middle|\,
X_u,a_u
\right]
=
0
\]
    \item Conditional on \((X_u,a_u)\), the fills
    $
    \bigl(Q_{u+1}^{\operatorname{bid}},Q_{u+1}^{\operatorname{ask}}\bigr)
    $
    are independent of the price innovation \(S_{u+1}-S_u\)  at all times $u\in \{0, \dots, T-1\}$
\end{enumerate}
Then,  we have
\[
\mathbb E_{\PP}\!\left[
W_T + S_T I_T
-\sum_{u=t}^{T-1}
\frac{\gamma_u}{2}I_{u+1}^2\sigma^2\Delta_u 
\,\middle|\, X_t=x
\right]
=
W_t+S_tI_t
+
\mathbb E_{\mathbb P}
\left[
\sum_{u=t}^{T-1}
r(X_u,a_u,X_{u+1})\,\middle|\,
X_t=x
\right].
\]
\end{lem}
Lemma~\ref{lem_reward_approximation}  motivates to represent the objective from \eqref{eq:exact_value_function} by a cumulative reward formulation. Since the term
$W_t+S_tI_t$
does not depend on future decisions, it does not affect the optimizer. From now on we therefore consider a reduced robust value function of the form

\begin{equation}
\label{eq:final_additive_representation}
\sup_{\mathbf a\in\mathcal A_t}
\inf_{\PP\in\mathfrak P_t^{\varepsilon,\delta}(x)}
\mathbb E_{\PP}
\left[
\sum_{u=t}^{T-1}
r(X_u,a_u,X_{u+1})
\;\middle|\;
X_t=x
\right].
\end{equation}

The additive representation in \eqref{eq:final_additive_representation} transforms the robust terminal-value problem into a standard sequential decision problem and provides the foundation for the robust dynamic programming formulation developed in Section~\ref{sec:method}.

\subsubsection{Sinkhorn Ambiguity Sets}
\label{subsec:sinkhorn-ambiguity}

To model uncertainty in future market dynamics, we construct ambiguity sets using optimal transport distances. A desirable ambiguity set should allow transition dynamics to deviate from the reference model while preserving the geometric structure of the underlying state space. Optimal transport is particularly attractive because it measures uncertainty through the cost of transforming one probability distribution into another. Consequently, nearby distributions correspond to small perturbations of market dynamics, whereas distant distributions correspond to substantial structural changes.

The one-step evolution of the state $X_t$ from \eqref{eq:def_X} is only partially stochastic: part of the state evolves deterministically, while the remaining components depend on random market innovations. We therefore separate the deterministic state evolution from the stochastic innovation. Let $\Omega_{\operatorname{det}}\subset\Omloc$ be the product space of shifted history coordinates, namely the first $m-1$ entries of each rolling window, together with deterministic coordinates such as the updated time-to-close window. Let $\Omega_{\operatorname{rnd}}\subset\Omloc$ be the product space of the new stochastic market innovations appended at the next step, and write
\[
\pi:\Omloc\to\Omega_{\operatorname{det}}
\]
for the deterministic projection (or shift) map. Given $x\in\Omloc$, $a\in A$, and $z\in\Omega_{\operatorname{rnd}}$, the reconstruction map
\[
\Phi_\pi(x,a,z)\in\Omloc
\]
combines $\pi(x)$ with $z$ and the deterministic inventory, clipping, and time-to-close updates. Thus
\[
X_{t+1}=\Phi_\pi(X_t,a_t,\zeta_{t+1}),
\]
where $\zeta_{t+1}\in\Omega_{\operatorname{rnd}}$ has law $\widetilde{\PP}$. Accordingly, the product measure $\delta_{\pi(x)}\otimes\widetilde{\PP}$ is always interpreted through this reconstruction map, namely as its push-forward under
\[
(\pi(x),z)\mapsto \Phi_\pi(x,a,z).
\]
This convention ensures that ambiguity affects only the stochastic innovation law, while the history component of the next state remains deterministic.

Let $\widehat{\PP}(x,a)$ denote a reference transition law associated with state $x\in\Omloc$ and action $a\in A$, such as the market dynamics implied by the estimated transition model. Due to uncertainty, the true transition distribution may lie within a neighborhood of the reference distribution rather than coincide exactly with it. A natural ambiguity set can be constructed using the Wasserstein distance, which measures the minimum transportation cost required to transform one probability distribution into another. Specifically, for probability measures $\PP_1,\PP_2\in\mathcal M_1(\Omega_{\rm rnd})$, the Wasserstein distance is defined by

\begin{equation}
W(\PP_1,\PP_2)
:=
\inf_{\gamma\in\Pi(\PP_1,\PP_2)}
\int_{\Omega_{\rm rnd}\times\Omega_{\rm rnd}}
c(\tilde x,\tilde y)
\,d\gamma(\tilde x,\tilde y),
\label{eq:wasserstein}
\end{equation}
where $c(\cdot,\cdot)$ is a transportation cost function and $\Pi(\PP_1,\PP_2)$ denotes the collection of couplings having marginals $\PP_1$ and $\PP_2$.

Although Wasserstein ambiguity sets provide a flexible description of uncertainty, repeatedly solving Wasserstein transport problems within dynamic programming is computationally demanding. We therefore employ the entropically regularized optimal transport distance, commonly known as the Sinkhorn cost, which yields smoother transport plans, improves numerical efficiency, and remains a faithful approximation to the Wasserstein distance.

Let $\delta>0$ denote an entropic regularization parameter. The Sinkhorn cost between $\PP_1$ and $\PP_2$ is given by
\[
W_{\delta}(\mathbb{P}_1,\mathbb{P}_2)
:=
\inf_{\gamma \in \Pi(\mathbb{P}_1,\mathbb{P}_2)}
\left\{
\int_{\Omega_{\operatorname{rnd}}\times\Omega_{\operatorname{rnd}}}
c(\tilde x,\tilde y)\, d\gamma(\tilde x,\tilde y)
\;+\;
\delta\, H\!\left(\gamma \mid \mathbb{P}_1 \otimes \nu \right)
\right\},
\]
where $c$ is again the transport cost, $\nu$ is a reference sampling measure, and
$H(\gamma \mid \mathbb{P}_1 \otimes \nu)$ denotes the relative entropy of
$\gamma \in \Pi(\mathbb{P}_1,\mathbb{P}_2)$ with respect to the product measure
$\mathbb{P}_1 \otimes \nu$.

The first term measures how much the transition distribution must be altered relative to the reference model, while the second term penalizes highly concentrated transport plans through entropy regularization. Intuitively, the Wasserstein distance searches for the most adverse perturbation of the reference distribution, whereas the Sinkhorn formulation encourages uncertainty to be distributed more smoothly across alternative future states. The parameter $\delta$ governs this trade-off. Smaller values place greater emphasis on localized worst-case perturbations, whereas larger values produce smoother and more diffuse representations of uncertainty. As a result, $\delta$ influences not only the ambiguity set itself but also the manner in which decisions respond to adverse scenarios. 
The relative entropy is defined by
\[
H\!\left(\gamma \mid \mathbb{P}_1 \otimes \nu \right)
\;:=\;
\mathbb{E}_{(\tilde x,\tilde y)\sim\gamma}
\left[
\log\!\left(
\frac{d\gamma(\tilde x,\tilde y)}
     {d\mathbb{P}_1(\tilde x)\, d\nu(\tilde y)}
\right)
\right],
\]
where \(\frac{d\gamma}{d(\mathbb P_1\otimes\nu)}\) denotes the Radon--Nikodym derivative of $\gamma$ with respect to
$\mathbb P_1\otimes\nu$.
The entropy regularization yields smoother transport plans while preserving computational tractability. As $\delta\rightarrow 0$, the Sinkhorn cost converges to the Wasserstein distance, whereas larger values of $\delta$ place greater emphasis on entropy regularization.

The Sinkhorn ball characterizes uncertainty in one-step transition dynamics. To formulate a robust sequential decision problem, these local uncertainty sets must be extended to transition kernels and path measures over the entire trading horizon.

Given a reference innovation law $\widehat{\PP}(X_t,a_t) \in \mathcal{M}_1(\Omega_{\operatorname{rnd}})$, a radius $\varepsilon > 0$, and a regularization parameter $\delta > 0$, the Sinkhorn ball on the stochastic component is defined as
\begin{equation}\label{eq_definition_sinkhorn_ball_primal}
\mathcal{B}_{\varepsilon,\delta}\!\left(\widehat{\PP}(X_t,a_t)\right)
:=
\left\{
\mathbb{P} \in \mathcal{M}_1(\Omega_{\operatorname{rnd}}) ~\middle|~ W_{\delta}\!\left(\widehat{\mathbb{P}}(X_t,a_t),\mathbb{P}\right) \le \varepsilon
\right\}.
\end{equation}
As $W_\delta$ is the (non-symmetric) entropic OT cost rather than the debiased Sinkhorn divergence, the set $\mathcal{B}_{\varepsilon,\delta}\left(\widehat{\PP}(X_t,a_t)\right)$ need not be nonempty for arbitrary $\varepsilon$. The parameter $\varepsilon$ directly controls the size of the ambiguity set. Economically, it therefore determines the range of transition dynamics how much uncertainty the agent is willing to admit as plausible in the random part of future market dynamics. Larger values of $\varepsilon$ admit greater deviations from the reference model, whereas smaller values imply greater confidence in the estimated transition dynamics.

For $x\in\Omloc$ and $a\in A$, define the lifted local ambiguity set
\begin{equation}\label{eq_definition_projected_local_ambiguity}
\mathcal{P}^{\pi}_{\varepsilon,\delta}(x,a)
:=
\left\{
(\Phi_\pi^{x,a})_{\#}\!\left(\delta_{\pi(x)}\otimes\widetilde{\PP}\right)
~\middle|~
\widetilde{\PP}\in\mathcal{B}_{\varepsilon,\delta}\!\left(\widehat{\PP}(x,a)\right)
\right\}
\subset\mathcal{M}_1(\Omloc),
\end{equation}
where $\Phi_\pi^{x,a}$ denotes the reconstruction map $(\pi(x),z)\mapsto\Phi_\pi(x,a,z)$. Equivalently, under the product identification of the next-state space, one may write the elements of $\mathcal{P}^{\pi}_{\varepsilon,\delta}(x,a)$ as $\delta_{\pi(x)}\otimes\widetilde{\PP}$ with $\widetilde{\PP}$ ranging over the Sinkhorn ball around the reference innovation law.

\medskip
The construction extends one-step ambiguity sets to an entire trajectory by allowing each transition kernel to vary within its corresponding Sinkhorn neighborhood.
We lift the projected local correspondence to a global ambiguity set of path measures on $\mathcal{M}_1(\Omega)$: for every initial state $x \in \Omloc$ and every policy $\ab \in \mathcal{A}$,

\begin{equation}\label{eq_definition_Pepsilon}
\begin{aligned}
\mathfrak{P}_{x,\ab}^{\varepsilon, \delta}:=\bigg\{\delta_x \otimes \PP_0\otimes \cdots \otimes \PP_{T-1}~\bigg|~&\text{ for all } 0 \leq t \leq T-1:~\PP_t:\Omloc \rightarrow \mathcal{M}_1(\Omloc) \\
&\text{ Borel-measurable, and} \\
&\PP_t(\omega_t) \in
\left\{(\Phi_\pi^{\omega_t,a_t(\omega_t)})_{\#}\!\left(\delta_{\pi(\omega_t)}\otimes\widetilde{\PP}\right)~\middle|~
\widetilde{\PP}\in \mathcal{B}_{\varepsilon,\delta}\!\left(\widehat{\PP}(\omega_t,a_t(\omega_t))\right)\right\} \\
&\hspace{8.5cm}\text{for all } \omega_t\in \Omloc \bigg\} \subset \mathcal{M}_1(\Omega),
\end{aligned}
\end{equation}
where the notation {$\PP=\delta_x \otimes\PP_0\otimes \cdots \otimes \PP_{T-1}\in \mathfrak{P}_{x,\ab}^{\varepsilon,\delta}$} abbreviates, for all measurable sets $B \subset \Omega$,
\[
\PP(B):=\int_{\Omloc}\cdots \int_{\Omloc}\one_{B}\!\left((\omega_t)_{0 \leq t \leq T}\right) \PP_{T-1}(\omega_{T-1};\D\omega_T)\cdots \PP_0(\omega_0;\D\omega_1)\,\delta_x(\D \omega_0),
\]
and $\delta_x$ denotes the Dirac measure at $x$. Thus $\mathfrak{P}_{x,\ab}^{\varepsilon, \delta}$ contains path measures whose stochastic innovation laws remain within a Sinkhorn ball around the reference innovation model, while the deterministic history-shift component of each next state is fixed by $\pi$. Throughout the paper, robustness is interpreted relative to the baseline transition model. Accordingly, the ambiguity set is intended to capture uncertainty around the baseline market dynamics rather than introducing additional structural features such as adverse selection.

\subsection[Robust Dynamic Programming]{Robust Dynamic Programming}\label{sec:method}

The ambiguity sets introduced in Section~\ref{subsec:sinkhorn-ambiguity} define a family of plausible transition dynamics around the reference model. Since the robust objective is additive over time, the sequential decision problem can be solved recursively through dynamic programming. At each decision epoch, the market maker chooses an action that maximizes the current reward while accounting for the worst-case continuation value over all admissible transition distributions. This recursive formulation extends the classical Bellman principle to a distributionally robust setting and provides the foundation for the reinforcement learning algorithms developed later.

\subsubsection{Dynamic Programming Principle}
\label{subsec:global-robust-control}

To establish the dynamic programming principle rigorously, we first introduce a stopped and clipped reward function. Clipping ensures bounded rewards and facilitates the contraction arguments underlying the existence and uniqueness of the robust value function, while leaving the economic reward unchanged on the relevant state region.
\begin{equation}\label{eq}
r^C(x,a,y)
:= r\left(\Pi_{C_{\operatorname{global}}}(x),a,\Pi_{C_{\operatorname{global}}}(y)\right)\mathbf 1_{\{\tau(x)>0\}},
\end{equation}
where $\Pi_{C_{\operatorname{global}}}$ is the componentwise two-sided projection on the real-valued state coordinates defined in Section~\ref{transition_prob}. On trajectories for which the projection is inactive and before maturity, $r^C$ coincides with the economic reward $r$ in \eqref{eq_reward}.



Proposition~\ref{prop_2} establishes that the robust sequential decision problem admits a Bellman recursion despite the presence of distributional ambiguity. As in the classical MDP framework, the optimal value can be computed recursively by balancing the current reward against the continuation value. The key difference is that the continuation value is evaluated under the worst-case transition distribution within the Sinkhorn ambiguity set. Consequently, the optimal policy simultaneously optimizes current decisions and hedges against future model uncertainty.

\begin{prop}[Dynamic programming principle]\label{prop_2}~
The following assertions hold.

\begin{itemize}
    \item[(i)]
    The value function \(V\in C_b(\Omega_{\le T})\), defined in \eqref{eq_robust_problem_1}, is the unique fixed point of
    the robust Bellman operator
    \begin{equation}\label{eq_TV}       
    \T v(x)
    :=
    \sup_{a\in A}
    \inf_{\PP\in\mathcal{P}^{\pi}_{\varepsilon,\delta}(x,a)}
    \E_{\PP}
    \left[
    r^C(x,a,X_1)+\alpha v(X_1)
    \right],
    \qquad x\in\Omega_{\le T}.
    \end{equation}
    Moreover, there exist Borel-measurable selectors
    \[
    a_{\operatorname{loc}}^*:\Omega_{\le T}\to A,
    \qquad
    \PP_{\operatorname{loc}}^*:\Omega_{\le T}\to\mathcal{M}_1(\Omega_{\le T}),
    \]
    such that
    \[
    \PP_{\operatorname{loc}}^*(x)
    \in
    \mathcal{P}^{\pi}_{\varepsilon,\delta}
    \bigl(x,a_{\operatorname{loc}}^*(x)\bigr),
    \qquad x\in\Omega_{\le T},
    \]
    and, for every \(x\in\Omega_{\le T}\),
    \begin{equation}\label{eq_thm_assertion_1}
    \begin{aligned}
    V(x)=\T V(x)
    &=
    \sup_{a\in A}
    \inf_{\PP\in\mathcal{P}^{\pi}_{\varepsilon,\delta}(x,a)}
    \E_{\PP}
    \left[
    r^C(x,a,X_1)+\alpha V(X_1)
    \right]=
    \E_{\PP_{\operatorname{loc}}^*(x)}
    \left[
    r^C\bigl(x,a_{\operatorname{loc}}^*(x),X_1\bigr)
    +\alpha V(X_1)
    \right].
    \end{aligned}
    \end{equation}

    \item[(ii)]
    Define the stationary Markov policy
$
    \ab^*
    :=
    \bigl(a_{\operatorname{loc}}^*(X_t)\bigr)_{t\ge 0}
$
    and the corresponding infinite-horizon product measure
    \[
    \PP_x^{*,\infty}(d\omega_0,d\omega_1,\ldots)
    :=
    \delta_x(d\omega_0)
    \prod_{t=0}^{\infty}
    \PP_{\operatorname{loc}}^*(\omega_t;d\omega_{t+1}).
    \]
    Then
    $
    \PP_x^{*,\infty}
    \in
    \mathfrak{P}_{x,\ab^*}^{\varepsilon,\delta,\infty},
    $
    and, for every \(x\in\Omega_{\le T}\),
    \begin{equation}\label{eq_thm_assertion_3}
    \begin{aligned}
    V(x)
    &=
    \E_{\PP_x^{*,\infty}}
    \left[
    \sum_{t=0}^{\infty}
    \alpha^t
    r^C\bigl(X_t,a_{\operatorname{loc}}^*(X_t),X_{t+1}\bigr)
    \right] =
    \inf_{\PP\in\mathfrak{P}_{x,\ab^*}^{\varepsilon,\delta,\infty}}
    \E_{\PP}
    \left[
    \sum_{t=0}^{\infty}
    \alpha^t
    r^C\bigl(X_t,a_{\operatorname{loc}}^*(X_t),X_{t+1}\bigr)
    \right].
    \end{aligned}
    \end{equation}
If \(k:=\tau_{\operatorname{cur}}(x)\), then
    \begin{equation}\label{eq_thm_assertion_stopped}
    \begin{aligned}
    V(x)
    &=
    \E_{\PP_x^{*,\infty}}
    \left[
    \sum_{t=0}^{k-1}
    \alpha^t
    r^C\bigl(X_t,a_{\operatorname{loc}}^*(X_t),X_{t+1}\bigr)
    \right]=
    \inf_{\PP\in\mathfrak{P}_{x,\ab^*}^{\varepsilon,\delta}}
    \E_{\PP}
    \left[
    \sum_{t=0}^{k-1}
    \alpha^t
    r^C\bigl(X_t,a_{\operatorname{loc}}^*(X_t),X_{t+1}\bigr)
    \right].
    \end{aligned}
    \end{equation}
    In particular, for every initial full-horizon state
    \(x\in\Omega_T^{\operatorname{init}}\), we have \(k=T\), and therefore
    \begin{equation}\label{eq_thm_assertion_full_horizon}
    \begin{aligned}
    V(x)
    &=
    \E_{\PP_x^{*,\infty}}
    \left[
    \sum_{t=0}^{T-1}
    \alpha^t
    r^C\bigl(X_t,a_{\operatorname{loc}}^*(X_t),X_{t+1}\bigr)
    \right] =
    \inf_{\PP\in\mathfrak{P}_{x,\ab^*}^{\varepsilon,\delta}}
    \E_{\PP}
    \left[
    \sum_{t=0}^{T-1}
    \alpha^t
    r^C\bigl(X_t,a_{\operatorname{loc}}^*(X_t),X_{t+1}\bigr)
    \right].
    \end{aligned}
    \end{equation}
\end{itemize}
\end{prop}
The proposition shows that distributional robustness modifies the Bellman recursion only through the transition operator. Consequently, classical reinforcement learning algorithms can be adapted by replacing the standard expectation with a worst-case expectation over the ambiguity set. This observation forms the basis of the robust reinforcement learning algorithm developed in Section \ref{sec:algorithms}.

\subsubsection{Robust Bellman Operator via Sinkhorn Duality} \label{subsec:sinkhorn-dual}

The robust Bellman operator in \eqref{eq_TV} requires evaluating the worst-case continuation value over all transition distributions in the ambiguity set. Solving this inner optimization directly is computationally challenging because it is an infinite-dimensional optimization over probability measures. To obtain a tractable Bellman operator, we exploit the dual formulation of the Sinkhorn optimal transport problem, which transforms the inner minimization into an equivalent finite-dimensional optimization. Since the lifting map $(\Phi_\pi^{x,a})_\#$ is deterministic, the worst-case optimization can be carried out directly over the stochastic innovation law on $\Omega_{\mathrm{rnd}}$. Accordingly, fix $x\in\Omega_{\le T}$, $a\in A$, and a bounded measurable function $g:\Omega_{\mathrm{rnd}}\to\mathbb{R}$. The corresponding local Sinkhorn primal problem is
\begin{equation}\label{eq:sinkhorn_primal}
\begin{aligned}
V^P_{\varepsilon,\delta}(x,a;g)=
\inf_{\widetilde{\PP}\in\mathcal B_{\varepsilon,\delta}(\widehat\PP(x,a))}
\mathbb E_{\widetilde Z\sim\widetilde{\PP}}
\big[g(\widetilde Z)\big].
\end{aligned}
\end{equation}
This is the primal problem whose dual value is denoted by $V^D_{\varepsilon,\delta}(x,a;g)$ below. In particular, when
\begin{equation}\label{eq:bellman_innovation_payoff}
G_v^{x,a}(z)
:=
r^C\!\left(x,a,\Phi_\pi(x,a,z)\right)
+
\alpha v\!\left(\Phi_\pi(x,a,z)\right),
\qquad z\in\Omega_{\operatorname{rnd}},
\end{equation}
the inner minimization in \eqref{eq_TV} satisfies
\begin{equation}\label{eq:bellman_inner_primal_equivalence}
\inf_{\PP\in\mathcal P^{\pi}_{\varepsilon,\delta}(x,a)}
\mathbb E_{\PP}\!\left[r^C(x,a,X_1)+\alpha v(X_1)\right]
=
V^P_{\varepsilon,\delta}\!\left(x,a;G_v^{x,a}\right).
\end{equation}

We write $\tilde x\in\Omega_{\operatorname{rnd}}$ for an innovation sampled from the reference law $\widehat\PP(x,a)$, and $\tilde z\in\Omega_{\operatorname{rnd}}$ for a candidate perturbed innovation sampled from the reference measure $\nu$ or from the Sinkhorn kernel. Define
\begin{equation}\label{eq:sinkhorn_normalizer}
Z(\tilde x,\delta)
:=
\mathbb E_{\tilde z\sim\nu}
\left[
\exp\!\left(-\frac{c(\tilde x,\tilde z)}{\delta}\right)
\right],
\end{equation}
the shifted radius
\begin{equation*}\label{eq:shifted_radius_local}
\bar\varepsilon_{x,a}
:=
\varepsilon
+
\delta\,
\mathbb E_{\tilde x\sim\widehat\PP(x,a)}
\left[\log Z(\tilde x,\delta)\right],
\end{equation*}
and the Sinkhorn kernel
\begin{equation}\label{eq:sinkhorn_kernel_general}
dQ_{\tilde x,\delta}(\tilde z)
:=
\frac{
\exp\!\left(-c(\tilde x,\tilde z)/\delta\right)
}{
Z(\tilde x,\delta)
}
\,d\nu(\tilde z).
\end{equation}
Applying the strong duality result of \citet{wang2025sinkhorn} gives the reward-side dual
\begin{equation}\label{eq:dual_Sinkhorn1}
V^D_{\varepsilon,\delta}(x,a;g)
:=
\sup_{\lambda>0}
\left\{
-\lambda\bar\varepsilon_{x,a}
-
\lambda\delta\,
\mathbb E_{\tilde x\sim\widehat\PP(x,a)}
\left[
\log
\mathbb E_{\tilde z\sim Q_{\tilde x,\delta}}
\left[
\exp\!\left(-\frac{g(\tilde z)}{\lambda\delta}\right)
\right]
\right]
\right\}.
\end{equation}
Thus $V^D_{\varepsilon,\delta}$ is the Lagrange-dual value of the explicitly specified primal problem \eqref{eq:sinkhorn_primal}. If the cost $c(\tilde x,\tilde z)$ is $\widehat\PP(x,a)\otimes\nu$-measurable and satisfies
\begin{equation*}
\nu\!\left(\{\tilde z:0\le c(\tilde x,\tilde z)<\infty\}\right)=1
\quad\text{for }\widehat\PP(x,a)\text{-almost every }\tilde x,
\end{equation*}
then the primal problem is feasible if and only if $\bar\varepsilon_{x,a}\ge0$, and, whenever it is feasible, strong duality yields
\begin{equation}\label{eq:sinkhorn_primal_dual_equality}
V^P_{\varepsilon,\delta}(x,a;g)
=
V^D_{\varepsilon,\delta}(x,a;g).
\end{equation}

When the normalizing term in \eqref{eq:sinkhorn_normalizer} is independent of the anchor innovation, the shifted radius is a scalar tuning parameter. In that case we write $\bar\varepsilon$ instead of $\bar\varepsilon_{x,a}$. Substituting the payoff $G_v^{x,a}$ from \eqref{eq:bellman_innovation_payoff} into the dual representation \eqref{eq:dual_Sinkhorn1} gives a finite-dimensional representation of the robust Bellman update. For $\lambda>0$, define the projected Sinkhorn dual target
\begin{equation}\label{eq:sinkhorn_H}
\begin{aligned}
H_{\delta,v}^{\mathrm{sink}}(x,a;\lambda)
:=
&-\lambda\bar\varepsilon_{x,a}
-
\lambda\delta\,
\mathbb E_{\tilde x\sim\widehat\PP(x,a)}
\Bigg[
\log\,
\mathbb E_{\tilde z\sim Q_{\tilde x,\delta}}
\Bigg[
\exp\!\left(
-\frac{
G_v^{x,a}(\tilde z)
}{\lambda\delta}
\right)
\Bigg]
\Bigg]  \\
=
&-\lambda\bar\varepsilon_{x,a}
-
\lambda\delta\,
\mathbb E_{\tilde x\sim\widehat\PP(x,a)}
\Bigg[
\log\,
\mathbb E_{\tilde z\sim Q_{\tilde x,\delta}}
\Bigg[
\exp\!\left(
\frac{
-r^C\!\left(x,a,\Phi_\pi(x,a,\tilde z)\right)
-\alpha v\!\left(\Phi_\pi(x,a,\tilde z)\right)
}{\lambda\delta}
\right)
\Bigg]
\Bigg].
\end{aligned}
\end{equation}
The variable $\tilde x$ is the reference innovation around which the Sinkhorn kernel is centered, while $\tilde z$ is the perturbed innovation used to evaluate the next-state reward and continuation value. The Sinkhorn robust Bellman operator can therefore be written as
\begin{equation}\label{eq:sinkhorn_bellman_op}
\begin{aligned}
\mathcal T v(x)
&:=
\sup_{a\in A}
V^D_{\varepsilon,\delta}\!\left(x,a;G_v^{x,a}\right) =
\sup_{a\in A}\sup_{\lambda>0}
H_{\delta,v}^{\mathrm{sink}}(x,a;\lambda),
\qquad x\in\Omega_{\le T}.
\end{aligned}
\end{equation}
By \eqref{eq:bellman_inner_primal_equivalence} and strong duality, $\mathcal T$ is equivalent to the abstract Bellman operator $\T$ in \eqref{eq_TV}. Consequently, Proposition~\ref{prop_2} implies that the robust value function is the unique fixed point of \eqref{eq:sinkhorn_bellman_op}.

The dual representation also provides a clear interpretation of the two dimensions of robustness used throughout the paper. The shifted radius $\bar\varepsilon_{x,a}$, or simply $\bar\varepsilon$ under the weighted $\ell_1$ specialization, determines the size of the ambiguity set and therefore controls how much the transition distribution is allowed to deviate from the reference model. In contrast, the regularization parameter $\delta$ governs how the worst-case transition distribution is constructed within this admissible set. Specifically, $\delta$ enters both the Sinkhorn kernel $Q_{\tilde{x},\delta}$ and the exponential tilting in \eqref{eq:sinkhorn_H}, thereby controlling how probability mass is redistributed across perturbed innovations. Consequently, $\bar\varepsilon$ determines the amount of uncertainty admitted into the decision problem, whereas $\delta$ determines how conservatively the decision maker responds to that uncertainty. This distinction underlies the comparative statics and empirical analyses presented in the remainder of the paper.

\begin{cor}[Optimal policy]\label{sol_1}
A deterministic Markov policy $a^*:\Omega_{\le T}\to A$ is optimal for \eqref{eq_robust_problem_1} if and only if
\begin{equation*}
a^*(x)\in
\arg\max_{a\in A}\sup_{\lambda>0}
H_{\delta,V}^{\mathrm{sink}}(x,a;\lambda)
\qquad\text{for all }x\in\Omega_{\le T}.
\end{equation*}
\end{cor}

\begin{rem}
The inner $\sup_{\lambda > 0}$ in \eqref{eq:sinkhorn_bellman_op} is the one-dimensional Lagrange-dual problem associated with the local Sinkhorn primal \eqref{eq:sinkhorn_primal}. For computational efficiency, $\lambda$ is updated at each iteration via a gradient step, as in Algorithm~\ref{robust-rl}. The log-sum-exp structure shows that the worst-case innovation distribution is tilted proportionally to
\[
Q_{\tilde x,\delta}(d\tilde z)
\exp\!\left(-\frac{G_v^{x,a}(\tilde z)}{\lambda\delta}\right),
\]
so, for a reward-maximization problem with an inner infimum, the adversary emphasizes perturbed innovations that lead to low one-step reward plus continuation value after deterministic reconstruction, subject to the Sinkhorn transport penalty.
\end{rem}

The following result establishes convergence of value iteration which is our main method to solve the problem numerically.


\begin{prop}[Value iteration]\label{prop_VI}
For every $v_0\in C_b(\Omloc)$, we have
\[
\lim_{n\to\infty}\|\T^n v_0-V\|_\infty=0.
\]
Moreover, the convergence is geometric with contraction factor at most $\alpha$.
\end{prop}


\section{Reinforcement Learning: Algorithm and Implementation}
\label{sec:algorithms}

This section describes the implementation of the robust sequential decision framework developed in Section~\ref{sec:setting}. We first present the state representation and action space used by the market-making agent, then describe the reference transition model and the construction of the Sinkhorn perturbation kernel. Finally, we introduce the neural-network architecture and the fitted actor--critic algorithms used to solve the robust Bellman equation.

\subsection{State Representation}
The market-making agent observes a rolling-window state vector comprising inventory, recent returns, order flow, liquidity, volatility, and the remaining trading horizon. The rolling-window representation allows the agent to exploit short-run temporal dependence, volatility clustering, and order-flow persistence that are well documented in high-frequency financial markets.

\subsubsection{State Variables}
At time $t$, the state is represented by a vector $X_t$ containing the most recent $m$ observations of the following variables:
\begin{itemize}
    \item[(i)] Inventory
    $I^{t} = (I_{t-m+1}, \dots, I_t) \in \mathbb{R}^{m}$,

    \item[(ii)] Returns of the underlying asset
    $R^{t} = (R_{t-m+1}, \dots, R_t) \in \mathbb{R}^{m}$,

    \item[(iii)] Number of market orders
    $N^{t,\operatorname{bid}} = (N_{t-m+1}^{\operatorname{bid}}, \dots, N_t^{\operatorname{bid}}) \in \mathbb{R}^{m}$,
    $N^{t,\operatorname{ask}} = (N_{t-m+1}^{\operatorname{ask}}, \dots, N_t^{\operatorname{ask}}) \in \mathbb{R}^{m}$,

    \item[(iv)] Average size of market orders
    $V^{t,\operatorname{bid}} = (V_{t-m+1}^{\operatorname{bid}}, \dots, V_t^{\operatorname{bid}}) \in \mathbb{R}^{m}$,
    $V^{t,\operatorname{ask}} = (V_{t-m+1}^{\operatorname{ask}}, \dots, V_t^{\operatorname{ask}}) \in \mathbb{R}^{m}$,

    \item[(v)] Spread in ticks
    $\delta^{t,\operatorname{tick}} =
    (\delta_{t-m+1}^{\operatorname{tick}}, \dots, \delta_t^{\operatorname{tick}})
    \in \mathbb{R}^{m}$

    \item[(vi)] Relative spread
    $\delta^{t,\operatorname{rel}} =
    (\delta_{t-m+1}^{\operatorname{rel}}, \dots, \delta_t^{\operatorname{rel}})
    \in \mathbb{R}^{m}$

    \item[(vii)] Microprice deviation
    $D^{t,\operatorname{micro}} =
    (D_{t-m+1}^{\operatorname{micro}}, \dots, D_t^{\operatorname{micro}})
    \in \mathbb{R}^{m}$

    \item[(viii)] Running volatility
    $\nu^{t} = (\nu_{t-m+1}, \dots, \nu_t) \in \mathbb{R}^{m}$

    \item[(ix)] Order flow imbalance
    $\rm{OFI}^{t} = (\rm{OFI}_{t-m+1}, \dots, \rm{OFI}_t) \in \mathbb{R}^{m}$

    \item[(x)] Market-order self-excitation states
    $Z^{t,\operatorname{bid}} = (Z_{t-m+1}^{\operatorname{bid}}, \dots, Z_t^{\operatorname{bid}}) \in \mathbb{R}^{m}$,
    
    $Z^{t,\operatorname{ask}} = (Z_{t-m+1}^{\operatorname{ask}}, \dots, Z_t^{\operatorname{ask}}) \in \mathbb{R}^{m}$,

    \item[(xi)] Time-to-close
    {$\tau^t = (\tau_{t-m+1}, \dots, \tau_t) \in \mathbb{N}_0^m$.}
\end{itemize}

The resulting state vector is
\[
\Big(
I^t,R^t,
N^{t,\operatorname{bid}},N^{t,\operatorname{ask}},
V^{t,\operatorname{bid}},V^{t,\operatorname{ask}},
\delta^{t,\operatorname{tick}},\delta^{t,\operatorname{rel}},
D^{t,\operatorname{micro}},
\nu^t,
\operatorname{OFI}^{t},
Z^{t,\operatorname{bid}},
Z^{t,\operatorname{ask}},
\tau^t
\Big)
\in \Omega_{\mathrm{loc}},\]
where $\Omega_{\mathrm{loc}}$ denotes the state space introduced in Section~\ref{sec:setting}. 

\subsubsection{Feature Construction}
The variables entering the state are defined as follows.

{The microprice deviation measures how far the current mid-price is from the order-book-implied microprice, and is defined by}
\begin{align*}
{D_t^{\operatorname{micro}}}
&{:=}
{\frac{S_t^{\operatorname{ask,best}} N_t^{\operatorname{bid}}
    + S_t^{\operatorname{bid,best}} N_t^{\operatorname{ask}}}
     {N_t^{\operatorname{bid}} + N_t^{\operatorname{ask}}}
     - S_t^{\operatorname{mid}}.}
\end{align*}

The running volatility follows an exponentially weighted moving-average (EWMA) recursion
\begin{align*}
\nu_{t+1}
=
(1-\lambda)\nu_t + \lambda R_{t+1}^2,
\qquad 0 < \lambda \le 1,
\end{align*}
with initial value estimated from a warm-up sample.

Order flow imbalance is defined by:
\begin{align*}
{\operatorname{OFI}_{t+1}}
{= N_t^{\operatorname{ask}} - N_t^{\operatorname{bid}}.}
\end{align*}

The self-excitation state $Z_t$ captures the clustering behavior of market order arrivals and is constructed as a discrete-time approximation of a linear Hawkes process. 
Specifically, given observed market order counts $N_t^{\operatorname{ask}}$ and $N_t^{\operatorname{bid}}$, the excitation state is computed recursively as
\begin{equation}\label{eq:Z_update}
Z_{t+1}^{\operatorname{ask}} = e^{-\beta} Z_t^{\operatorname{ask}} + \alpha N_t^{\operatorname{ask}},\qquad Z_{t+1}^{\operatorname{bid}} = e^{-\beta} Z_t^{\operatorname{bid}} + \alpha N_t^{\operatorname{bid}},
\end{equation}
where $\beta > 0$ controls the decay rate of past activity and $\alpha > 0$ governs the magnitude of excitation induced by new arrivals. In practice, this recursion can be applied directly to observed data: starting from an initial value $Z_0$, the sequences $(Z_t^{\operatorname{bid}})_{t=1,\dots,T}$, $(Z_t^{\operatorname{ask}})_{t=1,\dots,T}$ are obtained by iterating the update rule from \eqref{eq:Z_update} using realized market order counts. 
This construction allows $Z_t$ to serve as a low-dimensional, Markovian summary of past order flow, and is therefore suitable for both simulated environments and real high-frequency data.

Moreover, we choose the time-varying risk-aversion coefficient in the reward function as 
\begin{equation}\label{eq:gamma_t}
\gamma_t := \gamma\!\left(1 + \frac{1}{(T-t)\Delta_t + \xi}\right),\qquad \text{ with } \Delta_t =t_{i+1}-t_i,
\end{equation}
which eventually increases the inventory penalty as the remaining time to maturity shrinks. Here $\gamma > 0$ is the baseline risk-aversion parameter from the objective \eqref{eq_reward} and $\xi > 0$ controls the rate of increase; one has $\gamma_t \approx \gamma$ for large remaining time $T - t$ and $\gamma_t$ reaches its maximum $\gamma(1 + 1/\xi)$ at $t = T$. Consequently, inventory positions become progressively more costly near the end of the trading horizon, reflecting the reduced opportunity to unwind inventory before maturity.

The selected features summarize inventory conditions, liquidity, volatility, and order-flow dynamics, all of which are known to be informative for short-horizon market-making decisions \citep{ContKukanovStoikov2014,CarteaJaimungalPenalva2015}. Using rolling windows rather than contemporaneous observations allows the agent to condition its decisions on recent market evolution. This representation enables the transition model to capture short-run serial dependence, volatility clustering, and persistent order-flow effects commonly observed in limit-order-book data. The state variables exhibit substantially different scales and units. Prior to training, each continuous state variable is standardized using statistics computed from the training sample.

\subsection{Action Space}

At each decision epoch, the market-making agent determines both the placement of its quotes and the fraction of incoming market flow that it is willing to transact. The spread in ticks and relative spread are given by
\begin{align*}
{\delta_t^{\operatorname{tick}}} &= {\frac{S_t^{\operatorname{ask, best}} - S_t^{\operatorname{bid, best}}}{\operatorname{ticksize}}}, \\
\delta_t^{\operatorname{rel}} &= \frac{S_t^{\operatorname{ask, best}} - S_t^{\operatorname{bid, best}}}{S_t^{\operatorname{mid}}},
\qquad \text{ with }
S_t = S_t^{\operatorname{mid}} = \frac{S_t^{\operatorname{ask, best}} + S_t^{\operatorname{bid, best}}}{2}.
\end{align*}
where $S_t^{\operatorname{ask, best}}$ and $S_t^{\operatorname{bid, best}}$ are the best ask and bid price at time $t$ and {$\operatorname{ticksize}$ is the tick size.}

The spread variables determine the profitability and execution probability of posted quotes. Narrow spreads increase execution likelihood but reduce per-trade profit, whereas wider spreads increase profit per trade at the expense of lower execution probabilities. The participation variables control the market maker's exposure to incoming order flow and therefore directly influence inventory accumulation and inventory risk.

Unlike classical market-making models that only optimize quote placement, the proposed framework jointly determines quote placement and participation rates. This additional flexibility allows the agent to actively manage both execution opportunities and inventory exposure under uncertain market conditions.

In the implementation, actions are constrained to lie within compact intervals,

\begin{equation}
\delta_t^{\operatorname{bid}},
\delta_t^{\operatorname{ask}}
\in
[\delta_{\min},\delta_{\max}],
\qquad
q_t^{\operatorname{bid}},
q_t^{\operatorname{ask}}
\in
[0,1],
\end{equation}
which guarantees feasibility of the learned policy and stabilizes the reinforcement-learning procedure.

\subsection[Reference Transition Model]{Reference Transition Model} \label{transition_prob}

The robust sequential decision framework developed in Section~\ref{sec:setting} requires a probabilistic model for the one-step evolution of the market-making state. Such a model describes the conditional distribution of the next state given the current state and action, which provides the nominal transition dynamics used by the reinforcement learning agent and defines the reference transition law around which the Sinkhorn ambiguity sets are constructed. The one-step state transition consists of a deterministic update and a stochastic market innovation. The deterministic component captures the predictable evolution of the state, whereas the stochastic component models uncertainty in future market conditions. This decomposition allows robustness to be introduced solely through stochastic innovation. To accommodate the complex and nonlinear dynamics of high-frequency financial markets, we model stochastic innovation using a data-driven probabilistic forecasting model. Let $\Omega_{\mathrm{det}}\subset\Omega_{\mathrm{loc}}$ denote the deterministic component of the state and let $\Omega_{\mathrm{rnd}}\subset\Omega_{\mathrm{loc}}$ denote the stochastic innovation space. We have $\pi:\Omega_{\mathrm{loc}}\rightarrow\Omega_{\mathrm{det}}$ for the deterministic map. Given the current state $x$, action $a$, and innovation $z\in\Omega_{\mathrm{rnd}}$, the reconstruction map \[ \Phi_\pi(x,a,z) \] combines the deterministic update with the innovation to produce the next state. This conditional distribution defines the reference transition law that is used throughout the remainder of the paper. 

\subsubsection{Deterministic State Update} \label{subsec:deterministic-update} The deterministic component describes the predictable evolution of the state once the current state, action, and market innovation are given. In the current design, it consists of the rolling-window shift of historical variables together with the deterministic updates of the remaining trading horizon. Specifically, we have
\begin{align*}\label{eq:piwin}   
\mathbb{R}^m \ni (x_1,\dots,x_m) &\mapsto \pi_{\operatorname{win}}(x_1,\dots,x_m) := (x_2,\dots,x_m) \in \R^{m-1}, \\
 \pi(y_1,\dots,y_m) &=
\begin{cases} 
(\pi_{\operatorname{win}}((\tau_{t-m+1},\dots,\tau_t)),\max\{\tau_t-1,0\} ) & \text{if } (y_1,\dots,y_m) = (\tau_{t-m+1},\dots,\tau_t), \\
(\pi_{\operatorname{win}}(y_1,\dots,y_m), \emptyset ) & \text{otherwise.}
\end{cases} 
\end{align*}
for the deterministic map. 

\subsubsection{Stochastic Innovation Model} \label{subsec:stochastic-innovation} The remaining uncertainty is represented by the innovation vector $Z_{t+1}$, which contains the new market information that arrives between two consecutive decision time points. Specifically, 
\[ Z_{t+1} = \left( R_{t+1}, N^{\mathrm{bid}}_{t+1}, N^{\mathrm{ask}}_{t+1}, V^{\mathrm{bid}}_{t+1}, V^{\mathrm{ask}}_{t+1}, \delta^{\mathrm{tick}}_{t+1}, \delta^{\mathrm{rel}}_{t+1}, D^{\mathrm{micro}}_{t+1}, \nu_{t+1}, \mathrm{OFI}_{t+1}, Z^{\mathrm{bid}}_{t+1}, Z^{\mathrm{ask}}_{t+1} \right), \] where the components describe future returns, order flow, liquidity, volatility, and other market microstructure variables.

The conditional distribution of $Z_{t+1}$ depends on the current state and action. Because these relationships are highly nonlinear and evolve over time, we estimate the innovation distribution using the deep ensemble model \citep{lakshminarayanan2017simple, ChuaCML18}, a data-driven probabilistic forecasting model. We train $K$ neural networks $\{\theta_k\}_{k=1}^K$ independently, each outputting component-wise mean and standard deviation functions {$\mu^{(c)}_{\theta_k}:\Omega_{\operatorname{loc}}\times A \to \mathbb R$ and $\sigma^{(c)}_{\theta_k}:\Omega_{\operatorname{loc}}\times A \to \mathbb R_{>0}$}. At each transition step, we draw
\begin{equation*}\label{eq:ref_measure_ensemble}
\begin{aligned}
&k \sim \operatorname{Uniform}\{1,\dots,K\}, \\[4pt]
&Z^{(c)} \,\big|\, X_t, k \;\sim\;
\begin{cases}
{\operatorname{LogNormal}\!\bigl(\mu^{(c)}_{\theta_k}(X_t,a_t),\; \sigma^{(c)}_{\theta_k}(X_t,a_t)^{2}\bigr)}, & c \in \mathcal C_{\geq 0}, \\[4pt]
{\mathcal N\!\bigl(\mu^{(c)}_{\theta_k}(X_t,a_t),\; \sigma^{(c)}_{\theta_k}(X_t,a_t)^{2}\bigr)}, & c \in \mathcal C_{\mathbb R}.
\end{cases}
\end{aligned}
\end{equation*}
with components drawn independently given {$(X_t,a_t,k)$}. 

\subsubsection{Reference Transition Law} \label{subsec:reference-transition} Combining the deterministic state update with the innovation model yields the one-step state transition 
\[ X_{t+1} = \Phi_\pi(X_t,a_t,Z_{t+1}). \] 
Specifically, we have
\begin{equation*}
\begin{aligned}\label{eq_X_t_1}
X_{t+1} & = \bigg( (\pi_{\operatorname{win}}(I^t), I_t+Q_{t+1}^{\operatorname{bid}}-Q_{t+1}^{\operatorname{ask}}),~
(\pi_{\operatorname{win}}(R^t), Z_{t+1}^{(1)}),~
(\pi_{\operatorname{win}}(N^{t,\operatorname{bid}}),Z_{t+1}^{(2)}),~\\
&(\pi_{\operatorname{win}}(N^{t,\operatorname{ask}}),Z_{t+1}^{(3)}),~
(\pi_{\operatorname{win}}(V^{t,\operatorname{bid}}),Z_{t+1}^{(4)}),~
(\pi_{\operatorname{win}}(V^{t,\operatorname{ask}}),Z_{t+1}^{(5)}),~
(\pi_{\operatorname{win}}(\delta^{t,\operatorname{tick}}),Z_{t+1}^{(6)}),~\\
&(\pi_{\operatorname{win}}(\delta^{t,\operatorname{rel}}),Z_{t+1}^{(7)}),~
(\pi_{\operatorname{win}}(D^{t,\operatorname{micro}}),Z_{t+1}^{(8)}),~
(\pi_{\operatorname{win}}(\nu^{t}),Z_{t+1}^{(9)}),~
(\pi_{\operatorname{win}}(\operatorname{OFI}^{t}),Z_{t+1}^{(10)}),~\\
&(\pi_{\operatorname{win}}(Z^{t,\operatorname{bid}}),Z_{t+1}^{(11)}),~
(\pi_{\operatorname{win}}(Z^{t,\operatorname{ask}}),Z_{t+1}^{(12)}),~
(\pi_{\operatorname{win}}(\tau^t),\max\{\tau_t-1,0\})
\bigg), \\
Q_{t+1}^{\operatorname{bid}}
&:= Z_{t+1}^{(3)} Z_{t+1}^{(5)}\,\varphi^{\operatorname{bid}}(\delta^{\operatorname{bid}}, q^{\operatorname{bid}}), \quad Q_{t+1}^{\operatorname{ask}} := Z_{t+1}^{(2)} Z_{t+1}^{(4)}\,\varphi^{\operatorname{ask}}(\delta^{\operatorname{ask}}, q^{\operatorname{ask}}),
\end{aligned}
\end{equation*} where $Z_{t+1}^{(1)},\dots,Z_{t+1}^{(12)}$ corresponds to the individual components of the innovation vector $Z_{t+1}$, and \(\varphi^{\operatorname{bid}},\varphi^{\operatorname{ask}}\in[0,1]\) are continuous fill-ratio functions. Since the deterministic update is fixed, the conditional distribution of the next state is entirely determined by the conditional distribution of the innovation. Accordingly, the reference transition law is defined by \[ \widehat{\PP}(x,a) = \mathcal{L}\!\left( \Phi_\pi(x,a,Z) \right), \] where $Z$ follows the estimated conditional innovation distribution given the current state $x$ and action $a$, and $\mathcal{L}(\cdot)$ denotes the induced probability law. The resulting reference transition law provides the nominal dynamics for the reinforcement learning algorithm and serves as the center of the Sinkhorn ambiguity set introduced in Section~\ref{subsec:sinkhorn-ambiguity}.

\subsection[Training Algorithms]{Training Algorithms}\label{subsec:training-algorithms}
We train our agent according to the following numerical routine which is based on Proposition~\ref{prop_VI}. Algorithm~\ref{greedy-rl} describes the standard value iteration reinforcement learning approach, whereas Algorithm~\ref{robust-rl} incorporates distributional robustness via Sinkhorn divergence.

\subsubsection[Non-Robust Fitted Actor--Critic]{Non-Robust Fitted Actor--Critic}\label{subsubsec:nonrobust-actor-critic}~

\begin{algorithm}[H]
\scriptsize
\caption{Fitted Actor--Critic (Value Iteration)}
\KwIn{{\color{blue}Batch size $B$; number of outer iterations $N_{\mathrm{iter}}$; number of critic steps $I_V$; number of actor steps $I_A$; discount factor $\alpha\in(0,1)$; empirical innovation model $\widehat{P}$ and deterministic lift $\Phi_\pi$; reward function $r(x,a,x')$; number of innovation samples $N$ from $\widehat{P}$; learning rates for value net and policy net.}}
\KwOut{Critic network $V$; actor network $A$.}

Initialize critic network $V^{0}$ and actor network $A^{0}$\;

\For{$n=1,\ldots,N_{\mathrm{iter}}$}{
    \textbf{Target update:} set $V^{\mathrm{old}} \leftarrow V^{n-1}$ and freeze the weights of $V^{\mathrm{old}}$\;

    \textbf{Warm start:} initialize $V^{n} \leftarrow V^{n-1}$ and $A^{n} \leftarrow A^{n-1}$\;

    \textbf{Data collection:} collect a replay buffer of state--action pairs
    $(x_t,a_t)$ using the current policy $A^{n-1}$\;

    \tcc*[h]{\scriptsize Critic step: regress $V$ to the Bellman target.} \\
    \For{$s=1,\ldots,I_V$}{
        Sample a minibatch of state--action pairs $(x_i,a_i)_{i=1}^B$ from the replay buffer\;

        {\color{blue}For each $i$, draw $N$ innovation samples $\widetilde X'_{i,j} \sim \widehat{P}(\cdot\mid x_i,a_i)$ and set $X'_{i,j}:=\Phi_\pi(x_i,a_i,\widetilde X'_{i,j})$, $j=1,\ldots,N$\;}

        Compute the Monte Carlo Bellman target
        \[
        V_{\text{target}}(x_i, a_i)
        :=
        \frac{1}{N}\sum_{j=1}^{N}
        \Bigl[r\bigl(x_i,a_i,X'_{i,j}\bigr)
            +\alpha\, V^{\mathrm{old}}\!\bigl(X'_{i,j}\bigr)\Bigr]
        \]

        Minimize
        \[
        \sum_{i=1}^{B}
        \left(
        V^{n}(x_i)
        -
        V_{\text{target}}(x_i, a_i)
        \right)^2
        \]
        with respect to parameters of $V^{n}$\;
    }

    \tcc*[h]{\scriptsize Actor step: improve the policy by maximizing the Bellman target.} \\
    \For{$s=1,\ldots,I_A$}{
        {\color{blue}Sample a minibatch of states $(x_i)_{i=1}^B$ from the replay buffer; let $a_i := A^{n}(x_i)$, draw $\widetilde X'_{i,j} \sim \widehat{P}(\cdot\mid x_i,a_i)$, and set $X'_{i,j}:=\Phi_\pi(x_i,a_i,\widetilde X'_{i,j})$, $j=1,\ldots,N$\;}

        Maximize
        \[
        \sum_{i=1}^{B} V_{\text{target}}(x_i, a_i)
        \]
        with respect to parameters of $A^{n}$ (with $V^{n}$ fixed)\;
    }
}
\label{greedy-rl}
\end{algorithm}

\subsubsection[Sinkhorn-Robust Fitted Actor--Critic]{Sinkhorn-Robust Fitted Actor--Critic}\label{subsubsec:robust-actor-critic}~

\begin{algorithm}[H]
\scriptsize
\caption{Sinkhorn-Robust Fitted Actor--Critic (Value Iteration)}
\KwIn{{\color{blue}Batch size $B$; number of outer iterations $N_{\mathrm{iter}}$; number of critic steps $I_V$; number of actor steps $I_A$; number of dual steps $I_\lambda$; discount factor $\alpha\in(0,1)$; Sinkhorn parameters $(\bar{\varepsilon},\delta)$; innovation kernel $Q_{\tilde x,\delta}$; empirical innovation model $\widehat{P}$ and deterministic lift $\Phi_\pi$; cost $c(\cdot,\cdot)$ on $\Omega_{\operatorname{rnd}}$; reward function $r(x,a,x')$; number of samples $N$ from $\widehat{P}$ and $M$ from $Q_{\tilde x,\delta}$; learning rates for value net, policy net, and lambda net.}}
\KwOut{Critic network $V$; actor network $A$; lambda network $\Lambda$.}

Initialize critic network $V^{0}$, actor network $A^{0}$, and lambda network $\Lambda^{0}$ with parameters $\theta_\Lambda^{0}$ (constrained so that $\Lambda^{0}(x)\geq 0$ for all $x$, e.g., via a softplus output)\;

\For{$n=1,\ldots,N_{\mathrm{iter}}$}{
    \textbf{Target update:} set $V^{\mathrm{old}} \leftarrow V^{n-1}$ and freeze the weights of $V^{\mathrm{old}}$\;

    \textbf{Warm start:} initialize $V^{n} \leftarrow V^{n-1}$, $A^{n} \leftarrow A^{n-1}$, and $\theta_\Lambda^{n} \leftarrow \theta_\Lambda^{n-1}$\;

    \textbf{Data collection:} collect a replay buffer of state--action pairs
    $(x_t,a_t)$ using the current policy $A^{n-1}$\;

    \tcc*[h]{\scriptsize Dual step: update parameters $\theta_\Lambda$ of $\Lambda^{n}$ by maximizing the Sinkhorn dual on minibatches.} \\
    \For{$s=1,\ldots,I_\lambda$}{
        Sample a minibatch of state--action pairs $(x_i,a_i)_{i=1}^B$ from the replay buffer\;

        \For{$i=1,\ldots,B$}{
            Set $\lambda_i := \Lambda^{n}(x_i;\theta_\Lambda^{n})$\;

            {\color{blue}Draw $N$ reference innovation samples $\widetilde X'_{i,j} \sim \widehat{P}(\cdot\mid x_i,a_i)$, $j=1,\ldots,N$\;}

            {\color{blue}For each $j$, draw $M$ perturbed innovation samples $\widetilde Z_{i,j,l} \sim Q_{\widetilde X'_{i,j},\delta}$ and set $Z^{\pi}_{i,j,l}:=\Phi_\pi(x_i,a_i,\widetilde Z_{i,j,l})$, $l=1,\ldots,M$\;}

            Evaluate the Sinkhorn dual target
            \[
            H_{\delta}^{\mathrm{sink}}(x_i,a_i;\lambda_i)
            :=
            -\lambda_i\,\bar{\varepsilon}
            -\lambda_i\,\delta\;
            \frac{1}{N}\sum_{j=1}^{N}
            \log\!\left(
            \frac{1}{M}\sum_{l=1}^{M}
            \exp\!\left(
            \frac{
            -r(x_i,a_i,{\color{blue}Z^{\pi}_{i,j,l}})
            -\alpha V^{\mathrm{old}}({\color{blue}Z^{\pi}_{i,j,l}})
            }{\lambda_i\,\delta}
            \right)
            \right)\!.
            \]
        }

        Update $\theta_\Lambda^{n}$ by \textbf{gradient ascent} on
        $\sum_{i=1}^{B} H_{\delta}^{\mathrm{sink}}(x_i,a_i;\lambda_i)$
        (back-propagating through $\lambda_i = \Lambda^{n}(x_i;\theta_\Lambda^{n})$ only; treat $V^{\mathrm{old}}$ and the actor as fixed)\;
    }

    {\color{blue}\emph{In the critic and actor steps below, every evaluation of $H_\delta^{\mathrm{sink}}$ uses the Monte Carlo estimate from the dual step with fresh samples $\widetilde X'_{i,j} \sim \widehat{P}(\cdot\mid x_i,a_i)$, perturbed innovations $\widetilde Z_{i,j,l} \sim Q_{\widetilde X'_{i,j},\delta}$, and reconstructed states $Z^{\pi}_{i,j,l}=\Phi_\pi(x_i,a_i,\widetilde Z_{i,j,l})$.}}

    \tcc*[h]{\scriptsize Critic step: regress $V$ to the Sinkhorn robust Bellman target.} \\
    \For{$s=1,\ldots,I_V$}{
        Sample a minibatch of state--action pairs $(x_i,a_i)_{i=1}^B$; set $\lambda_i := \Lambda^{n}(x_i;\theta_\Lambda^{n})$\;

        Minimize
        \[
        \sum_{i=1}^{B}
        \left(
        V^{n}(x_i) - H_{\delta}^{\mathrm{sink}}(x_i,a_i;\lambda_i)
        \right)^2
        \]
        with respect to parameters of $V^{n}$ (with $\theta_\Lambda^{n}$ fixed)\;
    }

    \tcc*[h]{\scriptsize Actor step: improve the policy by maximizing the robust objective.} \\
    \For{$s=1,\ldots,I_A$}{
        Sample a minibatch of states $(x_i)_{i=1}^B$; let $a_i := A^{n}(x_i)$ and $\lambda_i := \Lambda^{n}(x_i;\theta_\Lambda^{n})$\;

        Maximize
        $\sum_{i=1}^{B} H_{\delta}^{\mathrm{sink}}(x_i,a_i;\lambda_i)$
        with respect to parameters of $A^{n}$ (with $\theta_\Lambda^{n}, V^{n}$ fixed)\;
    }
}
\label{robust-rl}
\end{algorithm}

\section{Numerical and Empirical Analysis} \label{sec:numericals}
This section evaluates the proposed robust reinforcement learning framework through simulation and empirical studies. The numerical analysis is organized around three questions. First, how does robustness reshape sequential decision behavior? Second, how do uncertainty tolerance and action robustness differ in their effects on policy adaptation? Third, under what market conditions does robustness improve performance, and when can it become overly conservative?

The simulation study considers six market environments: a stable baseline, price stress, liquidity dry-out, buy-arrival imbalance, sell-arrival imbalance, and fill stress. These scenarios isolate different forms of evolving uncertainty, including volatility shocks, deteriorating liquidity conditions, directional order-flow pressure, and execution uncertainty. To examine how robustness influences sequential decision making, we compare the greedy benchmark with robust policies corresponding to different combinations of the robustness parameters $(\bar{\varepsilon},\delta)$. As discussed in Section~\ref{sec:method}, $\bar{\varepsilon}$ determines the amount of uncertainty admitted into the decision problem, whereas $\delta$ governs how conservatively the policy responds to that uncertainty. Their interaction generates economically distinct market-making behaviors, ranging from aggressive liquidity provision to more defensive inventory management.

The empirical analysis further evaluates the framework on real high-frequency market data across multiple years and market regimes. Together, the simulation and empirical results show that robustness does not merely improve performance statistics; it fundamentally changes how the market maker adapts spreads, executions, and inventory exposure under evolving uncertainty.

\subsection{Simulation}\label{sec:simulation}

The simulation environment is calibrated using real high-frequency market data. The baseline market dynamics, including order arrivals, order sizes, volatility, liquidity conditions, and execution behavior, are estimated from historical observations and serve as the reference environment for policy learning. Building upon this baseline, we construct a collection of controlled stress scenarios that introduce deviations from the reference dynamics while preserving realistic market microstructure characteristics. These perturbations are designed to mimic economically relevant forms of evolving uncertainty commonly observed in practice, including volatility shocks, liquidity deterioration, directional order-flow imbalances, and execution uncertainty. Details of the simulation setup, scenario construction, and parameter choices are provided in Appendix~\ref{appendix:simulation_details}. Additional results such as full robustness hyperparameter tables, Pareto frontier performance plots, and full agent behavior plots can be found in Appendix~\ref{appendix:add_sim_results}.

\subsubsection{Robustness under Evolving Uncertainty}
\label{subsec:robustness_results}

Table~\ref{tab:rlmm_sim_results} reports the simulation performance across the six market environments. The robust policies are selected either by validation P\&L or by validation Sharpe ratio. Economically, these two selection rules correspond to different forms of robustness. The validation-P\&L policies remain relatively close to the greedy benchmark and primarily preserve profitability and execution opportunities, whereas the validation-Sharpe policies induce substantially stronger risk stabilization.

\begin{table}[htb]
\centering
\caption{Greedy versus robust performance across simulation scenarios. Percentage changes relative to the greedy benchmark are reported for the robust policy selected by validation Sharpe.}
\label{tab:rlmm_sim_results}

\scriptsize
\setlength{\tabcolsep}{3pt}
\renewcommand{\arraystretch}{1.08}

\begin{tabular}{@{}lrrr@{\hspace{1.5em}}lrrr@{}}
\toprule \toprule

\multicolumn{4}{c}{\textbf{Stable}} 
& \multicolumn{4}{c}{\textbf{Price Stress}} \\
\cmidrule(lr){1-4} \cmidrule(lr){5-8}
 & Greedy & \shortstack{Robust\\(Val PnL)} & \shortstack{Robust\\(Val Sharpe)}
 &  & Greedy & \shortstack{Robust\\(Val PnL)} & \shortstack{Robust\\(Val Sharpe)} \\
\midrule

Mean PnL 
& 1885.31 & \textbf{1886.04} & 1845.06 {\scriptsize $\downarrow 2.1\%$}
& Mean PnL 
& 1898.50 & \textbf{1900.88} & 1846.74 {\scriptsize $\downarrow 2.7\%$} \\

Std PnL 
& 49.37 & 49.17 & \textbf{39.23} {\scriptsize $\downarrow 20.5\%$}
& Std PnL 
& 155.96 & 165.11 & \textbf{124.65} {\scriptsize $\downarrow 20.1\%$} \\

\shortstack[l]{Mean Max\\Drawdown} 
& -6.26 & -5.60 & \textbf{-4.61} {\scriptsize $\downarrow 26.3\%$}
& \shortstack[l]{Mean Max\\Drawdown} 
& -5.20 & -4.09 & \textbf{-4.38} {\scriptsize $\downarrow 15.8\%$} \\

\shortstack[l]{Mean Final\\Inv.} 
& 7.91 & 3.04 & \textbf{-5.33} {\scriptsize $\downarrow 32.6\%$}
& \shortstack[l]{Mean Final\\Inv.} 
& 5.53 & 0.01 & \textbf{-1.99} {\scriptsize $\downarrow 64.0\%$} \\

Sharpe 
& 38.19 & 38.36 & \textbf{47.03} {\scriptsize $\uparrow 23.1\%$}
& Sharpe 
& 12.17 & 11.51 & \textbf{14.82} {\scriptsize $\uparrow 21.8\%$} \\

\addlinespace[0.8em]

\multicolumn{4}{c}{\textbf{Liquidity Dry-Out}} 
& \multicolumn{4}{c}{\textbf{Buy Arrival Imbalance}} \\
\cmidrule(lr){1-4} \cmidrule(lr){5-8}
 & Greedy & \shortstack{Robust\\(Val PnL)} & \shortstack{Robust\\(Val Sharpe)}
 &  & Greedy & \shortstack{Robust\\(Val PnL)} & \shortstack{Robust\\(Val Sharpe)} \\
\midrule

Mean PnL 
& 989.43 & \textbf{994.10} & 983.61 {\scriptsize $\downarrow 0.6\%$}
& Mean PnL 
& 2644.38 & 2574.74 & \textbf{2649.14} {\scriptsize $\uparrow 0.2\%$} \\

Std PnL 
& 31.44 & 33.66 & \textbf{23.64} {\scriptsize $\downarrow 24.8\%$}
& Std PnL 
& 156.26 & 190.57 & \textbf{93.38} {\scriptsize $\downarrow 40.2\%$} \\

\shortstack[l]{Mean Max\\Drawdown} 
& -22.25 & -7.80 & \textbf{-3.96} {\scriptsize $\downarrow 82.2\%$}
& \shortstack[l]{Mean Max\\Drawdown} 
& -1.00 & \textbf{-0.97} & -1.22 {\scriptsize $\downarrow 22.0\%$} \\

\shortstack[l]{Mean Final\\Inv.} 
& 6.17 & 2.17 & \textbf{-0.92} {\scriptsize $\downarrow 85.1\%$}
& \shortstack[l]{Mean Final\\Inv.} 
& -88.15 & -94.42 & \textbf{-34.05} {\scriptsize $\downarrow 61.4\%$} \\

Sharpe 
& 31.47 & 29.54 & \textbf{41.62} {\scriptsize $\uparrow 32.3\%$}
& Sharpe 
& 16.92 & 13.51 & \textbf{28.37} {\scriptsize $\uparrow 67.7\%$} \\

\addlinespace[0.8em]

\multicolumn{4}{c}{\textbf{Sell Arrival Imbalance}} 
& \multicolumn{4}{c}{\textbf{Fill Stress}} \\
\cmidrule(lr){1-4} \cmidrule(lr){5-8}
 & Greedy & \shortstack{Robust\\(Val PnL)} & \shortstack{Robust\\(Val Sharpe)}
 &  & Greedy & \shortstack{Robust\\(Val PnL)} & \shortstack{Robust\\(Val Sharpe)} \\
\midrule

Mean PnL 
& 2586.82 & \textbf{2608.21} & 2574.89 {\scriptsize $\downarrow 0.5\%$}
& Mean PnL 
& 428.72 & \textbf{433.26} & 384.29 {\scriptsize $\downarrow 10.4\%$} \\

Std PnL 
& 200.90 & 188.23 & \textbf{109.89} {\scriptsize $\downarrow 45.3\%$}
& Std PnL 
& 15.64 & 15.21 & \textbf{11.07} {\scriptsize $\downarrow 29.2\%$} \\

\shortstack[l]{Mean Max\\Drawdown} 
& -107.81 & -41.29 & \textbf{-33.98} {\scriptsize $\downarrow 68.4\%$}
& \shortstack[l]{Mean Max\\Drawdown} 
& -13.90 & -9.30 & \textbf{-7.72} {\scriptsize $\downarrow 44.4\%$} \\

\shortstack[l]{Mean Final\\Inv.} 
& 125.04 & 106.72 & \textbf{53.38} {\scriptsize $\downarrow 57.3\%$}
& \shortstack[l]{Mean Final\\Inv.} 
& 4.76 & -0.33 & \textbf{-2.42} {\scriptsize $\downarrow 49.2\%$} \\

Sharpe 
& 12.88 & 13.86 & \textbf{23.43} {\scriptsize $\uparrow 81.9\%$}
& Sharpe 
& 27.41 & 28.49 & \textbf{34.72} {\scriptsize $\uparrow 26.7\%$} \\

\bottomrule \bottomrule
\end{tabular}

\vspace{0.4em}

\begin{minipage}{0.96\linewidth}
\footnotesize
Notes: Robust (Val PnL) reports the best-performing robust configuration selected by validation P\&L; Robust (Val Sharpe) reports the best-performing robust configuration selected by validation Sharpe. Percentage changes are computed relative to the greedy benchmark. For volatility and inventory, downward arrows indicate reductions and therefore correspond to improved stabilization performance. For all six scenarios, the selected $(\bar{\epsilon}, \delta)$ is $(0.0001,0.1)$ under validation P\&L and $(4,1)$ under validation Sharpe.
\end{minipage}

\end{table}

The first important result is that robustness primarily improves the stability of sequential outcomes rather than uniformly increasing raw profitability. In the stable baseline, the validation-Sharpe robust policy increases the Sharpe ratio from 38.19 to 47.03, corresponding to a 23.1\% improvement, while reducing volatility from 49.37 to 39.23 and improving drawdowns from -6.26 to -4.61. Although average P\&L decreases slightly, the reduction in sequential risk dominates economically. This suggests that robustness changes the structure of the sequential policy rather than simply generating more aggressive spread capture.

The stabilizing role of robustness becomes substantially more pronounced under stressed environments. Under price stress, the robust policy improves Sharpe by 21.8\%, reduces volatility by 20.1\%, and improves drawdowns by 15.8\%. Under liquidity dry-out, Sharpe increases from 31.47 to 41.62, corresponding to a 32.3\% improvement, while inventory exposure declines by 85.1\%. The strongest effects appear under directional order-flow imbalance. Under buy-arrival imbalance, robustness improves Sharpe from 16.92 to 28.37, corresponding to a 67.7\% increase, while volatility declines by 40.2\%. Under sell-arrival imbalance, Sharpe improves from 12.88 to 23.43, representing an 81.9\% increase, while volatility falls by 45.3\%.

Inventory stabilization emerges as a particularly important mechanism. Under buy-arrival imbalance, the greedy benchmark accumulates a large negative inventory of -88.15, whereas the robust policy reduces this exposure to -34.05. Under sell-arrival imbalance, robustness sharply reduces positive inventory accumulation from 125.04 to 53.38. These results show that robustness becomes especially valuable when persistent order-flow pressure creates directional inventory risk.

The fill-stress scenario provides a complementary interpretation. When execution quality deteriorates, the validation-P\&L robust policy preserves profitability while the validation-Sharpe robust policy substantially stabilizes sequential outcomes through lower volatility and drawdowns. This confirms that the two robust selection rules correspond to different operating points on the profitability--stability frontier.

Overall, Table~\ref{tab:rlmm_sim_results} supports four conclusions. First, robustness becomes increasingly valuable as the testing environment departs from the stable training regime. Second, robustness primarily improves risk-adjusted performance through volatility and inventory stabilization rather than through uniformly higher profits. Third, robustness reacts differently across uncertainty sources and is therefore state dependent. Fourth, different robustness levels generate economically distinct market-making behaviors ranging from aggressive liquidity provision to defensive inventory stabilization.

\subsubsection{Sequential Policy Adaptation Under Uncertainty}
\label{sec:simulation_behavior}

We next examine how robustness changes the learned policy itself. The goal is not only to compare performance, but to identify how robust reinforcement learning reshapes the market maker's state-dependent quoting behavior. We therefore use three complementary pieces of evidence.  Table~\ref{tab:rlmm_spread_summary} summarizes the resulting spread behavior numerically. Figure~\ref{fig:price_intraday} illustrates how the robustness parameters affect the full action vector in one representative scenario. Figures~\ref{fig:cross_sharpe} and~\ref{fig:cross_state} then examine whether these behavioral changes persist across scenarios and across market states. 

\newcommand{\up}[1]{\hfill{\scriptsize($\uparrow$~#1)}}
\newcommand{\down}[1]{\hfill{\scriptsize($\downarrow$~#1)}}

\begin{table}[htb]
\centering
\caption{Quoting behavior across simulation scenarios. Percentage changes are reported relative to the greedy benchmark for the robust policy selected by validation Sharpe.}
\label{tab:rlmm_spread_summary}

\scriptsize
\setlength{\tabcolsep}{3pt}
\renewcommand{\arraystretch}{1.08}

\begin{tabular}{@{}lrrr@{\hspace{1.5em}}lrrr@{}}
\toprule\toprule

\multicolumn{4}{c}{\textbf{Stable}}
&
\multicolumn{4}{c}{\textbf{Price Stress}}
\\
\cmidrule(lr){1-4}
\cmidrule(lr){5-8}

&
Greedy
&
\shortstack{Robust\\(Val PnL)}
&
\shortstack{Robust\\(Val Sharpe)}
&
&
Greedy
&
\shortstack{Robust\\(Val PnL)}
&
\shortstack{Robust\\(Val Sharpe)}
\\

\midrule

Mean Bid Spread
& 0.083 & 0.081 & 0.090 \up{8\%}
&
Mean Bid Spread
& 0.083 & 0.082 & 0.090 \up{8\%}
\\

Std. Bid Spread
& 0.009 & 0.009 & 0.012 \up{33\%}
&
Std. Bid Spread
& 0.009 & 0.009 & 0.012 \up{33\%}
\\

Mean Ask Spread
& 0.076 & 0.075 & 0.084 \up{11\%}
&
Mean Ask Spread
& 0.077 & 0.075 & 0.083 \up{8\%}
\\

Std. Ask Spread
& 0.009 & 0.010 & 0.012 \up{33\%}
&
Std. Ask Spread
& 0.009 & 0.010 & 0.012 \up{33\%}
\\

\addlinespace[0.8em]

\multicolumn{4}{c}{\textbf{Liquidity Dry-Out}}
&
\multicolumn{4}{c}{\textbf{Buy Arrival Imbalance}}
\\
\cmidrule(lr){1-4}
\cmidrule(lr){5-8}

&
Greedy
&
\shortstack{Robust\\(Val PnL)}
&
\shortstack{Robust\\(Val Sharpe)}
&
&
Greedy
&
\shortstack{Robust\\(Val PnL)}
&
\shortstack{Robust\\(Val Sharpe)}
\\

\midrule

Mean Bid Spread
& 0.086 & 0.084 & 0.087 \up{1\%}
&
Mean Bid Spread
& 0.070 & 0.078 & 0.065 \down{7\%}
\\

Std. Bid Spread
& 0.012 & 0.010 & 0.011 \down{8\%}
&
Std. Bid Spread
& 0.016 & 0.021 & 0.012 \down{25\%}
\\

Mean Ask Spread
& 0.082 & 0.079 & 0.082 (--\%)
&
Mean Ask Spread
& 0.125 & 0.131 & 0.120 \down{4\%}
\\

Std. Ask Spread
& 0.012 & 0.012 & 0.011 \down{8\%}
&
Std. Ask Spread
& 0.023 & 0.025 & 0.011 \down{52\%}
\\

\addlinespace[0.8em]

\multicolumn{4}{c}{\textbf{Sell Arrival Imbalance}}
&
\multicolumn{4}{c}{\textbf{Fill Stress}}
\\
\cmidrule(lr){1-4}
\cmidrule(lr){5-8}

&
Greedy
&
\shortstack{Robust\\(Val PnL)}
&
\shortstack{Robust\\(Val Sharpe)}
&
&
Greedy
&
\shortstack{Robust\\(Val PnL)}
&
\shortstack{Robust\\(Val Sharpe)}
\\

\midrule

Mean Bid Spread
& 0.128 & 0.126 & 0.126 \down{2\%}
&
Mean Bid Spread
& 0.080 & 0.079 & 0.088 \up{10\%}
\\

Std. Bid Spread
& 0.018 & 0.015 & 0.012 \down{33\%}
&
Std. Bid Spread
& 0.005 & 0.005 & 0.006 \up{20\%}
\\

Mean Ask Spread
& 0.068 & 0.064 & 0.064 \down{6\%}
&
Mean Ask Spread
& 0.077 & 0.076 & 0.085 \up{10\%}
\\

Std. Ask Spread
& 0.016 & 0.010 & 0.011 \down{31\%}
&
Std. Ask Spread
& 0.006 & 0.007 & 0.007 \up{17\%}
\\

\bottomrule\bottomrule
\end{tabular}

\vspace{0.4em}

\begin{minipage}{0.96\linewidth}
\footnotesize
Notes: Robust (Val PnL) reports the best-performing robust configuration selected by validation P\&L. Robust (Val Sharpe) reports the best-performing robust configuration selected by validation Sharpe. Percentage changes are computed relative to the greedy benchmark. Upward arrows indicate wider quotes, while downward arrows indicate tighter quotes. For all six scenarios, the selected $(\bar{\epsilon},\delta)$ is $(0.0001,0.1)$ under validation P\&L and $(4,1)$ under validation Sharpe.
\end{minipage}

\end{table}

Table~\ref{tab:rlmm_spread_summary} provides direct evidence that robustness changes the quoting behavior of the market-making agent rather than merely affecting terminal performance. The results reveal three economically important effects of robustness.

First, \textbf{robustness generally leads to wider quotes when uncertainty increases.} In the stable, price-stress, and fill-stress scenarios, the robust policy selected by validation Sharpe consistently posts wider bid and ask spreads than the greedy benchmark. Under the stable scenario, for example, the mean bid spread widens by $8.4\%$ (from $0.083$ to $0.090$), while the mean ask spread widens by $10.5\%$ (from $0.076$ to $0.084$). Similar widening is observed under price stress, and the strongest effect occurs under fill stress, where the mean bid and ask spreads widen by $10.0\%$ and $10.4\%$, respectively. Economically, these wider quotes represent higher compensation for bearing adverse-selection and execution risks under increased market uncertainty, leading the market maker to provide liquidity more conservatively.

Second, \textbf{robustness is state dependent rather than uniformly conservative.} Under both buy- and sell-arrival imbalance, the robust policy tightens rather than widens quotes. Under buy-arrival imbalance, the mean bid and ask spreads narrow by $7.1\%$ (from $0.070$ to $0.065$) and $4.0\%$ (from $0.125$ to $0.120$), respectively. Under sell-arrival imbalance, the corresponding reductions are $1.6\%$ (from $0.128$ to $0.126$) and $5.9\%$ (from $0.068$ to $0.064$). Rather than simply withdrawing liquidity, the robust policy responds more aggressively when persistent order-flow imbalance provides informative market signals. This demonstrates that robustness enhances sequential adaptation to structural market conditions instead of merely producing uniformly defensive behavior.

Third, \textbf{Robustness also changes the stability of quoting behavior.} In the stable and price-stress regimes, the standard deviations of bid and ask spreads increase from $0.009$ to $0.012$, indicating that the robust policy explores a broader range of quote placements when uncertainty is primarily driven by prices. By contrast, robustness substantially reduces quote variability under order-flow imbalance. Under buy-arrival imbalance, the standard deviations of bid and ask spreads decrease by $25.0\%$ and $52.2\%$, respectively, while the corresponding reductions under sell-arrival imbalance are $33.3\%$ and $31.3\%$. These reductions indicate smoother and more stable quoting behavior when the market exhibits persistent directional order-flow signals.

Taken together, these results reveal the mechanism through which robustness improves performance. Rather than directly increasing terminal P\&L or Sharpe ratios, robustness first changes quote placement, which subsequently affects execution intensity, inventory dynamics, and ultimately trading performance. Robustness therefore acts as a state-dependent decision principle rather than a uniformly conservative adjustment, adapting liquidity provision to the prevailing source of market uncertainty. This mechanism supports the central message of the paper: robustness reshapes sequential decision behavior under evolving uncertainty.

We use the price-stress scenario to illustrate how the robustness parameters reshape the complete quoting action, providing a clean setting for studying uncertainty-aware behavior under elevated volatility. Figure~\ref{fig:price_intraday} compares the greedy benchmark with four robust policies,
\[
(\bar{\varepsilon},\delta)
\in
\{(0.0001,0.1),(0.0001,1),(4,0.1),(4,1)\},
\]
representing combinations of low and high uncertainty tolerance ($\bar{\varepsilon}$) and low and high action robustness ($\delta$). The greedy benchmark is shown as the grey dotted line. Robust policies with low action robustness ($\delta=0.1$) are shown as dashed lines, whereas those with high action robustness ($\delta=1$) are shown as solid lines. The four panels report the bid spread, ask spread, bid quantity, and ask quantity, thereby illustrating the complete quoting action. Action robustness has a substantially larger effect on quoting behavior than uncertainty tolerance. Holding $\bar{\varepsilon}$ fixed, increasing $\delta$ produces substantially larger adjustments in both quoted spreads and quantities than comparable changes in $\bar{\varepsilon}$. In particular, larger values of $\delta$ lead to wider spreads and more selective quantity provision, especially near the end of the trading horizon. These results indicate that while uncertainty tolerance determines how much uncertainty is admitted into the decision problem, action robustness primarily governs how conservatively the market maker responds to that uncertainty.

The Sharpe-selected policy $(\bar{\varepsilon},\delta)=(4,1)$ also illustrates that robustness is dynamic rather than uniformly conservative. The policy maintains relatively competitive spreads during much of the trading day, preserving execution opportunities, but becomes more defensive near the close, where terminal inventory risk is more important. Additional scenario-specific parameter-sensitivity plots are reported in Appendix~\ref{appendix:simulation_details}; they show qualitatively similar patterns, although the magnitude and direction of adjustment depend on the source of uncertainty.

\begin{figure}[htb]
    \centering
    \includegraphics[width=0.72\linewidth]
    {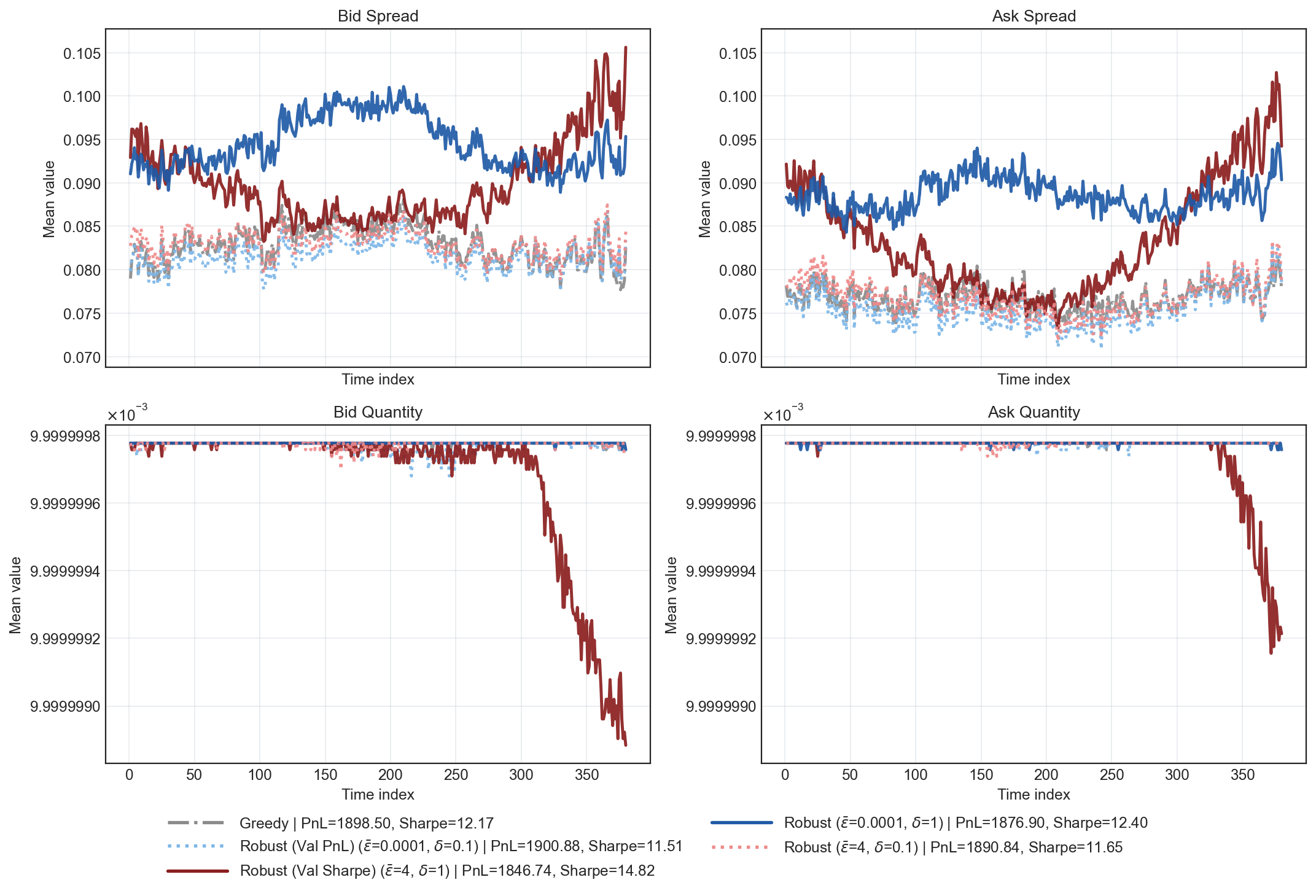}
    \caption{Intraday average quoting behavior under the Price Stress scenario for different robustness parameter combinations. The four panels report bid spread, ask spread, bid quantity, and ask quantity.}
    \label{fig:price_intraday}
\end{figure}

Figure~\ref{fig:cross_sharpe} broadens the analysis by comparing intraday behavior across all six simulation environments using the robust policy selected by validation Sharpe ratio. This figure provides information not contained in the summary statistics: it shows when during the trading day robust policies adjust their actions. The patterns are not scenario-agnostic. Under stable and price-stress conditions, robust policies preserve broadly similar intraday shapes but differ in the level of spread compensation. Under buy- and sell-arrival imbalances, the bid and ask sides respond asymmetrically, reflecting directional order-flow pressure. Under liquidity dry-out and fill stress, robust policies adjust compensation and exposure to execution uncertainty. Thus, robustness changes not only average quoting levels but also the timing and path of decisions.

\begin{figure}[htb]
    \centering
    \includegraphics[width=0.72\linewidth]
    {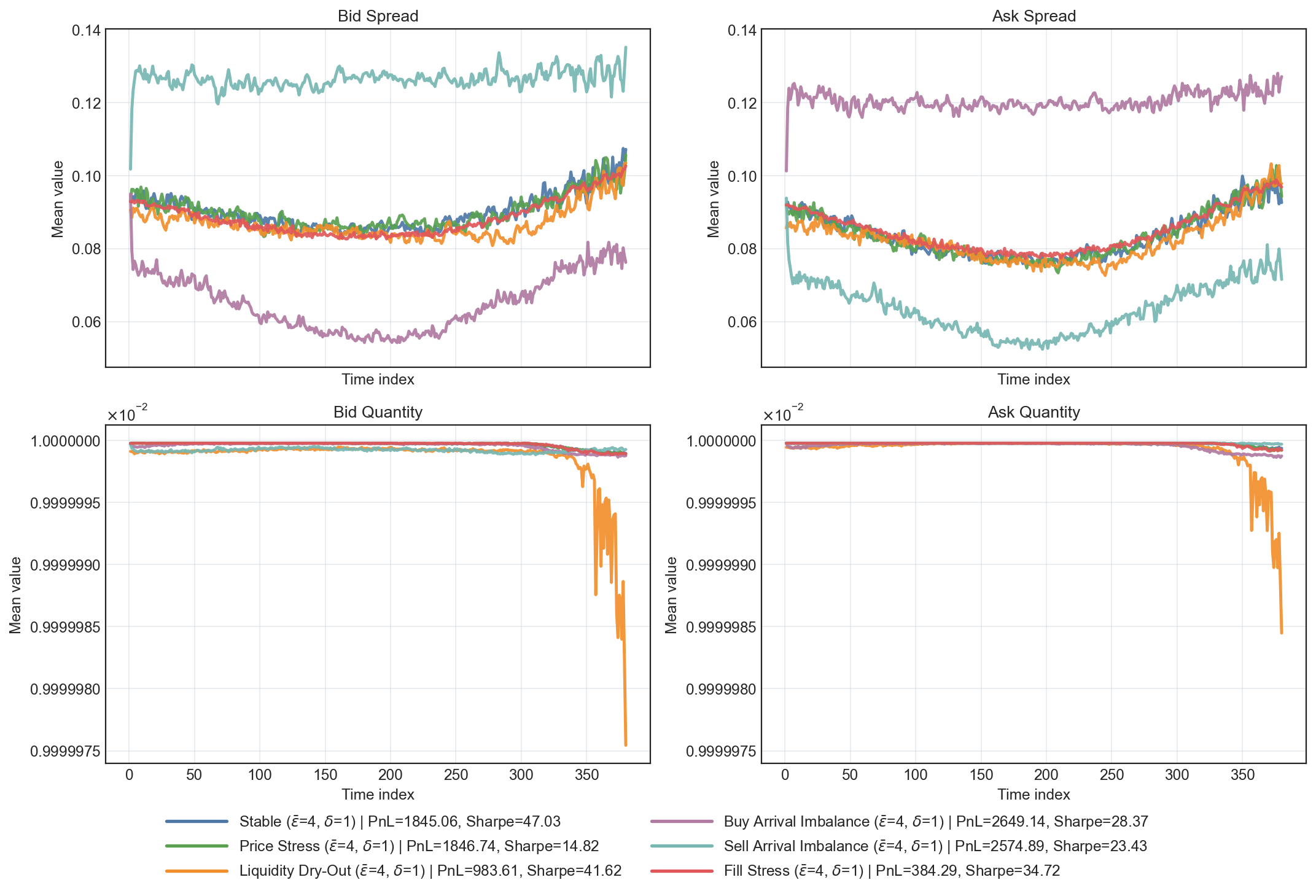}
    \caption{Cross-scenario intraday quoting behavior for the best robust policy selected using validation Sharpe ratio.}
    \label{fig:cross_sharpe}
\end{figure}

Figure~\ref{fig:cross_state} examines the state-conditional structure of the robust policy. Unlike Figure~\ref{fig:cross_sharpe}, which reports realized intraday actions, this figure shows how actions vary with market states. The robust policy responds systematically to economically meaningful variables, including order-flow measures, trading activity, spread conditions, time-to-close, and inventory. The responses are not flat: spreads vary with persistent state signals, while quantities remain stable in low-risk regions but adjust sharply in states associated with elevated inventory or execution risk. This indicates that robustness modifies the state-to-action mapping itself, rather than merely shifting average spreads.

\begin{figure}[htb]
    \centering
    \includegraphics[width=0.90\linewidth]
    {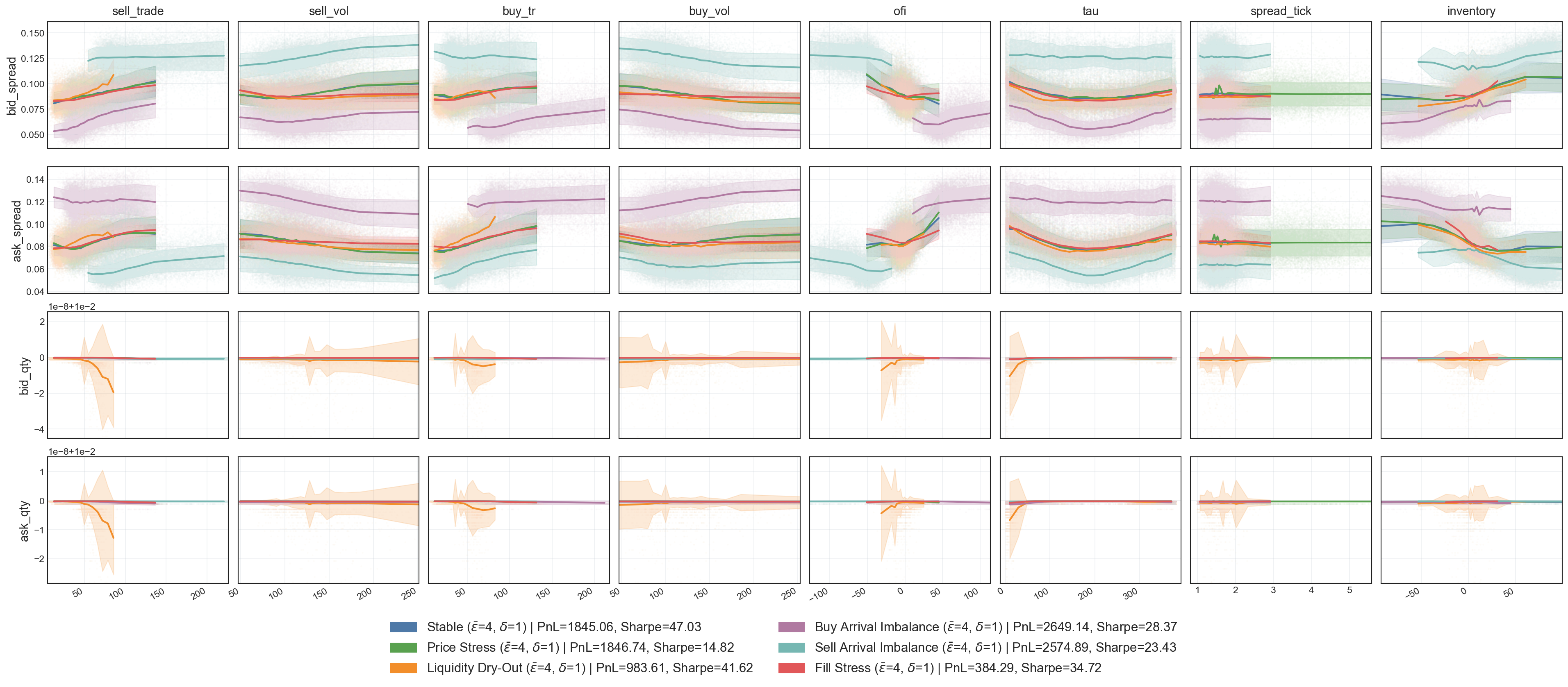}
    \caption{Cross-scenario state-conditional policy responses for the best robust policy selected using validation Sharpe ratio.}
    \label{fig:cross_state}
\end{figure}

These results show that robustness changes the policy function rather than only improving terminal performance. The figures reveal the mechanism: robustness changes the full action vector, the timing of actions, and the state-dependence of decisions. The table quantifies the resulting average spread behavior. Overall, the robust market maker becomes more selective under price and execution uncertainty, more responsive under persistent order-flow signals, and smoother in imbalance-driven regimes. This is precisely the sequential propagation channel emphasized in the paper: robustness first changes quote prices and quote quantities, which then affects execution intensity, inventory accumulation, and risk-adjusted performance.

\subsection{Empirical Evidence on When Robustness Creates Value}
\label{subsec:empirical}

Section \ref{sec:simulation} established that robustness alters the market maker's decision rule under uncertainty. We now examine whether these behavioral adjustments translate into economically meaningful improvements in real markets and, more importantly, under what market conditions robustness creates value.

The empirical evaluation uses the same execution model as the simulation. Executed quantities on real data are computed using the exponential fill probability described in Appendix~\ref{subsec:simulation_design}, so that fill rates depend only on the quoted spread. For both simulation and empirical evaluation, any inventory $I_T$ remaining at the terminal step is liquidated at a cost given by the square-root price-impact model described in Appendix~\ref{subsec:simulation_design} as well.

Our empirical analysis is organized around two key market characteristics: liquidity and volatility. We adopt a 2×2 factorial design comprising AAPL (\emph{high liquidity, low volatility}), TSLA (\emph{high liquidity, high volatility}), MKC (\emph{low liquidity, low volatility}), and TWLO (\emph{low liquidity, high volatility}). This design enables us to assess how robustness performs across distinct market environments and, in particular, whether its benefits depend on liquidity, volatility, or their interaction.

We also examine two distinct market regimes. The 2019 sample represents a relatively stable pre-shock environment, whereas the 2020 sample includes the onset of the COVID-19 crisis and therefore serves as a natural distribution-shift experiment. Consequently, the empirical analysis evaluates not only cross-sectional heterogeneity but also the ability of robust policies to maintain performance when future market conditions differ from those observed during training.

The central empirical question is straightforward: \emph{When does robustness create economic value in real-world market making?}

\subsubsection{Out-of-Sample Performance Across Market Environments}

Table~\ref{tab:empirical_summary} reports the out-of-sample performance of the benchmark strategies, the non-robust RL policy, and the robust RL policies selected using either validation P\&L or validation Sharpe ratio. Several findings emerge.

\begin{table}[H]
\centering
\caption{Economic Value of Robust RL Across Market Conditions (HL=high liquidity, LL=low liquidity, HV=high volatility, LV=low volatility; $\uparrow$/$\downarrow$ indicate improvement/deterioration relative to Greedy; absolute changes reported when Greedy benchmark is negative or near zero).}
\label{tab:empirical_summary}
\tiny
\setlength{\tabcolsep}{2pt}
\renewcommand{\arraystretch}{1}
\begin{tabular}{llcccccc}
\toprule
Env & Metric & AS & Rnd & Fix & Greedy & Robust (val. P\&L) & Robust (val. Sharpe) \\
\midrule
\multicolumn{8}{c}{\textbf{Panel A. Stable (2019)}}\\
\midrule
HL-LV & Sharpe & 0.50 & -0.14 & -0.14 & 0.58 & \textbf{1.52 ($\uparrow$162\%)} & 1.40 ($\uparrow$141\%) \\
& P\&L & 2.06 & -11,498 & -12,187 & 196.97 & \textbf{221.28 ($\uparrow$12\%)} & 175.13 ($\downarrow$11\%) \\
& MDD & -0.34 & -28.35 & -26.91 & -4.70 & \textbf{-2.25 ($\downarrow$52\%)} & -3.03 ($\downarrow$36\%) \\
\midrule
HL-HV & Sharpe & 2.01 & -0.14 & -0.17 & 1.56 & 1.39 ($\downarrow$11\%) & \textbf{4.74 ($\uparrow$204\%)} \\
& P\&L & 2.30 & -7,832 & -7,215 & \textbf{3,298} & 3,239 ($\downarrow$2\%) & 693 ($\downarrow$79\%) \\
& MDD & -0.51 & -35.60 & -37.04 & -1.18 & -1.50 ($\uparrow$27\%) & \textbf{-0.97 ($\downarrow$18\%)} \\
\midrule
LL-LV & Sharpe & \textbf{0.13} & -1.02 & -1.07 & -0.66 & -0.64 ($\uparrow$3\%) & -0.64 ($\uparrow$3\%) \\
& P\&L & \textbf{0.05} & -17,288 & -16,478 & -93.78 & -28.86 ($\uparrow$69\%) & -28.86 ($\uparrow$69\%) \\
& MDD & \textbf{-5.15} & -29.51 & -26.61 & -53.59 & -20.07 ($\downarrow$63\%) & -20.07 ($\downarrow$63\%) \\
\midrule
LL-HV & Sharpe & -0.28 & -0.55 & -0.51 & -0.03 & 0.32 (+0.35) & \textbf{0.50 (+0.53)} \\
& P\&L & -0.96 & -472,212 & -462,278 & -9.78 & \textbf{73.24 (+83.02)} & 54.22 (+64.00) \\
& MDD & -10.12 & -25.61 & -25.92 & -3.17 & -4.14 ($\uparrow$31\%) & \textbf{-2.59 ($\downarrow$18\%)} \\
\midrule
\multicolumn{8}{c}{\textbf{Panel B. Stress (2020)}}\\
\midrule
HL-LV & Sharpe & 0.33 & 0.04 & 0.06 & 1.87 & 4.55 ($\uparrow$143\%) & \textbf{5.52 ($\uparrow$195\%)} \\
& P\&L & 0.46 & 734 & 1,277 & 419.17 & \textbf{468.30 ($\uparrow$12\%)} & 206.98 ($\downarrow$51\%) \\
& MDD & \textbf{-0.30} & -1.11 & -0.93 & -1.43 & -0.63 ($\downarrow$56\%) & -0.52 ($\downarrow$64\%) \\
\midrule
HL-HV & Sharpe & 0.43 & -0.12 & -0.06 & 1.57 & 2.41 ($\uparrow$54\%) & \textbf{3.95 ($\uparrow$152\%)} \\
& P\&L & 12.90 & -1,605 & -860 & 1,394 & \textbf{1,547 ($\uparrow$11\%)} & 1,067 ($\downarrow$23\%) \\
& MDD & \textbf{-0.27} & -3.15 & -1.72 & -1.17 & -0.95 ($\downarrow$19\%) & -0.62 ($\downarrow$47\%) \\
\midrule
LL-LV & Sharpe & \textbf{-0.72} & -0.71 & -0.74 & -0.73 & -0.74 ($\downarrow$1\%) & -0.79 ($\downarrow$8\%) \\
& P\&L & \textbf{-1.12} & -296,553 & -278,583 & -669.12 & -609.78 ($\uparrow$9\%) & -720.59 ($\downarrow$8\%) \\
& MDD & \textbf{-16.06} & -31.12 & -27.97 & -107.62 & -123.61 ($\uparrow$15\%) & -115.55 ($\uparrow$7\%) \\
\midrule
LL-HV & Sharpe & -0.48 & -0.82 & -0.79 & -0.08 & \textbf{0.43 (+0.51)} & -0.27 (-0.19) \\
& P\&L & -1.70 & -439,931 & -403,622 & -73.61 & \textbf{91.91 (+165.52)} & -152.43 (-78.82) \\
& MDD & -131.52 & -55.87 & -27.43 & -5.00 & \textbf{-3.11 ($\downarrow$38\%)} & -7.33 ($\uparrow$47\%) \\
\bottomrule
\end{tabular}
\vspace{0.2em}
\begin{minipage}{0.96\linewidth}
\scriptsize
Notes: Environments: HL-LV (AAPL), HL-HV (TSLA), LL-LV (MKC), LL-HV (TWLO). Significance in Appendix.
\end{minipage}
\end{table}

First, \textbf{reinforcement-learning-based market makers substantially outperform the classical Avellaneda--Stoikov benchmark as well as the random and fixed quoting rules across most stock-year combinations.} For example, in AAPL during 2020, the Sharpe ratio increases from 0.33 under the best-performing AS specification to 1.87 under the learned non-robust policy. Similar improvements are observed in TSLA, where the Sharpe ratio increases from 0.43 under AS to 1.57 under the learned policy. These results confirm that adaptive state-dependent quoting generates substantial economic value beyond classical inventory-based market-making rules.

Second, \textbf{the effectiveness of robustness depends systematically on market liquidity, with volatility playing a secondary role.} The strongest gains occur in the two highly liquid assets, AAPL (HL--LV) and TSLA (HL--HV), whereas the improvements are considerably smaller for the two low-liquidity assets, MKC (LL--LV) and TWLO (LL--HV).

For the highly liquid assets, robustness consistently delivers large improvements in risk-adjusted performance. In AAPL (HL--LV), the stable 2019 market already exhibits substantial gains, with the Sharpe ratio increasing from 0.58 under the greedy policy to 1.52 under validation-P\&L selection, corresponding to a 162\% improvement, while maximum drawdown declines from $-4.70$ to $-2.25$. During the stressed 2020 market, the gains become even larger. Validation-P\&L robustness increases the Sharpe ratio from 1.87 to 4.55, while validation-Sharpe robustness further increases it to 5.52. At the same time, maximum drawdown decreases from $-1.43$ to $-0.63$ and $-0.52$, respectively. TSLA (HL--HV) exhibits a similar pattern. In 2020, validation-P\&L robustness increases the Sharpe ratio from 1.57 to 2.41 while simultaneously increasing the mean P\&L from 1,394 to 1,547. Validation-Sharpe robustness further increases the Sharpe ratio to 3.95 and reduces maximum drawdown from $-1.17$ to $-0.62$. Although some average profitability is sacrificed under validation-Sharpe selection, the improvement in risk-adjusted performance is substantial. Together, these results demonstrate that when execution opportunities are abundant, robustness can effectively translate improved sequential decisions into superior trading performance.

The gains are considerably weaker for the two low-liquidity assets. In MKC (LL--LV), robustness substantially reduces losses during the stable 2019 market, with the average P\&L improving from $-93.8$ to $-28.9$ and maximum drawdown declining from $-53.6$ to $-20.1$. Nevertheless, overall performance remains weak, and during the stressed 2020 market all learned policies continue to exhibit Sharpe ratios close to $-0.75$, regardless of the robustness specification. TWLO (LL--HV) occupies an intermediate position. Validation-P\&L robustness converts a negative Sharpe ratio of $-0.03$ into 0.32 during 2019, while validation-Sharpe robustness further improves it to 0.50. Similar improvements occur during 2020, where validation-P\&L robustness increases the Sharpe ratio from $-0.08$ to 0.43 while simultaneously converting negative average profitability into positive profitability. Although robustness remains beneficial under high volatility, the overall improvements remain smaller than those observed for the highly liquid assets. These results suggest that robustness cannot fully compensate for limited execution opportunities.

Overall, liquidity emerges as the dominant determinant of the value of robustness. High-liquidity assets provide sufficient execution opportunities for the agent to actively manage inventory risk, allowing improved sequential decisions to translate into substantially better trading performance. By contrast, when trading opportunities are inherently limited, as in MKC and, to a lesser extent, TWLO, the incremental value of robustness is considerably smaller.

Third, \textbf{robustness becomes increasingly valuable during periods of market stress.} Relative improvements are generally larger during the stressed 2020 market than during the stable 2019 environment. For example, the best robust Sharpe ratio for AAPL increases from 1.52 in 2019 to 5.52 in 2020. Similarly, robust policies continue to generate substantial improvements for TSLA across both periods, increasing the Sharpe ratio from 1.56 under the greedy policy to between 1.39--4.74 in 2019 and between 2.41--3.95 in 2020, depending on the validation criterion. More broadly, the robust policies maintain strong risk-adjusted performance despite the substantial distributional shift associated with the COVID-19 market disruption. These findings are consistent with the central rationale of robust reinforcement learning: policies that explicitly account for model uncertainty are better able to preserve decision quality when future market dynamics deviate from the historical environment used for training.

Figure~\ref{fig:pareto_2020} further illustrates these findings by reporting the out-of-sample risk--return frontiers during the 2020 stress regime. Robustness substantially expands the attainable frontier for the highly liquid assets AAPL (HL--LV) and TSLA (HL--HV), whereas the gains are more modest for MKC (LL--LV) and TWLO (LL--HV). The validation-selected policies lie close to the efficient frontier, suggesting that validation performance provides a practical criterion for selecting robustness parameters. Further Pareto-frontier diagnostics, validation-selection results, and additional empirical analyses are provided in Appendix~\ref{appendix:empirical_results}.

\begin{figure}[H]
    \centering
    \makebox[\textwidth][c]{\begin{minipage}{1.14\textwidth}\centering
        \begin{subfigure}{0.39\linewidth}
            \includegraphics[width=\linewidth]{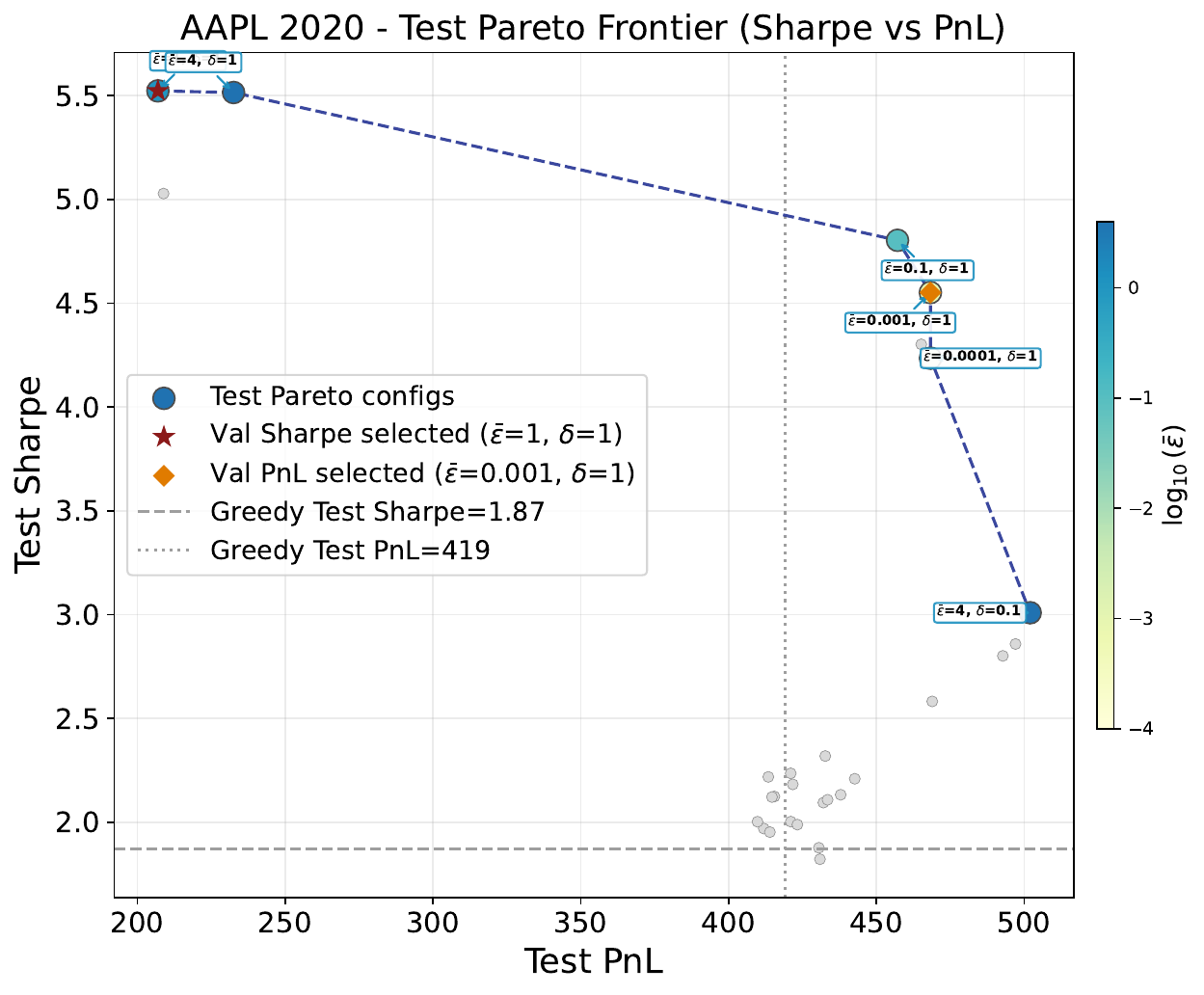}
            \caption{AAPL.}
        \end{subfigure}
        \begin{subfigure}{0.39\linewidth}
            \includegraphics[width=\linewidth]{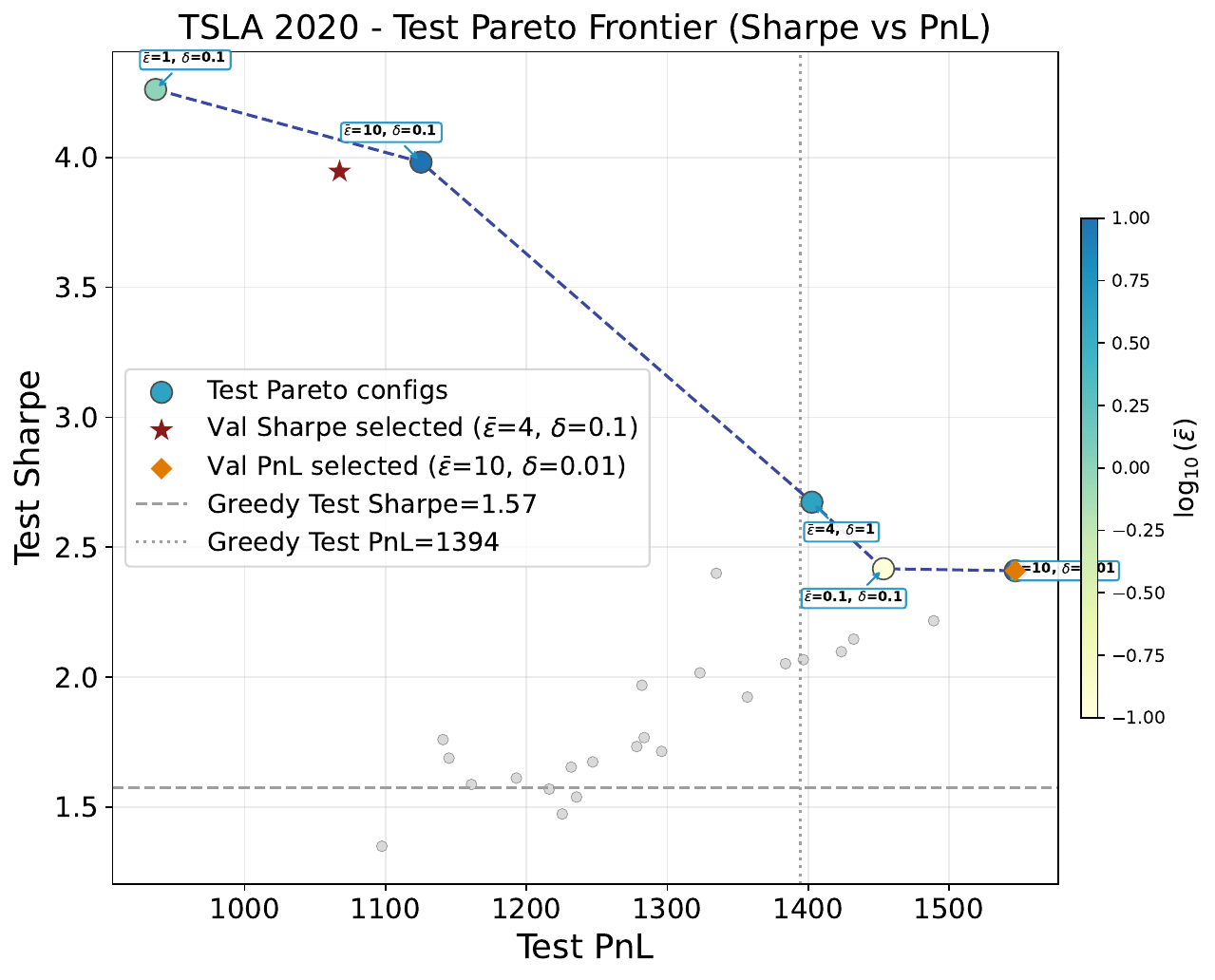}
            \caption{TSLA.}
        \end{subfigure}

        \vspace{0.6em}

        \begin{subfigure}{0.39\linewidth}
            \includegraphics[width=\linewidth]{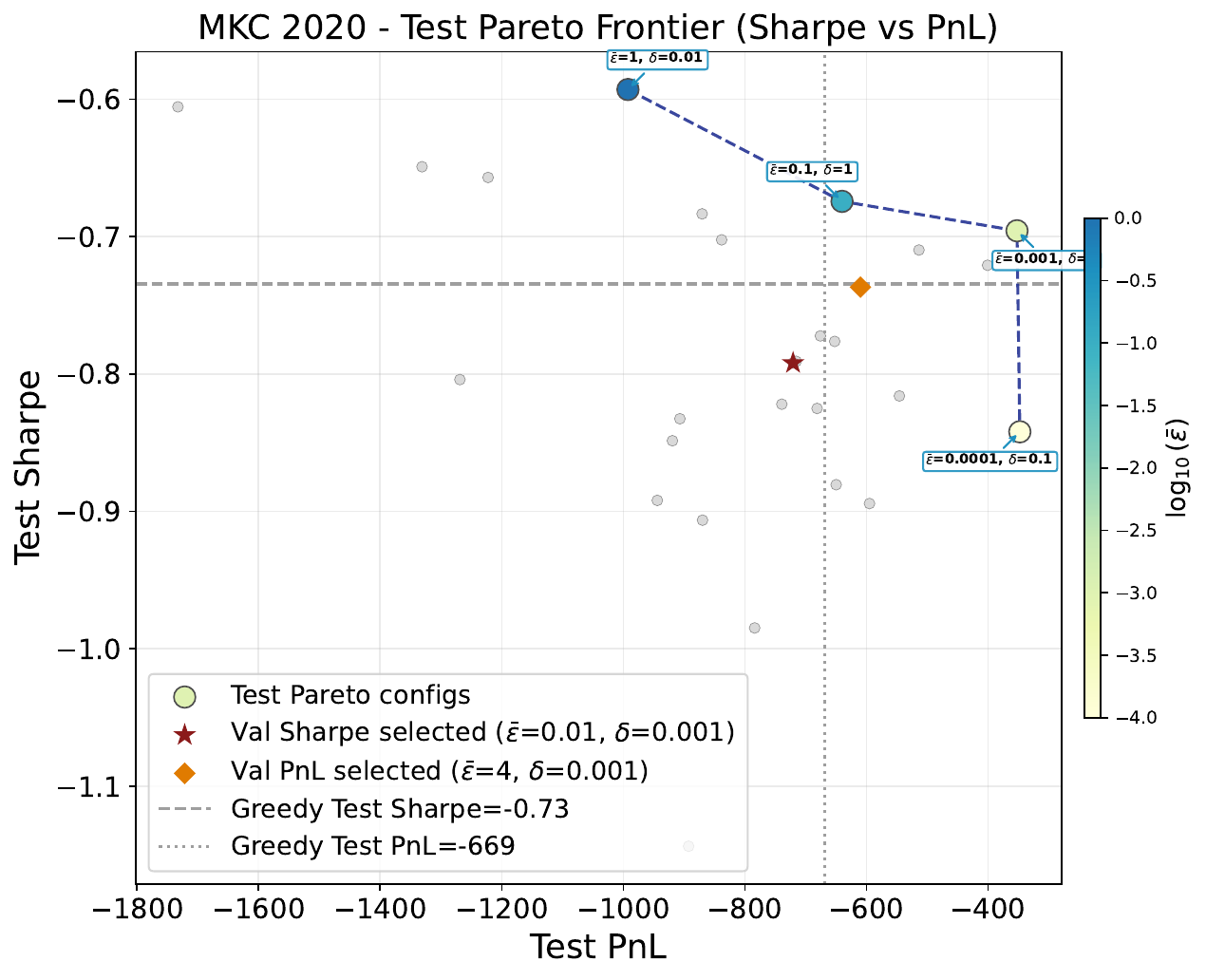}
            \caption{MKC.}
        \end{subfigure}
        \begin{subfigure}{0.39\linewidth}
            \includegraphics[width=\linewidth]{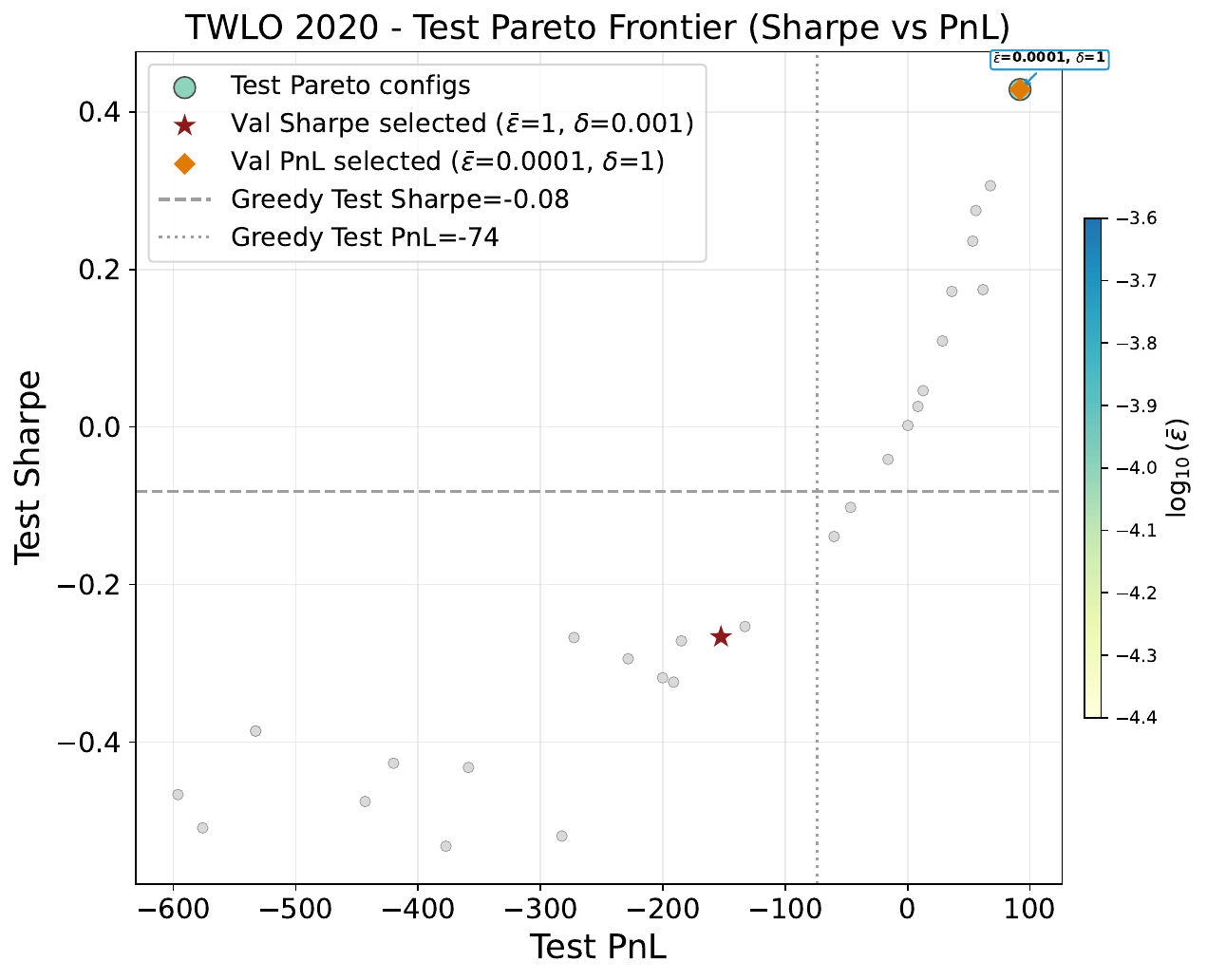}
            \caption{TWLO.}
        \end{subfigure}
    \end{minipage}}

    \caption{
 Out-of-sample risk--return frontiers during 2020 stress regime under robustness. Robustness substantially expands the attainable frontier for the highly liquid assets AAPL (HL--LV) and TSLA (HL--HV), while the gains are more modest for MKC (LL--LV) and TWLO (LL--HV). Highlighted points denote the validation-selected policies.
    }
    \label{fig:pareto_2020}
\end{figure}

\section[Conclusion]{Discussion and Conclusion}\label{sec:Conclusion}

This paper develops a robust reinforcement learning framework for high-frequency market making by incorporating distributional uncertainty directly into the transition dynamics. Rather than treating robustness solely as protection against model misspecification, the proposed framework views robustness as a mechanism that reshapes sequential decision making under evolving uncertainty. Both simulation and empirical evidence show that robust policies adapt their quoting behavior, inventory management, and execution decisions in economically meaningful ways.

An important implication is that robustness should not be interpreted as a uniformly conservative adjustment. Instead, robustness acts as a state-dependent decision principle. In liquid markets, it enables the agent to respond more effectively to distributional shifts while preserving trading opportunities. Under persistent order-flow signals, robustness may even lead to more aggressive liquidity provision rather than wider quotes. These findings suggest that the economic value of robustness depends jointly on market uncertainty and the availability of execution opportunities.

Beyond market making, the proposed framework provides a general approach for robust sequential decision making in environments with complex transition dynamics. The decomposition into deterministic state evolution and stochastic innovations allows ambiguity to be introduced only where uncertainty genuinely arises, preserving both computational tractability and economic interpretability. The two-dimensional robustness formulation further separates uncertainty tolerance from action robustness, providing a flexible framework for analyzing different forms of decision conservativeness.

The present study has several limitations. The present framework adopts the standard inventory-based market-making model and therefore does not explicitly model adverse selection and competing market makers. Second, the empirical analysis focuses on single-asset market making and does not consider cross-asset inventory constraints or multi-venue execution. Third, although the forecasting model captures complex nonlinear transition dynamics, its performance remains limited by the quality and representativeness of historical data.

These limitations suggest several promising directions for future research. Extending the framework to multi-agent market making, endogenous order-flow dynamics, and multi-venue execution would allow robustness to capture strategic interactions among market participants. Another important direction is to integrate online learning so that ambiguity sets evolve continuously as new market information arrives. Finally, the proposed framework may provide a useful foundation for robust sequential decision making in broader financial applications, including optimal execution, portfolio management, and market design.

Overall, the results demonstrate that robustness is more than a tool for guarding against model misspecification. By explicitly accounting for evolving uncertainty in sequential decision making, it provides a principled and economically interpretable framework for improving the resilience of learning-based market-making systems.

\section*{Acknowledgements}

J.~Sester gratefully acknowledges support from the NUS Start-Up Grant \emph{Tackling Model Uncertainty in Finance with Machine Learning} and the MOE Academic Research Fund (AcRF) Tier~1 Grant 25-0428-P0001. J.~Sester and Y.~Chen gratefully acknowledge support from the MOE Academic Research Fund (AcRF) Tier~2 Grant T2EP20225-0030. Y.~Chen also acknowledges support from the SIA--NUS Digital Aviation Corporate Laboratory Phase~2 under the Industry Alignment Fund--Industry Collaboration Projects (IAF--ICP).
%

\bibliographystyle{plainnat}
\IfFileExists{references.bib}{%
    \bibliography{references}%
}{%
    \bibliography{literature}%
}

\appendix

\section{Simulation Details}\label{appendix:simulation_details}

\subsection[Experimental Design]{Experimental Design}\label{subsec:simulation_design}

We evaluate the proposed market-making methods in a synthetic limit-order-book
environment that is empirically anchored to intraday equity data while remaining
transparent enough to isolate specific forms of model uncertainty. The simulation
design has two objectives. First, it provides a stable baseline environment that is
consistent with the state representation and transition structure used by the
reinforcement-learning agent. Second, it generates a small set of interpretable stress
environments that perturb one component of the market mechanism at a time, so that
performance differences can be attributed to specific sources of misspecification.

A trading day is discretized into \(T=390\) one-minute intervals, corresponding to
the regular U.S. equity session. The training and validation samples are always
generated from the stable baseline regime. Test samples are then generated under the
stable regime and under five stress scenarios: price stress, liquidity dry-out,
buy-side imbalance, sell-side imbalance, and execution stress. We assess the effectiveness of the proposed approach by examining terminal-day Profit and Loss (PnL). Market making agents are trained within a simulated market environment using the robust reinforcement learning procedure outlined in Algorithm~\ref{robust-rl}. In each experiment, we train both a robust and a non-robust agent to isolate the contribution of distributional robustness. The following subsections detail the data-generating mechanism of the artificial market used in our experiments.

\subsubsection{Returns}
In our simulation study, we assume that the asset \emph{log-returns}
\[
r_{t+1} := \log\!\left(\frac{S_{t+1}}{S_t}\right)
\]
follow an autoregressive AR(1) model of the form:
\begin{equation*}
    r_t = \mu + \phi r_{t-1} + \varepsilon_t, \quad \varepsilon_t \sim t_\nu(0,\sigma^2)
\end{equation*}\footnote{We use $\bar r$ to ensure that returns do not become excessively large.}
where $\mu$ is the long-run mean return, $\phi$ is the autoregressive coefficient, and $\varepsilon_t$ follows a Student's $t$-distribution with $\nu$ degrees of freedom and scale $\sigma^2$.

\subsubsection{Buy and Sell trades}
Buy and sell market order (MO) arrivals, denoted by 
$N_t^{\operatorname{bid}}$ and $N_t^{\operatorname{ask}}$, are modeled in discrete time as Poisson processes with time-varying intensities. The intensity process is given by $\lambda_t = \mu_t + Z_t$\footnote{We set a lower bound $\underline{\lambda}$ so that the intensity never becomes zero.}, where $Z_t$ is a self-exciting component. The deterministic component $\mu_t$ follows a U-shaped curve with a peak parameter $\mu_{\mathrm{peak}}$ for market open and market close, and a mid-day trough parameter $\mu_{\mathrm{mid}}$. Under a unit time discretization ($\mathrm{\Delta_t} = 1$), the excitation evolves according to
\[
Z_{t+1} = e^{-\beta} Z_t + \alpha N_t.
\]
Conditional on the current intensity, arrivals follow a Poisson distribution
\[
N_t \sim \mathrm{Poisson}(\lambda_t).
\]

This discrete-time specification can be viewed as a time-discretization of a continuous-time linear Markovian Hawkes process \cite{GAO20183807}. In continuous time, the intensity process is given by
\[
\lambda_t = \mu_t + Z_t,
\]
where $\mu_t$ denotes the baseline intensity and $Z_t$ is a self-exciting component evolving according to
\[
\mathrm{d}Z_t = -\beta Z_t\,\mathrm{d}t + \alpha\,\mathrm{d}N_t.
\]
The resulting system admits a Markov representation with state vector $(Z_t, N_t)$ \cite{GAO20183807}.

We assume that buy and sell MOs are independent and follow the same time-varying intensity dynamics. All MOs arriving within the time interval $[t, t+1)$ are aggregated, and may interact with the MM’s limit orders posted at time $t$. Accordingly, we model arrivals at the level of aggregated counts rather than individual events. Finally, the volumes of market orders on both sides, $V_t^{\operatorname{bid}}$ and $V_t^{\operatorname{ask}}$, are assumed to follow log-normal distributions, consistent with empirical evidence in the literature (\cite{MASLOV2001234}).

\subsubsection{Market Spread and Microstructure Features}

To simulate realistic order book conditions, we model the \emph{market bid--ask spread} as a mean-reverting process and derive four microstructure state features directly from the simulated best bid and ask prices.

The market spread $\delta^m_t := S_t^{\operatorname{best, ask}} - S_t^{\operatorname{best,bid}}$ is simulated via an AR(1) process on the log-spread:
\[
\log \delta^m_{t+1} = \phi_s \log \delta^m_t + (1 - \phi_s)\,\mu_s + \sigma_s\,\varepsilon_t, \quad \varepsilon_t \overset{\text{i.i.d.}}{\sim} \mathcal{N}(0,1),
\]
where $\mu_s$ is the long-run log-mean, $\phi_s \in (0,1)$ controls mean-reversion speed, and $\sigma_s > 0$ governs spread volatility. The log-normal specification ensures strict positivity and captures the empirically observed right-skewness of spreads. Given the mid-price $S_t$, the best bid and ask prices are set symmetrically as
\[
S_t^{\operatorname{best, bid}} = S_t - \tfrac{\delta^m_t}{2}, \qquad S_t^{\operatorname{best, ask}} = S_t + \tfrac{\delta^m_t}{2}.
\]

The four microstructure state variables $(\delta^{t,\operatorname{tick}},\, \delta^{t,\operatorname{rel}},\, D^{t,\operatorname{micro}},\, \mathit{\rm{OFI}}^t)$ are then computed directly from the simulated best prices and order counts:
\begin{align*}
\delta^{t,\operatorname{tick}} &= \frac{\delta^m_t}{size_{tick}}, \\[4pt]
\delta^{t,\operatorname{rel}}  &= \frac{\delta^m_t}{S_t}, \\[4pt]
D^{t,\operatorname{micro}}     &= \frac{N_t^{\operatorname{ask}}\,S_t^{\operatorname{best, bid}} + N_t^{\operatorname{bid}}\,S_t^{\operatorname{best, ask}}}{N_t^{\operatorname{ask}} + N_t^{\operatorname{bid}}} - S_t, \\[4pt]
\mathit{\rm{OFI}}^t                 &= N_t^{\operatorname{ask}} - N_t^{\operatorname{bid}},
\end{align*}
where $size_{tick}$ denotes the tick size and $S_t$ denotes the mid-price. Here $\delta^{t,\operatorname{tick}}$ measures the spread in tick units, $\delta^{t,\operatorname{rel}}$ expresses it as a fraction of the mid-price, $D^{t,\operatorname{micro}}$ is the deviation of the volume-weighted microprice from the mid-price, and $\mathit{OFI}^t$ captures the net directional order flow pressure.

\subsubsection{Executed Quantities} \label{linear-dd}

The \emph{limit order book} (LOB) records the demand for equity at different price levels from the bid/ask side, and its shape determines how deep a typical market order penetrates the book, i.e., how many LOs will be filled and lifted.

Following the classical market-making literature \citep{Avellaneda2008}, 
the execution of limit orders is modeled via exponentially decaying arrival intensities. 
For a quoted distance $\delta$ from the mid-price, the bid- and ask-side market order arrival rates are given by
\[
\lambda_{\text{limit}}^{\mathrm{bid}}(\delta) = A e^{-\kappa \delta^{\mathrm{bid}}}, 
\qquad
\lambda_{\text{limit}}^{\mathrm{ask}}(\delta) = A e^{-\kappa \delta^{\mathrm{ask}}}.
\]
Over a discrete time interval $\mathrm{d}t$, the probability of execution follows from the associated Poisson processes,
\[
p^{\mathrm{bid}}(\delta;\mathrm{d}t) 
= 1 - \exp\!\big(-\lambda_{\text{limit}}^{\mathrm{bid}}(\delta)\,\mathrm{d}t\big), 
\qquad
p^{\mathrm{ask}}(\delta;\mathrm{d}t) 
= 1 - \exp\!\big(-\lambda_{\text{limit}}^{\mathrm{ask}}(\delta)\,\mathrm{d}t\big).
\]
The expected executed quantities are given by
\[
\mathbb{E}[Q^{\mathrm{bid}} \mid \delta] 
= p^{\mathrm{bid}}(\delta;\mathrm{d}t)\, M^{\mathrm{ask}}, 
\qquad
\mathbb{E}[Q^{\mathrm{ask}} \mid \delta] 
= p^{\mathrm{ask}}(\delta;\mathrm{d}t)\, M^{\mathrm{bid}}.
\]
The functional form of the execution probability is assumed to be unknown to market participants.
This same function is applied for both simulation and empirical analysis for simplicity.
\subsubsection{Terminal Liquidation Cost}\label{app:liquidation}

Any inventory $I_T$ remaining at the end of the trading session must be unwound in the open market. Following the square-root price-impact model \citep{AlmgrenThum2005,Toth2011}, the liquidation cost for a meta-order of size $Q = |I_T|$ shares against an average daily volume $V_T$ is
\begin{equation}\label{eq:liquidation_cost}
c_{\mathrm{liq}}(I_T, V_T)
\;=\;
\eta\,\sigma\,|I_T|\,\sqrt{\frac{|I_T|}{V_T}},
\end{equation}
where $\sigma$ is the per-step price volatility, $\eta > 0$ is a price-impact scaling factor, and
\[
V_T \;=\; N_T^{\operatorname{bid}}\,V_T^{\operatorname{bid}} \;+\; N_T^{\operatorname{ask}}\,V_T^{\operatorname{ask}}
\]
is the prevailing average daily volume at the terminal period. The cost \eqref{eq:liquidation_cost} is subtracted from the terminal reward for both simulation and empirical evaluation.

\subsection[Data and Empirical Calibration]{Data and Empirical Calibration}\label{subsec:data-calibration}
The simulation parameters are calibrated from the intraday LOB data. The calibration procedure retains only full regular-session trading days, corresponding to 390 one-minute observations per day, and discards half-days and incomplete sessions. In total, 494 full AAPL days are used for calibration.

Each stock-day is calibrated separately. Returns are fitted with the AR(1) Student-\(t\)
model. Buy- and sell-side baseline intensities are fitted from the one-minute market
order counts. The Hawkes parameters \((\alpha_Z,\beta_Z)\), the lognormal size
parameters, the spread parameters, the deep-book proxies, and the fill parameters
\((A,k)\) are then estimated from the same day. This produces a cross section of
daily parameter vectors rather than a single pooled fit.

The final simulation package is constructed from the empirical distribution of daily
fits for AAPL. The train, validation, and test periods consist of 160, 48, and 48 days, respectively. Let \(\theta_d\) denote the daily
parameter vector for day \(d\) of the anchor stock. The stable baseline uses the
coordinate-wise median across days. The remaining scenarios alter only a targeted
subset of parameters:
\begin{enumerate}
\item \emph{Price stress.} Replace
\(\{\mu_r,\phi_r,\sigma_r,\nu_r,\bar r,\mu_s,\phi_s,\sigma_s\}\)
by their 90th-percentile anchor-stock values and keep all other parameters at the
baseline.
\item \emph{Liquidity dry-out.} Replace
\(\{\mu_{\mathrm{mid}}^b,\mu_{\mathrm{peak}}^b,\mu_{\mathrm{mid}}^a,
\mu_{\mathrm{peak}}^a,\underline{\lambda}\}\) by their 10th-percentile values and
keep all other parameters at the baseline.
\item \emph{Buy-side imbalance.} Replace
\(\{\mu_{\mathrm{mid}}^b,\mu_{\mathrm{peak}}^b\}\) by their 90th-percentile values and
leave the sell side at the baseline.
\item \emph{Sell-side imbalance.} Replace
\(\{\mu_{\mathrm{mid}}^a,\mu_{\mathrm{peak}}^a\}\) by their 90th-percentile values and
leave the buy side at the baseline.
\item \emph{Execution stress.} Keep the stable baseline state dynamics, but replace the
training and validation fill parameters by their stable median values and the test-time
fill parameters by a low-\(A\), high-\(k\) configuration, namely the 10th percentile
of \(A\) and the 90th percentile of \(k\).
\end{enumerate}

This scenario design is intentionally modular. Price stress changes only return and spread behavior; liquidity dry-out changes only baseline order-flow activity; imbalance stresses change only one side of the market-order intensity surface; and execution stress changes only the fill mechanism. As a result, the test environments separate distinct forms of model ambiguity instead of conflating them in a single shock.

\begin{table}[H]
\centering
\caption{Calibrated simulation parameters under the stable baseline and stress
scenarios}
\label{tab:simulation_scenarios}
\scriptsize
\setlength{\tabcolsep}{4pt}
\begin{tabular}{lrrrrrr}
\hline \hline
Parameter & Stable & Price & Dry-out & Buy imb. & Sell imb. & Exec. \\
\hline
\multicolumn{7}{l}{\emph{Return dynamics}} \\
Return mean \((\times 10^{5})\) & 0.39 & 4.49 & 0.39 & 0.39 & 0.39 & 0.39 \\
Return AR coefficient & 0.01 & 0.10 & 0.01 & 0.01 & 0.01 & 0.01 \\
Return scale \((\times 10^{4})\) & 5.98 & 11.87 & 5.98 & 5.98 & 5.98 & 5.98 \\
Degrees of freedom & 5.97 & 9.79 & 5.97 & 5.97 & 5.97 & 5.97 \\
Return clip \((\times 10^{3})\) & 3.05 & 5.95 & 3.05 & 3.05 & 3.05 & 3.05 \\
\multicolumn{7}{l}{\emph{Arrival intensities}} \\
Buy mid intensity & 25.84 & 25.84 & 12.73 & 66.83 & 25.84 & 25.84 \\
Buy peak intensity & 106.39 & 106.39 & 59.78 & 208.66 & 106.39 & 106.39 \\
Sell mid intensity & 29.24 & 29.24 & 14.13 & 29.24 & 67.35 & 29.24 \\
Sell peak intensity & 111.07 & 111.07 & 60.66 & 111.07 & 218.39 & 111.07 \\
Minimum intensity & 6.16 & 6.16 & 2.96 & 6.16 & 6.16 & 6.16 \\
Hawkes excitation \(\alpha_Z\) & 0.11 & 0.11 & 0.11 & 0.11 & 0.11 & 0.11 \\
Hawkes decay \(\beta_Z\) & 11.76 & 11.76 & 11.76 & 11.76 & 11.76 & 11.76 \\
\multicolumn{7}{l}{\emph{Sizes and spreads}} \\
Lognormal size mean & 121.44 & 121.44 & 121.44 & 121.44 & 121.44 & 121.44 \\
Lognormal size sigma & 0.32 & 0.32 & 0.32 & 0.32 & 0.32 & 0.32 \\
Spread mean & 0.02 & 0.04 & 0.02 & 0.02 & 0.02 & 0.02 \\
Spread AR coefficient & 0.55 & 0.69 & 0.55 & 0.55 & 0.55 & 0.55 \\
Spread sigma & 0.16 & 0.19 & 0.16 & 0.16 & 0.16 & 0.16 \\
EWMA volatility \(\lambda_v\) & 0.05 & 0.05 & 0.05 & 0.05 & 0.05 & 0.05 \\
\multicolumn{7}{l}{\emph{Execution}} \\
Fill intensity \(A_{\mathrm{test}}\) & 5.23 & 5.23 & 5.23 & 5.23 & 5.23 & 2.33 \\
Fill elasticity \(k_{\mathrm{test}}\) & 14.05 & 14.05 & 14.05 & 14.05 & 14.05 & 22.88 \\
\hline \hline
\end{tabular}

\end{table}

\subsection[Implementation Details and Hyperparameter Choices]{Implementation Details and Hyperparameters}\label{subsec:implementation-hyperparameters}

The value network consists of an initial LayerNorm layer, two linear--SiLU blocks, and two residual MLP blocks, each residual block combining LayerNorm, two linear layers, a SiLU activation, and a skip connection. The policy network uses the same stem, but with separate residual branches for spreads and quantities; each branch ends with a sigmoid output layer to enforce the admissible action bounds. The lambda network is simpler, using a one-hidden-layer MLP with a softplus output to ensure positivity of the dual variable. We use SiLU because its smooth nonlinearity tends to yield more stable optimization than piecewise-linear alternatives in this high-dimensional, noisy setting ~\cite{Ramachandran2018}. All networks are trained with AdamW ~\cite{LoshchilovHutter2019}. Table~\ref{tab:key_implementation_defaults} summarizes the main hyperparameter choices.

\begin{table}[H]
\centering
\caption{Key hyperparameters in the current implementation}
\label{tab:key_implementation_defaults}
\footnotesize
\setlength{\tabcolsep}{6pt}
\begin{tabular}{lr}
\hline \hline
Hyperparameter & Value \\
\hline
Trading intervals per day \(T\) & 390.00 \\
State lookback \(m\) & 10.00 \\
Discount parameter \(\alpha\) & 0.90 \\
Risk-aversion parameter \(\gamma\) & 0.10 \\
Liquidation price impact factor \(\eta\) & 1 \\
Transaction cost \((\times 10^{4})\) & 1.00 \\
Maximum spread & 0.30 \\
Maximum quantity share & 0.01 \\
Training epochs & 100.00 \\
Minimum epochs & 25.00 \\
Patience & 5.00 \\
Actor learning rate \((\times 10^{4})\) & 3.00 \\
Critic learning rate \((\times 10^{4})\) & 3.00 \\
Dual learning rate \((\times 10^{4})\) & 3.00 \\
Critic batch size & 4.00 \\
Actor batch size & 4.00 \\
Sampler epochs & 10.00 \\
Sampler batch size & 16.00 \\
\hline \hline
\end{tabular}

\vspace{0.5em}
\parbox{0.92\textwidth}{\footnotesize Notes. This table reports the important training and robustness settings hyperparameters for the current implementation.}
\end{table}

\section{Additional Simulation Results}\label{appendix:add_sim_results}

\subsection{Hyperparameter Performance Heatmaps and Pareto frontiers}

Figures \ref{fig:app_sim_heatmap_1} and \ref{fig:app_sim_heatmap_2} present the test-period performance heatmaps of all robust hyperparameter combinations under all six simulation scenarios.  Figure \ref{fig:app_sim_pareto} then presents the PnL-to-Sharpe Pareto frontier plots in all six simulation scenarios. The labeled points are those on the frontier. In general, we observe a downward sloping trend, as a higher Sharpe value would accompany a lower PnL value. This suggests that the robust agent trades in some of the profit for higher stability. The gray dot lines pinpoint the positions of the greedy agent. In some scenarios, it offers a competitive baseline as it can outperform some of the robust hyperparameter combinations, while in the sell arrival imbalance scenario, it is almost dominated by the robust frontier.

\begin{figure}[htbp]
    \centering
    \makebox[\textwidth][c]{\begin{minipage}{1.18\textwidth}\centering
        \begin{subfigure}{0.32\linewidth}
            \includegraphics[width=\linewidth]{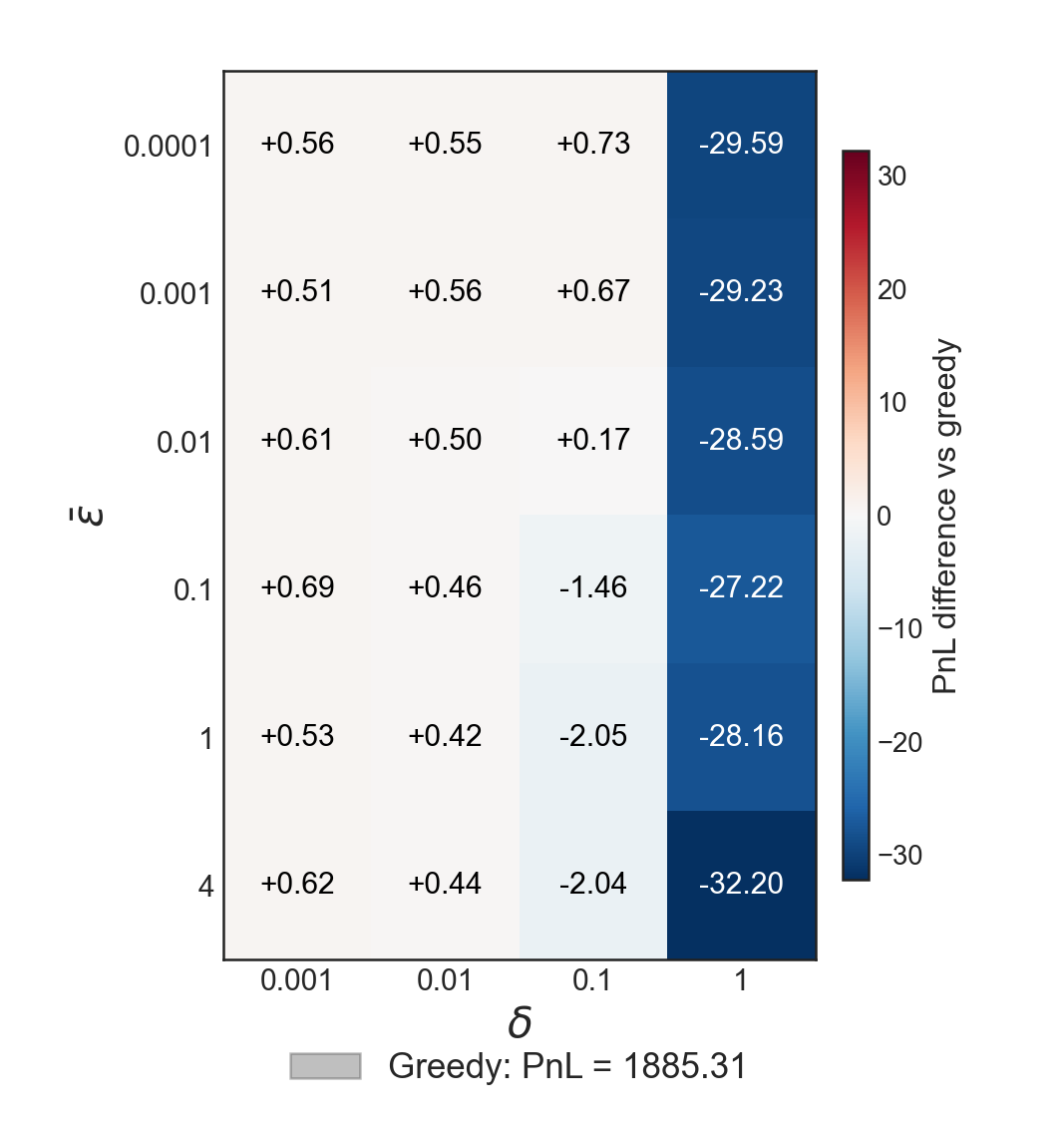}
            \caption{Stable, PnL Heatmap}
        \end{subfigure}
        \begin{subfigure}{0.32\linewidth}
            \includegraphics[width=\linewidth]{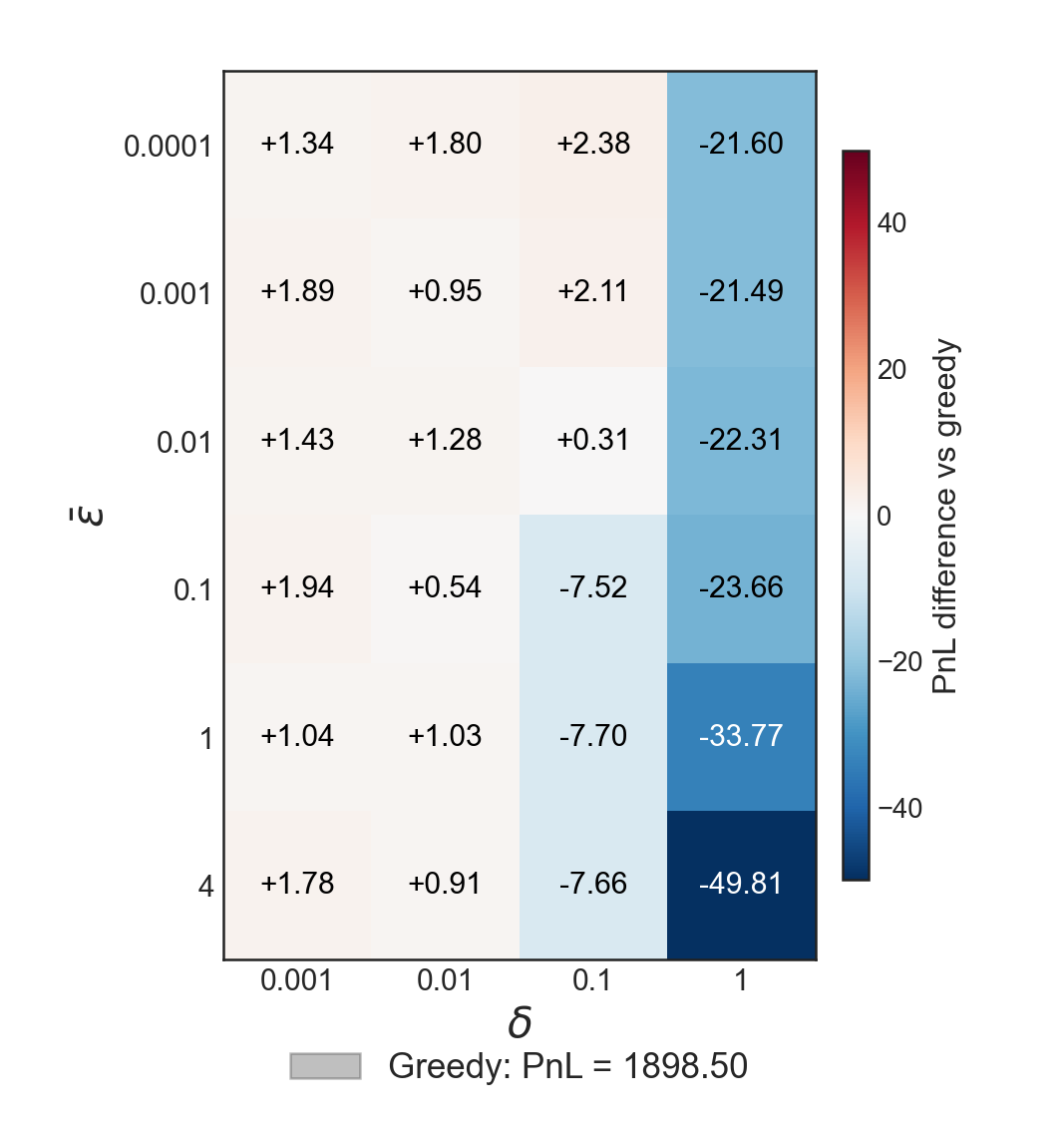}
            \caption{Price stress, PnL Heatmap}
        \end{subfigure}
        \begin{subfigure}{0.32\linewidth}
            \includegraphics[width=\linewidth]{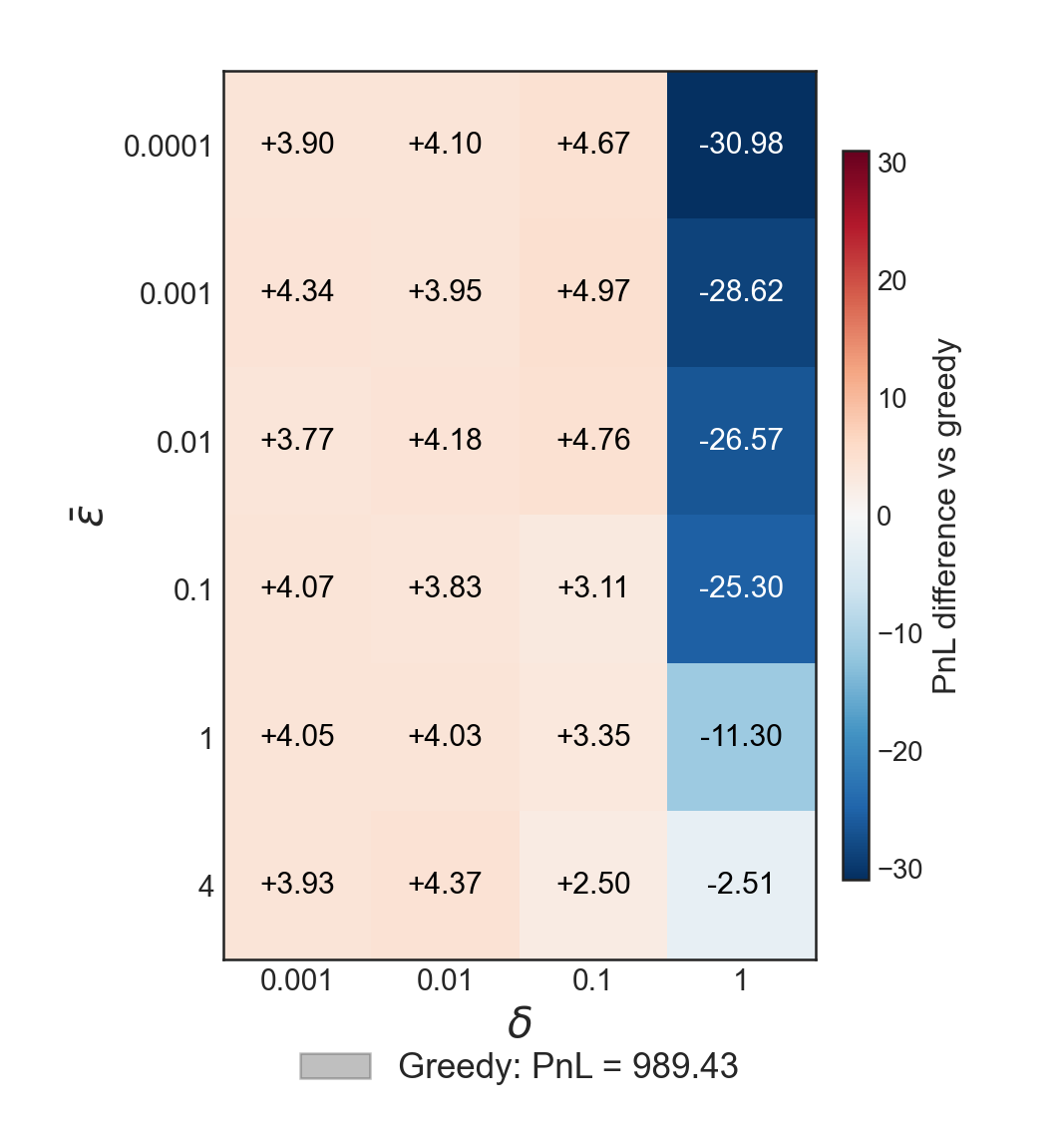}
            \caption{Liquidity dry-out, PnL Heatmap}
        \end{subfigure}

        \par\vspace{0.1em}

        \begin{subfigure}{0.32\linewidth}
            \includegraphics[width=\linewidth]{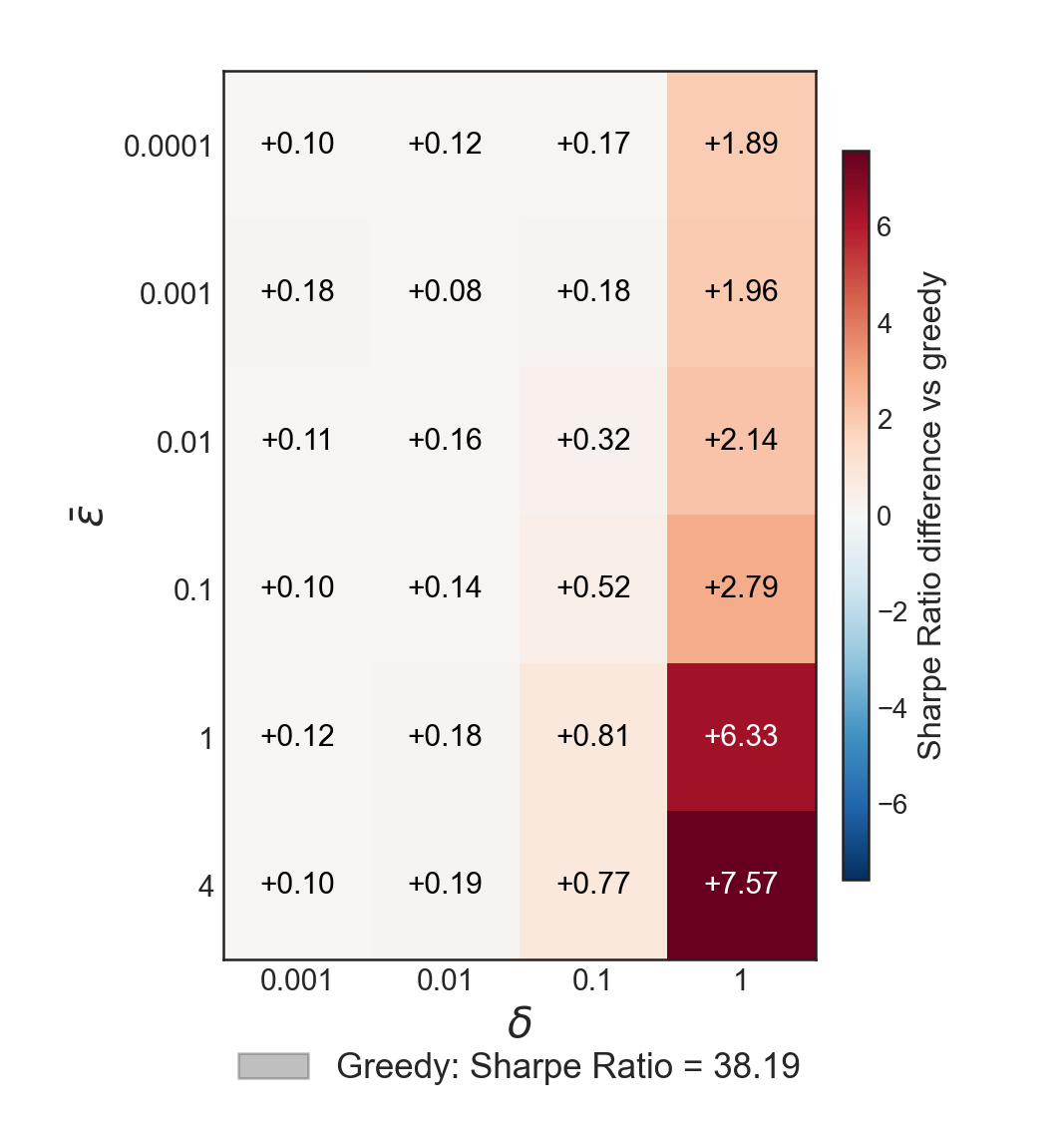}
            \caption{Stable, Sharpe Heatmap}
        \end{subfigure}
        \begin{subfigure}{0.32\linewidth}
            \includegraphics[width=\linewidth]{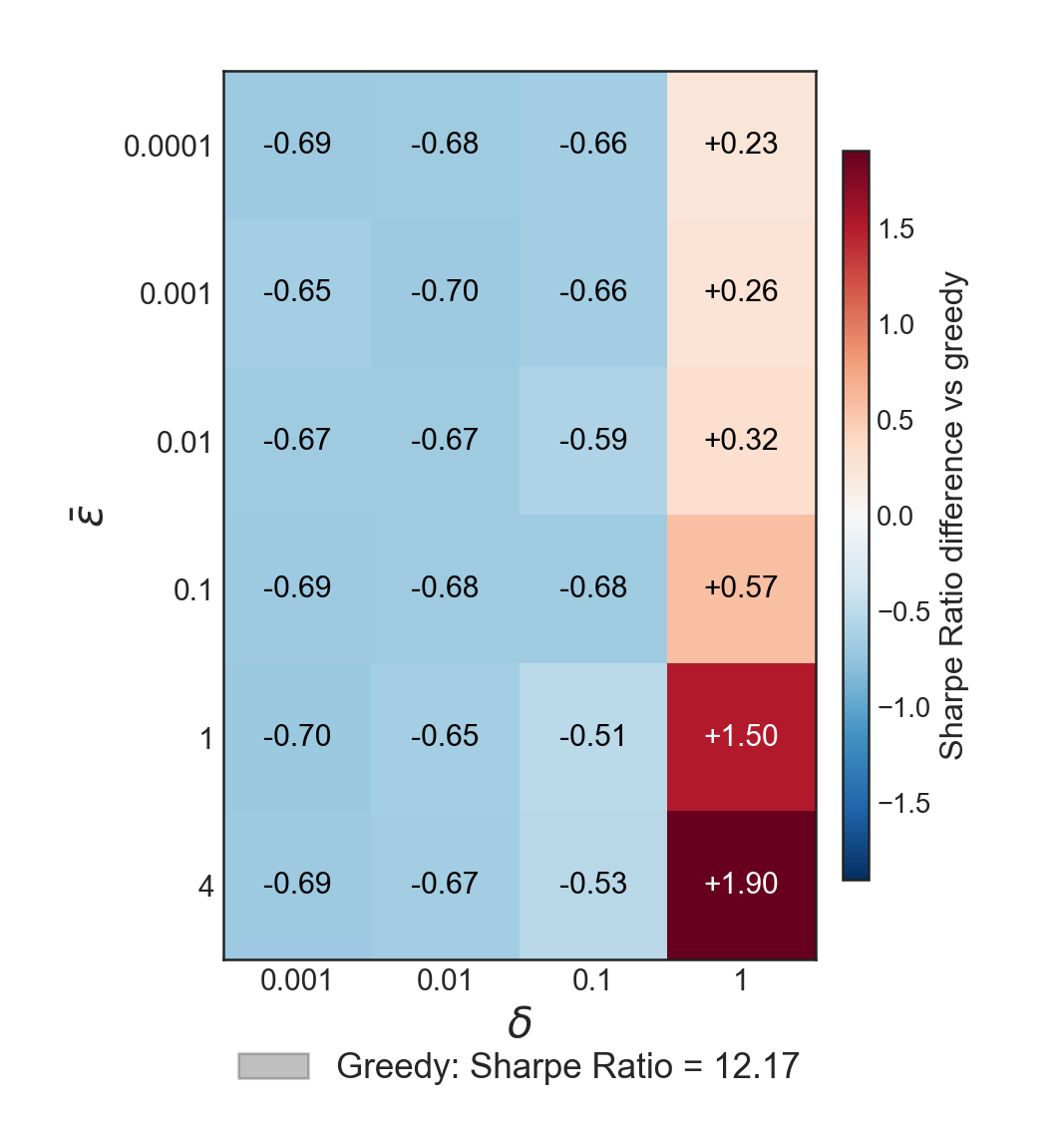}
            \caption{Price stress, Sharpe Heatmap}
        \end{subfigure}
        \begin{subfigure}{0.32\linewidth}
            \includegraphics[width=\linewidth]{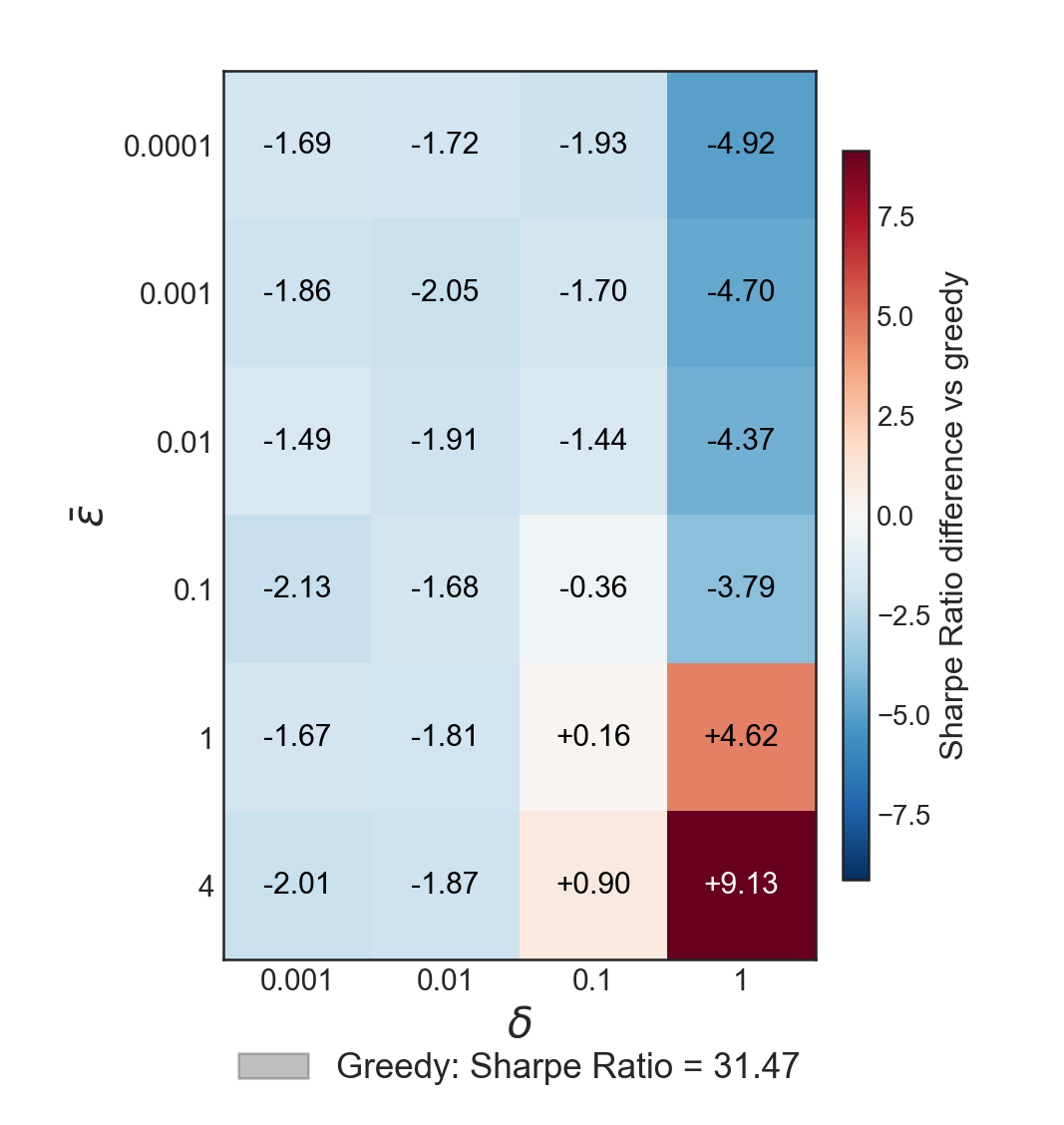}
            \caption{Liquidity dry-out, Sharpe Heatmap}
        \end{subfigure}
    \end{minipage}}
    \caption{Test-period mean PnL and Sharpe ratio over the $(\bar\varepsilon,\delta)$ grid for the stable, price-stress, and liquidity-dry-out scenarios. The first row reports mean PnL and the second row reports Sharpe ratio, with columns corresponding to the three scenarios; the greedy benchmark value is shown below each panel. The results show that out-of-sample performance depends materially on both the robustness parameters and the stress scenario.}
    \label{fig:app_sim_heatmap_1}
\end{figure}

\begin{figure}[htbp]
    \centering
    \makebox[\textwidth][c]{\begin{minipage}{1.18\textwidth}\centering
        \begin{subfigure}{0.32\linewidth}
            \includegraphics[width=\linewidth]{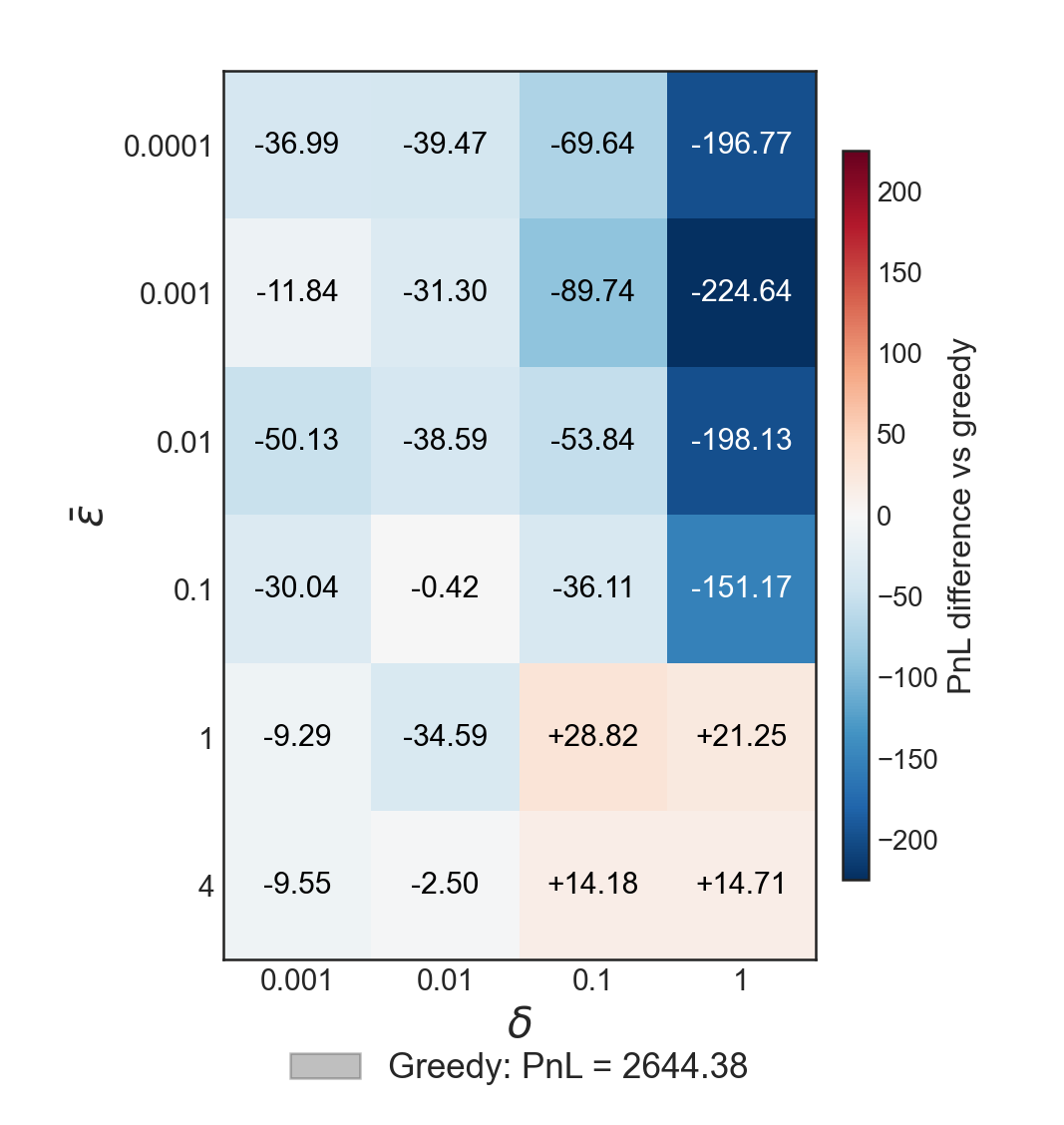}
            \caption{Buy arrival imbalance, PnL Heatmap}
        \end{subfigure}
        \begin{subfigure}{0.32\linewidth}
            \includegraphics[width=\linewidth]{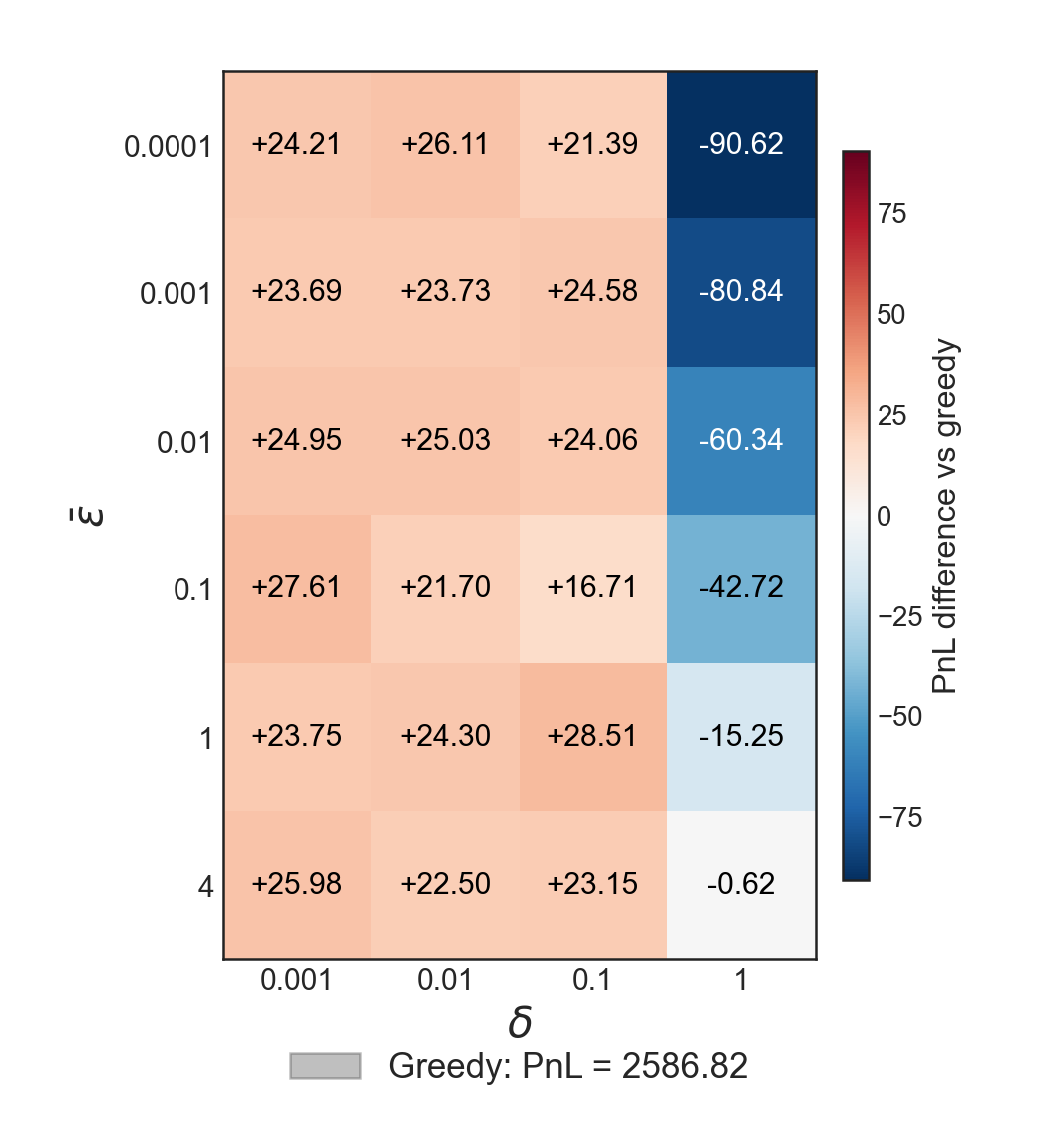}
            \caption{Sell arrival imbalance, PnL Heatmap}
        \end{subfigure}
        \begin{subfigure}{0.32\linewidth}
            \includegraphics[width=\linewidth]{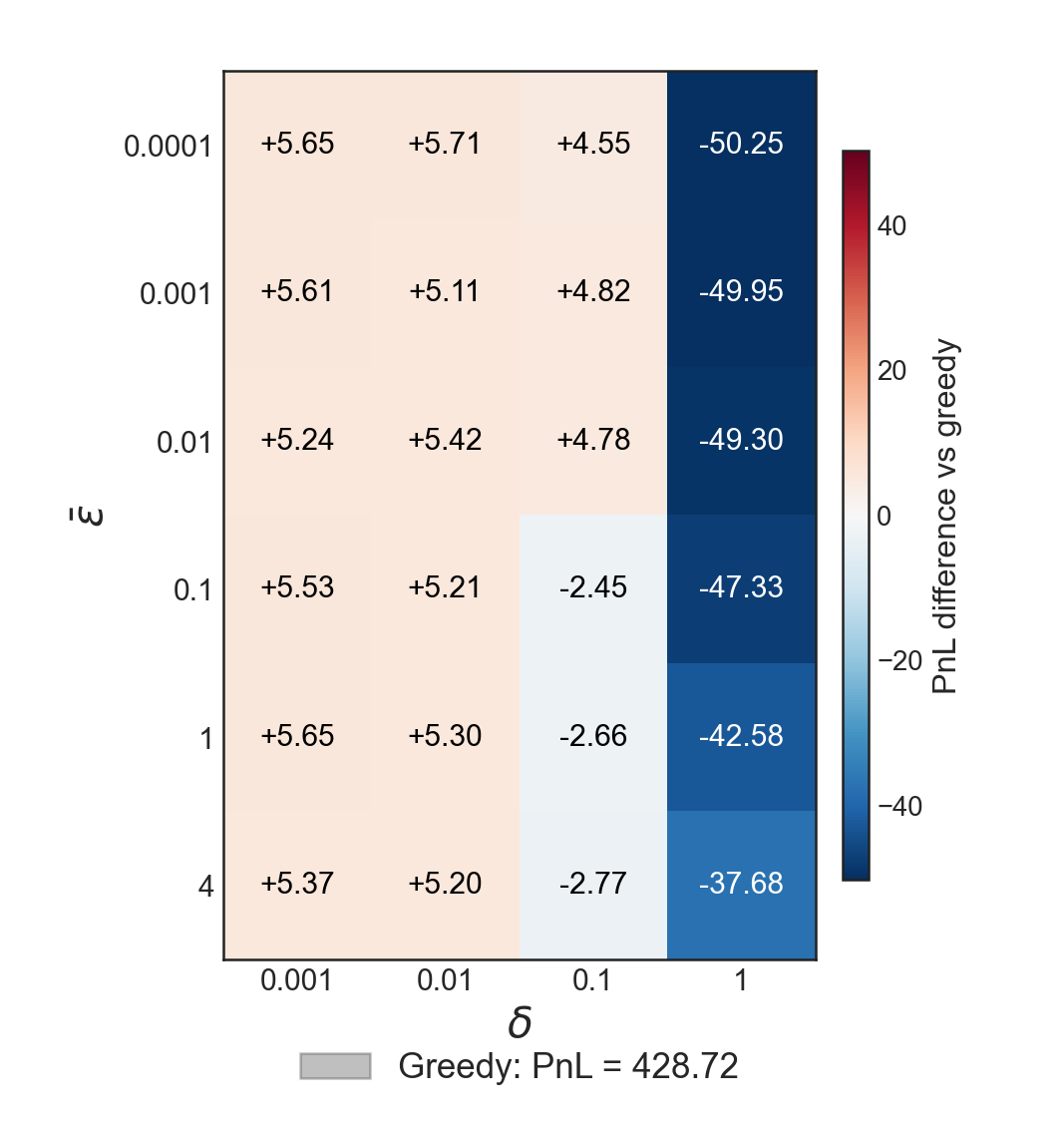}
            \caption{Fill stress, PnL Heatmap}
        \end{subfigure}

        \par\vspace{0.1em}

        \begin{subfigure}{0.32\linewidth}
            \includegraphics[width=\linewidth]{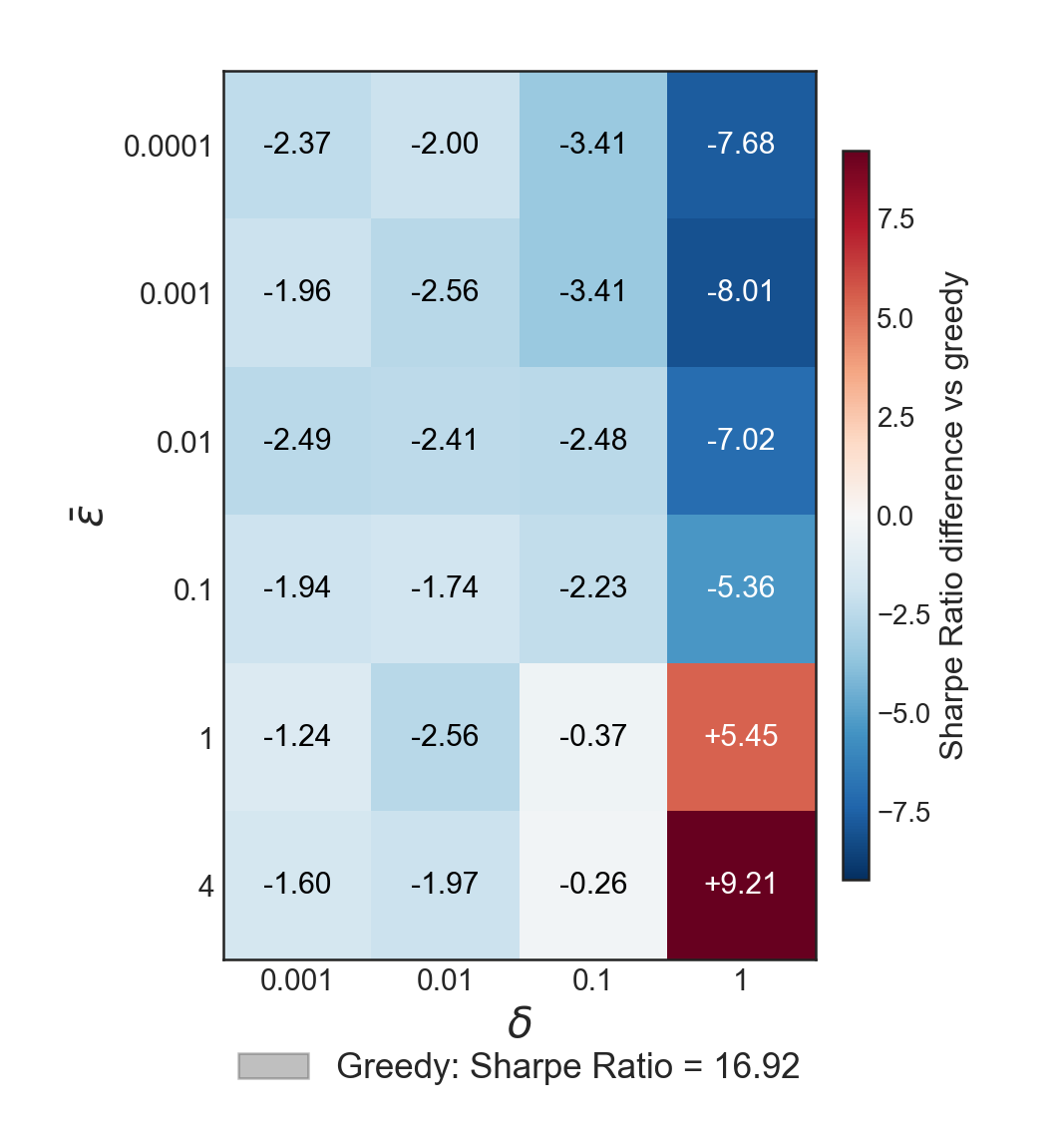}
            \caption{Buy arrival imbalance, Sharpe Heatmap}
        \end{subfigure}
        \begin{subfigure}{0.32\linewidth}
            \includegraphics[width=\linewidth]{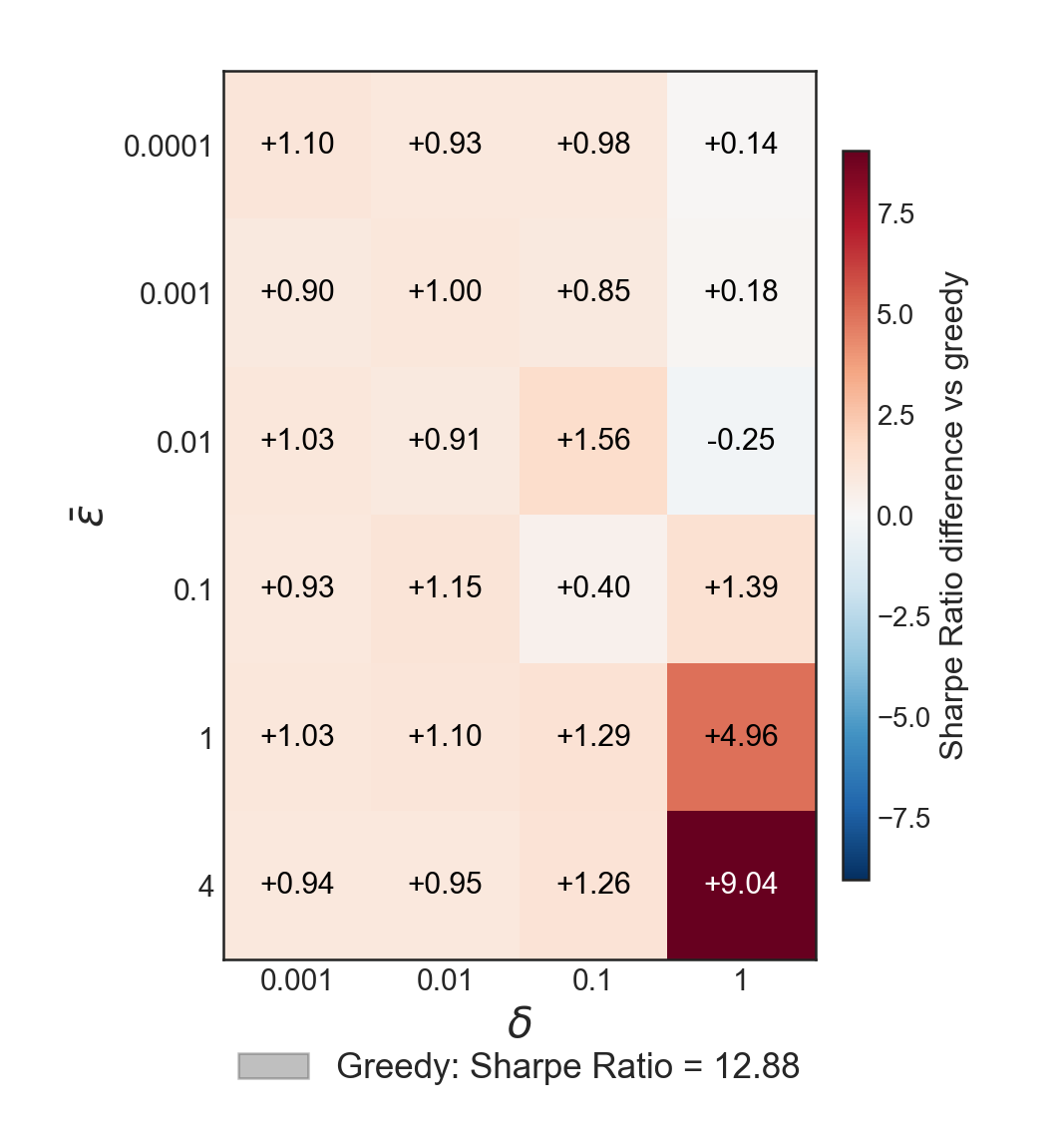}
            \caption{Sell arrival imbalance, Sharpe Heatmap}
        \end{subfigure}
        \begin{subfigure}{0.32\linewidth}
            \includegraphics[width=\linewidth]{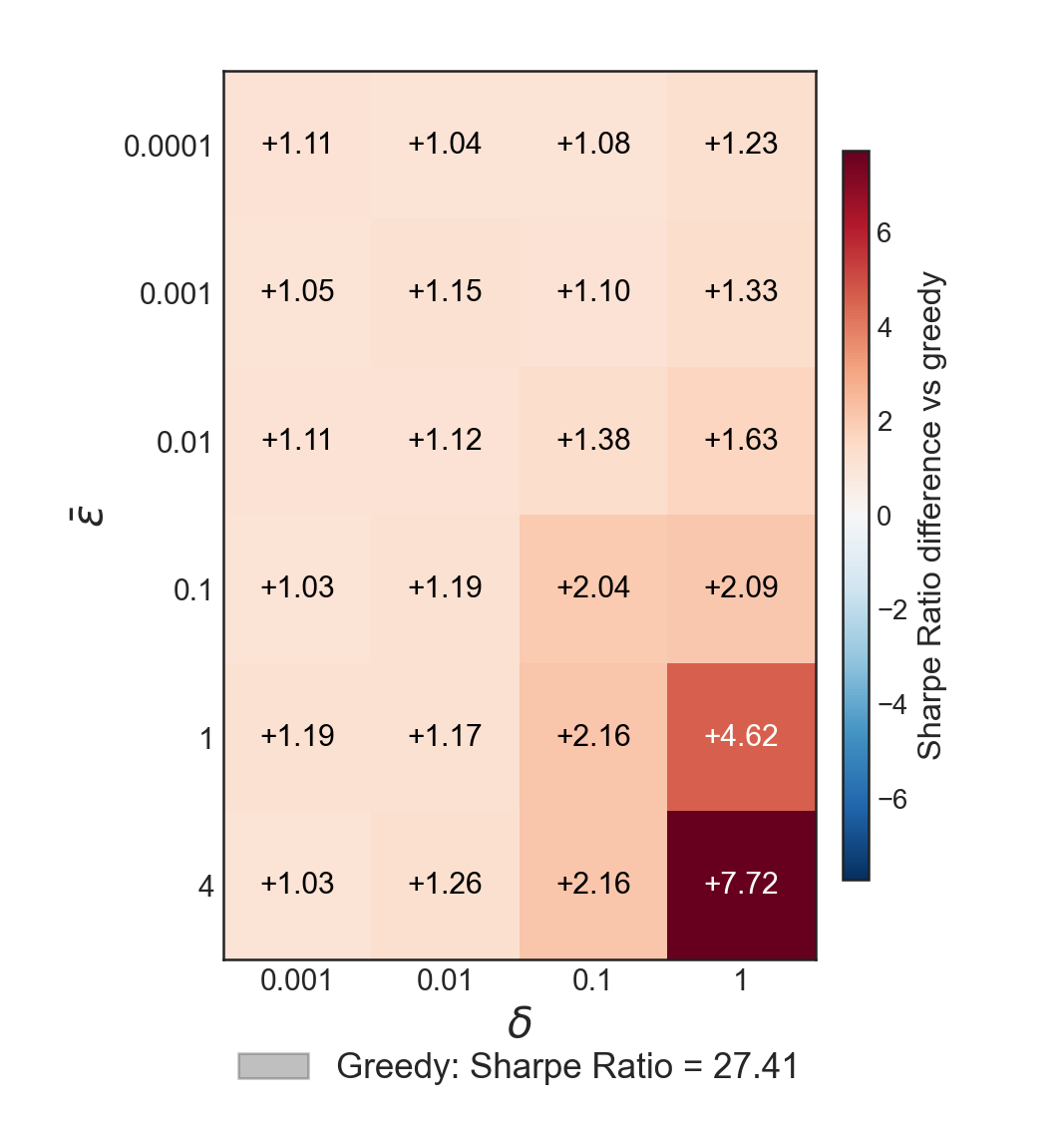}
            \caption{Fill stress, Sharpe Heatmap}
        \end{subfigure}
    \end{minipage}}
    \caption{Test-period mean PnL and Sharpe ratio over the $(\bar\varepsilon,\delta)$ grid for the buy-arrival-imbalance, sell-arrival-imbalance, and fill-stress scenarios. The first row reports mean PnL and the second row reports Sharpe ratio, with columns corresponding to the three scenarios; the greedy benchmark value is shown below each panel. The results show that out-of-sample performance depends materially on both the robustness parameters and the stress scenario.}
    \label{fig:app_sim_heatmap_2}
\end{figure}

\begin{figure}[htbp]
    \centering
    \makebox[\textwidth][c]{\begin{minipage}{1.16\textwidth}\centering
        \begin{subfigure}{0.31\linewidth}
            \includegraphics[width=\linewidth]{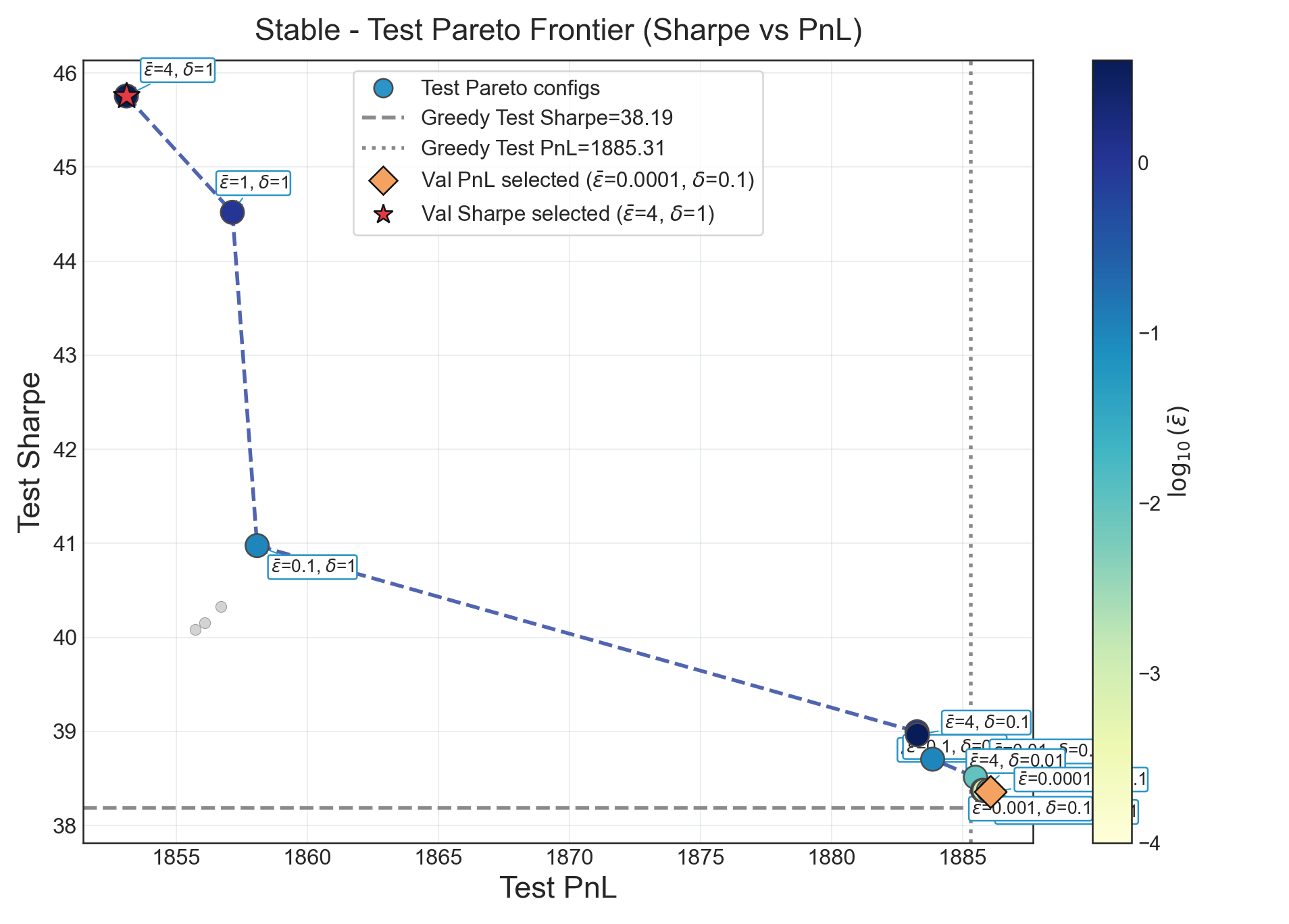}
            \caption{Stable, PnL vs Sharpe}
        \end{subfigure}
        \begin{subfigure}{0.31\linewidth}
            \includegraphics[width=\linewidth]{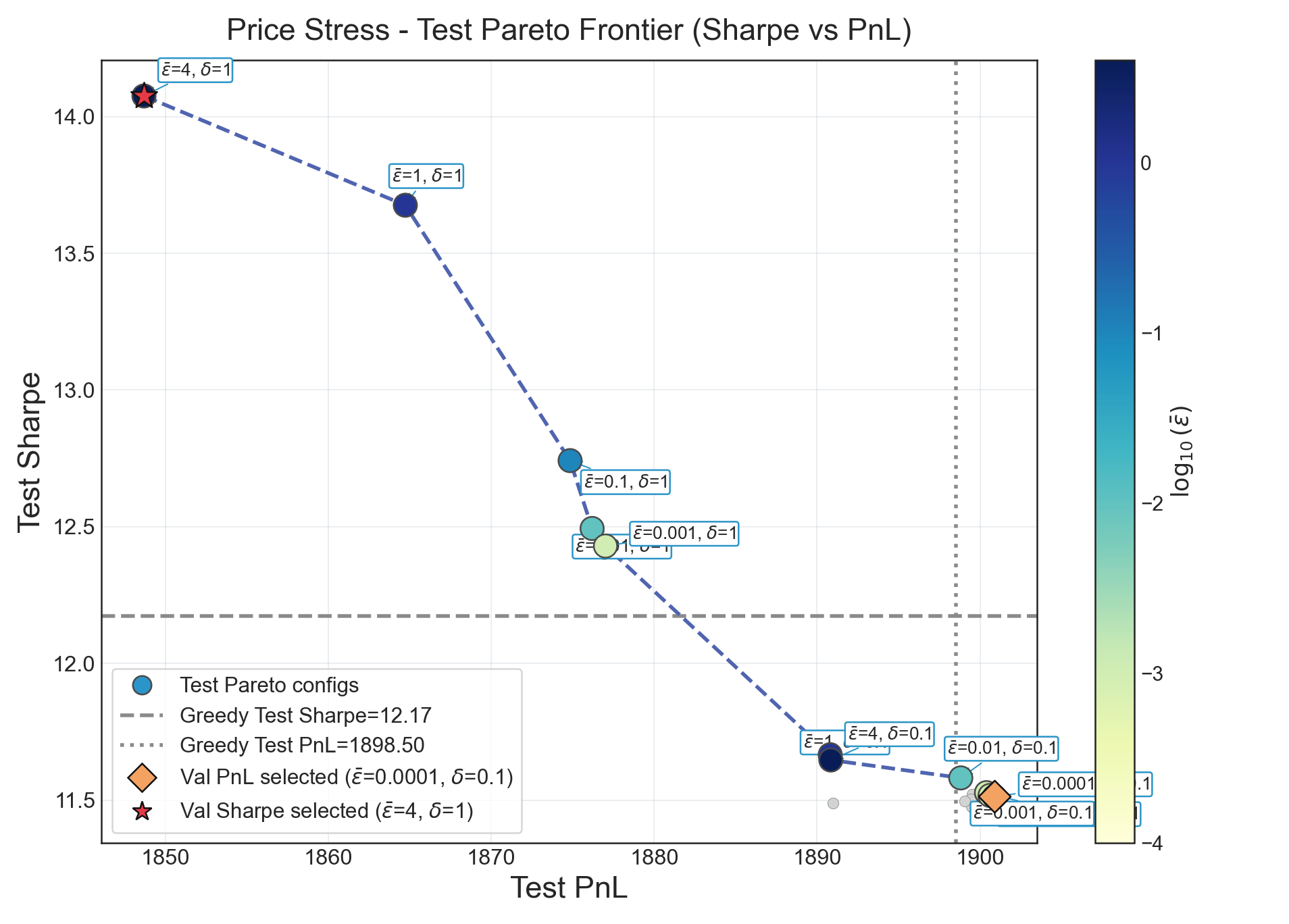}
            \caption{Price stress, PnL vs Sharpe}
        \end{subfigure}
        \begin{subfigure}{0.31\linewidth}
            \includegraphics[width=\linewidth]{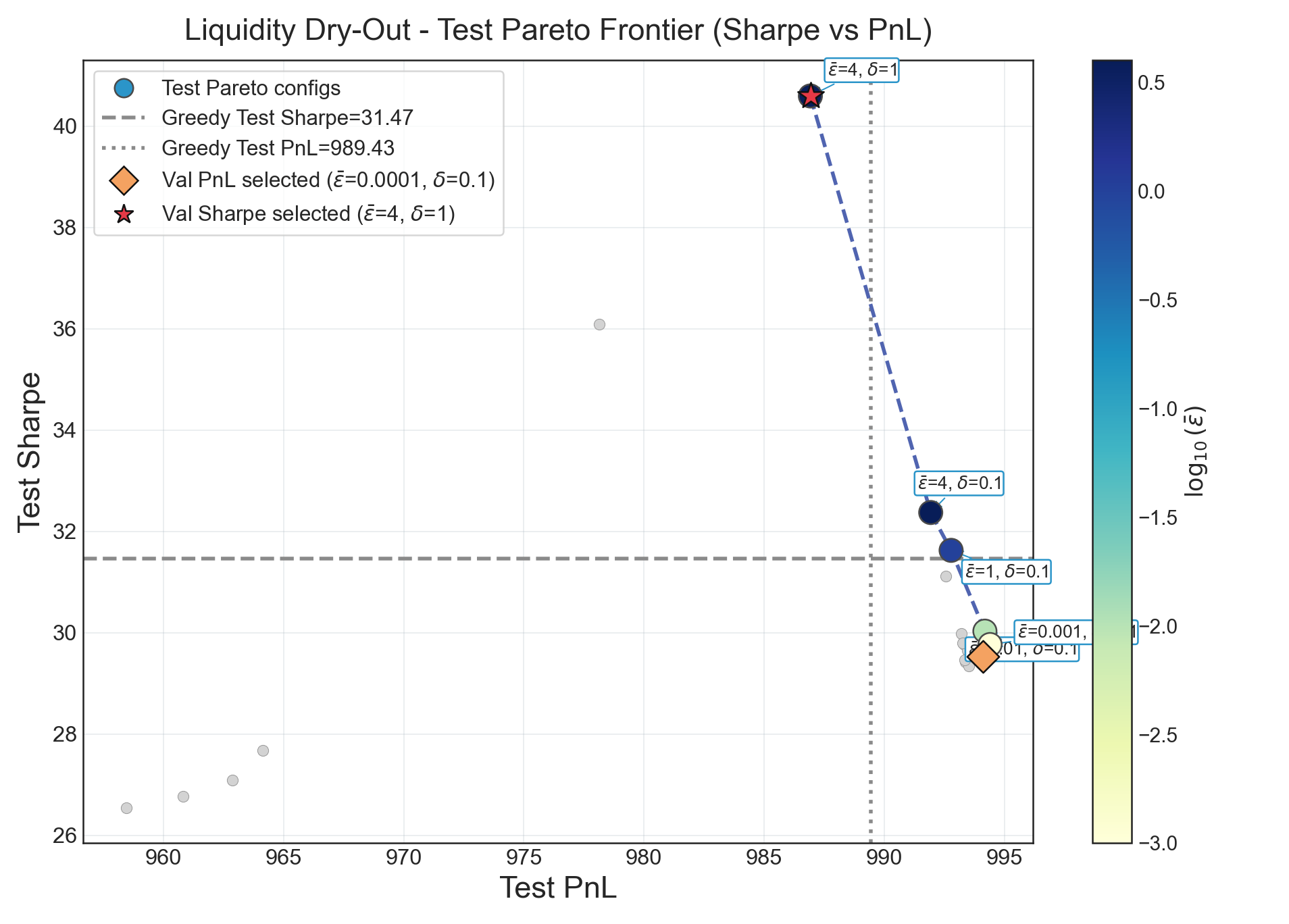}
            \caption{Liquidity dry-out, PnL vs Sharpe}
        \end{subfigure}

        \par\vspace{0.1em}

        \begin{subfigure}{0.31\linewidth}
            \includegraphics[width=\linewidth]{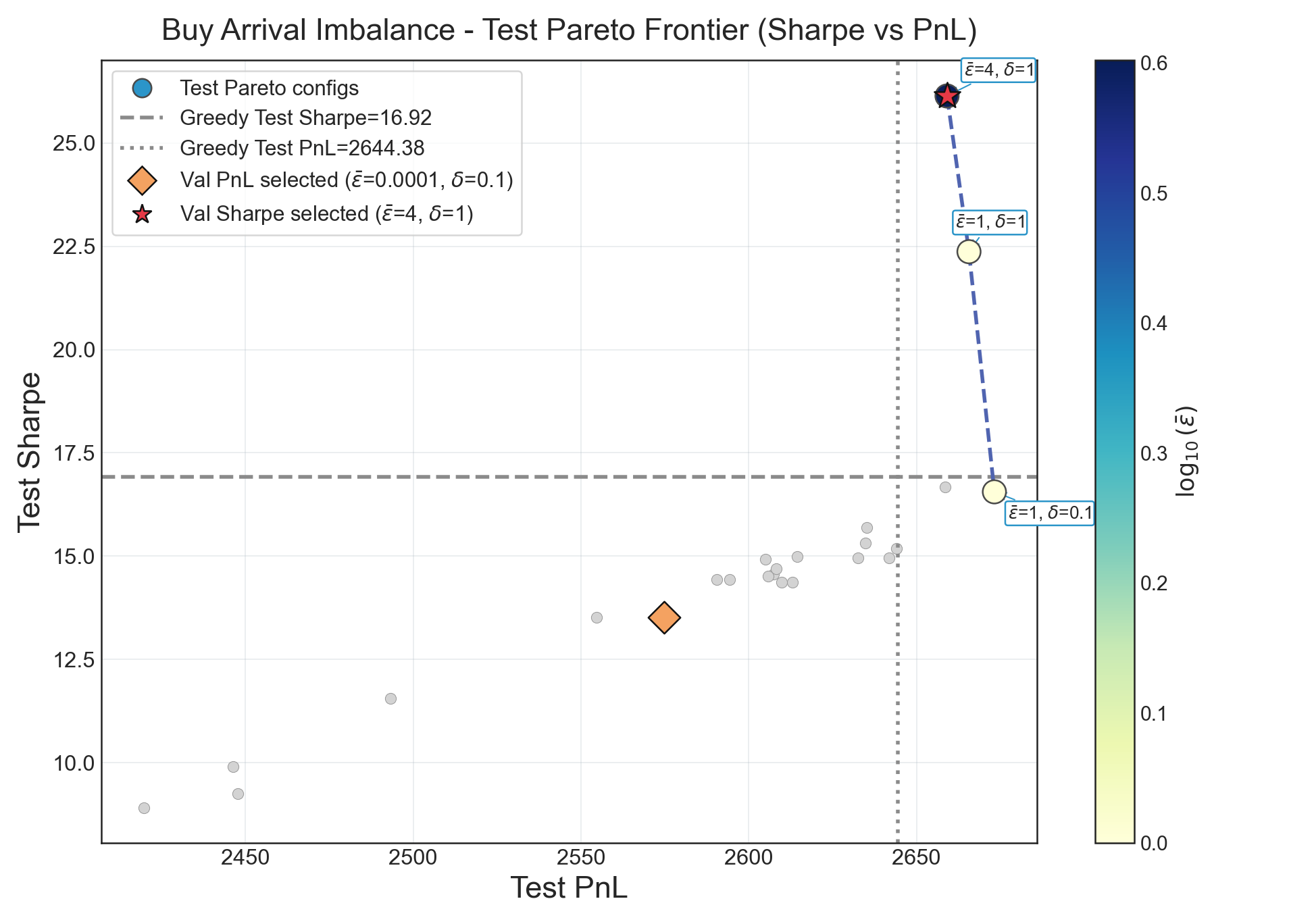}
            \caption{Buy arrival imbalance, PnL vs Sharpe}
        \end{subfigure}
        \begin{subfigure}{0.31\linewidth}
            \includegraphics[width=\linewidth]{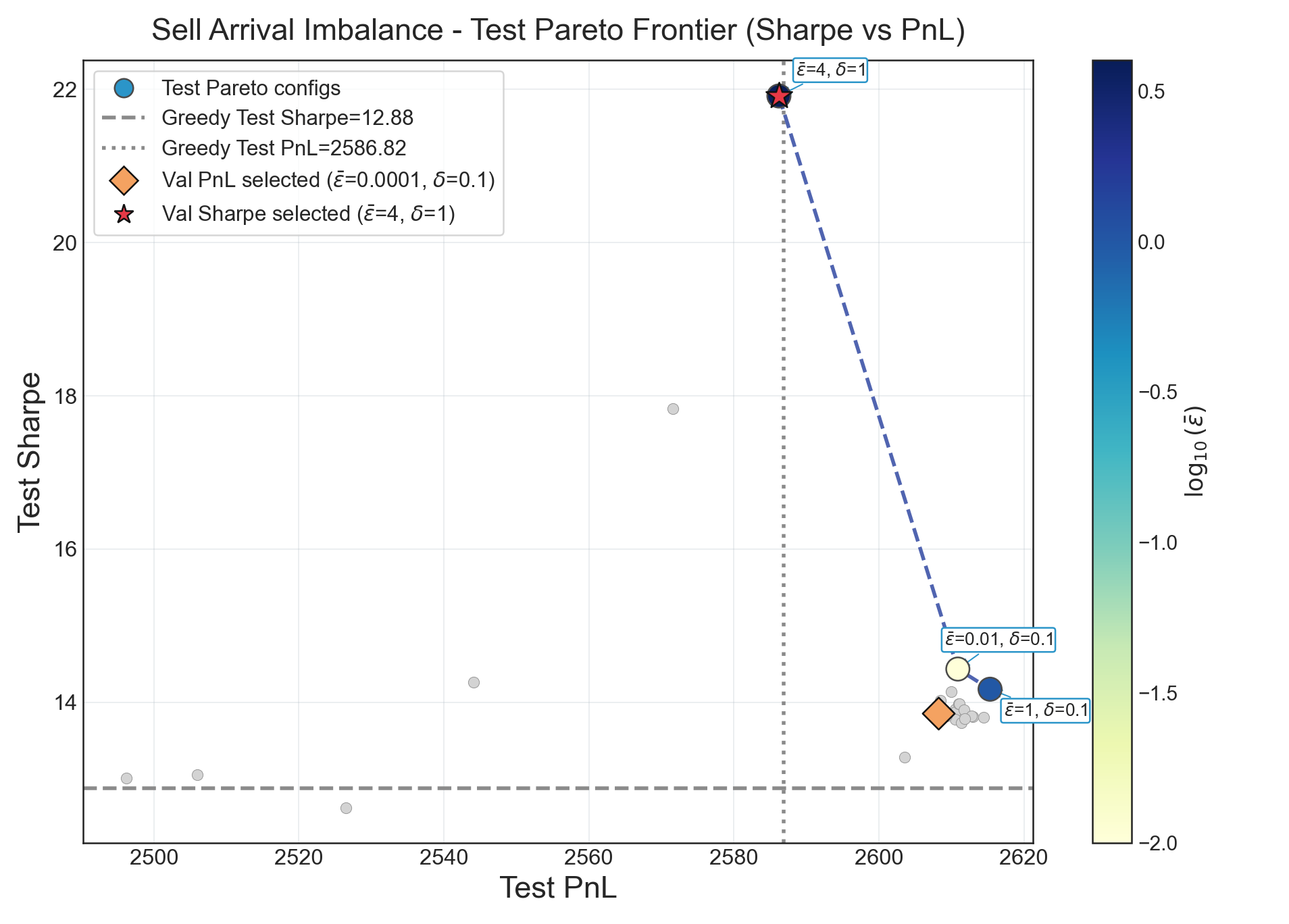}
            \caption{Sell arrival imbalance, PnL vs Sharpe}
        \end{subfigure}
        \begin{subfigure}{0.31\linewidth}
            \includegraphics[width=\linewidth]{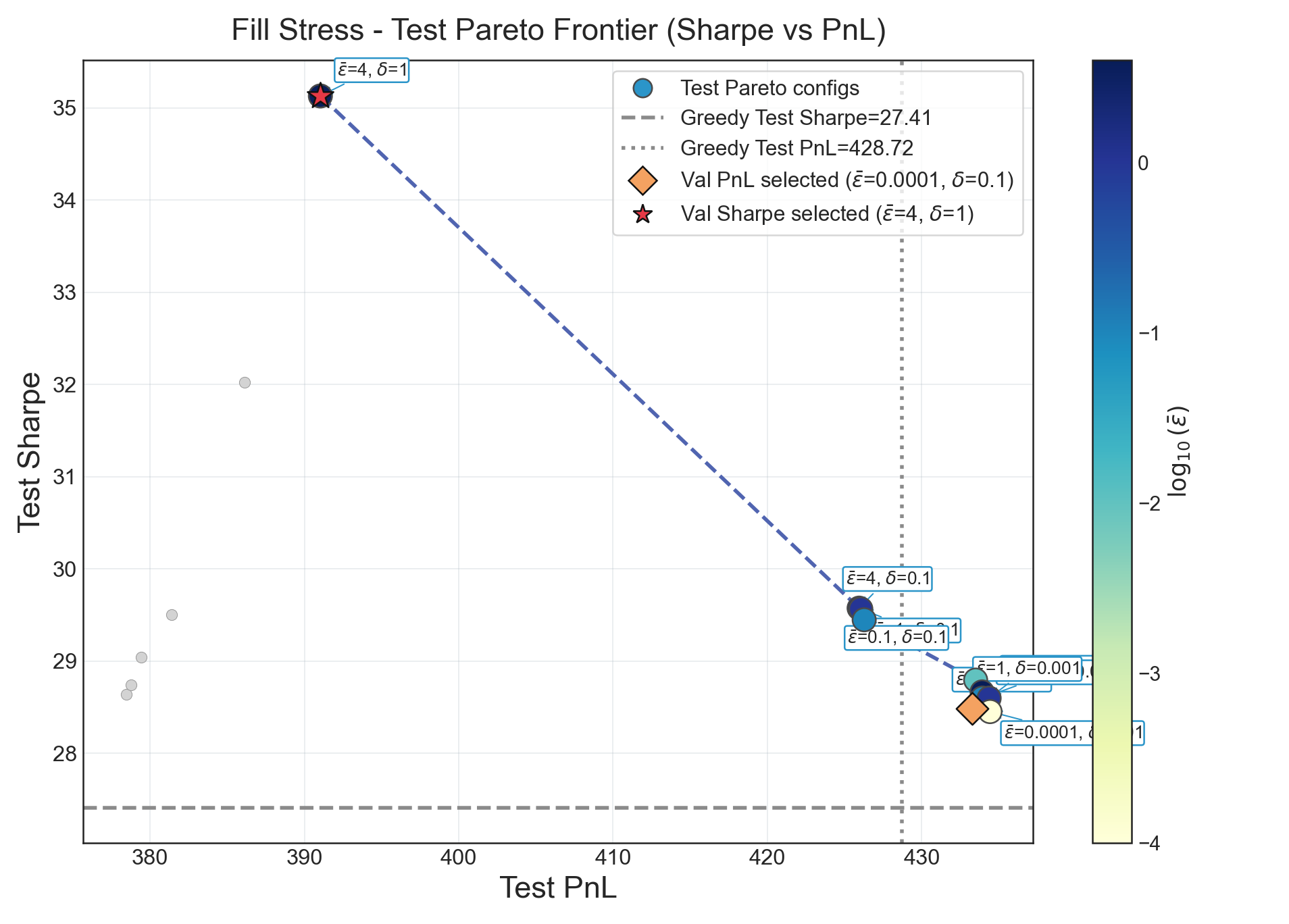}
            \caption{Fill stress, PnL vs Sharpe}
        \end{subfigure}

    \end{minipage}}
    \caption{Test-period Pareto frontiers in mean PnL-versus-Sharpe-ratio space for the six simulation stress scenarios. The panels show that the profitability--risk trade-off varies across scenarios, with some robustness choices delivering better risk-adjusted performance than others.}
    \label{fig:app_sim_pareto}
\end{figure}

\subsection{Intraday Agent Quoting Behavior}

In this section, we present the agent behavior plots in the six scenarios. Overall, the parameters $(\bar{\varepsilon},\delta)$ jointly shape sequential market-making behavior, ranging from more aggressive liquidity provision to more defensive inventory control. Intuitively, $\varepsilon$ is the original uncertainty budget, while $\bar{\varepsilon}$ is the adjusted uncertainty budget that enters the Sinkhorn-robust formulation. The parameter $\delta$ governs how the robustness adjustment is distributed within that budget. A smaller $\delta$ localizes the perturbation kernel around the reference next state and makes the entropic reweighting sharper within that neighborhood, whereas a larger $\delta$ spreads both the kernel and the reweighting more diffusely across perturbed future states. The four plotted policies therefore correspond to economically distinct sequential behaviors: the greedy policy is the least sensitive to uncertainty; $(\bar{\varepsilon},\delta)=(0.0001,0.1)$ combines a small uncertainty budget with a localized robustness adjustment; $(0.0001,1)$ keeps the same budget but spreads the robustness adjustment more broadly across perturbed future states, inducing more diffuse hedging; and $(4,1)$ combines a larger uncertainty budget with a more broadly distributed robustness adjustment, generating the strongest robust response.

The figures suggest that $\delta$ plays the dominant role in shaping sequential policy behavior. While increasing $\bar{\varepsilon}$ enlarges the adjusted uncertainty budget, the resulting policy adjustments remain relatively modest within the considered parameter range. By contrast, increasing $\delta$ substantially alters spread placement, execution intensity, and inventory stabilization. Economically, this indicates that, for the market-making dynamics considered here, how robustness is distributed within the uncertainty budget matters more than the budget size itself.


Figure \ref{fig:app_sim_agent_behav_stable} showcases the quoting behavior in the stable scenario. Figure~\ref{fig:app_sim_agent_behav_stable_action} examines the state-conditional structure of the robust Sharpe policy in the stable scenario. The stable baseline already reveals the fundamental mechanism of robustness. The policy $(0.0001,0.1)$ remains relatively close to the greedy benchmark, preserving tight spreads and strong execution intensity throughout the trading day. By contrast, increasing $\delta$ from $0.1$ to $1$ immediately generates visibly smoother spread adjustments and lower inventory fluctuations even when $\bar{\varepsilon}$ remains fixed at $0.0001$. The policy $(4,0.1)$ behaves similarly to $(0.0001,0.1)$, indicating that, within this localized regime, enlarging the adjusted uncertainty budget generates comparatively smaller behavioral changes than broadening the robustness adjustment through $\delta$. Importantly, all robust policies continue responding to economically meaningful market signals, including order arrivals, order-flow imbalance, volume pressure, and time variation. Robustness therefore preserves the economic structure of the sequential policy while changing the intensity and stability of the response. Examining the best Sharpe-validated robust policy $(\bar{\varepsilon},\delta)=(4,1)$ in greater detail reveals a characteristic intraday pattern: the agent progressively widens spreads and reduces quoted quantities toward the end of the trading session as terminal liquidation risk intensifies, while maintaining more competitive spreads during the trading day to sustain sufficient fill rates. This time-varying adjustment demonstrates that the robust agent responds meaningfully to the evolving risk profile across the session rather than applying a static conservative shift.

\begin{figure}[htbp]
    \centering
    \includegraphics[width=0.72\linewidth]
    {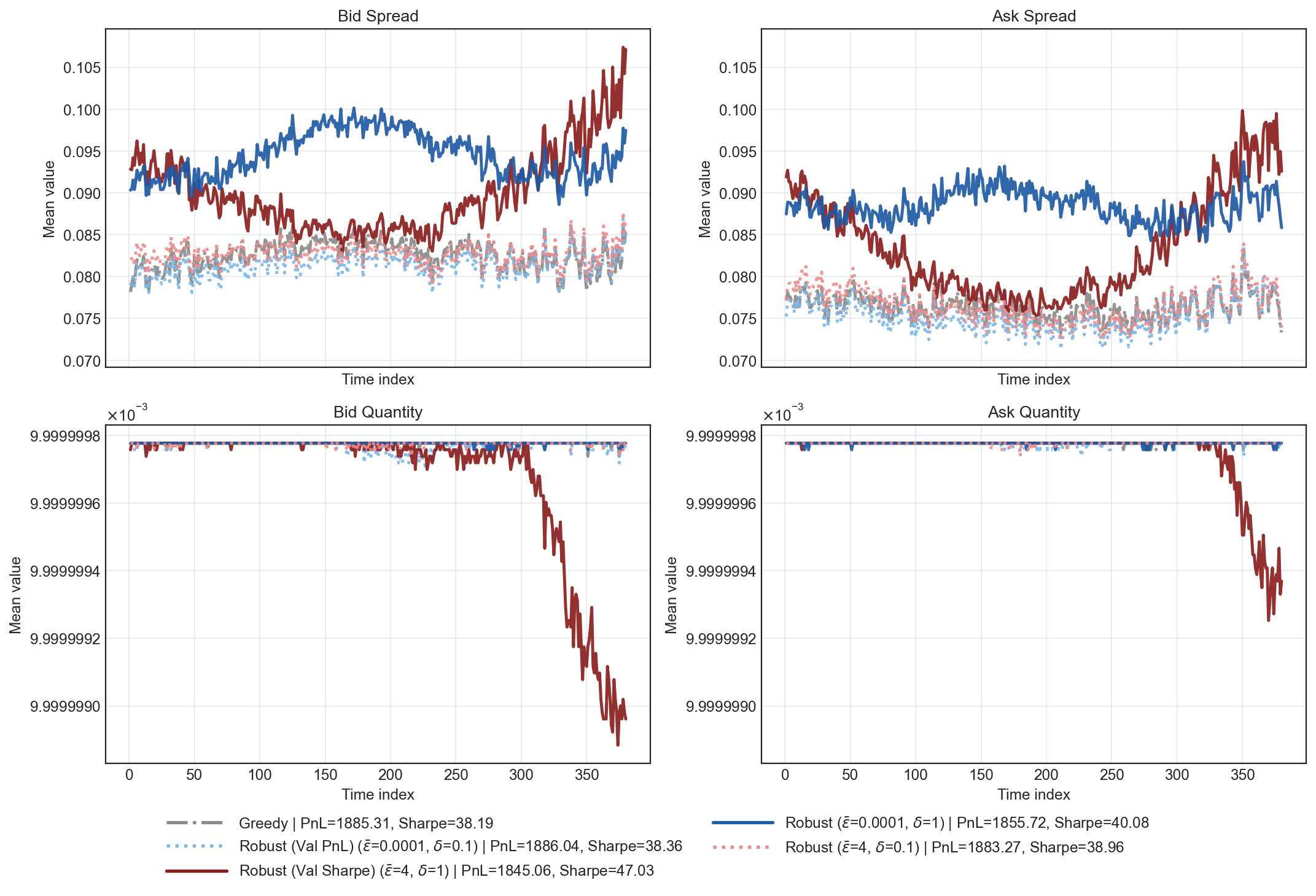}
    \caption{Intraday average bid spread, ask spread, bid quantity, and ask quantity in the stable baseline scenario for the greedy benchmark and the four robust policies. The figure shows that increasing $\delta$ has a stronger effect than increasing $\bar{\varepsilon}$, mainly through competitive spread adjustments and more selective end-of-day quoting.}
    \label{fig:app_sim_agent_behav_stable}
\end{figure}

\begin{figure}[htbp]
    \centering
    \includegraphics[width=1\linewidth]{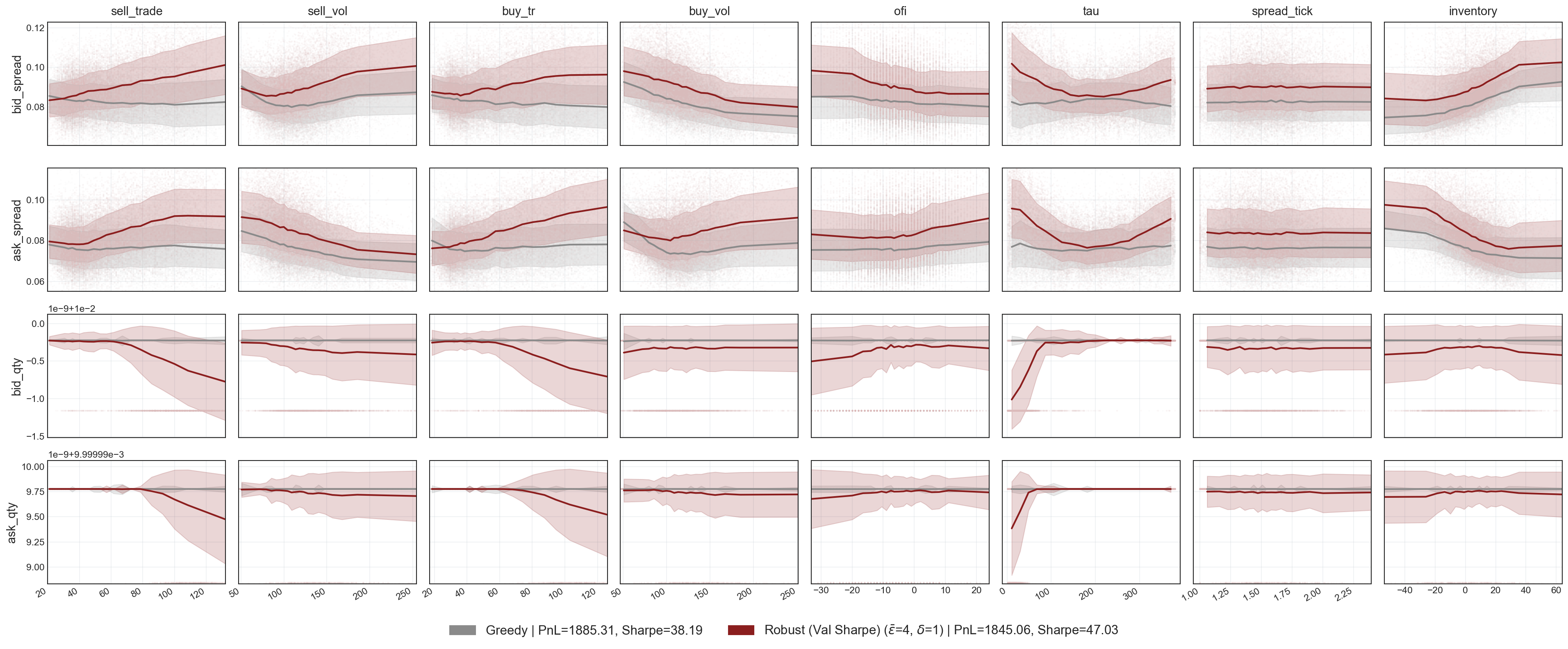}
    \caption{State-conditional mean bid quantity, ask quantity, bid spread, and ask spread of the robust and greedy agents in the stable baseline scenario, with $\pm1$ s.d.\ bands. The robust policy preserves the same broad state dependence as the greedy benchmark but becomes more selective in higher-risk states.}
    \label{fig:app_sim_agent_behav_stable_action}
\end{figure}

Figure \ref{fig:app_sim_agent_behav_dryout} showcases the quoting behavior in the liquidity dry-out scenario. Figure~\ref{fig:app_sim_agent_behav_dryout_action} examines the state-conditional structure of the robust Sharpe policy in the liquidity dry-out scenario. The liquidity dry-out scenario reveals that robustness is not uniformly conservative. The policy $(0.0001,0.1)$ keeps the spreads relatively close to the greedy benchmark in order to preserve execution opportunities under thinner market conditions. Increasing $\delta$ to $1$ produces substantially stronger stabilization, wider spreads, and lower inventory fluctuations. The incremental effect of increasing $\bar{\varepsilon}$ from $0.0001$ to $4$ remains comparatively modest relative to the behavioral transition induced by $\delta$. Robustness therefore balances two competing objectives: preserving fills versus reducing exposure to uncertain execution conditions. The best Sharpe-validated policy $(\bar{\varepsilon},\delta)=(4,1)$ exhibits a notably different intraday pattern from the stable and price-stress scenarios. Rather than uniformly widening spreads, the agent actively narrows spreads below the greedy benchmark during the trading day, responding to the heightened risk of insufficient executions in thin market conditions. By quoting more competitively, the agent prioritizes fill rate preservation when liquidity is scarce. Nevertheless, the agent still widens spreads and reduces quoted quantities toward the end of the trading session, maintaining awareness of terminal liquidation risk. This nuanced balance between competitive intraday quoting and end-of-session caution reflects a meaningful and economically coherent response to the specific structure of dry-out uncertainty.

\begin{figure}[htbp]
    \centering
    \includegraphics[width=0.72\linewidth]
    {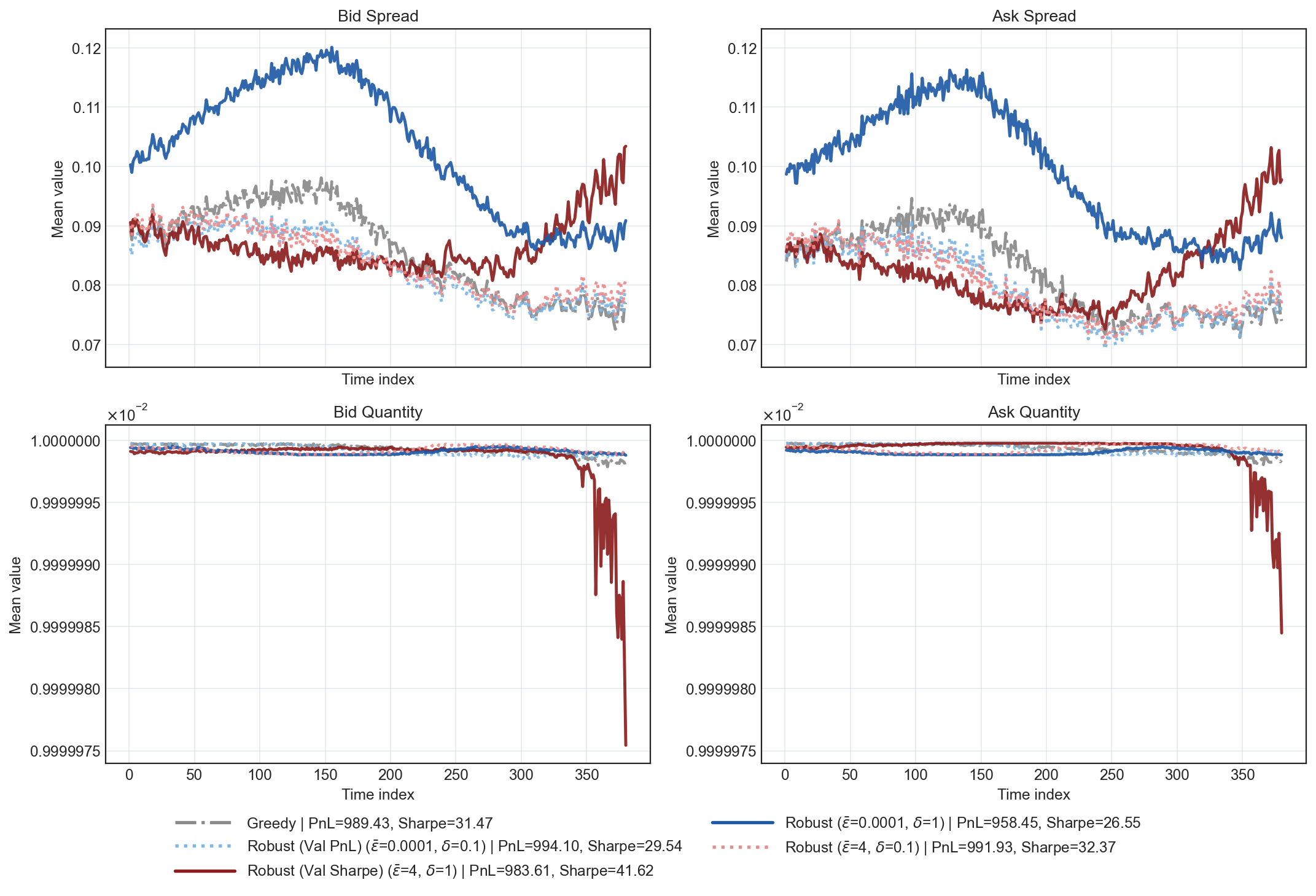}
    \caption{Intraday average bid spread, ask spread, bid quantity, and ask quantity in the liquidity dry-out scenario for the greedy benchmark and the four robust policies. The figure shows that stronger robustness mainly affects the policy through $\delta$, with the best robust policy quoting more competitively during the trading day to preserve fills under thin liquidity and more conservatively at the end of the day.}
    \label{fig:app_sim_agent_behav_dryout}
\end{figure}

\begin{figure}[htbp]
    \centering
    \includegraphics[width=1\linewidth]{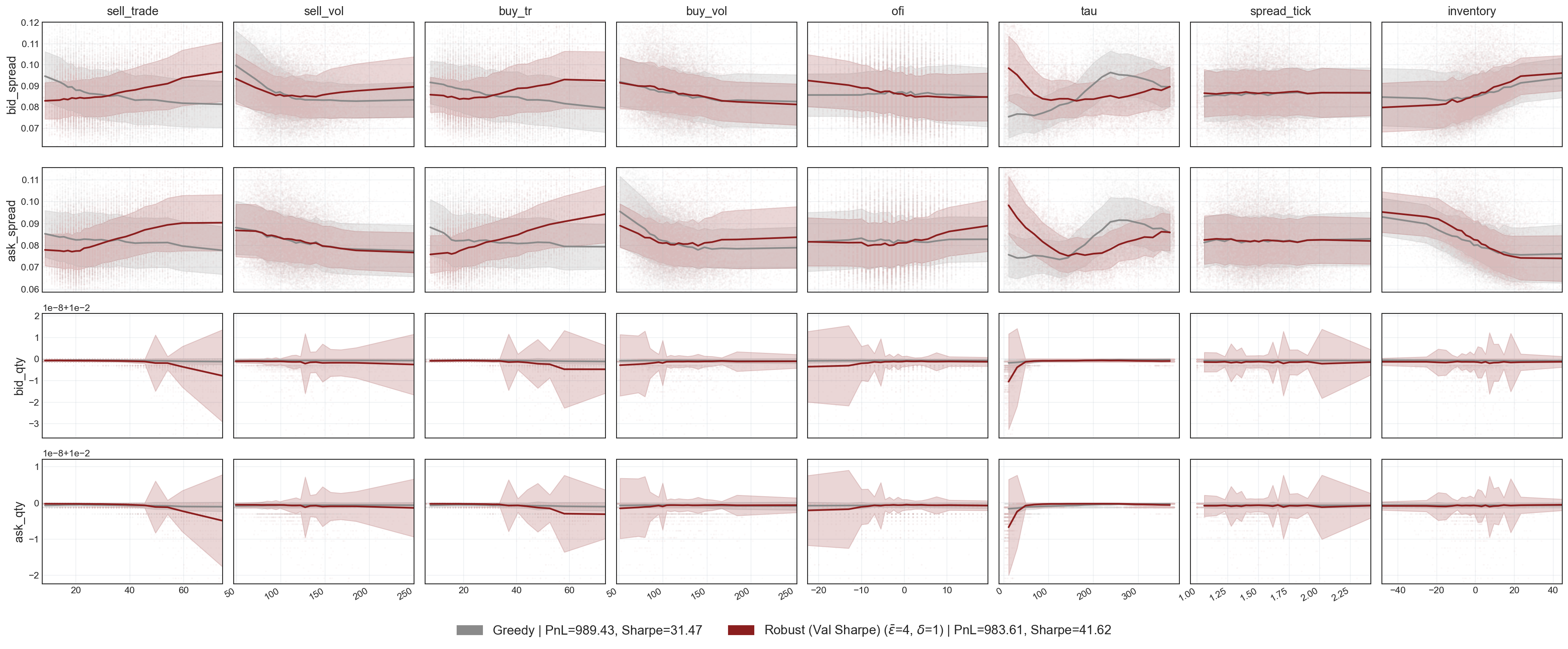}
    \caption{State-conditional mean bid quantity, ask quantity, bid spread, and ask spread of the robust and greedy agents in the liquidity dry-out scenario, with $\pm1$ s.d.\ bands. The robust policy remains sensitive to the same broad state signals as the greedy benchmark but quotes more competitively in states where preserving execution opportunities is especially important.}
    \label{fig:app_sim_agent_behav_dryout_action}
\end{figure}

Figure \ref{fig:app_sim_agent_buy_arr} showcases the quoting behavior in the buy arrival imbalance scenario. Figure~\ref{fig:app_sim_agent_behav_buy_arr_action} examines the state-conditional structure of the robust Sharpe policy in the buy arrival imbalance scenario. It provide the clearest illustration of how robustness reshapes sequential adaptation. Under buy-arrival imbalance, persistent buy market orders generate strong ask-side execution pressure and large negative inventory accumulation. As $\delta$ increases, the policy reacts by widening ask spreads more aggressively, thereby slowing inventory depletion and reducing directional exposure. The transition from $(0.0001,0.1)$ to $(0.0001,1)$ produces visibly larger changes than the transition from $(0.0001,1)$ to $(4,1)$, again suggesting that the distribution of the robustness adjustment matters more than the size of the adjusted uncertainty budget within the considered parameter range. The best Sharpe-validated policy $(\bar{\varepsilon},\delta)=(4,1)$ responds to persistent buy pressure with a directionally targeted adjustment: rather than broadly widening spreads, the agent selectively manages ask-side exposure by exploiting the predictable directionality of the order flow while keeping spreads sufficiently tight to maintain execution opportunities. Terminal liquidation risk shapes end-of-session behavior, as the agent tightens inventory positions when the horizon shortens. This directionally specific response confirms that robust agents adapt to the sign and structure of the imbalance signal rather than applying generic conservatism.

\begin{figure}[htbp]
    \centering
    \includegraphics[width=0.72\linewidth]
    {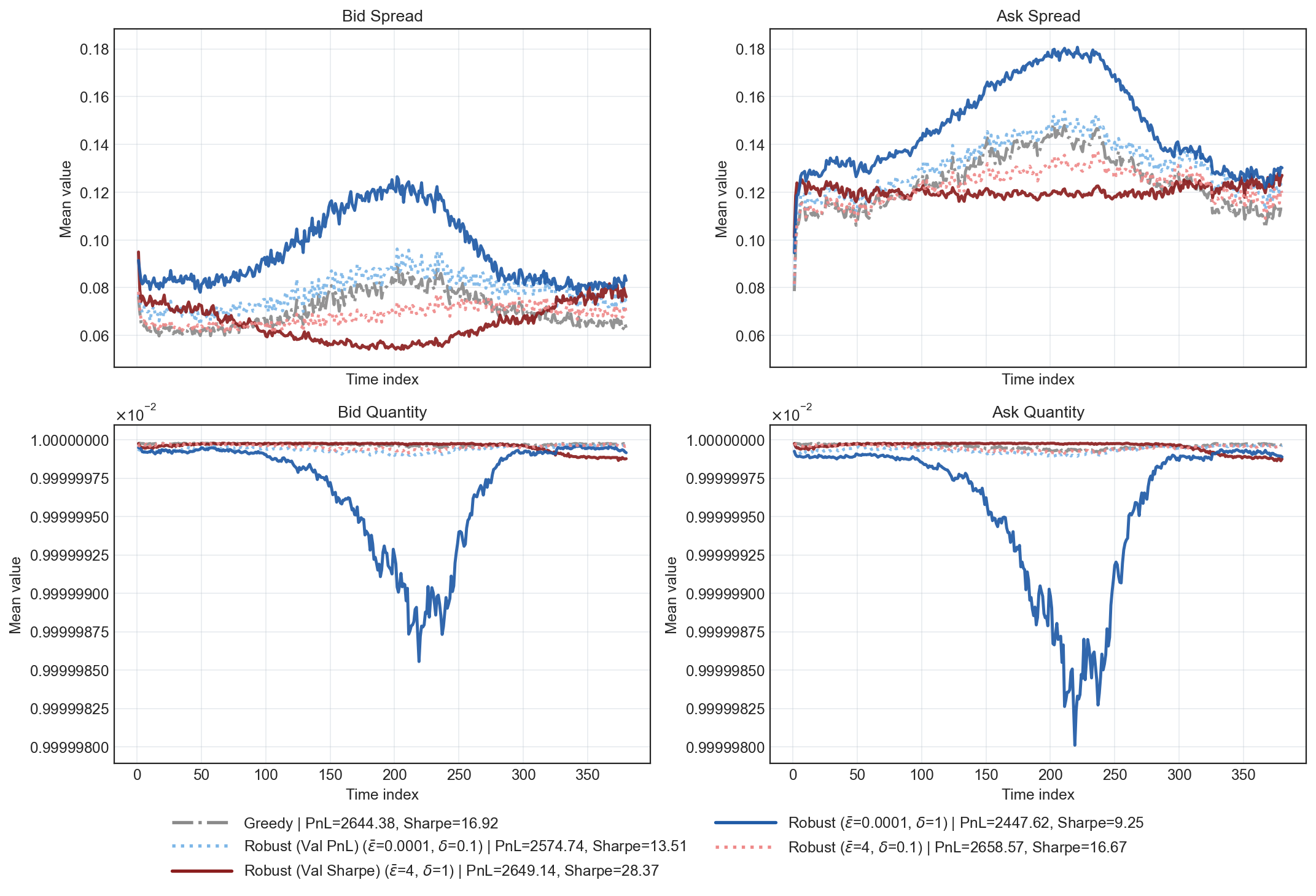}
    \caption{Intraday average bid spread, ask spread, bid quantity, and ask quantity in the buy-arrival-imbalance scenario for the greedy benchmark and the four robust policies. The figure shows that stronger robustness mainly affects the policy through $\delta$, with the best robust policy responding to persistent buy pressure by managing ask-side exposure more selectively.}
    \label{fig:app_sim_agent_buy_arr}
\end{figure}

\begin{figure}[htbp]
    \centering
    \includegraphics[width=1\linewidth]{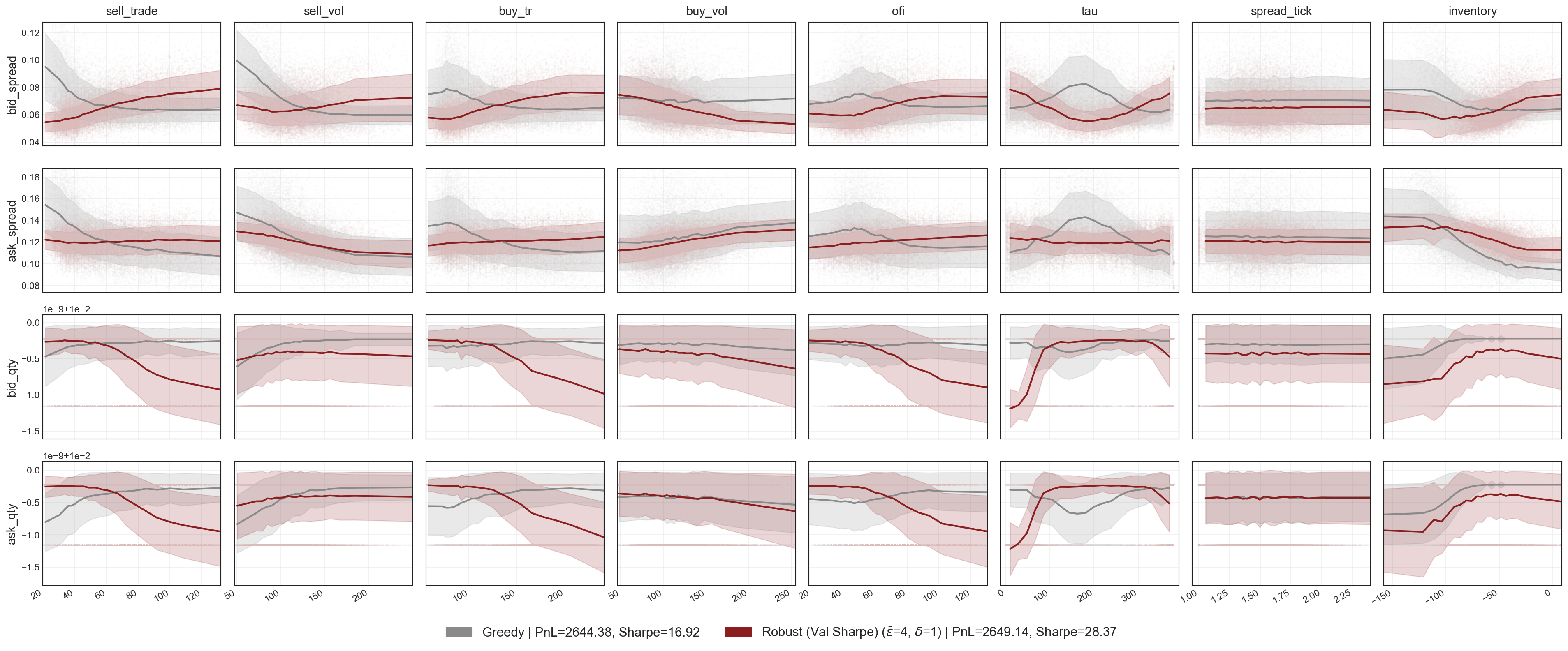}
    \caption{State-conditional mean bid quantity, ask quantity, bid spread, and ask spread of the robust and greedy agents in the buy-arrival-imbalance scenario, with $\pm1$ s.d.\ bands. The robust policy preserves the directional dependence on the state variables but adjusts the ask side more selectively in response to persistent buy pressure.}
    \label{fig:app_sim_agent_behav_buy_arr_action}
\end{figure}

Figure \ref{fig:app_sim_agent_sell_arr} showcases the quoting behavior in the sell arrival imbalance scenario. Figure~\ref{fig:app_sim_agent_sell_arr_action} examines the state-conditional structure of the robust Sharpe policy in the sell arrival imbalance scenario. The sell arrival imbalance scenario exhibits the mirror image. Persistent sell market orders create positive inventory accumulation through repeated bid-side executions. Increasing $\delta$ substantially widens bid spreads and stabilizes inventory trajectories, whereas increasing $\bar{\varepsilon}$ produces comparatively smaller additional adjustments. The best Sharpe-validated policy $(\bar{\varepsilon},\delta)=(4,1)$ mirrors the buy-arrival response but on the bid side: the agent selectively manages bid-side exposure generated by the persistent sell pressure while maintaining competitive ask spreads to preserve execution opportunities on the other side. The directional symmetry between the buy- and sell-arrival responses confirms that the robust policy adapts specifically to the sign of the imbalance signal, and terminal liquidation risk again governs the end-of-session spread and quantity dynamics. The figures therefore show that robustness modifies policy behavior in a directionally meaningful way rather than through simple uniform conservatism.

\begin{figure}[htbp]
    \centering
    \includegraphics[width=0.72\linewidth]
    {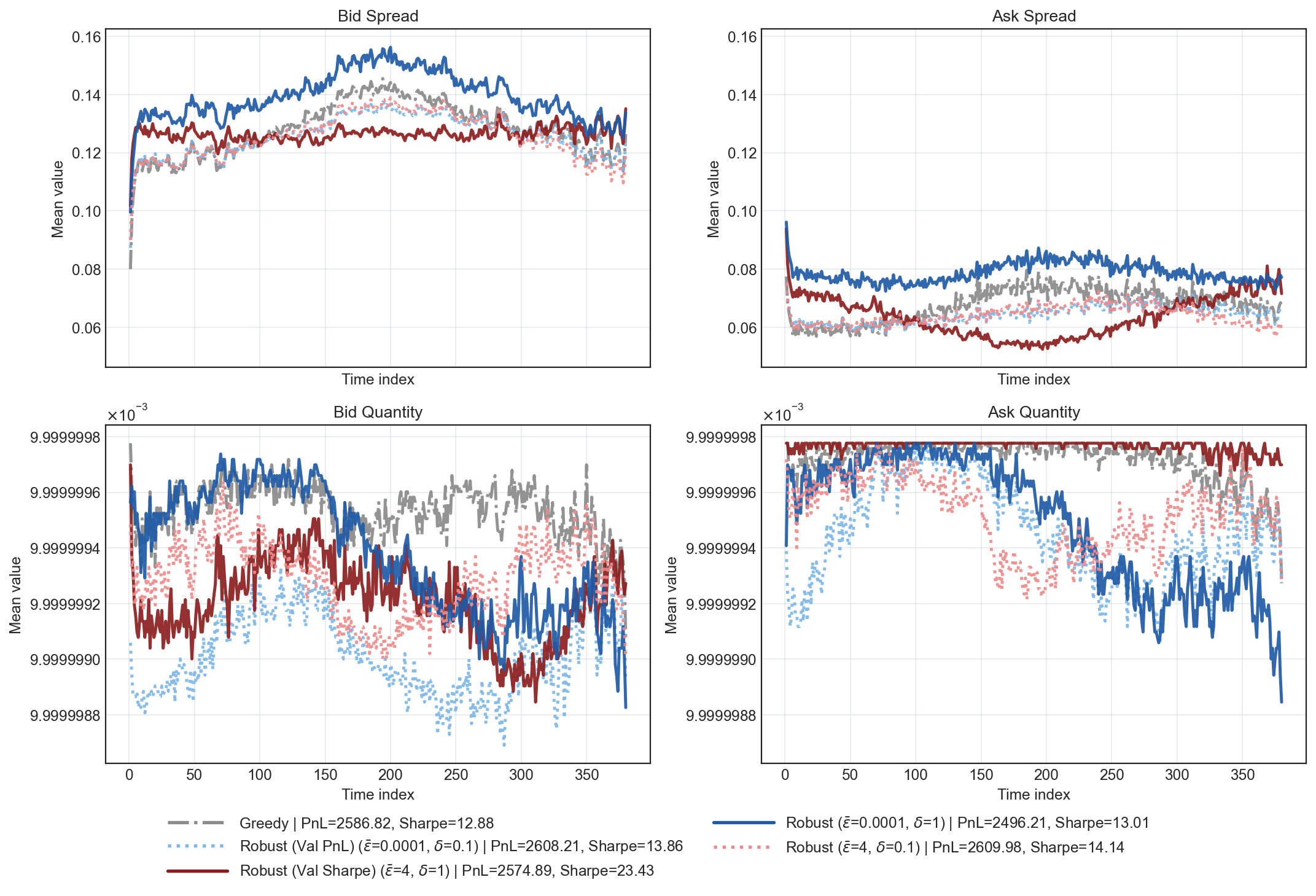}
    \caption{Intraday average bid spread, ask spread, bid quantity, and ask quantity in the sell-arrival-imbalance scenario for the greedy benchmark and the four robust policies. The figure shows that stronger robustness mainly affects the policy through $\delta$, with the best robust policy responding to persistent sell pressure by managing bid-side exposure more selectively.}
    \label{fig:app_sim_agent_sell_arr}
\end{figure}

\begin{figure}[htbp]
    \centering
    \includegraphics[width=1\linewidth]{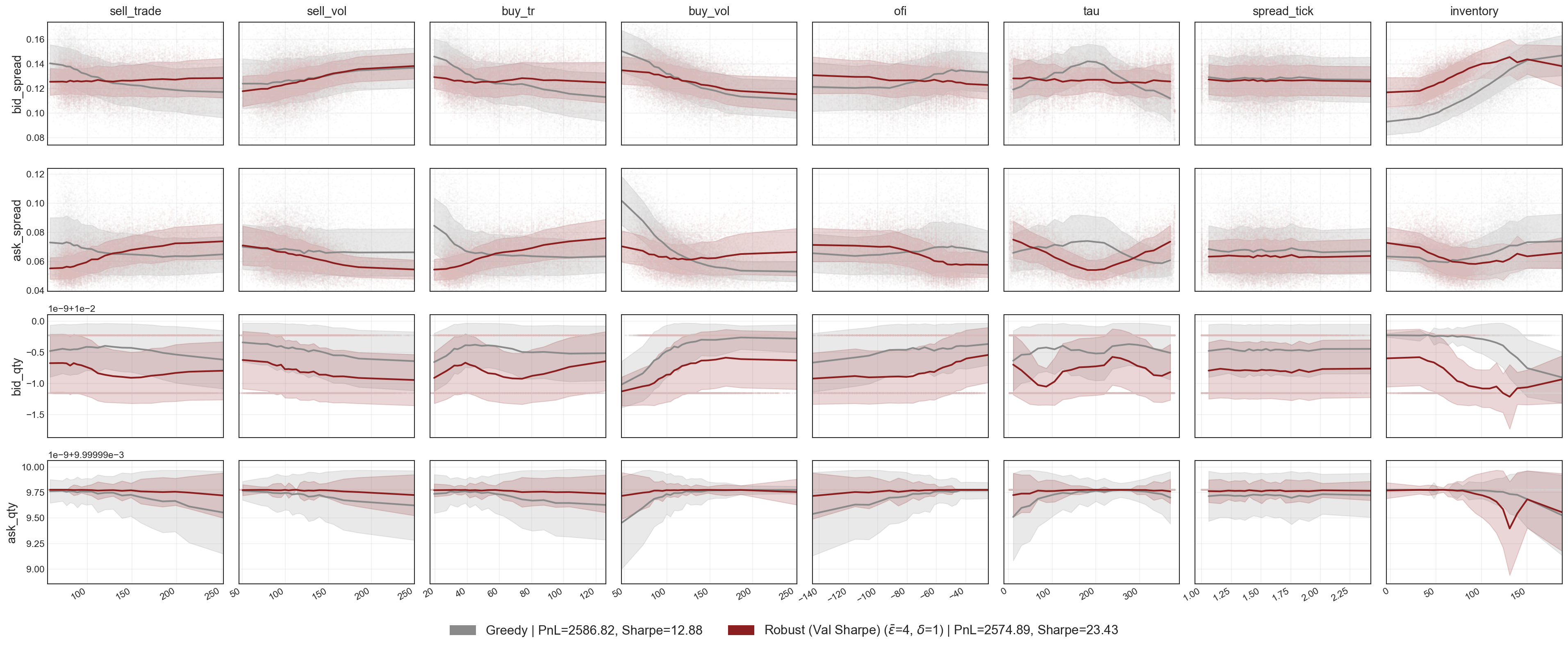}
    \caption{State-conditional mean bid quantity, ask quantity, bid spread, and ask spread of the robust and greedy agents in the sell-arrival-imbalance scenario, with $\pm1$ s.d.\ bands. The robust policy preserves the directional dependence on the state variables but adjusts the bid side more selectively in response to persistent sell pressure.}
    \label{fig:app_sim_agent_sell_arr_action}
\end{figure}

Figure \ref{fig:app_sim_agent_fill_stress} showcases the quoting behavior in the fill stress scenario. Figure~\ref{fig:app_sim_agent_fill_stress_action} examines the state-conditional structure of the robust Sharpe policy in the fill stress scenario. The fill stress scenario provides a different form of uncertainty. Here the state dynamics remain relatively stable, but execution quality deteriorates. The policy $(0.0001,0.1)$ largely preserves execution opportunities, while larger values of $\delta$ progressively shift the policy toward more defensive quoting and smoother risk exposure. Once again, the additional behavioral effect induced by increasing $\bar{\varepsilon}$ remains comparatively limited relative to the changes generated by $\delta$. The best Sharpe-validated policy $(\bar{\varepsilon},\delta)=(4,1)$ responds by widening spreads substantially on both bid and ask sides, demanding greater compensation for fill uncertainty. This stands in direct contrast to the liquidity dry-out scenario, where the agent narrowed spreads to compete for scarce executions: under fill stress, the source of execution risk is unreliable fill quality rather than thin liquidity, so the agent widens spreads to extract higher compensation per fill rather than competing more aggressively for volume. Terminal liquidation risk further reinforces spread widening and quantity reduction toward session end. The contrasting responses to two structurally different forms of execution uncertainty illustrate that robust agents adapt to the nature of the risk rather than applying a uniform defensive posture.

\begin{figure}[htbp]
    \centering
    \includegraphics[width=0.72\linewidth]
    {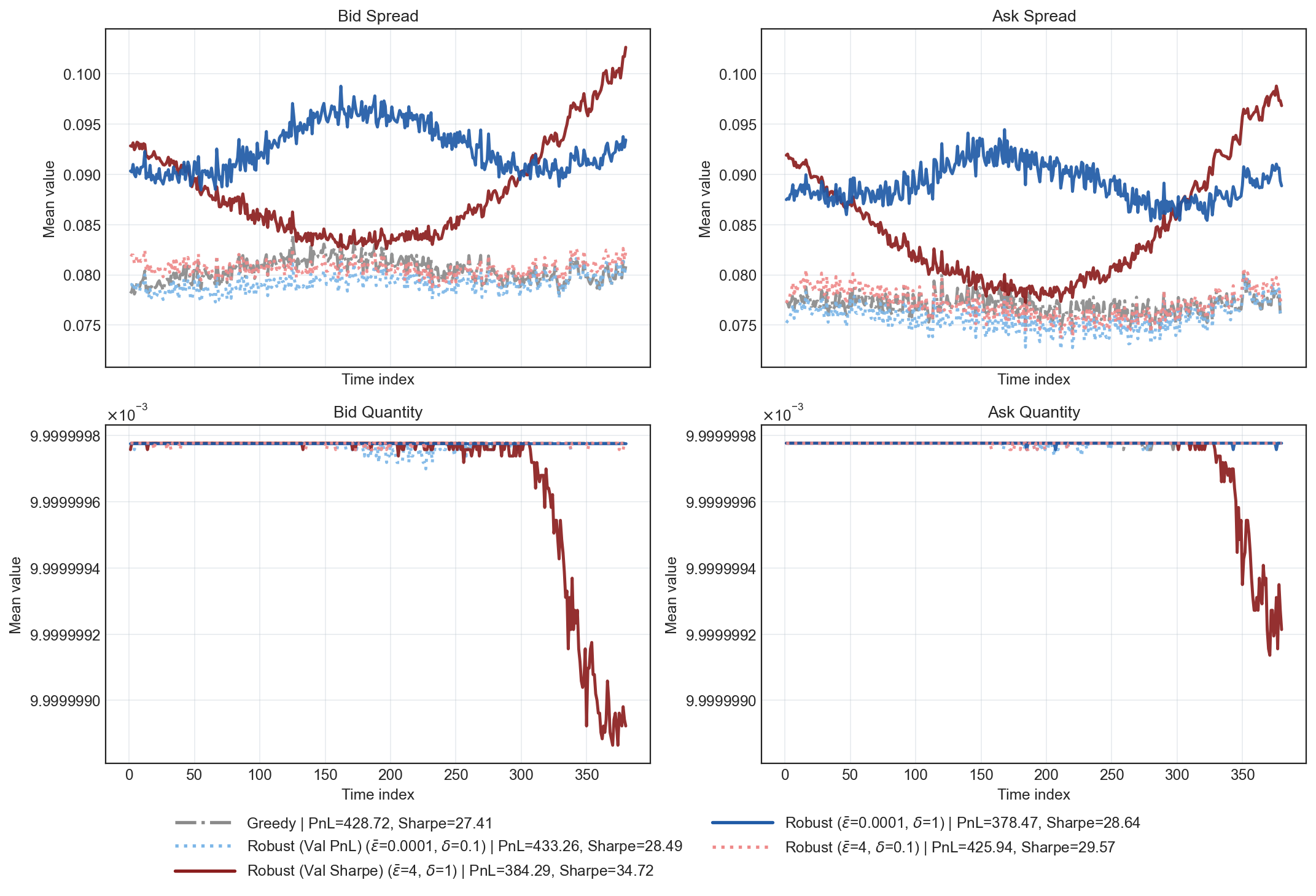}
    \caption{Intraday average bid spread, ask spread, bid quantity, and ask quantity in the fill-stress scenario for the greedy benchmark and the four robust policies. The figure shows that stronger robustness mainly affects the policy through $\delta$, with the best robust policy widening spreads on both sides to demand greater compensation for execution uncertainty.}
    \label{fig:app_sim_agent_fill_stress}
\end{figure}

\begin{figure}[H]
    \centering
    \includegraphics[width=1\linewidth]{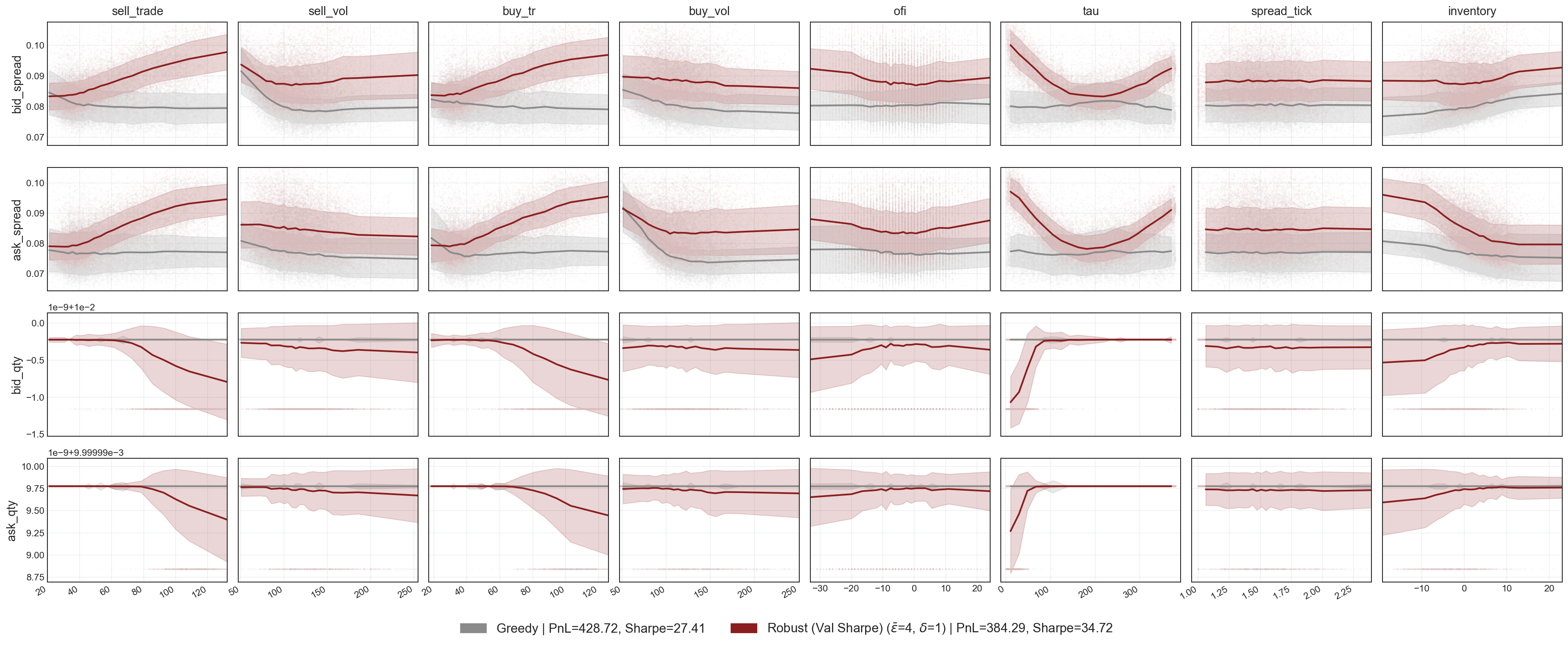}
    \caption{State-conditional mean bid quantity, ask quantity, bid spread, and ask spread of the robust and greedy agents in the fill-stress scenario, with $\pm1$ s.d.\ bands. The robust policy preserves the same broad state dependence as the greedy benchmark but adjusts both spreads and quantities more defensively when execution uncertainty is elevated.}
    \label{fig:app_sim_agent_fill_stress_action}
\end{figure}
Figure~\ref{fig:app_sim_inventory} complements the quoting-behavior plots by showing how these policy adjustments propagate into inventory dynamics. Across scenarios, the robust policy typically produces more stable inventory paths and responds more selectively to persistent directional or execution-related risks, with the largest differences appearing in the stressed environments.

Taken together, the simulation figures support the central mechanism of the paper. Robustness does not simply make the market maker quote wider in all states. Instead, it reshapes sequential policy adaptation by dynamically reallocating the trade-off between liquidity provision, execution intensity, and inventory stabilization under evolving uncertainty.

\begin{figure}[htbp]
    \centering
    \makebox[\textwidth][c]{\begin{minipage}{0.95\textwidth}\centering
        \begin{subfigure}{0.5\linewidth}
            \includegraphics[width=\linewidth]{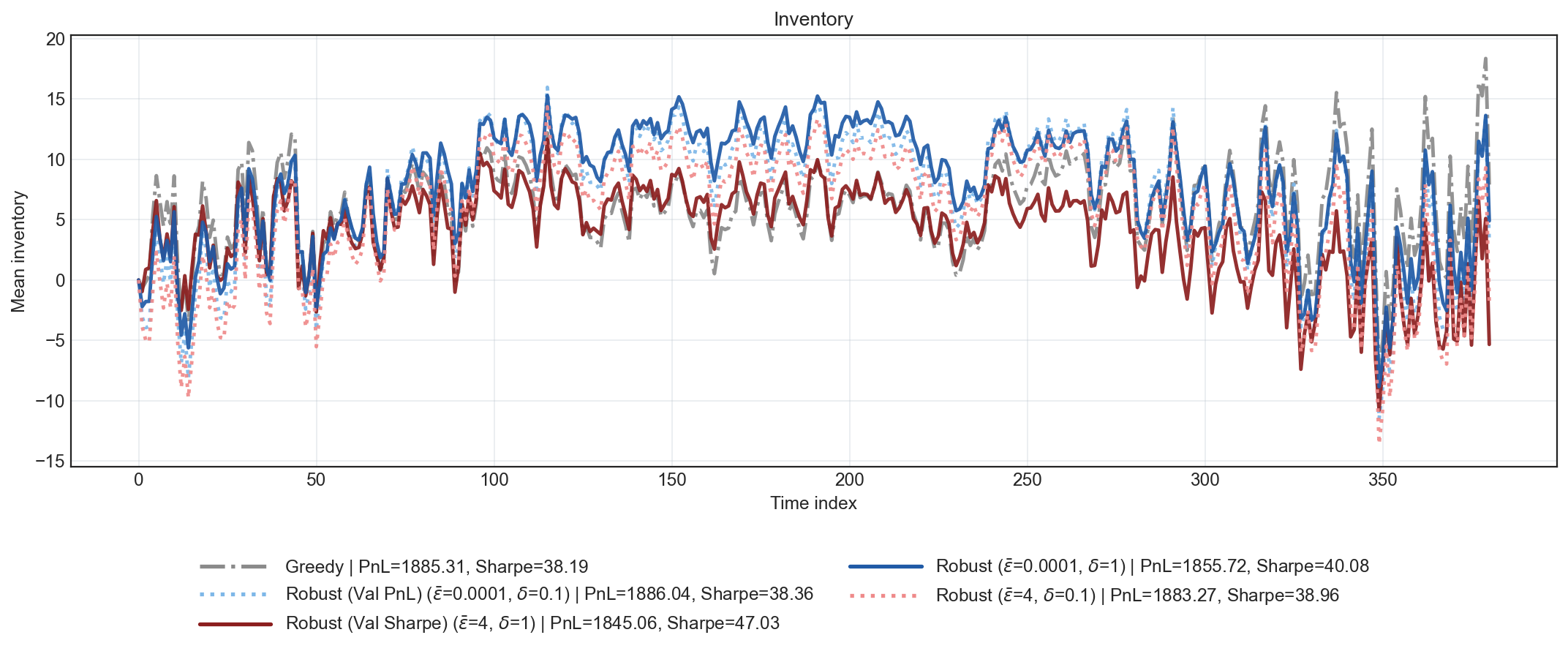}
            \caption{Stable, Inventory path}
        \end{subfigure}\hfill
        \begin{subfigure}{0.5\linewidth}
            \includegraphics[width=\linewidth]{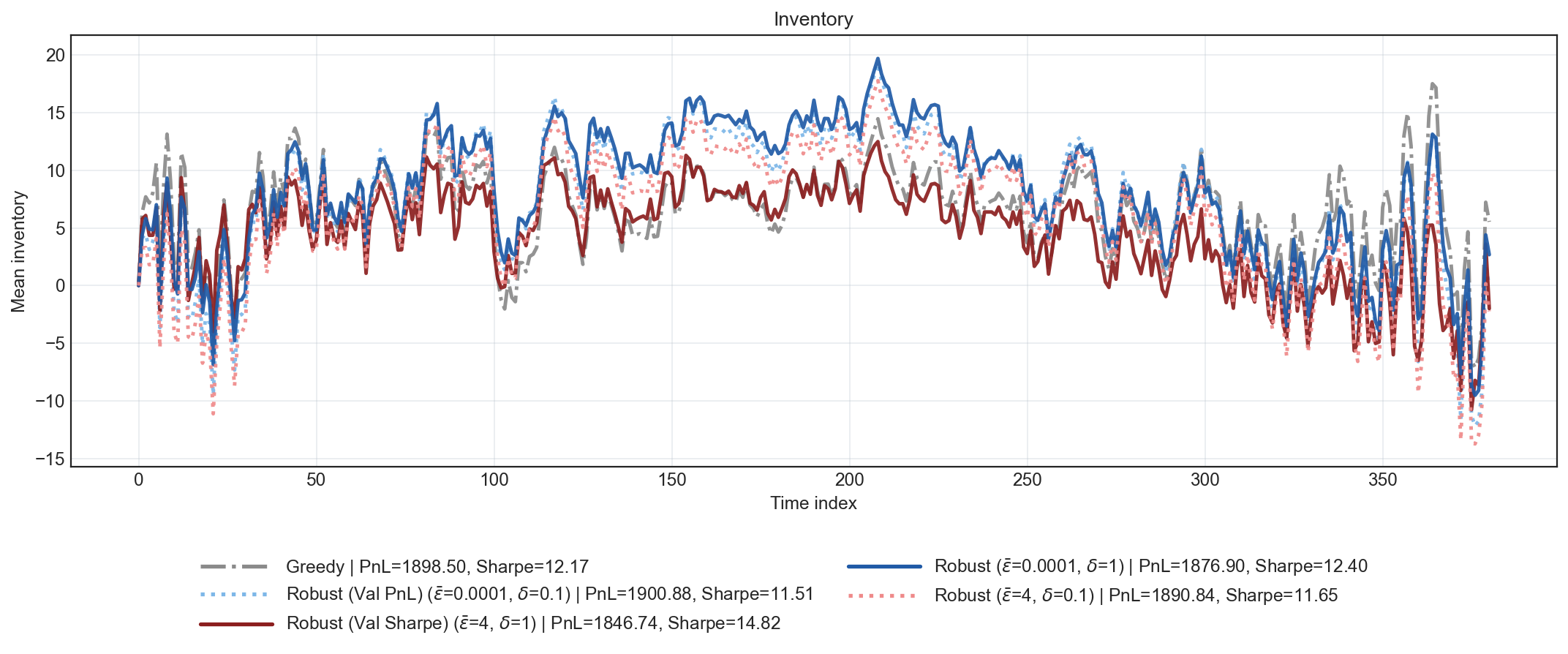}
            \caption{Price stress, Inventory path}
        \end{subfigure}

        \vspace{0.0em}

        \begin{subfigure}{0.5\linewidth}
            \includegraphics[width=\linewidth]{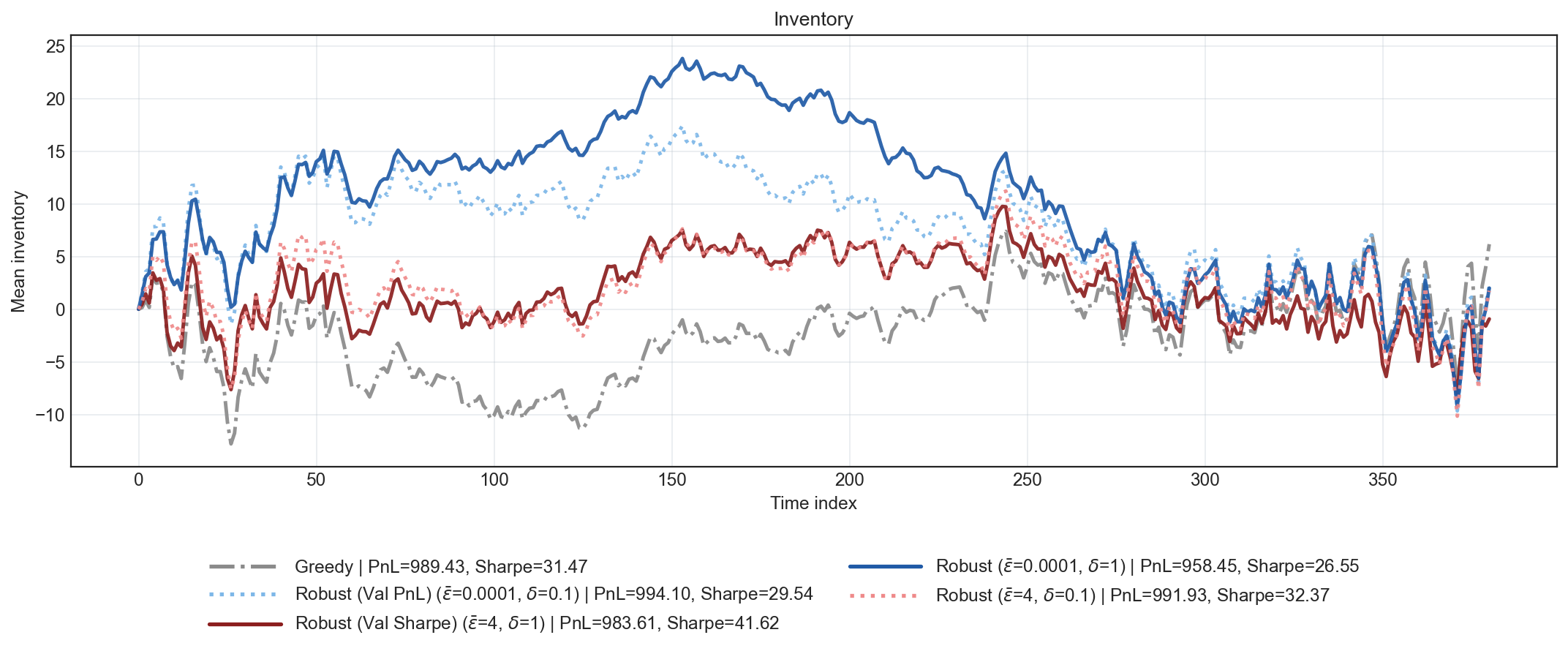}
            \caption{Liquidity dry-out, Inventory path}
        \end{subfigure}\hfill
        \begin{subfigure}{0.5\linewidth}
            \includegraphics[width=\linewidth]{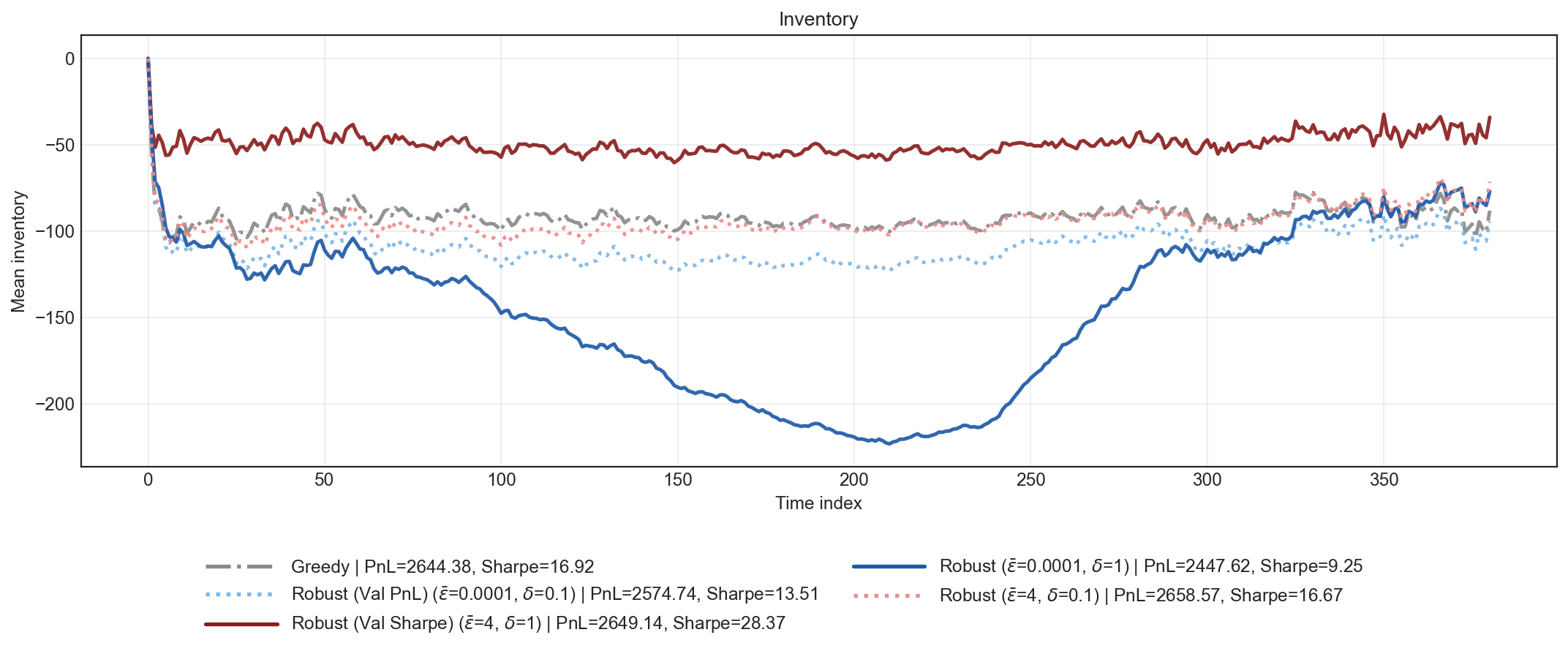}
            \caption{Buy arrival imbalance, Inventory path}
        \end{subfigure}

        \vspace{0.0em}
        
        \begin{subfigure}{0.5\linewidth}
            \includegraphics[width=\linewidth]{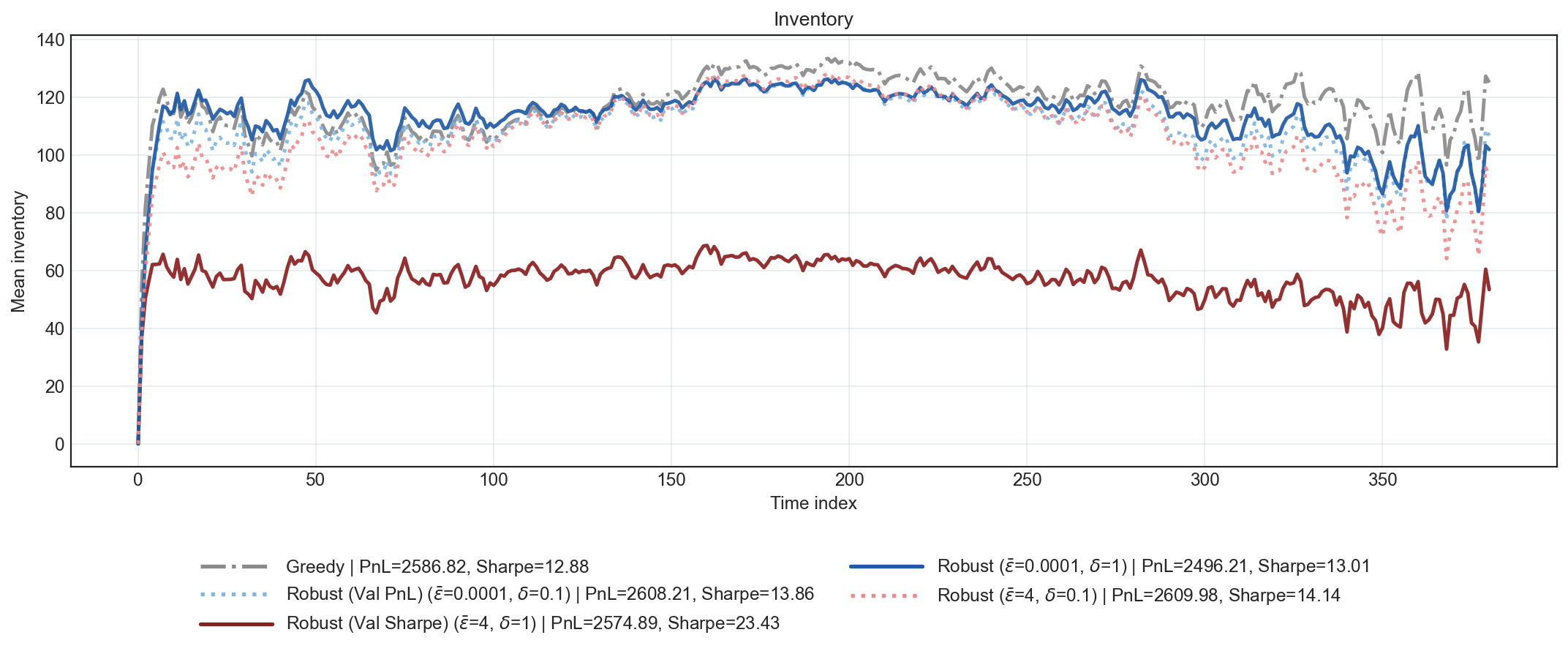}
            \caption{Sell arrival imbalance, Inventory path}
        \end{subfigure}\hfill
        \begin{subfigure}{0.5\linewidth}
            \includegraphics[width=\linewidth]{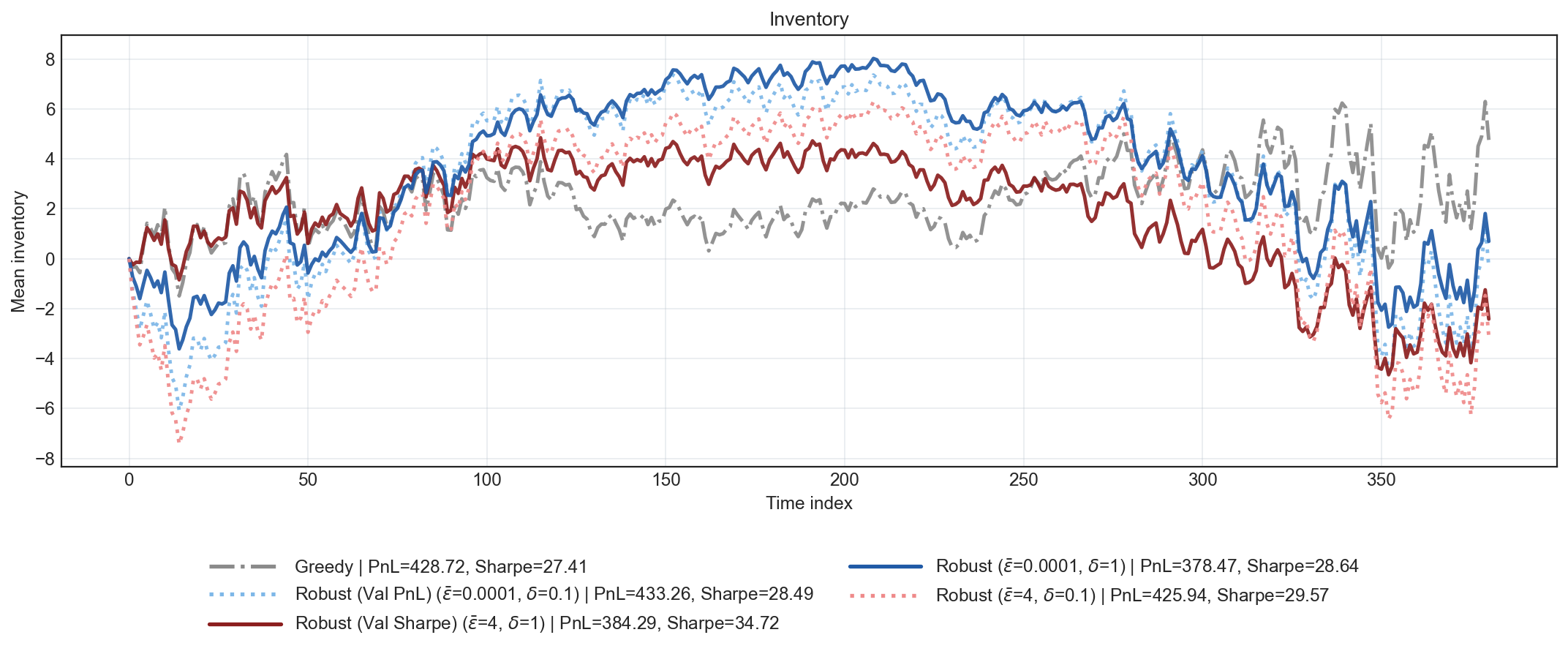}
            \caption{Fill stress, Inventory path}
        \end{subfigure}
    \end{minipage}}
    \caption{Intraday mean inventory paths of the robust and greedy agents in the six simulation scenarios. The panels show that the robust policy generally stabilizes inventory more effectively, with the largest reductions in inventory exposure appearing in the stressed scenarios.}
    \label{fig:app_sim_inventory}
\end{figure}

\section[Additional Numerical Results]{Additional Empirical Results on Real Data}\label{appendix:empirical_results}
\subsection{Empirical Analysis Design and Full Summary Statistics}\label{sec:full-summary}
For empirical analysis, we test along two dimensions \textendash{} liquidity (high vs.\ low) and return volatility (high vs.\ low) \textendash{} giving a $2\times 2$ design (Table~\ref{tab:universe}). In the empirical analysis below we report the high-liquidity cell, represented by AAPL (low volatility) and TSLA (high volatility), and the low-liquidity cell, presented by MKC (low volatility) and TWLO (high volatility).
The summary statistics in Table~\ref{tab:state-summary} confirm the $2\times 2$ classification.
Liquidity is captured by trade-arrival rates: AAPL and TSLA average 43--60 and 23--35 buy-side arrivals per interval, respectively, whereas MKC and TWLO average only 2.6 and 5.8--6.0, placing the latter two firmly in the low-liquidity tier.
Return volatility is measured by the standard deviation of mid-price returns: AAPL and MKC record standard deviations of $5.6$--$9.9\times10^{-4}$, roughly half those of TSLA and TWLO at $10.7$--$17.2\times10^{-4}$, confirming the low- and high-volatility groupings respectively.
For each stock, we adopt a train\textendash validation\textendash test protocol: hyperparameters are selected on validation data, and all reported numbers are out-of-sample. Strategies are compared under identical simulation settings and seeds, using mean P\&L, Sharpe ratio, maximum drawdown, and inventory volatility.

\begin{table}[H]
    \centering
    \caption{Stock universe: $2\times 2$ design over liquidity and return volatility.}
    \label{tab:universe}
    \begin{tabular}{l|cc}
        \hline
        \hline
        & \textbf{Low volatility} & \textbf{High volatility} \\
        \hline
        \textbf{High liquidity} & AAPL & TSLA \\
        \textbf{Low liquidity}  & MKC  & TWLO \\
        \hline
        \hline
    \end{tabular}
\end{table}

\scriptsize
\setlength{\tabcolsep}{3pt}
\begin{longtable}{lllrrrrrrrr}
\caption{Summary statistics of key state variables across stocks and sample periods.}
\label{tab:state-summary}\\
\toprule
Feature & Stock & Year & Count & Mean & Std & Min & 25\% & 50\% & 75\% & Max \\
\midrule
\endfirsthead

\multicolumn{11}{c}{\tablename\ \thetable{} -- continued from previous page} \\
\toprule
Feature & Stock & Year & Count & Mean & Std & Min & 25\% & 50\% & 75\% & Max \\
\midrule
\endhead

\midrule
\multicolumn{11}{r}{continued on next page}
\endfoot

\bottomrule
\endlastfoot

\multicolumn{11}{l}{\textbf{AAPL 2019}} \\
buy\_count & AAPL & 2019 & 93600 & 42.9945 & 45.6852 & 0.0000 & 18.0000 & 30.0000 & 51.0000 & 1019.0000 \\
sell\_count & AAPL & 2019 & 93600 & 45.8349 & 46.8980 & 0.0000 & 20.0000 & 33.0000 & 56.0000 & 1418.0000 \\
returns & AAPL & 2019 & 93600 & 0.0000 & 0.0006 & -0.0196 & -0.0002 & 0.0000 & 0.0003 & 0.0153 \\
\addlinespace

\multicolumn{11}{l}{\textbf{AAPL 2020}} \\
buy\_count & AAPL & 2020 & 93600 & 56.3747 & 59.9237 & 0.0000 & 21.0000 & 39.0000 & 70.0000 & 1377.0000 \\
sell\_count & AAPL & 2020 & 93600 & 59.6547 & 62.2363 & 0.0000 & 24.0000 & 41.0000 & 73.0000 & 2118.0000 \\
returns & AAPL & 2020 & 93600 & 0.0000 & 0.0008 & -0.0196 & -0.0003 & 0.0000 & 0.0003 & 0.0277 \\
\addlinespace

\multicolumn{11}{l}{\textbf{MKC 2019}} \\
buy\_count & MKC & 2019 & 93600 & 2.6485 & 3.8147 & 0.0000 & 0.0000 & 1.0000 & 4.0000 & 95.0000 \\
sell\_count & MKC & 2019 & 93600 & 2.5973 & 3.5593 & 0.0000 & 0.0000 & 1.0000 & 4.0000 & 77.0000 \\
returns & MKC & 2019 & 93600 & 0.0000 & 0.0006 & -0.0118 & -0.0002 & 0.0000 & 0.0002 & 0.0132 \\
\addlinespace

\multicolumn{11}{l}{\textbf{MKC 2020}} \\
buy\_count & MKC & 2020 & 93600 & 2.5874 & 3.8922 & 0.0000 & 0.0000 & 1.0000 & 4.0000 & 79.0000 \\
sell\_count & MKC & 2020 & 93600 & 2.5266 & 3.7156 & 0.0000 & 0.0000 & 1.0000 & 4.0000 & 154.0000 \\
returns & MKC & 2020 & 93600 & 0.0000 & 0.0010 & -0.0180 & -0.0002 & 0.0000 & 0.0002 & 0.0215 \\
\addlinespace

\multicolumn{11}{l}{\textbf{TESLA 2019}} \\
buy\_count & TESLA & 2019 & 93600 & 23.3313 & 25.0334 & 0.0000 & 8.0000 & 16.0000 & 30.0000 & 695.0000 \\
sell\_count & TESLA & 2019 & 93600 & 23.5483 & 24.3968 & 0.0000 & 9.0000 & 17.0000 & 30.0000 & 472.0000 \\
returns & TESLA & 2019 & 93600 & 0.0000 & 0.0011 & -0.0153 & -0.0005 & 0.0000 & 0.0005 & 0.0159 \\
\addlinespace

\multicolumn{11}{l}{\textbf{TESLA 2020}} \\
buy\_count & TESLA & 2020 & 93600 & 33.7404 & 39.8158 & 0.0000 & 11.0000 & 22.0000 & 42.0000 & 746.0000 \\
sell\_count & TESLA & 2020 & 93600 & 34.5712 & 40.1440 & 0.0000 & 11.0000 & 23.0000 & 43.0000 & 1865.0000 \\
returns & TESLA & 2020 & 93600 & 0.0000 & 0.0017 & -0.0323 & -0.0006 & 0.0000 & 0.0006 & 0.0296 \\
\addlinespace

\multicolumn{11}{l}{\textbf{TWLO 2019}} \\
buy\_count & TWLO & 2019 & 93600 & 5.8326 & 7.1609 & 0.0000 & 1.0000 & 4.0000 & 8.0000 & 232.0000 \\
sell\_count & TWLO & 2019 & 93600 & 6.0630 & 7.3630 & 0.0000 & 1.0000 & 4.0000 & 8.0000 & 185.0000 \\
returns & TWLO & 2019 & 93600 & 0.0000 & 0.0013 & -0.0199 & -0.0005 & 0.0000 & 0.0005 & 0.0174 \\
\addlinespace

\multicolumn{11}{l}{\textbf{TWLO 2020}} \\
buy\_count & TWLO & 2020 & 93600 & 6.0091 & 7.7962 & 0.0000 & 1.0000 & 4.0000 & 8.0000 & 173.0000 \\
sell\_count & TWLO & 2020 & 93600 & 6.1801 & 7.8214 & 0.0000 & 1.0000 & 4.0000 & 8.0000 & 185.0000 \\
returns & TWLO & 2020 & 93600 & 0.0000 & 0.0015 & -0.0265 & -0.0006 & 0.0000 & 0.0006 & 0.0423 \\
\end{longtable}
\normalsize

\begin{minipage}{0.96\linewidth}
\end{minipage}

\bigskip

We evaluate the strategies over two sample designs. In the first design, the data from 2019 are partitioned into 144 trading days for training, 48 trading days for validation, and 48 trading days for testing. In the second design, the training sample again consists of 144 trading days from 2019, while the validation sample comprises the 48 trading days up to 28 February 2020; the corresponding test sample is the subsequent 48 trading days beginning on 1 March 2020. This second test window spans the onset of the COVID-19 market disruption and therefore provides a natural stress scenario for assessing the value of distributional robustness. We benchmark the robust reinforcement learning strategies against six strategies: the Avellaneda--Stoikov benchmark (AS); a random baseline, in which bid--ask spreads and quote sizes are drawn randomly around their admissible upper thresholds, with performance averaged over ten resampled action paths; a fixed quoting rule, under which the agent always posts the midpoint between the minimum and maximum admissible spreads and quantities; the learned non-robust agent; and two robust agents, selected on the validation set according to either mean P\&L or Sharpe ratio.

The figures ~\ref{fig:full_state_dist_aapl}-\ref{fig:full_state_dist_twlo} report, for each stock\textendash period, the train, validation, and test distributions of the state variables used in the empirical analysis. They show that the 2020 sample exhibits a markedly larger cross-split distributional shift than the corresponding 2019 sample, reflecting the more pronounced market instability in 2020.

%
%
\clearpage
\begin{figure}[H]
    \centering
    \begin{subfigure}{\linewidth}
    \centering
    \makebox[\textwidth][c]{\begin{minipage}{1.16\textwidth}\centering
        \setlength{\tabcolsep}{0.5pt}
        \begin{tabular}{cccccc}
            \includegraphics[width=0.16\textwidth]{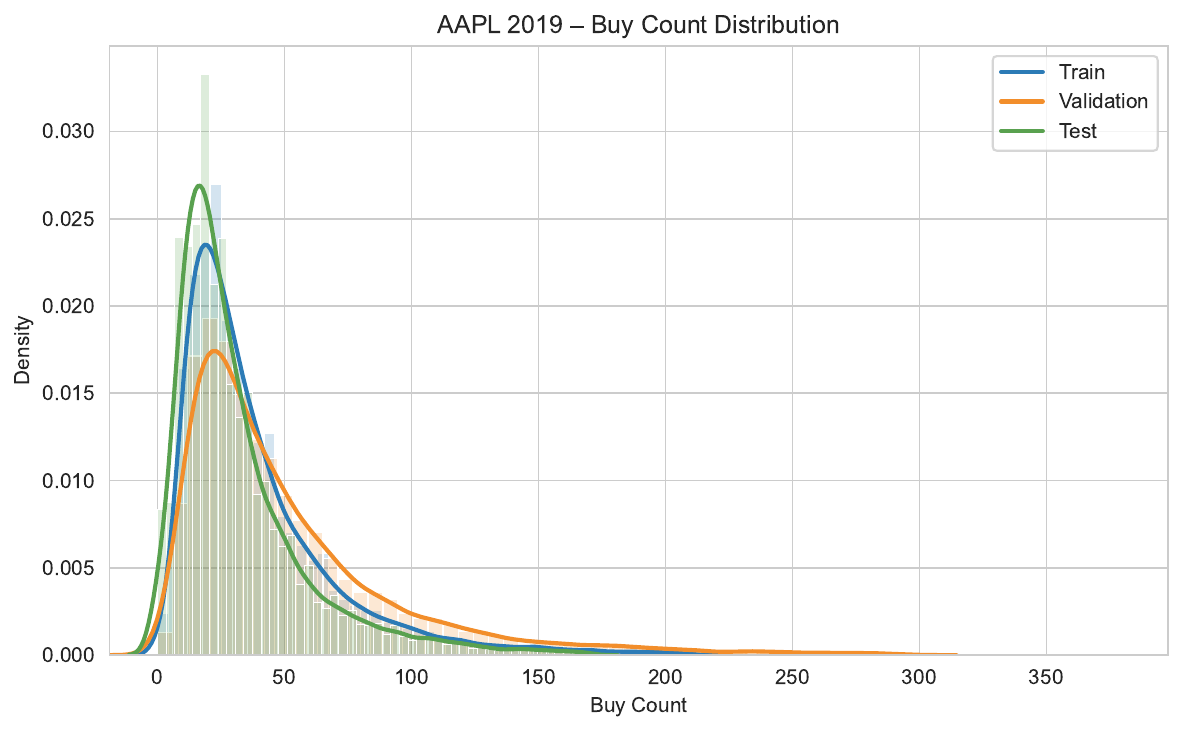} &
            \includegraphics[width=0.16\textwidth]{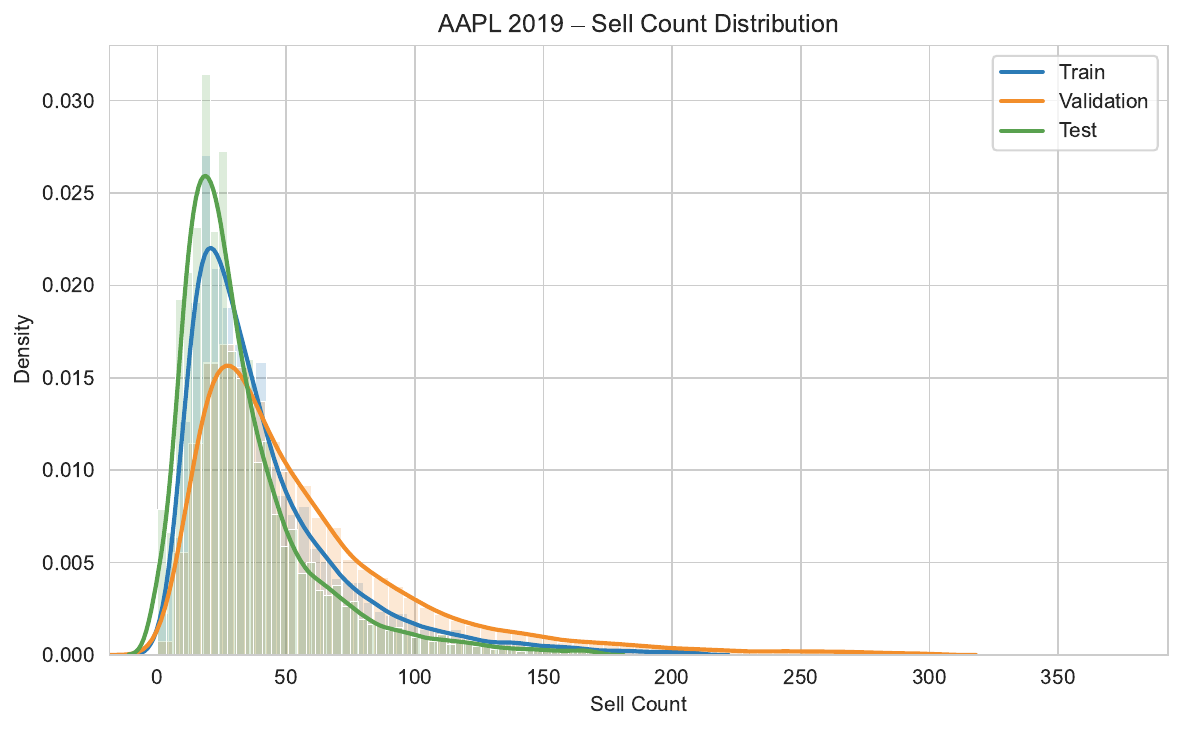} &
            \includegraphics[width=0.16\textwidth]{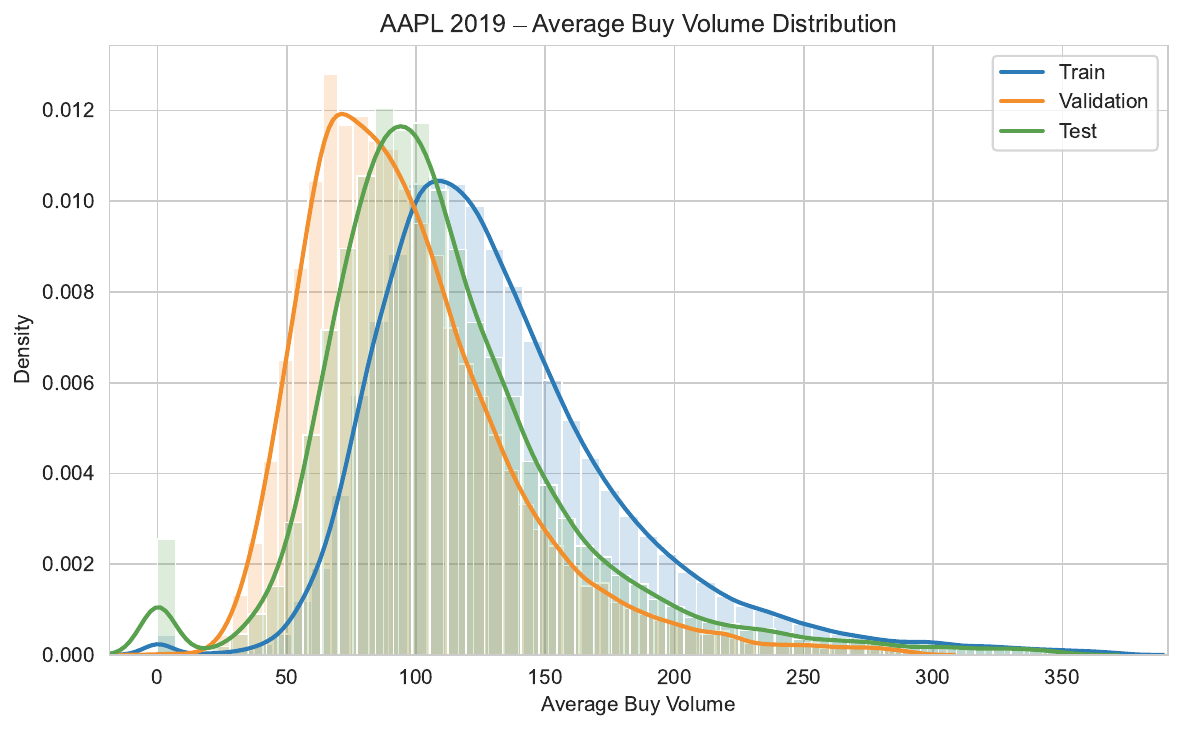} &
            \includegraphics[width=0.16\textwidth]{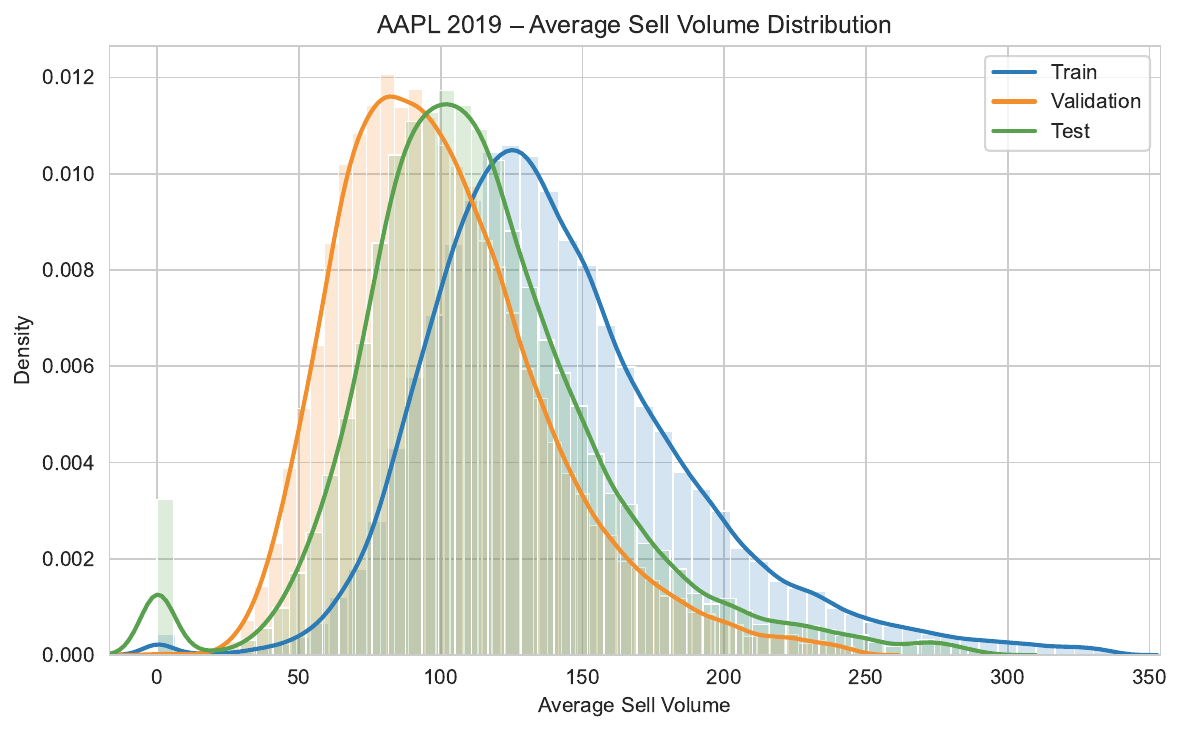} &
            \includegraphics[width=0.16\textwidth]{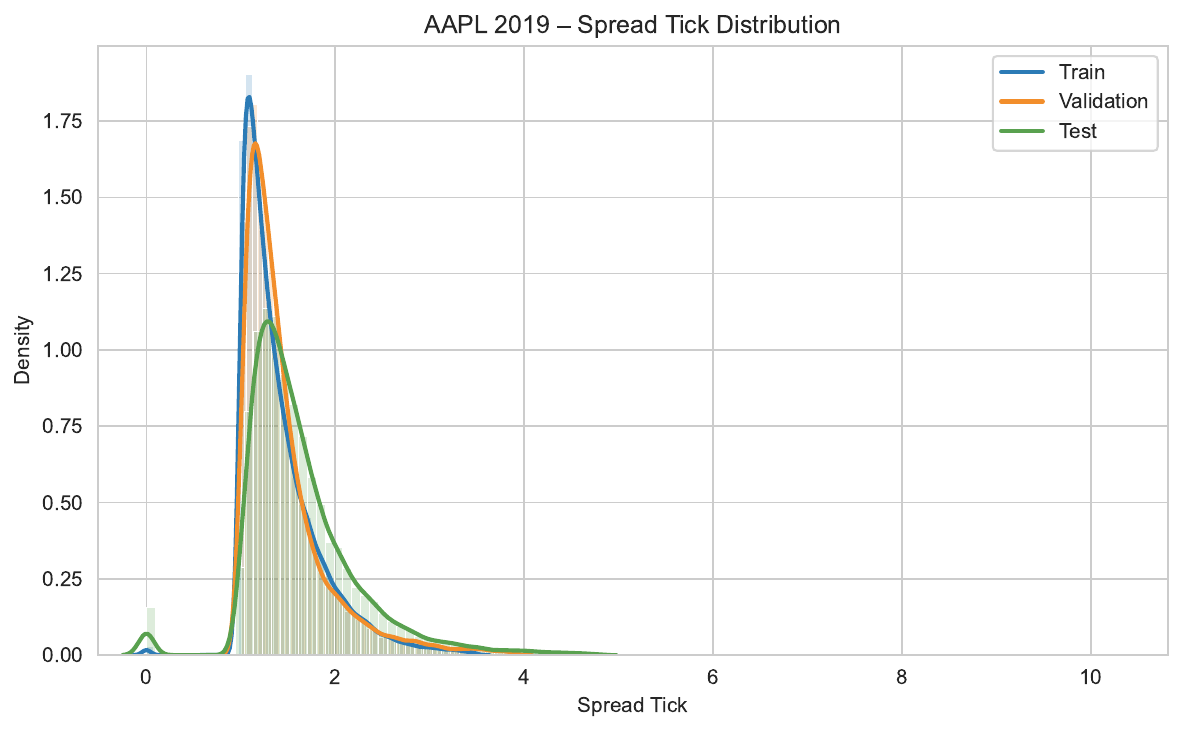} &
            \includegraphics[width=0.16\textwidth]{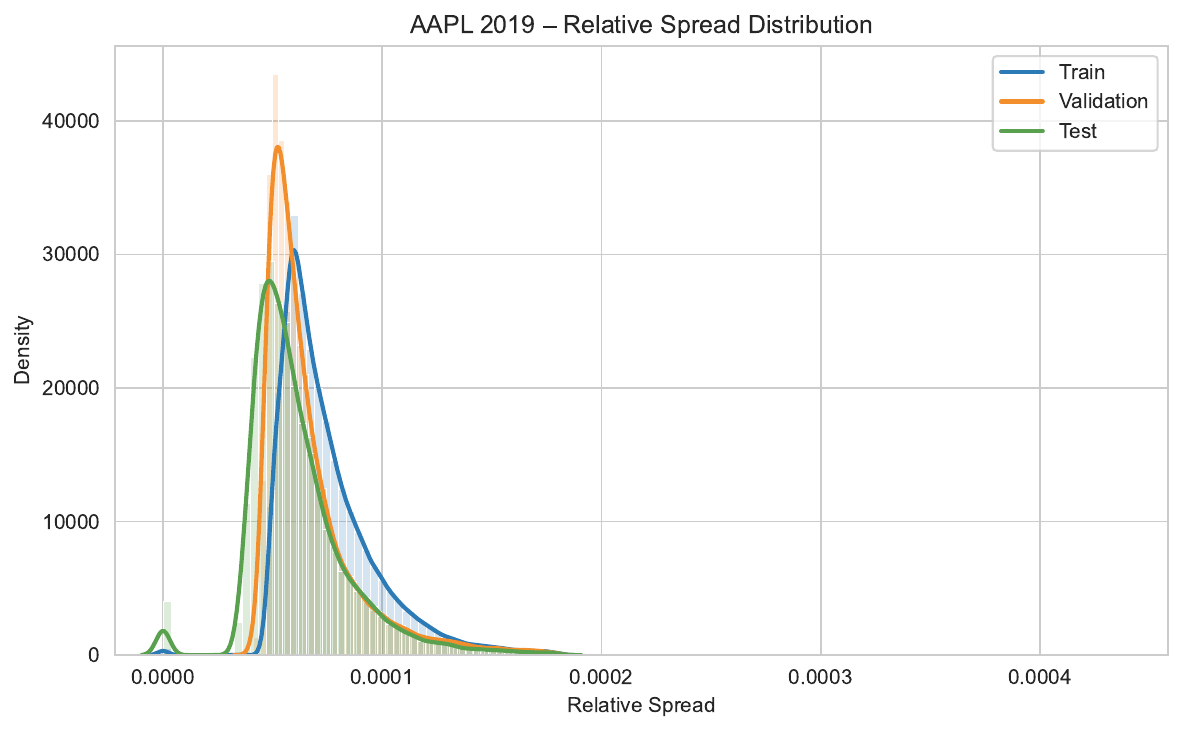} \\
            \scriptsize Buy count &
            \scriptsize Sell count &
            \scriptsize Average buy size &
            \scriptsize Average sell size &
            \scriptsize Spread (ticks) &
            \scriptsize Relative spread \\
            \includegraphics[width=0.16\textwidth]{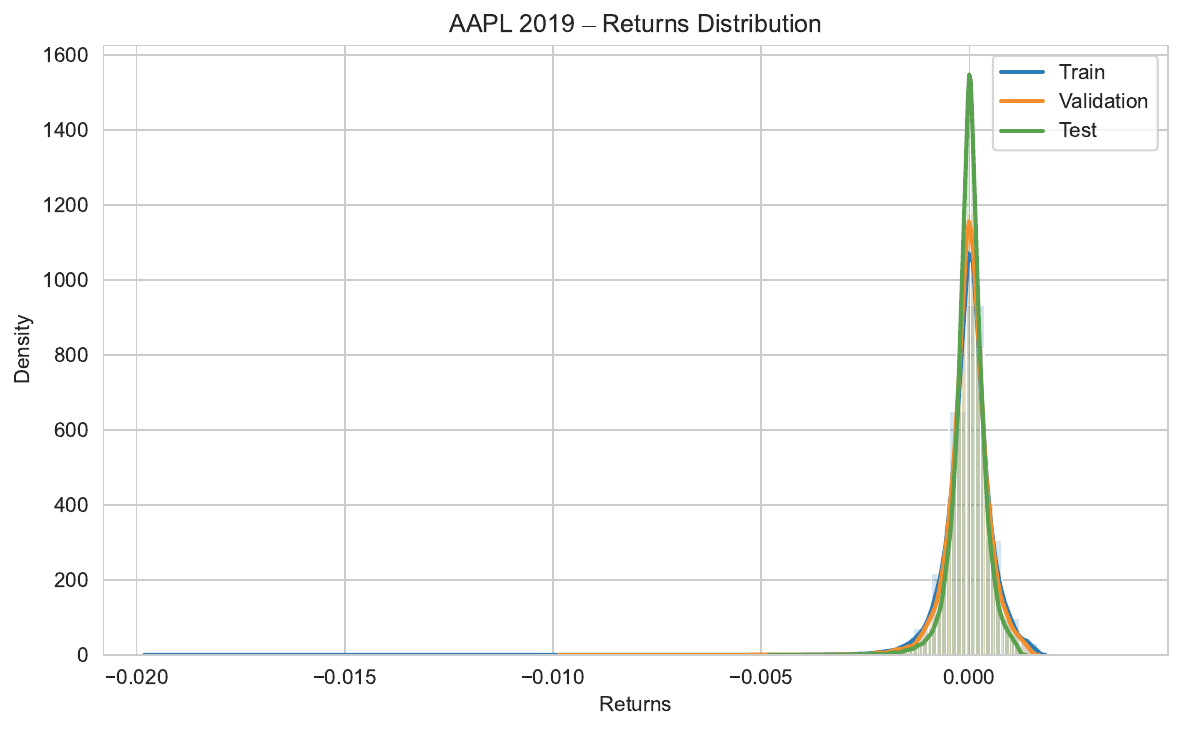} &
            \includegraphics[width=0.16\textwidth]{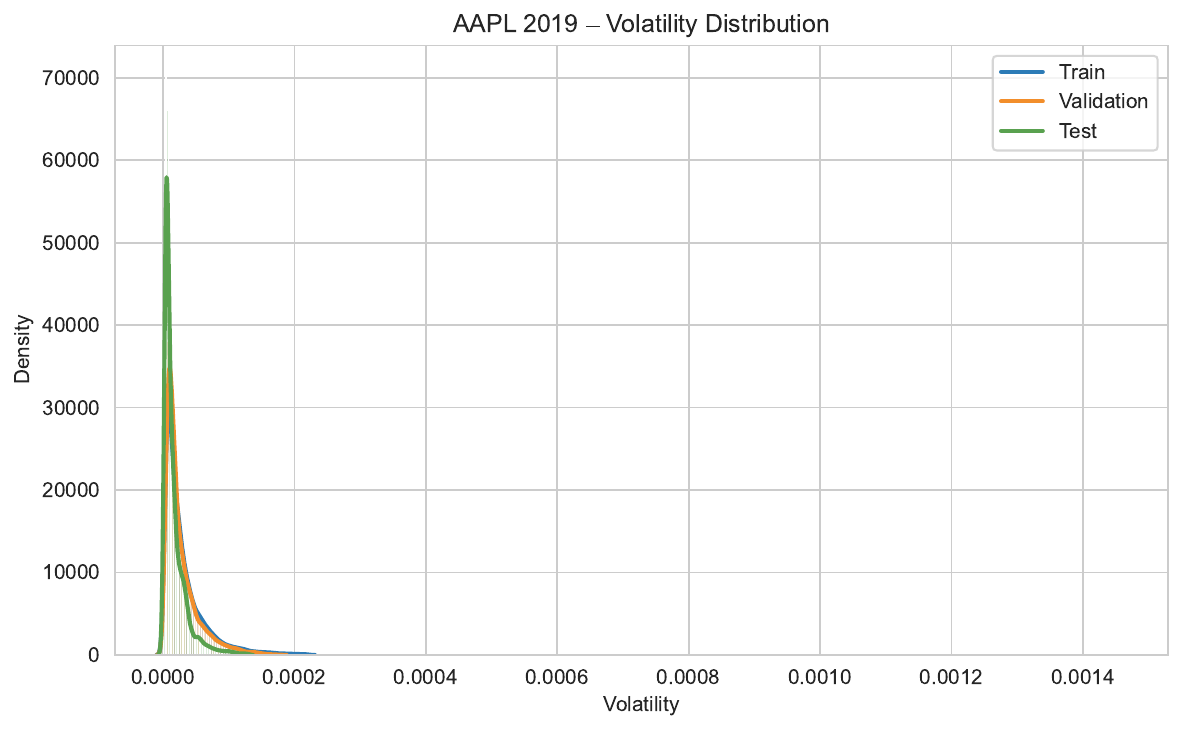} &
            \includegraphics[width=0.16\textwidth]{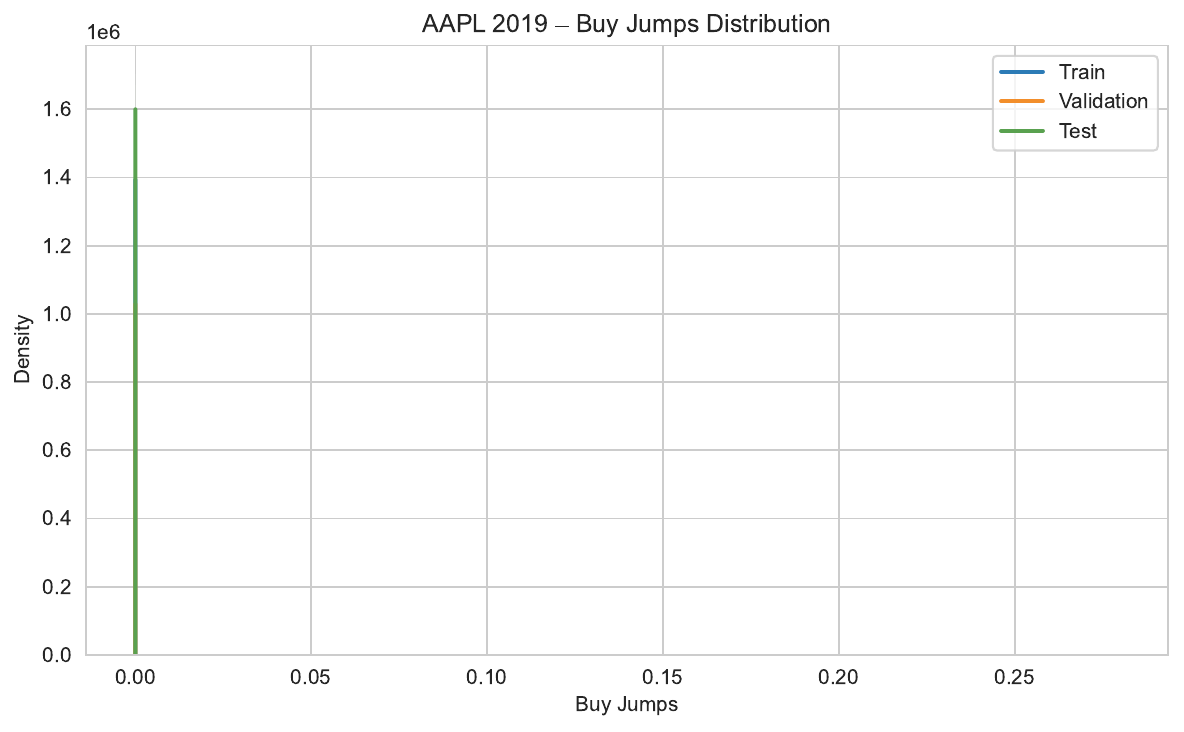} &
            \includegraphics[width=0.16\textwidth]{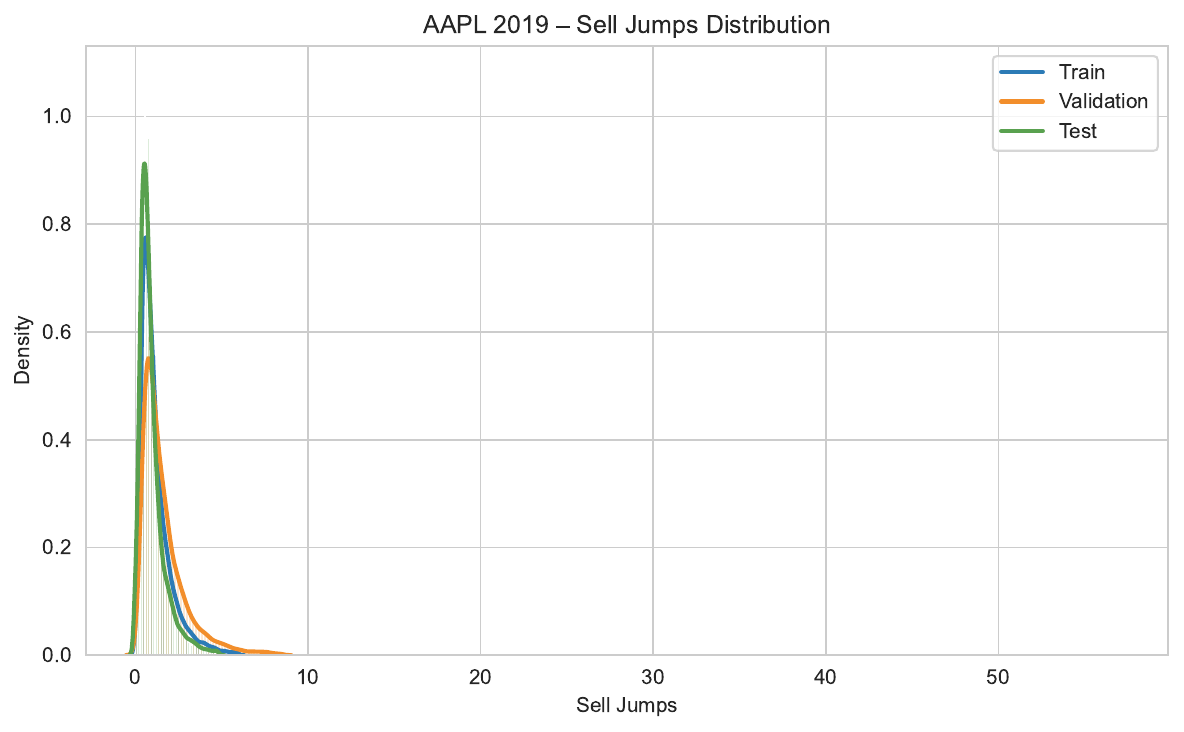} &
            \includegraphics[width=0.16\textwidth]{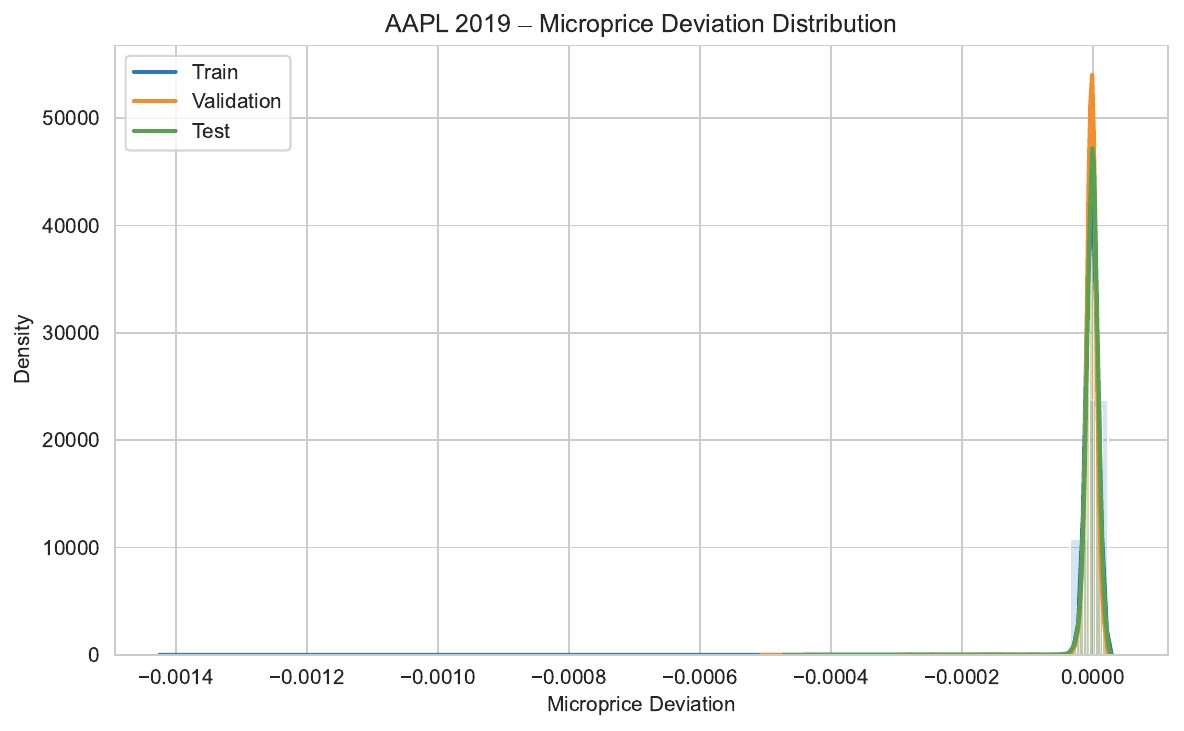} &
            \includegraphics[width=0.16\textwidth]{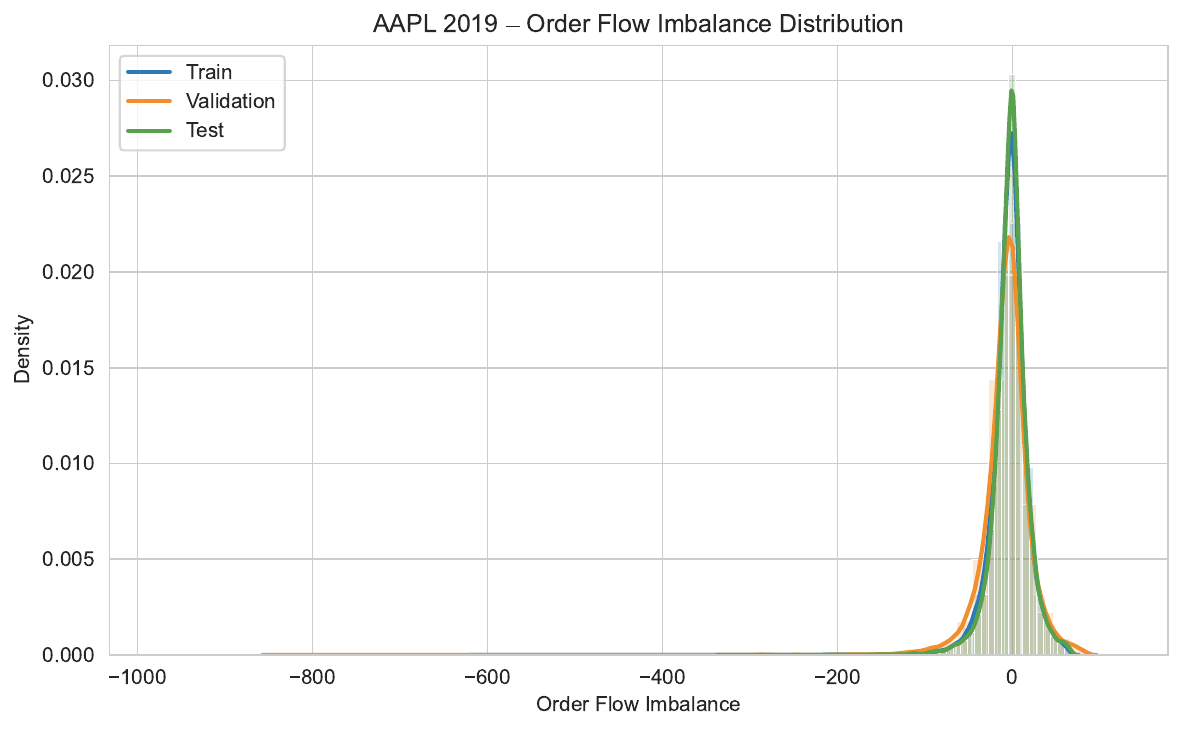} \\
            \scriptsize Returns &
            \scriptsize Volatility &
            \scriptsize Buy-order excitation &
            \scriptsize Sell-order excitation &
            \scriptsize Microprice deviation &
            \scriptsize OFI \\
        \end{tabular}
    \end{minipage}}
    \caption{AAPL, 2019.}
    \end{subfigure}

    \vspace{0.2em}

    \begin{subfigure}{\linewidth}
    \centering
    \makebox[\textwidth][c]{\begin{minipage}{1.16\textwidth}\centering
        \setlength{\tabcolsep}{0.5pt}
        \begin{tabular}{cccccc}
            \includegraphics[width=0.16\textwidth]{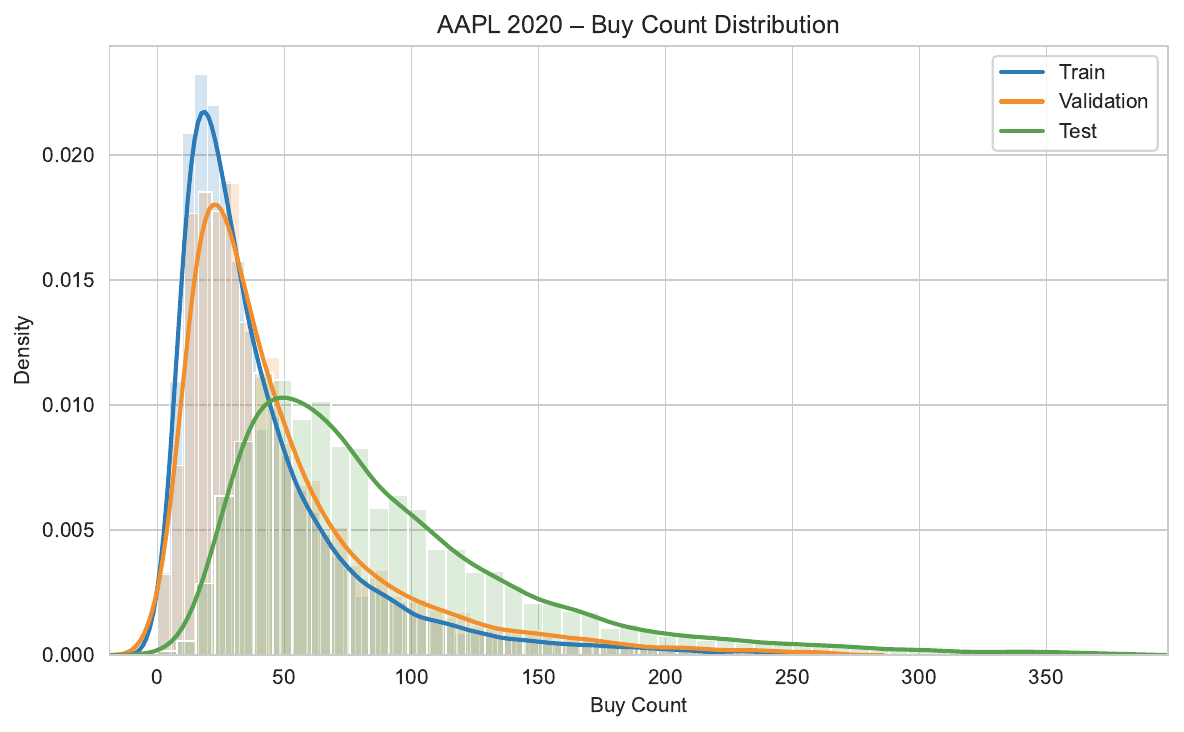} &
            \includegraphics[width=0.16\textwidth]{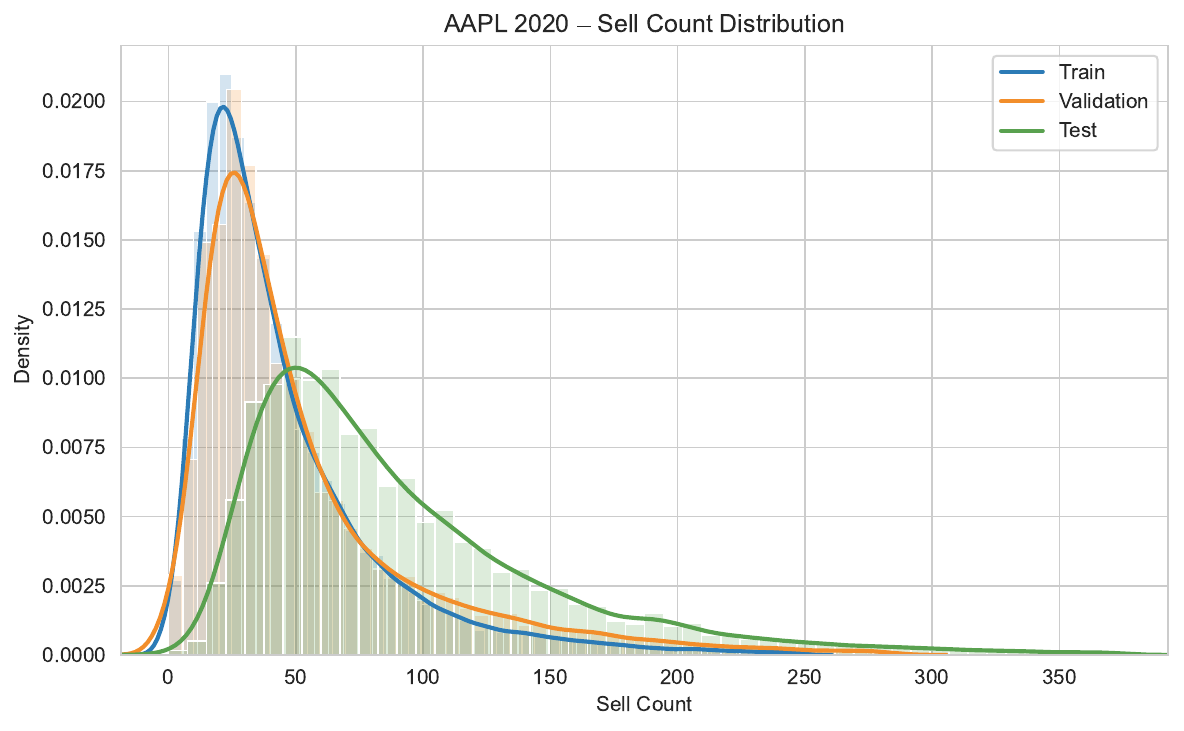} &
            \includegraphics[width=0.16\textwidth]{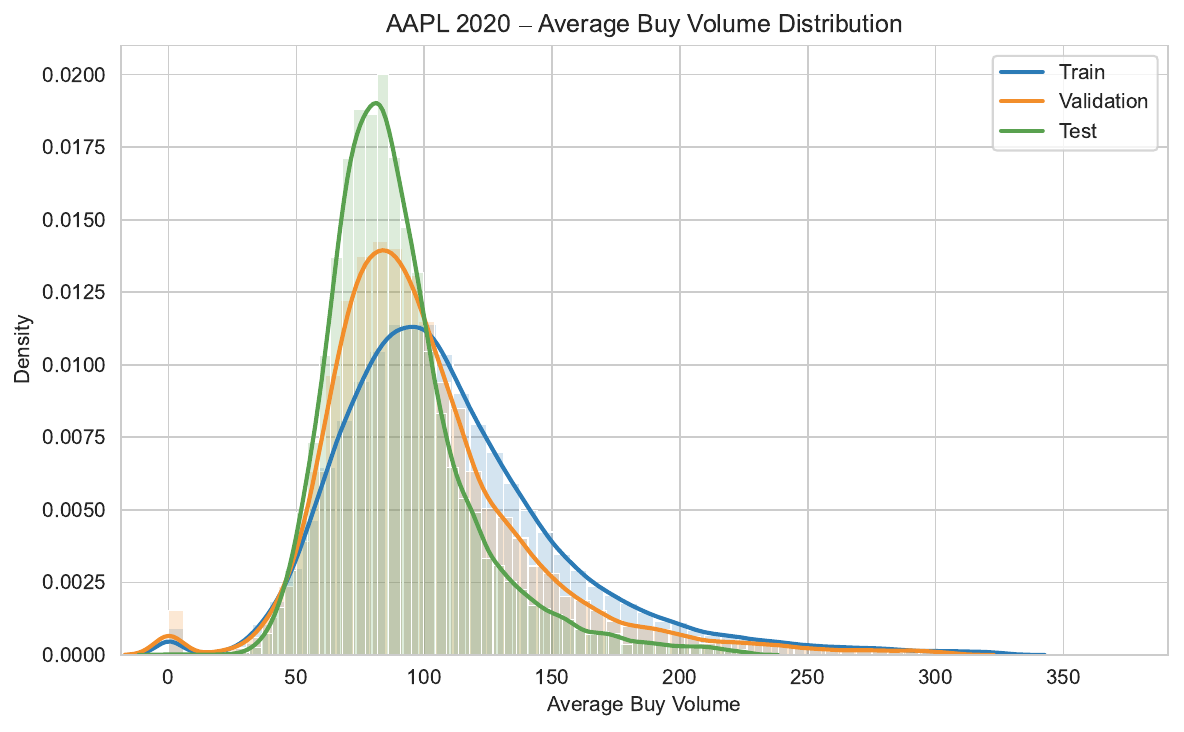} &
            \includegraphics[width=0.16\textwidth]{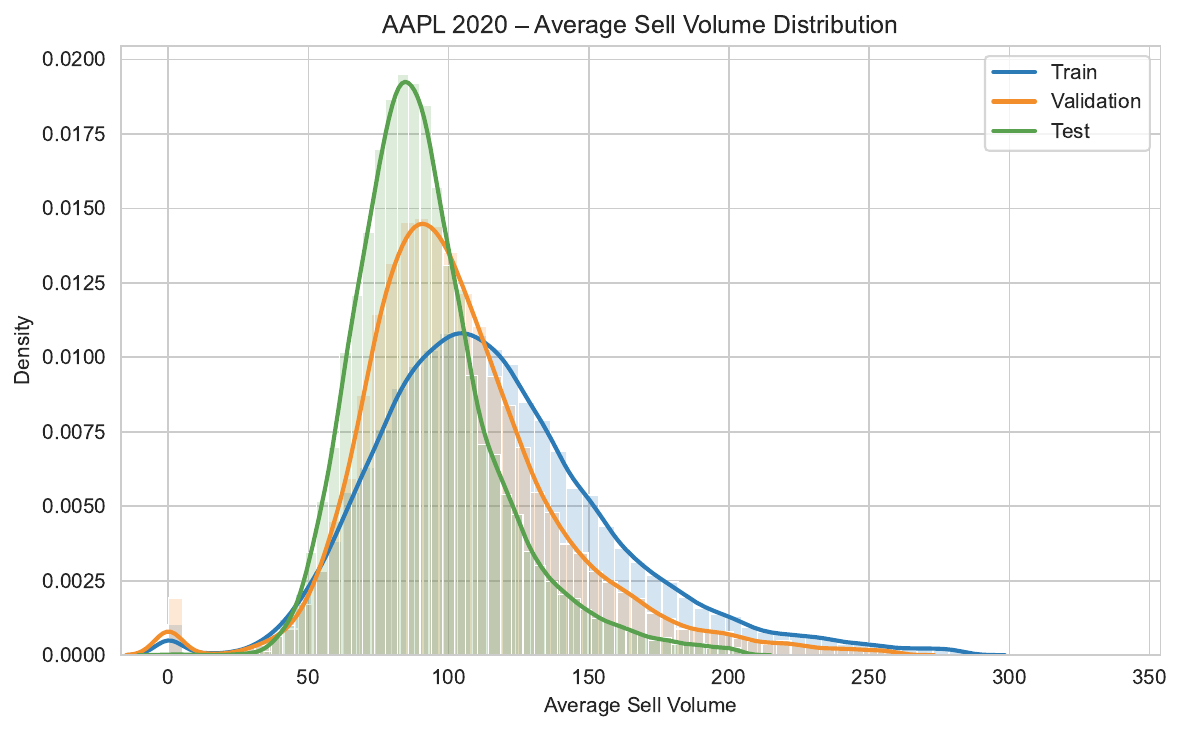} &
            \includegraphics[width=0.16\textwidth]{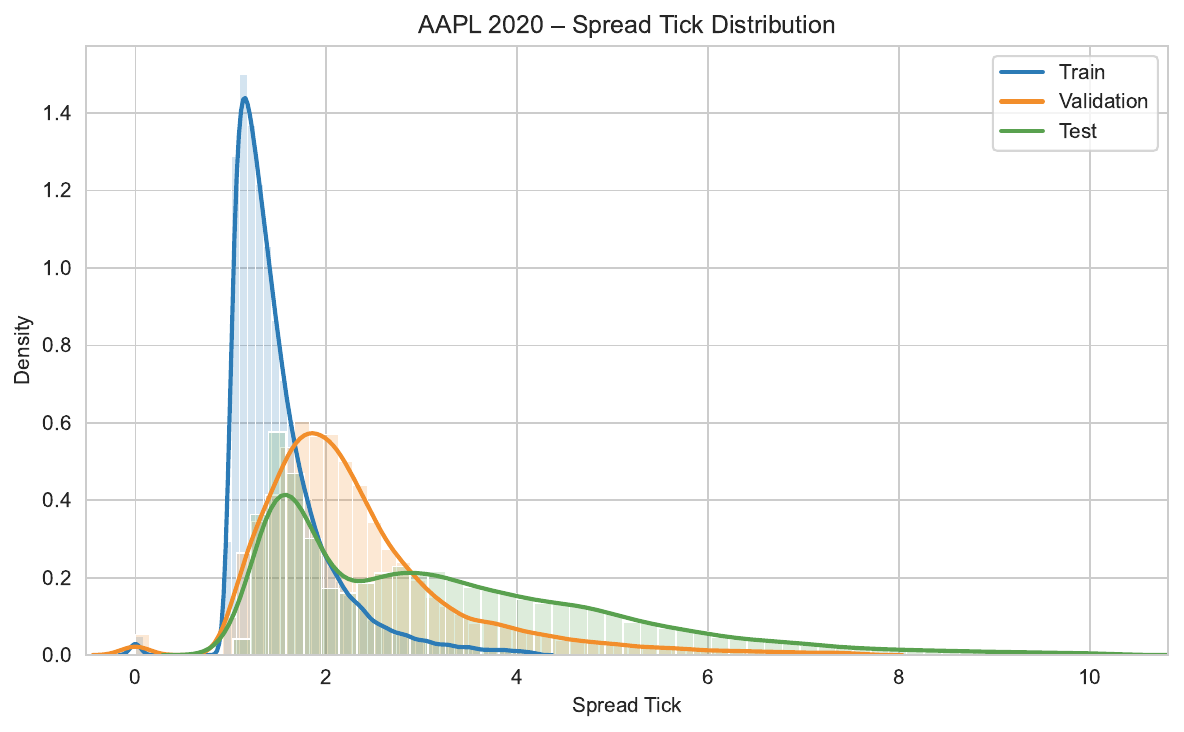} &
            \includegraphics[width=0.16\textwidth]{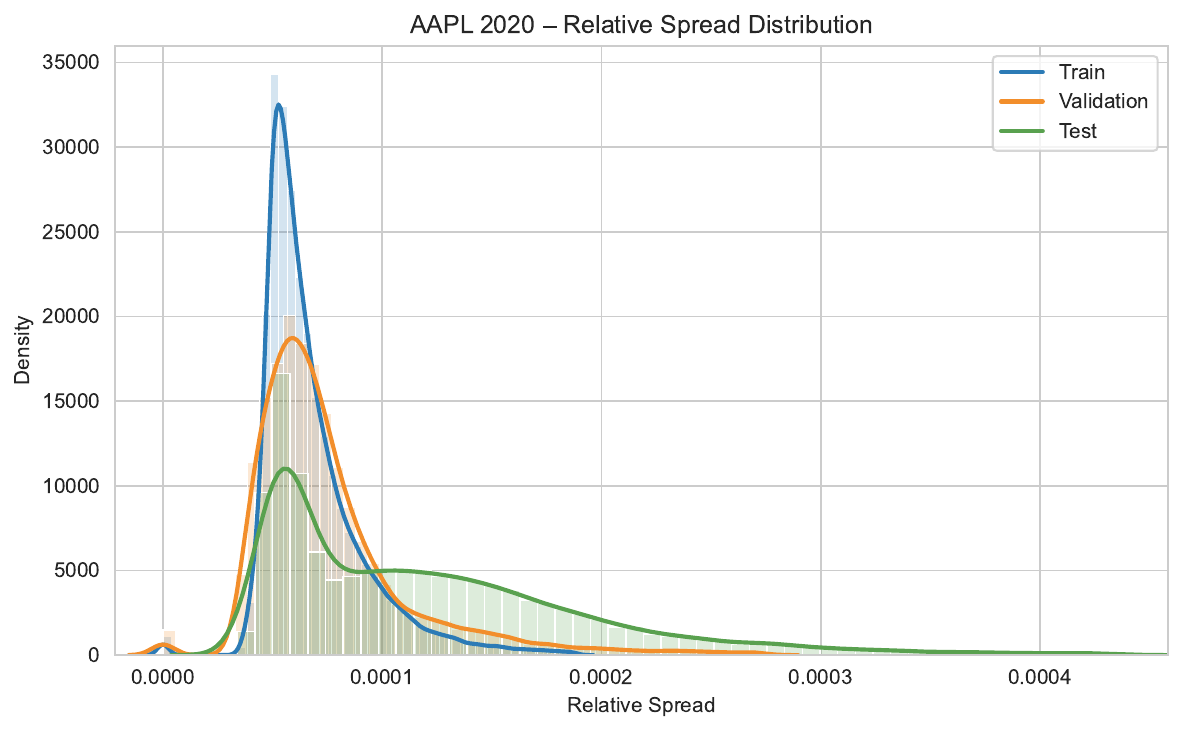} \\
            \scriptsize Buy count &
            \scriptsize Sell count &
            \scriptsize Average buy size &
            \scriptsize Average sell size &
            \scriptsize Spread (ticks) &
            \scriptsize Relative spread \\
            \includegraphics[width=0.16\textwidth]{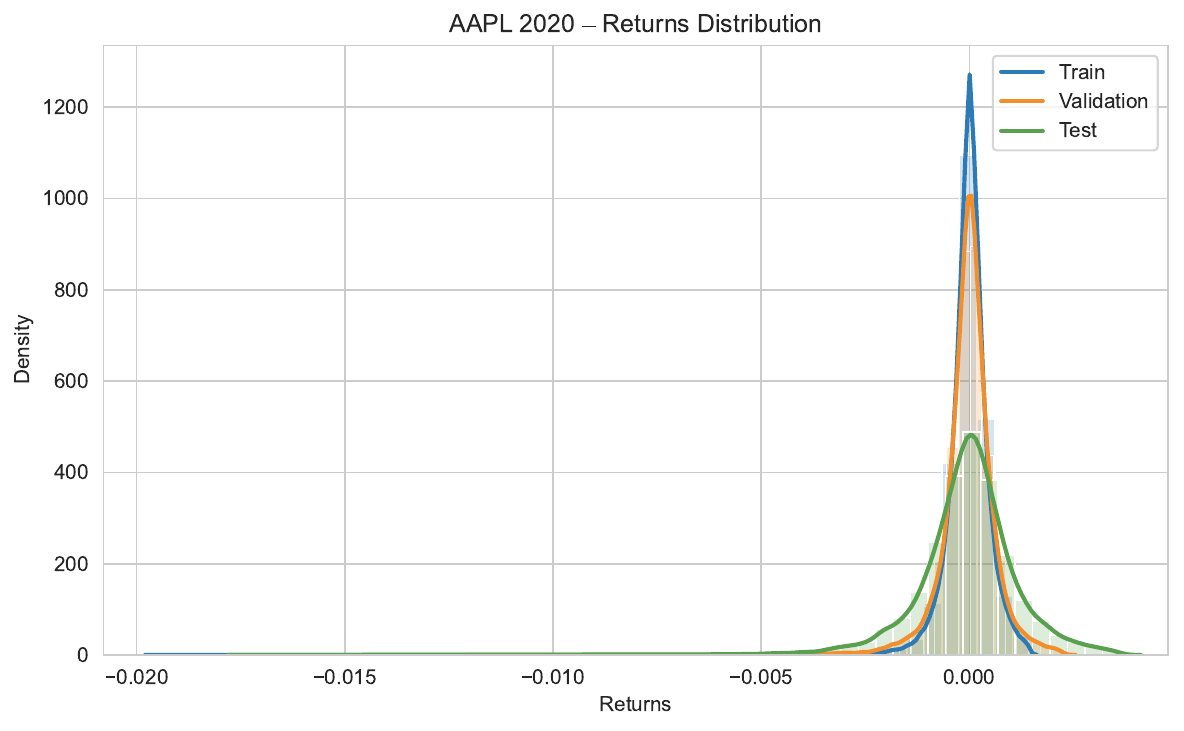} &
            \includegraphics[width=0.16\textwidth]{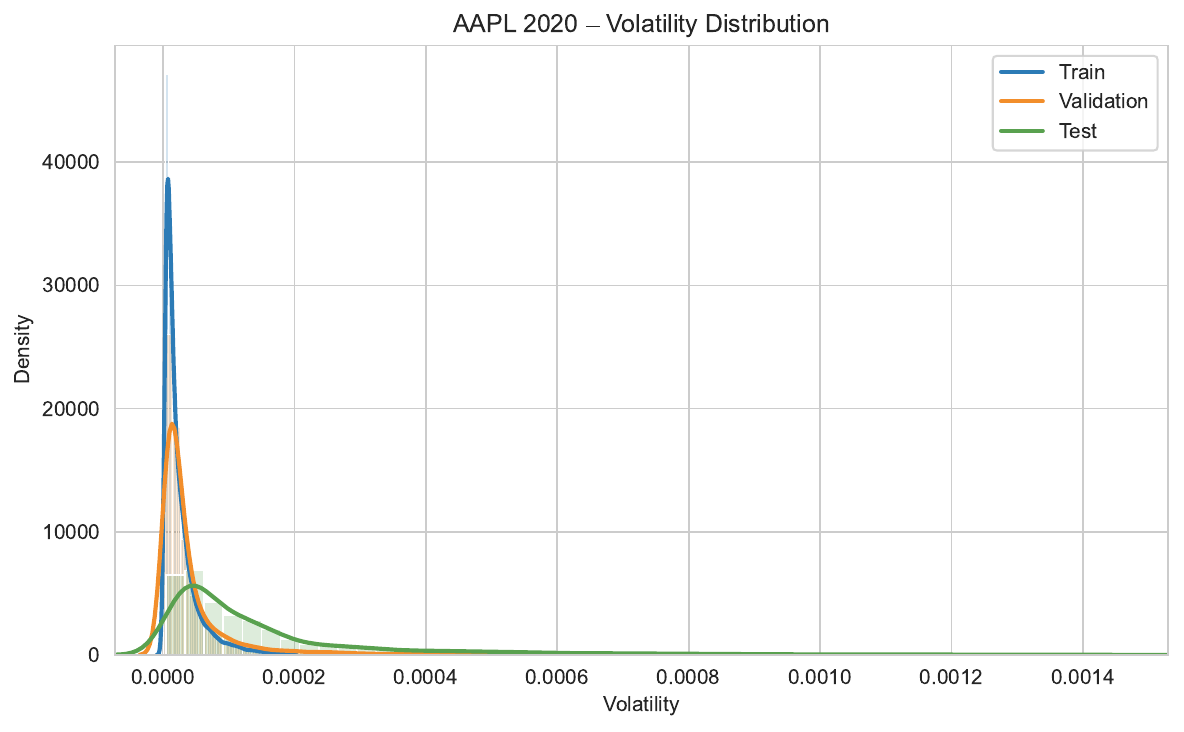} &
            \includegraphics[width=0.16\textwidth]{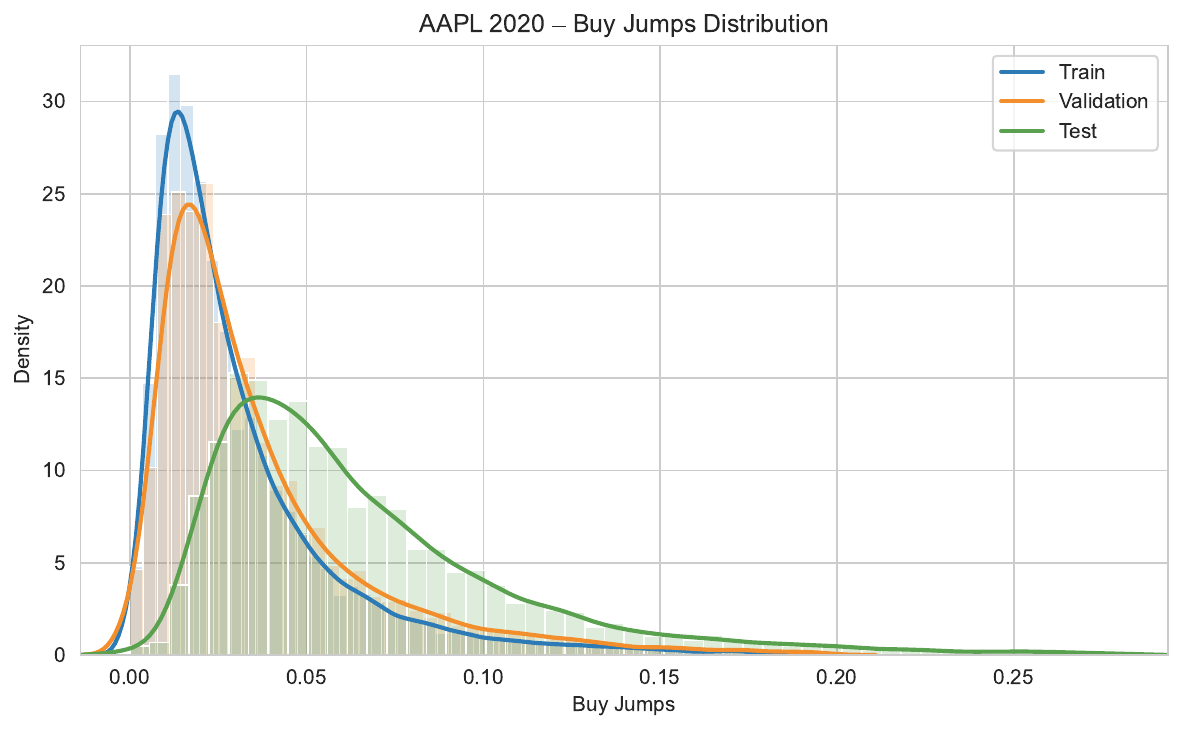} &
            \includegraphics[width=0.16\textwidth]{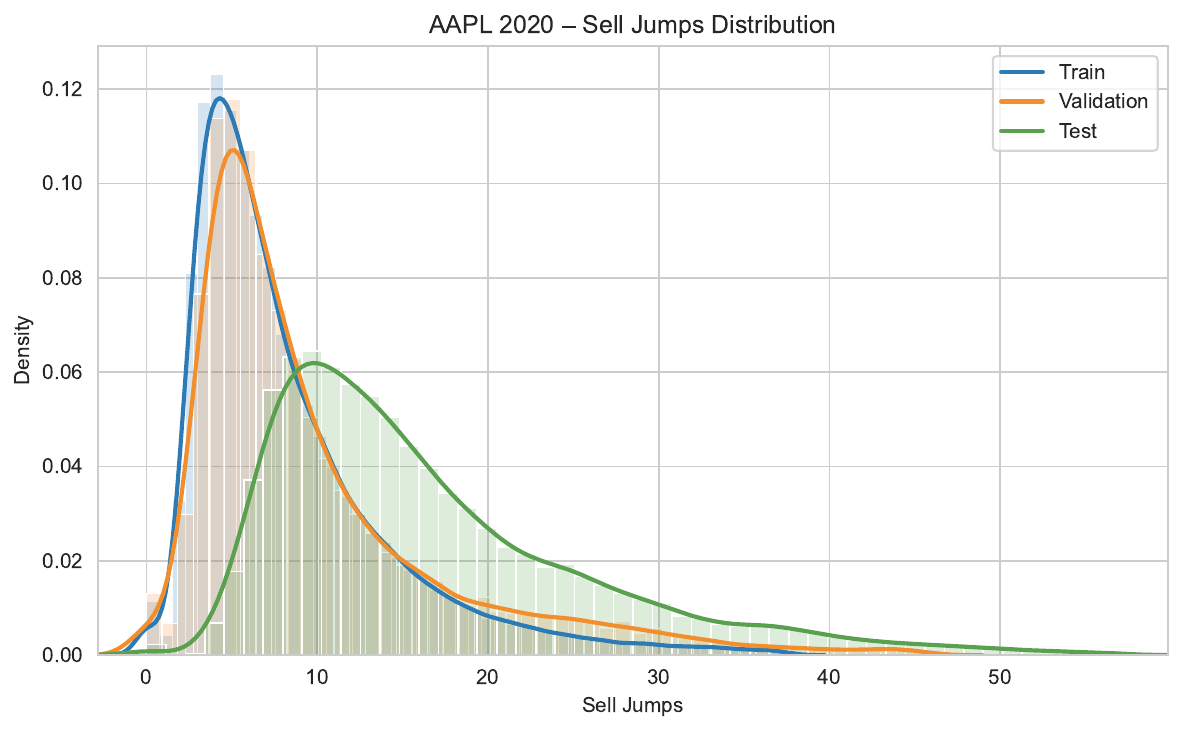} &
            \includegraphics[width=0.16\textwidth]{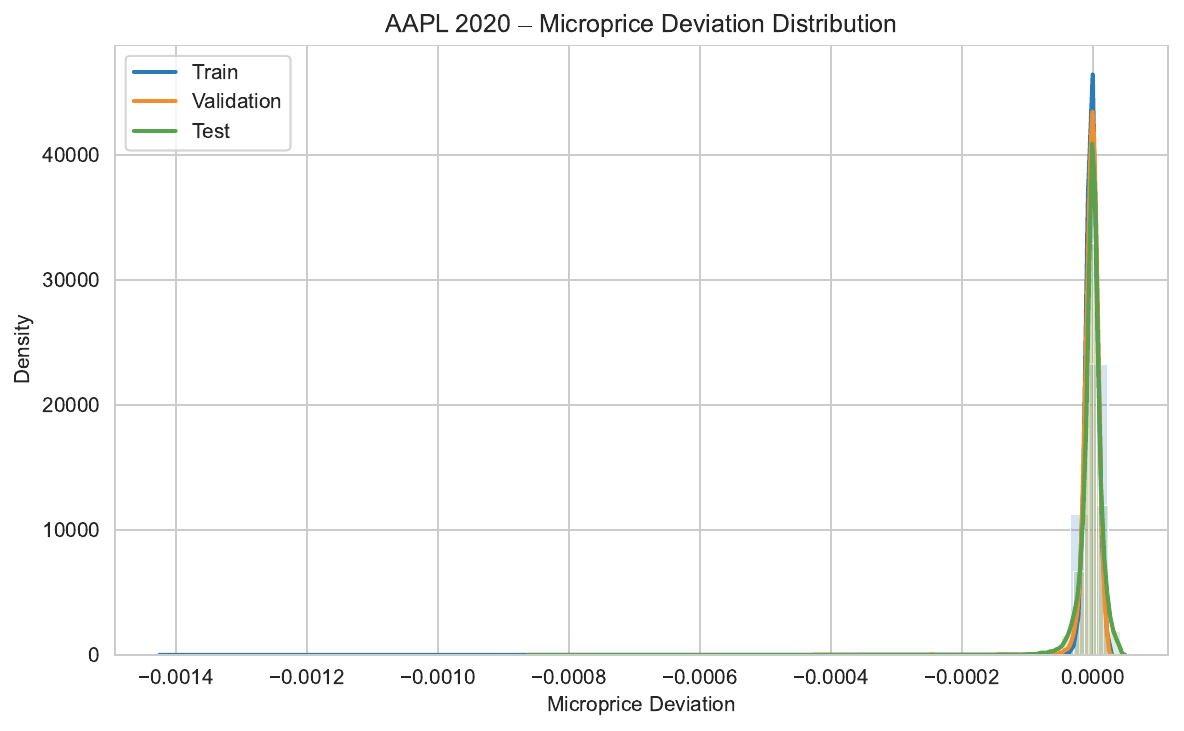} &
            \includegraphics[width=0.16\textwidth]{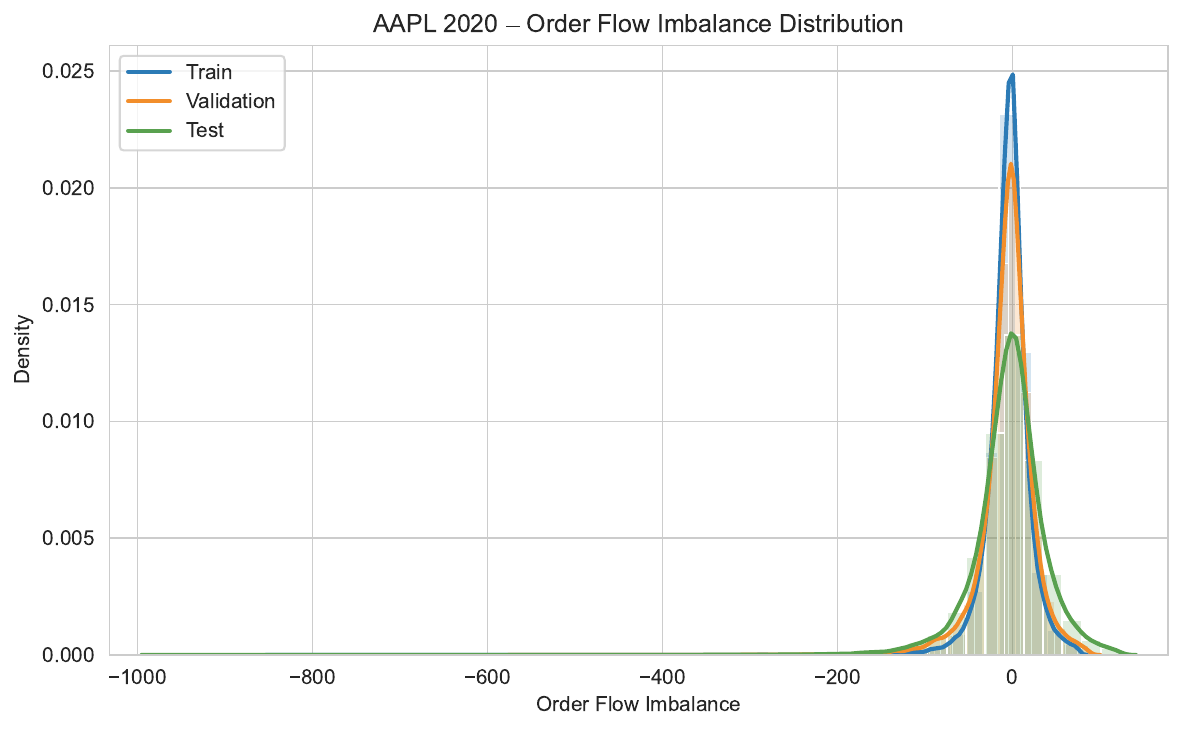} \\
            \scriptsize Returns &
            \scriptsize Volatility &
            \scriptsize Buy-order excitation &
            \scriptsize Sell-order excitation &
            \scriptsize Microprice deviation &
            \scriptsize OFI \\
        \end{tabular}
    \end{minipage}}
    \caption{AAPL, 2020.}
    \end{subfigure}

    \caption{Train, validation, and test distributions of all state variables for AAPL in 2019 and 2020. The comparison makes the stronger cross-split distributional shift in 2020 visually apparent.}
    \label{fig:full_state_dist_aapl}
\end{figure}

\begin{figure}[H]
    \centering
    \begin{subfigure}{\linewidth}
    \centering
    \makebox[\textwidth][c]{\begin{minipage}{1.12\textwidth}\centering
        \setlength{\tabcolsep}{0.5pt}
        \begin{tabular}{cccccc}
            \includegraphics[width=0.16\textwidth]{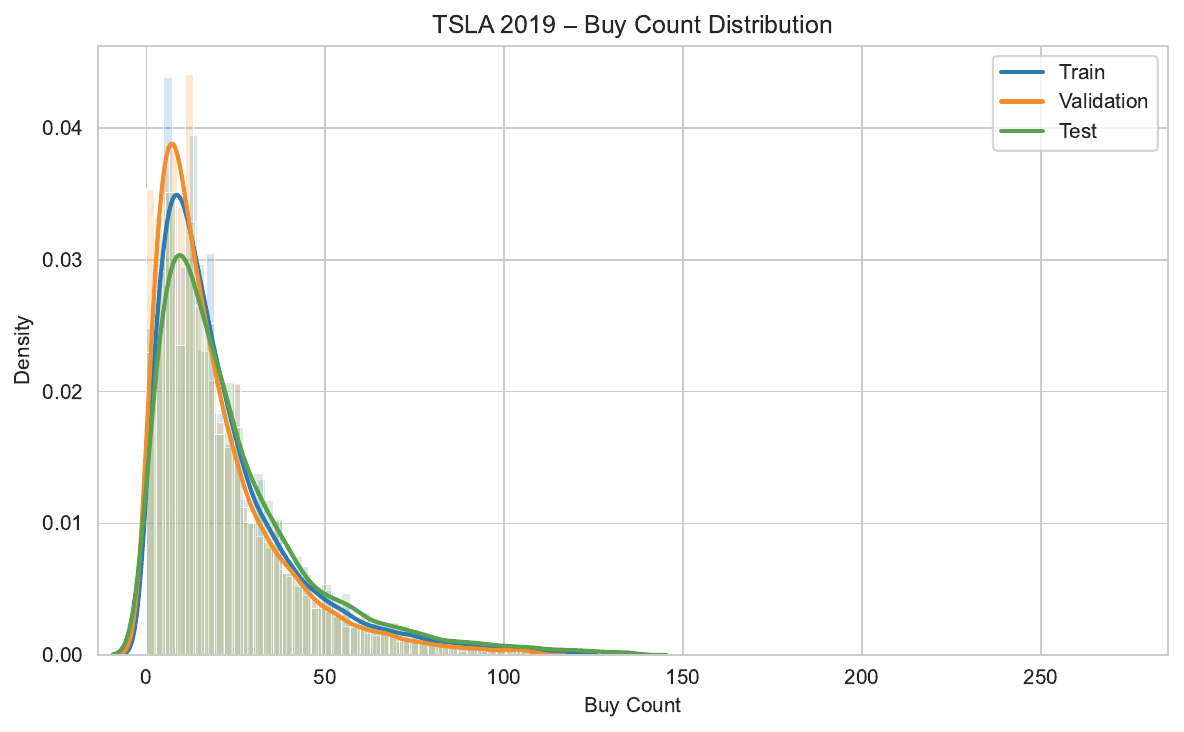} &
            \includegraphics[width=0.16\textwidth]{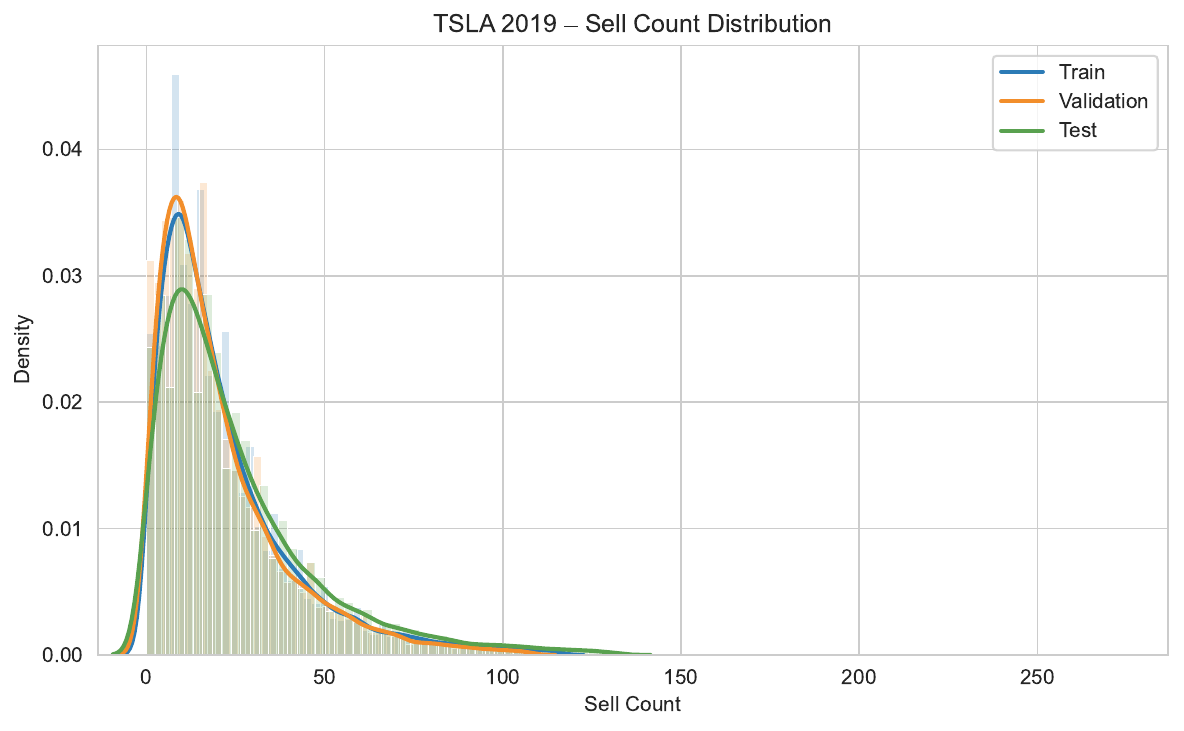} &
            \includegraphics[width=0.16\textwidth]{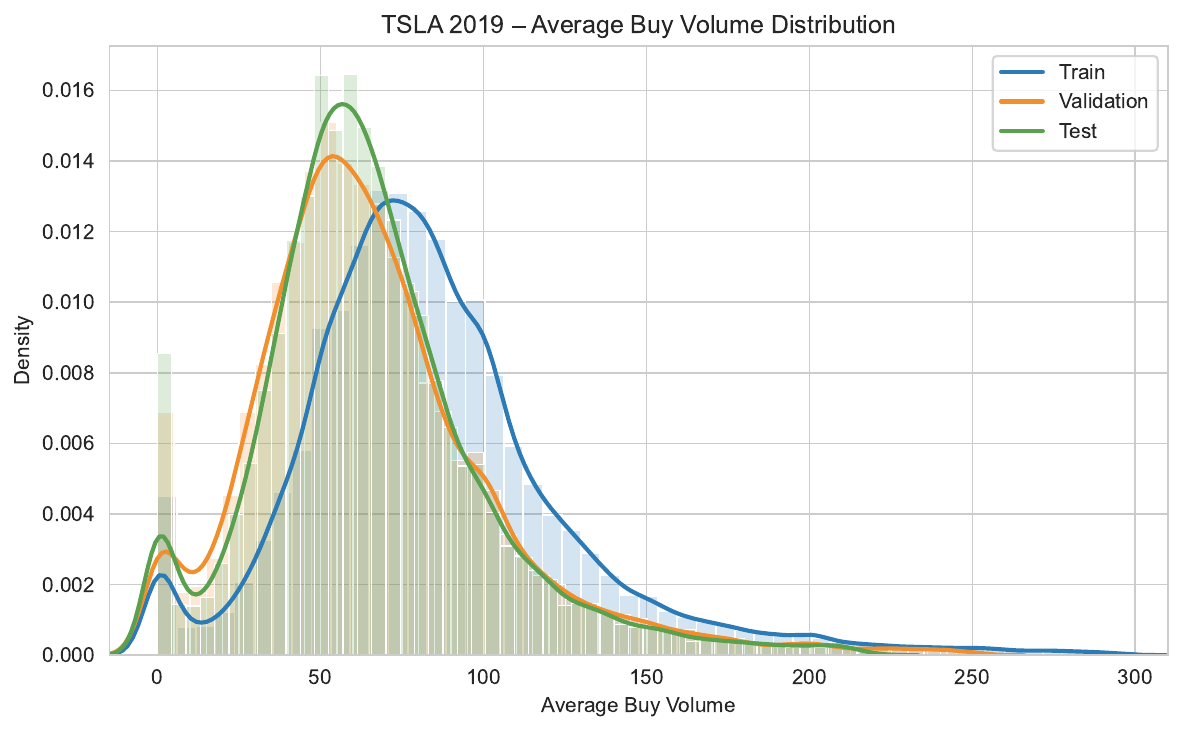} &
            \includegraphics[width=0.16\textwidth]{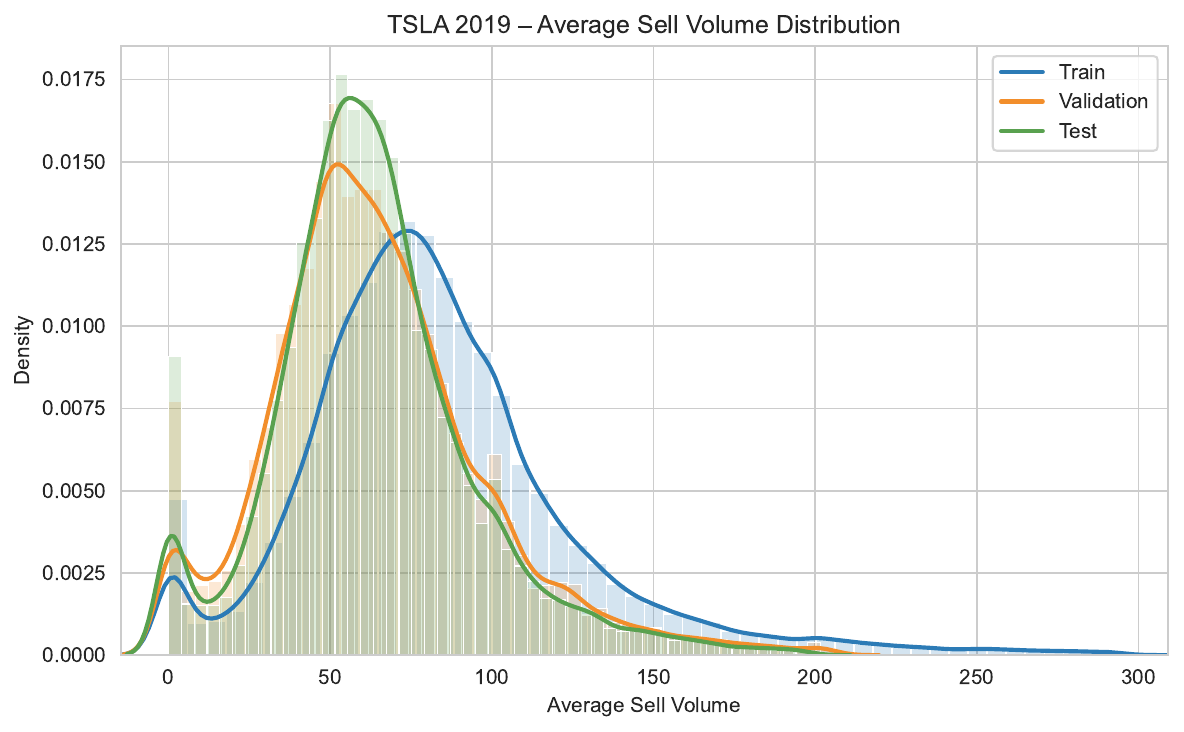} &
            \includegraphics[width=0.16\textwidth]{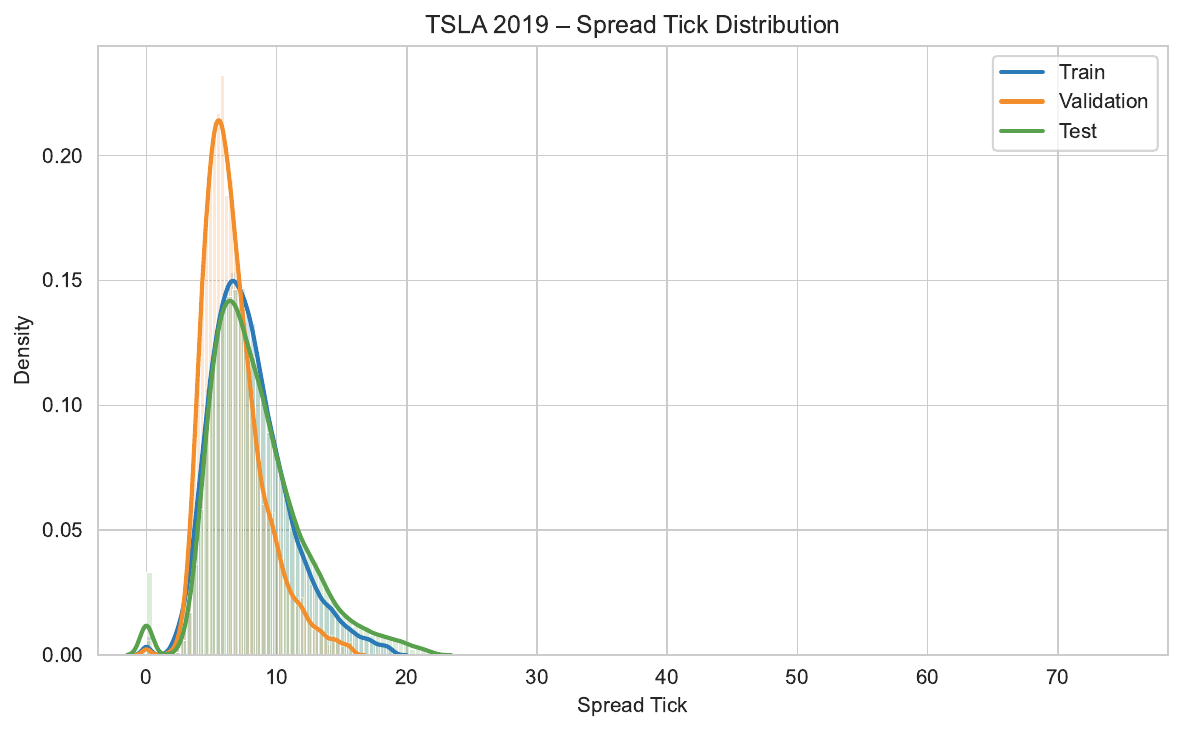} &
            \includegraphics[width=0.16\textwidth]{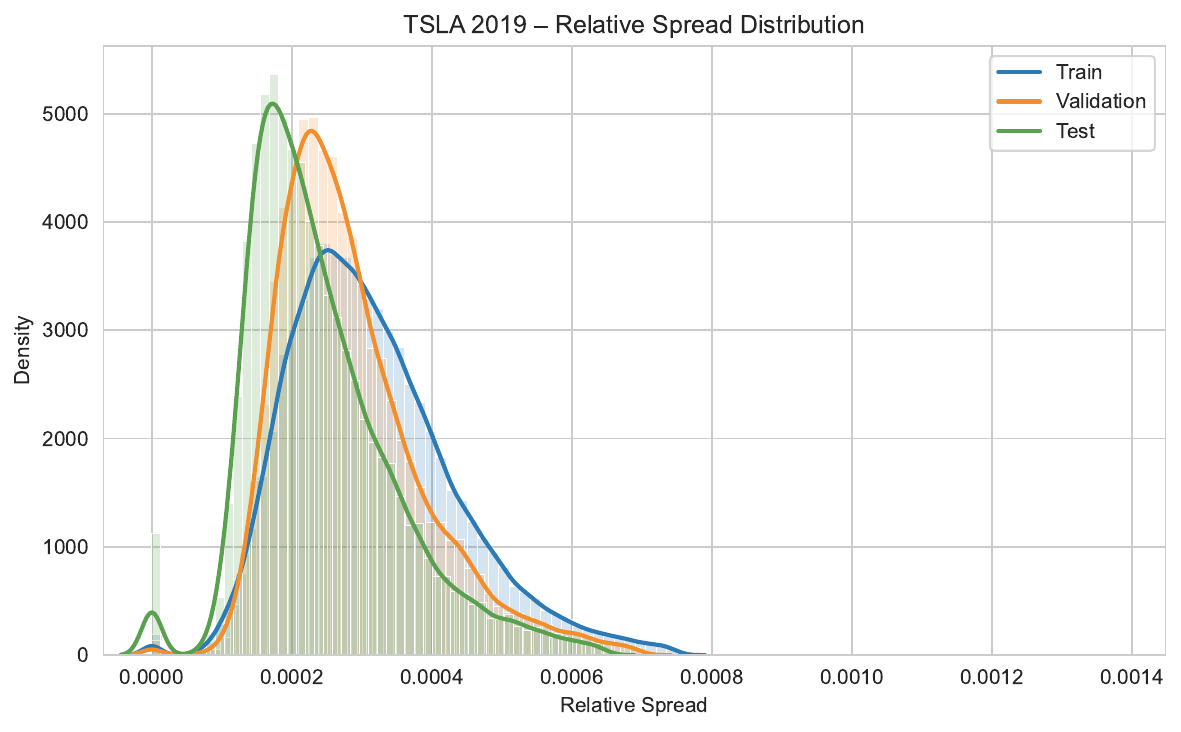} \\
            \scriptsize Buy count & \scriptsize Sell count & \scriptsize Average buy size & \scriptsize Average sell size & \scriptsize Spread (ticks) & \scriptsize Relative spread \\
            \includegraphics[width=0.16\textwidth]{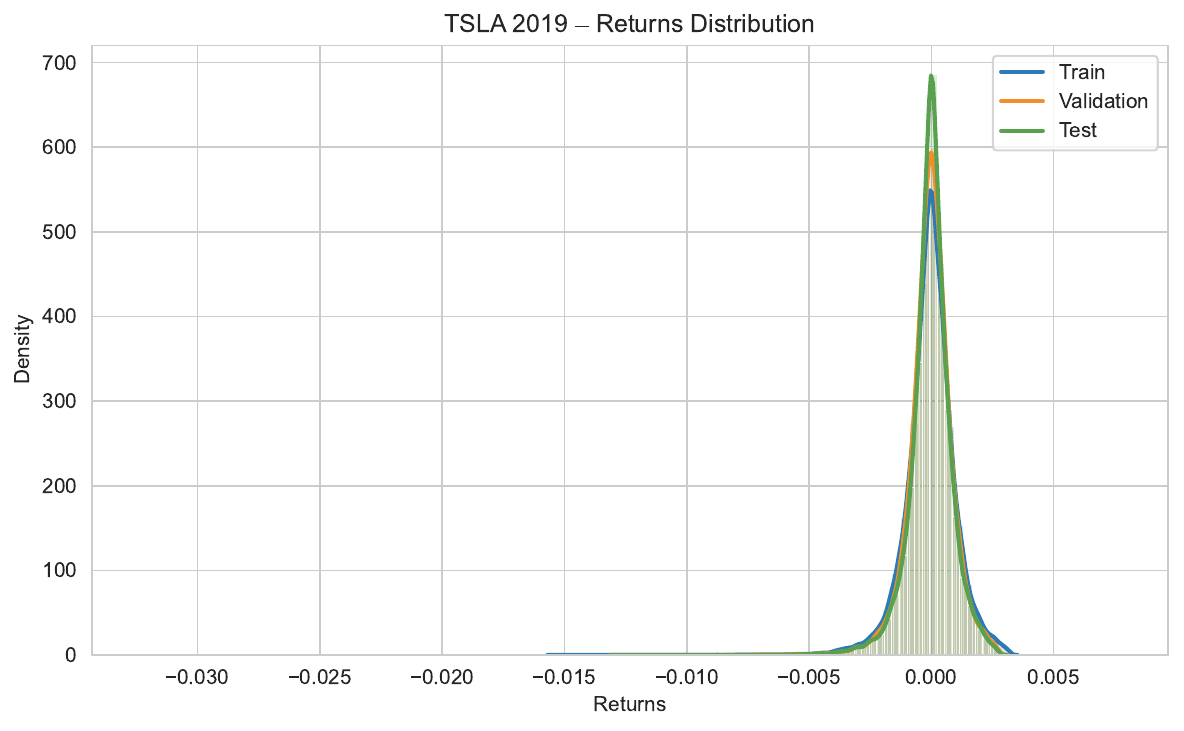} &
            \includegraphics[width=0.16\textwidth]{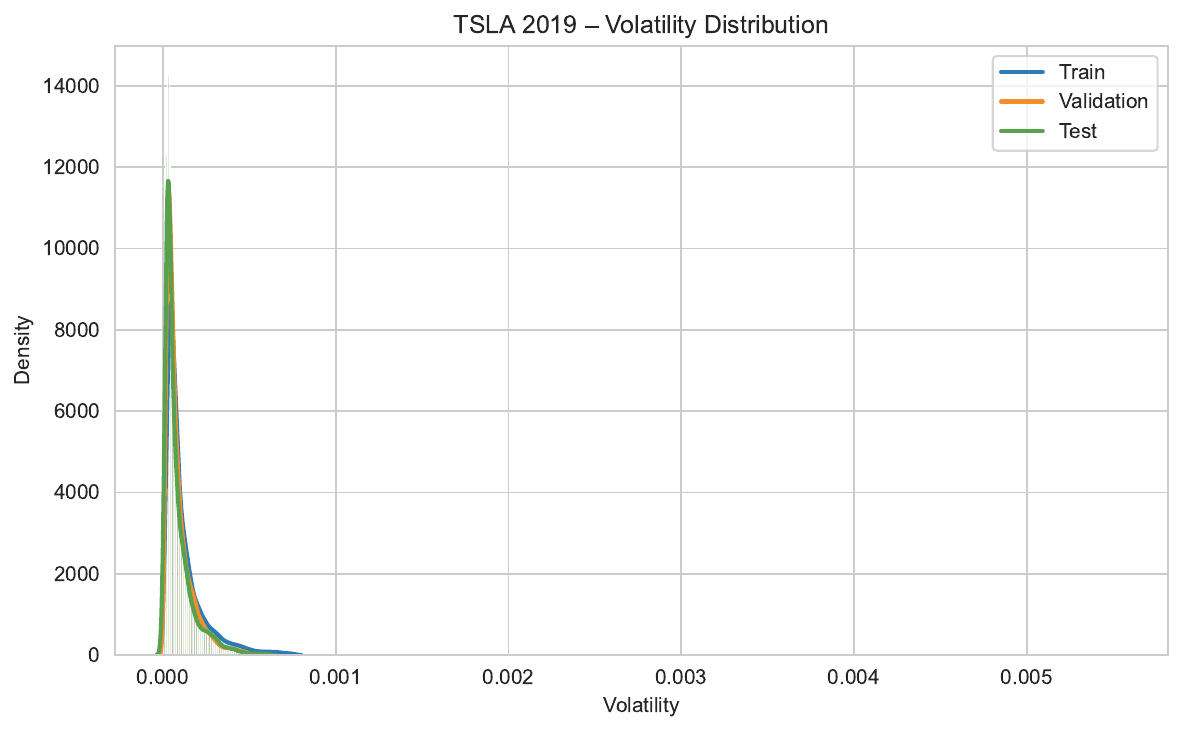} &
            \includegraphics[width=0.16\textwidth]{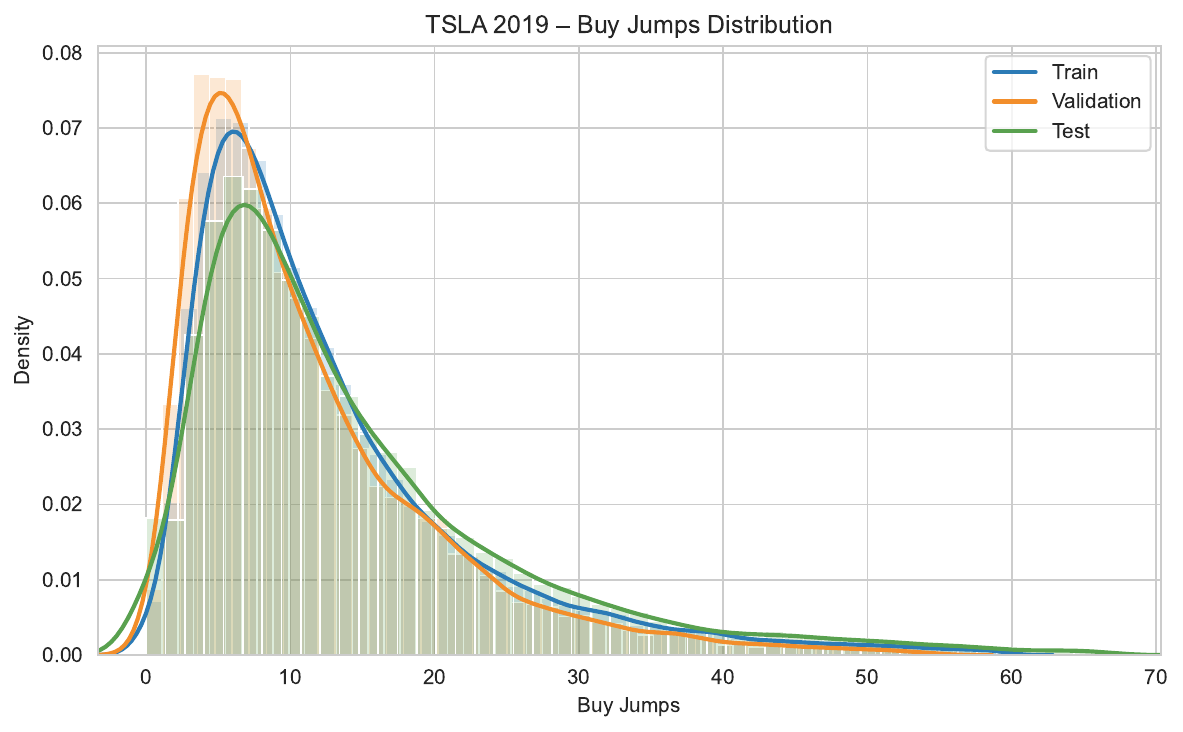} &
            \includegraphics[width=0.16\textwidth]{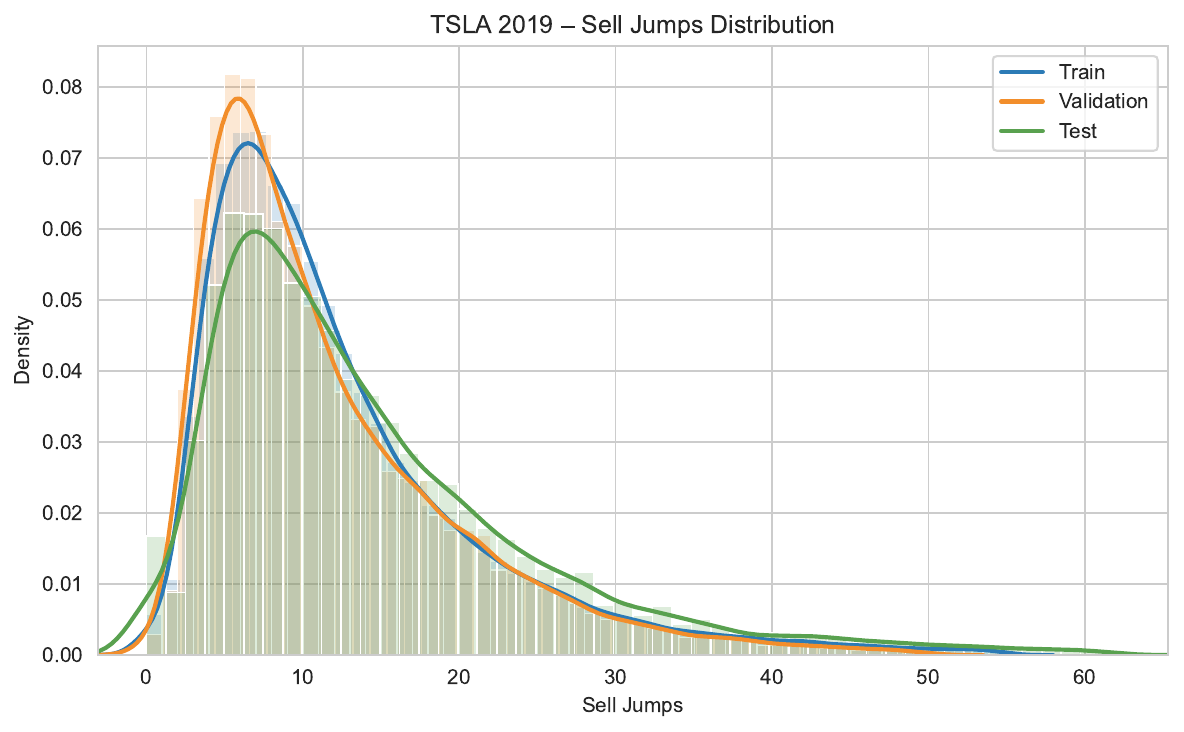} &
            \includegraphics[width=0.16\textwidth]{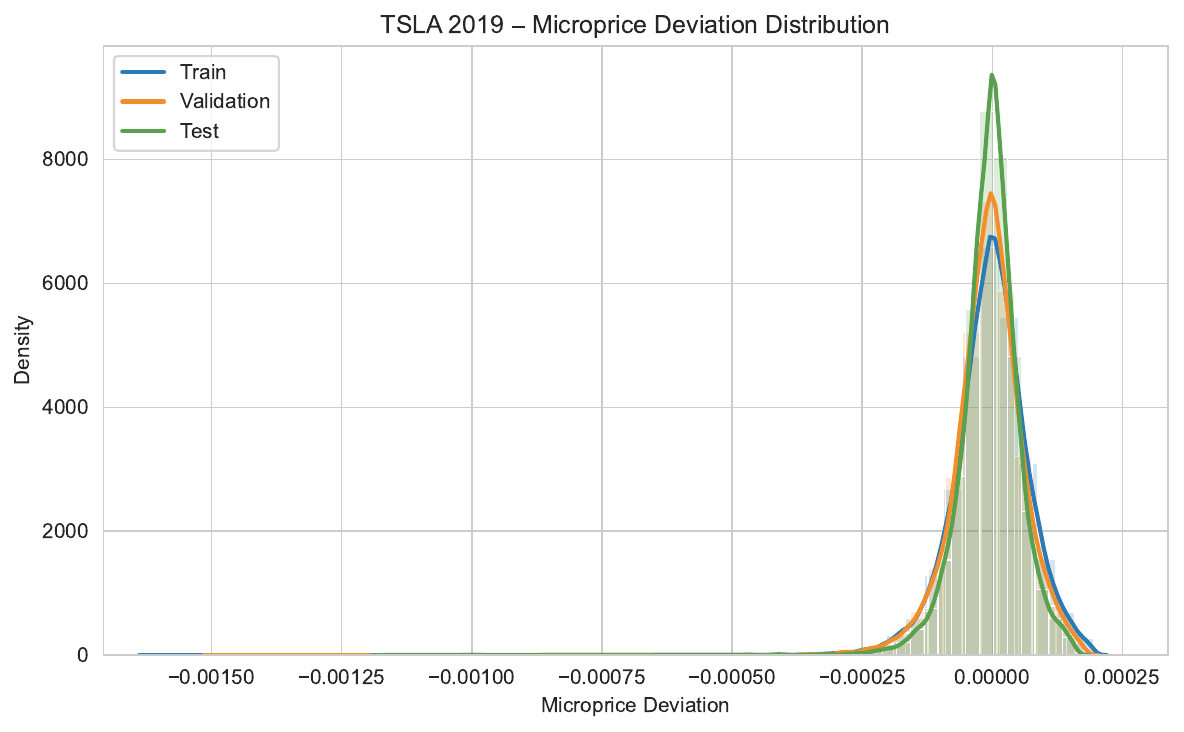} &
            \includegraphics[width=0.16\textwidth]{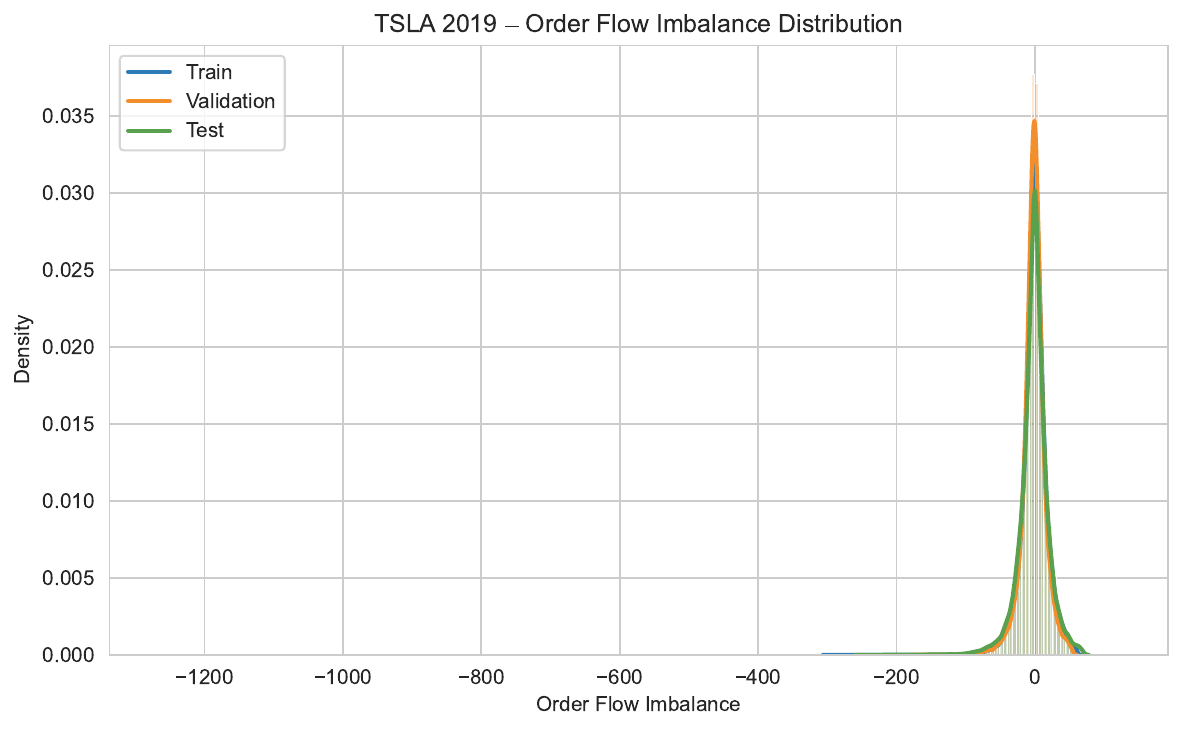} \\
            \scriptsize Returns & \scriptsize Volatility & \scriptsize Buy-order excitation & \scriptsize Sell-order excitation & \scriptsize Microprice deviation & \scriptsize OFI \\
        \end{tabular}
    \end{minipage}}
    \caption{TSLA, 2019.}
    \end{subfigure}

    \vspace{0.2em}

    \begin{subfigure}{\linewidth}
    \centering
    \makebox[\textwidth][c]{\begin{minipage}{1.12\textwidth}\centering
        \setlength{\tabcolsep}{0.5pt}
        \begin{tabular}{cccccc}
            \includegraphics[width=0.16\textwidth]{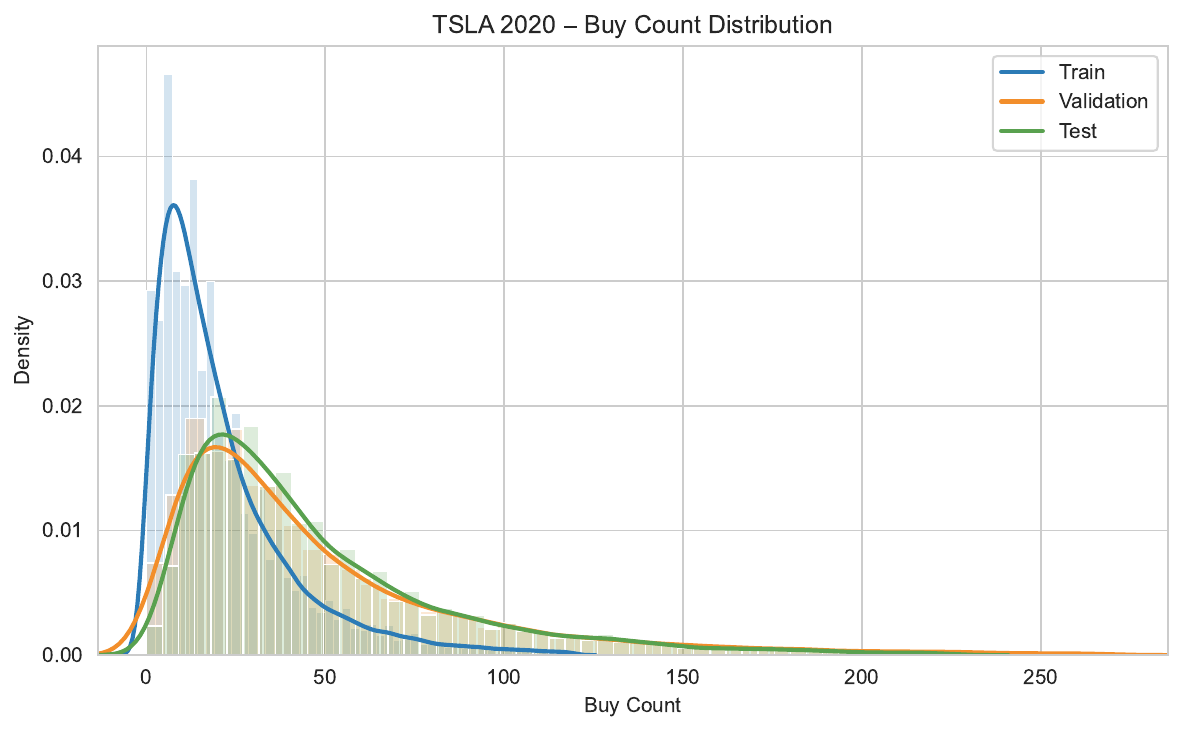} &
            \includegraphics[width=0.16\textwidth]{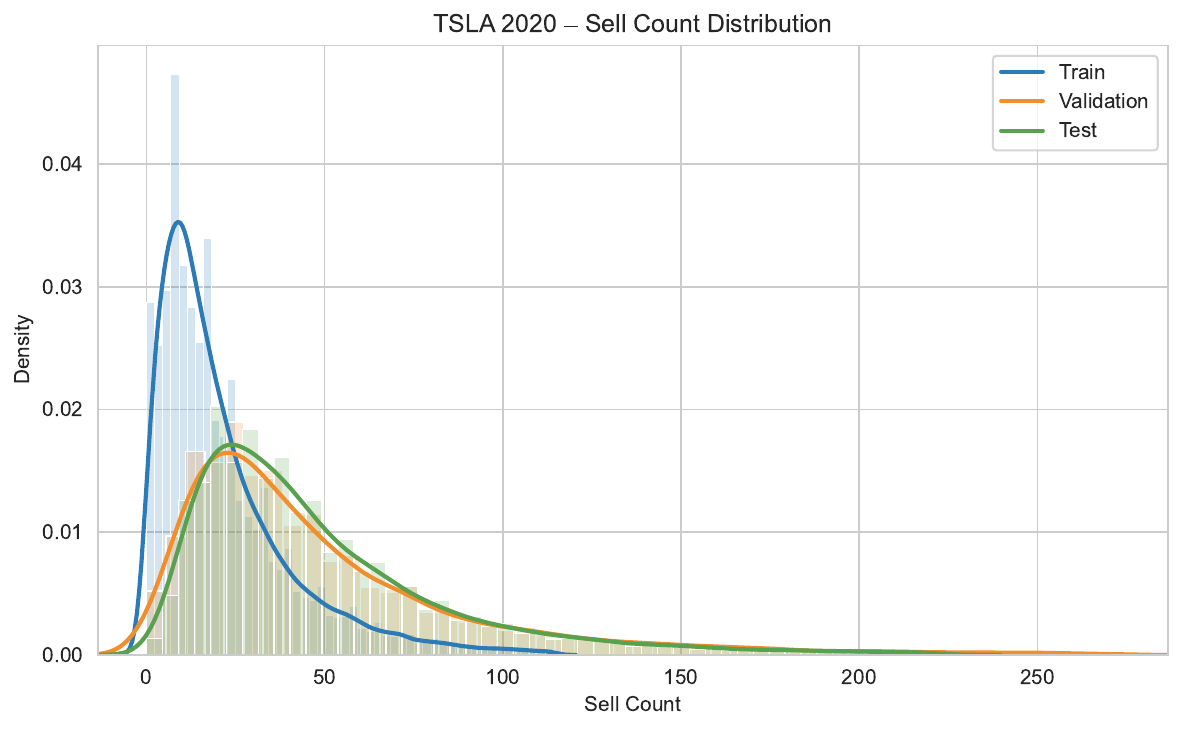} &
            \includegraphics[width=0.16\textwidth]{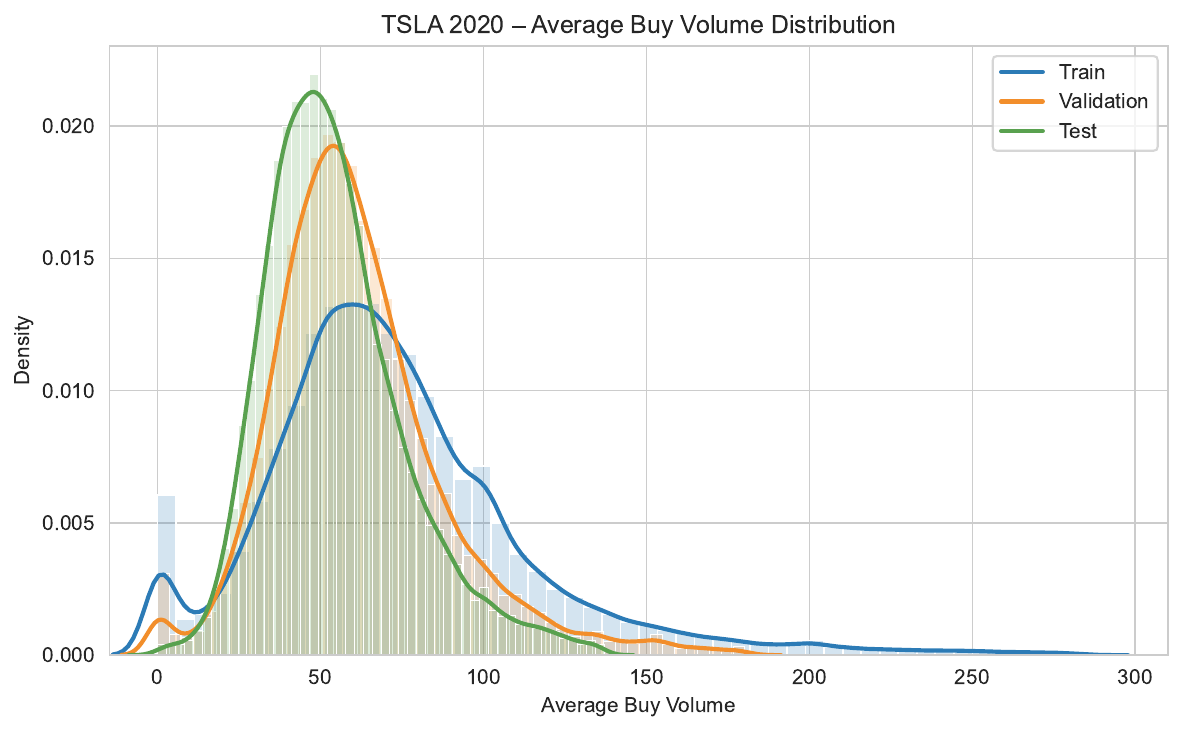} &
            \includegraphics[width=0.16\textwidth]{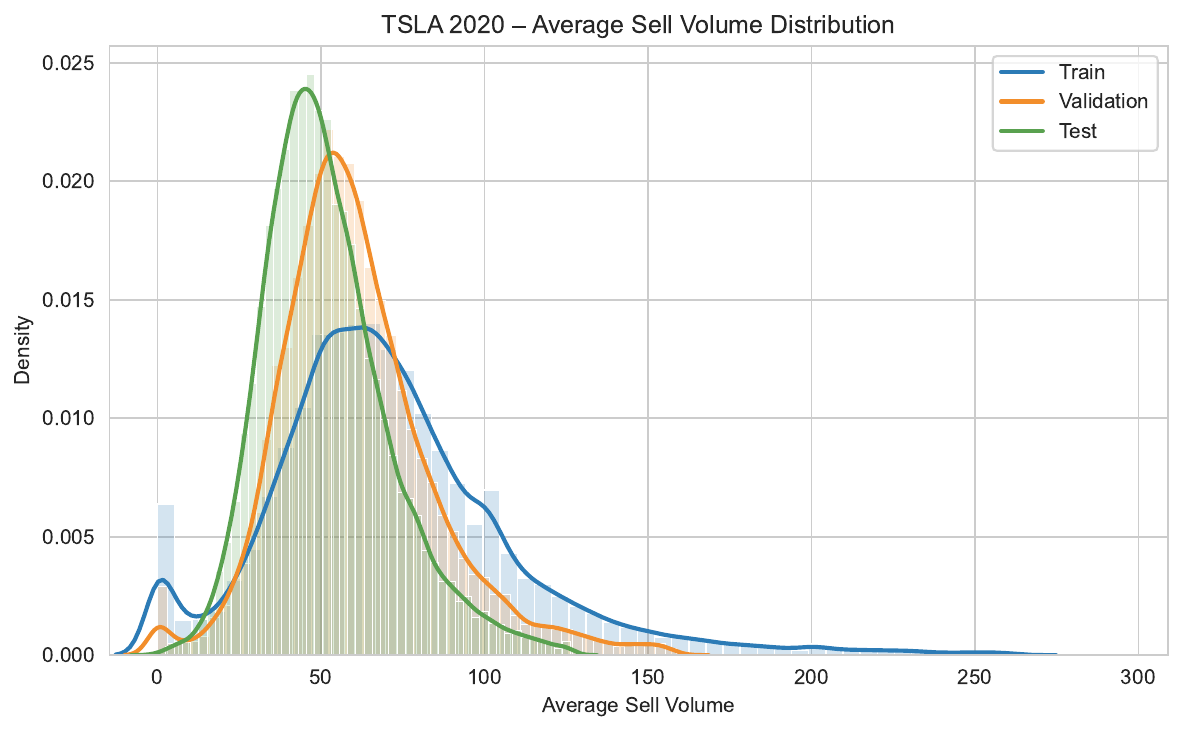} &
            \includegraphics[width=0.16\textwidth]{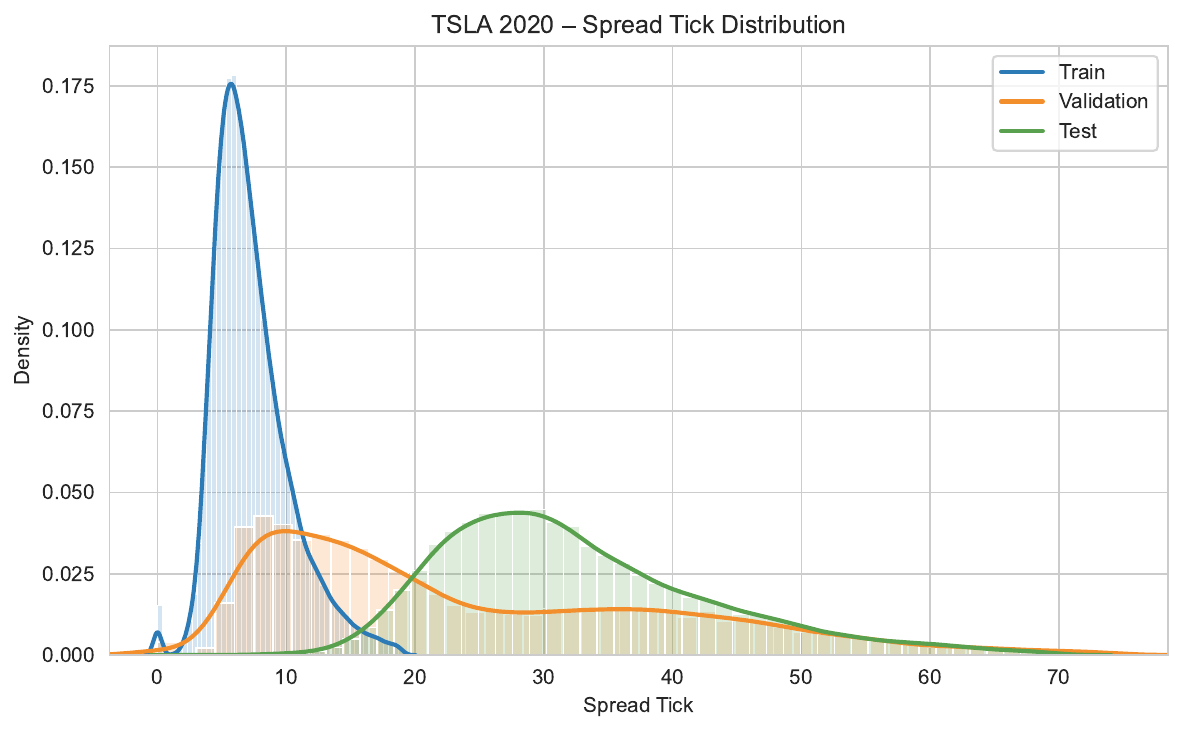} &
            \includegraphics[width=0.16\textwidth]{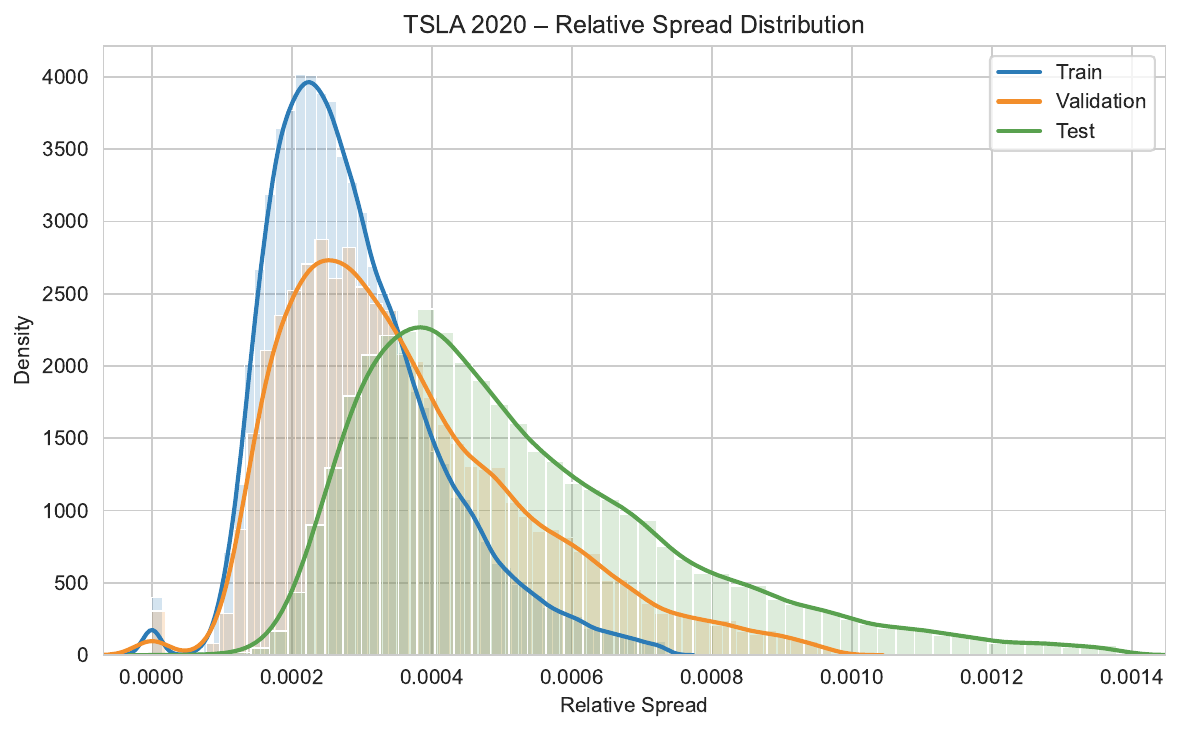} \\
            \scriptsize Buy count & \scriptsize Sell count & \scriptsize Average buy size & \scriptsize Average sell size & \scriptsize Spread (ticks) & \scriptsize Relative spread \\
            \includegraphics[width=0.16\textwidth]{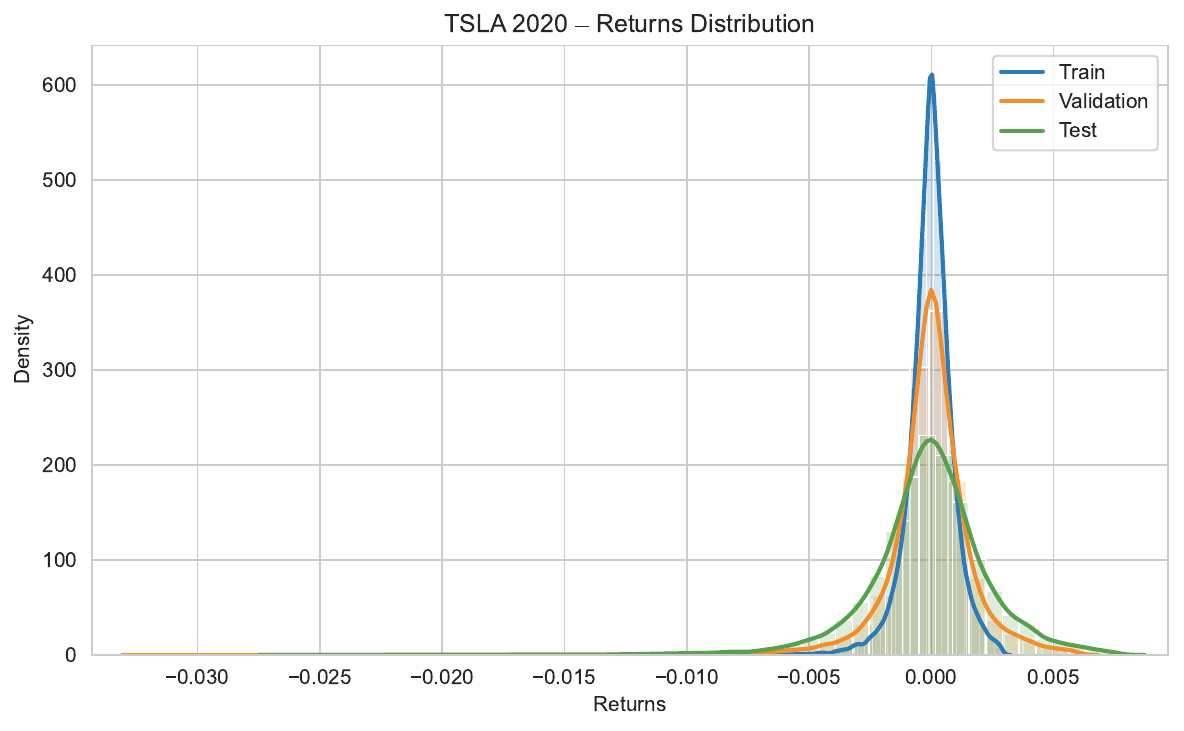} &
            \includegraphics[width=0.16\textwidth]{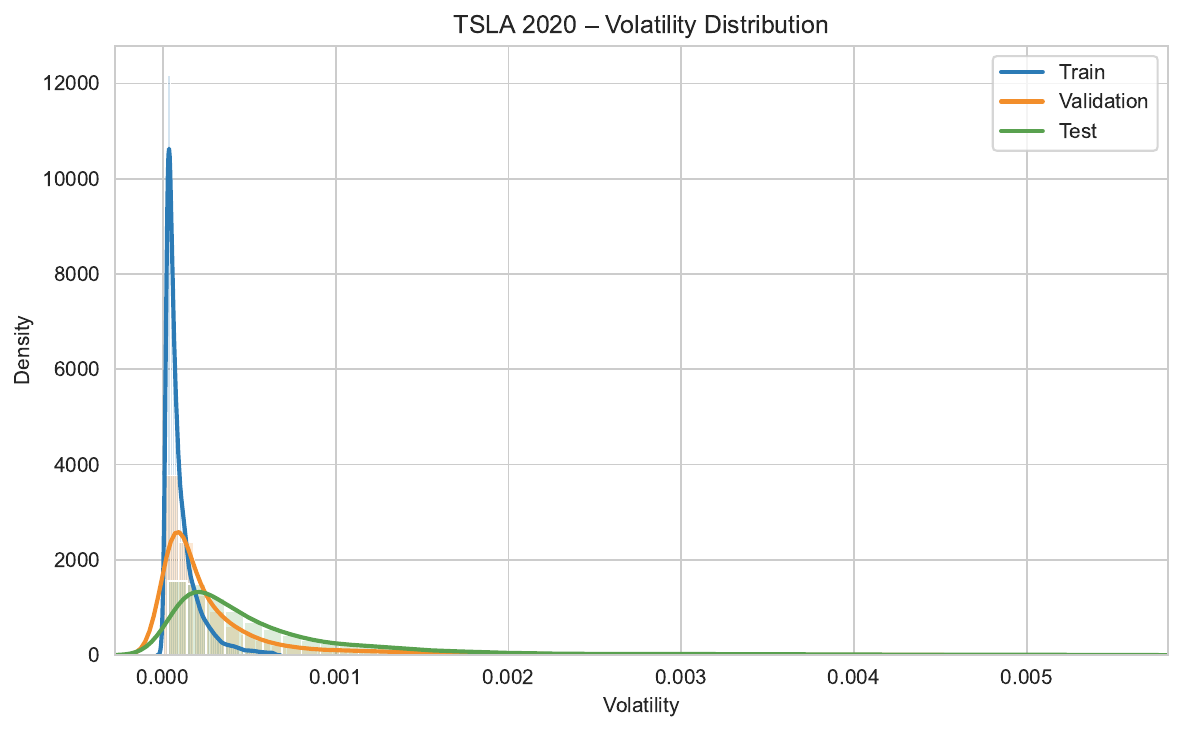} &
            \includegraphics[width=0.16\textwidth]{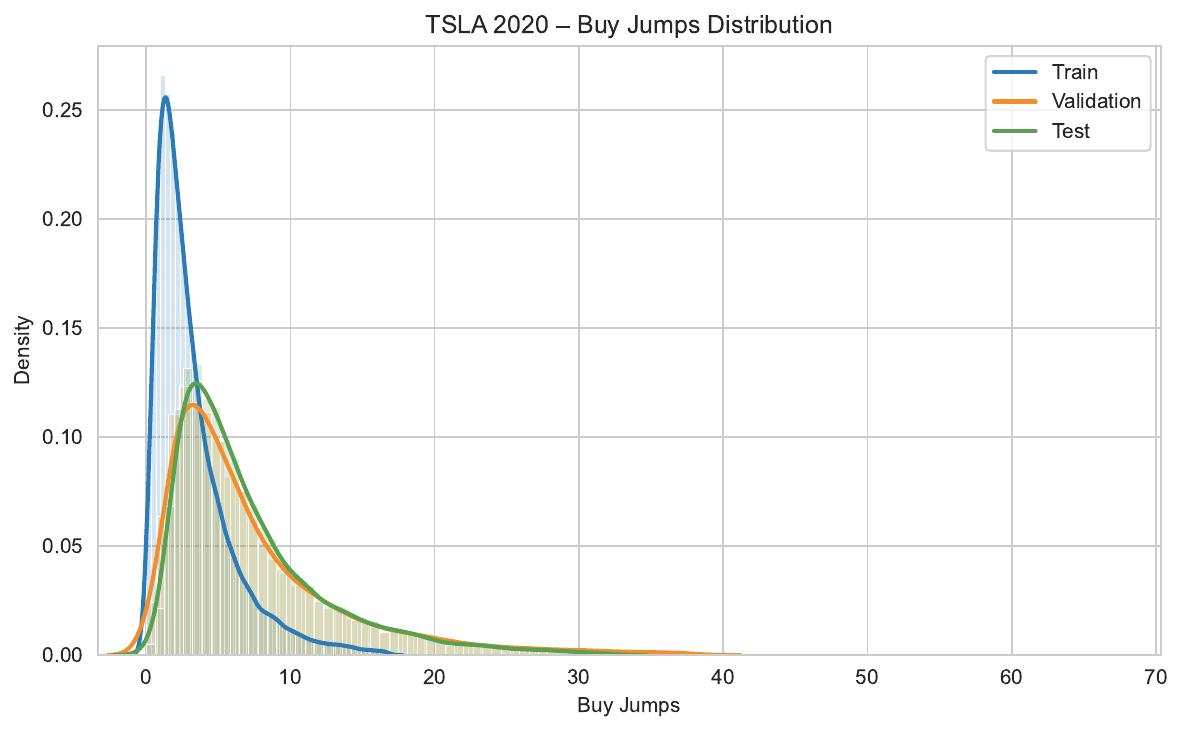} &
            \includegraphics[width=0.16\textwidth]{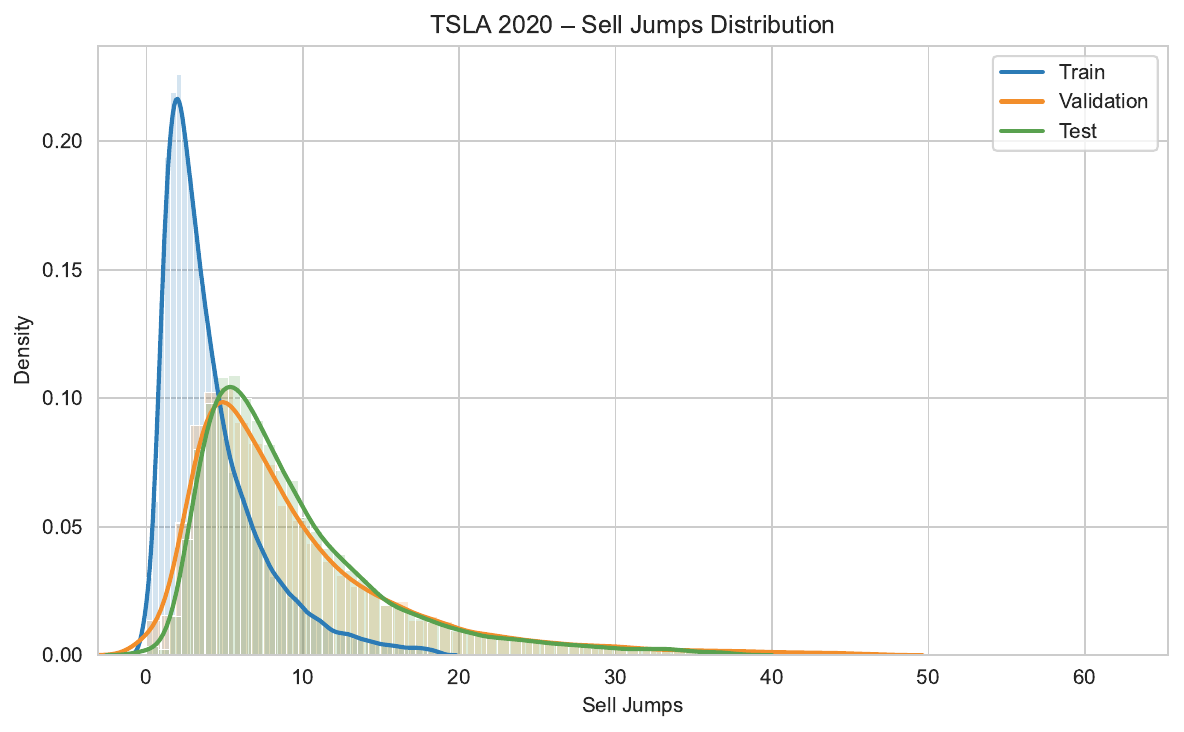} &
            \includegraphics[width=0.16\textwidth]{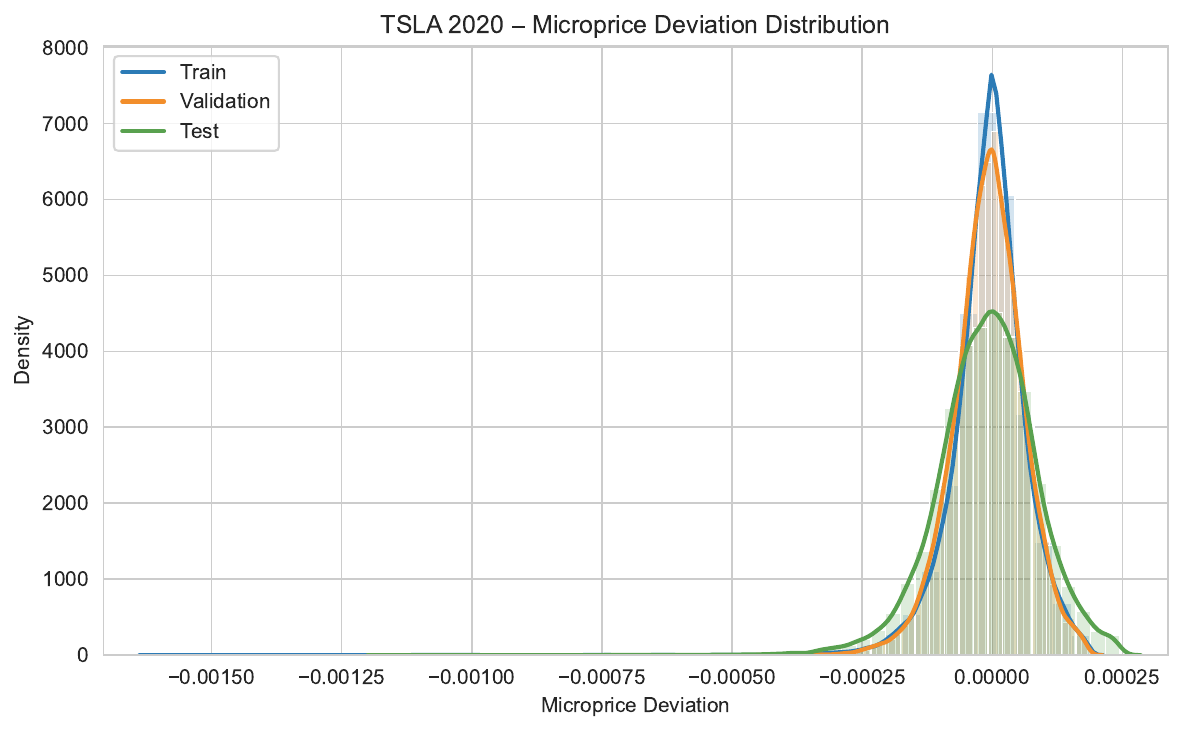} &
            \includegraphics[width=0.16\textwidth]{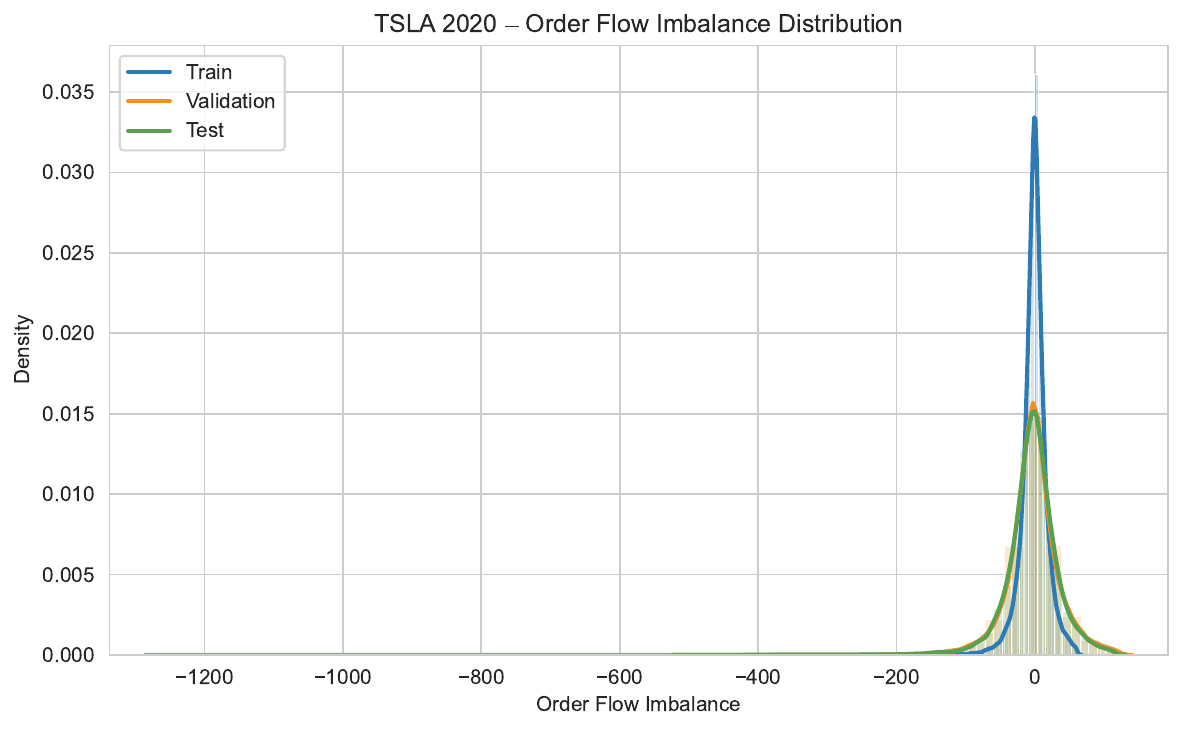} \\
            \scriptsize Returns & \scriptsize Volatility & \scriptsize Buy-order excitation & \scriptsize Sell-order excitation & \scriptsize Microprice deviation & \scriptsize OFI \\
        \end{tabular}
    \end{minipage}}
    \caption{TSLA, 2020.}
    \end{subfigure}

    \caption{Train, validation, and test distributions of all state variables for TSLA in 2019 and 2020. The comparison makes the stronger cross-split distributional shift in 2020 visually apparent.}
    \label{fig:full_state_dist_tsla}
\end{figure}

\begin{figure}[H]
    \centering
    \begin{subfigure}{\linewidth}
    \centering
    \makebox[\textwidth][c]{\begin{minipage}{1.12\textwidth}\centering
        \setlength{\tabcolsep}{0.5pt}
        \begin{tabular}{cccccc}
            \includegraphics[width=0.16\textwidth]{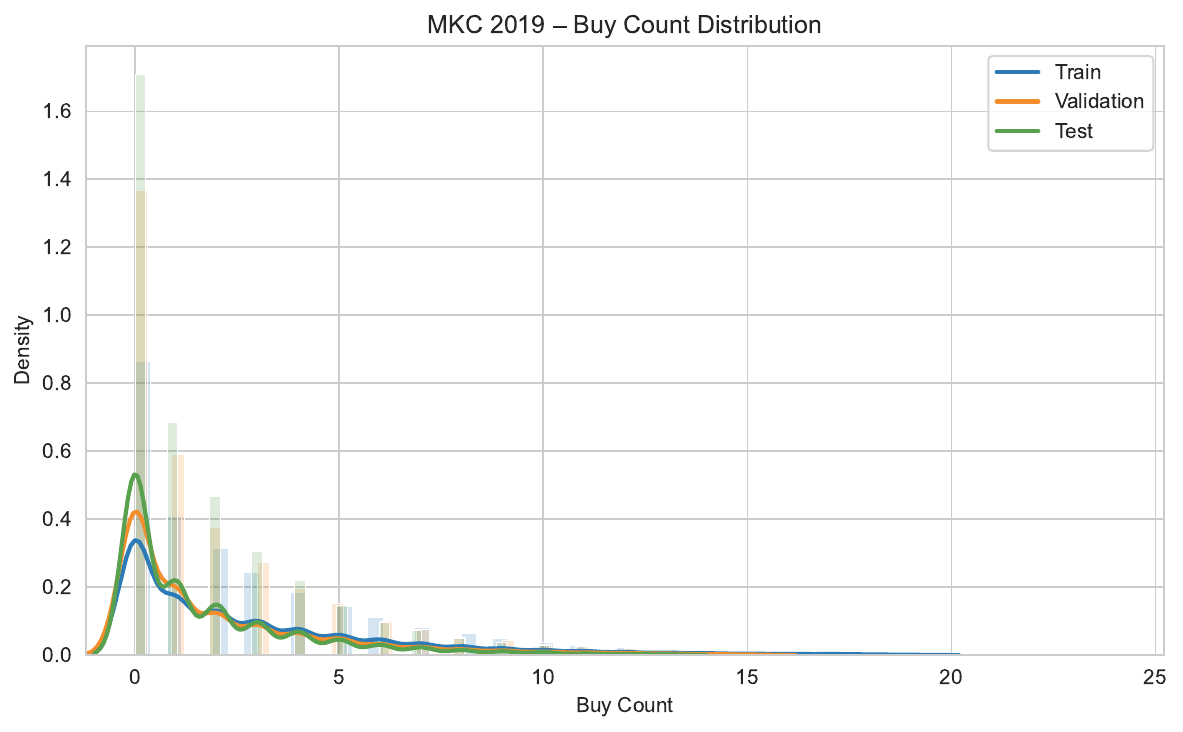} &
            \includegraphics[width=0.16\textwidth]{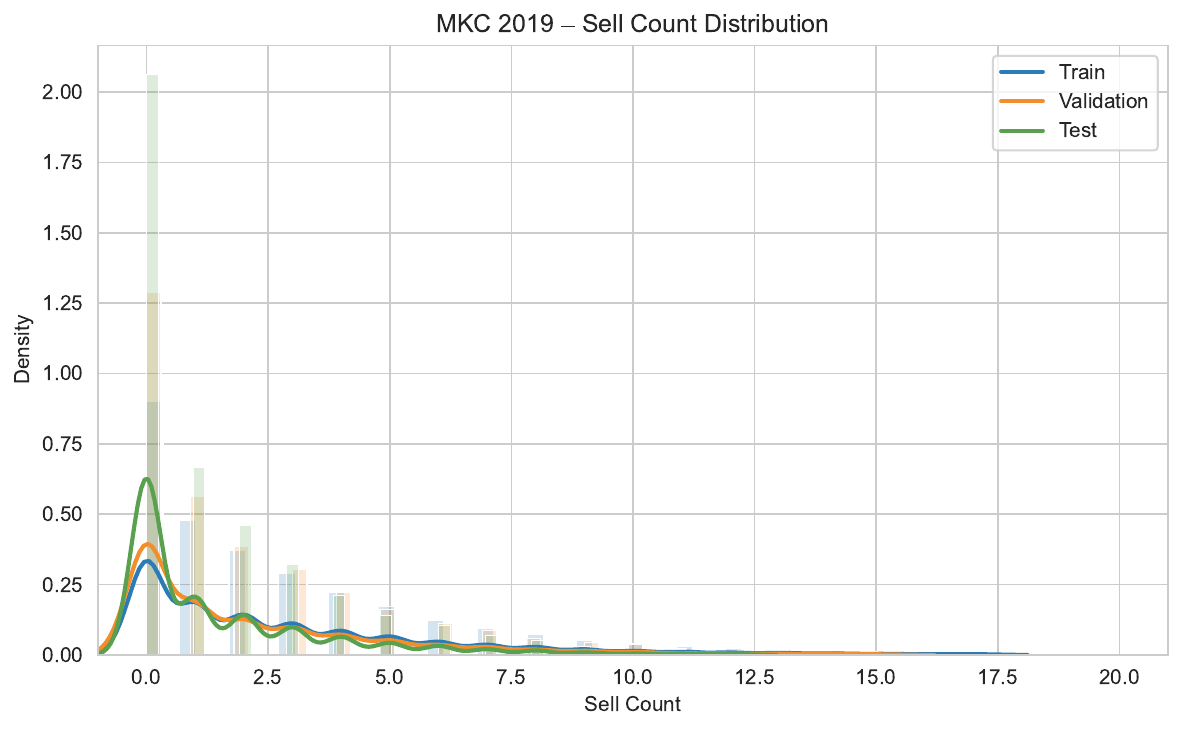} &
            \includegraphics[width=0.16\textwidth]{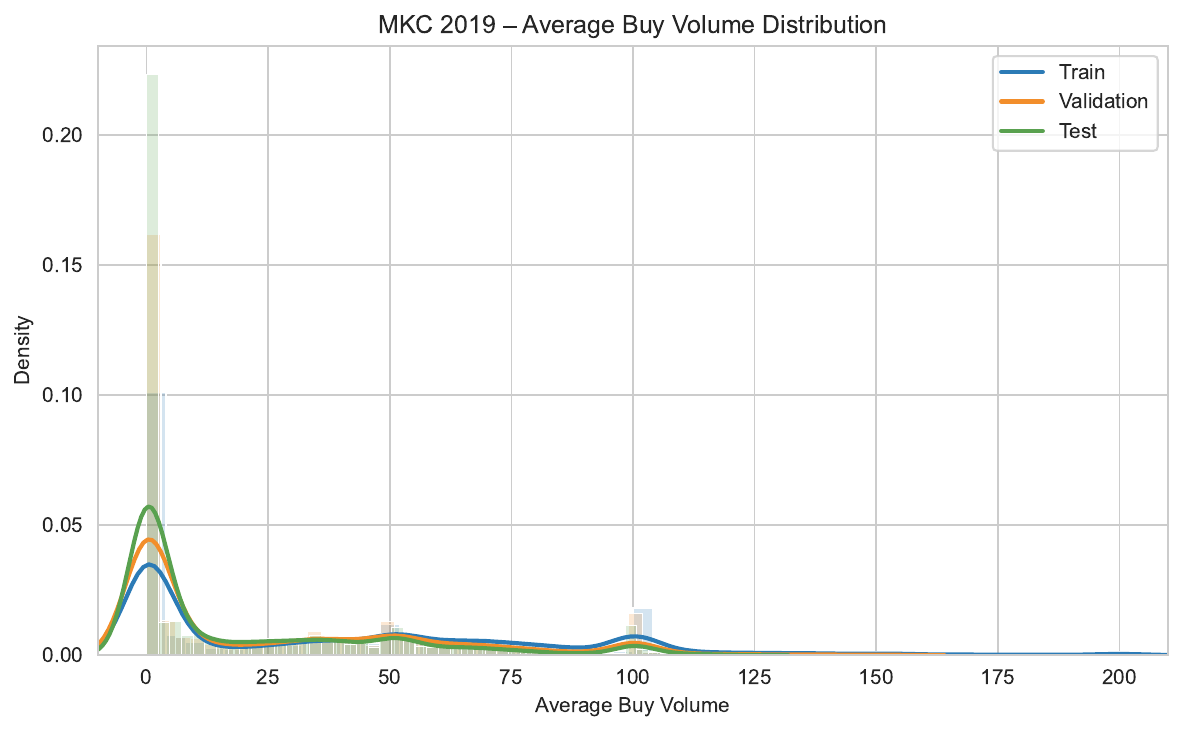} &
            \includegraphics[width=0.16\textwidth]{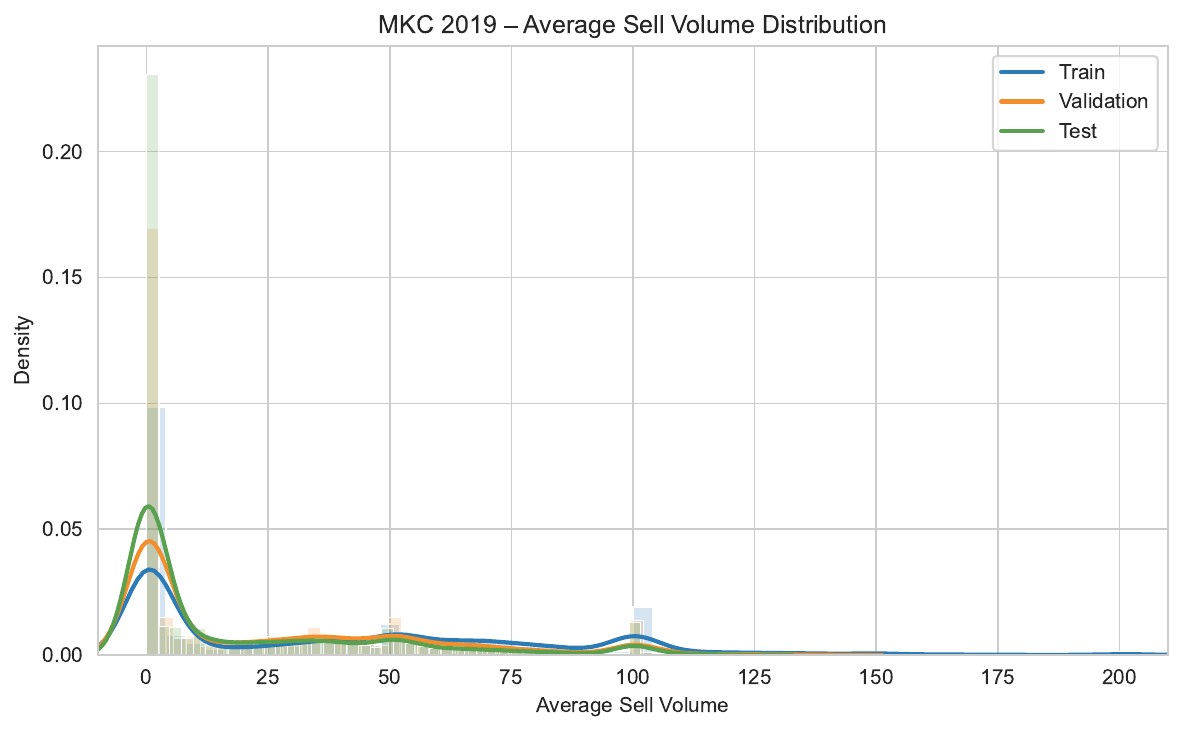} &
            \includegraphics[width=0.16\textwidth]{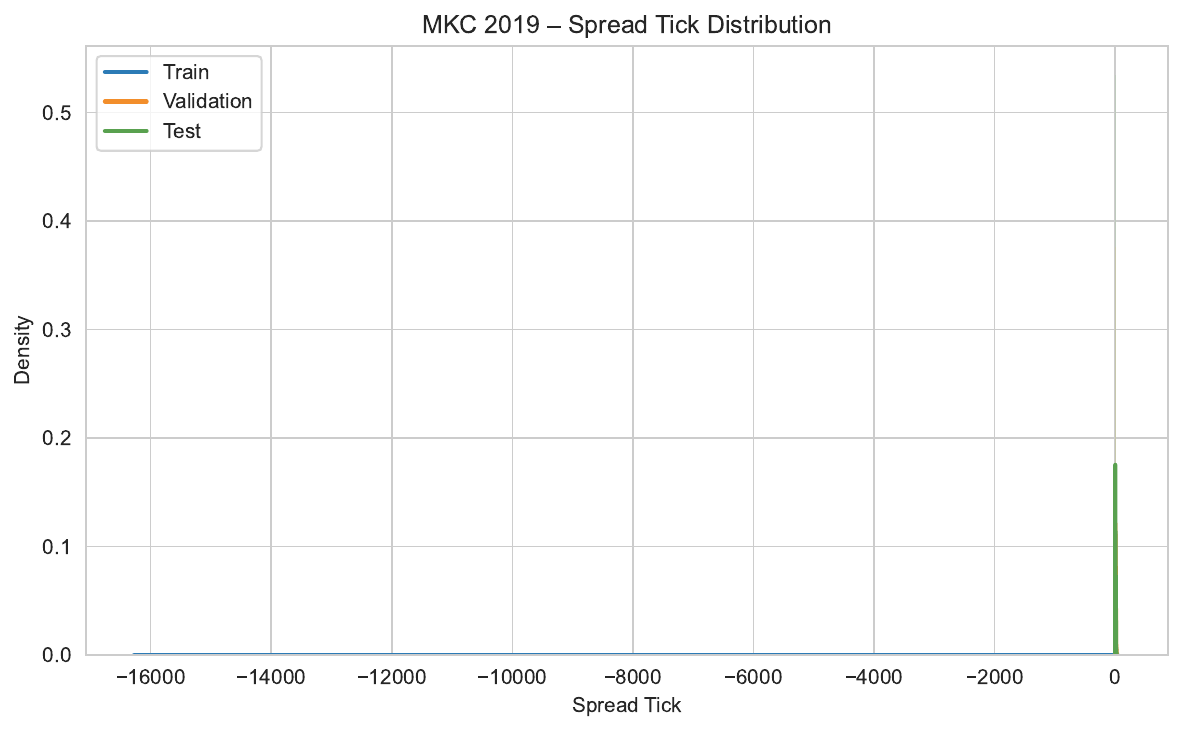} &
            \includegraphics[width=0.16\textwidth]{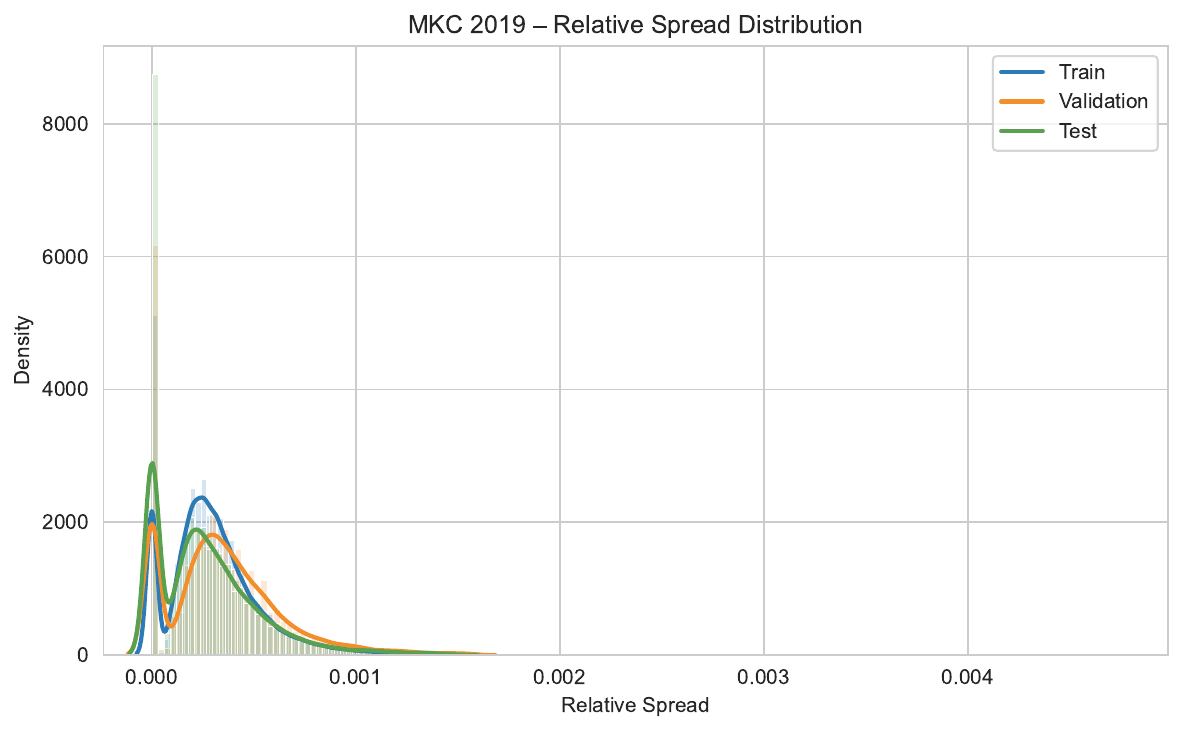} \\
            \scriptsize Buy count & \scriptsize Sell count & \scriptsize Average buy size & \scriptsize Average sell size & \scriptsize Spread (ticks) & \scriptsize Relative spread \\
            \includegraphics[width=0.16\textwidth]{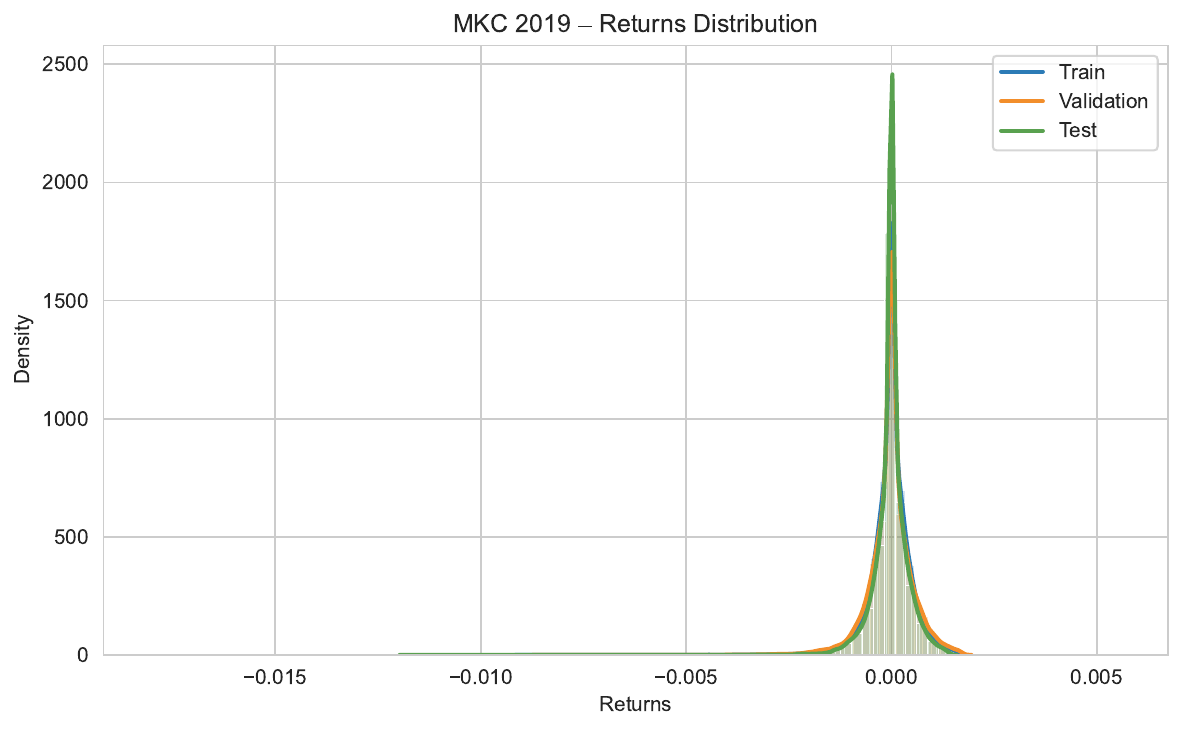} &
            \includegraphics[width=0.16\textwidth]{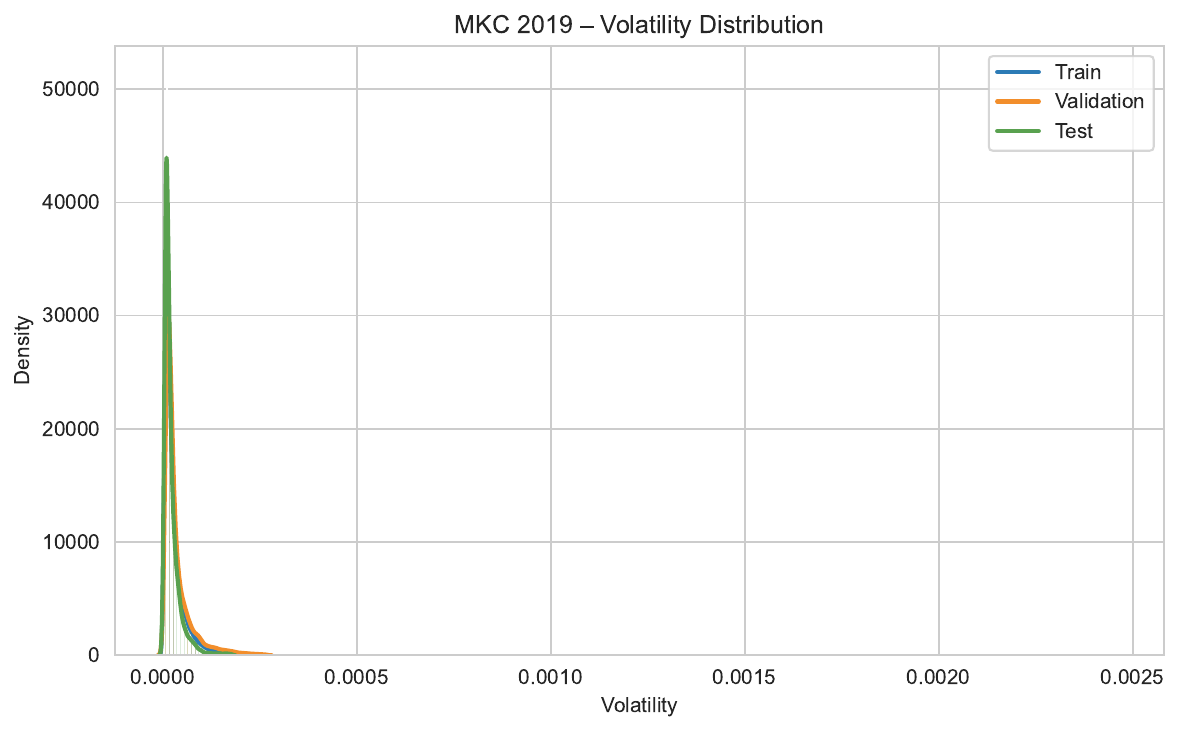} &
            \includegraphics[width=0.16\textwidth]{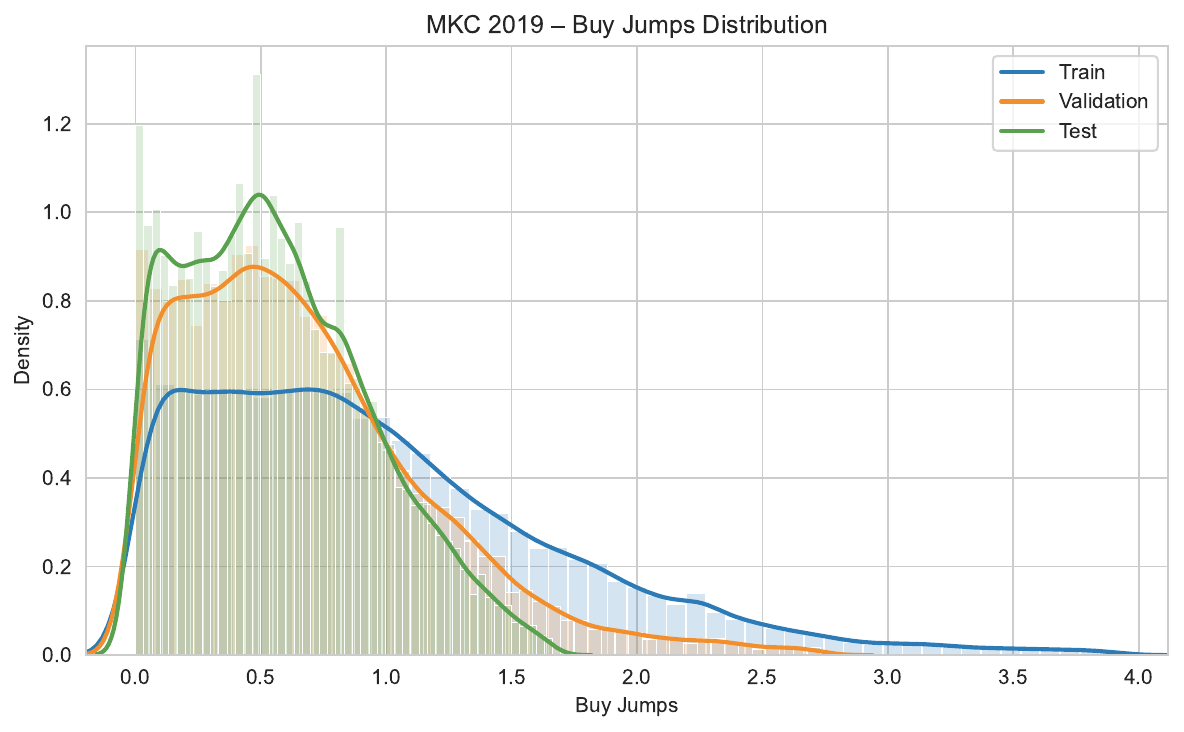} &
            \includegraphics[width=0.16\textwidth]{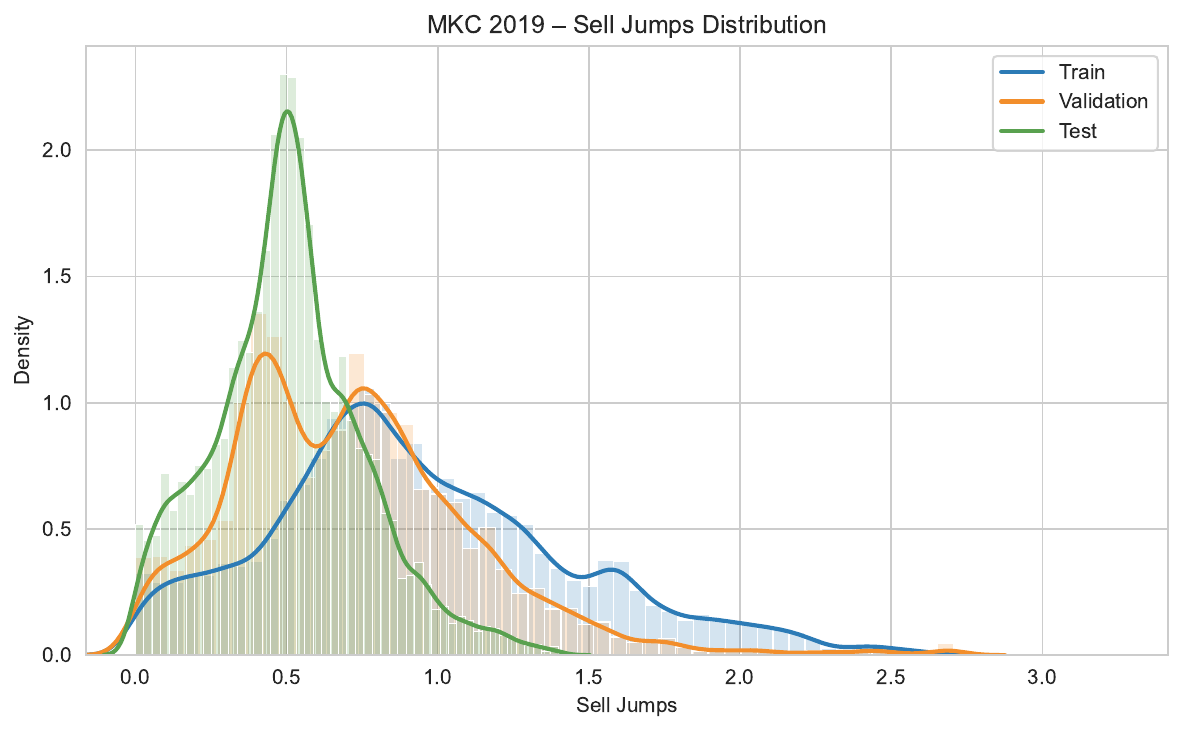} &
            \includegraphics[width=0.16\textwidth]{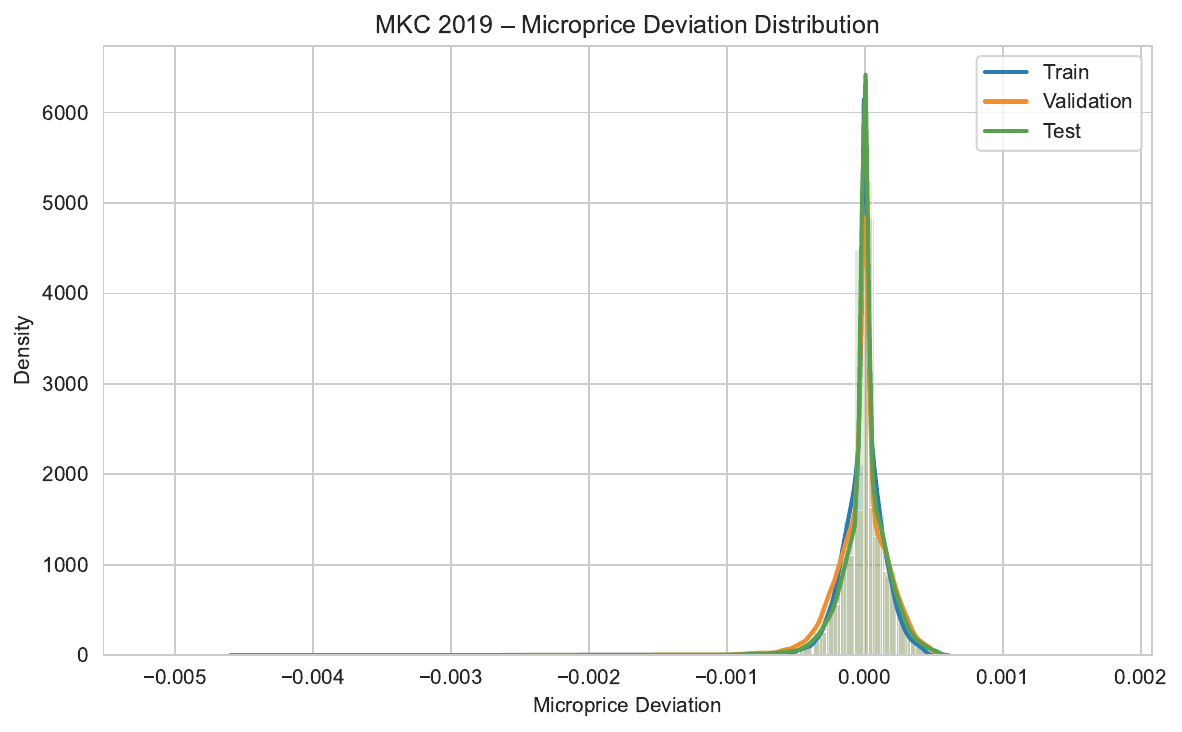} &
            \includegraphics[width=0.16\textwidth]{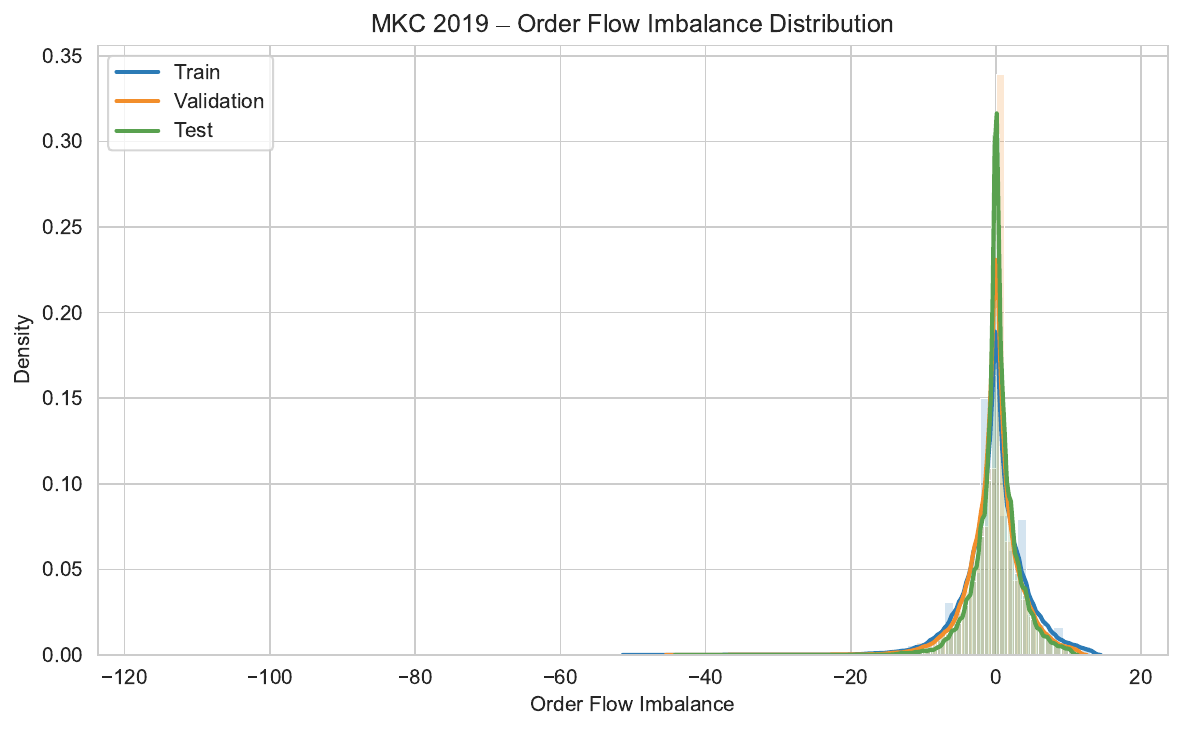} \\
            \scriptsize Returns & \scriptsize Volatility & \scriptsize Buy-order excitation & \scriptsize Sell-order excitation & \scriptsize Microprice deviation & \scriptsize OFI \\
        \end{tabular}
    \end{minipage}}
    \caption{MKC, 2019.}
    \end{subfigure}

    \vspace{0.2em}

    \begin{subfigure}{\linewidth}
    \centering
    \makebox[\textwidth][c]{\begin{minipage}{1.12\textwidth}\centering
        \setlength{\tabcolsep}{0.5pt}
        \begin{tabular}{cccccc}
            \includegraphics[width=0.16\textwidth]{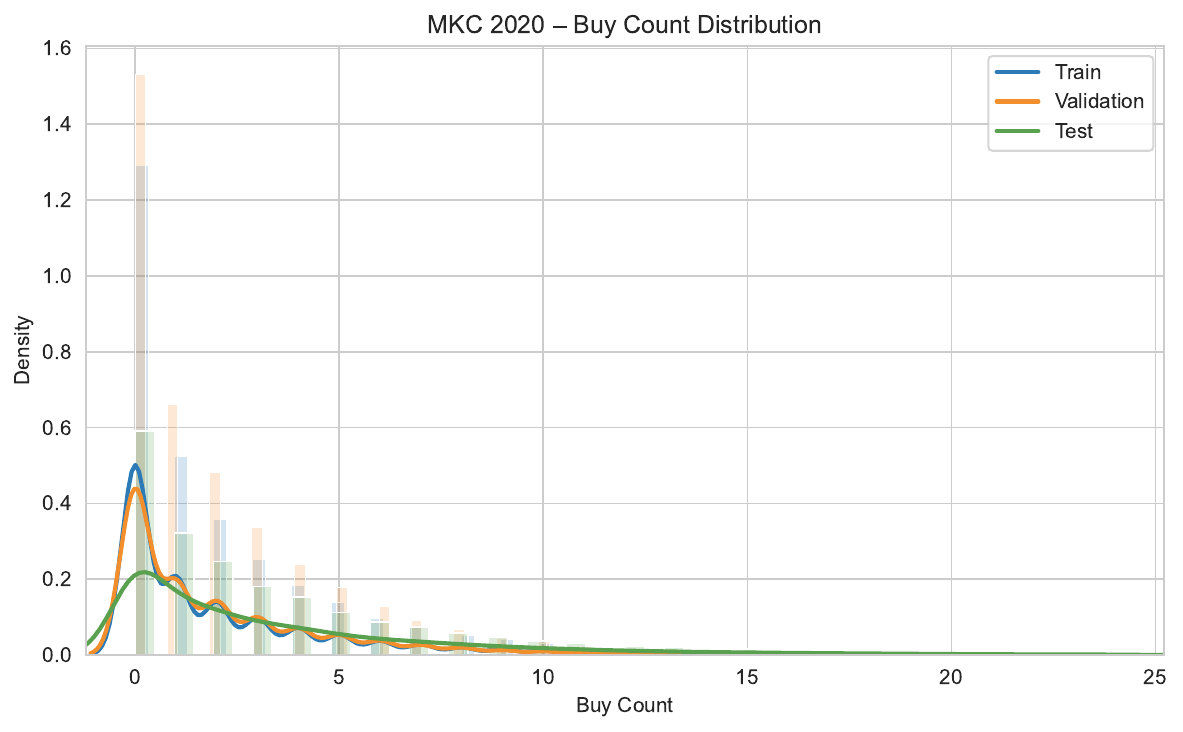} &
            \includegraphics[width=0.16\textwidth]{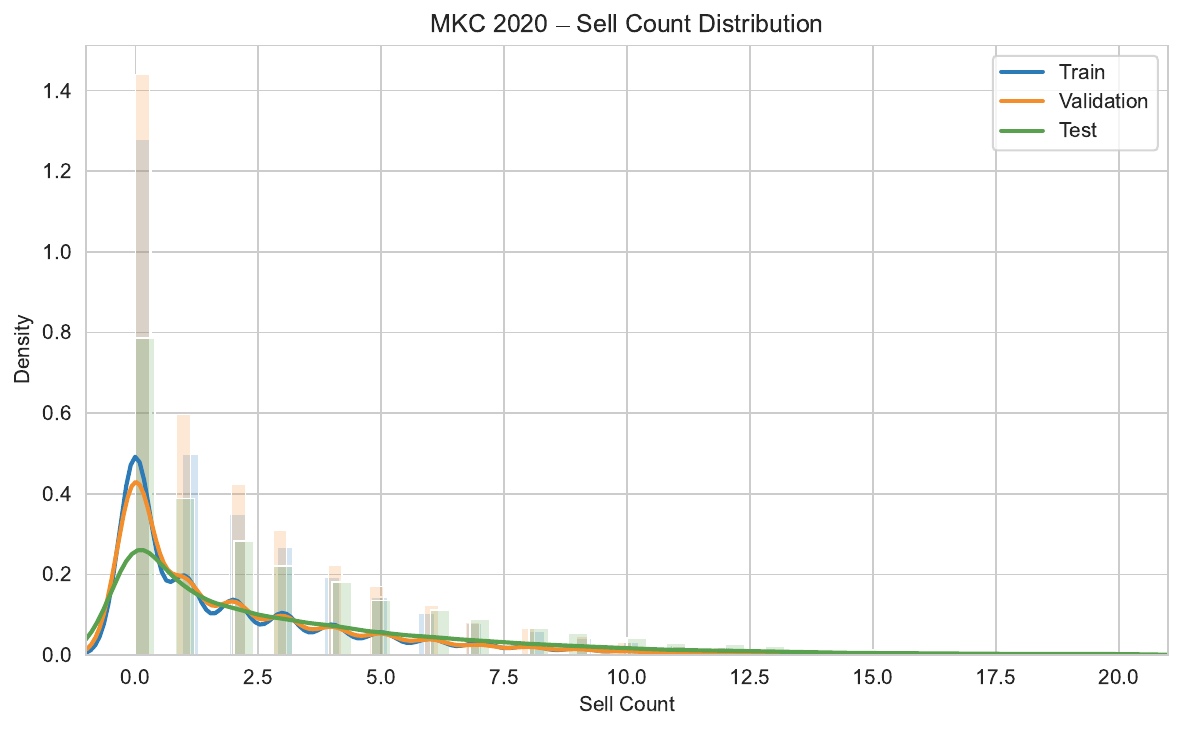} &
            \includegraphics[width=0.16\textwidth]{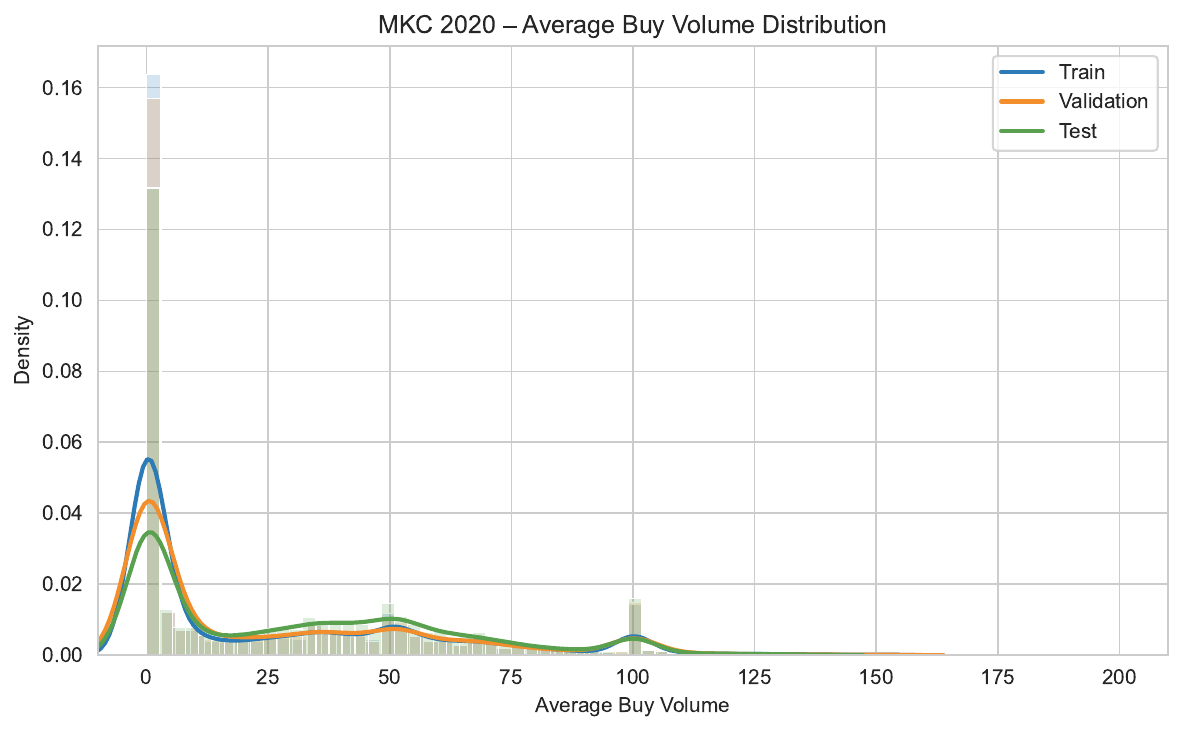} &
            \includegraphics[width=0.16\textwidth]{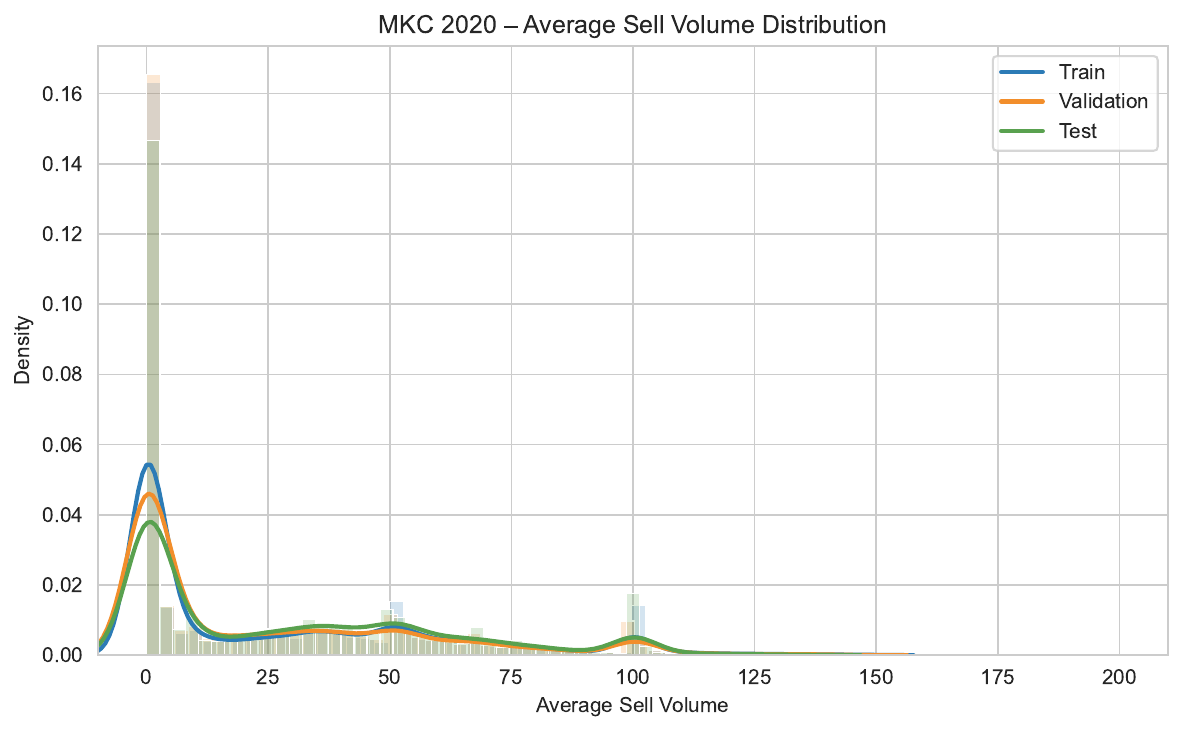} &
            \includegraphics[width=0.16\textwidth]{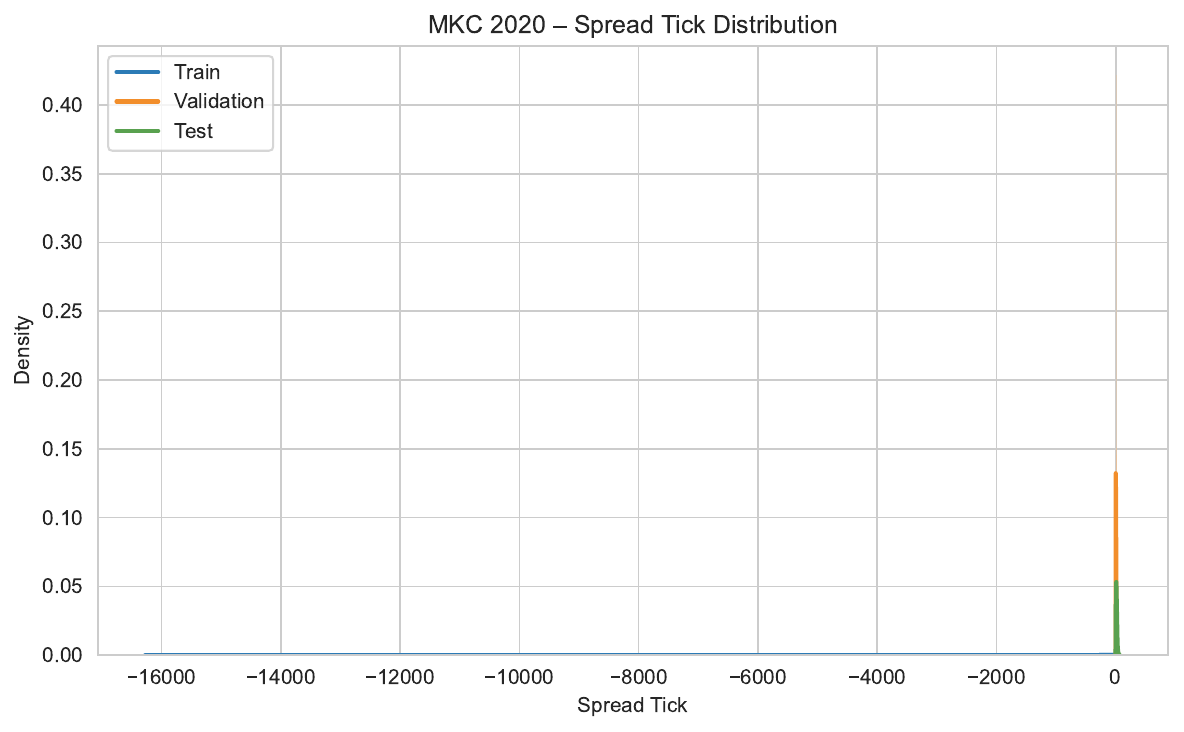} &
            \includegraphics[width=0.16\textwidth]{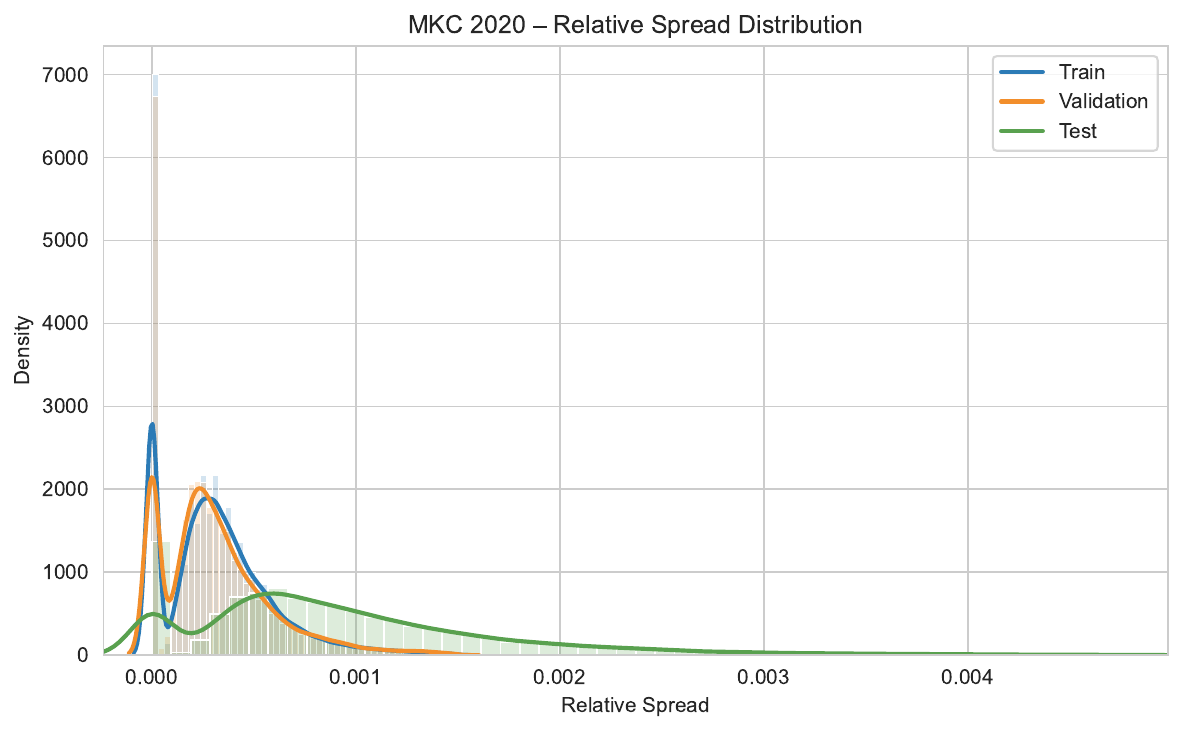} \\
            \scriptsize Buy count & \scriptsize Sell count & \scriptsize Average buy size & \scriptsize Average sell size & \scriptsize Spread (ticks) & \scriptsize Relative spread \\
            \includegraphics[width=0.16\textwidth]{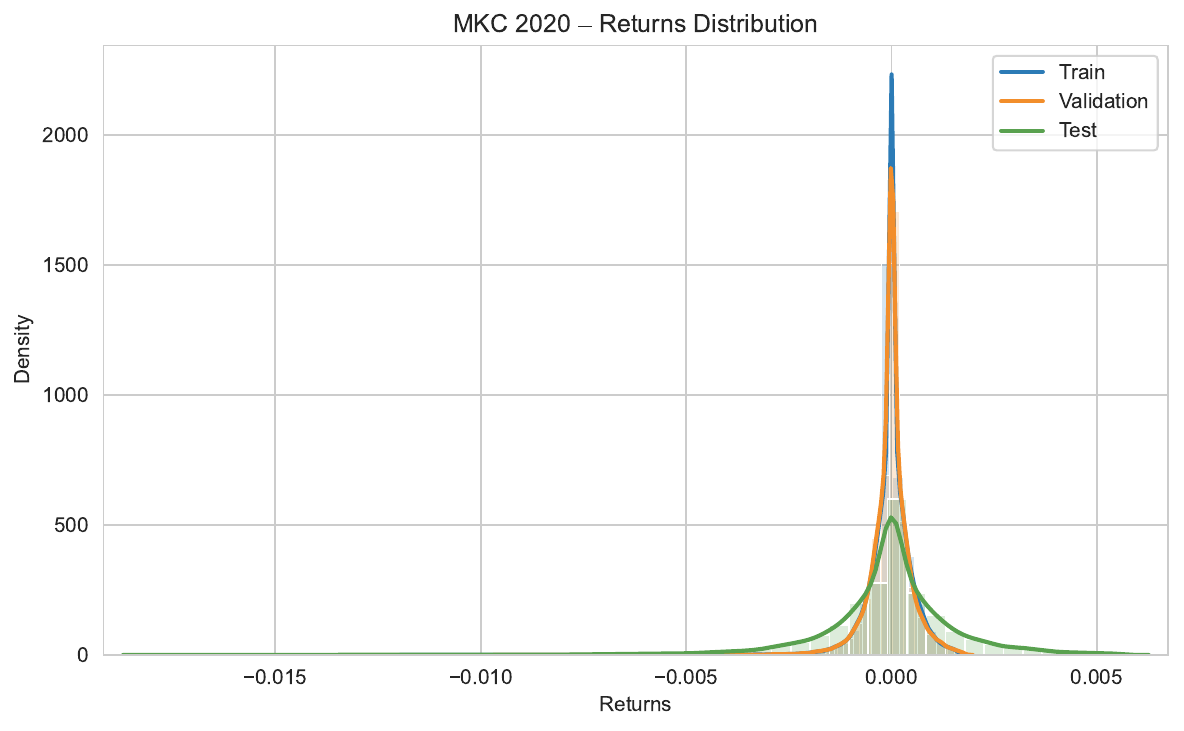} &
            \includegraphics[width=0.16\textwidth]{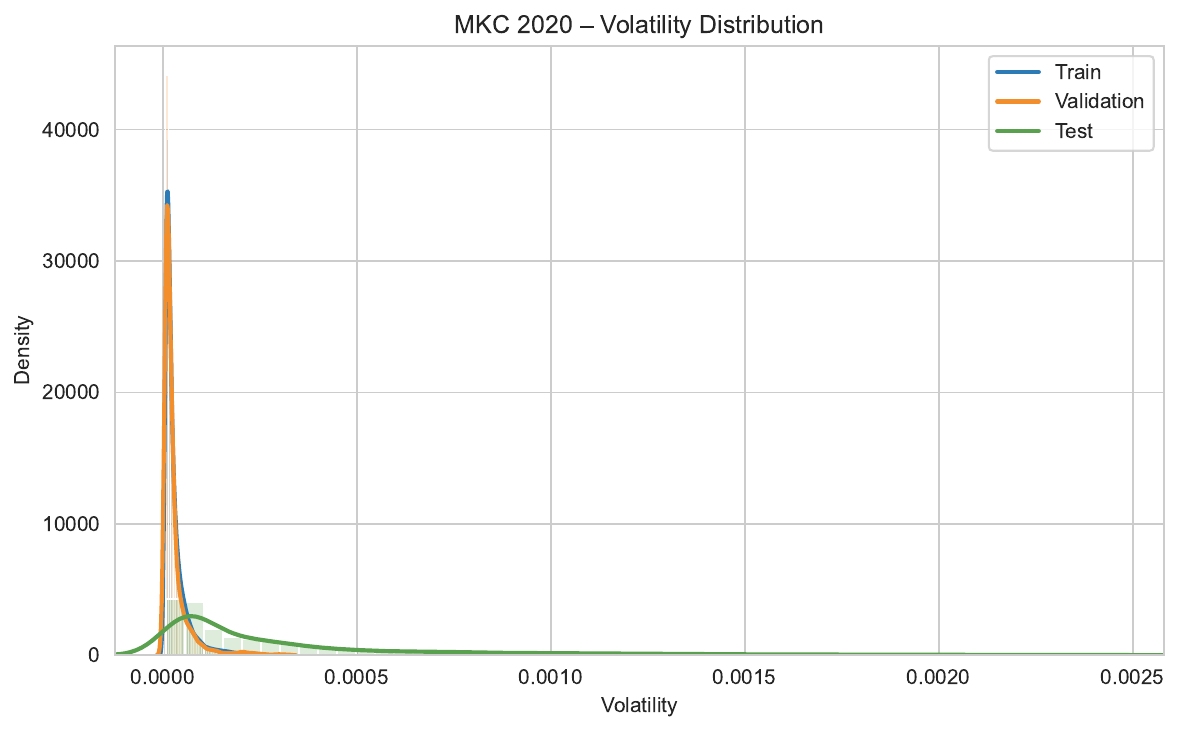} &
            \includegraphics[width=0.16\textwidth]{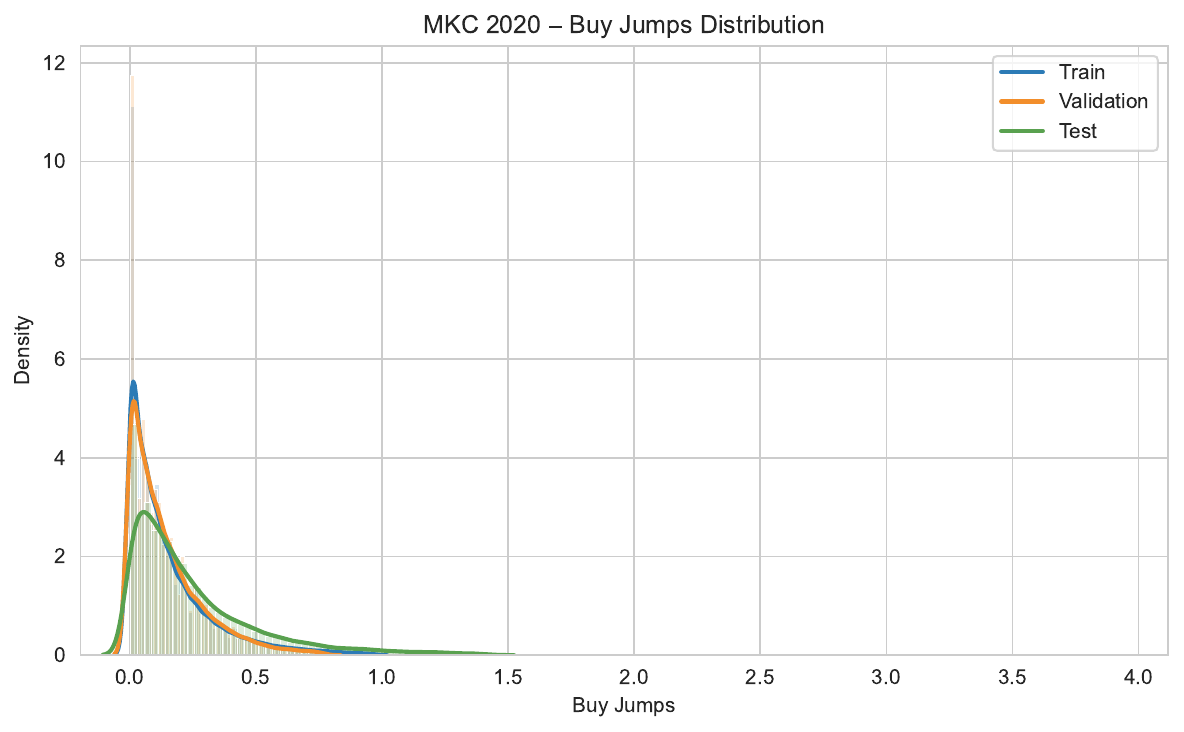} &
            \includegraphics[width=0.16\textwidth]{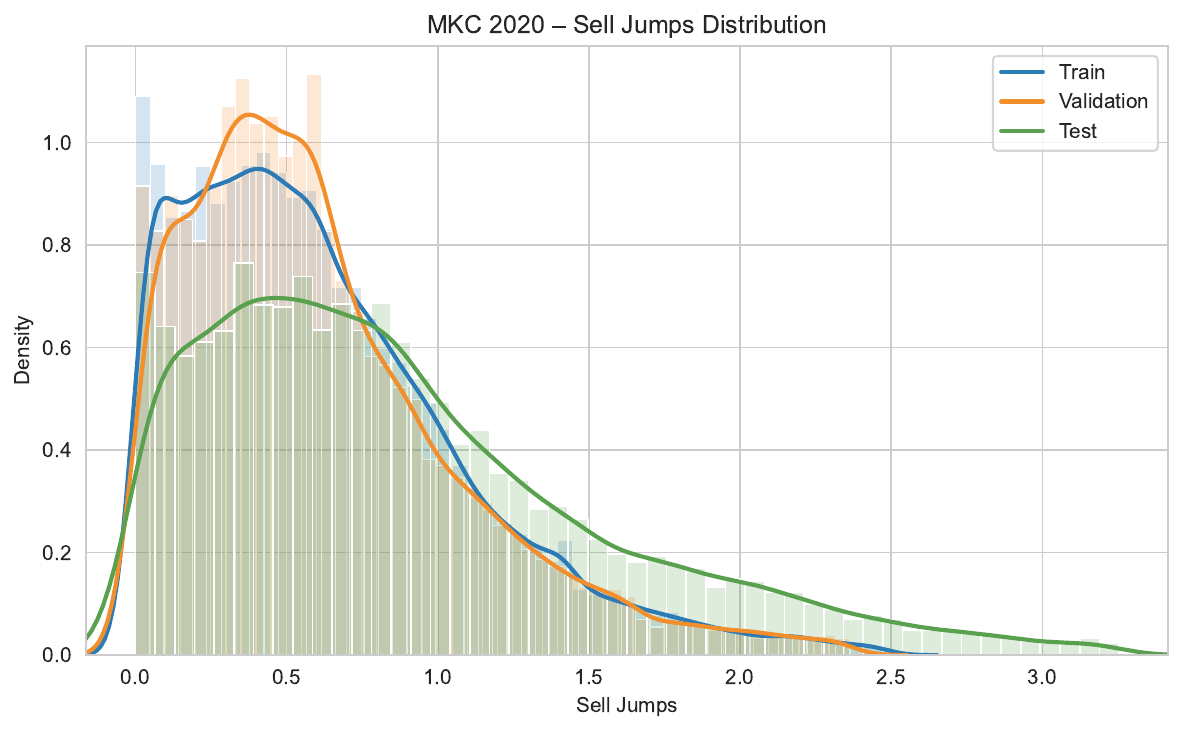} &
            \includegraphics[width=0.16\textwidth]{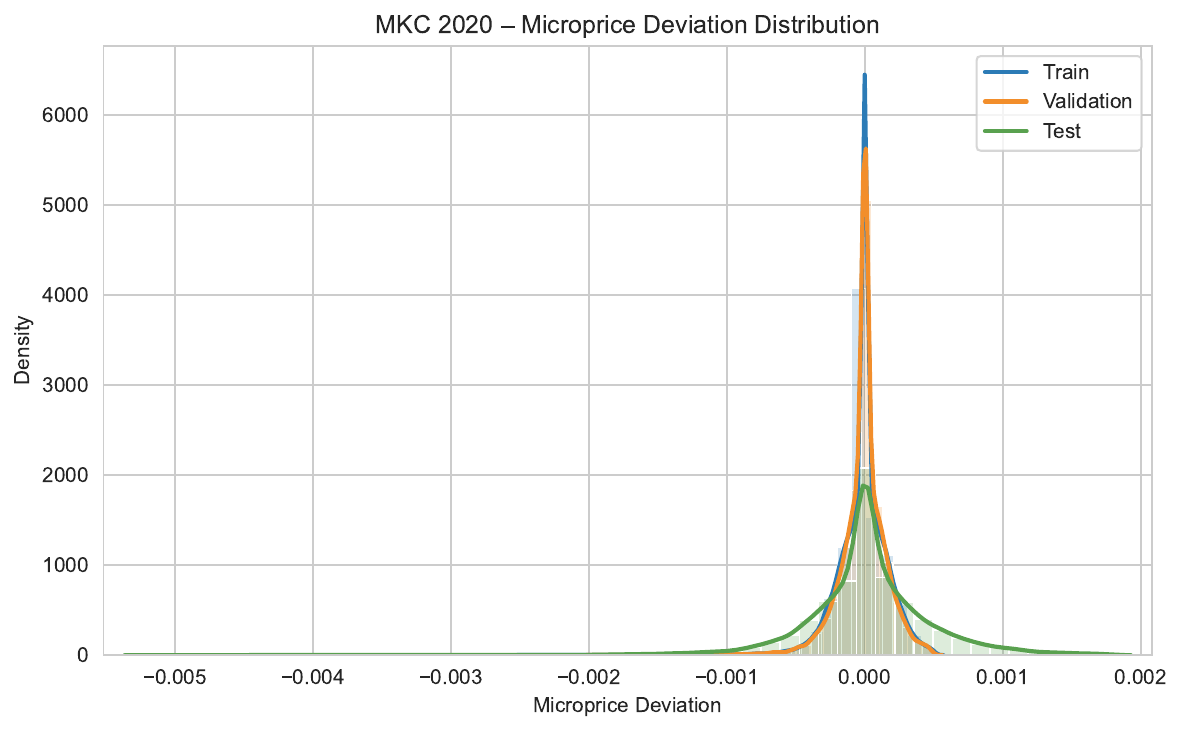} &
            \includegraphics[width=0.16\textwidth]{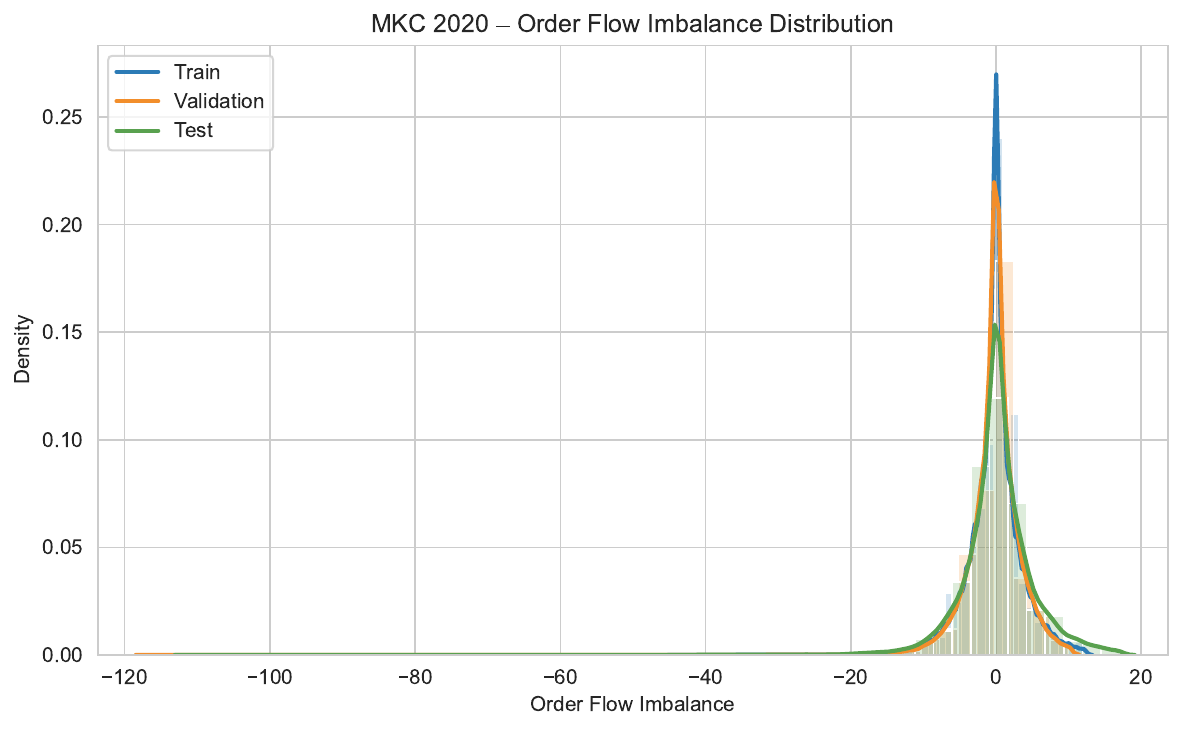} \\
            \scriptsize Returns & \scriptsize Volatility & \scriptsize Buy-order excitation & \scriptsize Sell-order excitation & \scriptsize Microprice deviation & \scriptsize OFI \\
        \end{tabular}
    \end{minipage}}
    \caption{MKC, 2020.}
    \end{subfigure}

    \caption{Train, validation, and test distributions of all state variables for MKC in 2019 and 2020. The comparison makes the stronger cross-split distributional shift in 2020 visually apparent.}
    \label{fig:full_state_dist_mkc}
\end{figure}

\begin{figure}[H]
    \centering
    \begin{subfigure}{\linewidth}
    \centering
    \makebox[\textwidth][c]{\begin{minipage}{1.12\textwidth}\centering
        \setlength{\tabcolsep}{0.5pt}
        \begin{tabular}{cccccc}
            \includegraphics[width=0.16\textwidth]{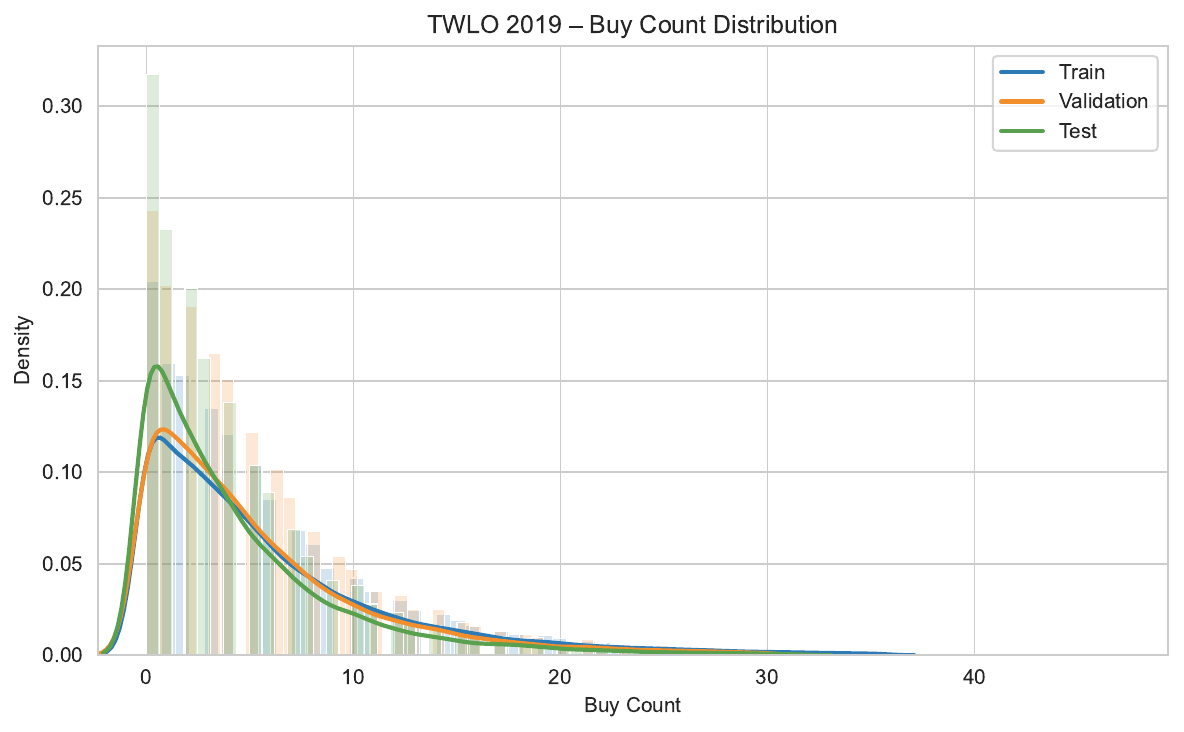} &
            \includegraphics[width=0.16\textwidth]{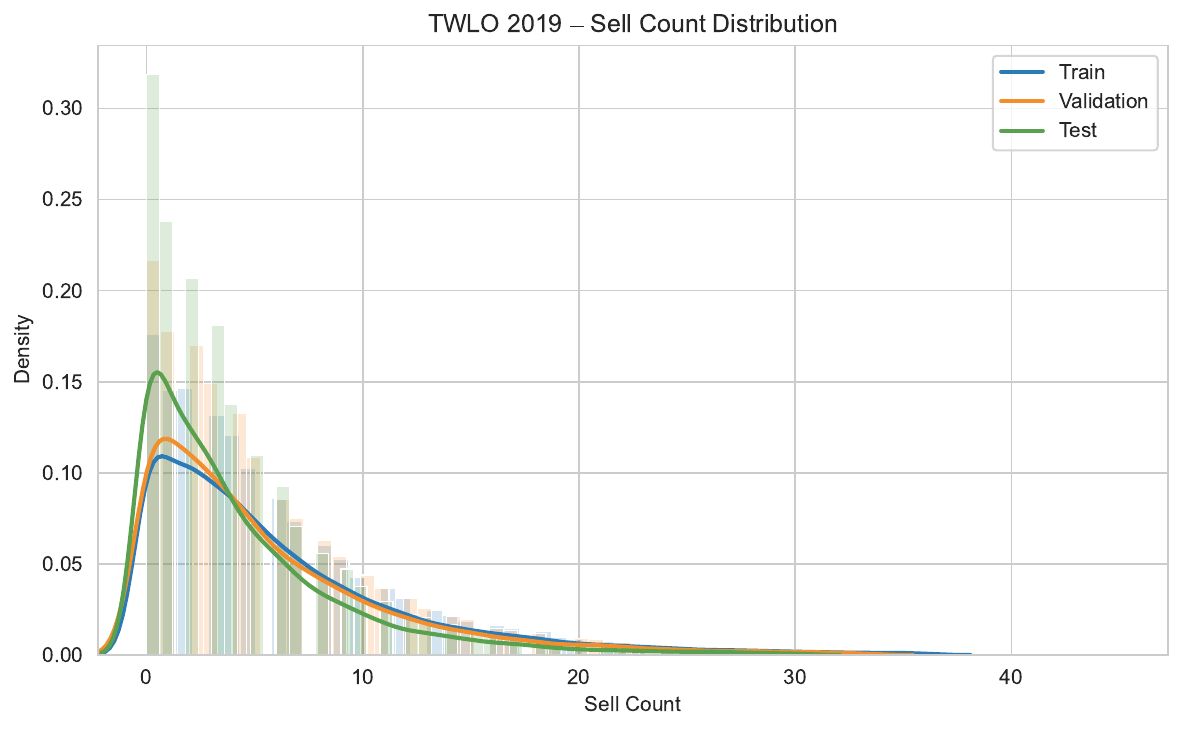} &
            \includegraphics[width=0.16\textwidth]{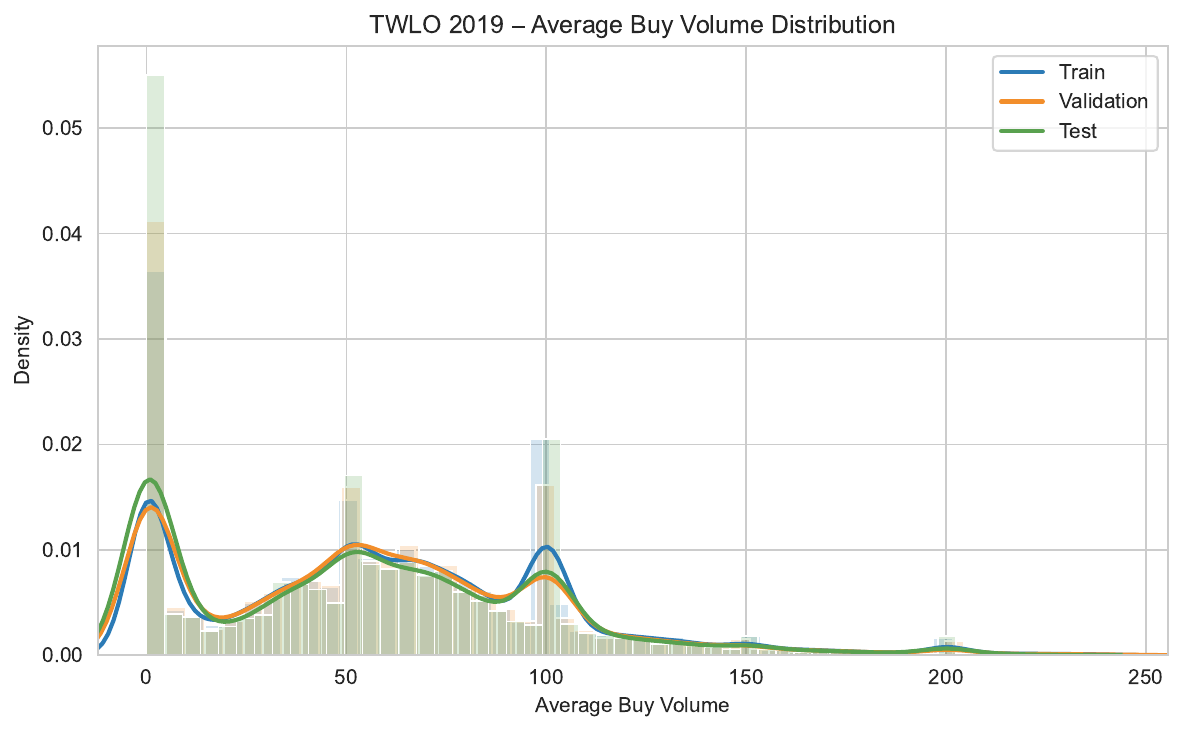} &
            \includegraphics[width=0.16\textwidth]{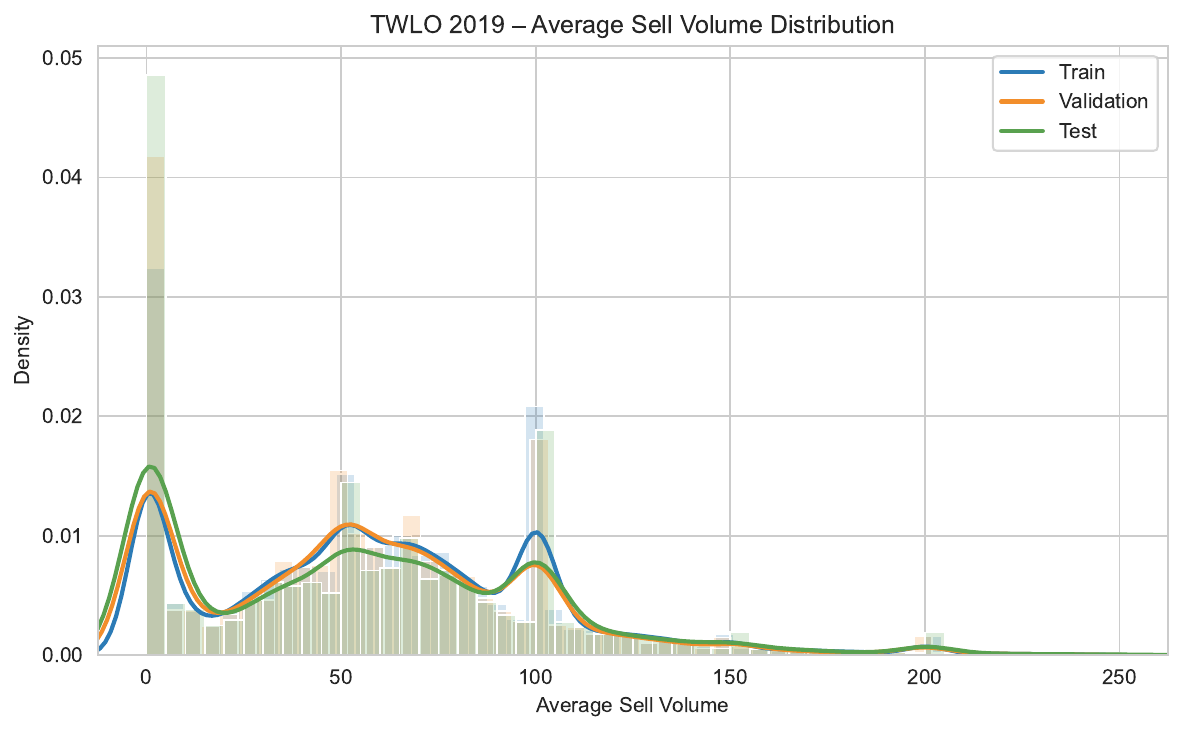} &
            \includegraphics[width=0.16\textwidth]{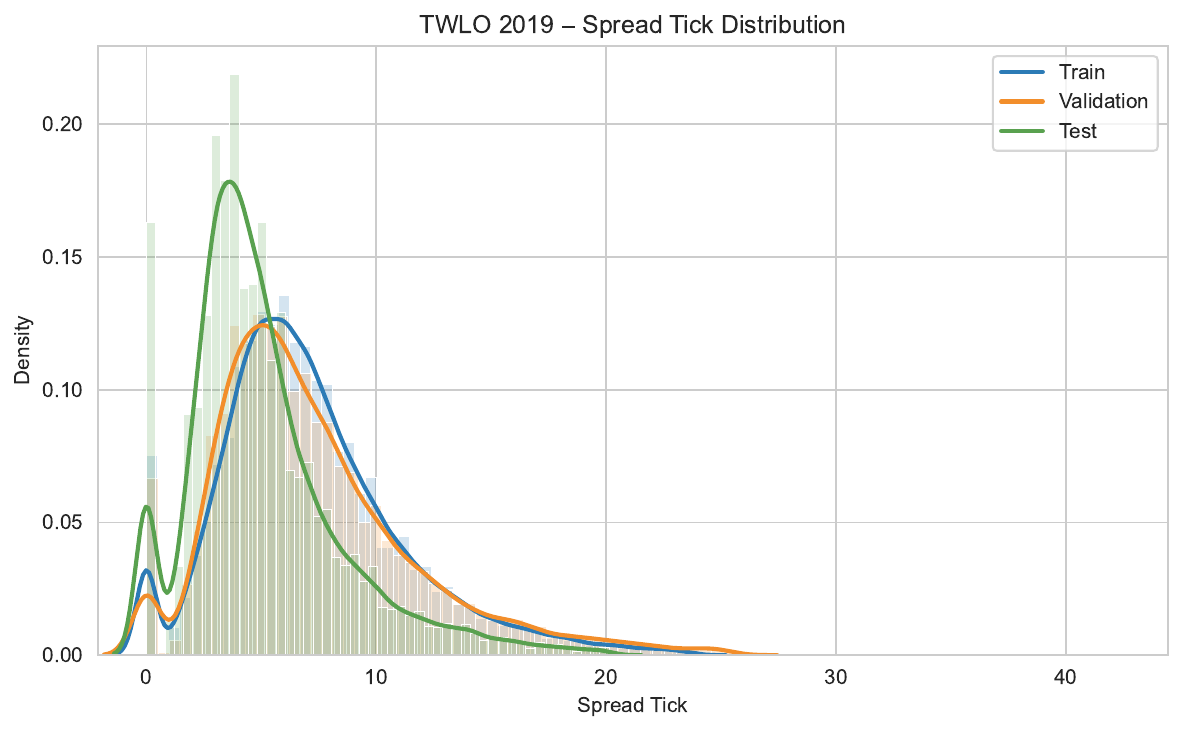} &
            \includegraphics[width=0.16\textwidth]{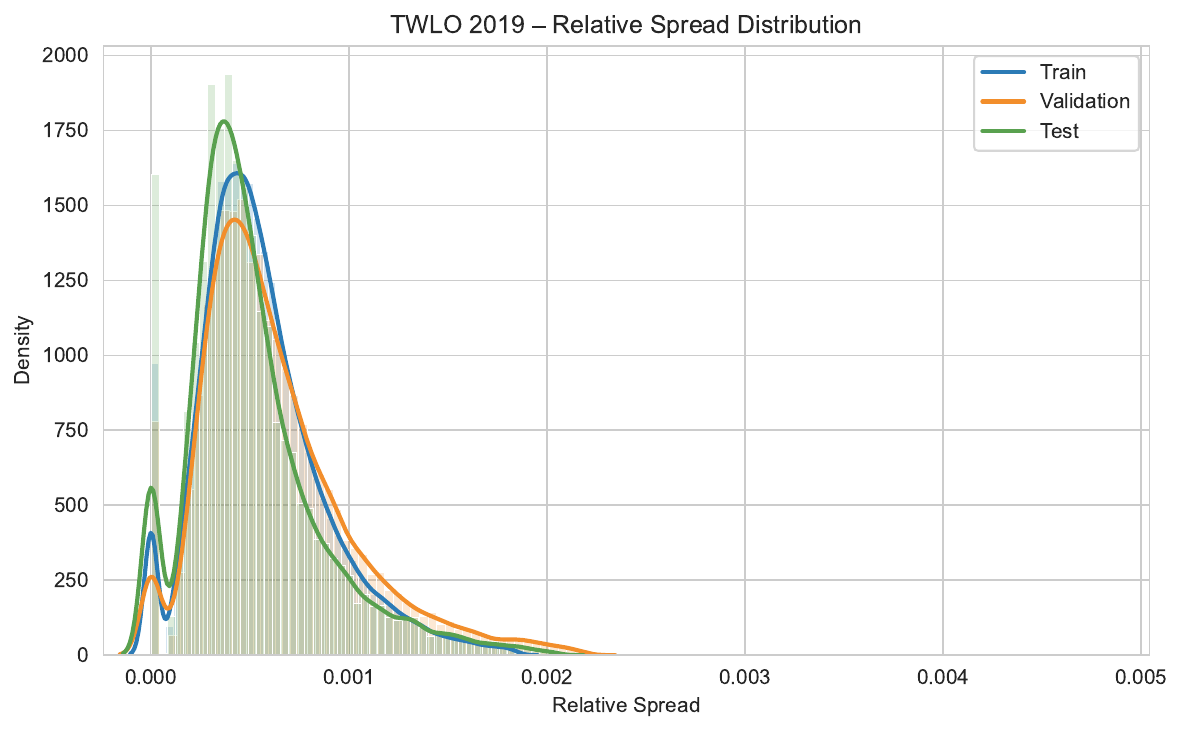} \\
            \scriptsize Buy count & \scriptsize Sell count & \scriptsize Average buy size & \scriptsize Average sell size & \scriptsize Spread (ticks) & \scriptsize Relative spread \\
            \includegraphics[width=0.16\textwidth]{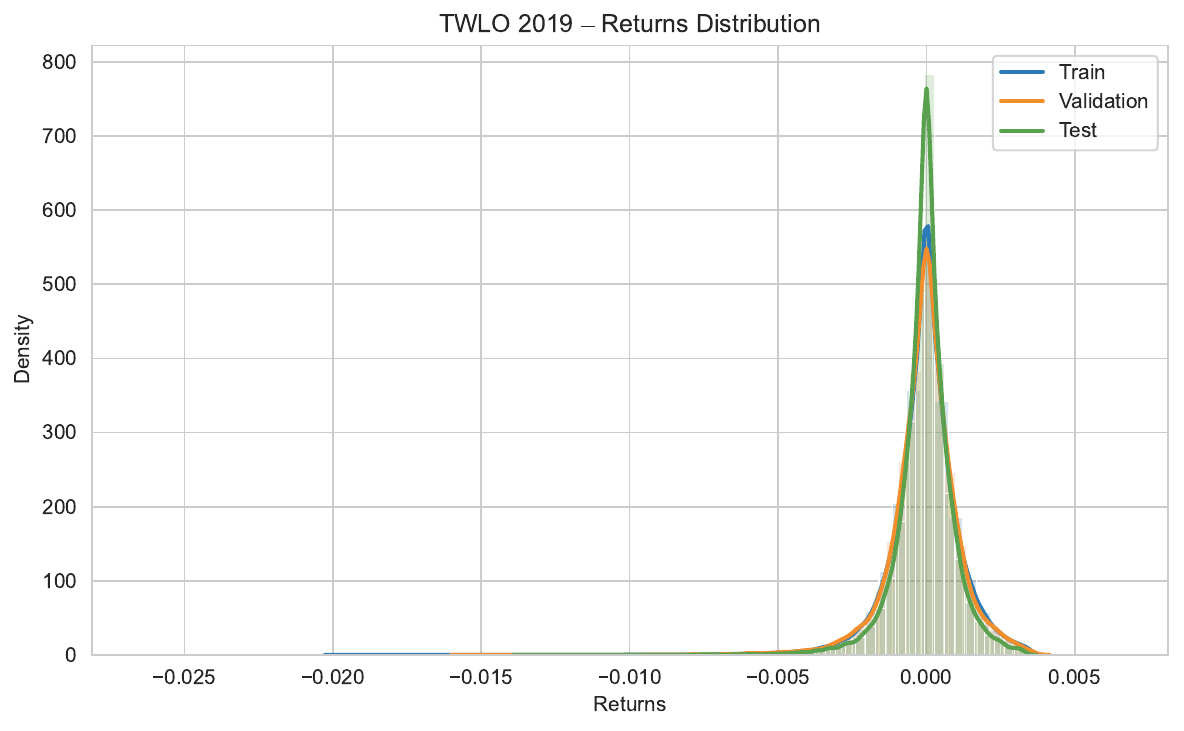} &
            \includegraphics[width=0.16\textwidth]{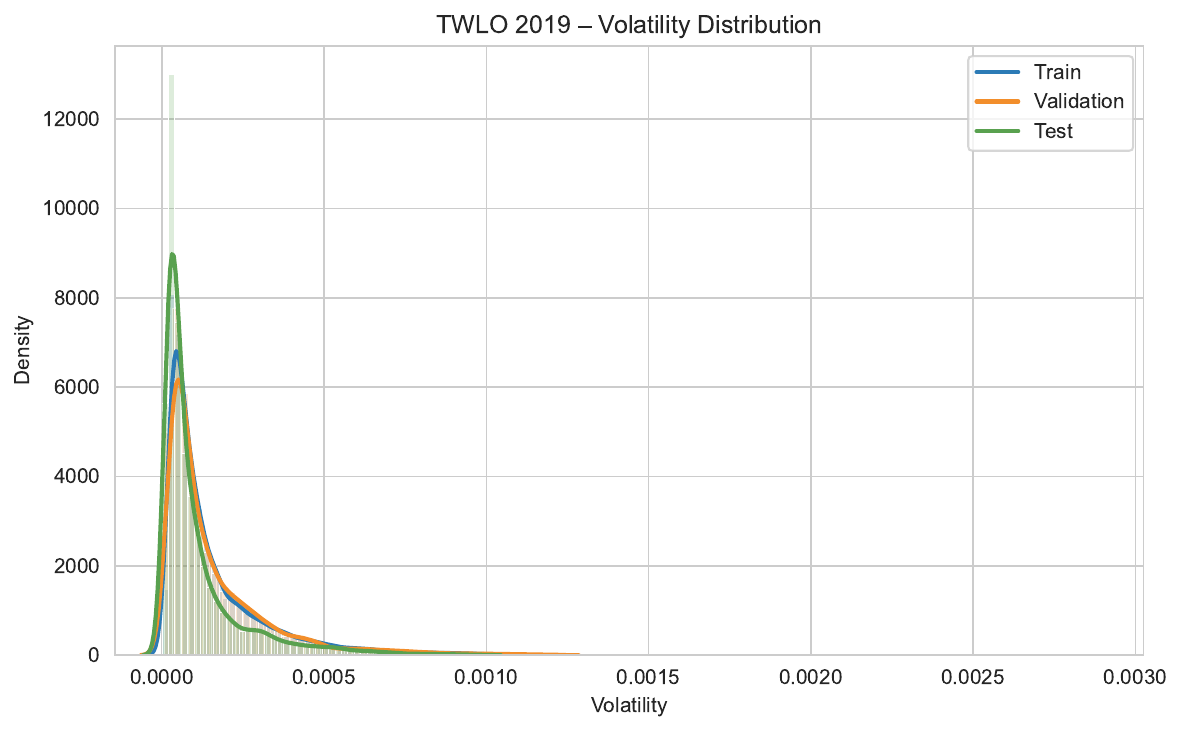} &
            \includegraphics[width=0.16\textwidth]{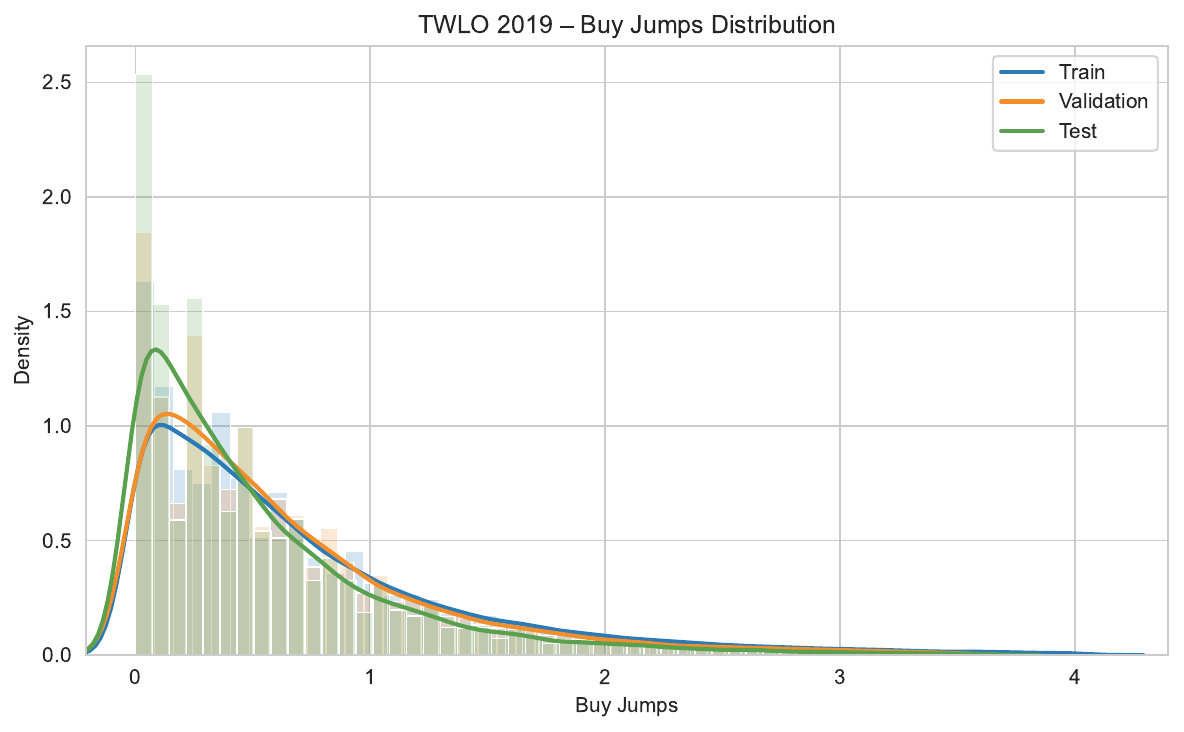} &
            \includegraphics[width=0.16\textwidth]{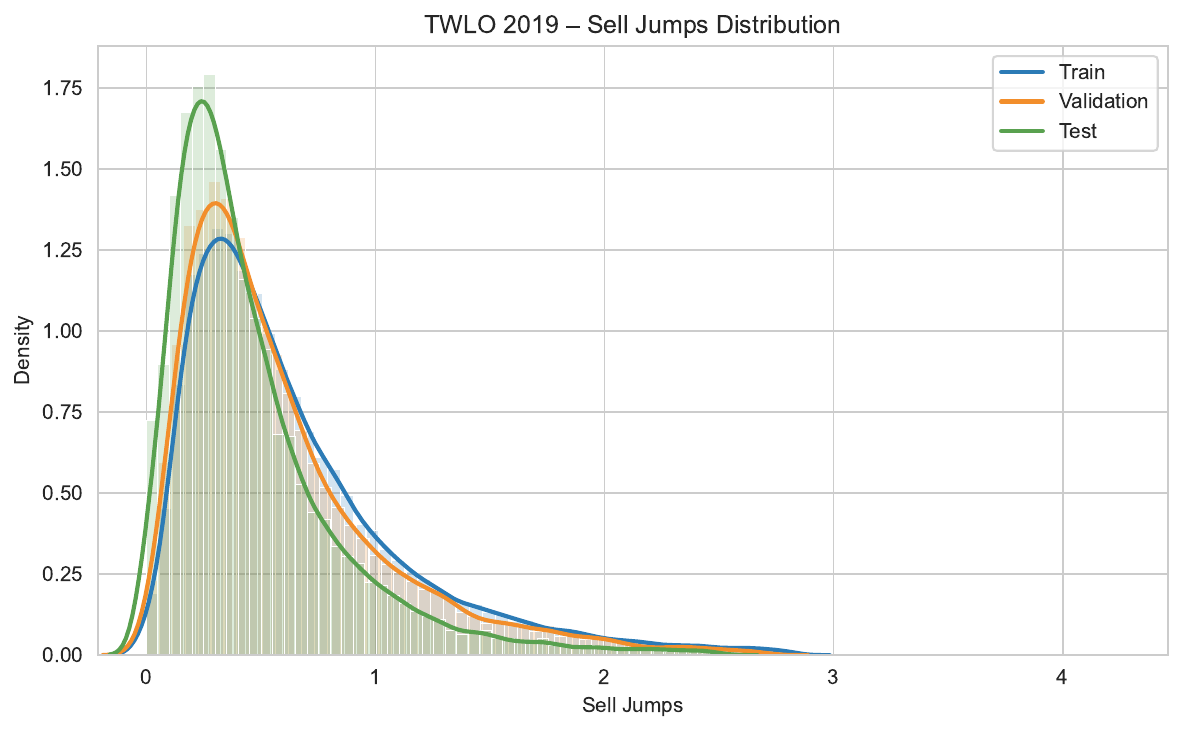} &
            \includegraphics[width=0.16\textwidth]{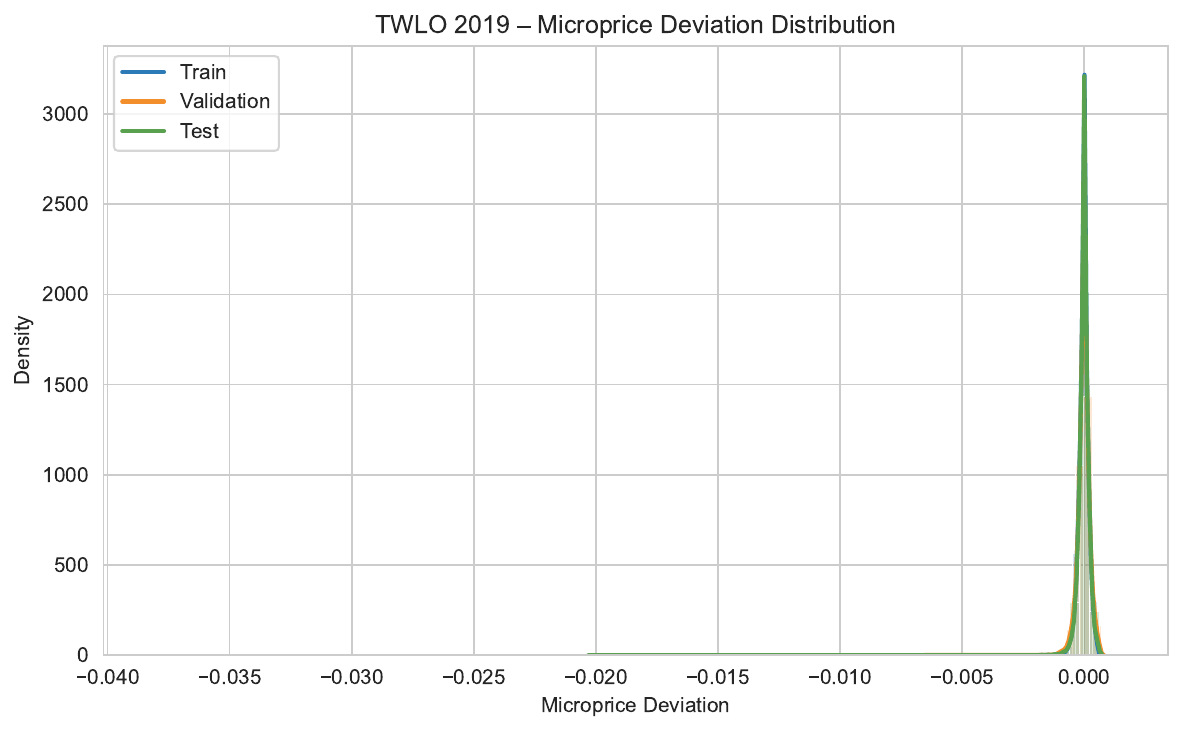} &
            \includegraphics[width=0.16\textwidth]{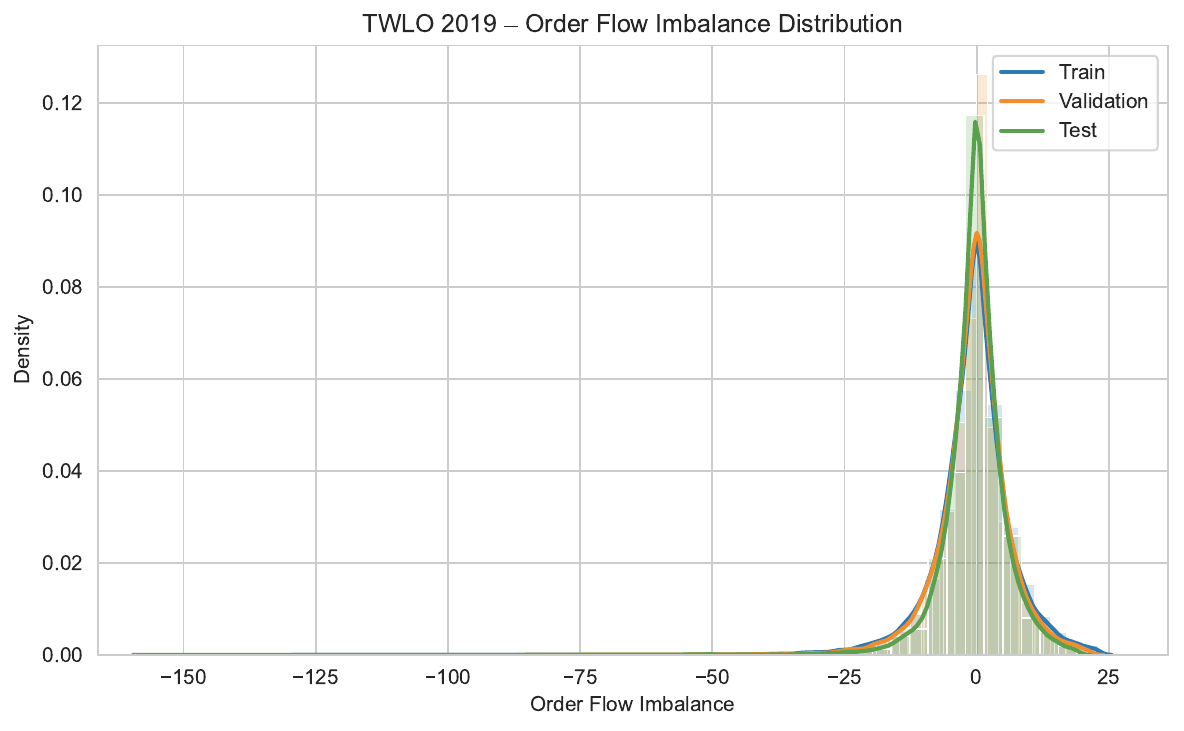} \\
            \scriptsize Returns & \scriptsize Volatility & \scriptsize Buy-order excitation & \scriptsize Sell-order excitation & \scriptsize Microprice deviation & \scriptsize OFI \\
        \end{tabular}
    \end{minipage}}
    \caption{TWLO, 2019.}
    \end{subfigure}

    \vspace{0.2em}

    \begin{subfigure}{\linewidth}
    \centering
    \makebox[\textwidth][c]{\begin{minipage}{1.12\textwidth}\centering
        \setlength{\tabcolsep}{0.5pt}
        \begin{tabular}{cccccc}
            \includegraphics[width=0.16\textwidth]{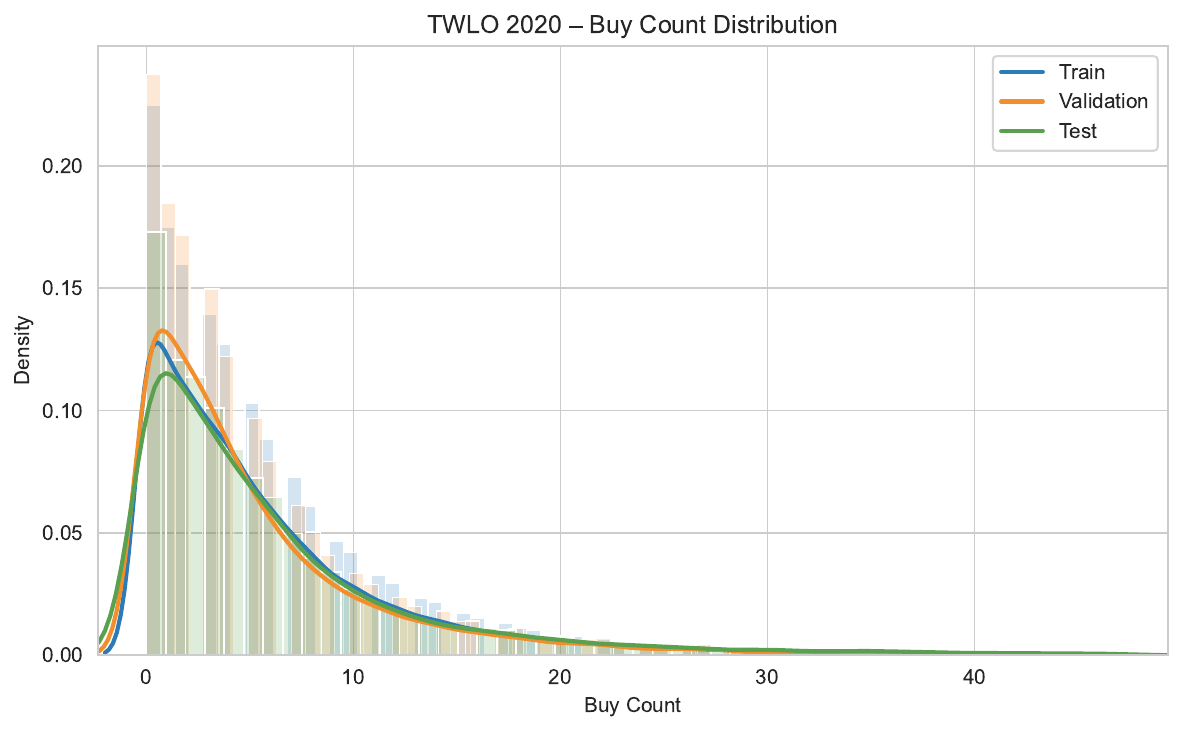} &
            \includegraphics[width=0.16\textwidth]{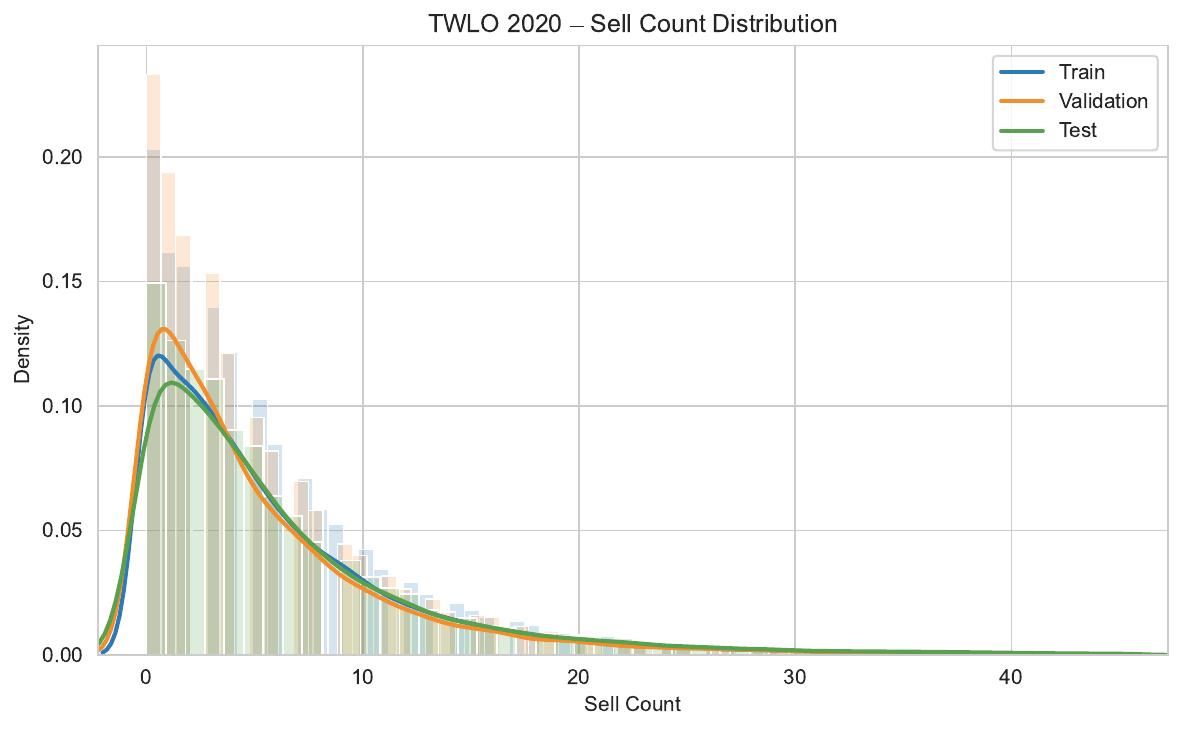} &
            \includegraphics[width=0.16\textwidth]{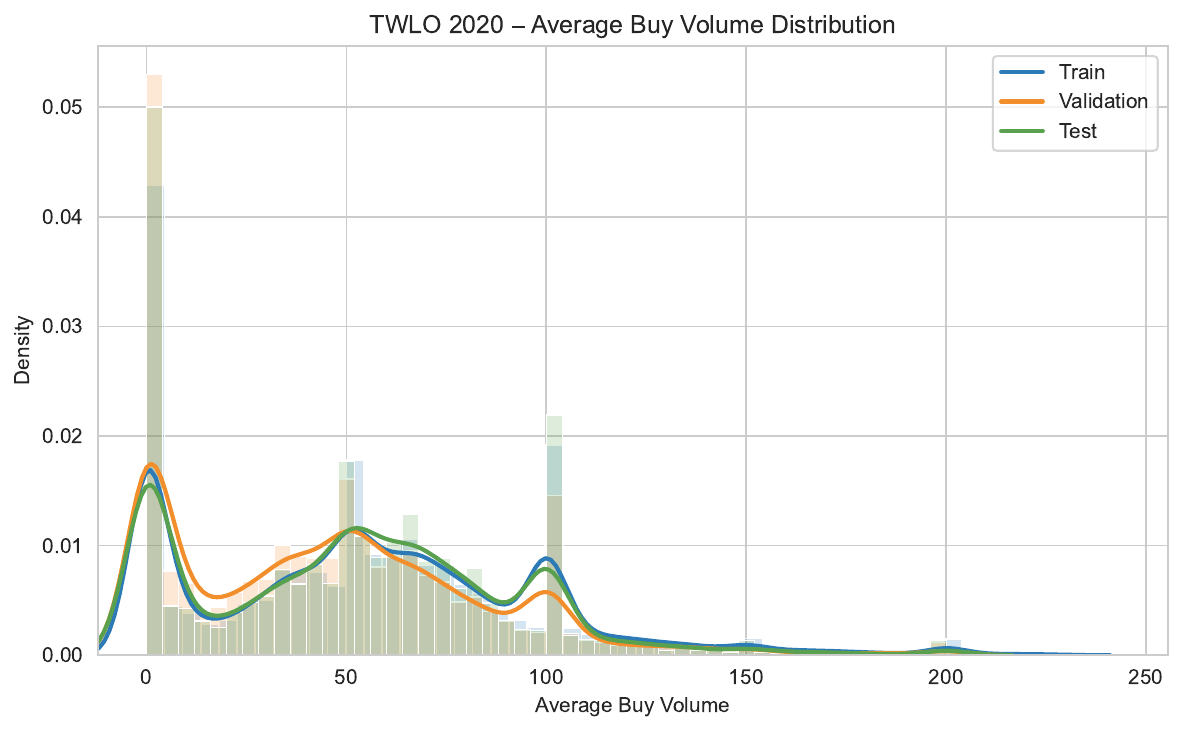} &
            \includegraphics[width=0.16\textwidth]{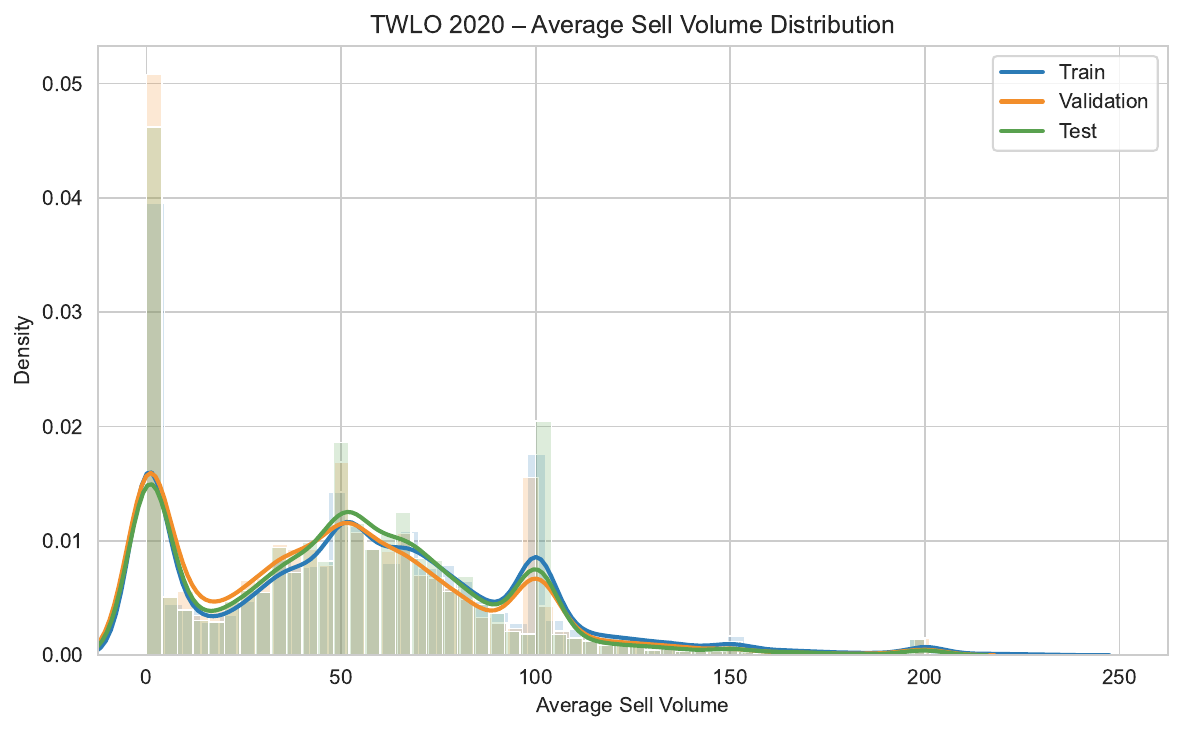} &
            \includegraphics[width=0.16\textwidth]{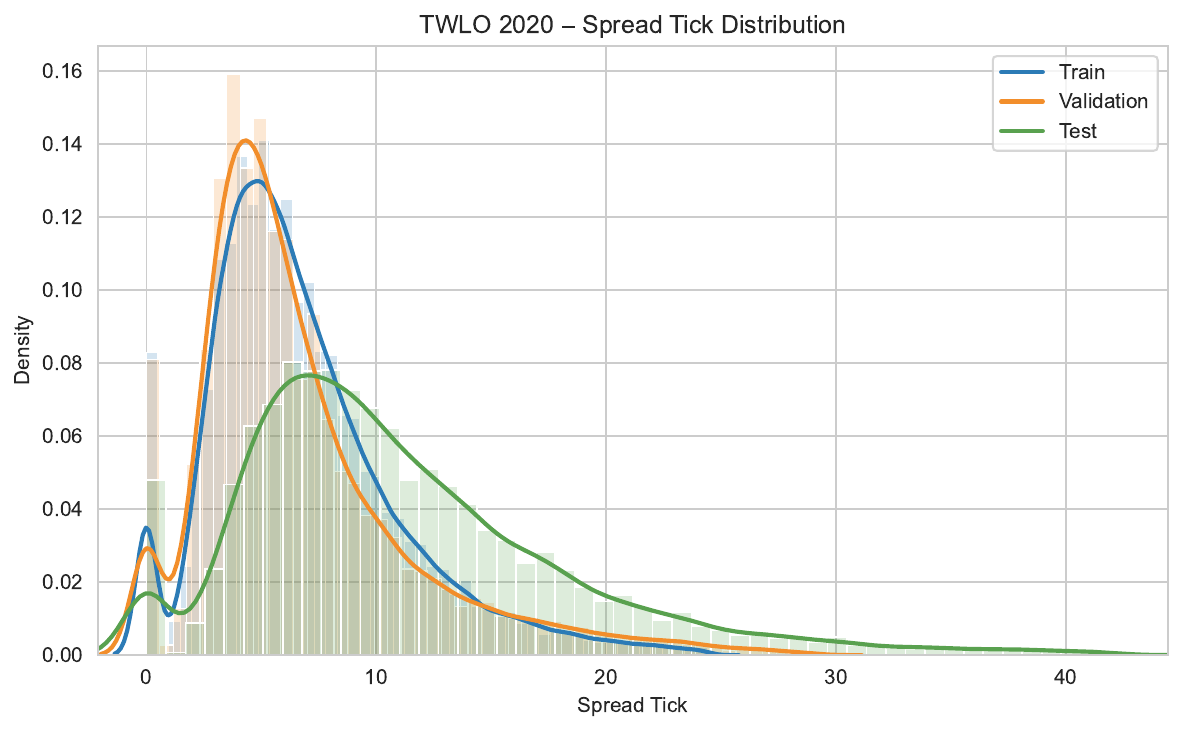} &
            \includegraphics[width=0.16\textwidth]{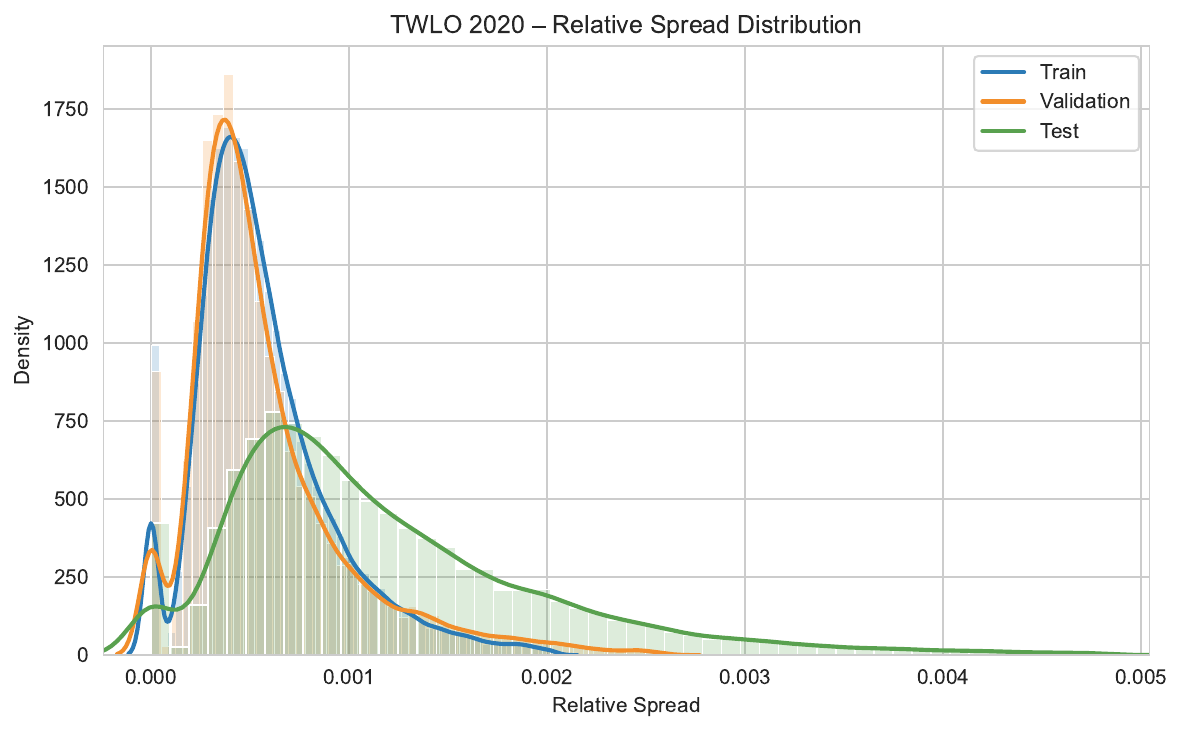} \\
            \scriptsize Buy count & \scriptsize Sell count & \scriptsize Average buy size & \scriptsize Average sell size & \scriptsize Spread (ticks) & \scriptsize Relative spread \\
            \includegraphics[width=0.16\textwidth]{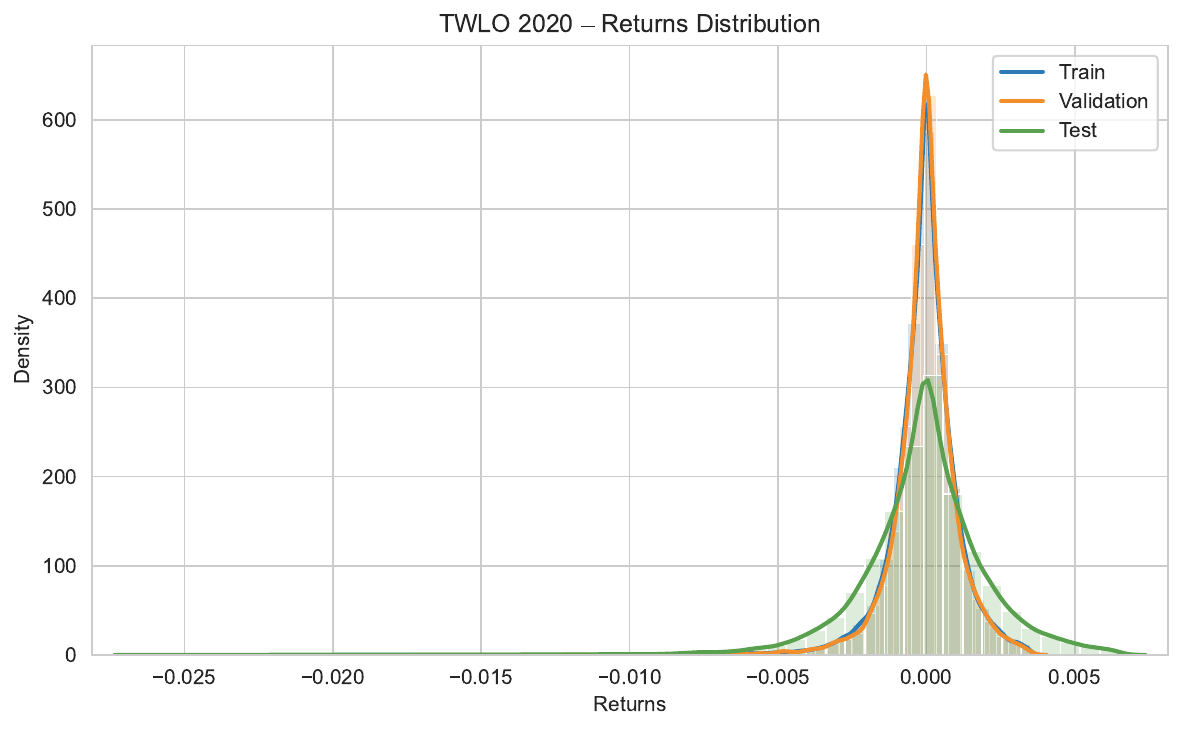} &
            \includegraphics[width=0.16\textwidth]{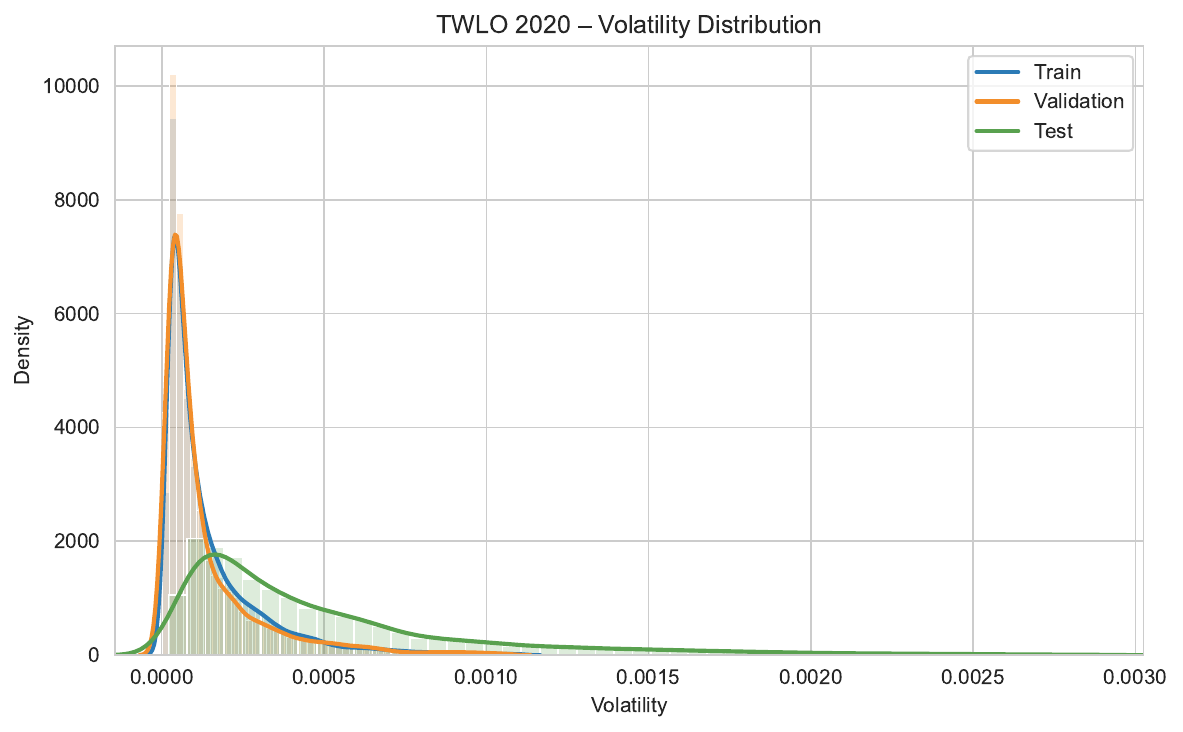} &
            \includegraphics[width=0.16\textwidth]{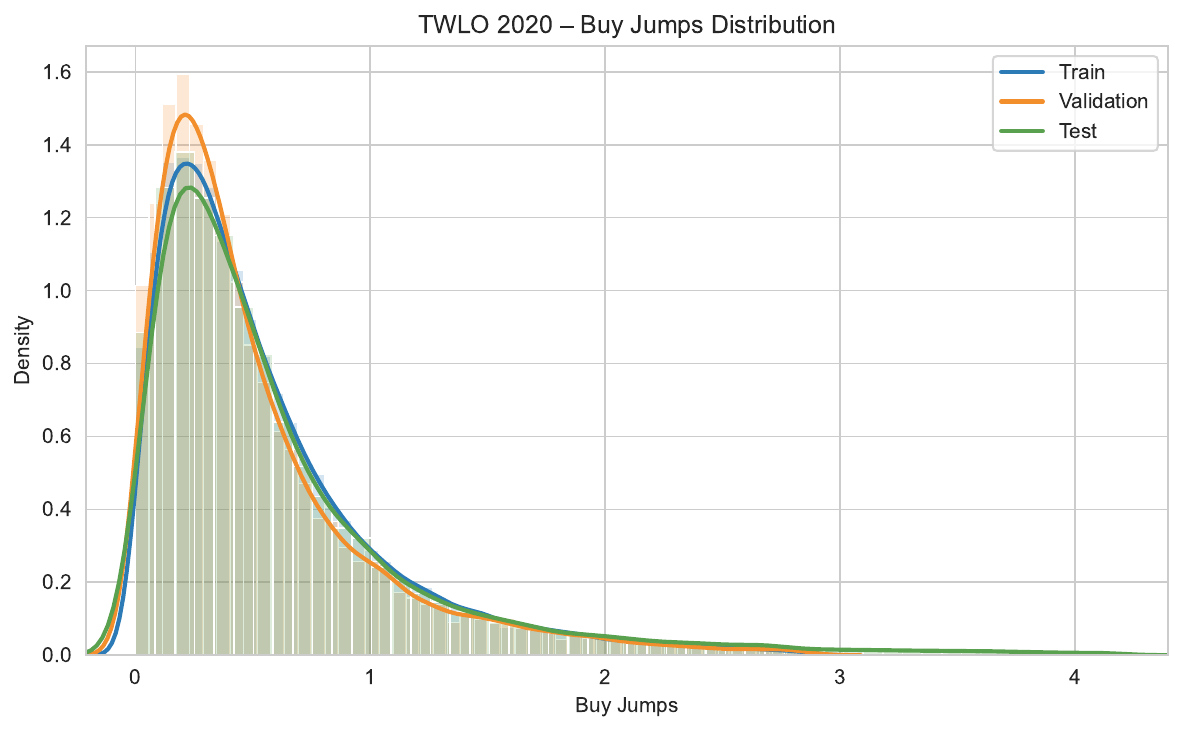} &
            \includegraphics[width=0.16\textwidth]{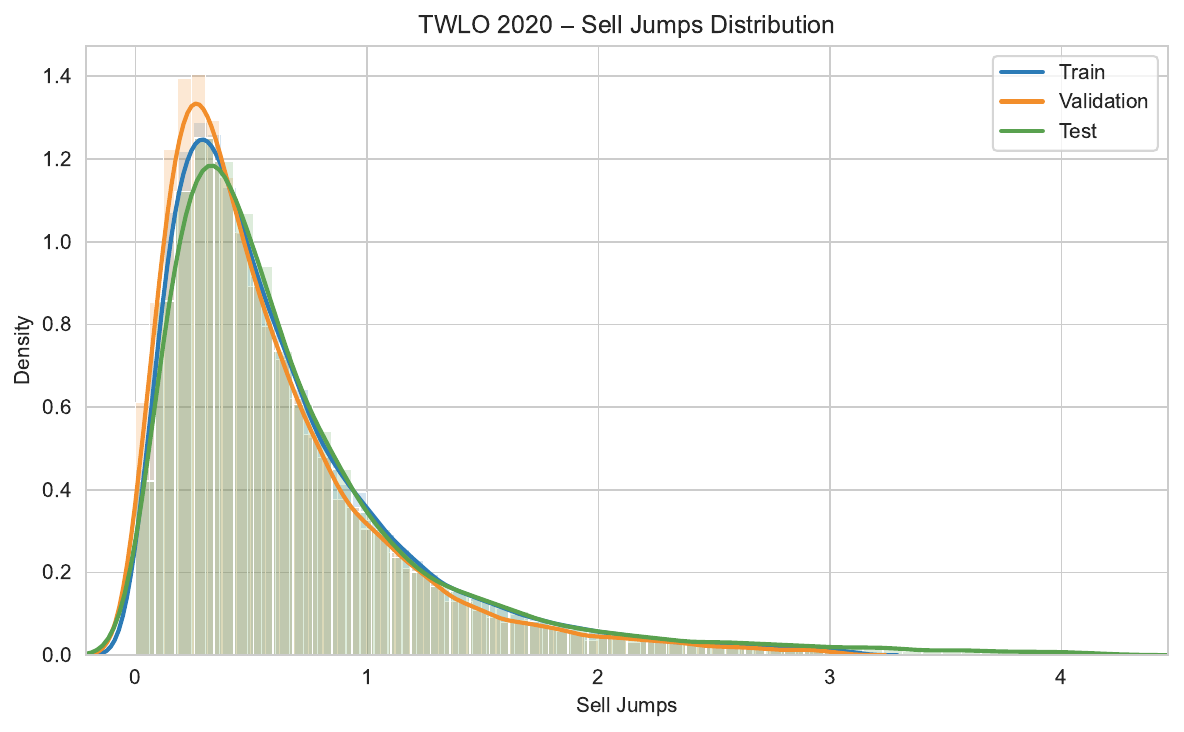} &
            \includegraphics[width=0.16\textwidth]{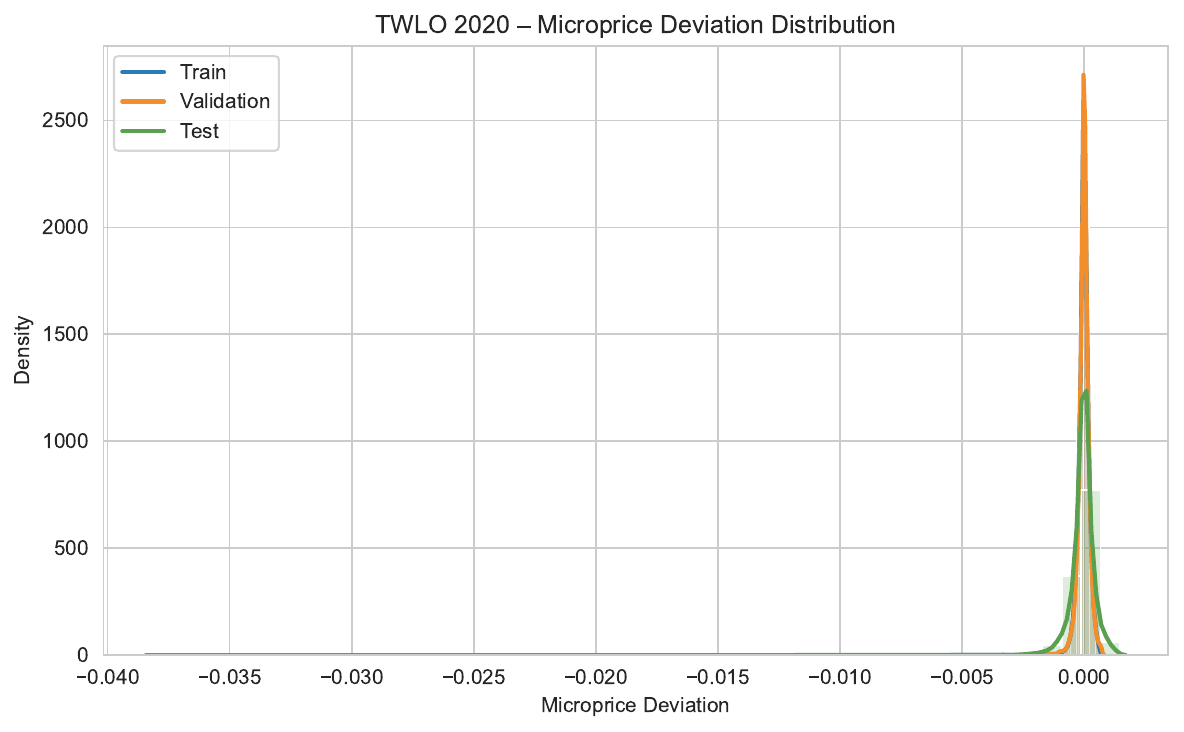} &
            \includegraphics[width=0.16\textwidth]{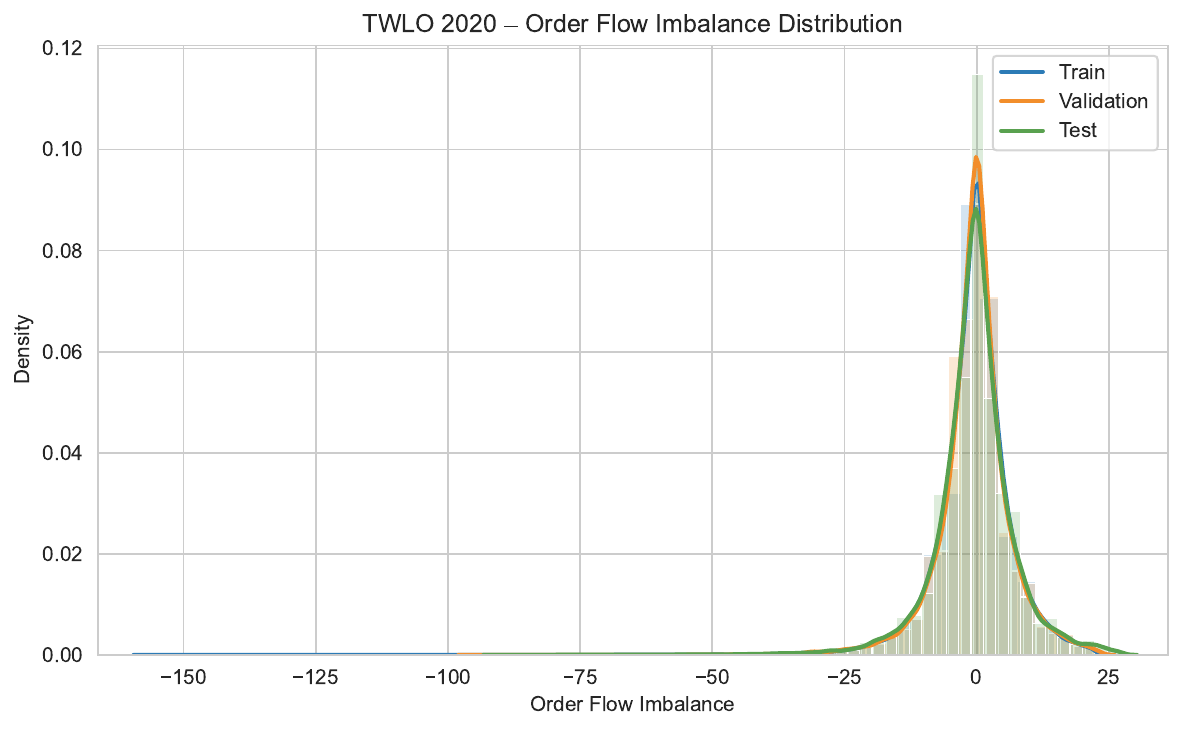} \\
            \scriptsize Returns & \scriptsize Volatility & \scriptsize Buy-order excitation & \scriptsize Sell-order excitation & \scriptsize Microprice deviation & \scriptsize OFI \\
        \end{tabular}
    \end{minipage}}
    \caption{TWLO, 2020.}
    \end{subfigure}

    \caption{Train, validation, and test distributions of all state variables for TWLO in 2019 and 2020. The comparison makes the stronger cross-split distributional shift in 2020 visually apparent.}
    \label{fig:full_state_dist_twlo}
\end{figure}

\subsection[Sensitivity to Uncertainty Tolerance and Action Robustness]{Sensitivity to Uncertainty Tolerance (\texorpdfstring{$\bar\varepsilon$}{epsilon-bar}) and Action Robustness (\texorpdfstring{$\delta$}{delta})}

Figures~\ref{fig:sharpe_heatmap_real2} and \ref{fig:pnl_heatmap_real} report the test-period Sharpe ratio and mean P\&L over the $(\bar\varepsilon,\delta)$ grid for all four stocks in the empirical universe, where $\bar\varepsilon$ is the uncertainty-tolerance parameter and $\delta$ is the action-robustness parameter. The amount of outperformance relative to the greedy benchmark varies materially across names. For AAPL, there is a broad region of the grid in both 2019 and 2020 where the robust agent outperforms the greedy benchmark, so the gain is not tied to a single finely tuned hyperparameter choice. TSLA is more selective in 2019, but in 2020 the set of outperforming robust configurations is noticeably larger, consistent with the more volatile COVID-period environment where robustness is more valuable. MKC shows a narrower improvement region: only a few $(\bar\varepsilon,\delta)$ combinations dominate the greedy benchmark, so the benefits of robustness are present but less pervasive. TWLO again exhibits a comparatively broad outperforming region in both years, indicating that the robust agent transfers well across a wide range of ambiguity specifications.

Across all four stocks, uncertainty tolerance and action robustness play structurally different roles. Increasing $\bar\varepsilon$ raises uncertainty tolerance by enlarging the ambiguity set directly, so the agent becomes more risk-averse in the usual sense: it hedges against a wider class of adverse transition laws, which typically leads to less aggressive quoting, lower P\&L volatility, and a higher Sharpe ratio up to a point. By contrast, $\delta$ does not primarily control the size of the ambiguity set. Rather, as the action-robustness parameter, it controls the degree of entropic regularization in the Sinkhorn transport problem and therefore the concentration of the worst-case transport plan. For smaller $\delta$, the worst-case measure can concentrate more tightly on particularly adverse next-state realizations; for larger $\delta$, the transport plan is forced to be more diffuse and to distribute mass across a broader set of candidate states. In this sense, a larger $\delta$ produces a more spatially dispersed stress scenario rather than a uniformly more adverse one.

This interpretation helps explain the heatmaps. A larger $\bar\varepsilon$ can be viewed as higher uncertainty tolerance: the agent admits a broader range of adverse transition distortions as plausible and therefore guards against a wider set of model deviations. Its effect on Sharpe is therefore comparatively stable. By contrast, $\delta$ is better interpreted as action robustness, namely how strongly and how diffusely the policy responds to adverse scenarios within that uncertainty set. Moving from a very small to a moderate $\delta$ can improve Sharpe because the policy becomes more conservative in a useful way, typically giving up some mean P\&L in exchange for lower dispersion of outcomes. But once $\delta$ becomes sufficiently large, that action-robust response may become too diffuse, so the policy is no longer calibrated sharply to the most decision-relevant adverse configurations. The relation between $\delta$ and performance therefore need not be monotone. In this sense, $\bar\varepsilon$ governs how much uncertainty the agent is willing to entertain, whereas $\delta$ governs how that uncertainty is translated into the policy response.
Whether $\delta$ should be viewed as small or large depends on the asset’s volatility over the relevant period. This also helps explain why, for AAPL, the Sharpe ratio falls with $\bar{\varepsilon}$ when $\delta=1$ in 2019, but rises with $\bar{\varepsilon}$ in 2020. In practice, uncertainty tolerance and action robustness can be selected jointly using the validation Pareto frontier reported in Appendix~\ref{sec:appendix_val_test_frontier}.

\begin{figure}[H]
    \centering
    \makebox[\textwidth][c]{\begin{minipage}{1.16\textwidth}\centering
        \begin{subfigure}{0.245\linewidth}
            \includegraphics[width=\linewidth]{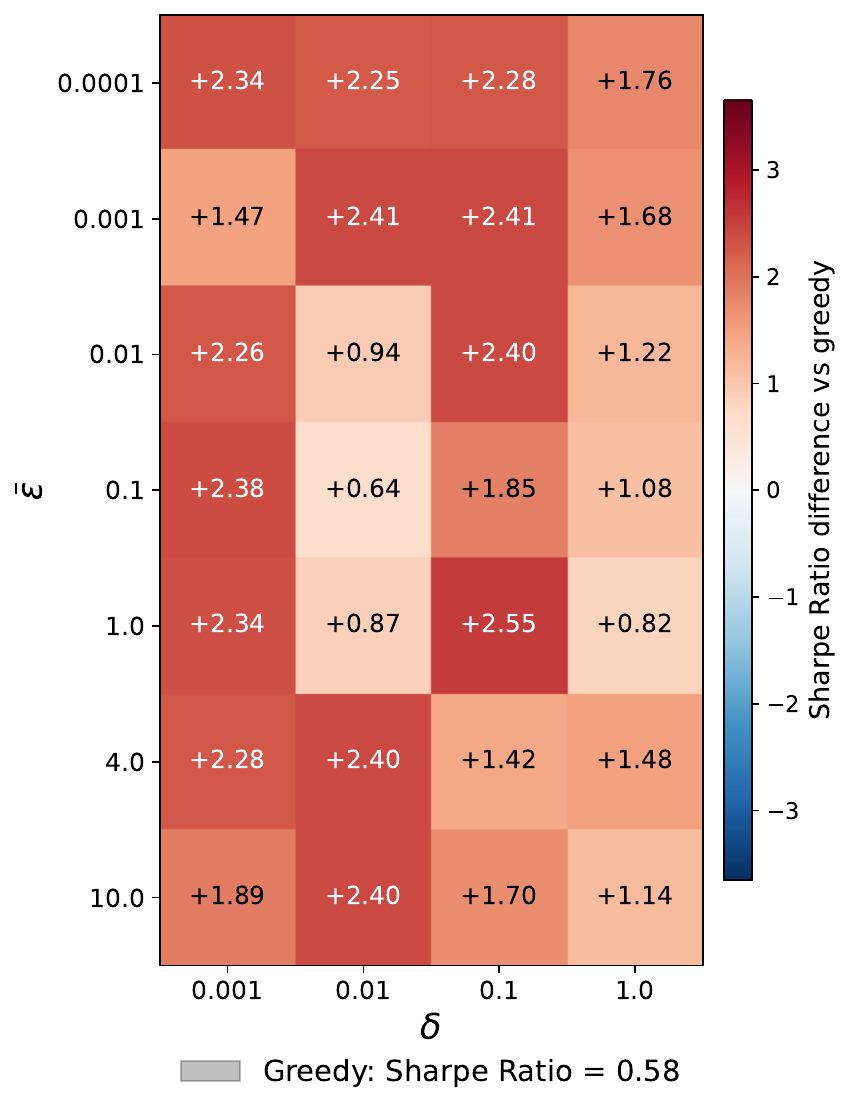}
            \caption{AAPL, 2019.}
        \end{subfigure}%
        \begin{subfigure}{0.245\linewidth}
            \includegraphics[width=\linewidth]{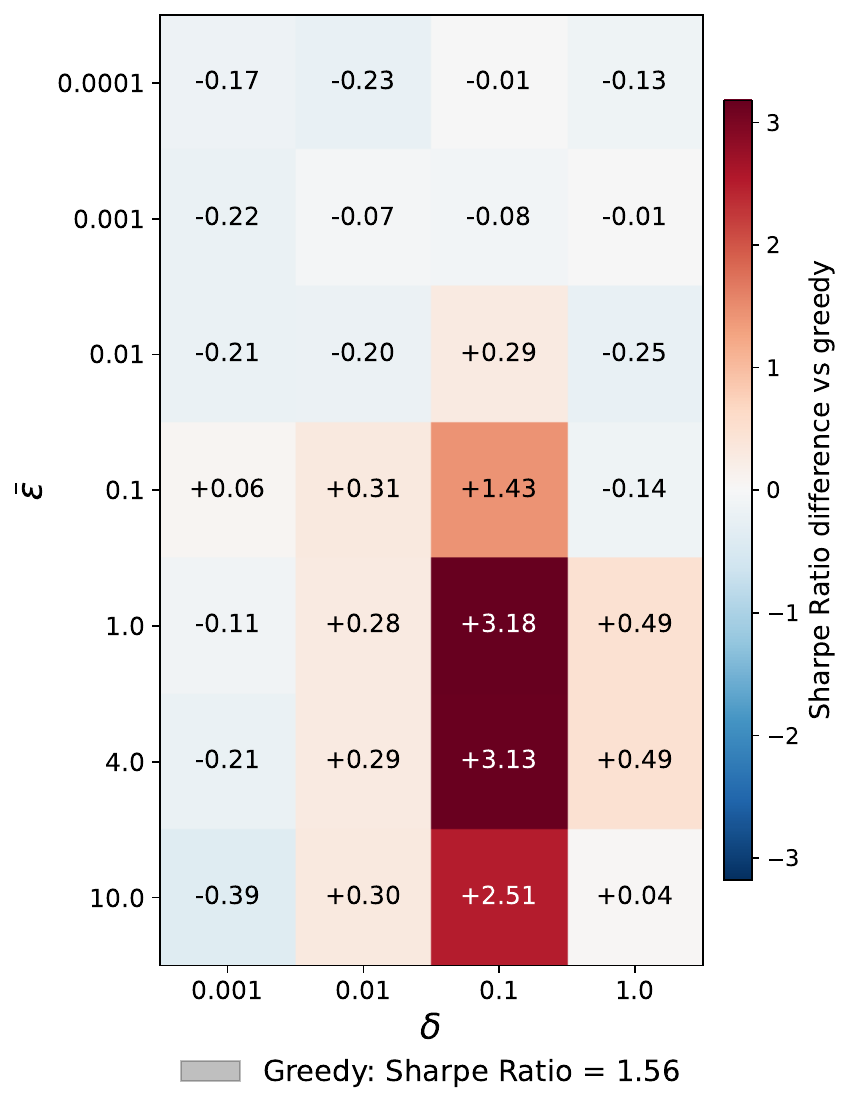}
            \caption{TSLA, 2019.}
        \end{subfigure}%
        \begin{subfigure}{0.245\linewidth}
            \includegraphics[width=\linewidth]{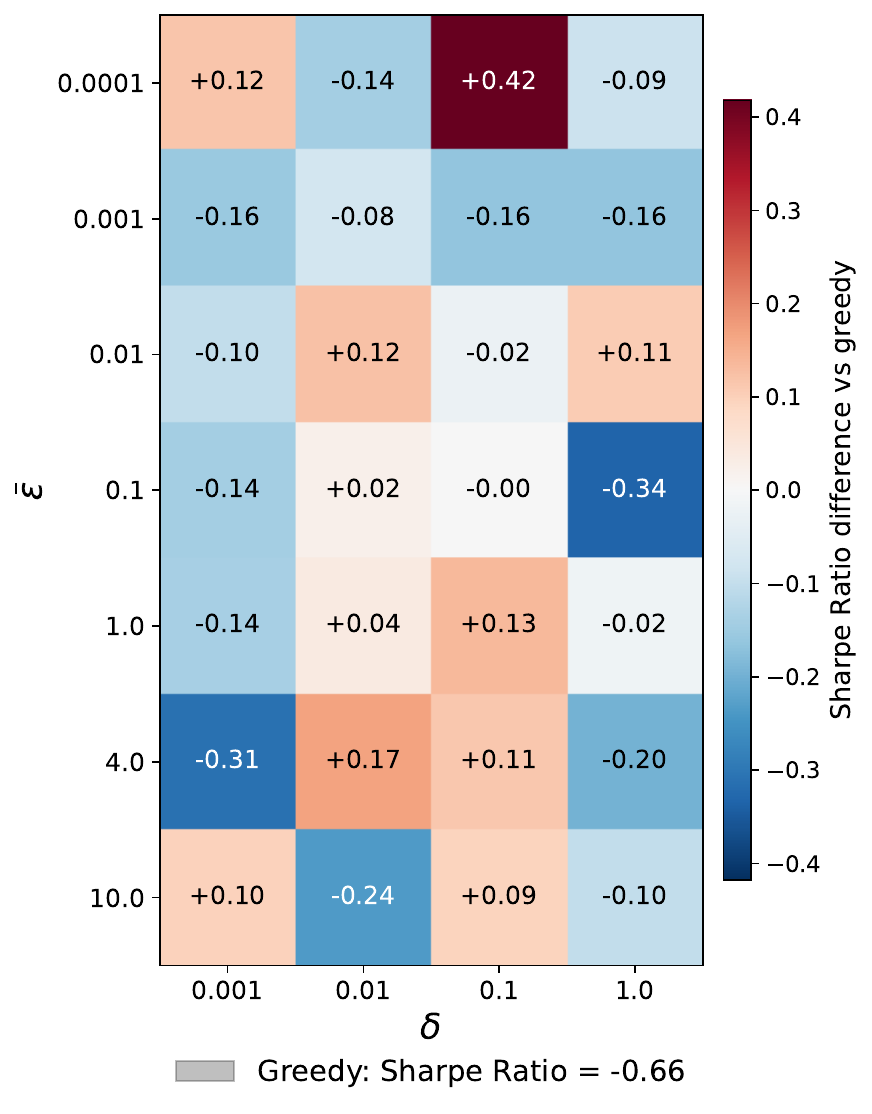}
            \caption{MKC, 2019.}
        \end{subfigure}%
        \begin{subfigure}{0.245\linewidth}
            \includegraphics[width=\linewidth]{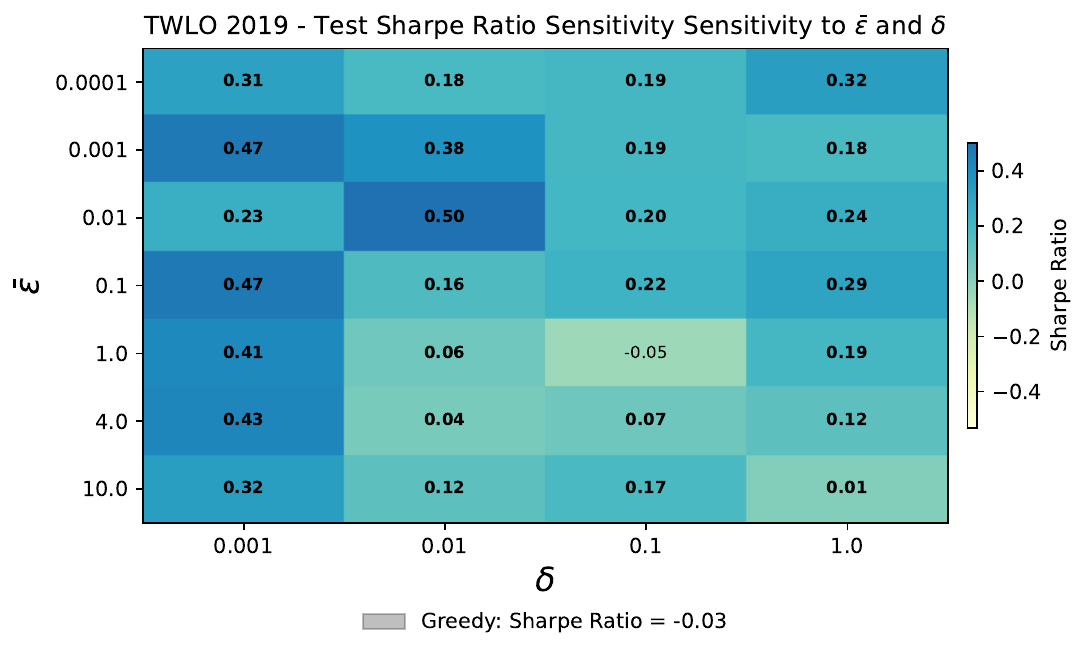}
            \caption{TWLO, 2019.}
        \end{subfigure}

        \par\vspace{0.1em}

        \begin{subfigure}{0.245\linewidth}
            \includegraphics[width=\linewidth]{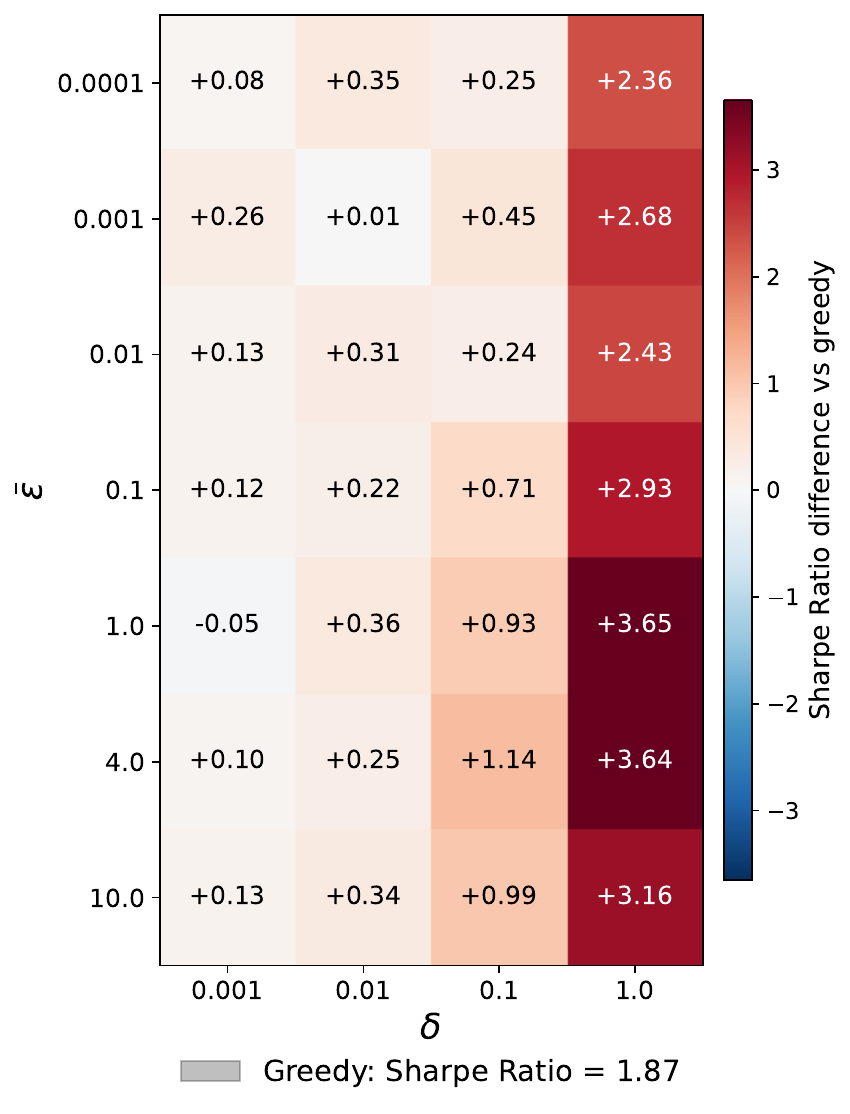}
            \caption{AAPL, 2020.}
        \end{subfigure}%
        \begin{subfigure}{0.245\linewidth}
            \includegraphics[width=\linewidth]{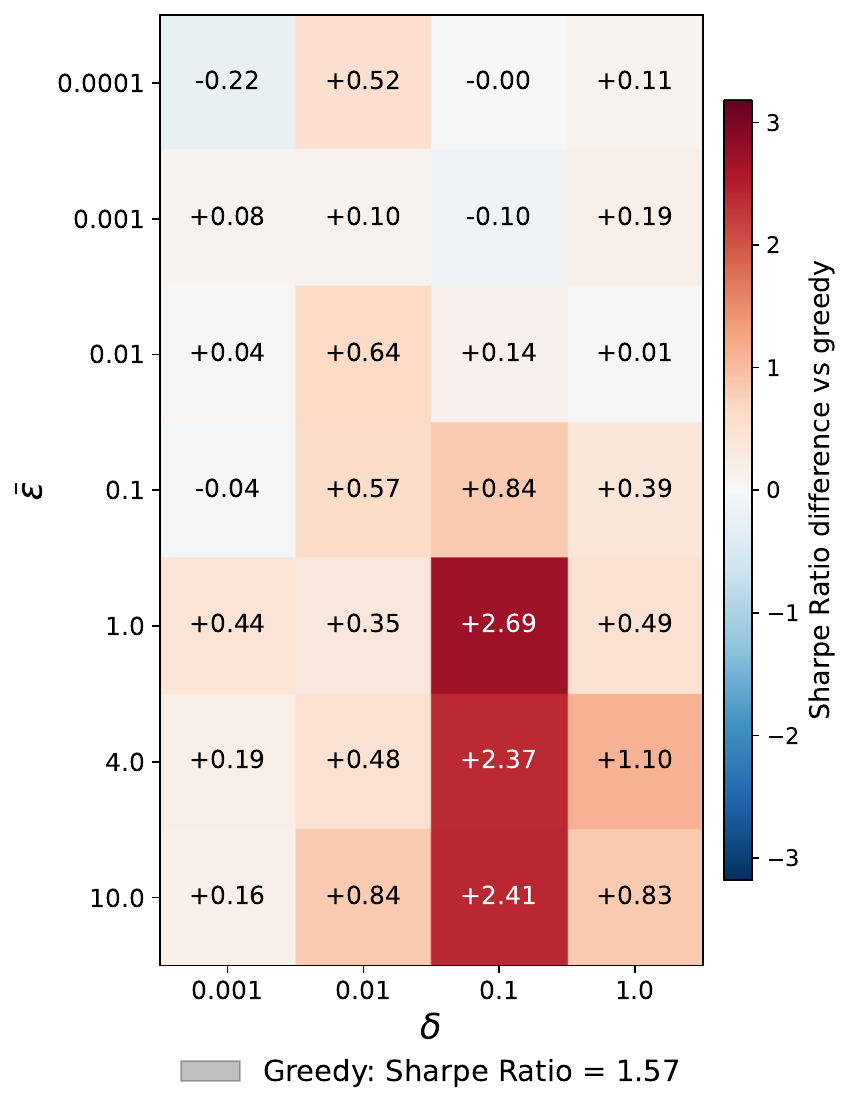}
            \caption{TSLA, 2020.}
        \end{subfigure}%
        \begin{subfigure}{0.245\linewidth}
            \includegraphics[width=\linewidth]{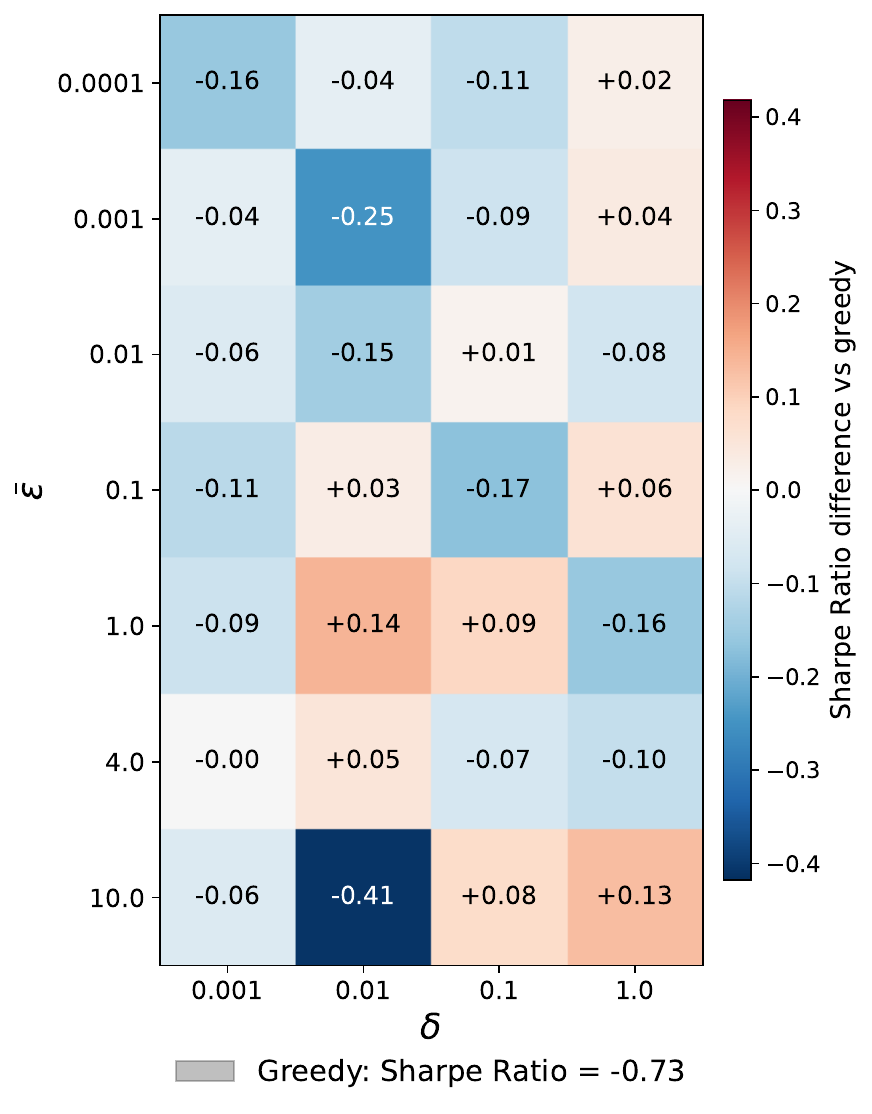}
            \caption{MKC, 2020.}
        \end{subfigure}%
        \begin{subfigure}{0.245\linewidth}
            \includegraphics[width=\linewidth]{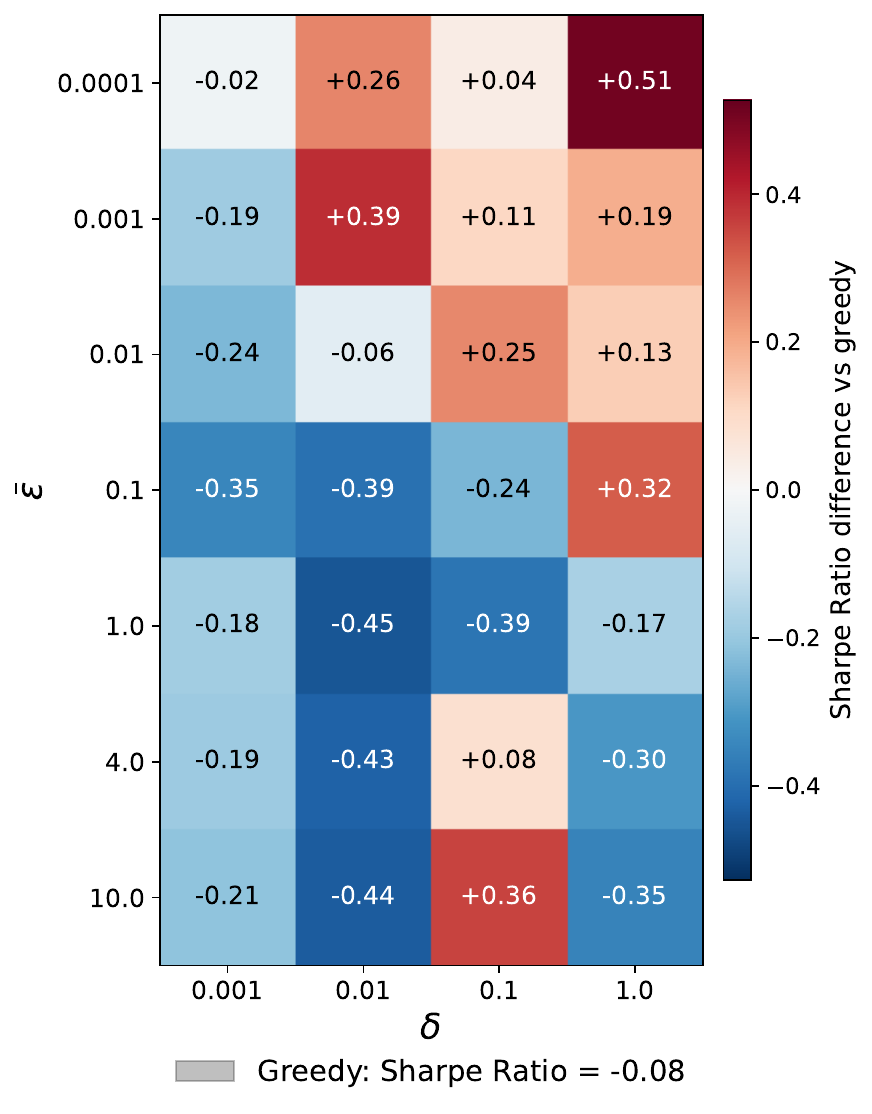}
            \caption{TWLO, 2020.}
        \end{subfigure}
    \end{minipage}}
    \caption[Test-period Sharpe ratio Difference between Greedy Policy and Robust Policy across the hyperparameter grid.]{Test-period Sharpe ratio Difference between Greedy Policy and Robust Policy across the $(\bar\varepsilon,\delta)$ grid, with the 2019 stock panels in the first row and the 2020 panels in the second row. The greedy benchmark is reported below each panel. The panels show that the region of robust configurations outperforming the greedy benchmark is typically broader in 2020, consistent with the larger distributional shift in that period.}
    \label{fig:sharpe_heatmap_real2}
\end{figure}

\begin{figure}[H]
    \centering
    \makebox[\textwidth][c]{\begin{minipage}{1.16\textwidth}\centering
        \begin{subfigure}{0.245\linewidth}
            \includegraphics[width=\linewidth]{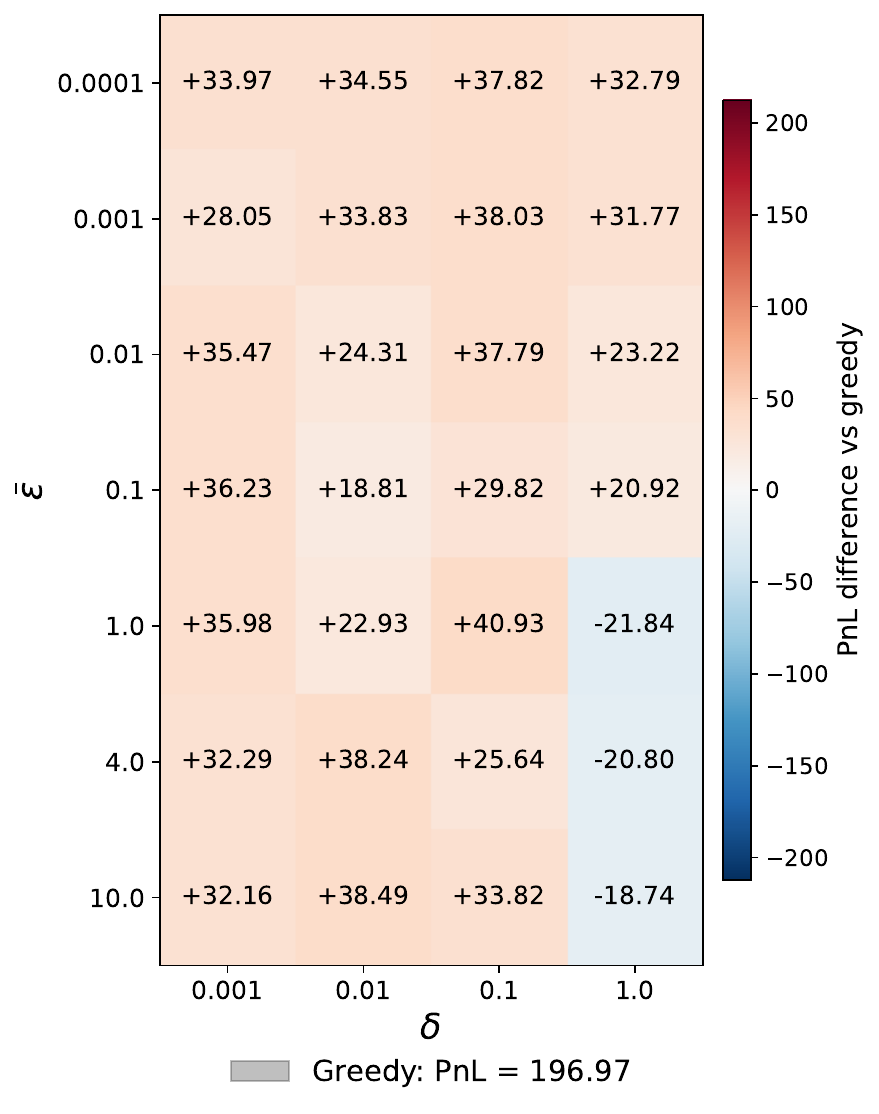}
            \caption{AAPL, 2019.}
        \end{subfigure}%
        \begin{subfigure}{0.245\linewidth}
            \includegraphics[width=\linewidth]{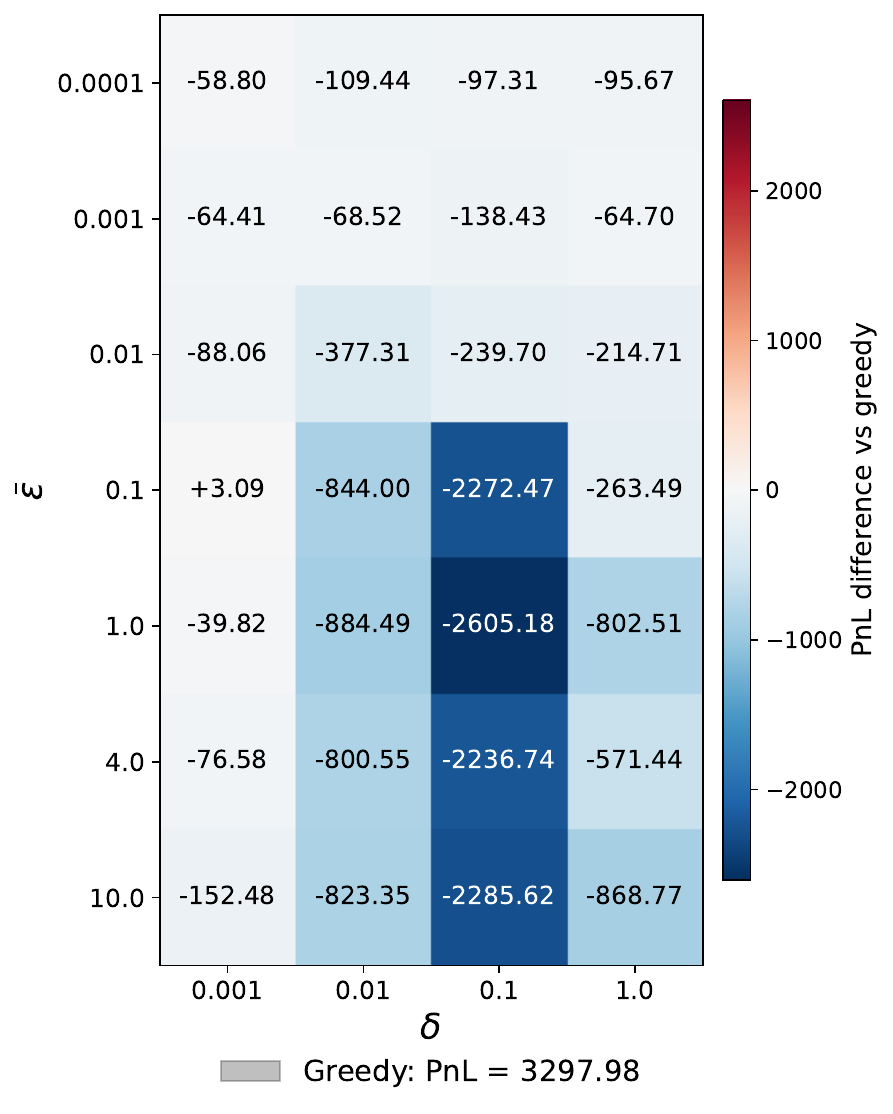}
            \caption{TSLA, 2019.}
        \end{subfigure}%
        \begin{subfigure}{0.245\linewidth}
            \includegraphics[width=\linewidth]{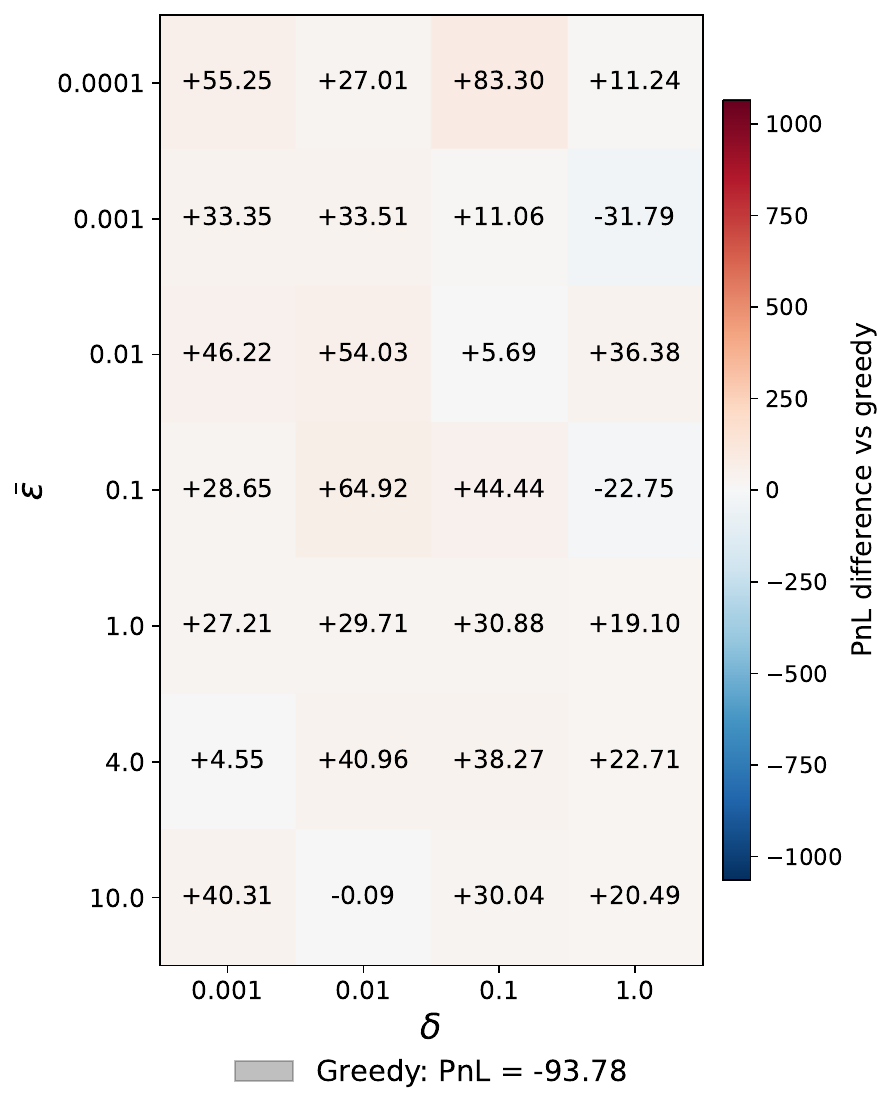}
            \caption{MKC, 2019.}
        \end{subfigure}%
        \begin{subfigure}{0.245\linewidth}
            \includegraphics[width=\linewidth]{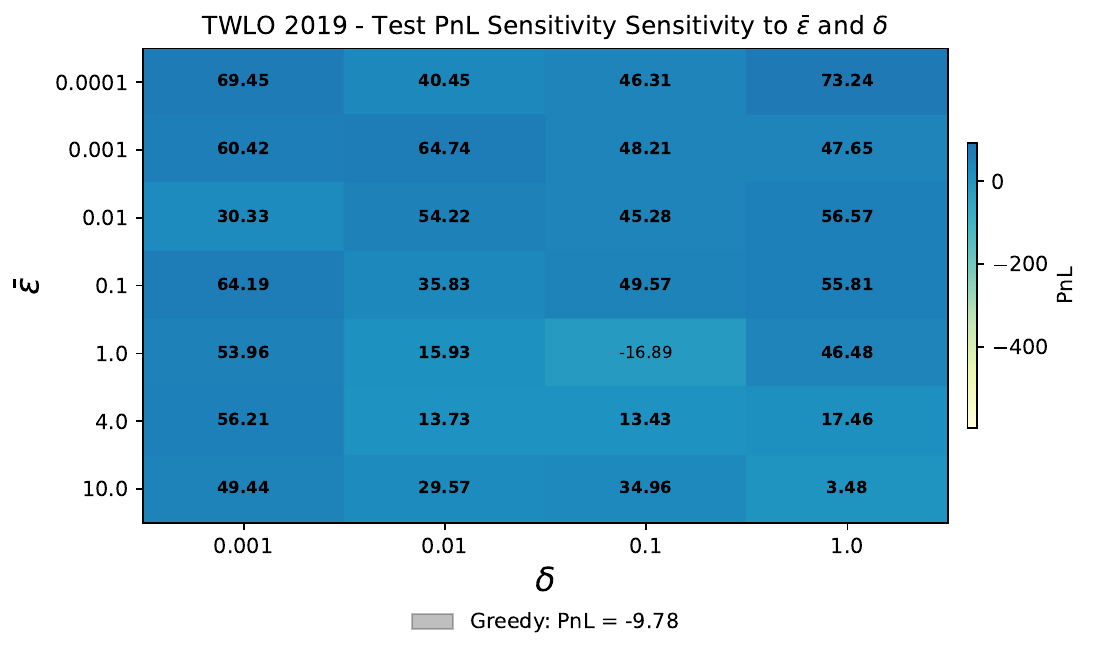}
            \caption{TWLO, 2019.}
        \end{subfigure}

        \par\vspace{0.1em}

        \begin{subfigure}{0.245\linewidth}
            \includegraphics[width=\linewidth]{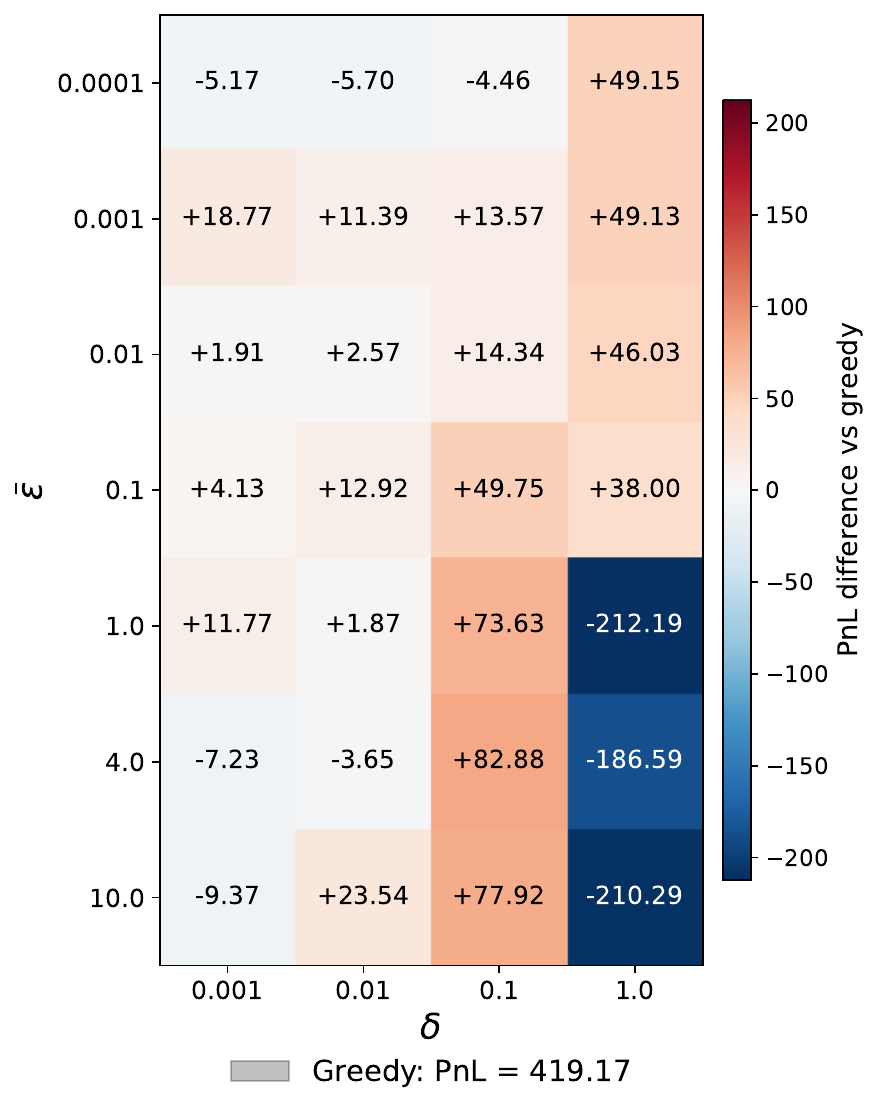}
            \caption{AAPL, 2020.}
        \end{subfigure}%
        \begin{subfigure}{0.245\linewidth}
            \includegraphics[width=\linewidth]{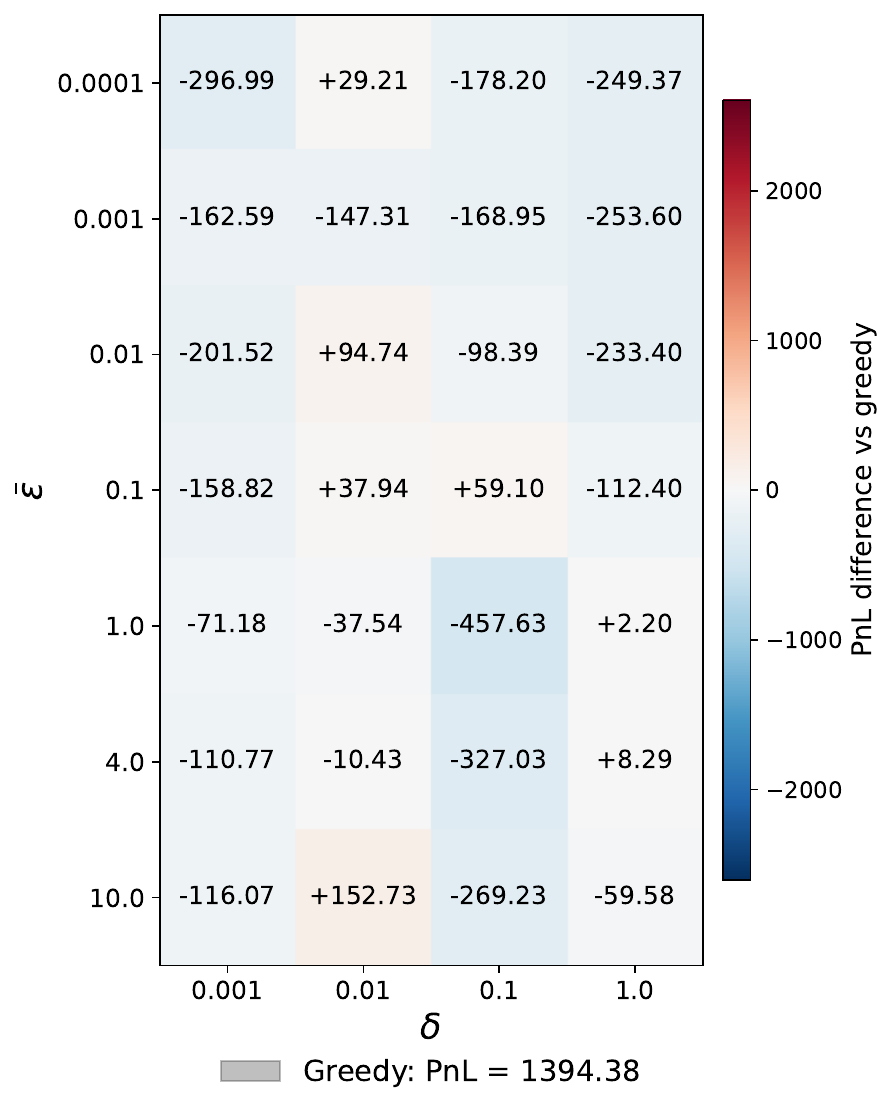}
            \caption{TSLA, 2020.}
        \end{subfigure}%
        \begin{subfigure}{0.245\linewidth}
            \includegraphics[width=\linewidth]{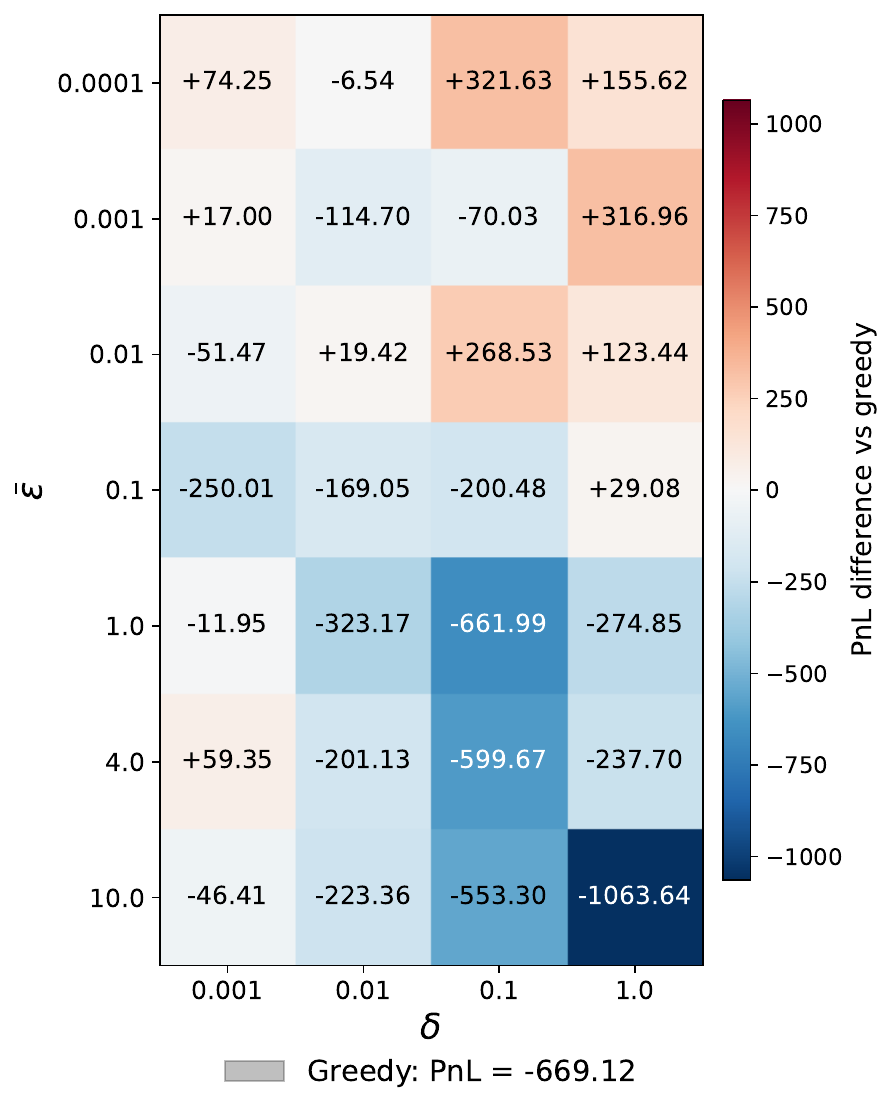}
            \caption{MKC, 2020.}
        \end{subfigure}%
        \begin{subfigure}{0.245\linewidth}
            \includegraphics[width=\linewidth]{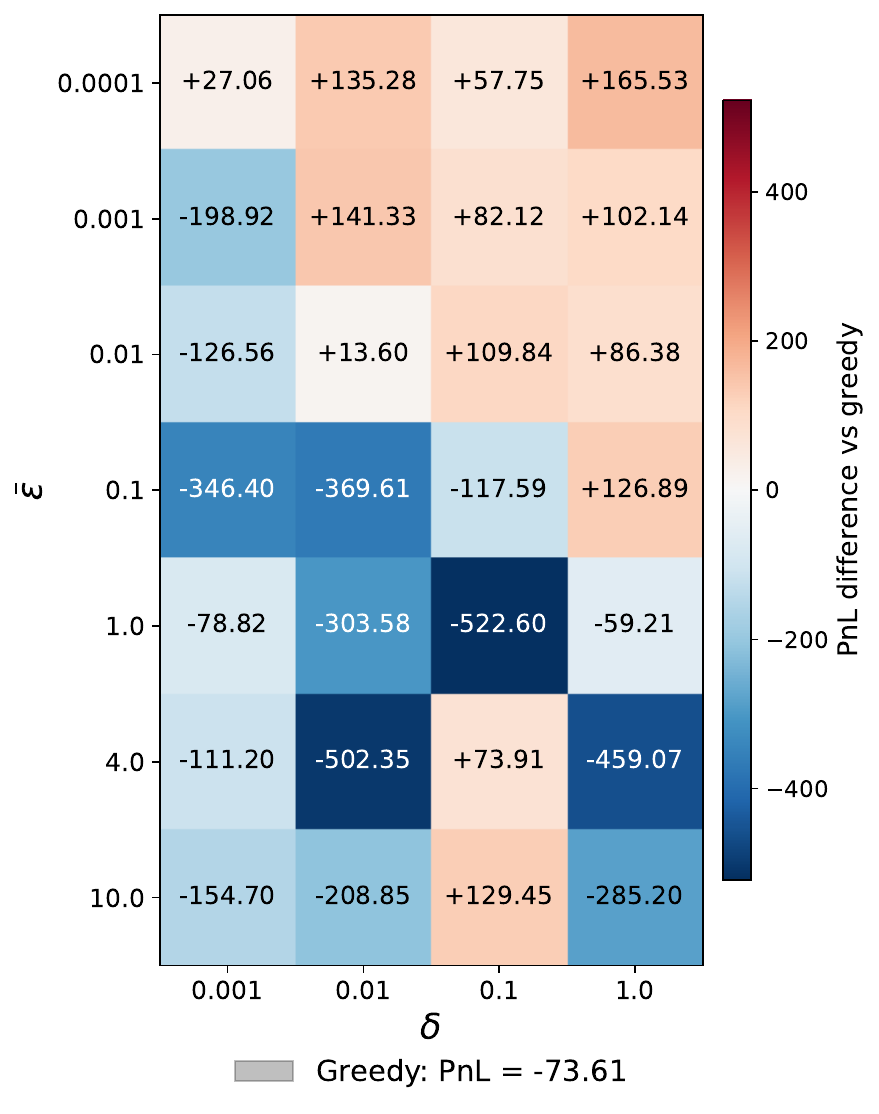}
            \caption{TWLO, 2020.}
        \end{subfigure}
    \end{minipage}}
    \caption[Test-period mean P\&L across the hyperparameter grid.]{Test-period mean P\&L Difference between Greedy Policy and Robust Policy across the $(\bar\varepsilon,\delta)$ grid, with the 2019 stock panels in the first row and the 2020 panels in the second row. The panels show that robustness often trades off some average profitability against improved stability, although the size and direction of the effect differ across stocks and are more pronounced in 2020.}
    \label{fig:pnl_heatmap_real}
\end{figure}

\subsection[Risk-Return Test Pareto Frontier in 2019]{2019 Test Pareto Frontiers}

Figure~\ref{fig:pareto_2019_appendix} reports the out-of-sample test Pareto frontiers for the four representative assets during the 2019 regime. The corresponding 2020 test Pareto frontiers are reported in Figure~\ref{fig:pareto_2020} in the main text.

\begin{figure}[H]
    \centering
    \makebox[\textwidth][c]{\begin{minipage}{1.16\textwidth}\centering
        \begin{subfigure}{0.245\linewidth}
            \includegraphics[width=\linewidth]{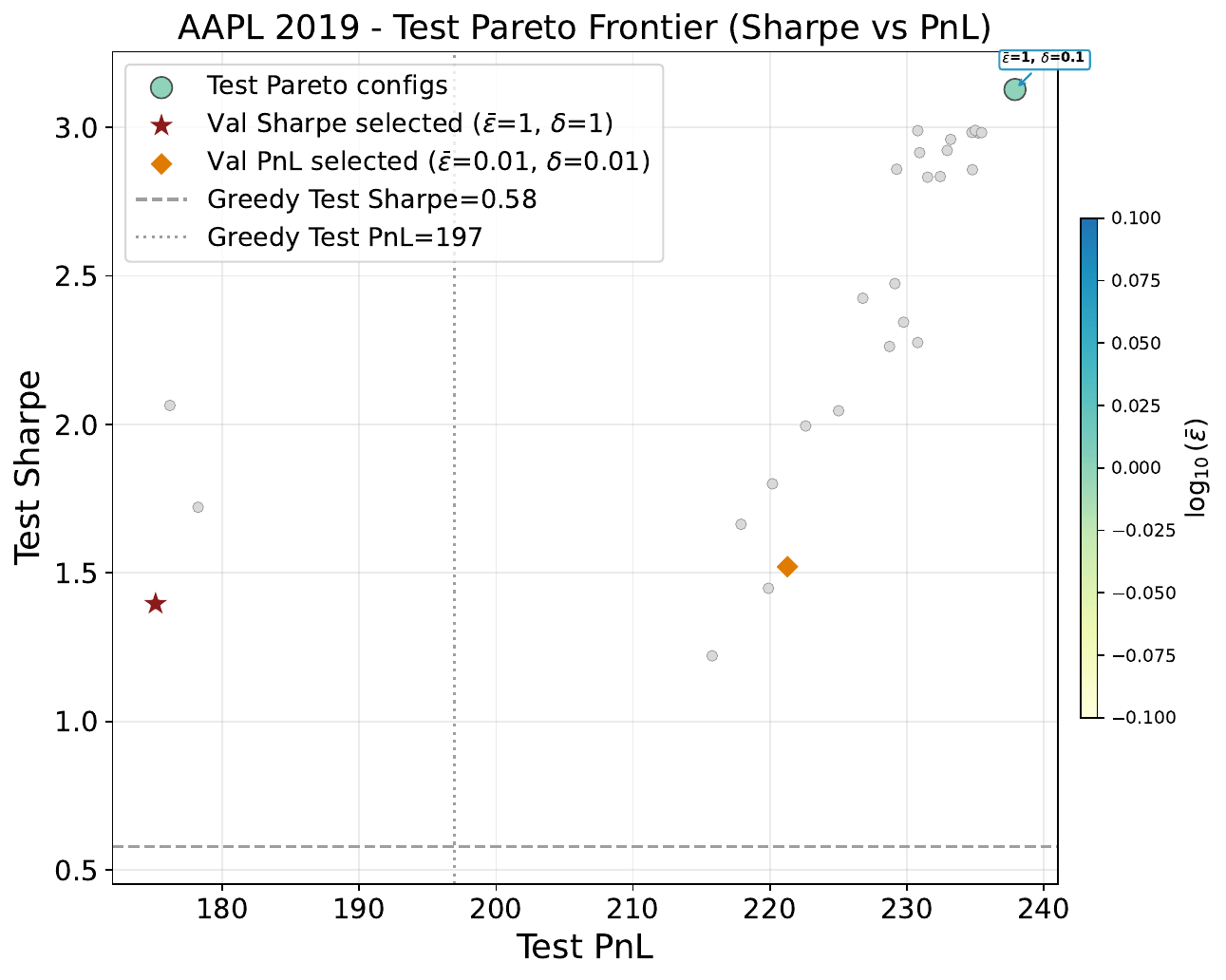}
            \caption{AAPL.}
        \end{subfigure}\hfill
        \begin{subfigure}{0.245\linewidth}
            \includegraphics[width=\linewidth]{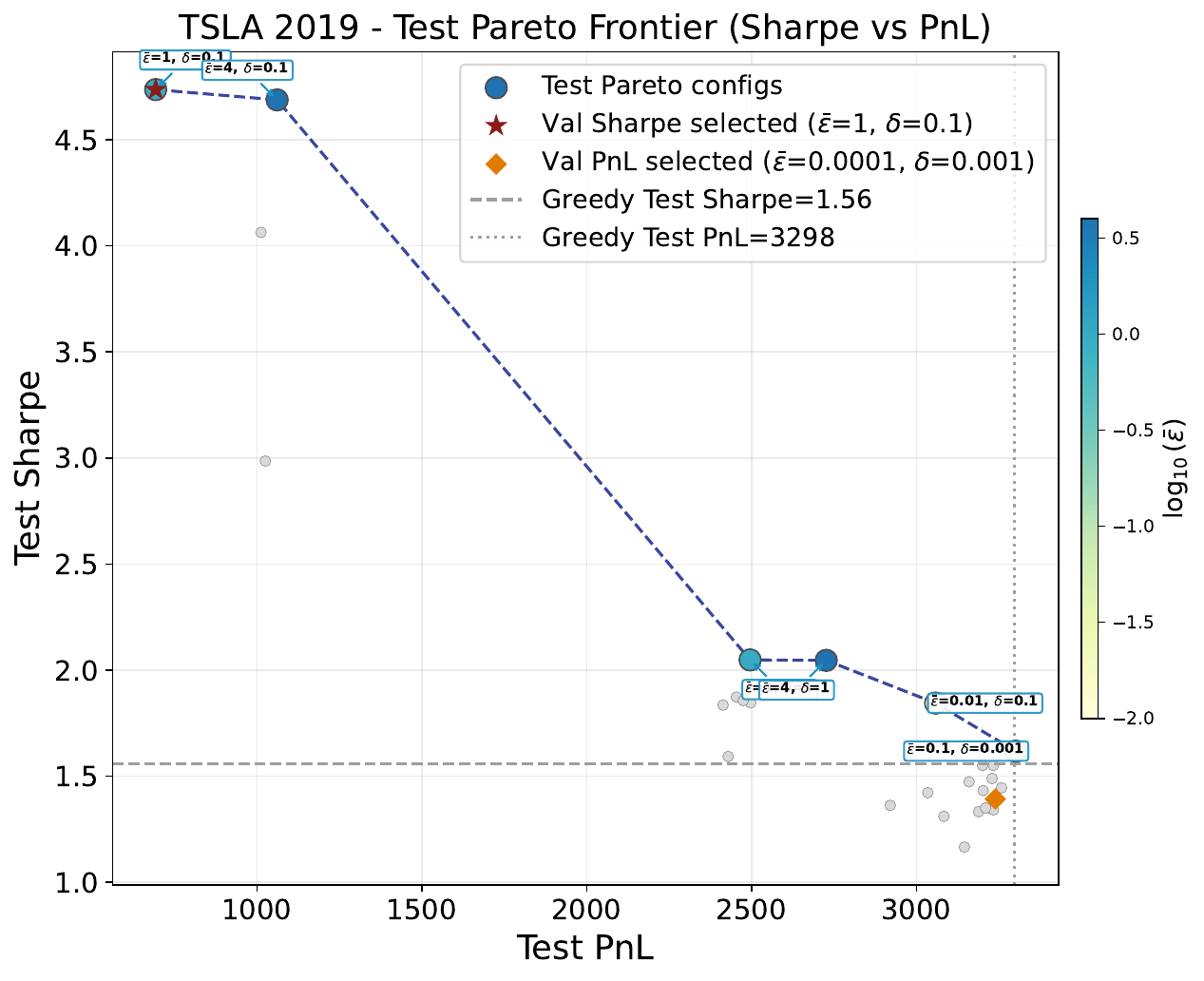}
            \caption{TSLA.}
        \end{subfigure}\hfill
        \begin{subfigure}{0.245\linewidth}
            \includegraphics[width=\linewidth]{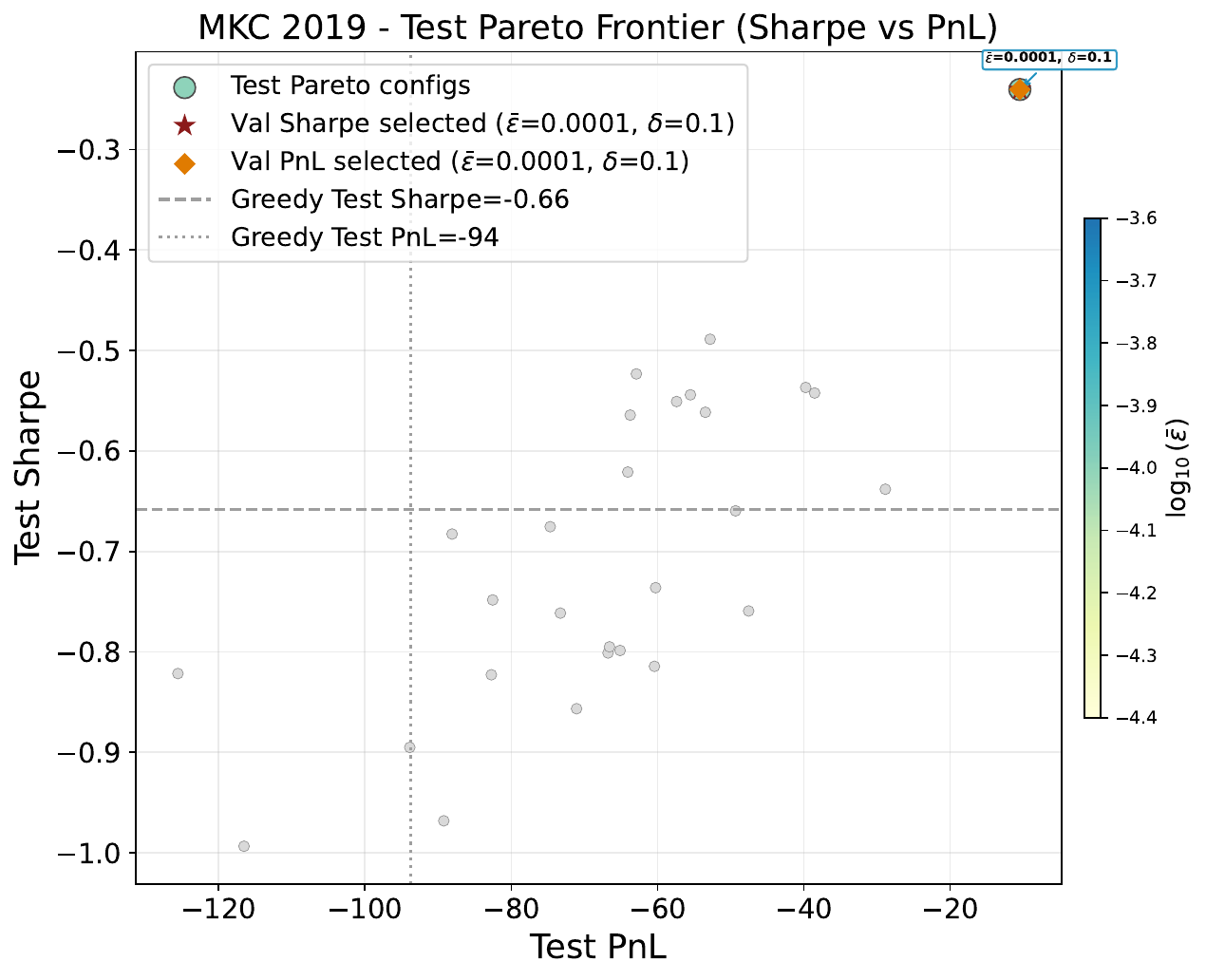}
            \caption{MKC.}
        \end{subfigure}\hfill
        \begin{subfigure}{0.245\linewidth}
            \includegraphics[width=\linewidth]{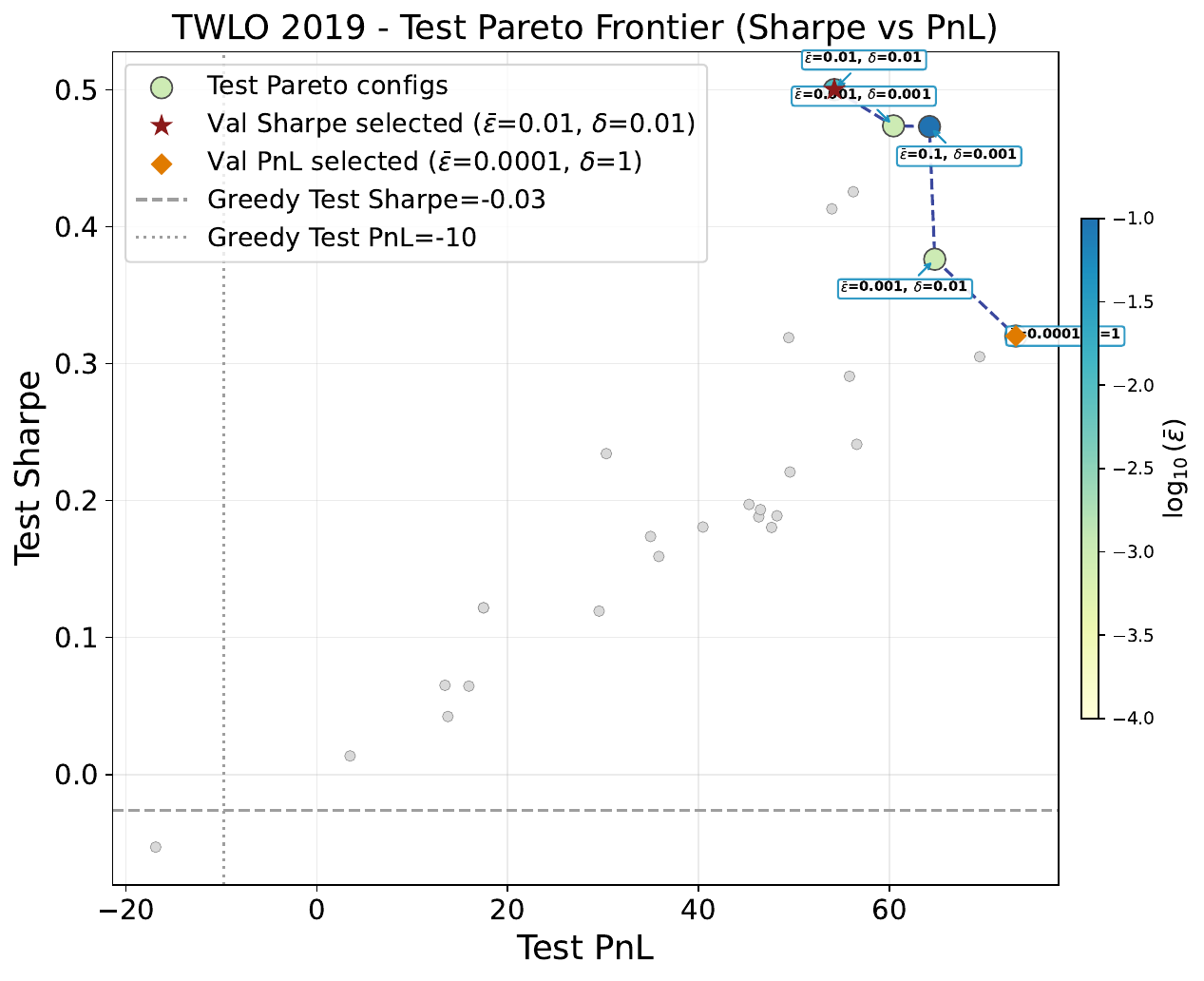}
            \caption{TWLO.}
        \end{subfigure}
    \end{minipage}}
    \caption{Out-of-sample test Pareto frontiers during the 2019 regime. Each point corresponds to a robust hyperparameter configuration, and highlighted points denote the policies selected from the validation frontier. The 2019 frontiers are generally tighter than their 2020 counterparts, indicating a less pronounced but still economically meaningful risk-return trade-off, with TSLA showing the clearest separation across robustness choices.}
    \label{fig:pareto_2019_appendix}
\end{figure}

The 2019 panels provide only limited evidence of a pronounced risk-return trade-off. The effect is most apparent for TSLA, where the frontier exhibits clearer separation across robustness specifications. For AAPL, MKC, and TWLO, by contrast, the frontiers are comparatively tight, but this also means that in the calmer 2019 regime there remain many hyperparameter choices delivering both relatively high mean P\&L and relatively high Sharpe ratio. Figure~\ref{fig:pareto_2020}, however, displays a substantially more visible trade-off in 2020: the frontier is more clearly stretched, but the set of policies that jointly attain high P\&L and high Sharpe ratio is correspondingly more selective. This pattern is consistent with the interpretation that robustness becomes more economically consequential in the more dislocated 2020 environment, where liquidity frictions and return uncertainty force sharper trade-offs across candidate policies.

\subsection{Intraday Quoting Behaviour}

Figures~\ref{fig:intraday_behavior_real_aapl_tsla} and~\ref{fig:intraday_behavior_real_mkc_twlo} report the mean intraday paths of bid quantity, ask quantity, bid spread, and ask spread for the robust strategies and the benchmark rules. Two regularities stand out. First, for the learned agents, intraday variation is driven primarily by quantities rather than spreads. This reflects the fact that the empirical maximum spread is calibrated from the 95th to 99th percentile of the historical one-minute spread distribution in order to obtain a reasonable fill-probability specification. At a one-minute decision frequency, this admissible spread range is comparatively tight, so the learned policies have limited scope to differentiate themselves through spreads and instead optimize mainly through quoted size. Second, robustness mainly affects the level of quoted size rather than the shape of the intraday spread. Relative to the greedy benchmark, the robust agents, especially the Sharpe-selected specification, post smaller quantities over most of the day and reduce them more sharply near the close.

The effect is clearest in the more active panels, particularly AAPL and TSLA, where the robust agents preserve the same broad intraday pattern as the greedy benchmark but with systematically more conservative sizing. This is consistent with the earlier cross-state analysis and, more broadly, with the simulation study, where the robust agents are also significantly more risk-averse than their non-robust countvvgbherparts. The same qualitative pattern is visible in TWLO, although the separation across policies is smaller. By contrast, MKC exhibits relatively limited dispersion across the learned policies, which accords with the weaker gains from robustness in that stock. Overall, the intraday evidence supports the interpretation that the benefit of robustness is strongest in the more liquid and volatile panels, where the agent has greater scope to trade off trading opportunities against end-of-day inventory risk.

\begin{figure}[H]
    \centering
    \makebox[\textwidth][c]{\begin{minipage}{1.16\textwidth}\centering
        \begin{subfigure}{0.49\linewidth}
            \includegraphics[width=\linewidth]{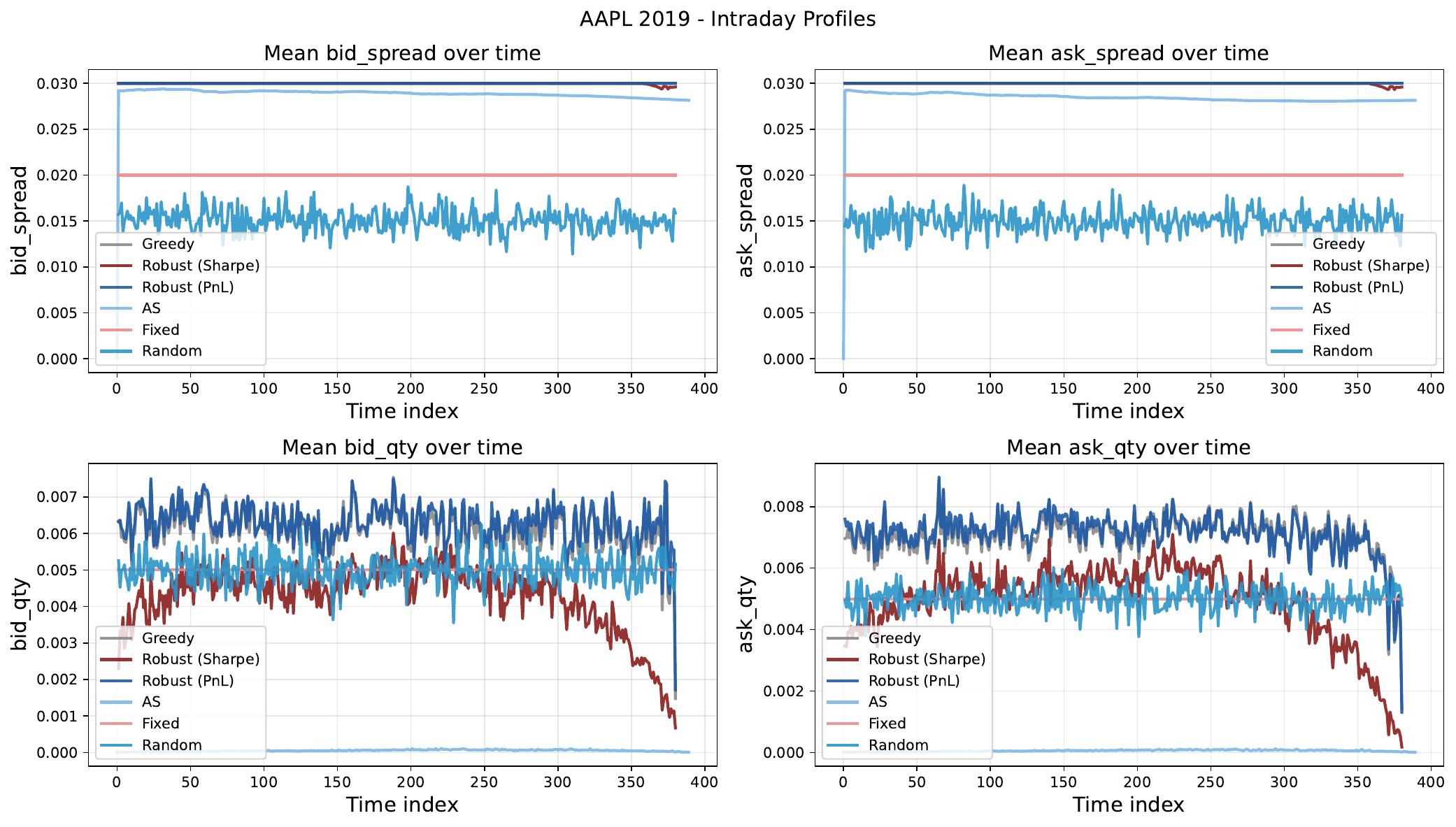}
            \caption{AAPL, 2019.}
        \end{subfigure}%
        \begin{subfigure}{0.49\linewidth}
            \includegraphics[width=\linewidth]{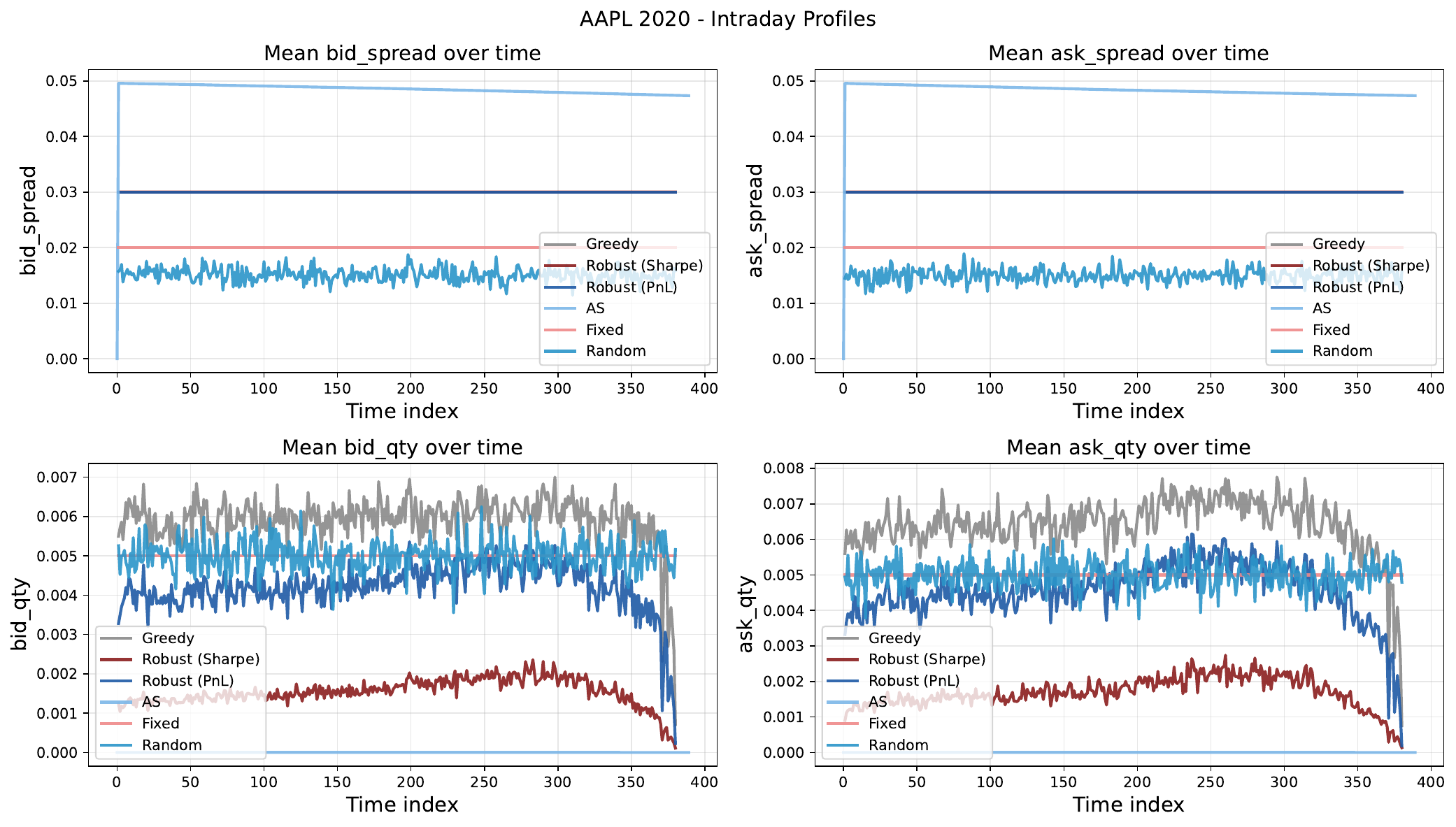}
            \caption{AAPL, 2020.}
        \end{subfigure}

        \par\vspace{0.4em}

        \begin{subfigure}{0.49\linewidth}
            \includegraphics[width=\linewidth]{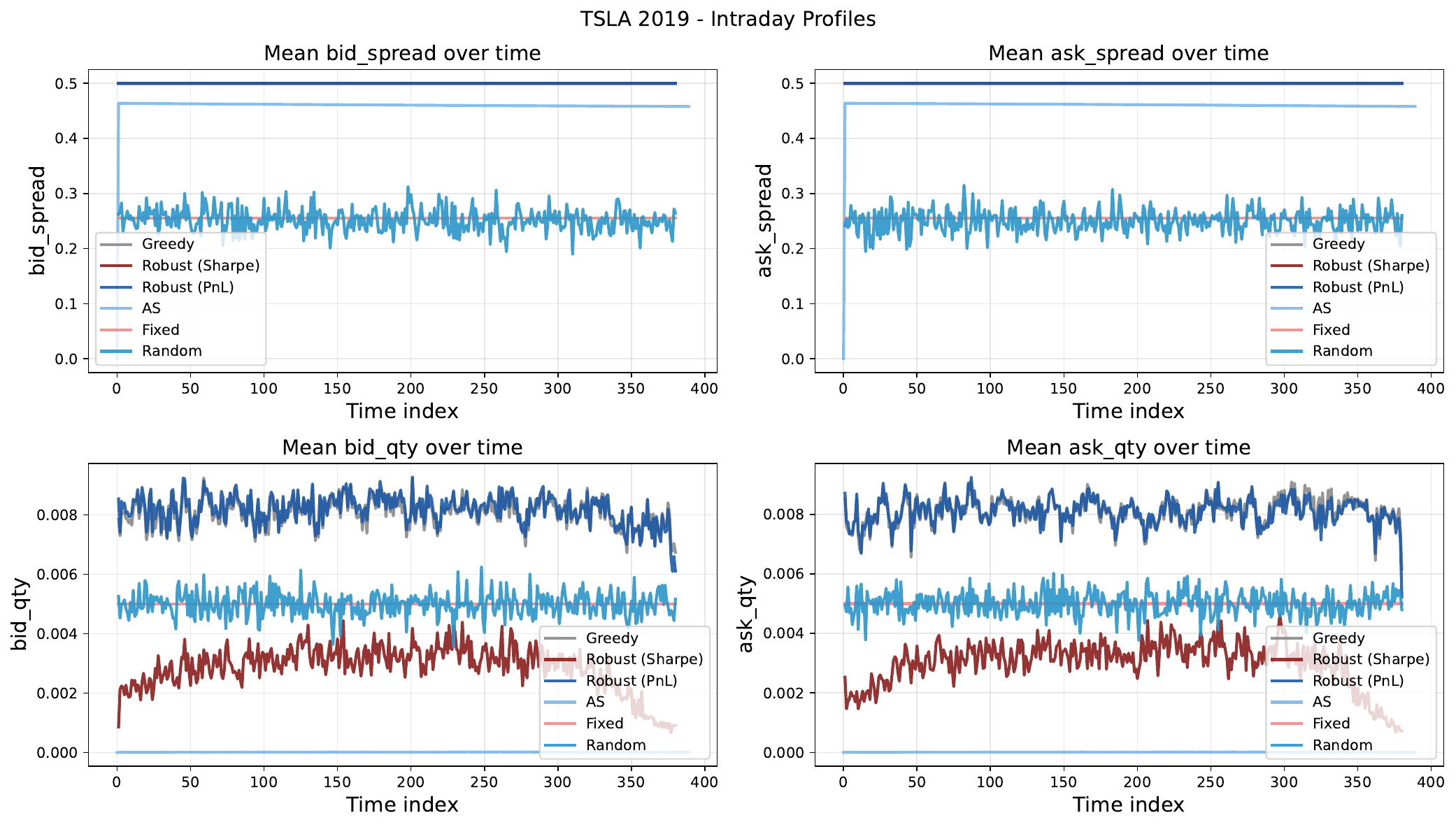}
            \caption{TSLA, 2019.}
        \end{subfigure}%
        \begin{subfigure}{0.49\linewidth}
            \includegraphics[width=\linewidth]{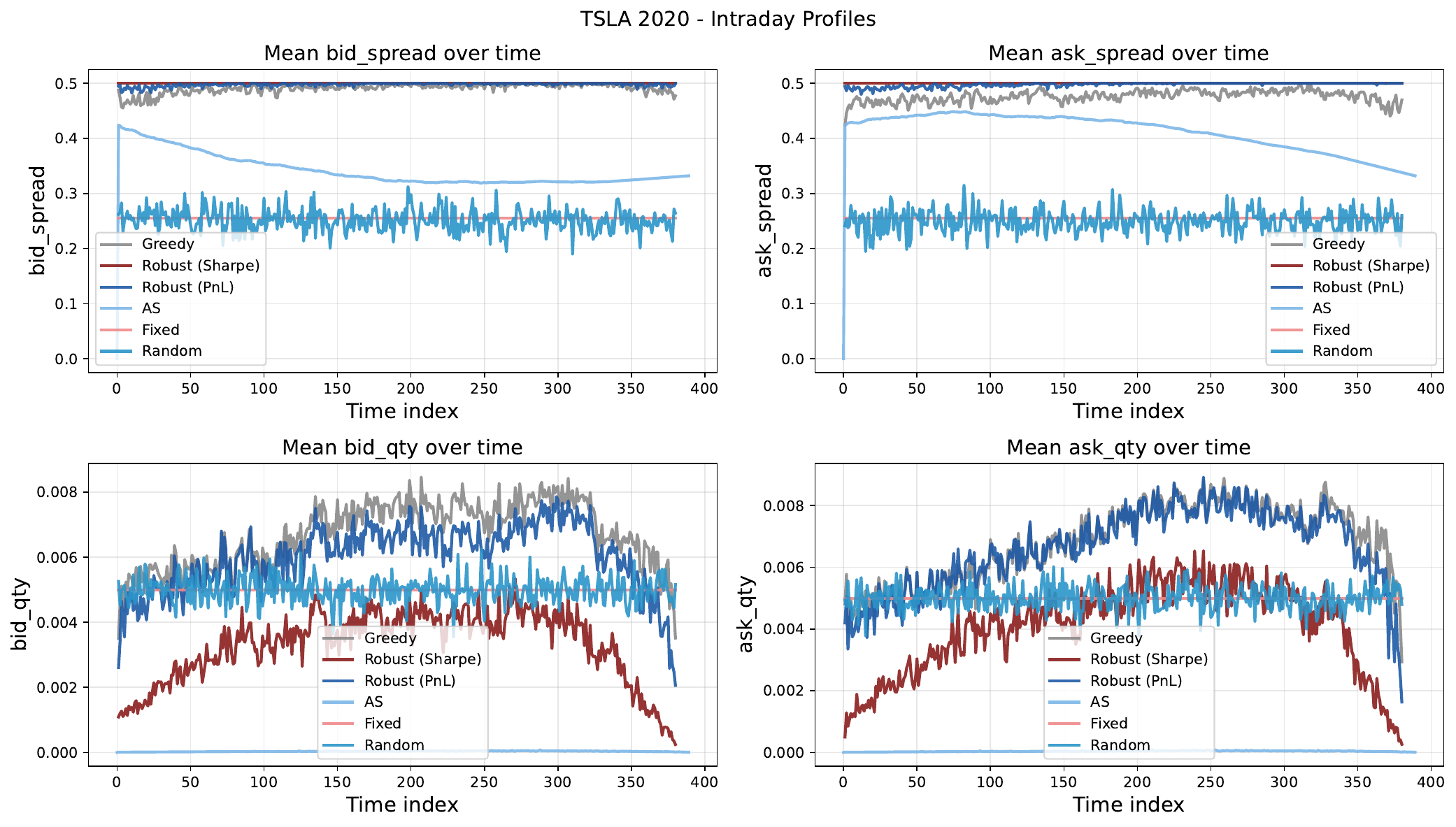}
            \caption{TSLA, 2020.}
        \end{subfigure}
    \end{minipage}}
    \caption{Mean quote (bid quantity, ask quantity, bid spread, ask spread) of robust strategies against baselines for AAPL and TSLA. Robust hyperparameters are those selected from the best validation Sharpe and best validation PnL. For the AS strategy, the maximum spread is 0.6. In both stocks, the robust agents mainly differ from the benchmark through more conservative quantity provision, while the intraday spread profile remains comparatively similar.}
    \label{fig:intraday_behavior_real_aapl_tsla}
\end{figure}

\begin{figure}[H]
    \centering
    \makebox[\textwidth][c]{\begin{minipage}{1.16\textwidth}\centering
        \begin{subfigure}{0.49\linewidth}
            \includegraphics[width=\linewidth]{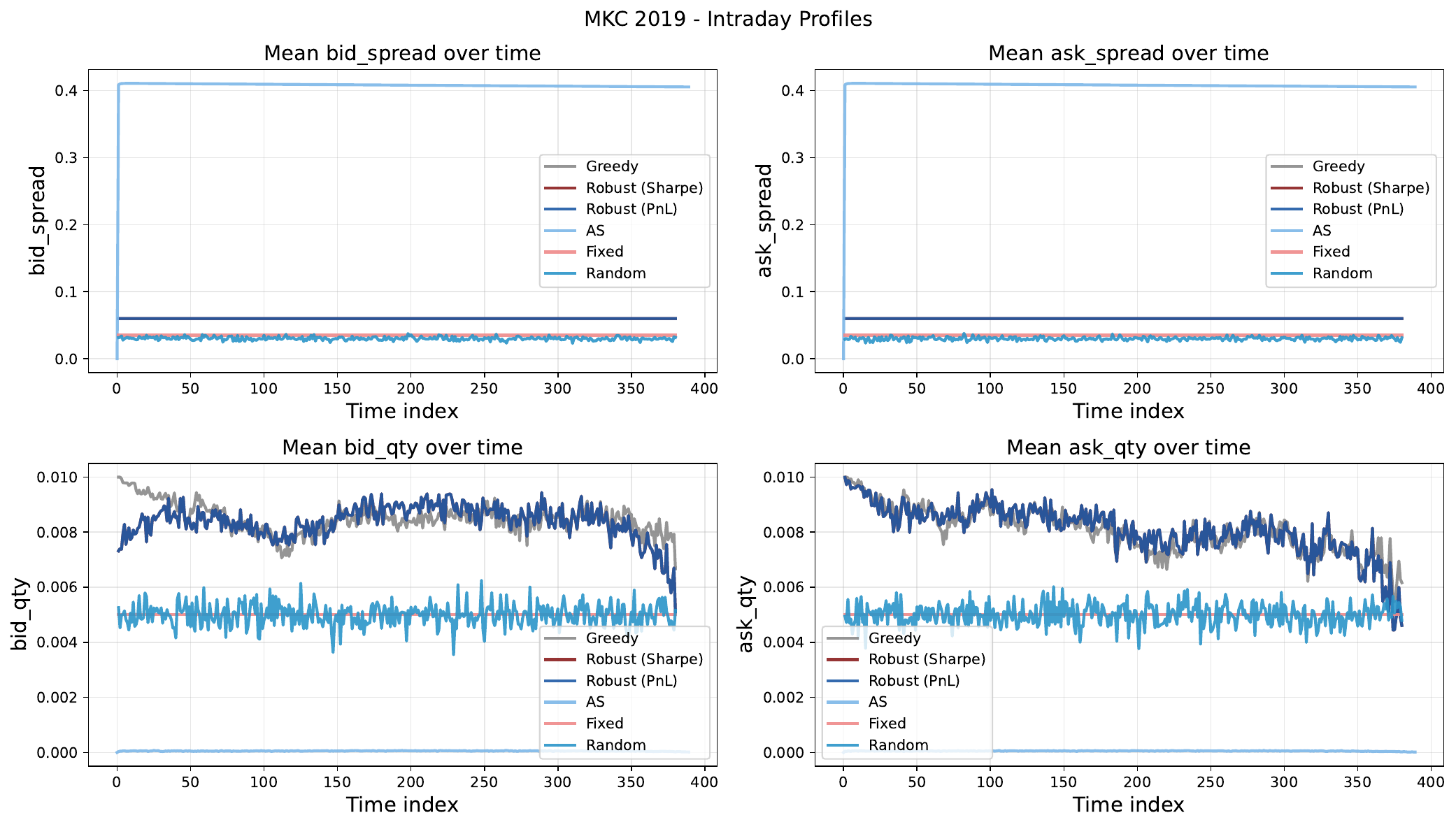}
            \caption{MKC, 2019.}
        \end{subfigure}%
        \begin{subfigure}{0.49\linewidth}
            \includegraphics[width=\linewidth]{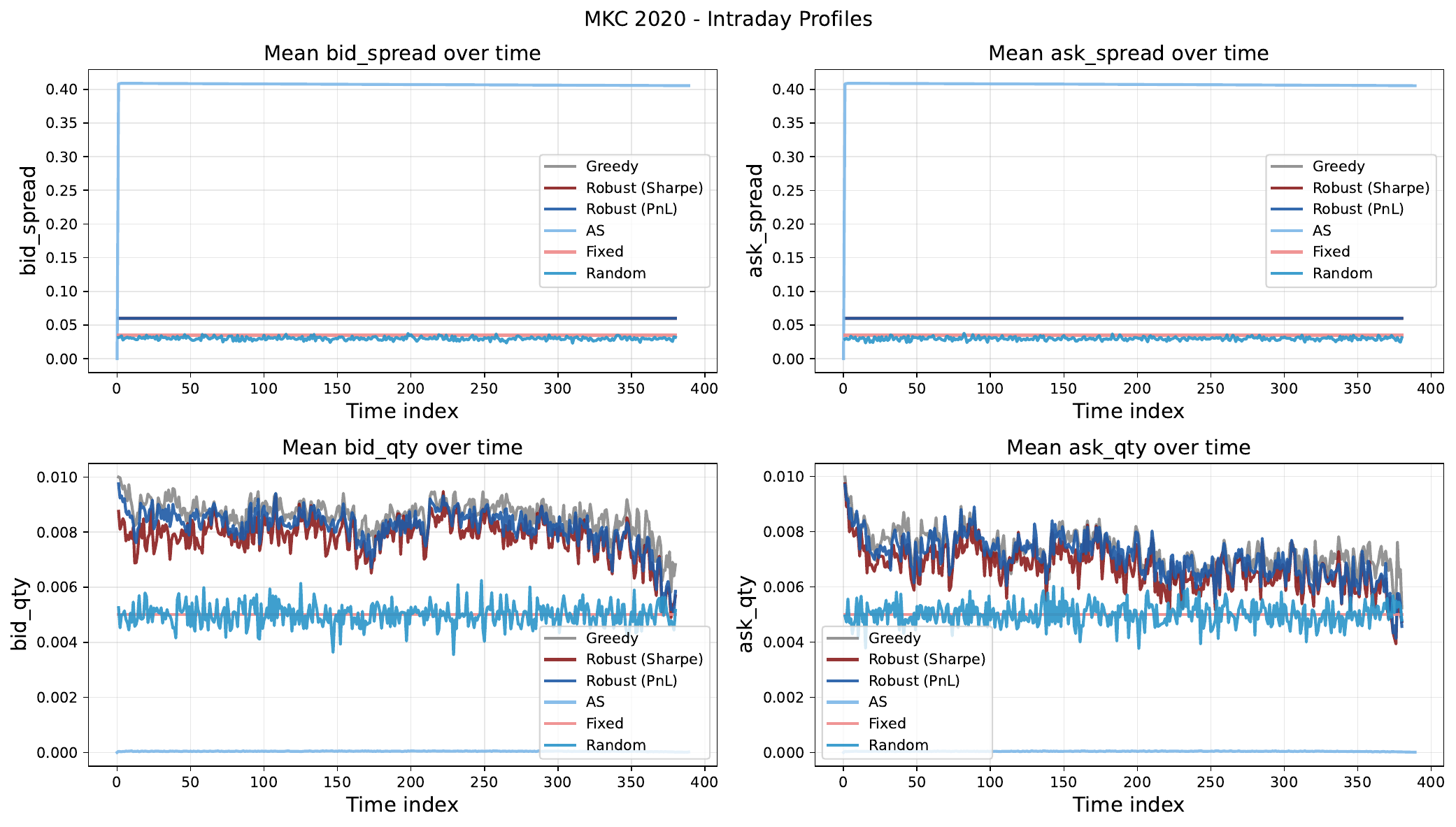}
            \caption{MKC, 2020.}
        \end{subfigure}

        \par\vspace{0.4em}

        \begin{subfigure}{0.49\linewidth}
            \includegraphics[width=\linewidth]{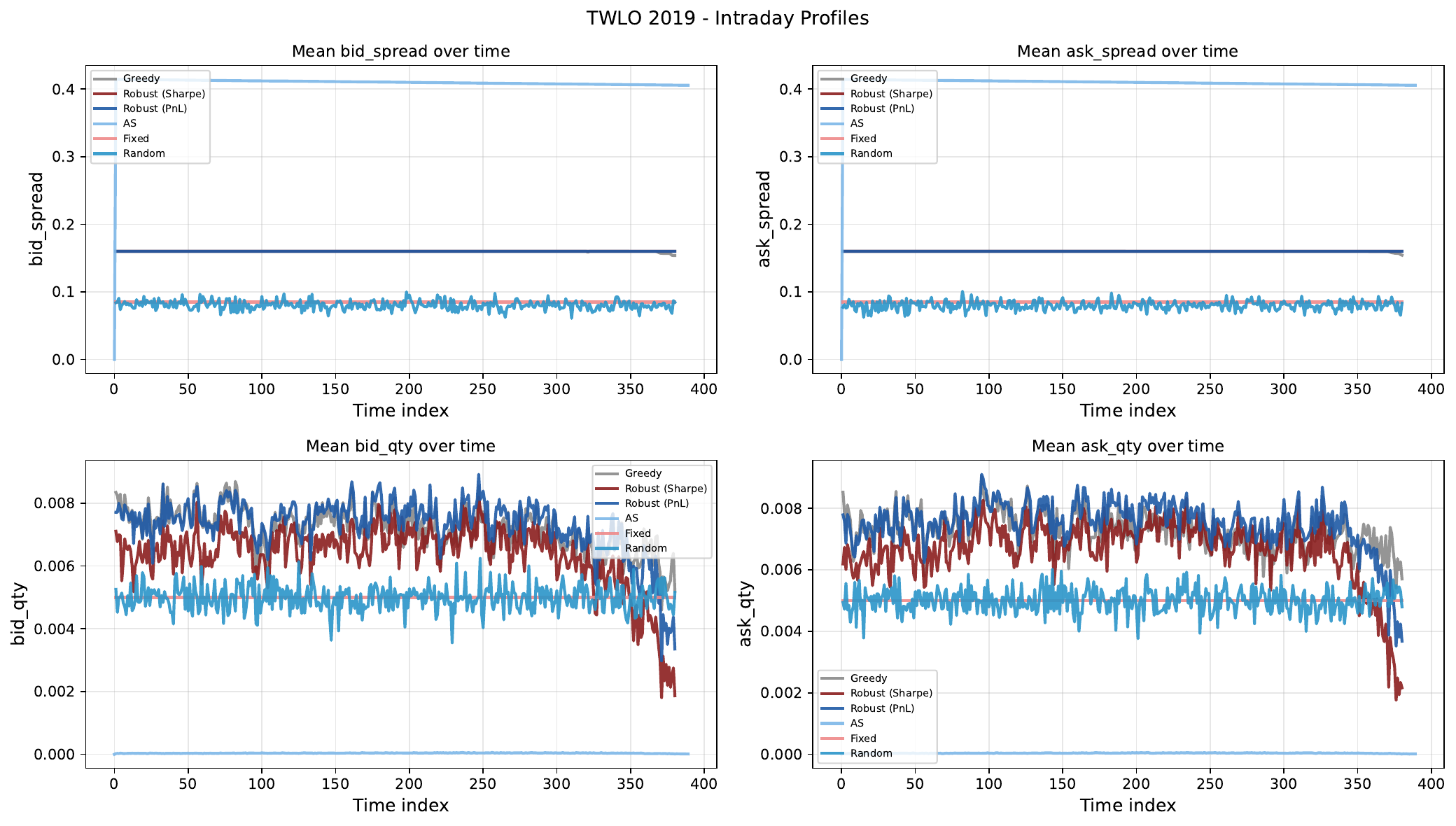}
            \caption{TWLO, 2019.}
        \end{subfigure}%
        \begin{subfigure}{0.49\linewidth}
            \includegraphics[width=\linewidth]{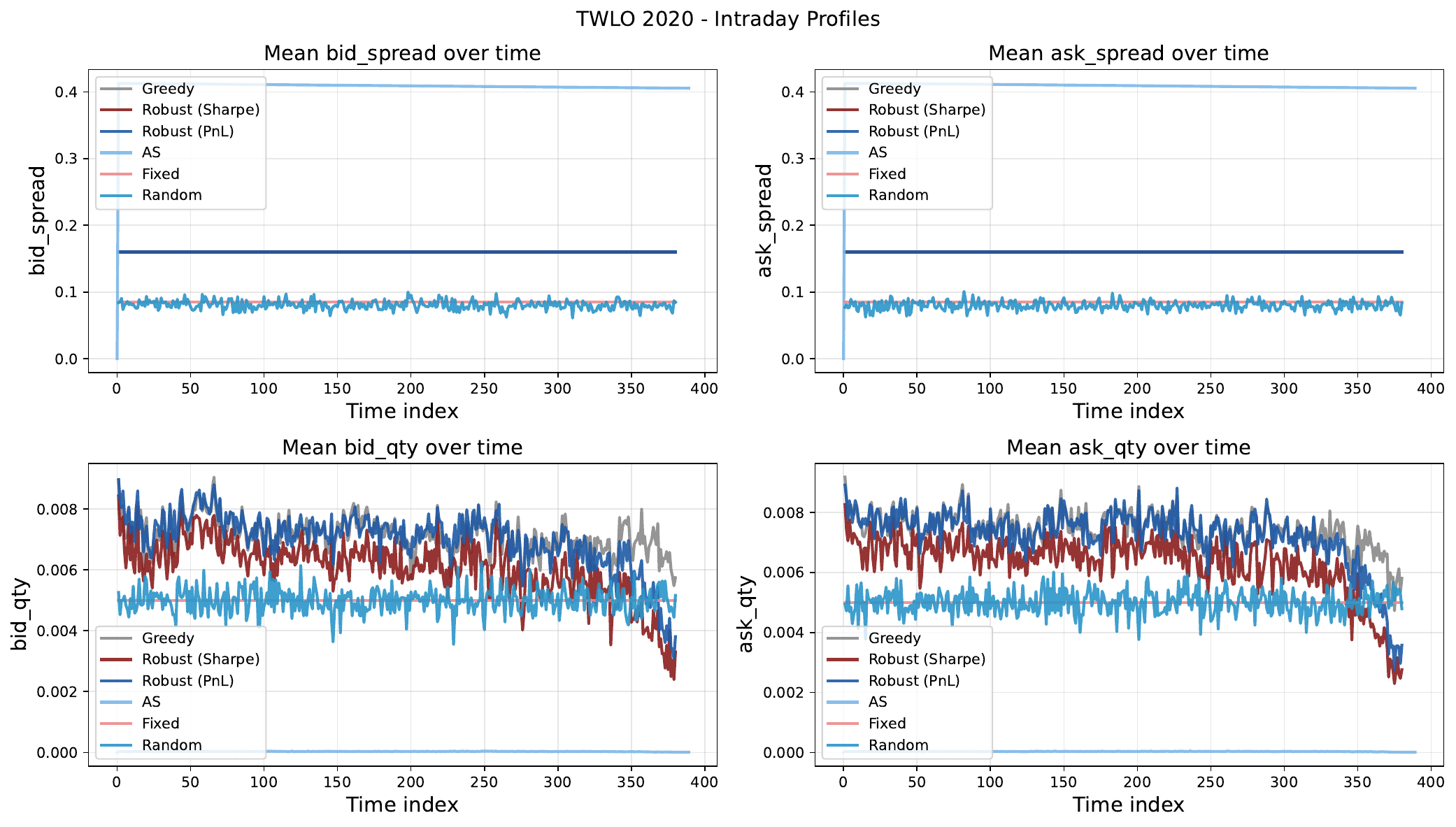}
            \caption{TWLO, 2020.}
        \end{subfigure}
    \end{minipage}}
    \caption{Mean quote (bid quantity, ask quantity, bid spread, ask spread) of robust strategies against baselines for MKC and TWLO. Robust hyperparameters are those selected from the best validation Sharpe and best validation PnL. For the AS strategy, the maximum spread is 0.6. The separation across policies is more modest in these two stocks, especially for MKC, consistent with the weaker gains from robustness outside the more active panels.}
    \label{fig:intraday_behavior_real_mkc_twlo}
\end{figure}

\subsection[Average Sharpe Ratio Conditional on \texorpdfstring{$\varepsilon$}{epsilon}]{Average Sharpe Ratio Conditional on $\varepsilon$}\label{sec:appendix_sharpe_eps}

Figure~\ref{fig:sharpe_vs_eps_real} reports the average test Sharpe ratio of the robust agent conditional on the \emph{original} Sinkhorn radius $\varepsilon$ (recovered from the tuned shifted radius), against the greedy benchmark (dashed), for all four stocks in 2019 and 2020. This view is meant to expose how performance depends on the size of the ambiguity ball itself, rather than on the tuned uncertainty-tolerance parameter $\bar\varepsilon$ used during optimisation. Across stocks and periods, the robust agent typically matches or improves upon the greedy benchmark for higher values of $\varepsilon$ in terms of Sharpe ratio. The gain is generally stronger on the higher-liquidity names. As $\varepsilon$ grows, the agent quotes more conservatively and tends to deliver a higher Sharpe ratio, although the relationship is not strictly monotone: the action-robustness parameter $\delta$, which controls the shape of the Sinkhorn sampling kernel $Q_{x,\delta}$, also influences the policy and can shift the conditional Sharpe curve, so $\varepsilon$ alone does not fully determine performance.

\begin{figure}[H]
    \centering
    \makebox[\textwidth][c]{\begin{minipage}{1.16\textwidth}\centering
        \begin{subfigure}{0.245\linewidth}
            \includegraphics[width=\linewidth]{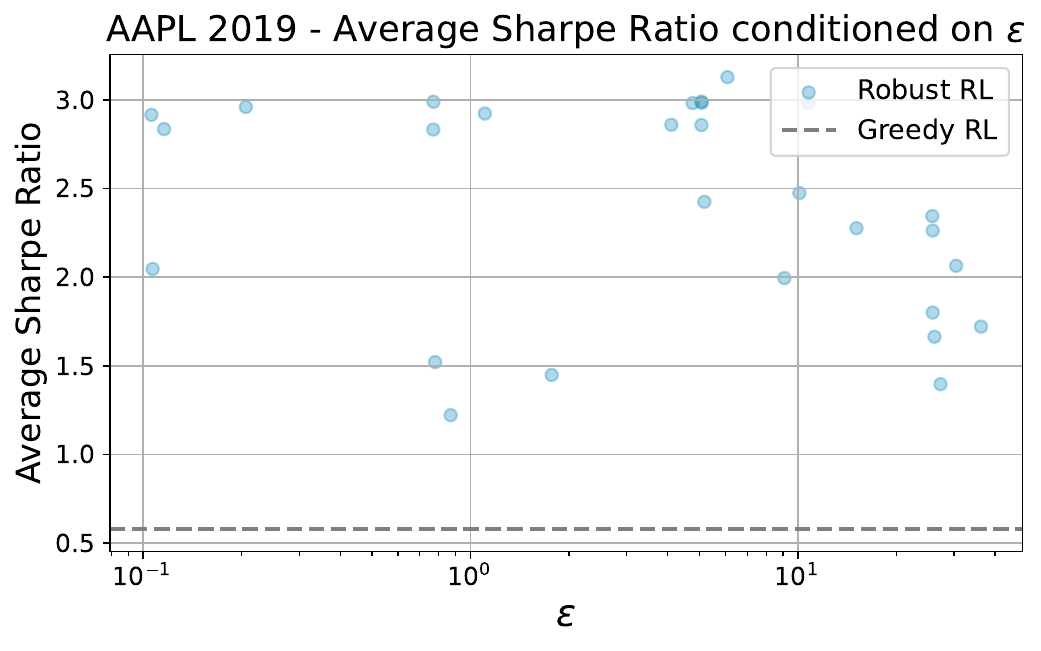}
            \caption{AAPL, 2019.}
        \end{subfigure}%
        \begin{subfigure}{0.245\linewidth}
            \includegraphics[width=\linewidth]{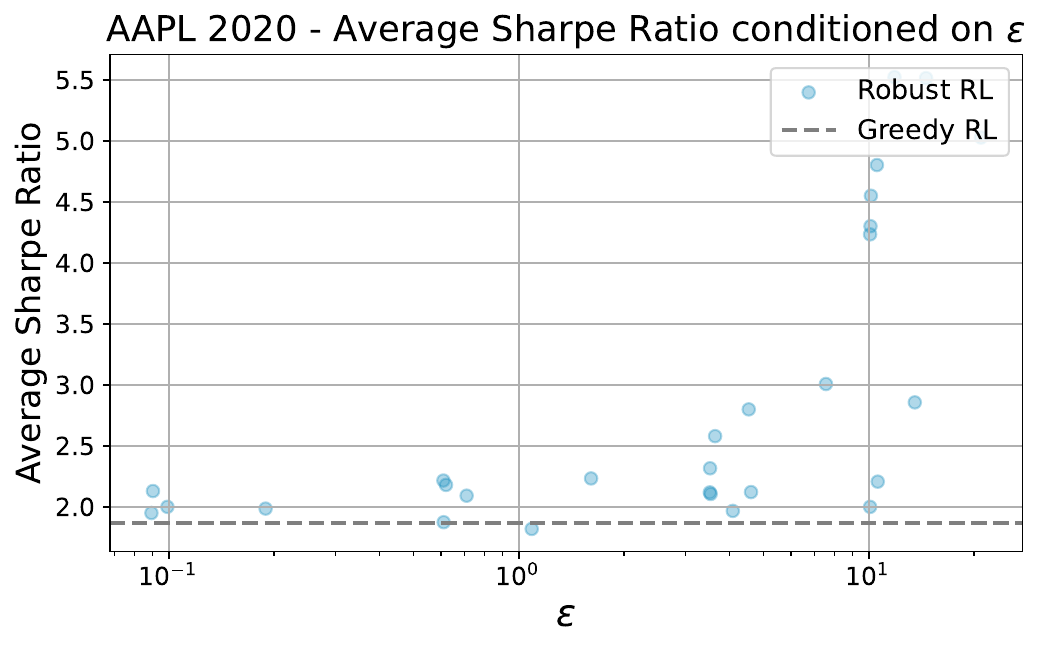}
            \caption{AAPL, 2020.}
        \end{subfigure}%
        \begin{subfigure}{0.245\linewidth}
            \includegraphics[width=\linewidth]{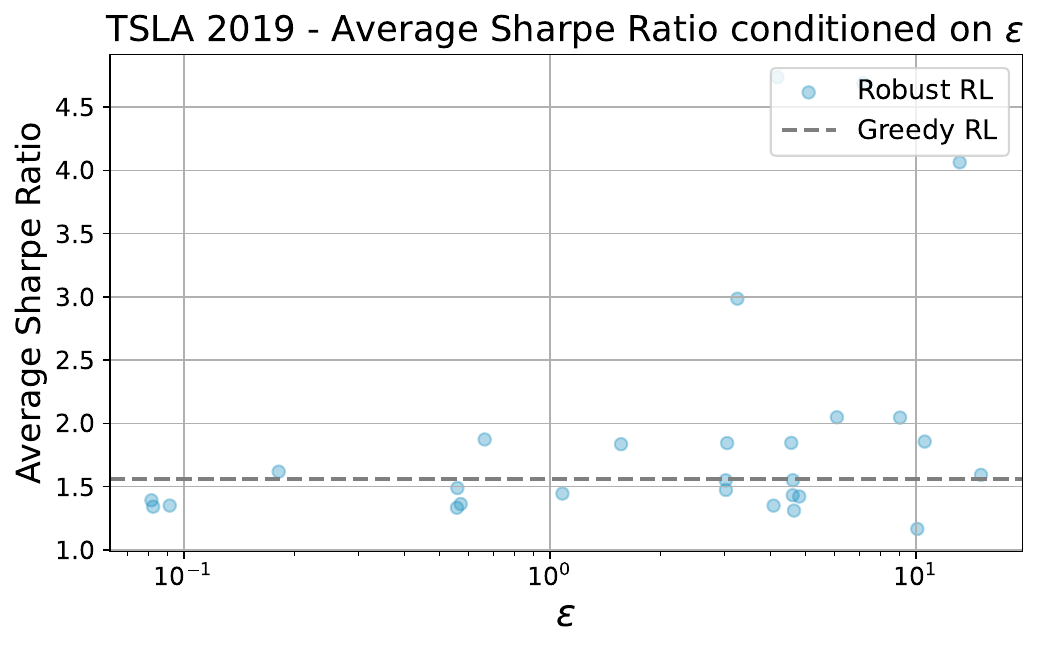}
            \caption{TSLA, 2019.}
        \end{subfigure}%
        \begin{subfigure}{0.245\linewidth}
            \includegraphics[width=\linewidth]{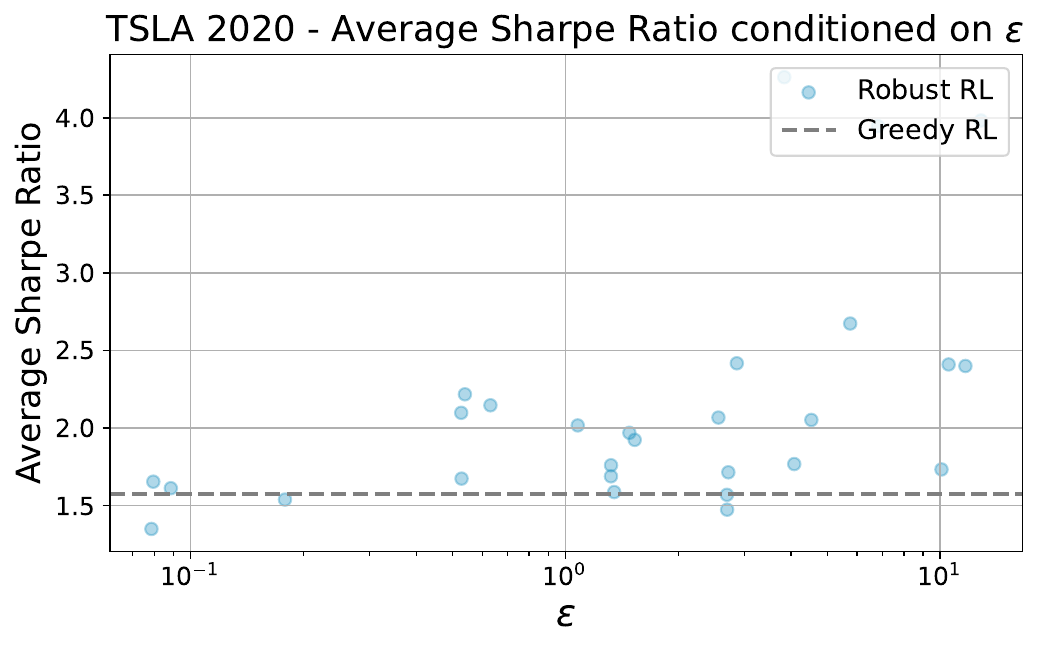}
            \caption{TSLA, 2020.}
        \end{subfigure}

        \par\vspace{0.1em}

        \begin{subfigure}{0.245\linewidth}
            \includegraphics[width=\linewidth]{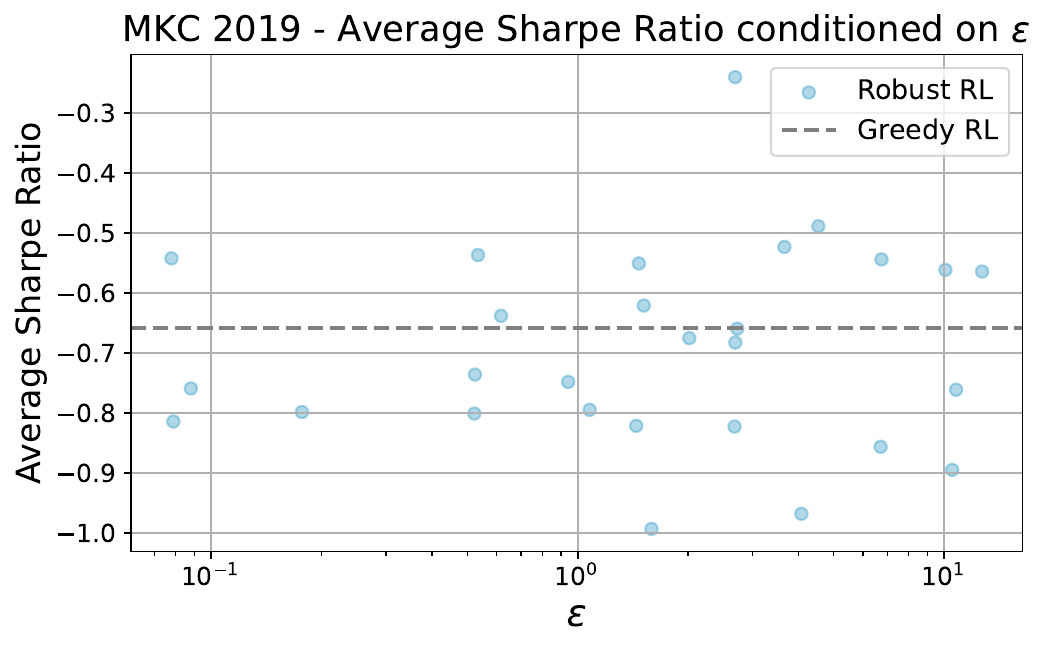}
            \caption{MKC, 2019.}
        \end{subfigure}%
        \begin{subfigure}{0.245\linewidth}
            \includegraphics[width=\linewidth]{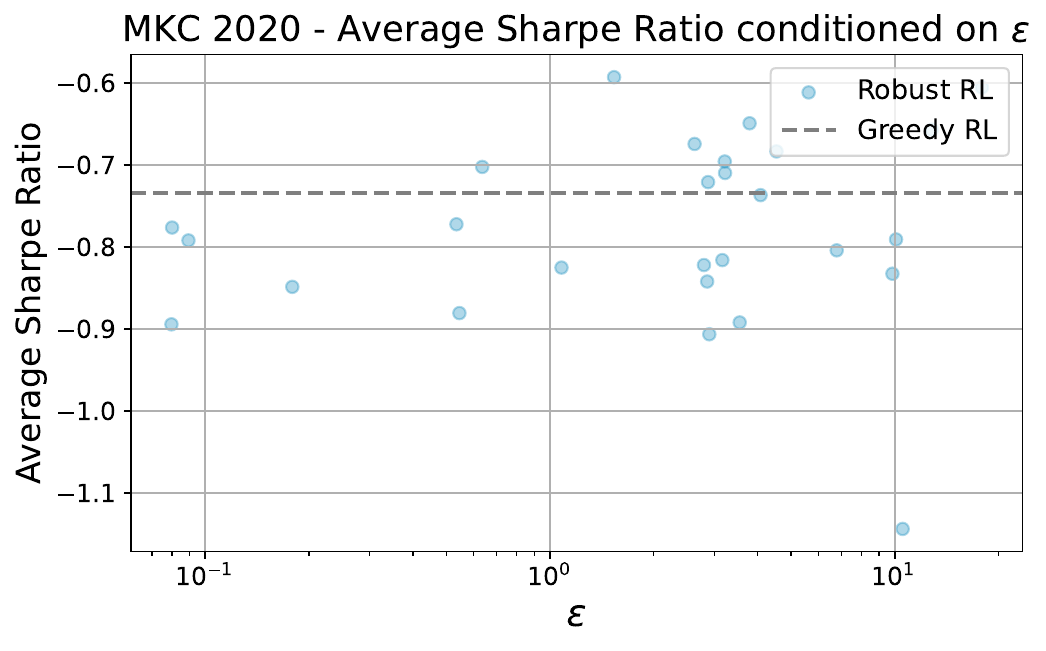}
            \caption{MKC, 2020.}
        \end{subfigure}%
        \begin{subfigure}{0.245\linewidth}
            \includegraphics[width=\linewidth]{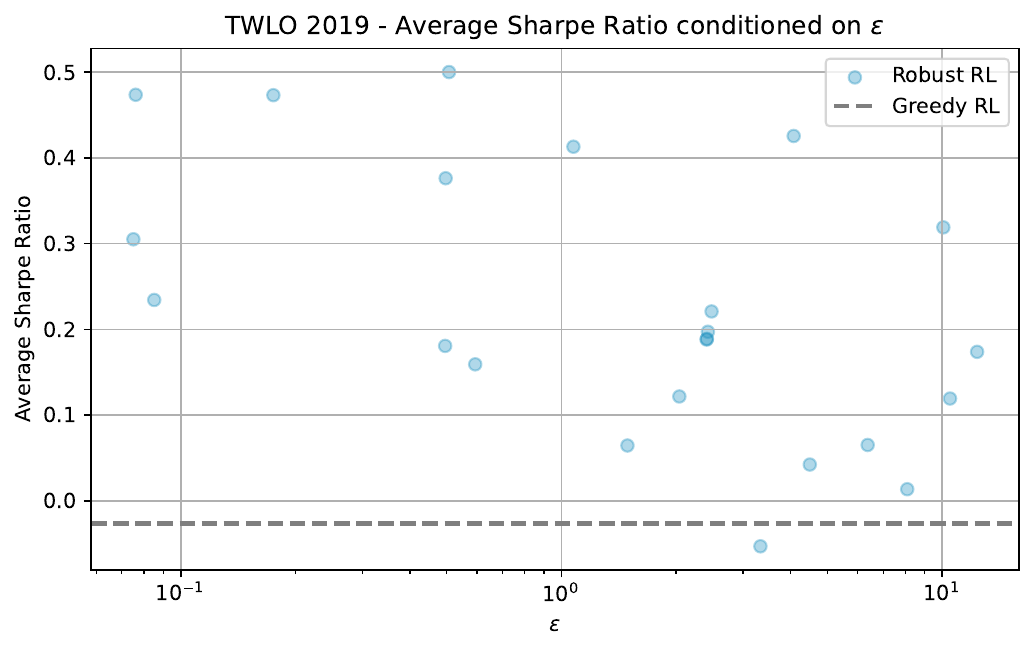}
            \caption{TWLO, 2019.}
        \end{subfigure}%
        \begin{subfigure}{0.245\linewidth}
            \includegraphics[width=\linewidth]{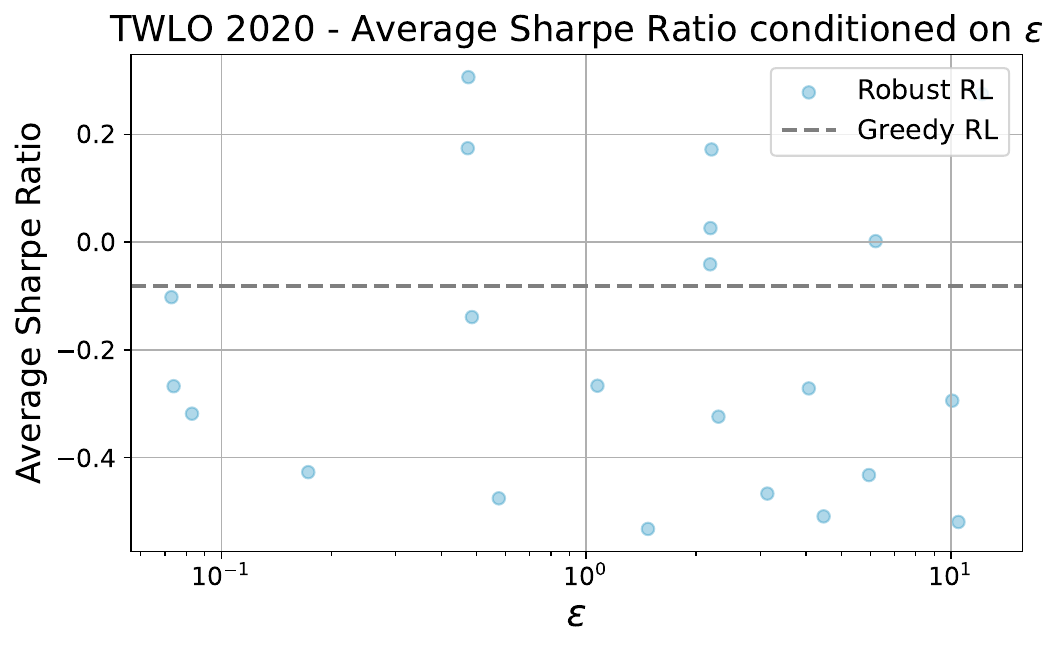}
            \caption{TWLO, 2020.}
        \end{subfigure}
    \end{minipage}}
    \caption{Average test Sharpe ratio of the robust agent conditional on the original Sinkhorn radius $\varepsilon$ (markers), with the greedy benchmark shown as a dashed line. Across panels, larger values of $\varepsilon$ often deliver comparable or higher Sharpe ratios, especially in 2020, although the relationship is not fully monotone across stocks.}
    \label{fig:sharpe_vs_eps_real}
\end{figure}

\subsection[Feature Importance via Shapley Values]{Feature Importance via Shapley Values}\label{sec:appendix_shapley}

To diagnose which state features drive the reinforcement learning agent's policy, we report the distribution of absolute Shapley values across the test period for each of the four quote components: bid spread, ask spread, bid quantity, and ask quantity in figures ~\ref{fig:shapley_aapl_2019}-\ref{fig:shapley_twlo_2020}. For illustration, we report these diagnostics for the best-Sharpe configuration on the test Pareto frontier in each stock-year panel, and we include the corresponding decomposition for the greedy agent.
Across stocks and periods, the quoted spreads and quantities are driven primarily by trade-flow features\textemdash the number and average size of buy and sell trades, the order-flow imbalance, and the time-to-maturity. Consistent with classical inventory-management intuition, the decomposition is asymmetric across sides: sell-trade features carry more weight on the bid quote, whereas buy-trade features dominate on the ask quote. The greedy and robust decompositions are also broadly similar, indicating that robustness mainly shifts the level of the quoting policy rather than the underlying set of state variables to which the agent is most responsive.

\begin{figure}[H]
    \centering
    \makebox[\textwidth][c]{\begin{minipage}{1.16\textwidth}\centering
        \begin{subfigure}{0.24\linewidth}
            \includegraphics[width=\linewidth]{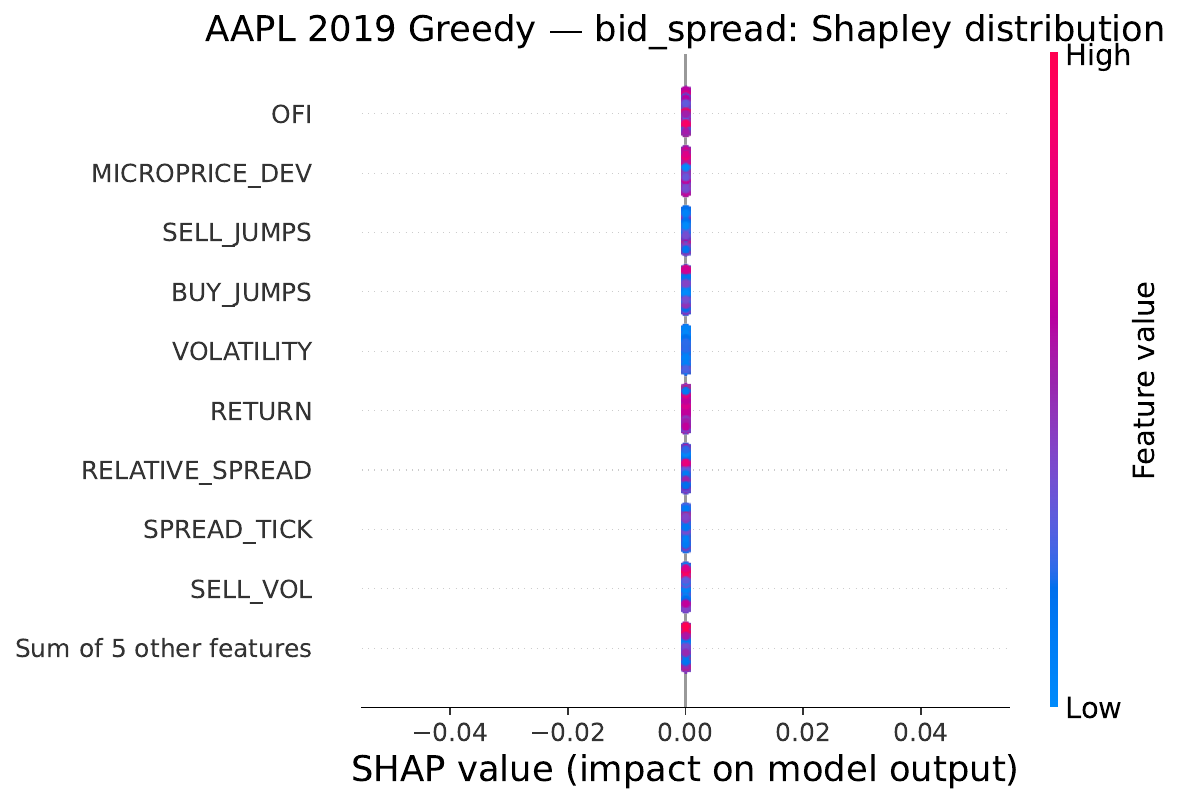}
            \caption{Greedy - Bid spread.}
        \end{subfigure}\hfill
        \begin{subfigure}{0.24\linewidth}
            \includegraphics[width=\linewidth]{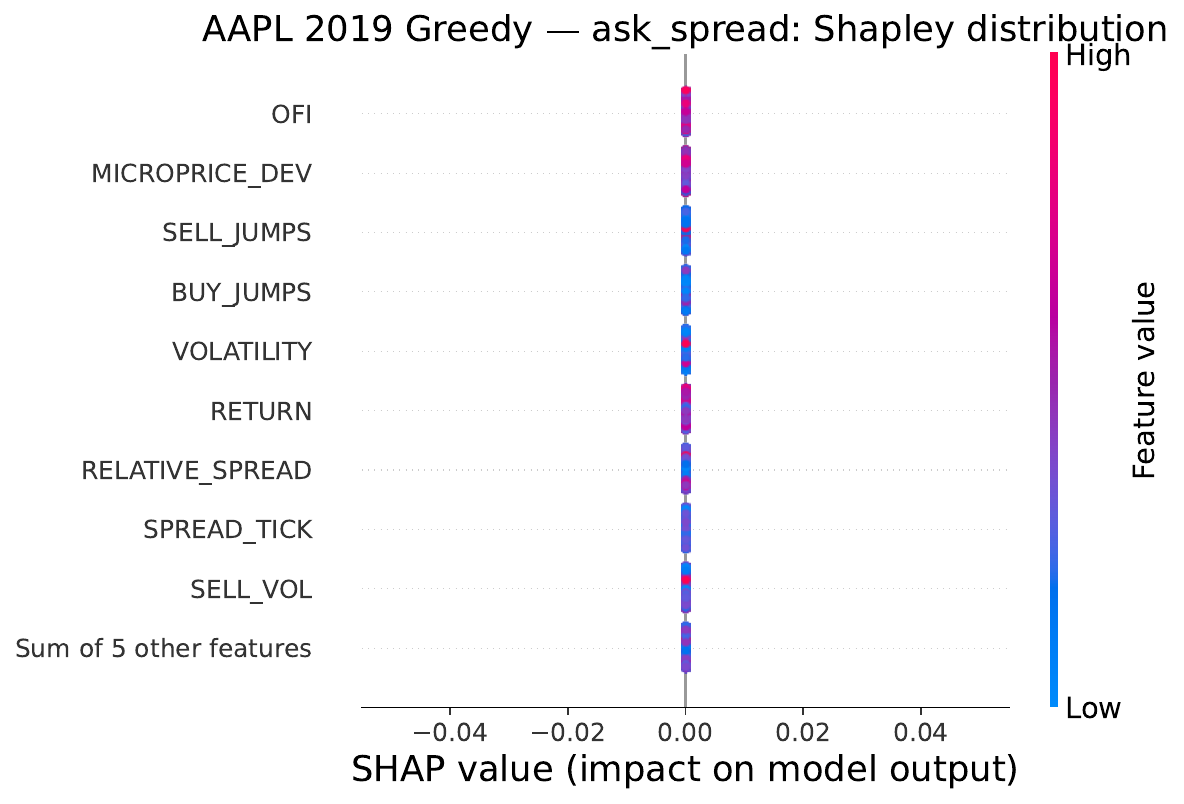}
            \caption{Greedy - Ask spread.}
        \end{subfigure}\hfill
        \begin{subfigure}{0.24\linewidth}
            \includegraphics[width=\linewidth]{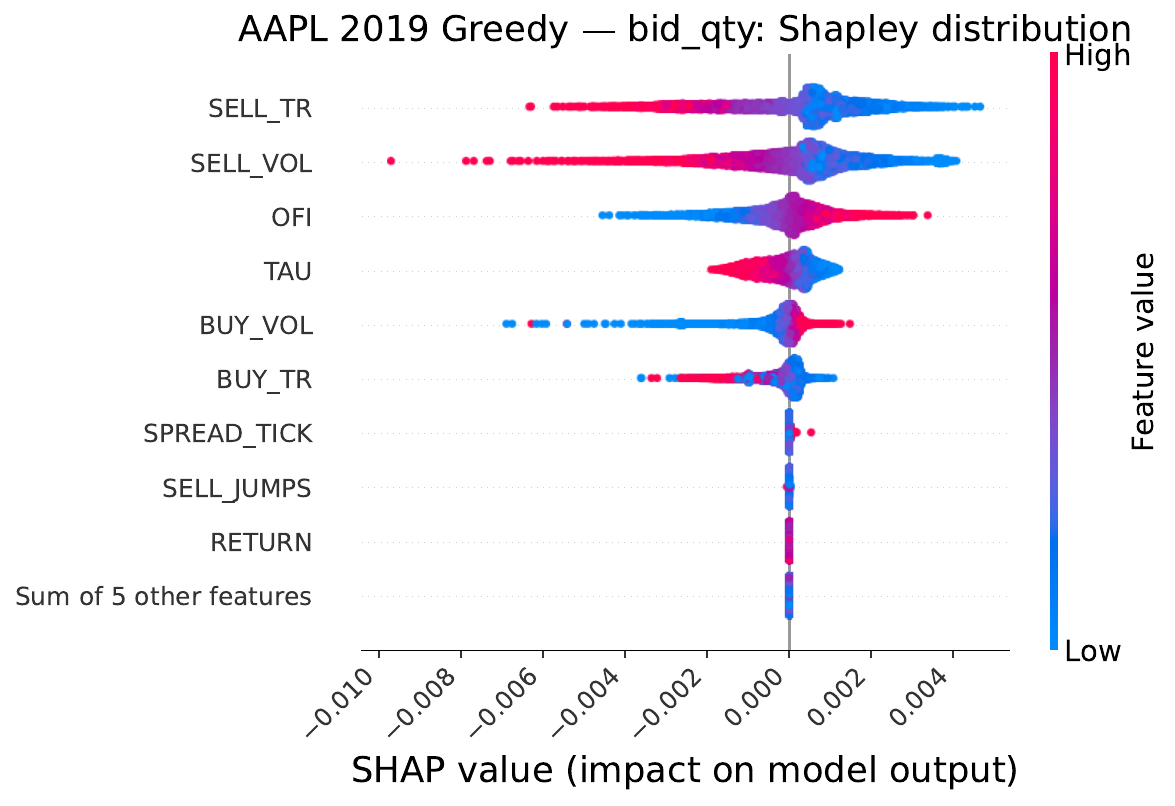}
            \caption{Greedy - Bid quantity.}
        \end{subfigure}\hfill
        \begin{subfigure}{0.24\linewidth}
            \includegraphics[width=\linewidth]{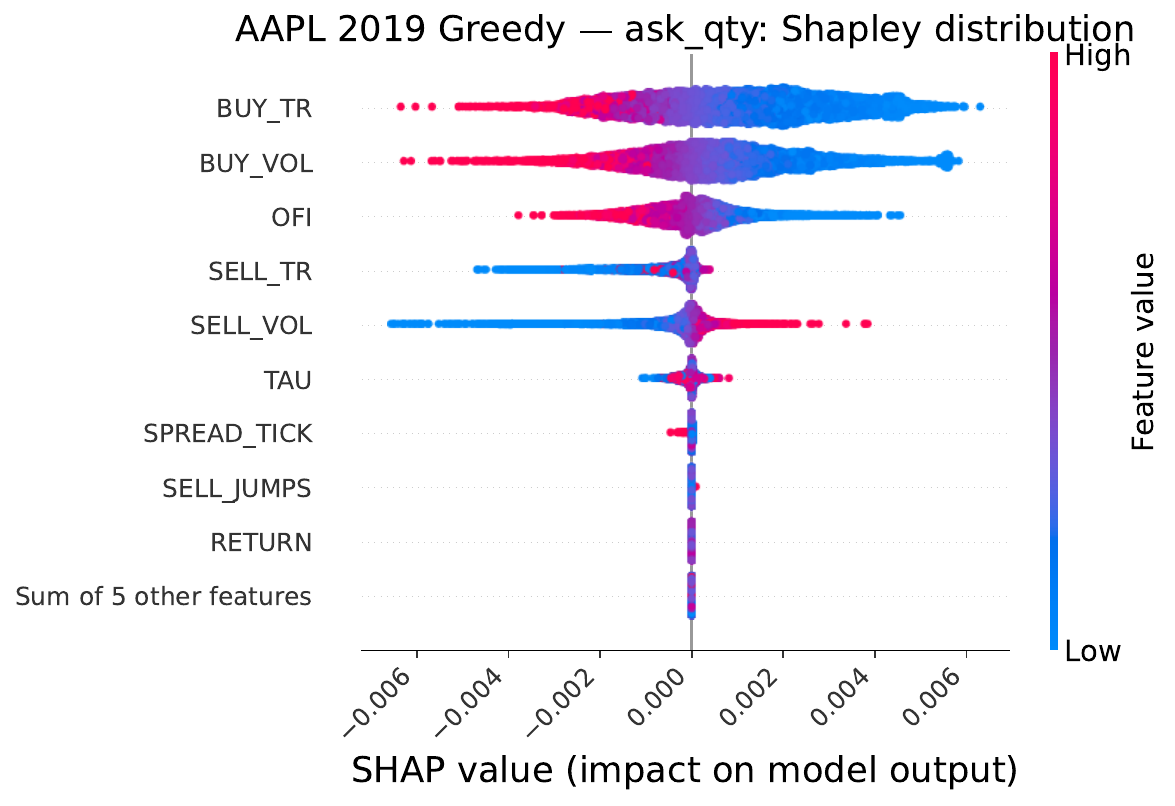}
            \caption{Greedy - Ask quantity.}
        \end{subfigure}
\vspace{0.6em}
        \begin{subfigure}{0.24\linewidth}
            \includegraphics[width=\linewidth]{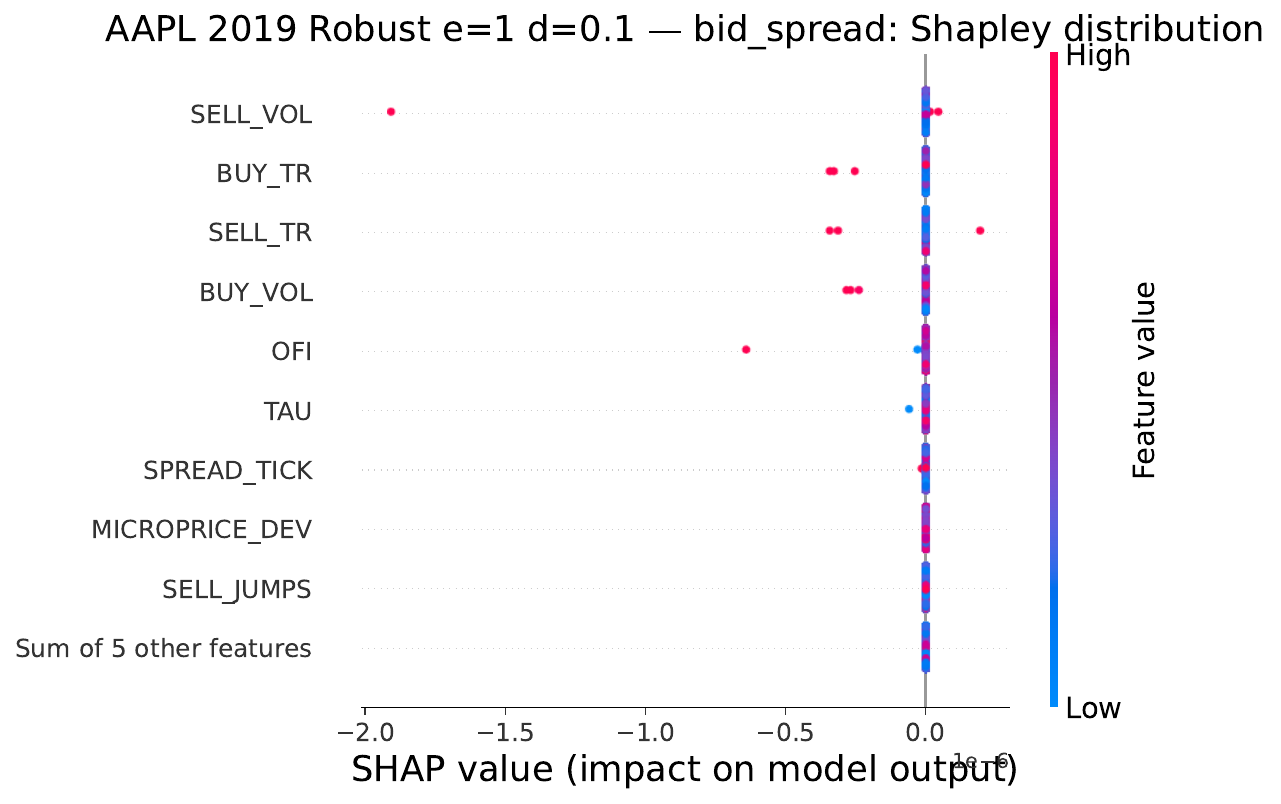}
            \caption{Robust - Bid spread.}
        \end{subfigure}\hfill
        \begin{subfigure}{0.24\linewidth}
            \includegraphics[width=\linewidth]{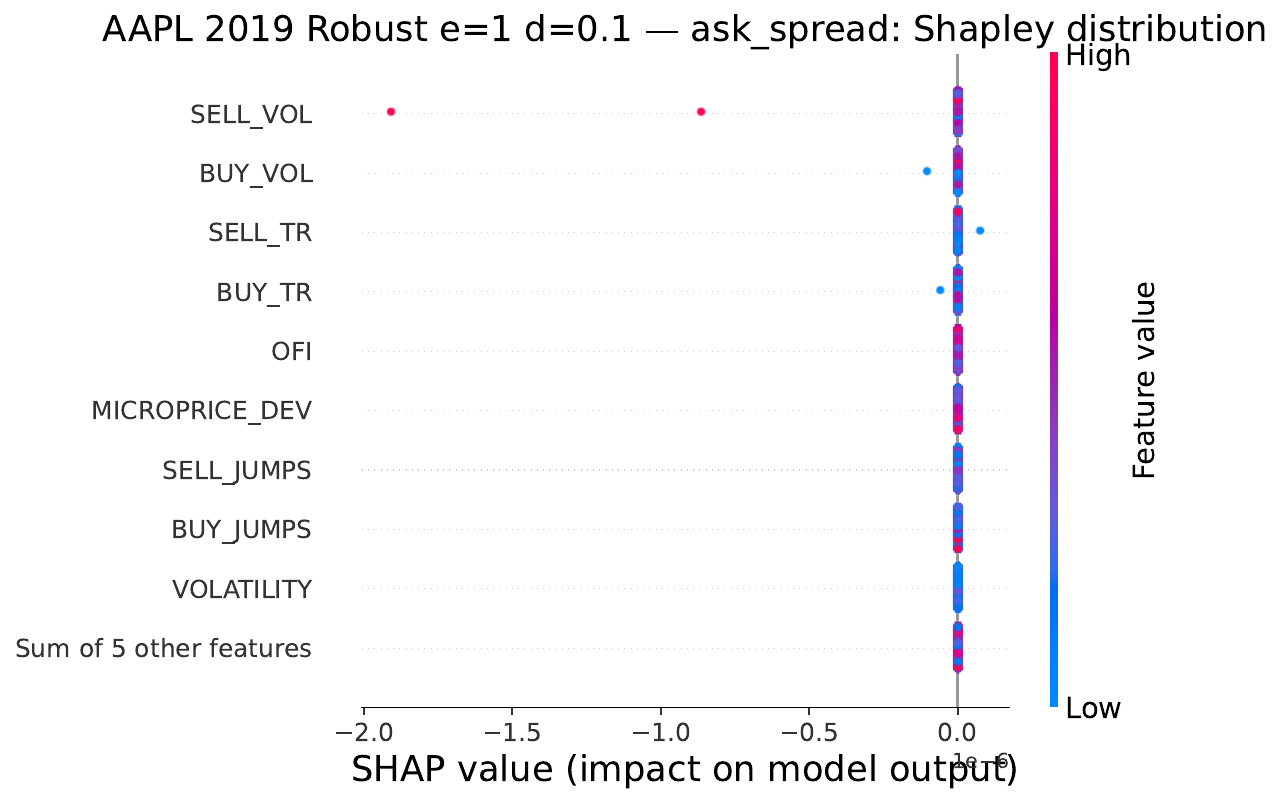}
            \caption{Robust - Ask spread.}
        \end{subfigure}\hfill
        \begin{subfigure}{0.24\linewidth}
            \includegraphics[width=\linewidth]{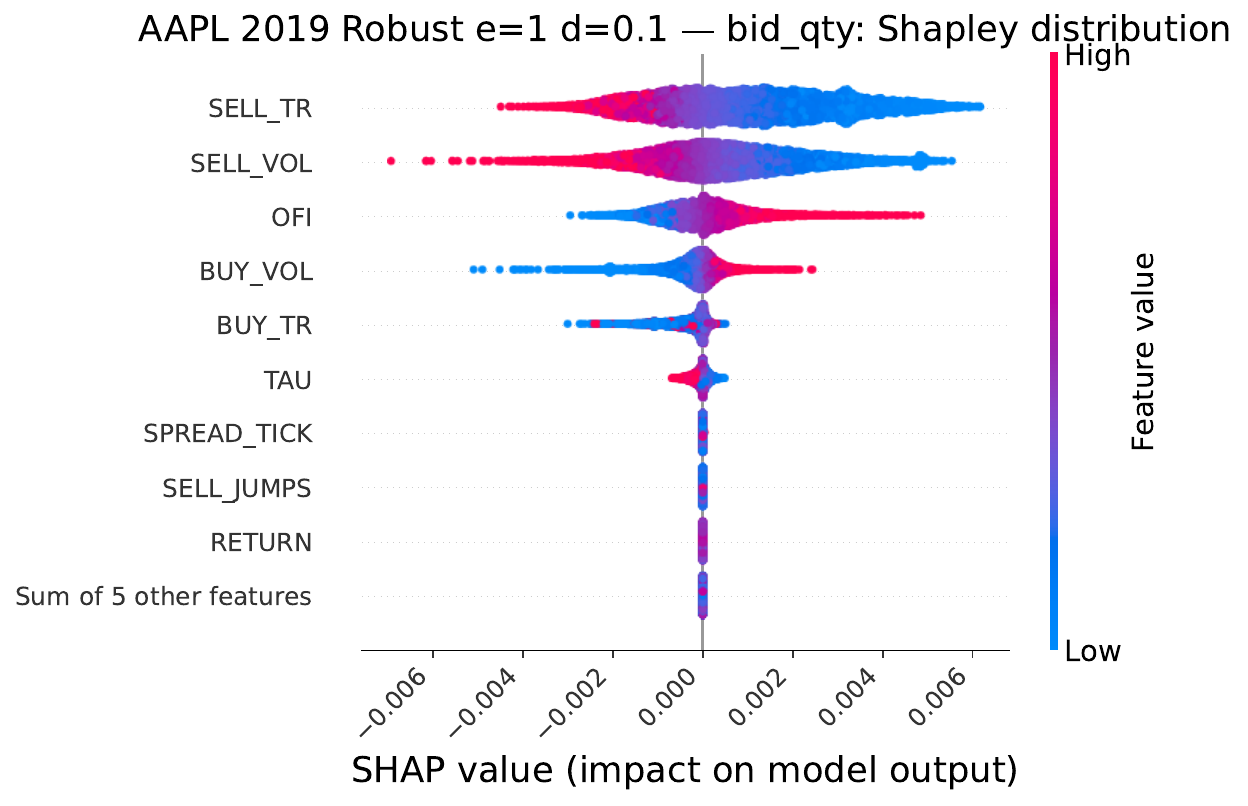}
            \caption{Robust - Bid quantity.}
        \end{subfigure}\hfill
        \begin{subfigure}{0.24\linewidth}
            \includegraphics[width=\linewidth]{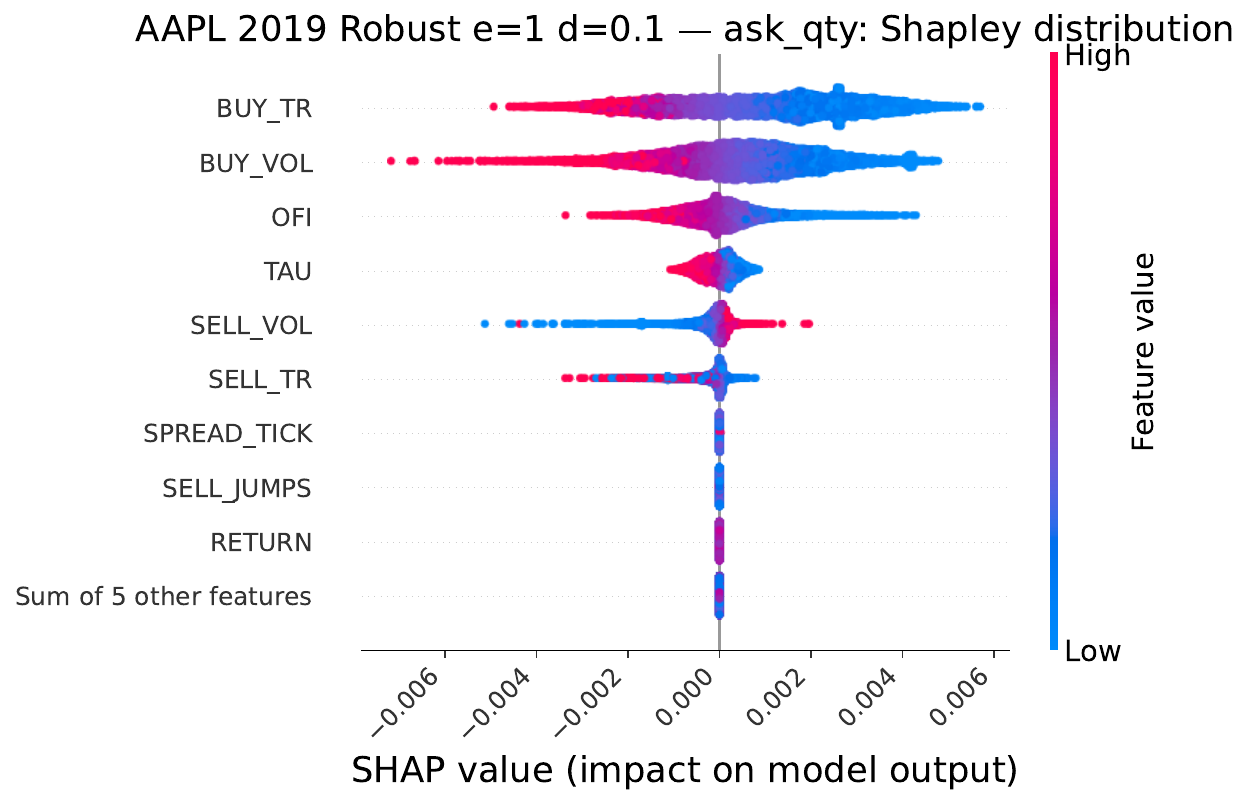}
            \caption{Robust - Ask quantity.}
        \end{subfigure}
    \end{minipage}}
    \caption{Distribution of absolute Shapley values for greedy and robust policies with $\bar{\varepsilon} = 1$ and $\delta = 0.1$ on AAPL, 2019, where the robust policy is the best test Sharpe Pareto configuration. The figure shows that both policies are driven primarily by trade-flow and timing variables, with robustness changing the relative weight of these drivers only modestly.}
    \label{fig:shapley_aapl_2019}
\end{figure}

\begin{figure}[H]
    \centering
    \makebox[\textwidth][c]{\begin{minipage}{1.16\textwidth}\centering
        \begin{subfigure}{0.24\linewidth}
            \includegraphics[width=\linewidth]{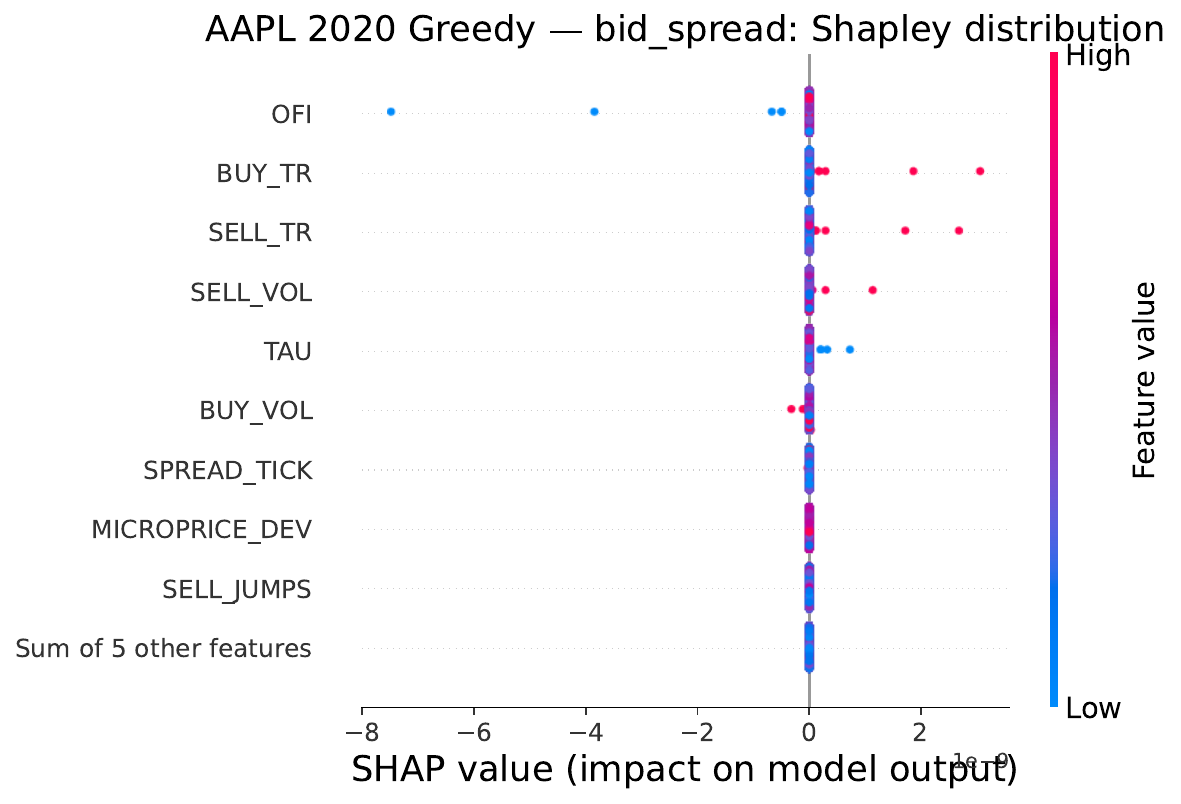}
            \caption{Greedy - Bid spread.}
        \end{subfigure}\hfill
        \begin{subfigure}{0.24\linewidth}
            \includegraphics[width=\linewidth]{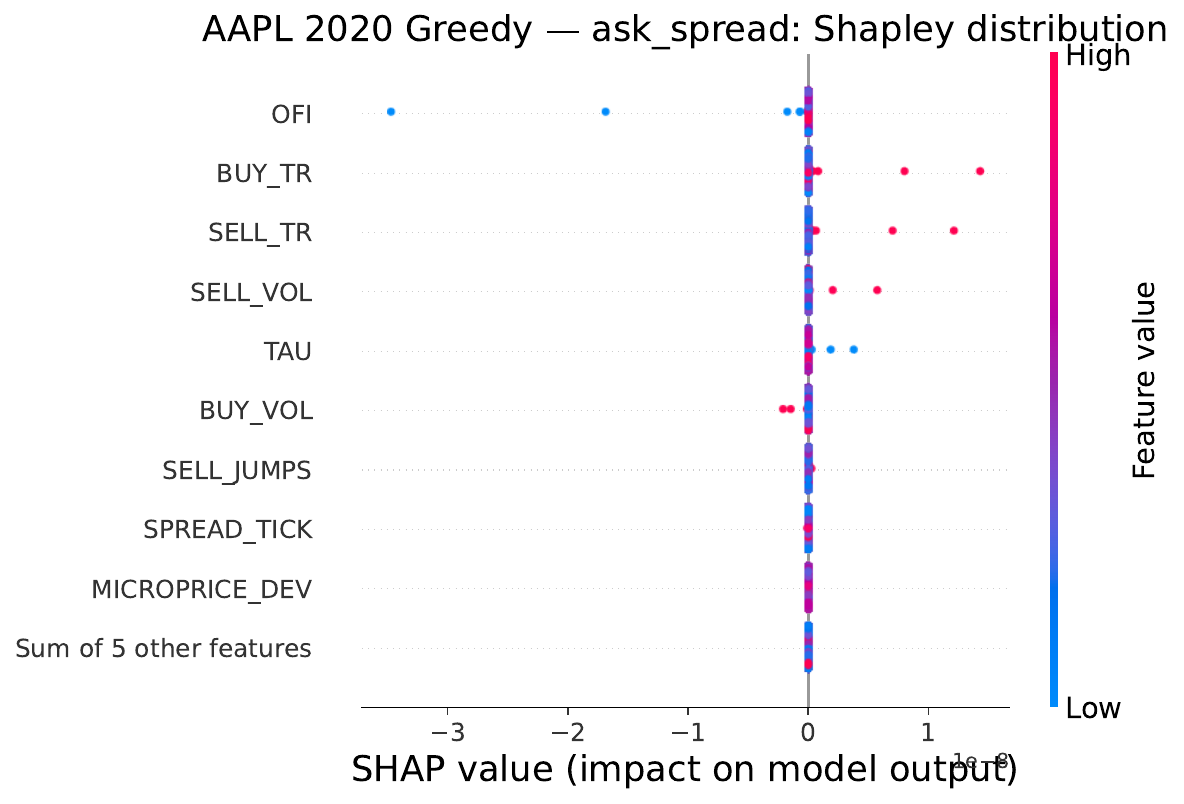}
            \caption{Greedy - Ask spread.}
        \end{subfigure}\hfill
        \begin{subfigure}{0.24\linewidth}
            \includegraphics[width=\linewidth]{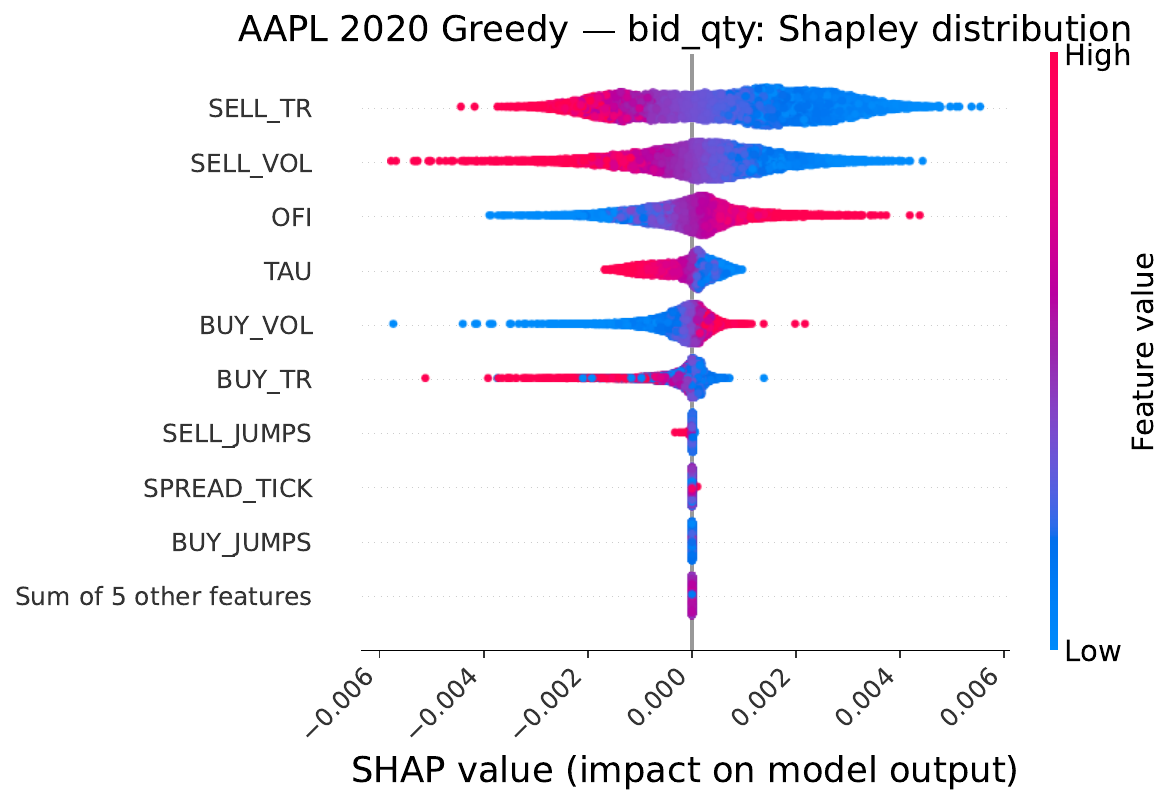}
            \caption{Greedy - Bid quantity.}
        \end{subfigure}\hfill
        \begin{subfigure}{0.24\linewidth}
            \includegraphics[width=\linewidth]{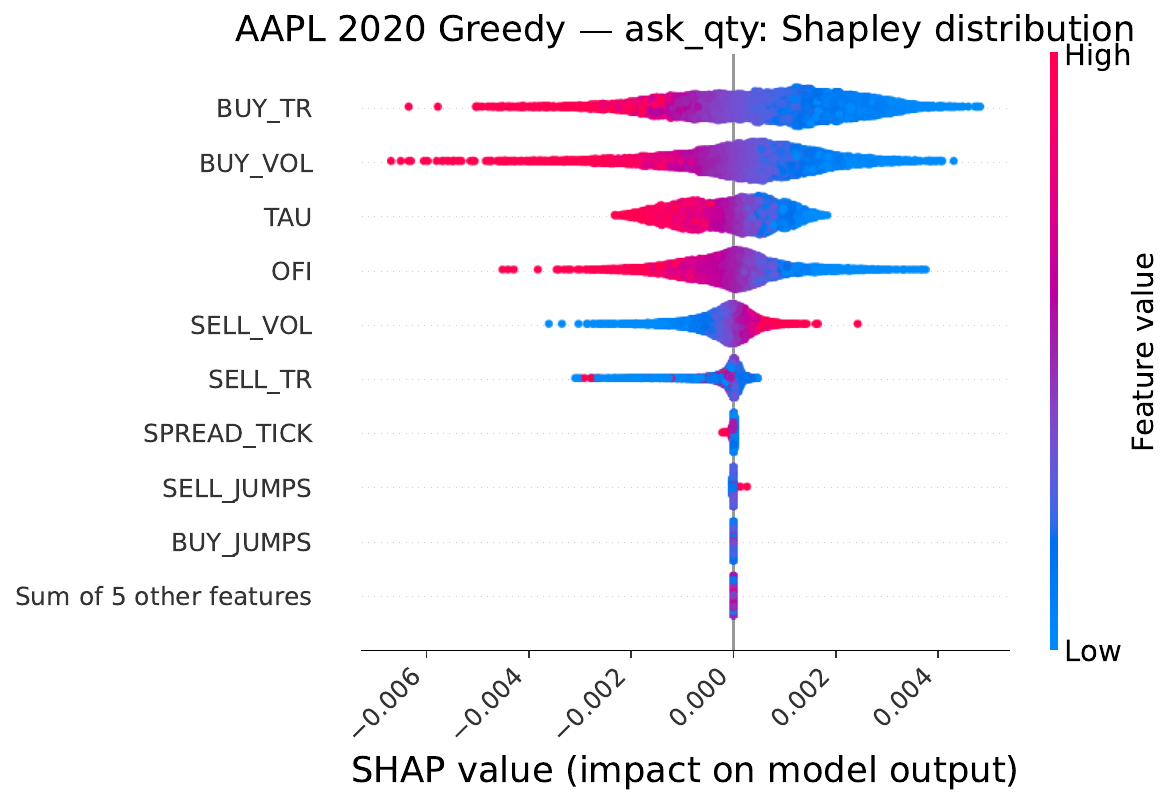}
            \caption{Greedy - Ask quantity.}
        \end{subfigure}
\vspace{0.6em}
        \begin{subfigure}{0.24\linewidth}
            \includegraphics[width=\linewidth]{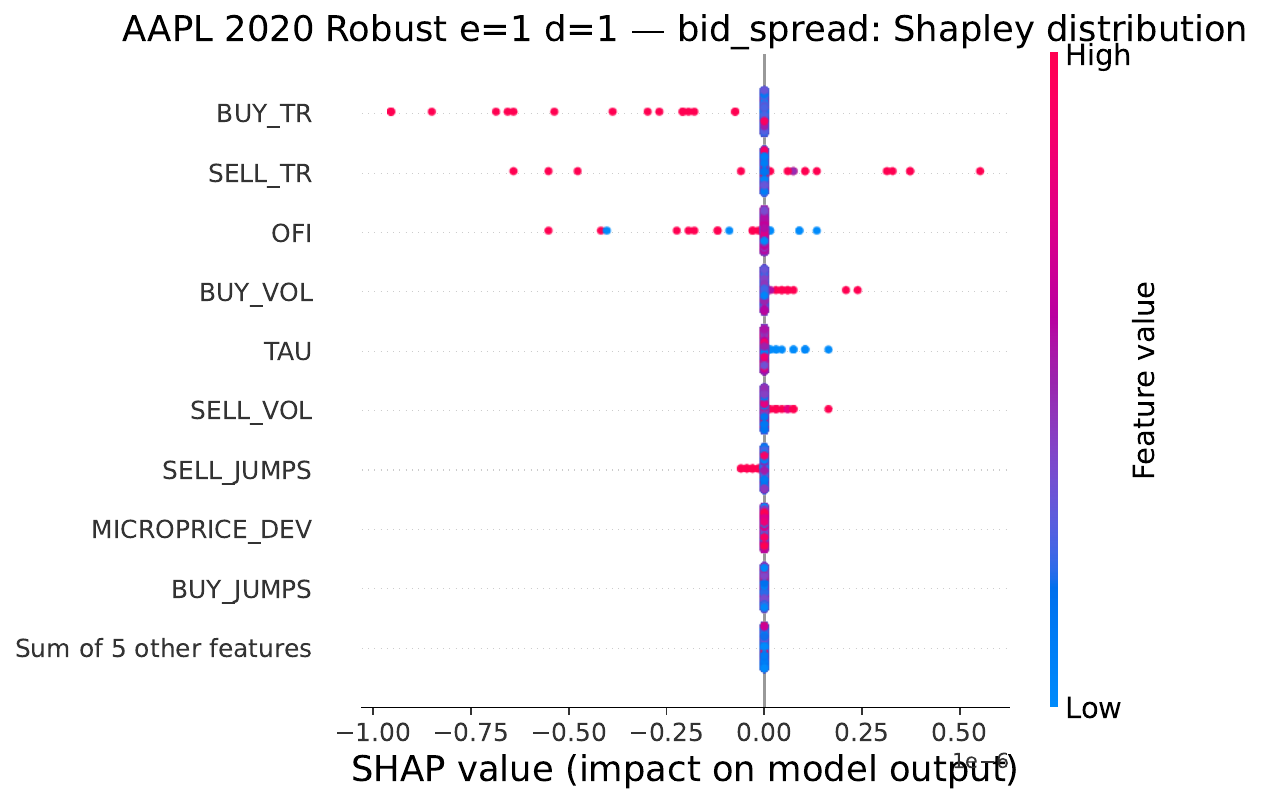}
            \caption{Robust - Bid spread.}
        \end{subfigure}\hfill
        \begin{subfigure}{0.24\linewidth}
            \includegraphics[width=\linewidth]{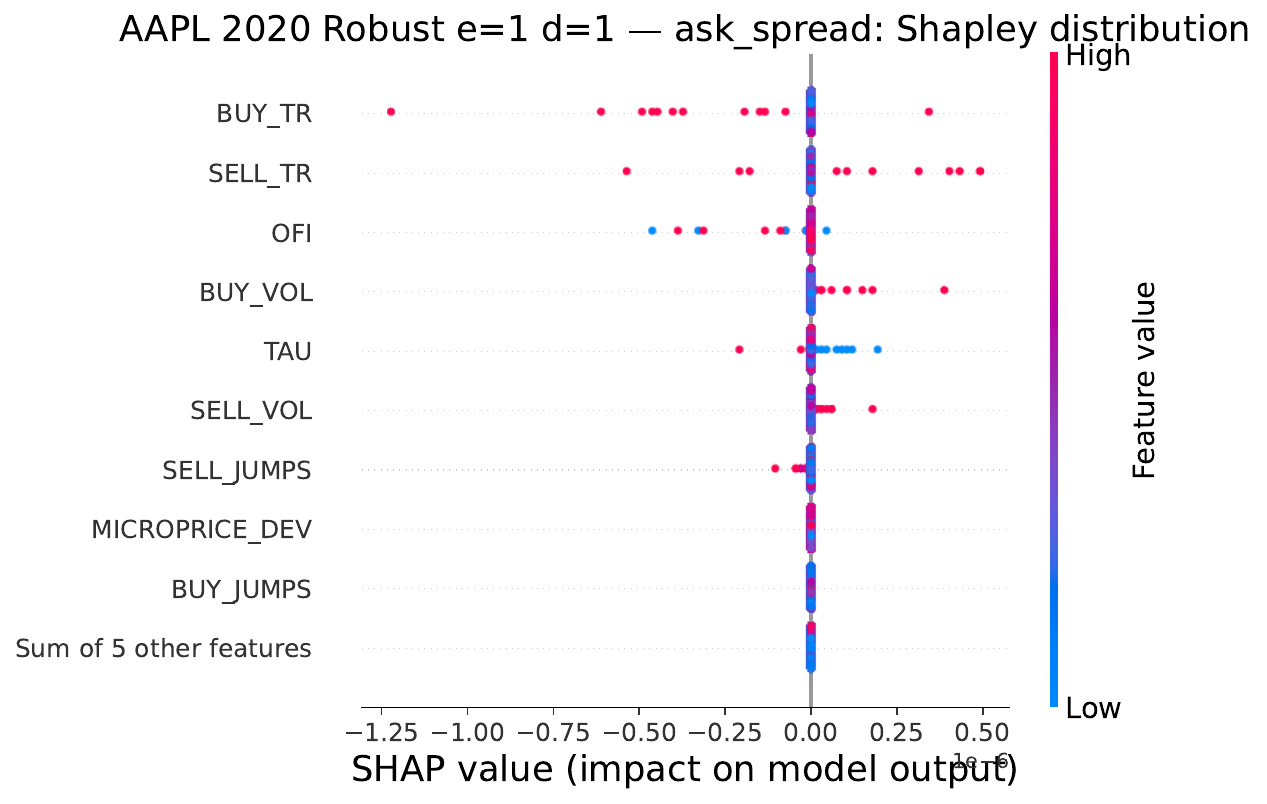}
            \caption{Robust - Ask spread.}
        \end{subfigure}\hfill
        \begin{subfigure}{0.24\linewidth}
            \includegraphics[width=\linewidth]{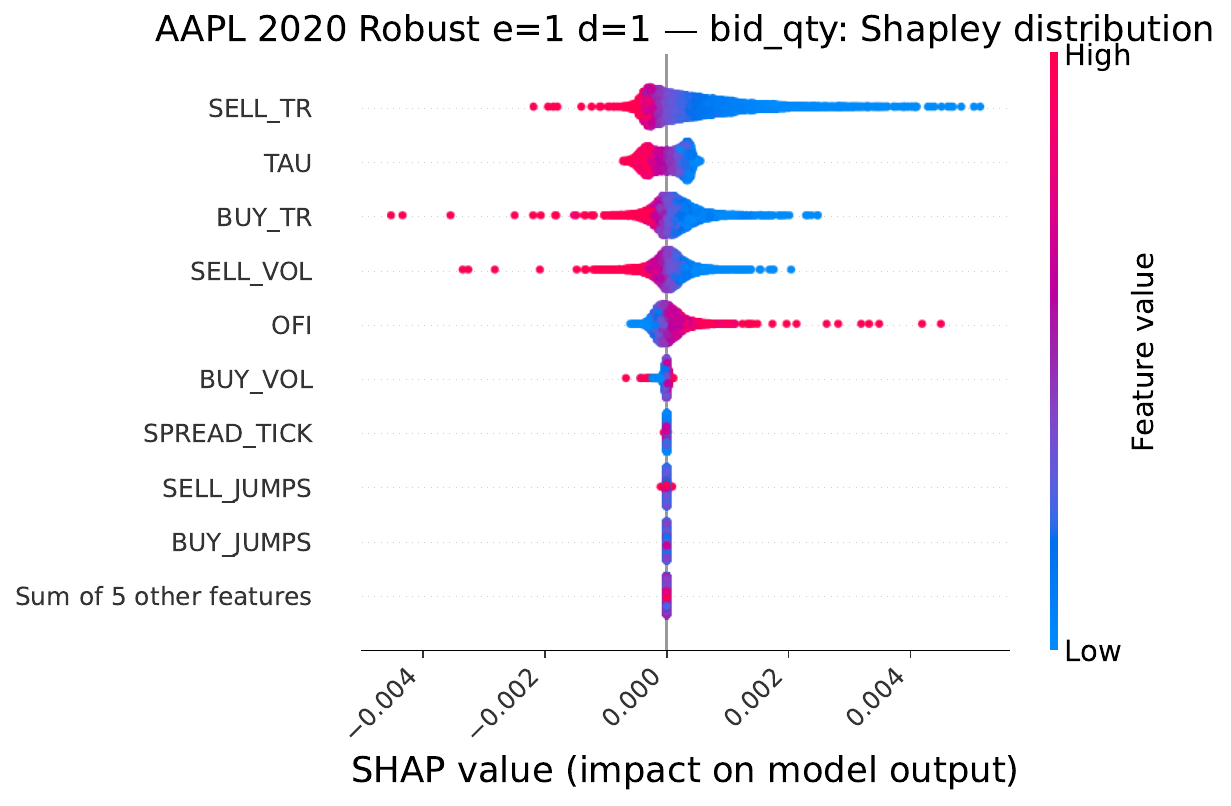}
            \caption{Robust - Bid quantity.}
        \end{subfigure}\hfill
        \begin{subfigure}{0.24\linewidth}
            \includegraphics[width=\linewidth]{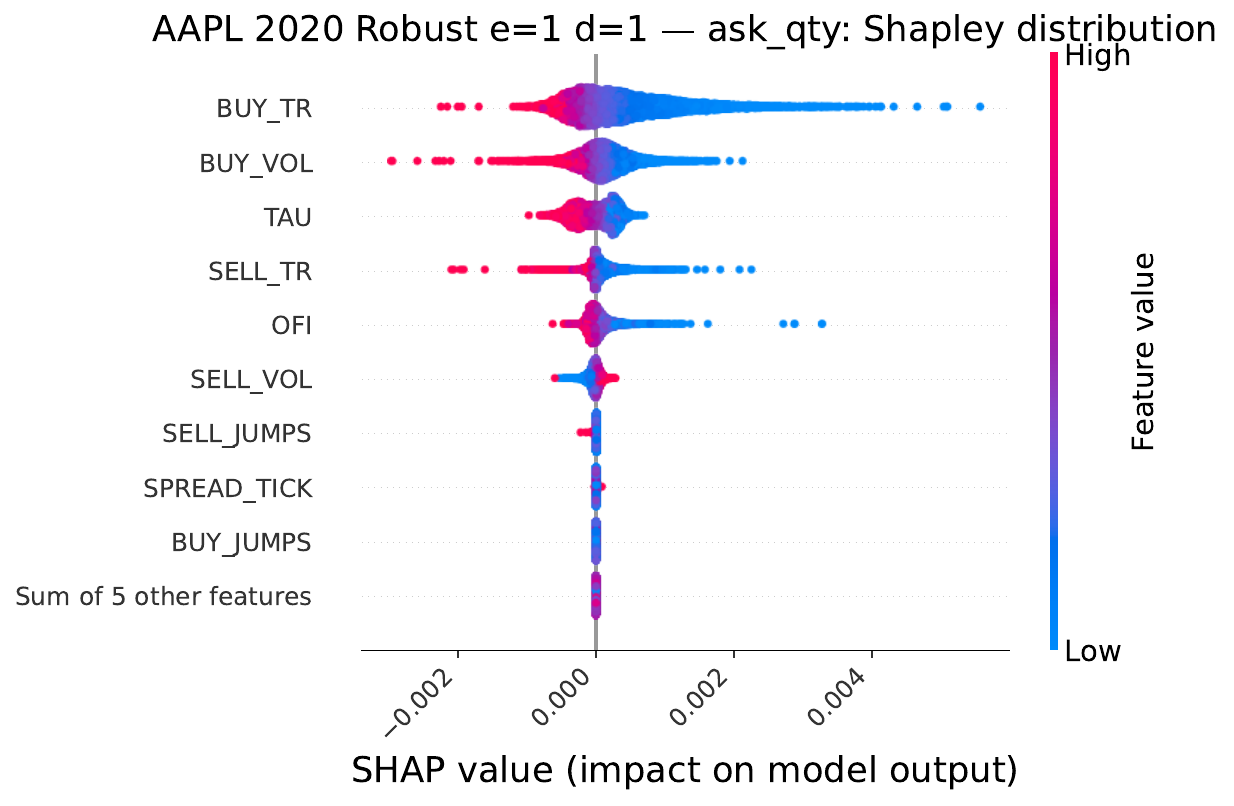}
            \caption{Robust - Ask quantity.}
        \end{subfigure}
    \end{minipage}}
    \caption{Distribution of absolute Shapley values for greedy and robust policies with $\bar{\varepsilon} = 1$ and $\delta = 1$ on AAPL, 2020, where the robust policy is the best test Sharpe Pareto configuration. The figure shows that both policies are driven primarily by trade-flow and timing variables, with robustness changing the relative weight of these drivers more than in 2019.}
    \label{fig:shapley_aapl_1920}
\end{figure}

\begin{figure}[H]
    \centering
    \makebox[\textwidth][c]{\begin{minipage}{1.16\textwidth}\centering
        \begin{subfigure}{0.24\linewidth}
            \includegraphics[width=\linewidth]{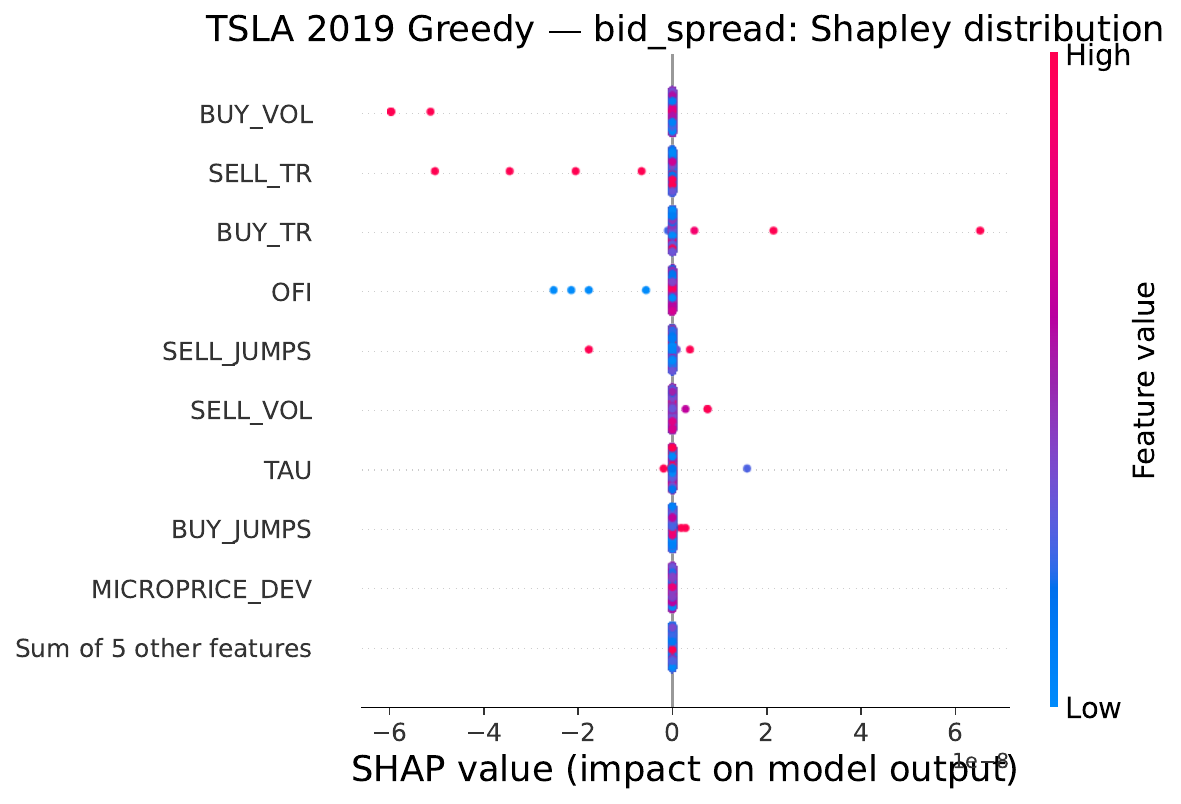}
            \caption{Greedy - Bid spread.}
        \end{subfigure}\hfill
        \begin{subfigure}{0.24\linewidth}
            \includegraphics[width=\linewidth]{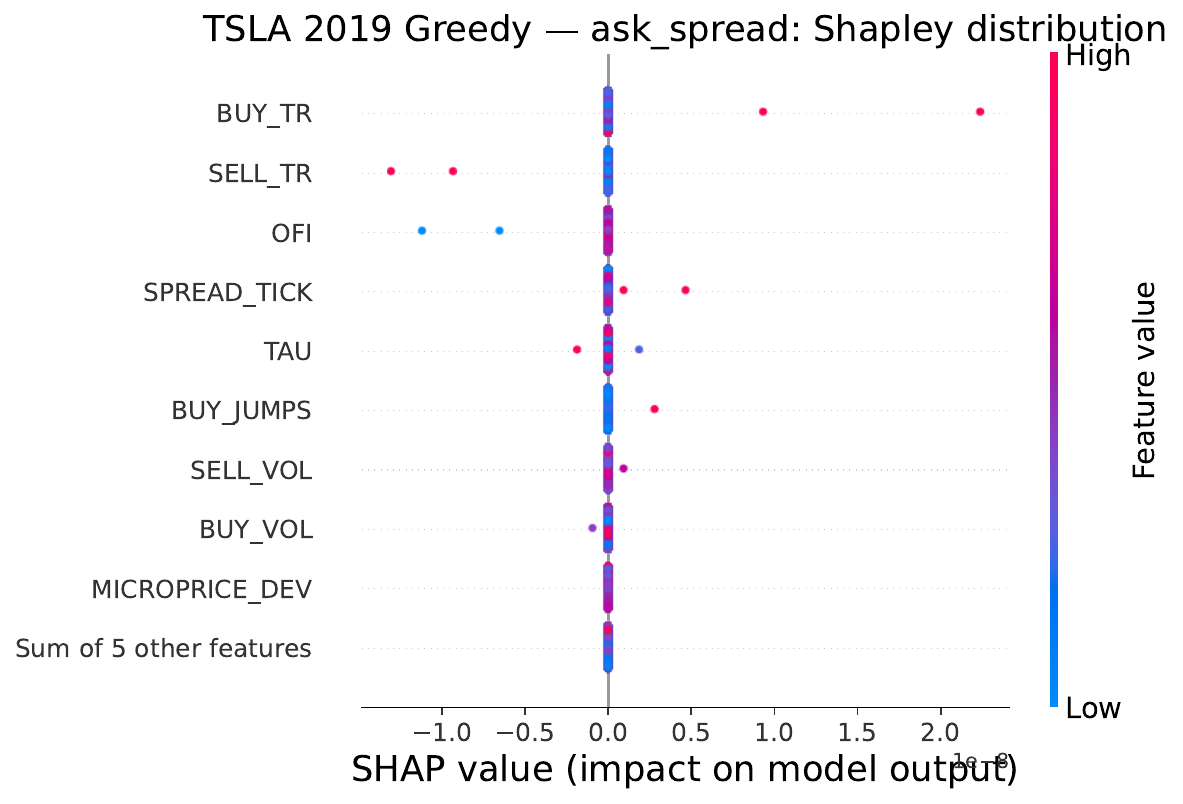}
            \caption{Greedy - Ask spread.}
        \end{subfigure}\hfill
        \begin{subfigure}{0.24\linewidth}
            \includegraphics[width=\linewidth]{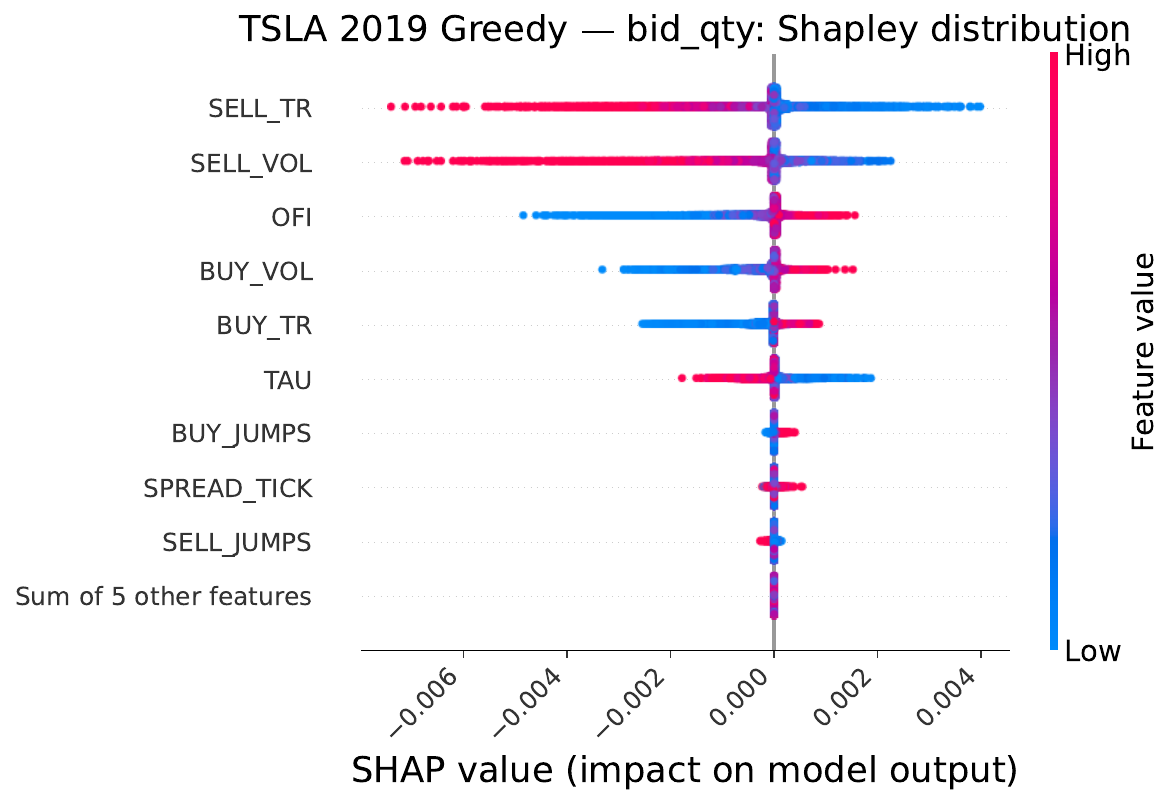}
            \caption{Greedy - Bid quantity.}
        \end{subfigure}\hfill
        \begin{subfigure}{0.24\linewidth}
            \includegraphics[width=\linewidth]{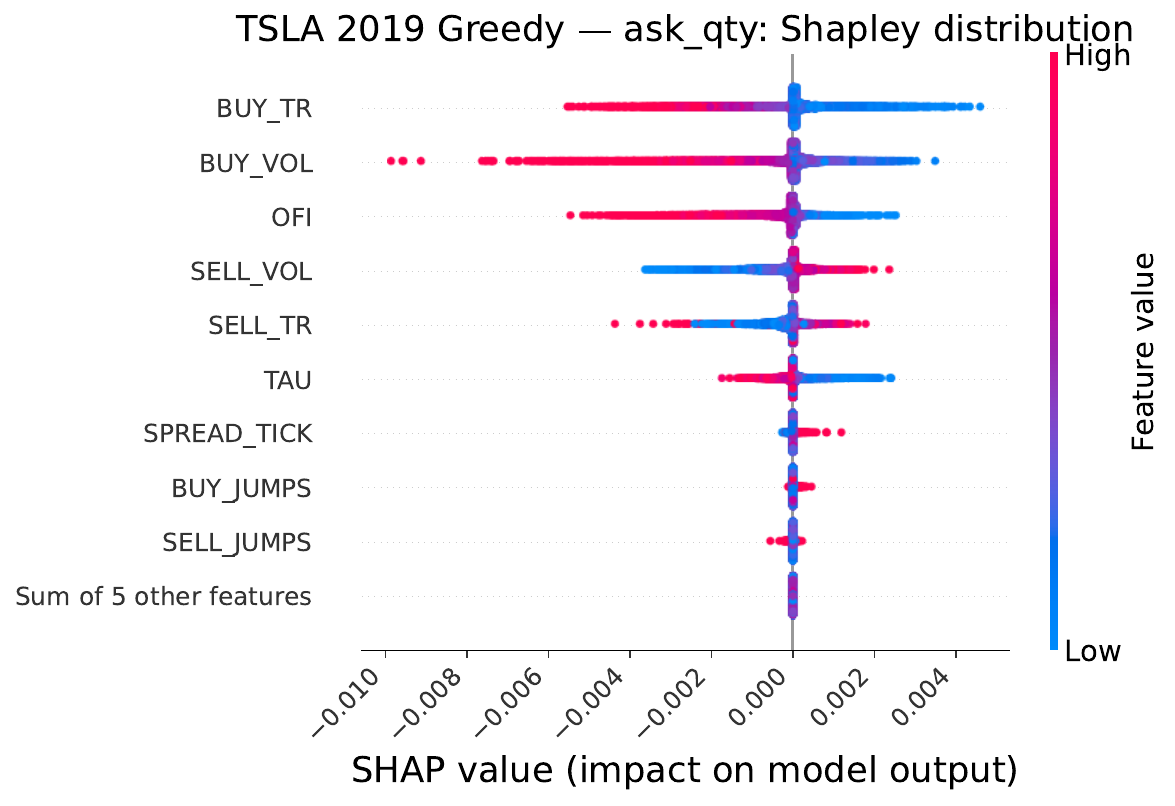}
            \caption{Greedy - Ask quantity.}
        \end{subfigure}
\vspace{0.6em}
        \begin{subfigure}{0.24\linewidth}
            \includegraphics[width=\linewidth]{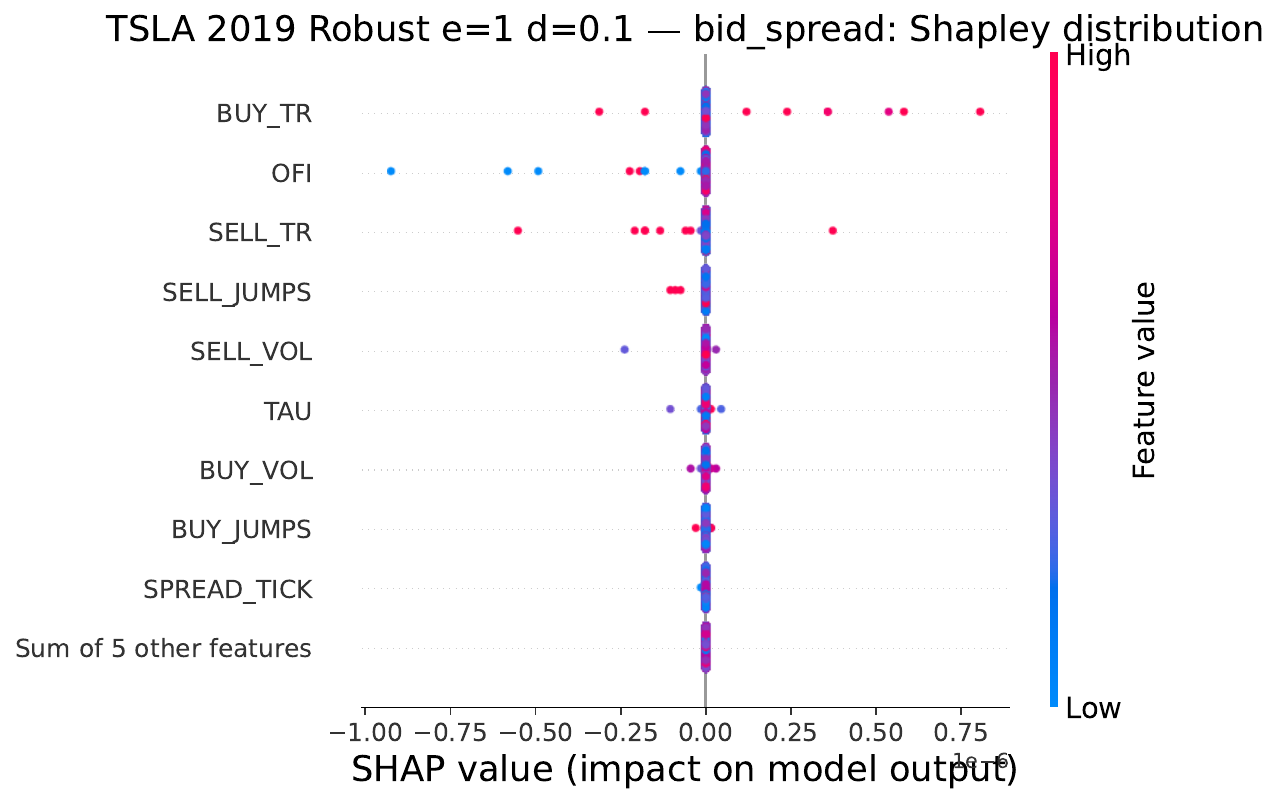}
            \caption{Robust - Bid spread.}
        \end{subfigure}\hfill
        \begin{subfigure}{0.24\linewidth}
            \includegraphics[width=\linewidth]{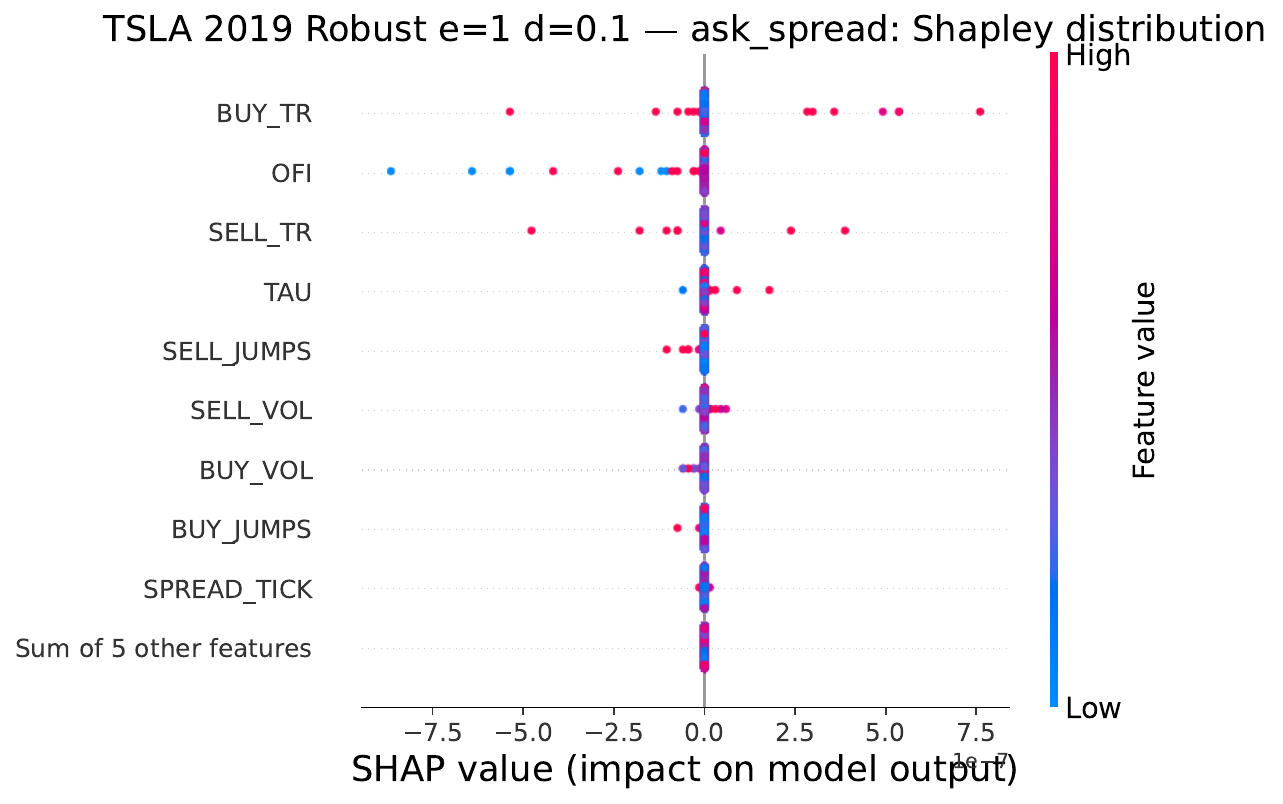}
            \caption{Robust - Ask spread.}
        \end{subfigure}\hfill
        \begin{subfigure}{0.24\linewidth}
            \includegraphics[width=\linewidth]{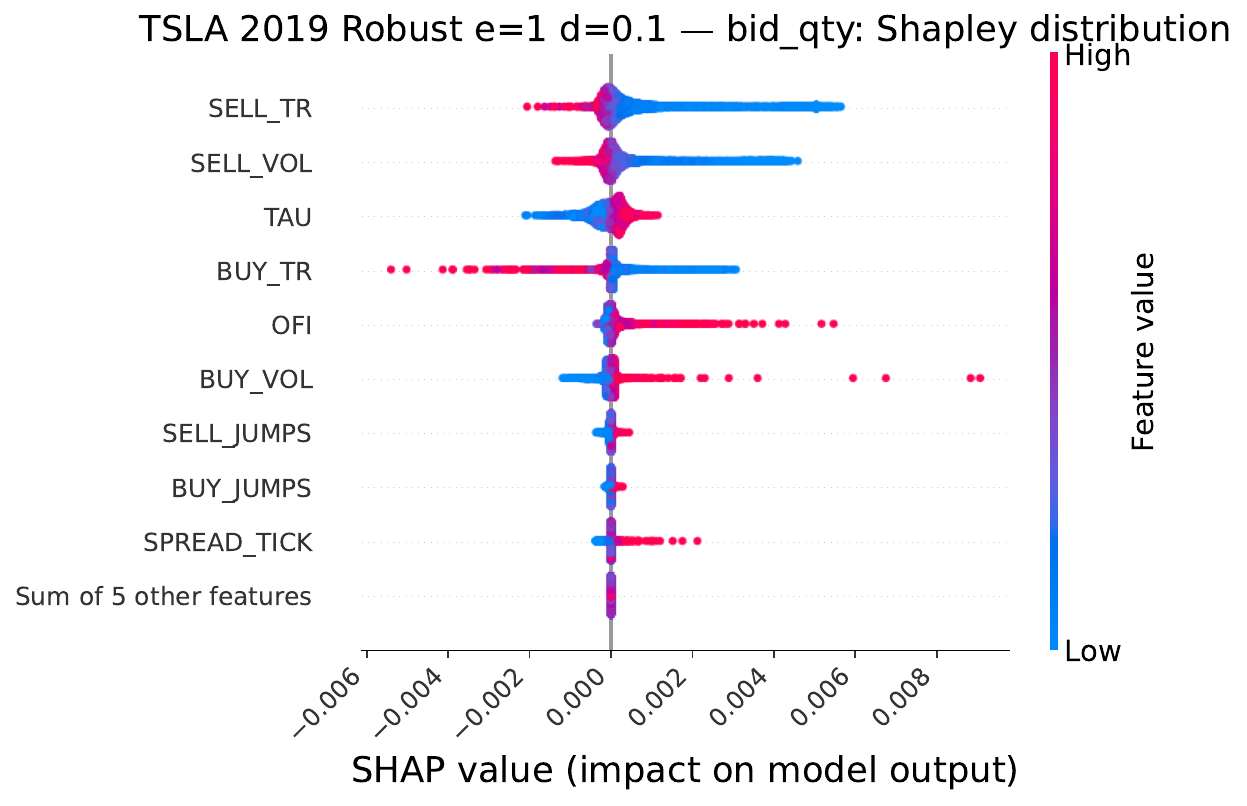}
            \caption{Robust - Bid quantity.}
        \end{subfigure}\hfill
        \begin{subfigure}{0.24\linewidth}
            \includegraphics[width=\linewidth]{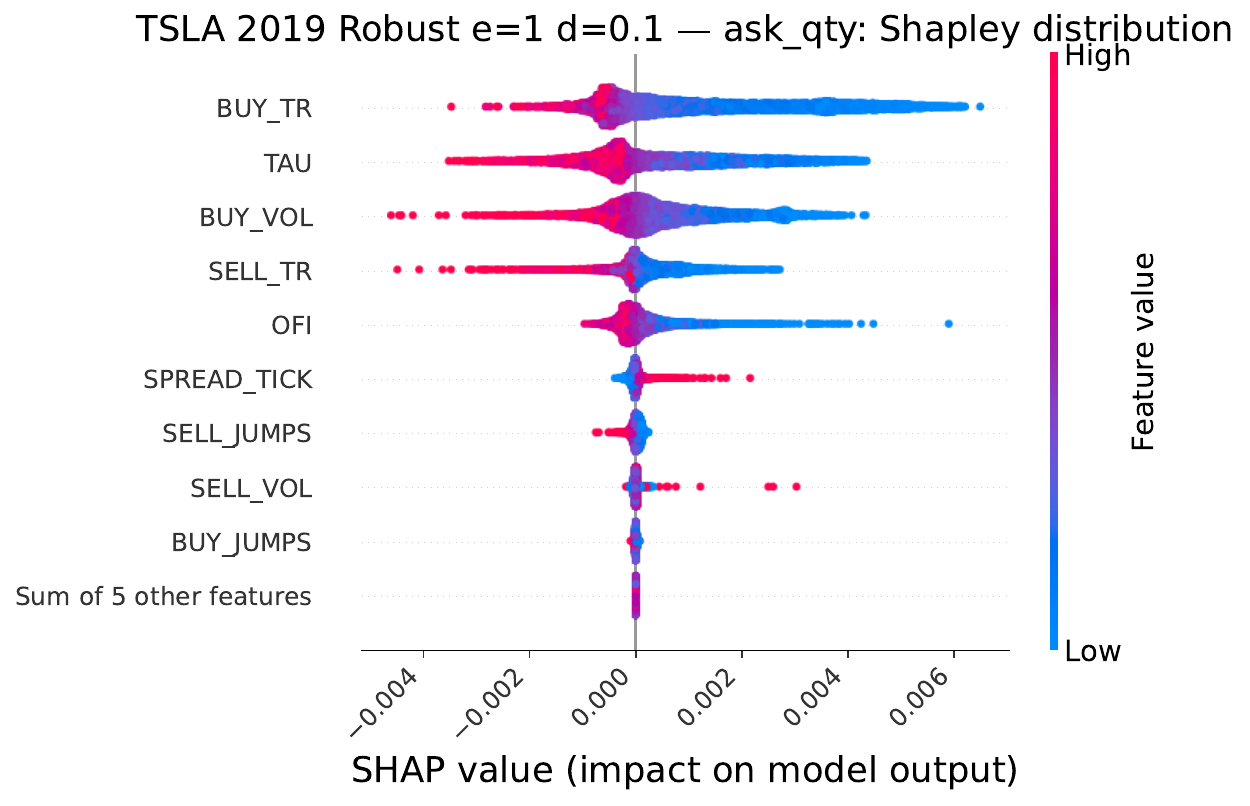}
            \caption{Robust - Ask quantity.}
        \end{subfigure}
    \end{minipage}}
    \caption{Distribution of absolute Shapley values for greedy and robust policies with $\bar{\varepsilon} = 1$ and $\delta = 0.1$ on TSLA, 2019, where the robust policy is the best test Sharpe Pareto configuration. As in AAPL, the dominant drivers are trade-flow and timing variables, but the spread components play a somewhat more visible role in TSLA's more volatile environment.}
    \label{fig:shapley_tsla_2019}
\end{figure}

\begin{figure}[H]
    \centering
    \makebox[\textwidth][c]{\begin{minipage}{1.16\textwidth}\centering
        \begin{subfigure}{0.24\linewidth}
            \includegraphics[width=\linewidth]{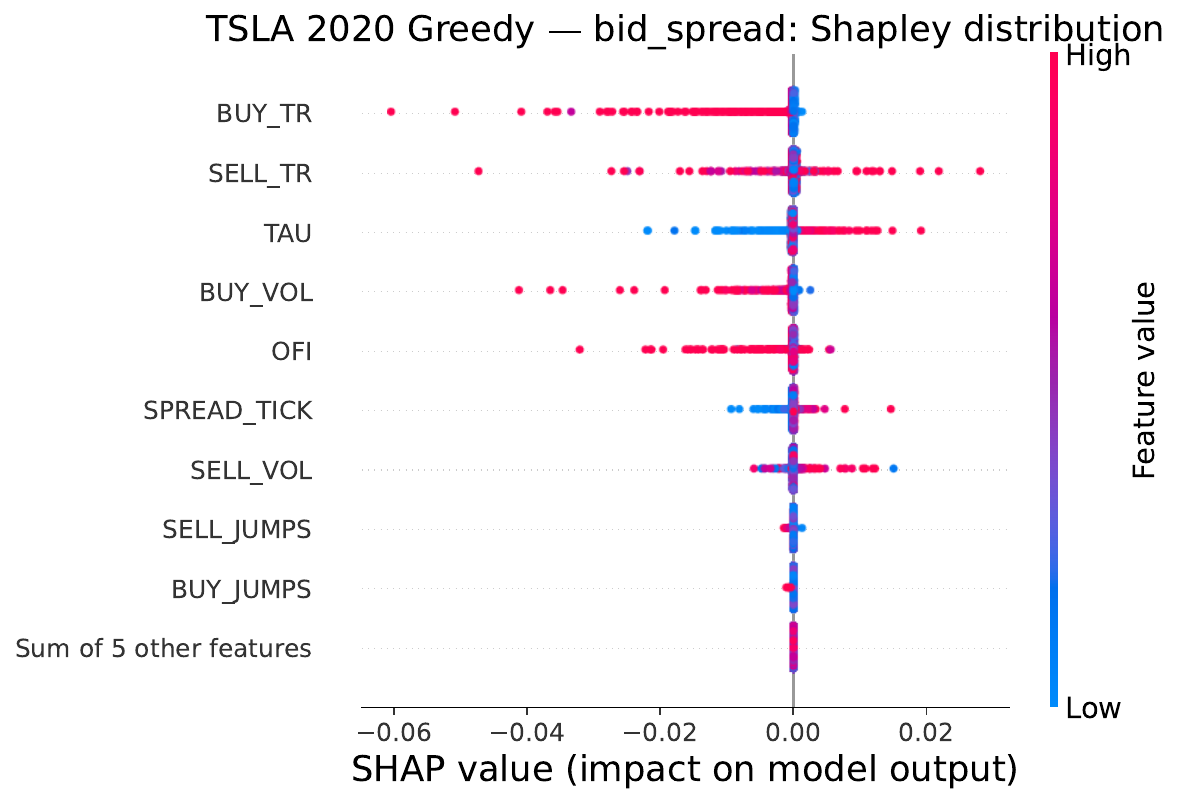}
            \caption{Greedy - Bid spread.}
        \end{subfigure}\hfill
        \begin{subfigure}{0.24\linewidth}
            \includegraphics[width=\linewidth]{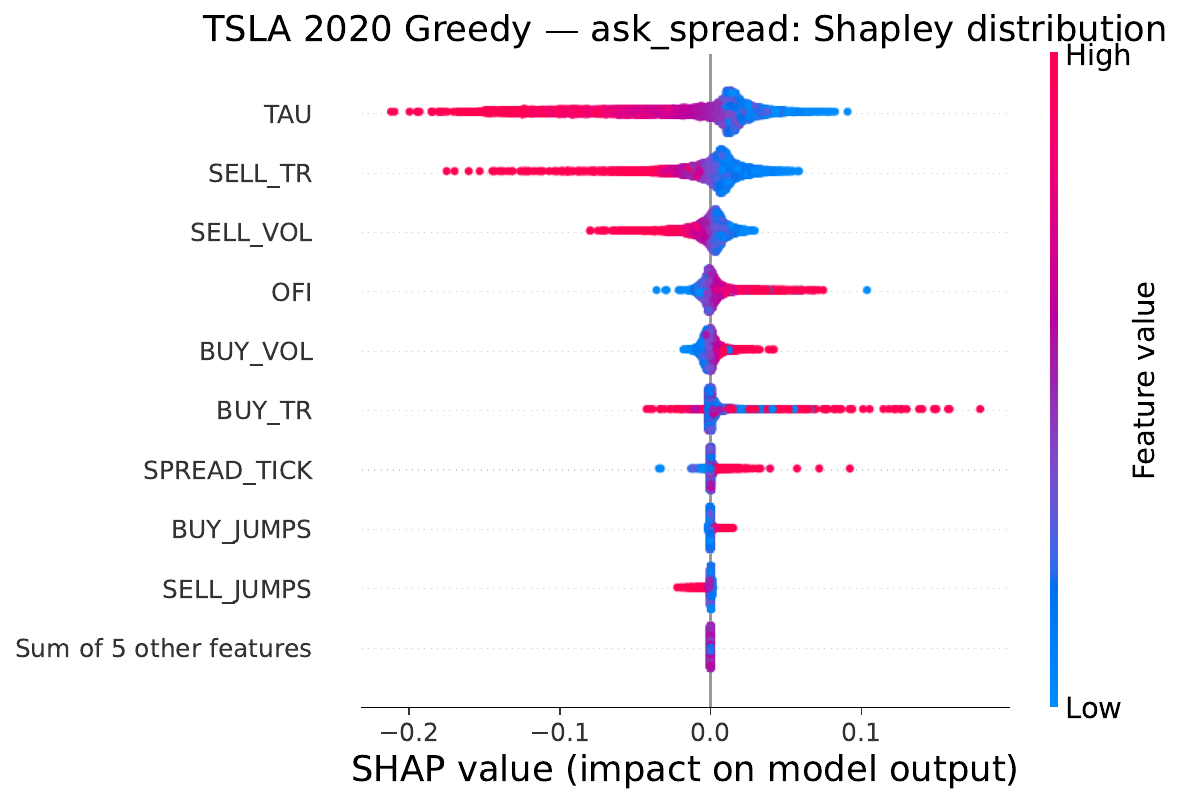}
            \caption{Greedy - Ask spread.}
        \end{subfigure}\hfill
        \begin{subfigure}{0.24\linewidth}
            \includegraphics[width=\linewidth]{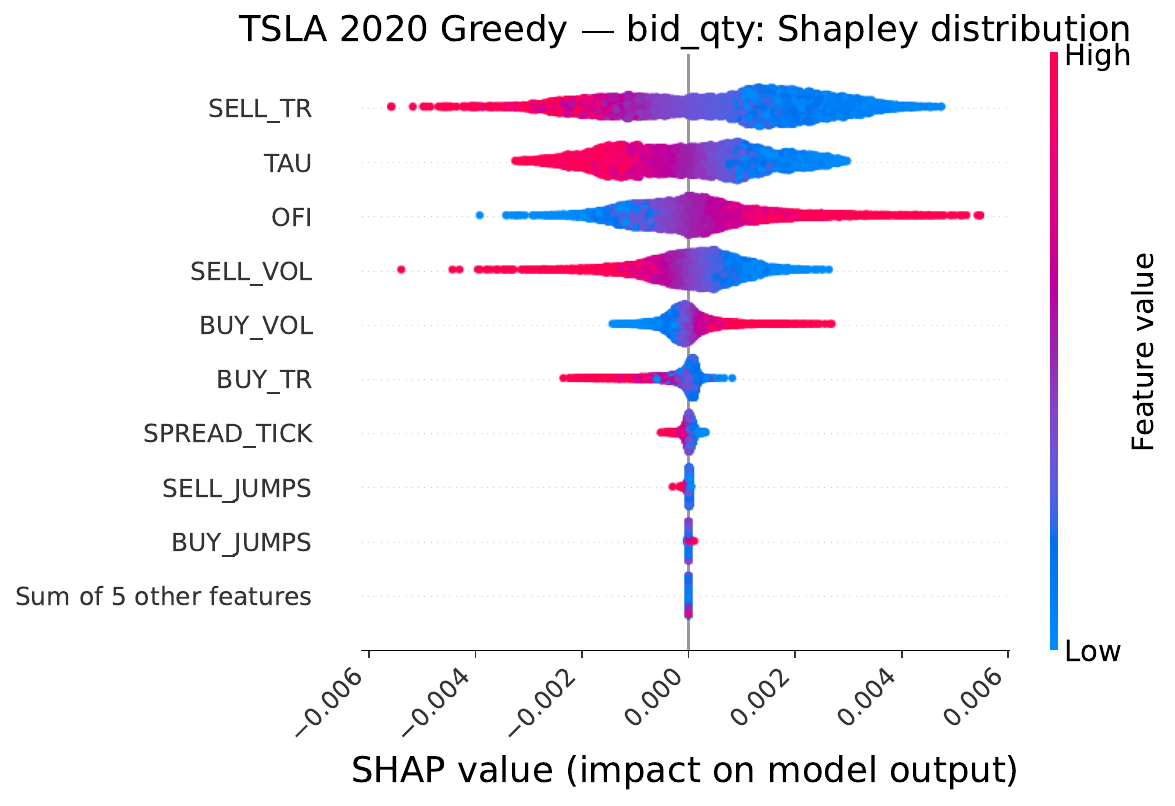}
            \caption{Greedy - Bid quantity.}
        \end{subfigure}\hfill
        \begin{subfigure}{0.24\linewidth}
            \includegraphics[width=\linewidth]{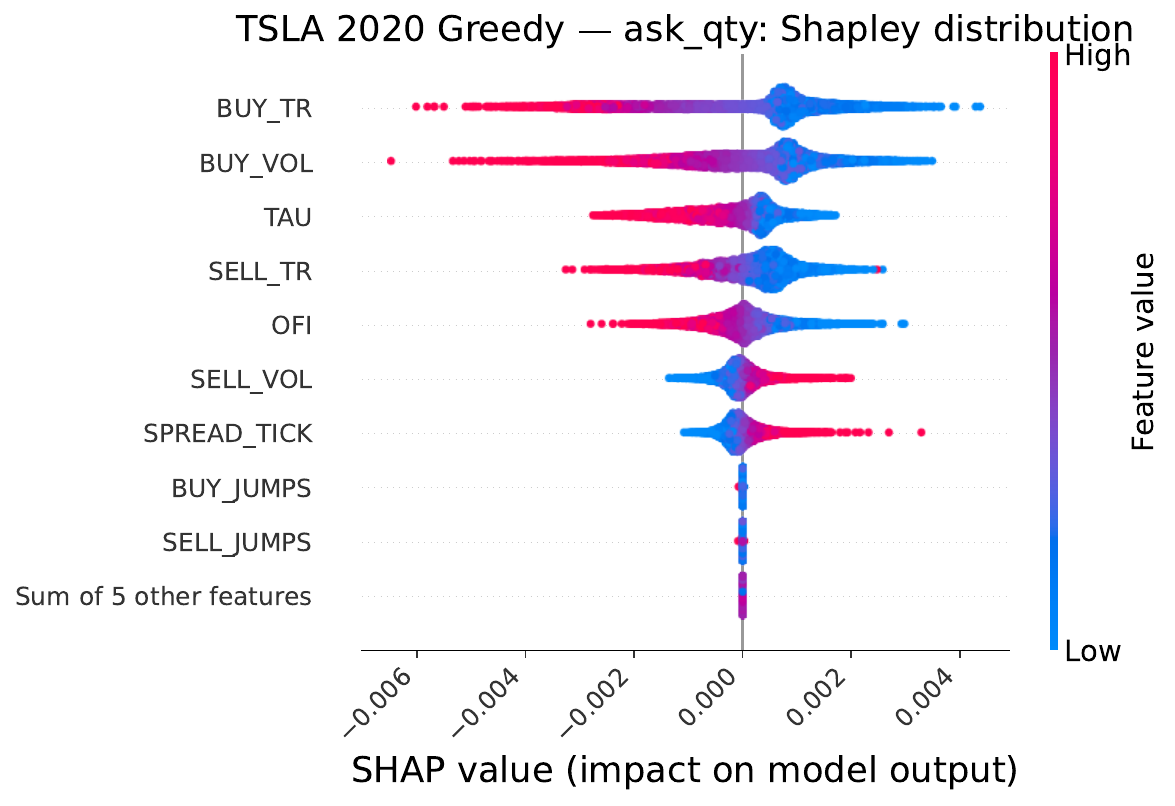}
            \caption{Greedy - Ask quantity.}
        \end{subfigure}
\vspace{0.6em}
        \begin{subfigure}{0.24\linewidth}
            \includegraphics[width=\linewidth]{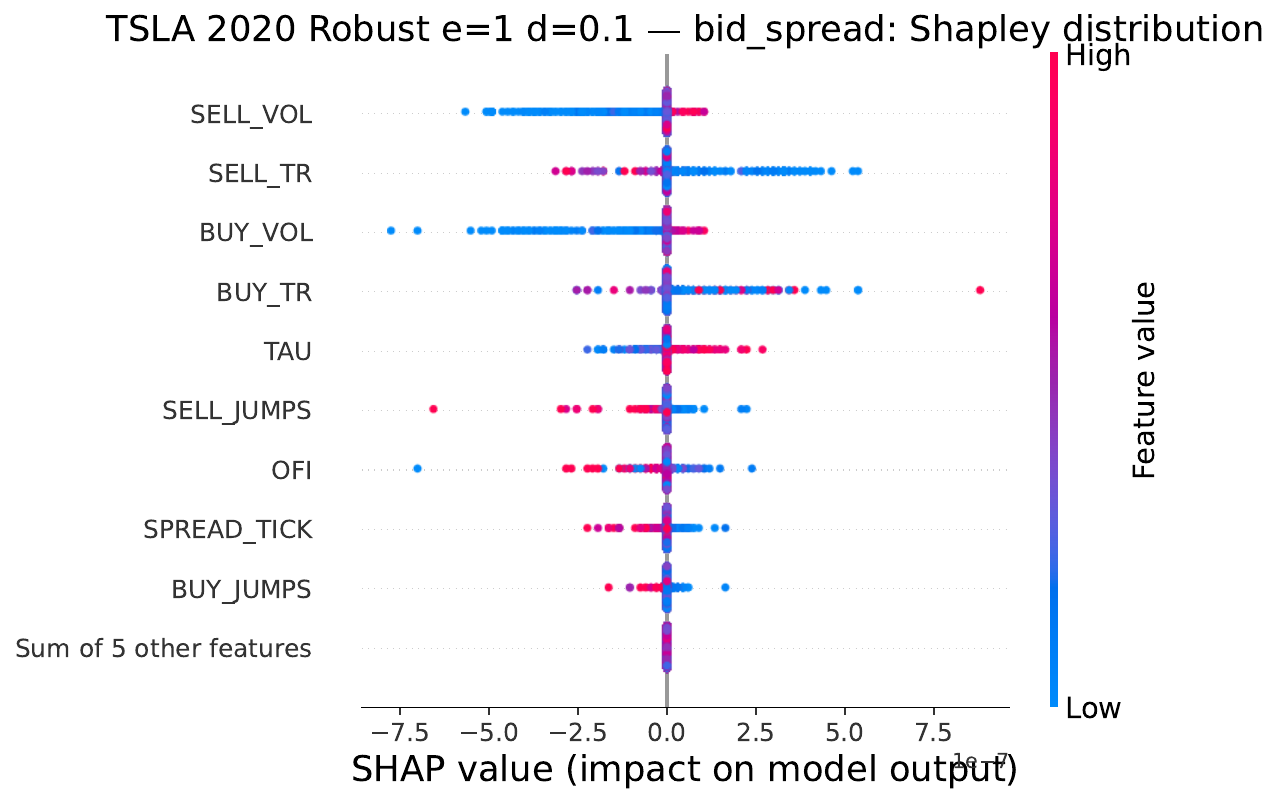}
            \caption{Robust - Bid spread.}
        \end{subfigure}\hfill
        \begin{subfigure}{0.24\linewidth}
            \includegraphics[width=\linewidth]{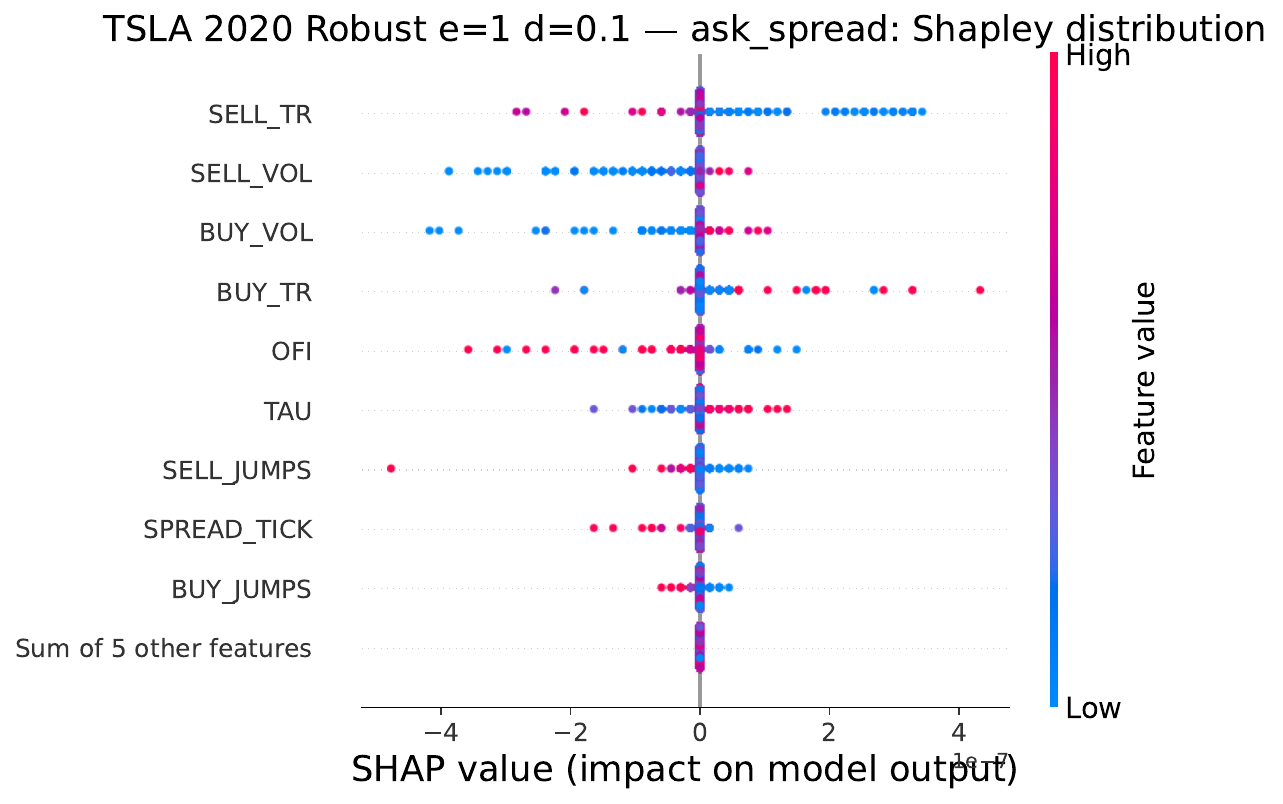}
            \caption{Robust - Ask spread.}
        \end{subfigure}\hfill
        \begin{subfigure}{0.24\linewidth}
            \includegraphics[width=\linewidth]{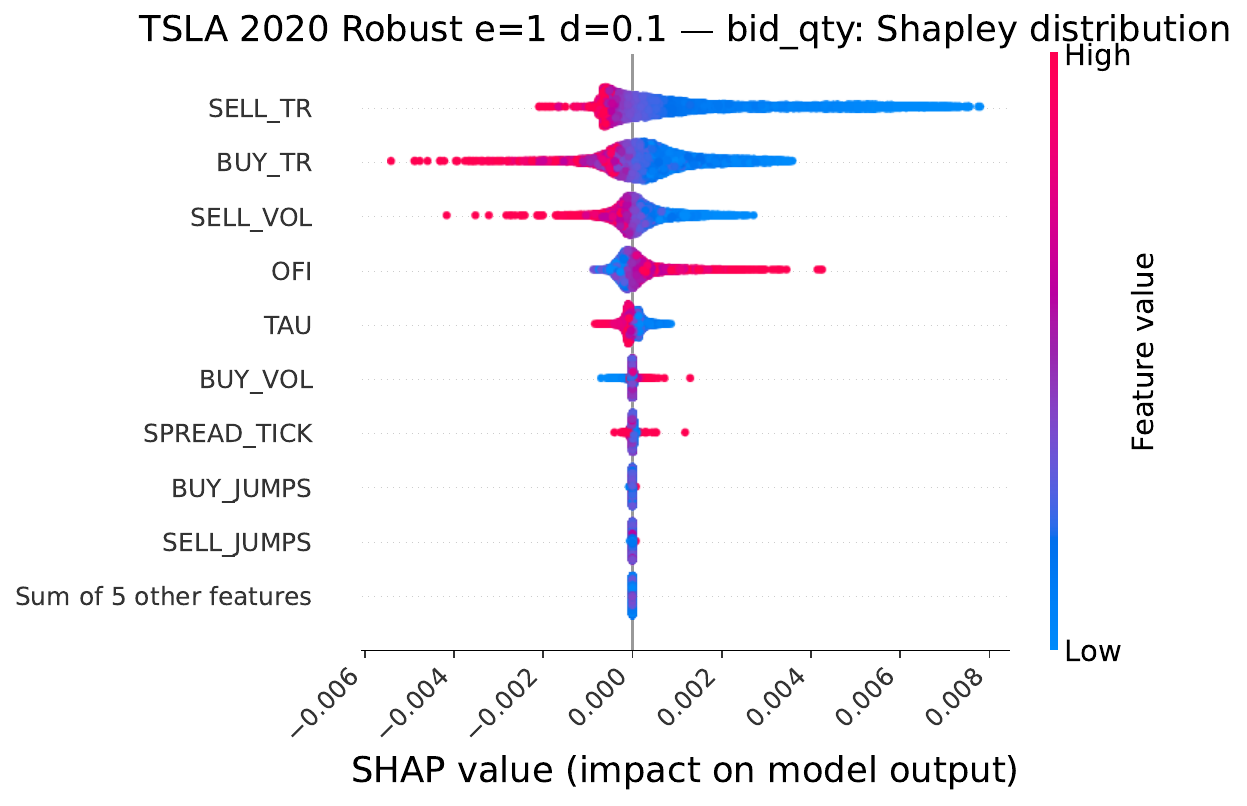}
            \caption{Robust - Bid quantity.}
        \end{subfigure}\hfill
        \begin{subfigure}{0.24\linewidth}
            \includegraphics[width=\linewidth]{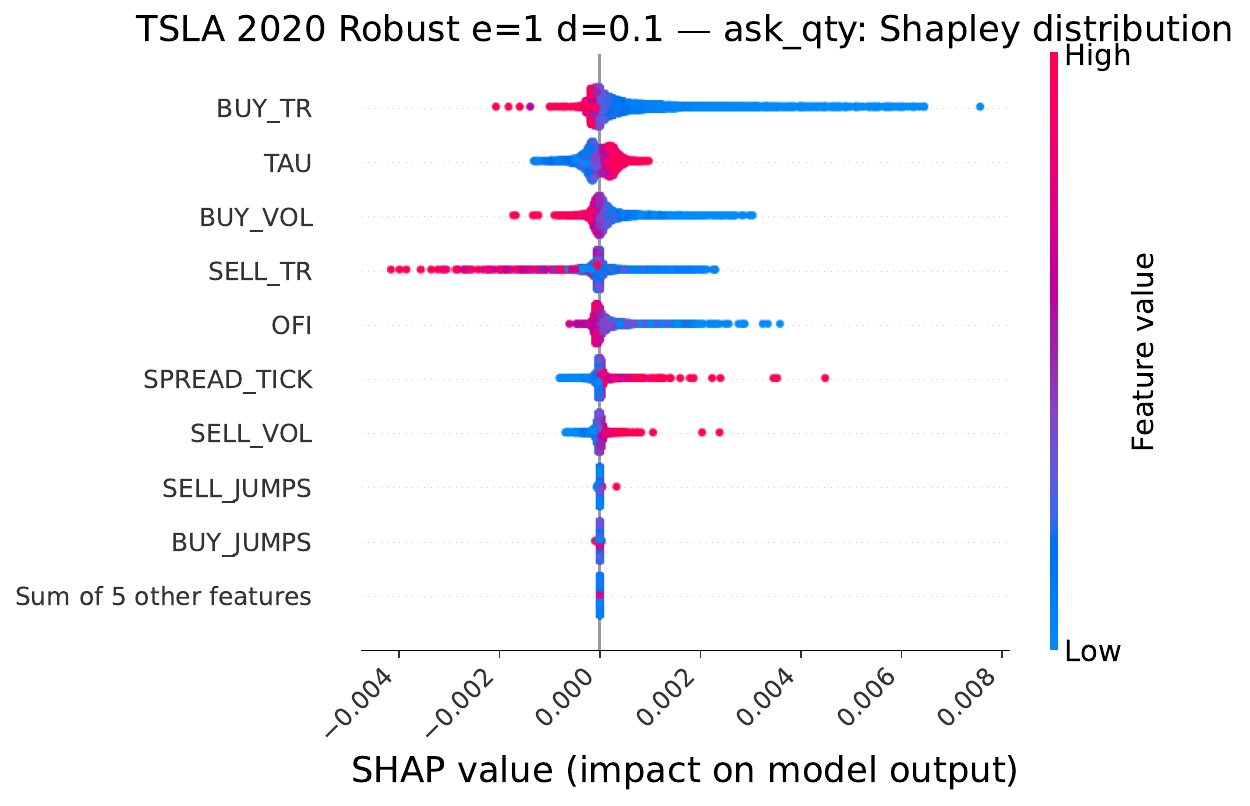}
            \caption{Robust - Ask quantity.}
        \end{subfigure}
    \end{minipage}}
    \caption{Distribution of absolute Shapley values for greedy and robust policies with $\bar{\varepsilon} = 1$ and $\delta = 0.1$ on TSLA, 2020, where the robust policy is the best test Sharpe Pareto configuration. The figure again points to trade-flow and timing variables as the main drivers, with robustness leaving the overall importance ranking broadly intact despite the more stressed market conditions.}
    \label{fig:shapley_tsla_1920}
\end{figure}

\begin{figure}[H]
    \centering
    \makebox[\textwidth][c]{\begin{minipage}{1.16\textwidth}\centering
        \begin{subfigure}{0.24\linewidth}
            \includegraphics[width=\linewidth]{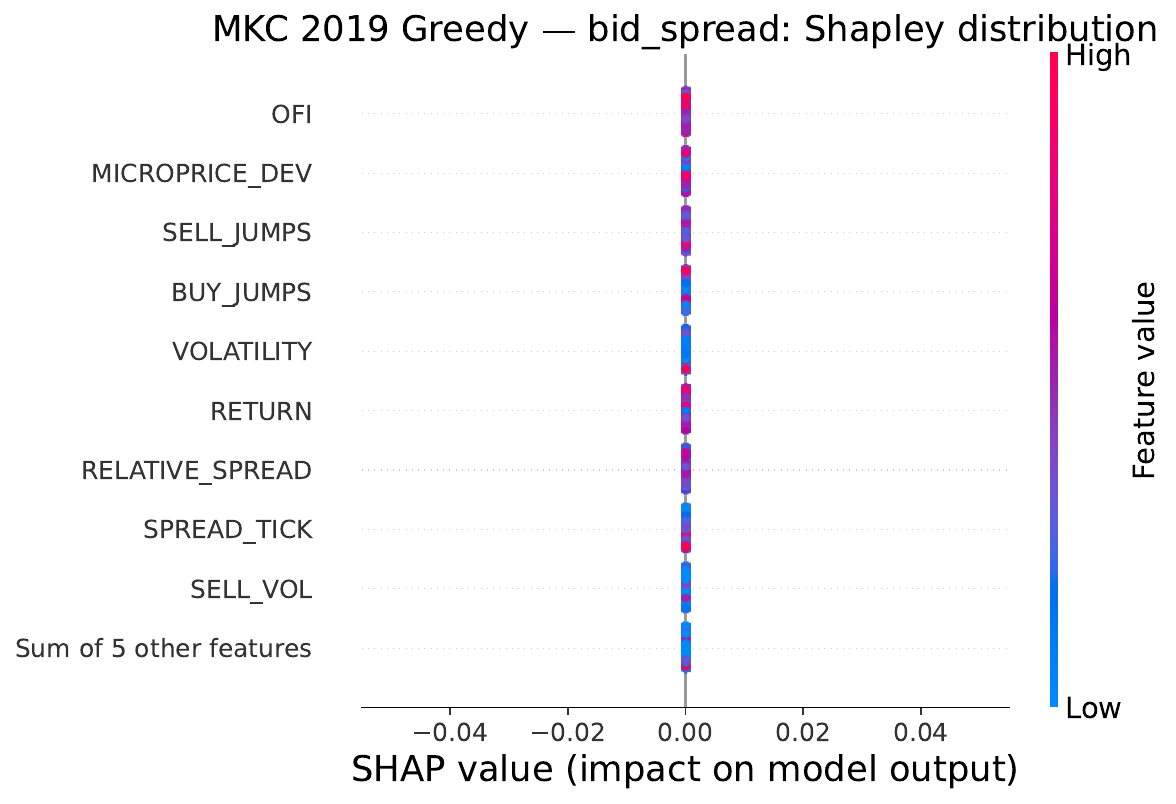}
            \caption{Greedy - Bid spread.}
        \end{subfigure}\hfill
        \begin{subfigure}{0.24\linewidth}
            \includegraphics[width=\linewidth]{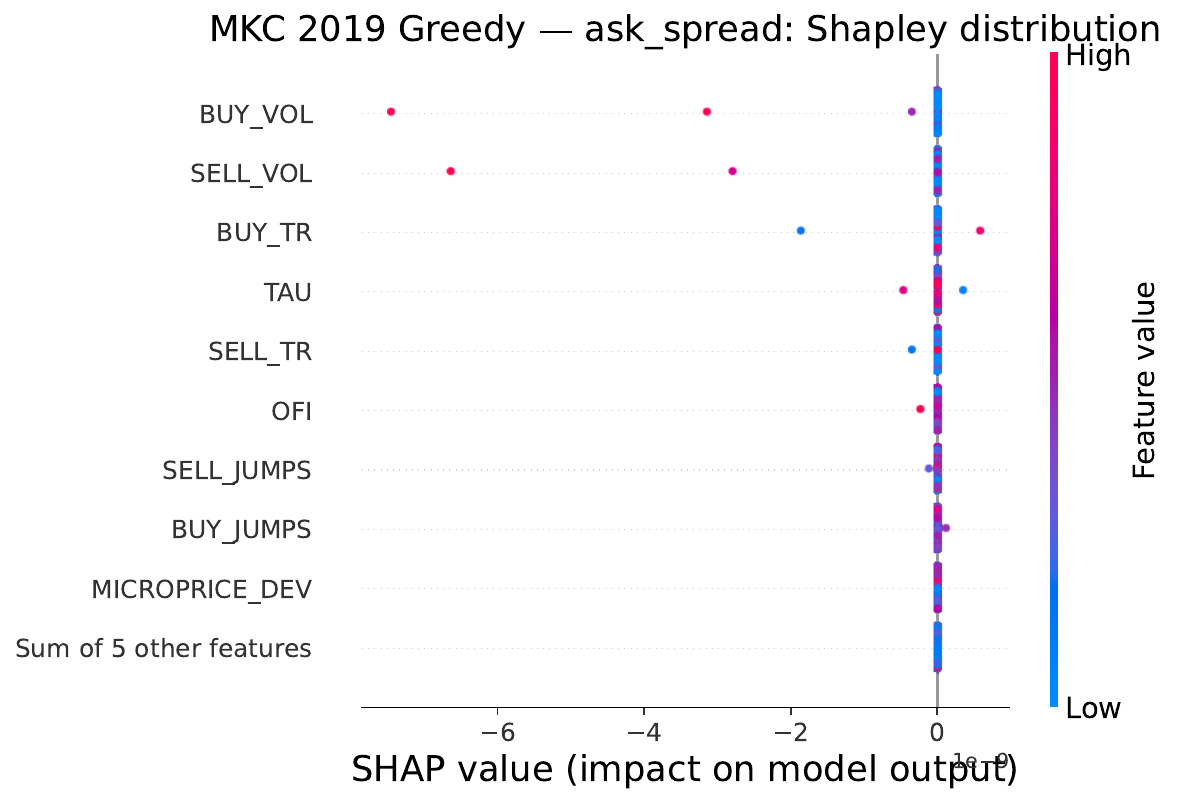}
            \caption{Greedy - Ask spread.}
        \end{subfigure}\hfill
        \begin{subfigure}{0.24\linewidth}
            \includegraphics[width=\linewidth]{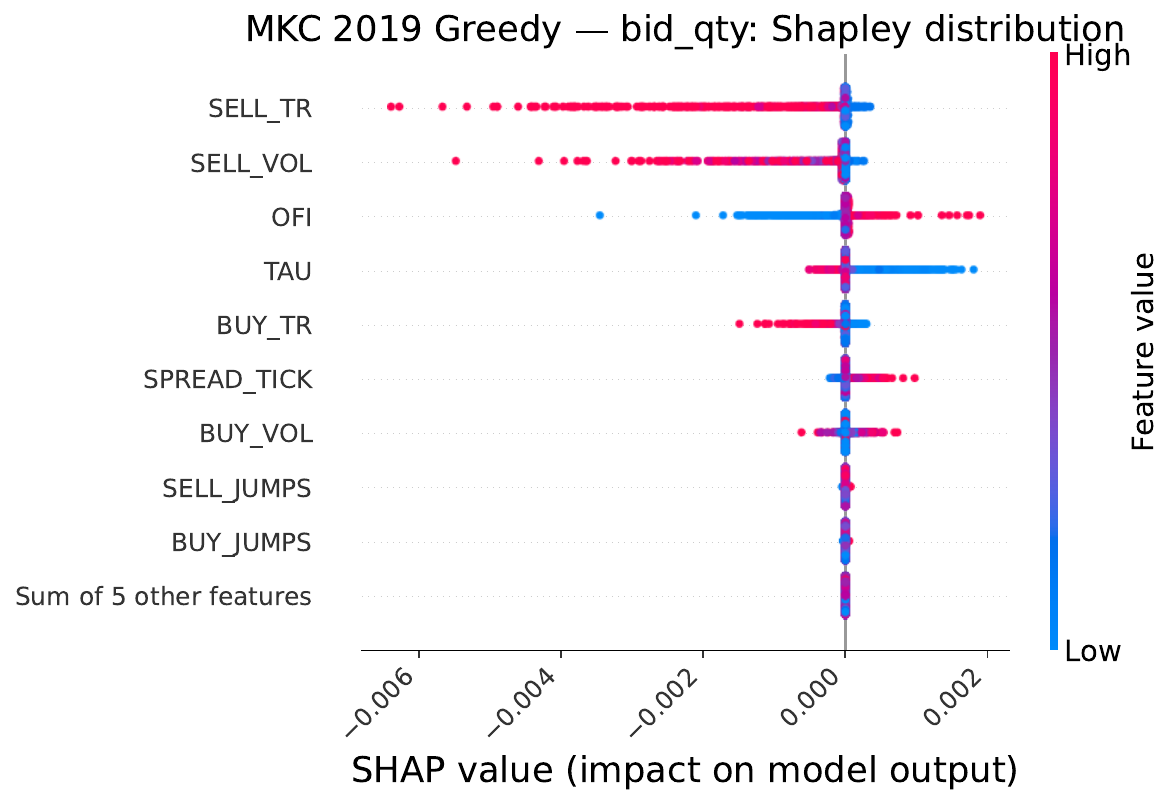}
            \caption{Greedy - Bid quantity.}
        \end{subfigure}\hfill
        \begin{subfigure}{0.24\linewidth}
            \includegraphics[width=\linewidth]{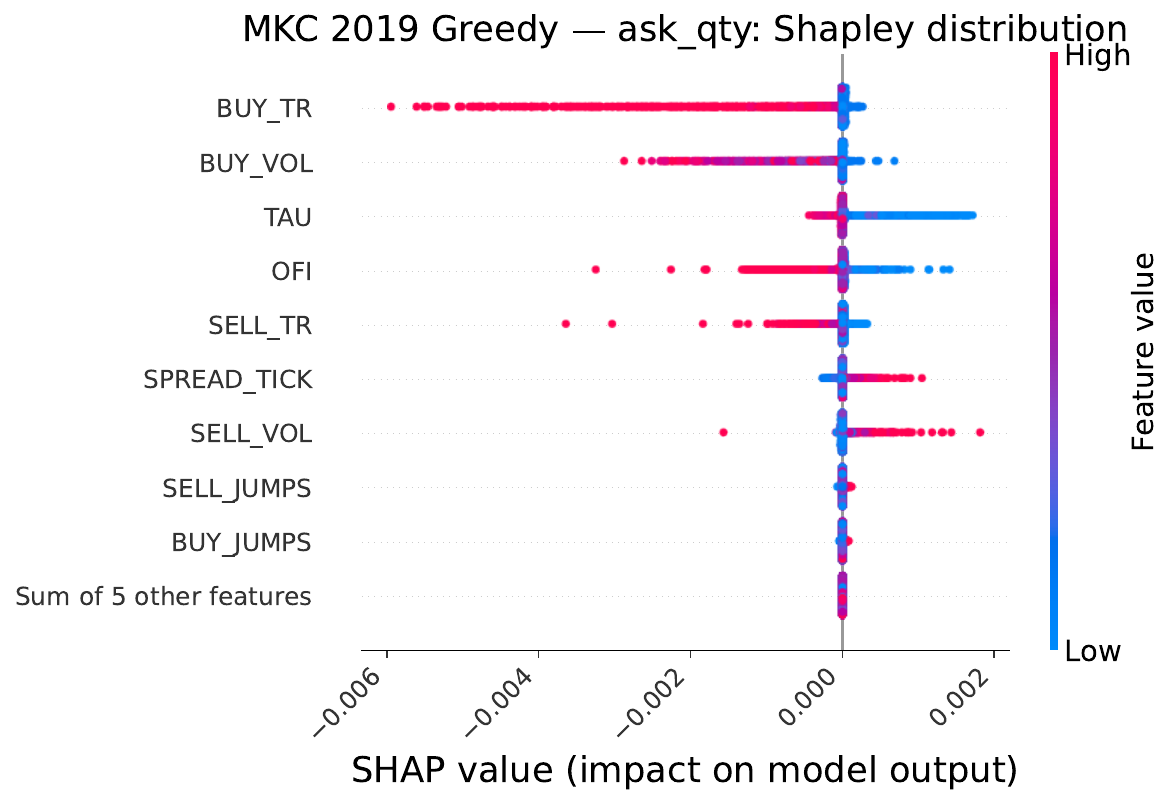}
            \caption{Greedy - Ask quantity.}
        \end{subfigure}
\vspace{0.6em}
        \begin{subfigure}{0.24\linewidth}
            \includegraphics[width=\linewidth]{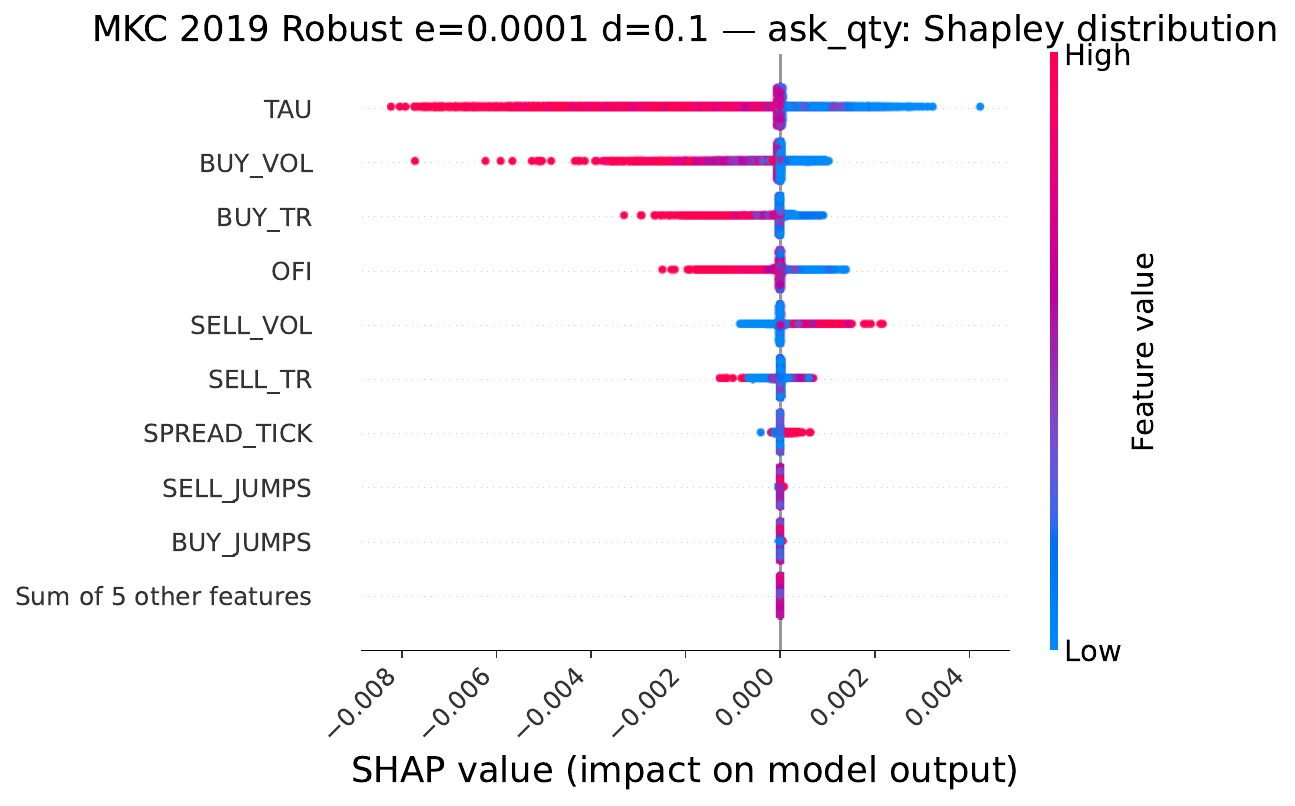}
            \caption{Robust - Bid spread.}
        \end{subfigure}\hfill
        \begin{subfigure}{0.24\linewidth}
            \includegraphics[width=\linewidth]{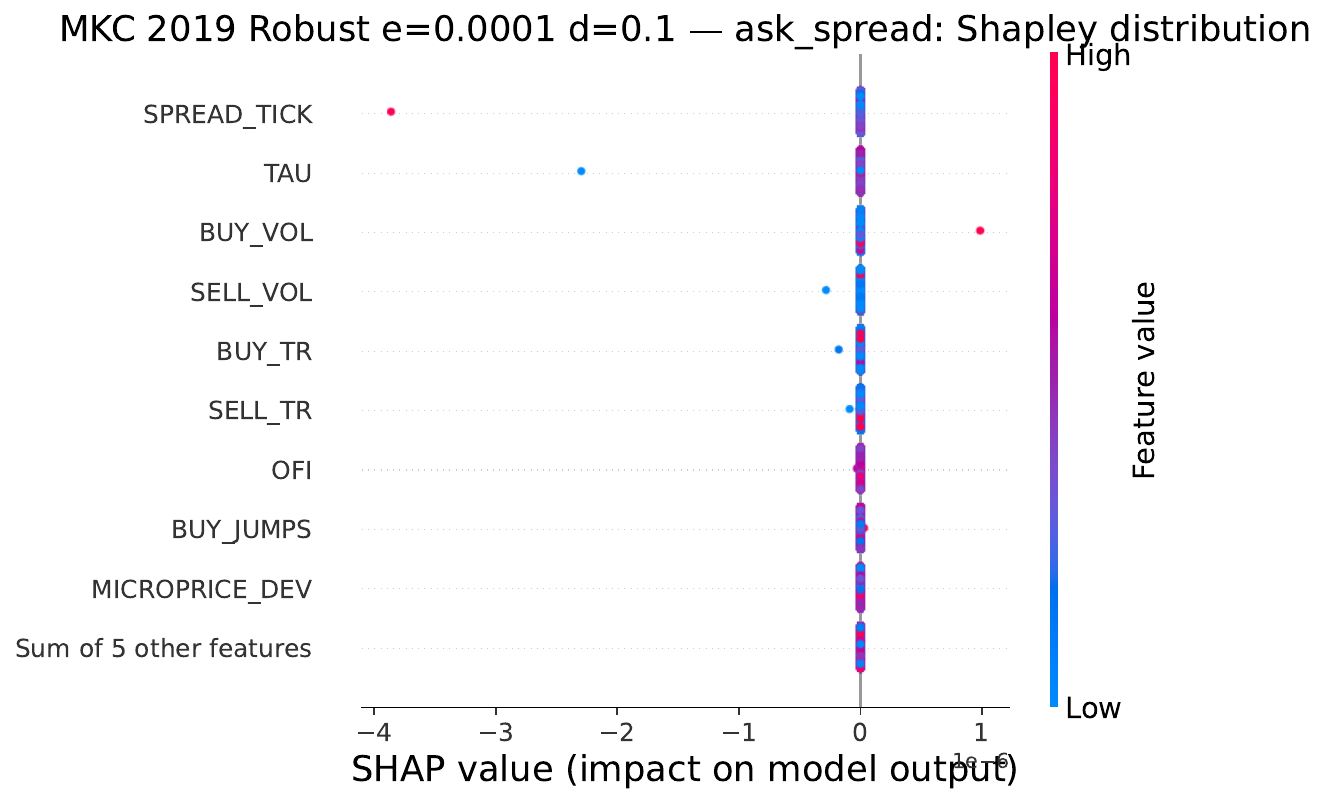}
            \caption{Robust - Ask spread.}
        \end{subfigure}\hfill
        \begin{subfigure}{0.24\linewidth}
            \includegraphics[width=\linewidth]{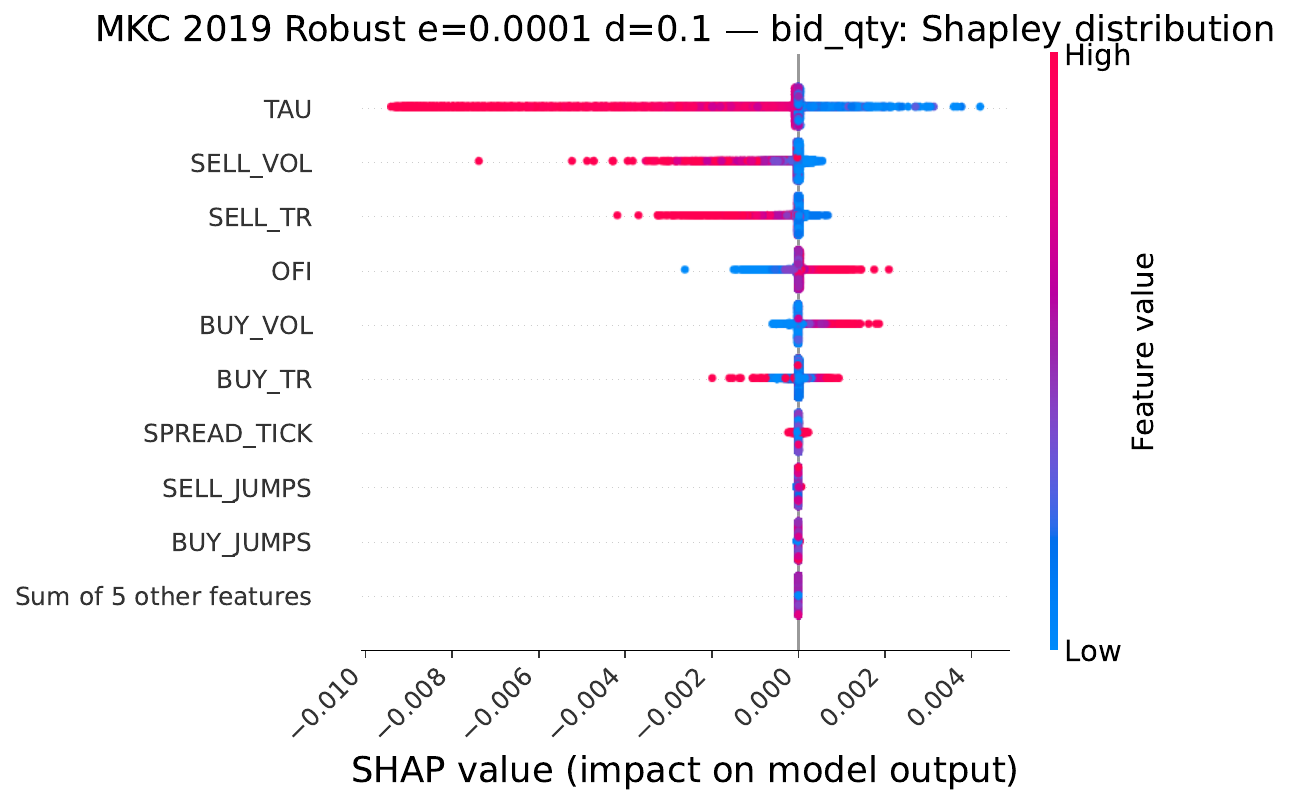}
            \caption{Robust - Bid quantity.}
        \end{subfigure}\hfill
        \begin{subfigure}{0.24\linewidth}
            \includegraphics[width=\linewidth]{image/MKC_2019/MKC_2019_Robust_e0.0001_d0.1_ask_qty_shapley_distribution.pdf}
            \caption{Robust - Ask quantity.}
        \end{subfigure}
    \end{minipage}}
    \caption{Distribution of absolute Shapley values for greedy and robust policies with $\bar{\varepsilon} = 0.0001$ and $\delta = 0.1$ on MKC, 2019, where the robust policy is the best test Sharpe Pareto configuration. The importance profiles are comparatively stable across the greedy and robust policies, consistent with the weaker behavioural impact of robustness in MKC.}
    \label{fig:shapley_mkc_2019}
\end{figure}

\begin{figure}[H]
    \centering
    \makebox[\textwidth][c]{\begin{minipage}{1.16\textwidth}\centering
        \begin{subfigure}{0.24\linewidth}
            \includegraphics[width=\linewidth]{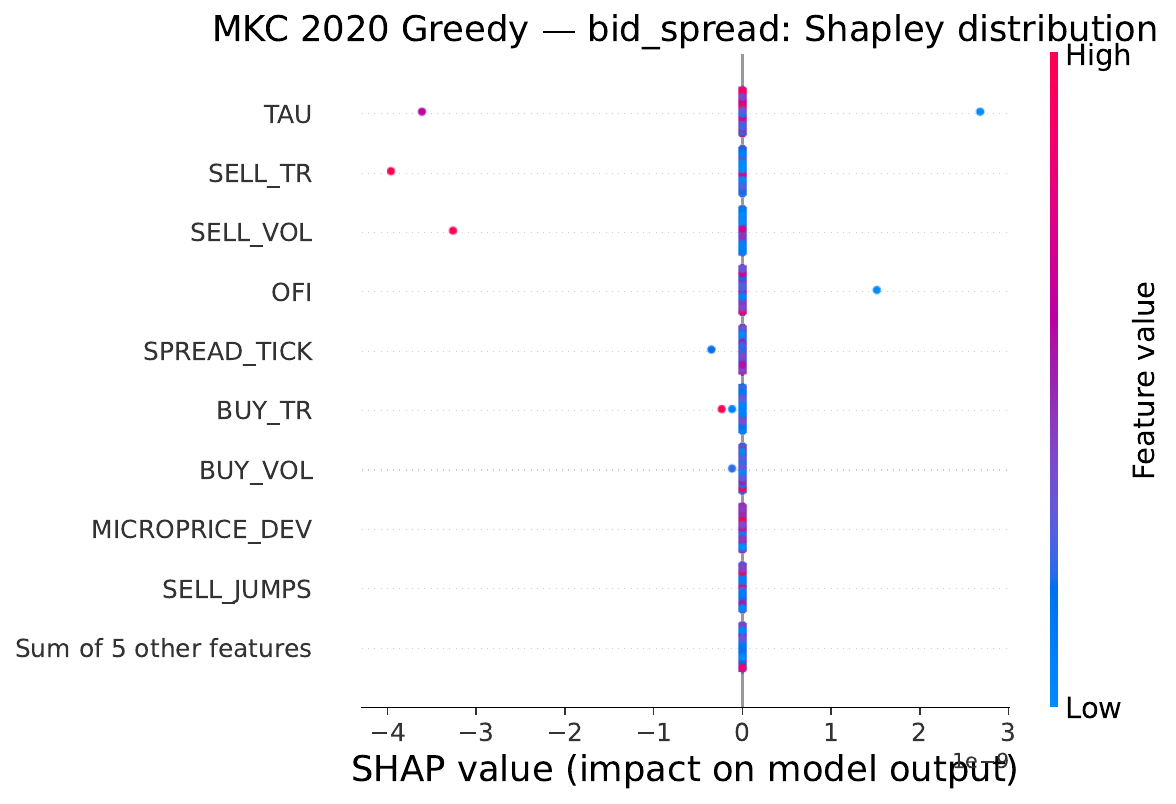}
            \caption{Greedy - Bid spread.}
        \end{subfigure}\hfill
        \begin{subfigure}{0.24\linewidth}
            \includegraphics[width=\linewidth]{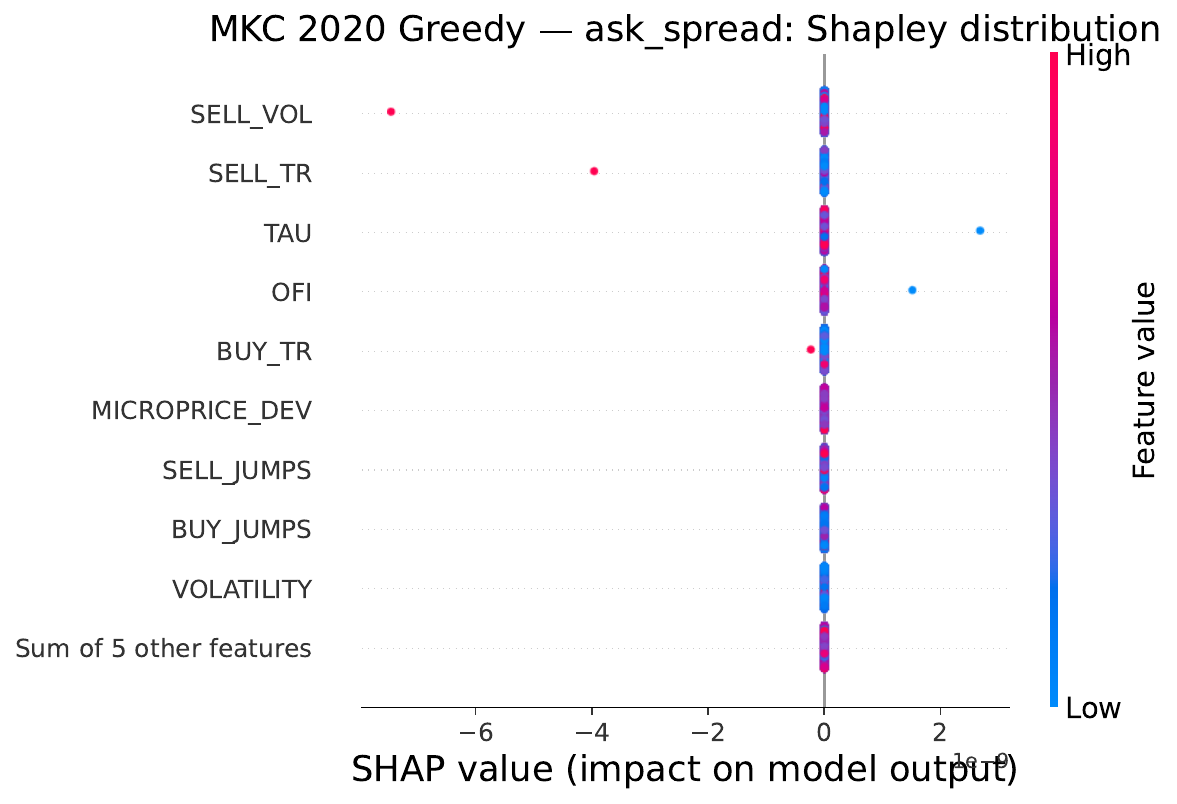}
            \caption{Greedy - Ask spread.}
        \end{subfigure}\hfill
        \begin{subfigure}{0.24\linewidth}
            \includegraphics[width=\linewidth]{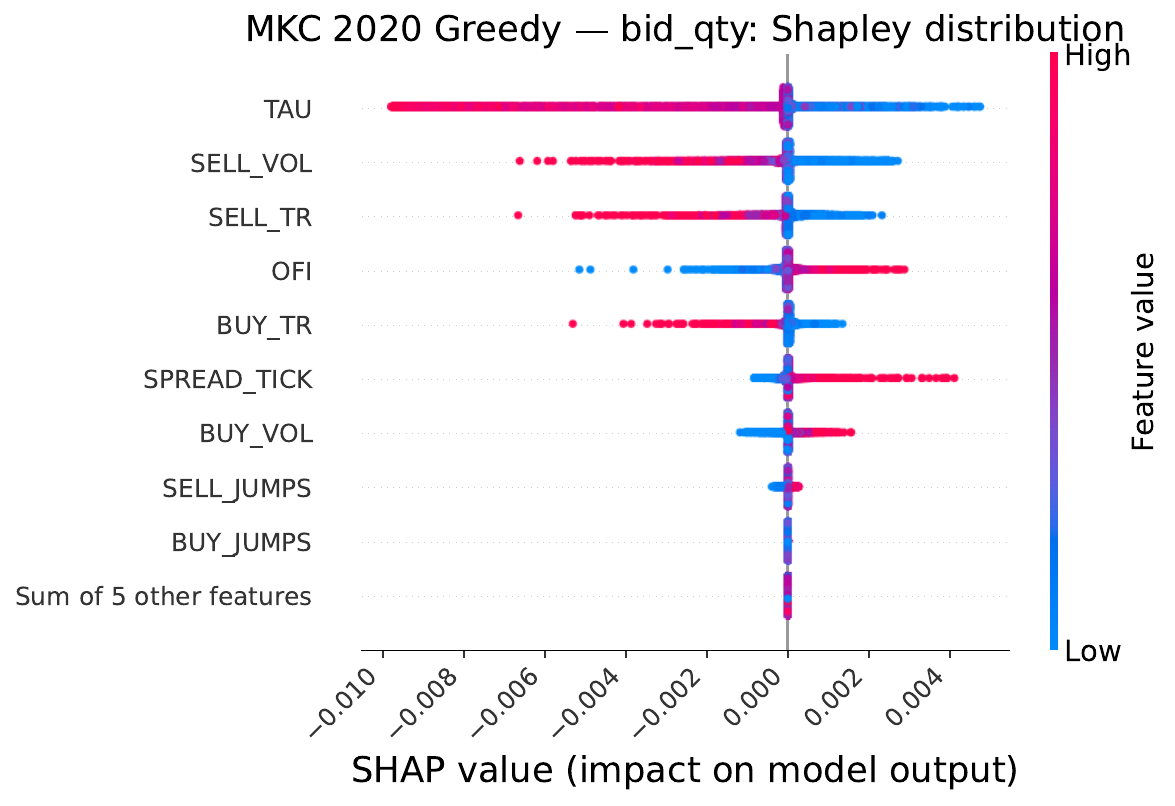}
            \caption{Greedy - Bid quantity.}
        \end{subfigure}\hfill
        \begin{subfigure}{0.24\linewidth}
            \includegraphics[width=\linewidth]{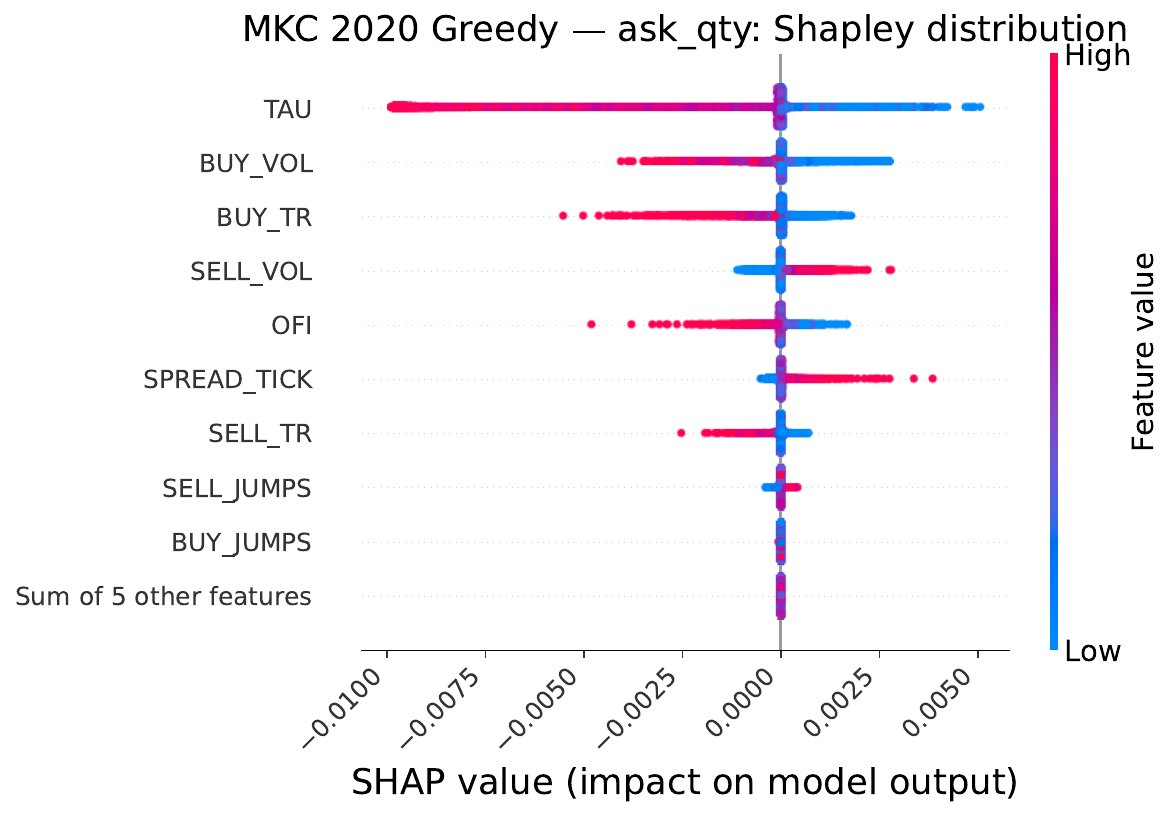}
            \caption{Greedy - Ask quantity.}
        \end{subfigure}
\vspace{0.6em}
        \begin{subfigure}{0.24\linewidth}
            \includegraphics[width=\linewidth]{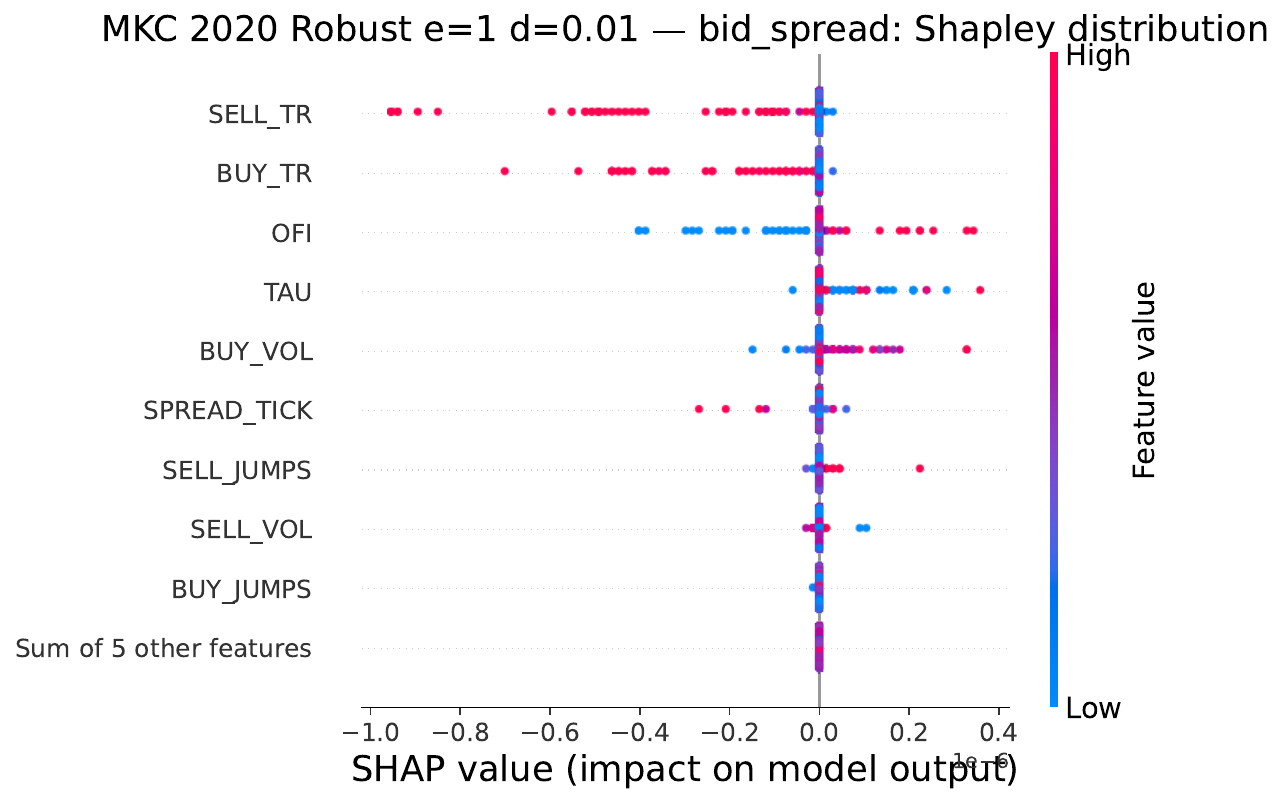}
            \caption{Robust - Bid spread.}
        \end{subfigure}\hfill
        \begin{subfigure}{0.24\linewidth}
            \includegraphics[width=\linewidth]{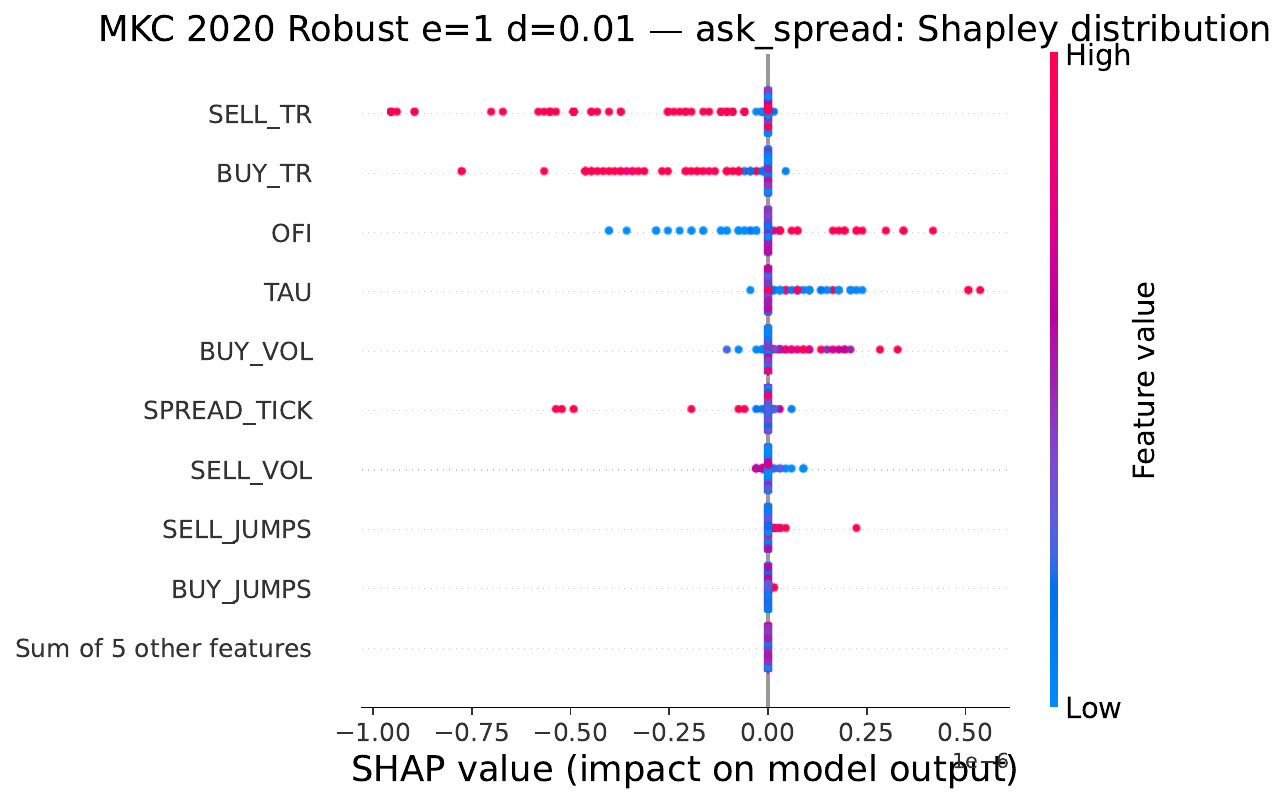}
            \caption{Robust - Ask spread.}
        \end{subfigure}\hfill
        \begin{subfigure}{0.24\linewidth}
            \includegraphics[width=\linewidth]{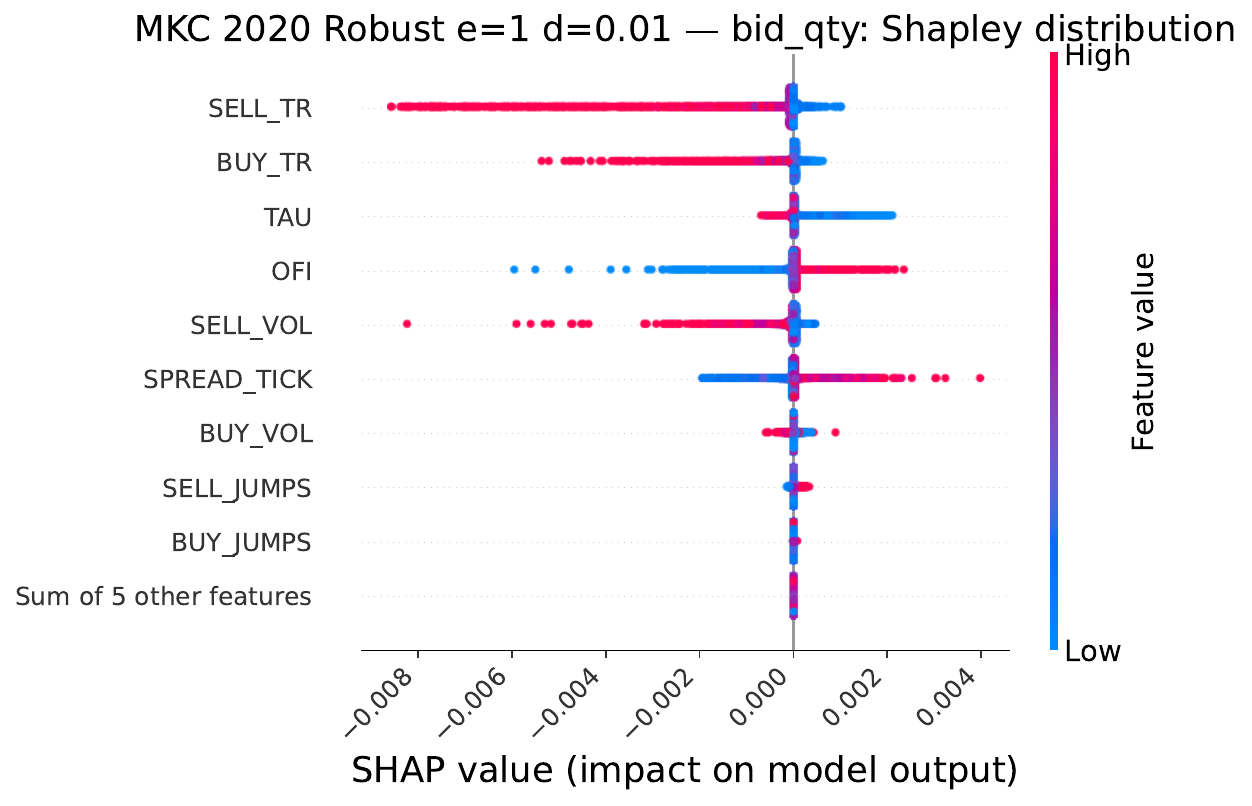}
            \caption{Robust - Bid quantity.}
        \end{subfigure}\hfill
        \begin{subfigure}{0.24\linewidth}
            \includegraphics[width=\linewidth]{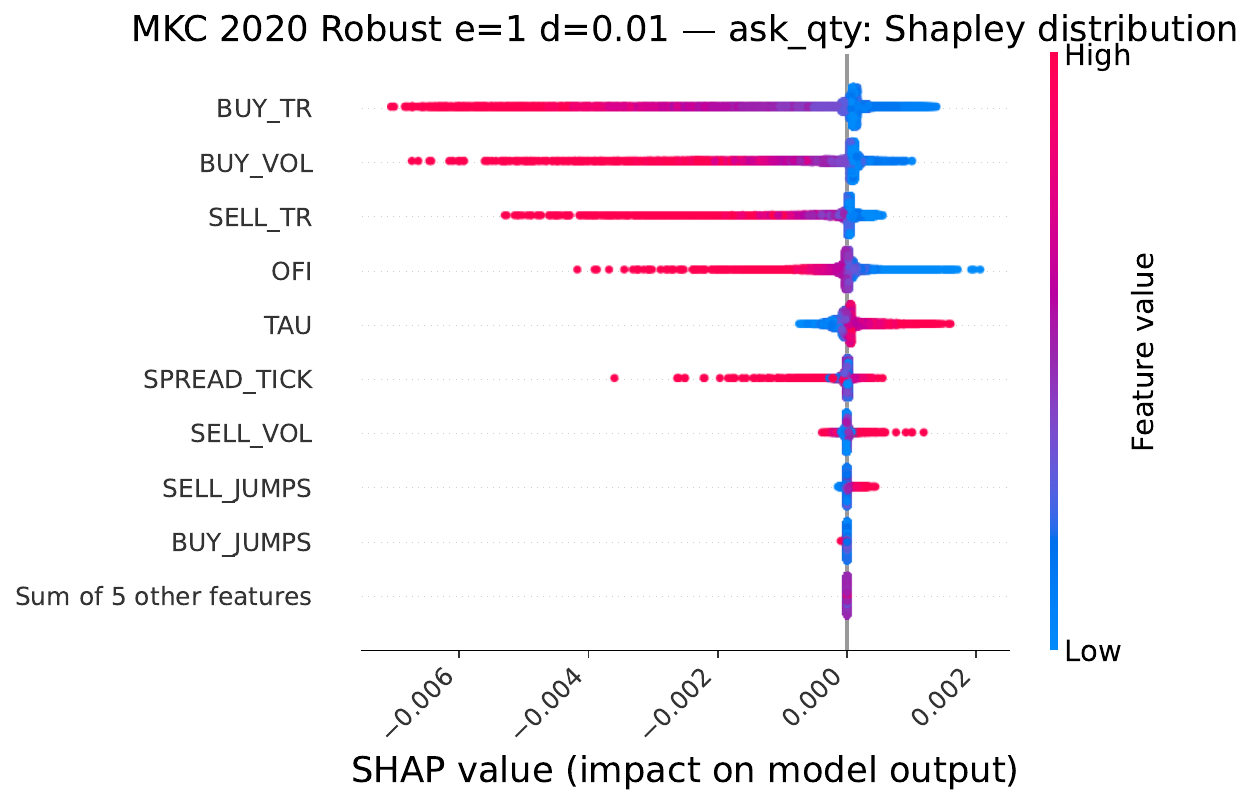}
            \caption{Robust - Ask quantity.}
        \end{subfigure}
    \end{minipage}}
    \caption{Distribution of absolute Shapley values for greedy and robust policies with $\bar{\varepsilon} = 1$ and $\delta = 0.01$ on MKC, 2020, where the robust policy is the best test Sharpe Pareto configuration. The same broad set of trade-flow and timing variables remains dominant, and the robust specification changes their relative weights quite significantly.}
    \label{fig:shapley_mkc_2020}
\end{figure}

\begin{figure}[H]
    \centering
    \makebox[\textwidth][c]{\begin{minipage}{1.16\textwidth}\centering
        \begin{subfigure}{0.24\linewidth}
            \includegraphics[width=\linewidth]{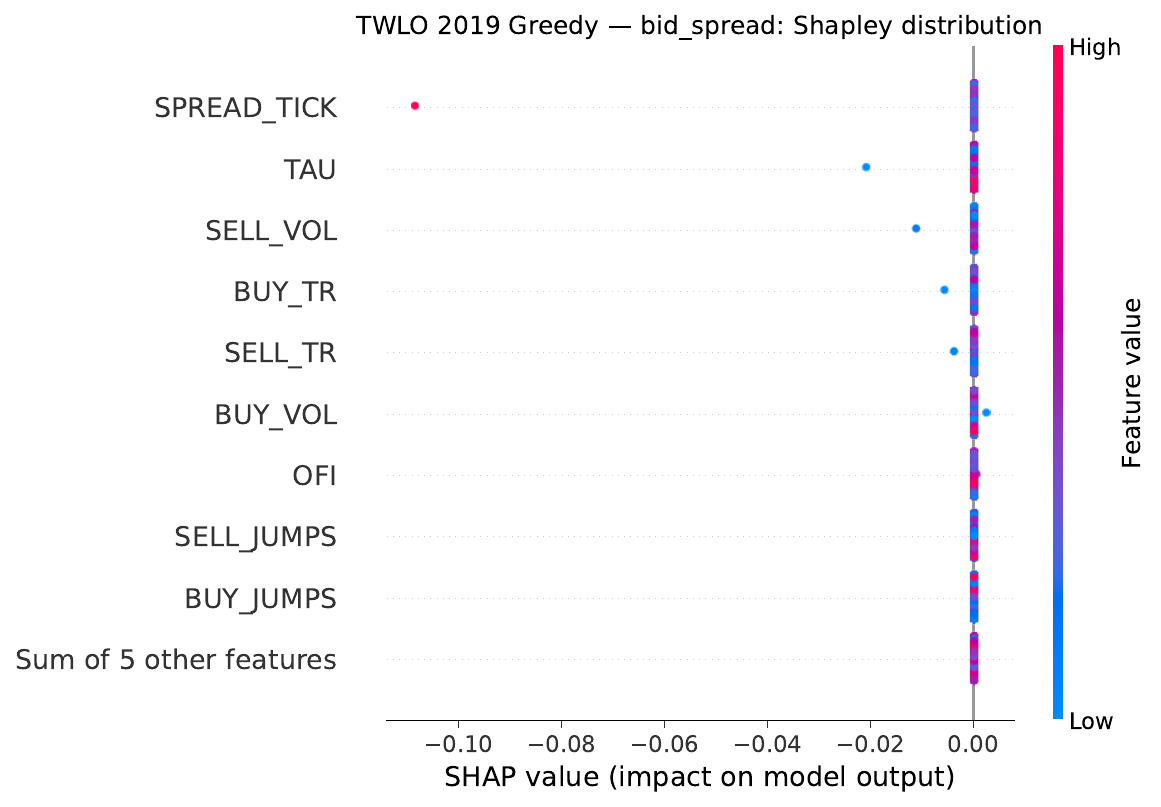}
            \caption{Greedy - Bid spread.}
        \end{subfigure}\hfill
        \begin{subfigure}{0.24\linewidth}
            \includegraphics[width=\linewidth]{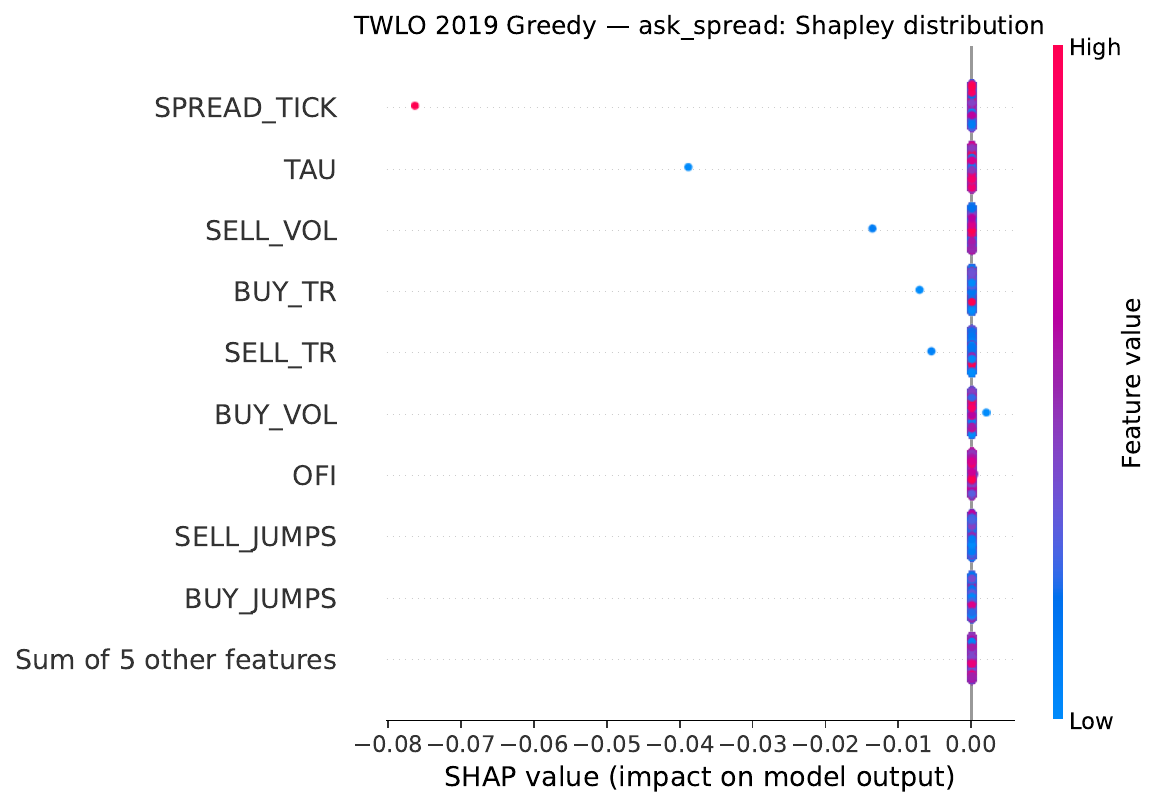}
            \caption{Greedy - Ask spread.}
        \end{subfigure}\hfill
        \begin{subfigure}{0.24\linewidth}
            \includegraphics[width=\linewidth]{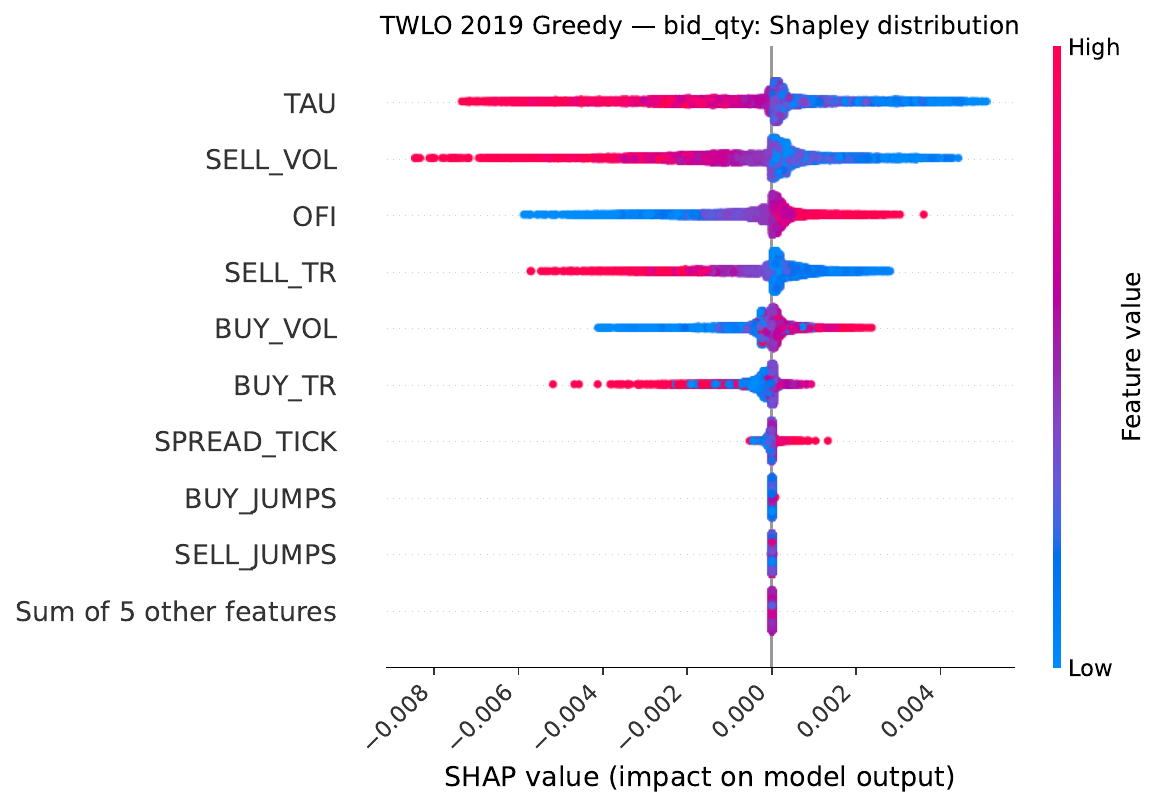}
            \caption{Greedy - Bid quantity.}
        \end{subfigure}\hfill
        \begin{subfigure}{0.24\linewidth}
            \includegraphics[width=\linewidth]{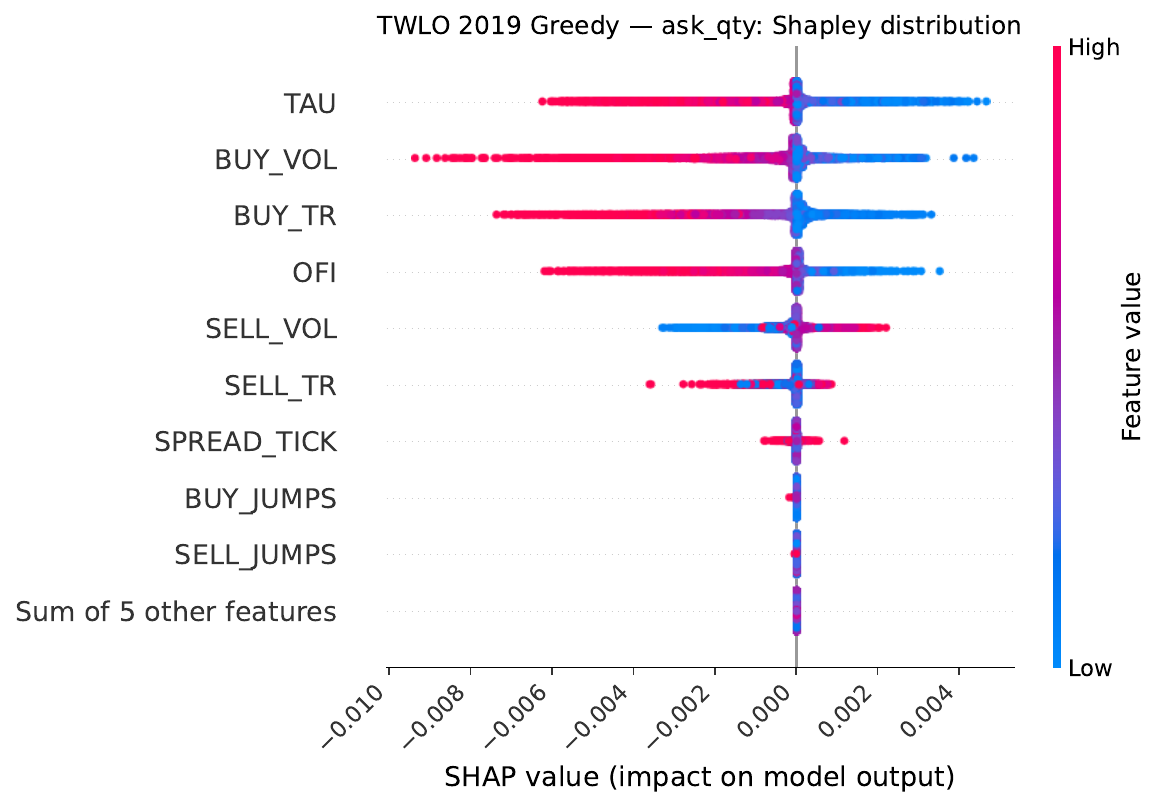}
            \caption{Greedy - Ask quantity.}
        \end{subfigure}
\vspace{0.6em}
        \begin{subfigure}{0.24\linewidth}
            \includegraphics[width=\linewidth]{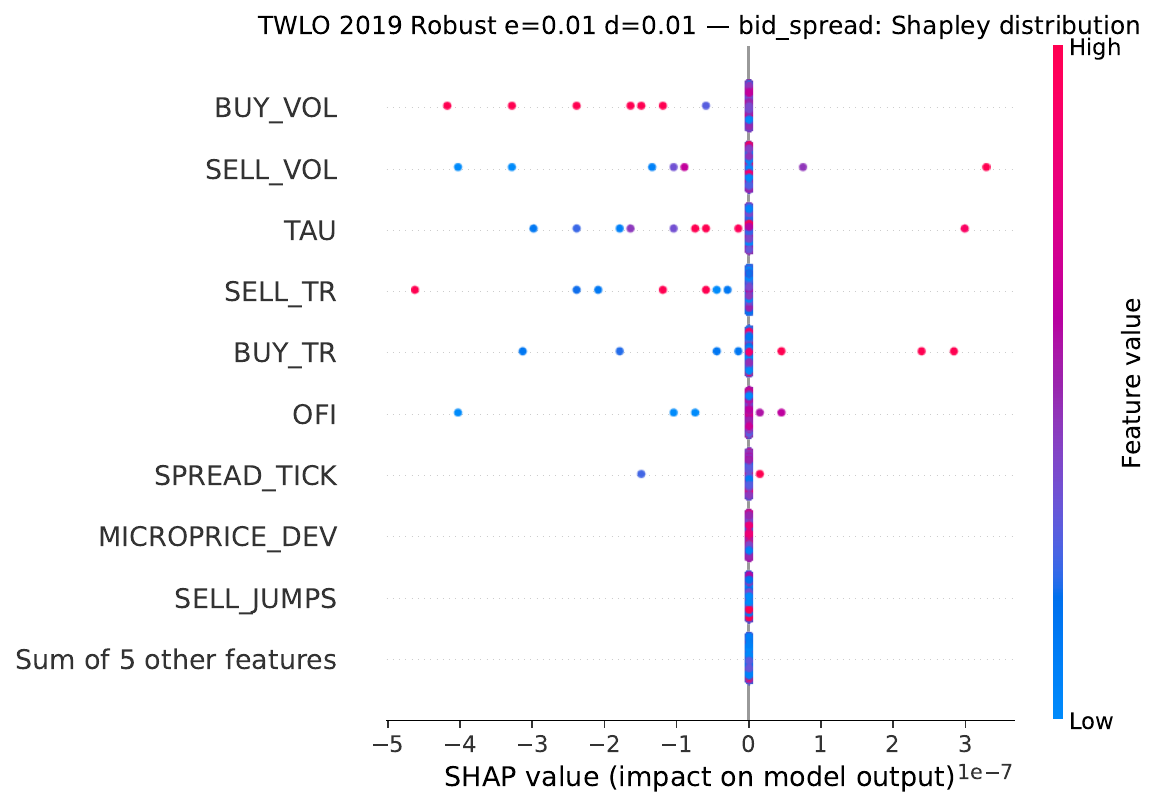}
            \caption{Robust - Bid spread.}
        \end{subfigure}\hfill
        \begin{subfigure}{0.24\linewidth}
            \includegraphics[width=\linewidth]{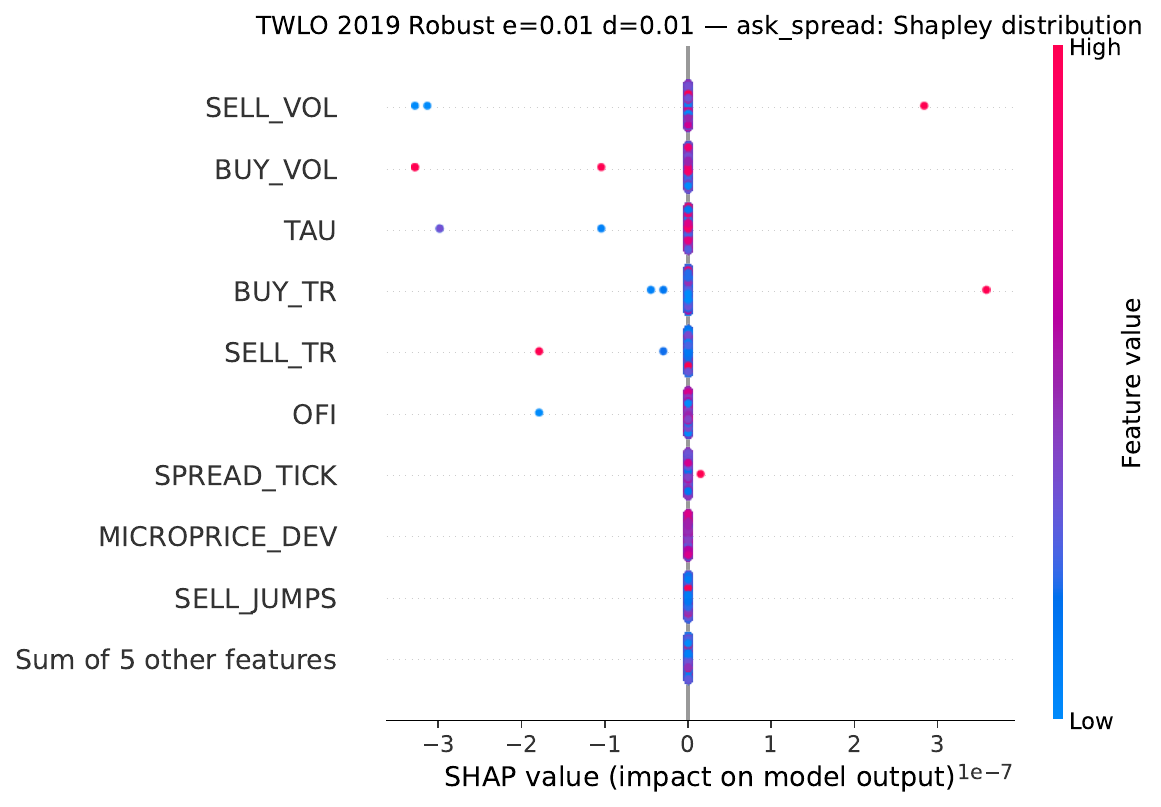}
            \caption{Robust - Ask spread.}
        \end{subfigure}\hfill
        \begin{subfigure}{0.24\linewidth}
            \includegraphics[width=\linewidth]{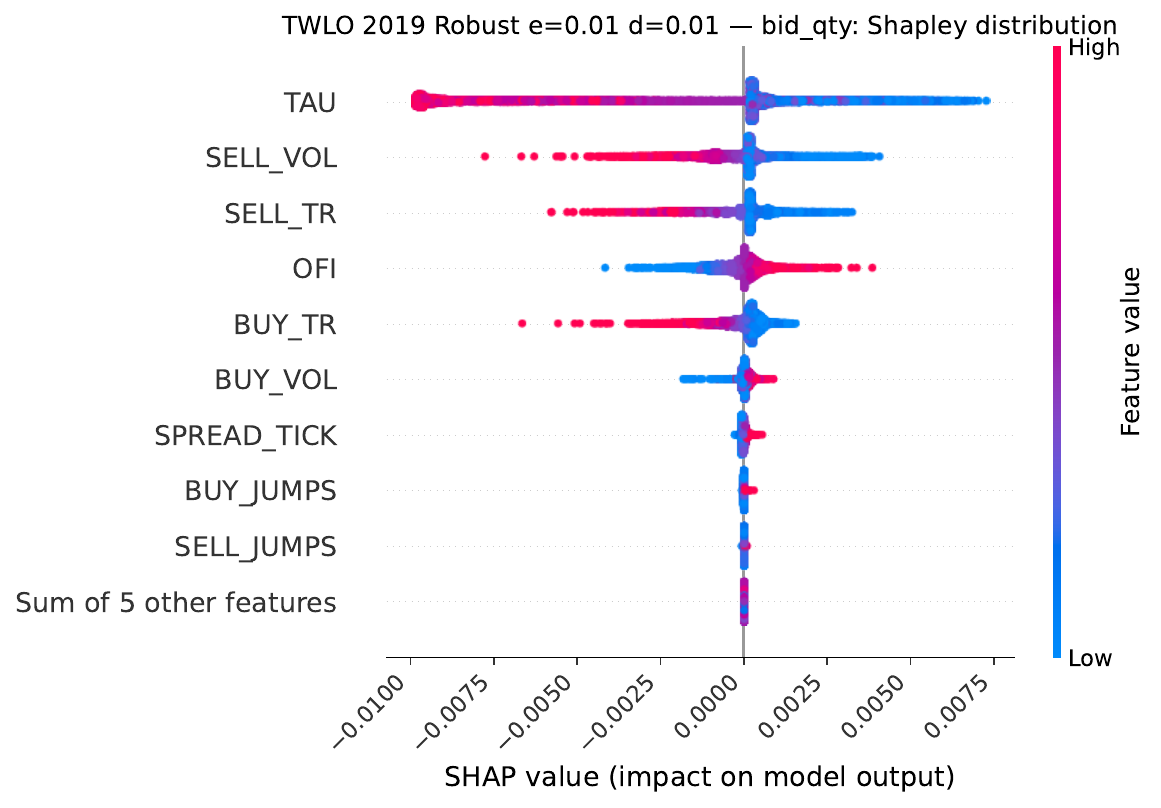}
            \caption{Robust - Bid quantity.}
        \end{subfigure}\hfill
        \begin{subfigure}{0.24\linewidth}
            \includegraphics[width=\linewidth]{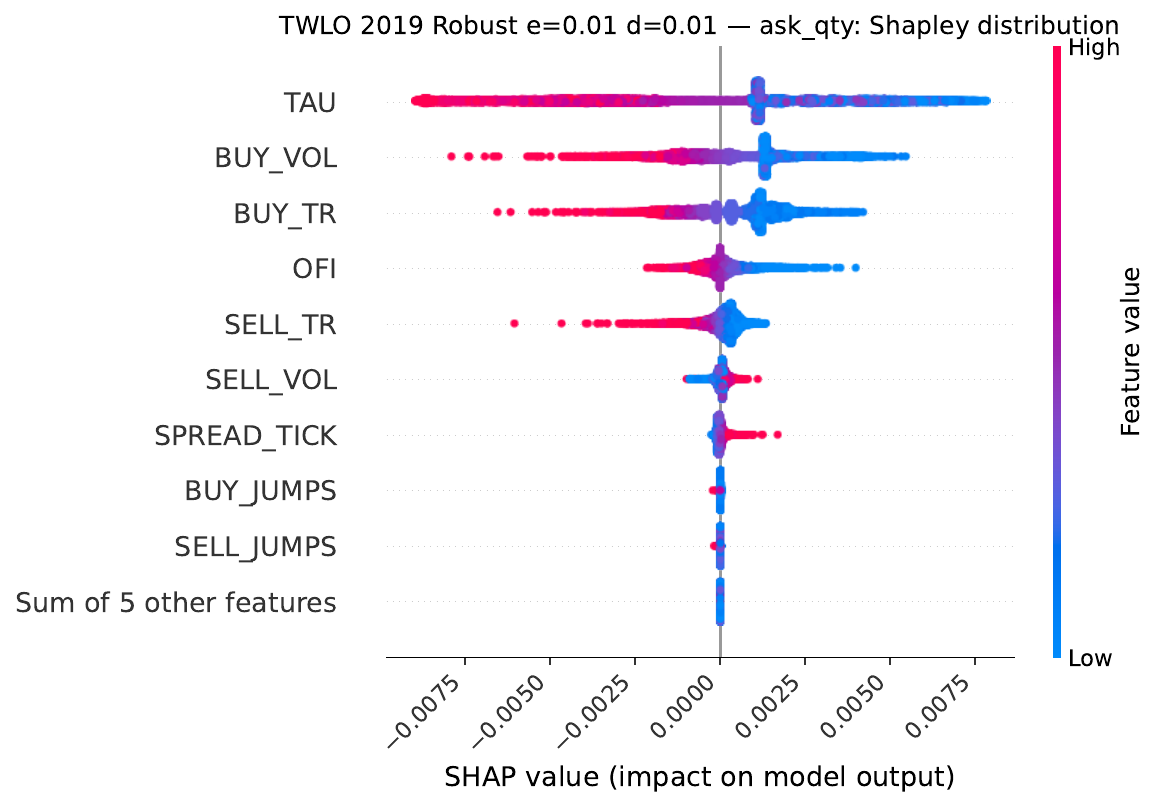}
            \caption{Robust - Ask quantity.}
        \end{subfigure}
    \end{minipage}}
    \caption{Distribution of absolute Shapley values for greedy and robust policies with $\bar{\varepsilon} = 0.01$ and $\delta = 0.01$ on TWLO, 2019, where the robust policy is the best test Sharpe Pareto configuration. The main drivers are again trade-flow and timing variables, although the relative weights appear somewhat more uneven than in the larger-liquidity names.}
    \label{fig:shapley_twlo_2019}
\end{figure}

\begin{figure}[H]
    \centering
    \makebox[\textwidth][c]{\begin{minipage}{1.16\textwidth}\centering
        \begin{subfigure}{0.24\linewidth}
            \includegraphics[width=\linewidth]{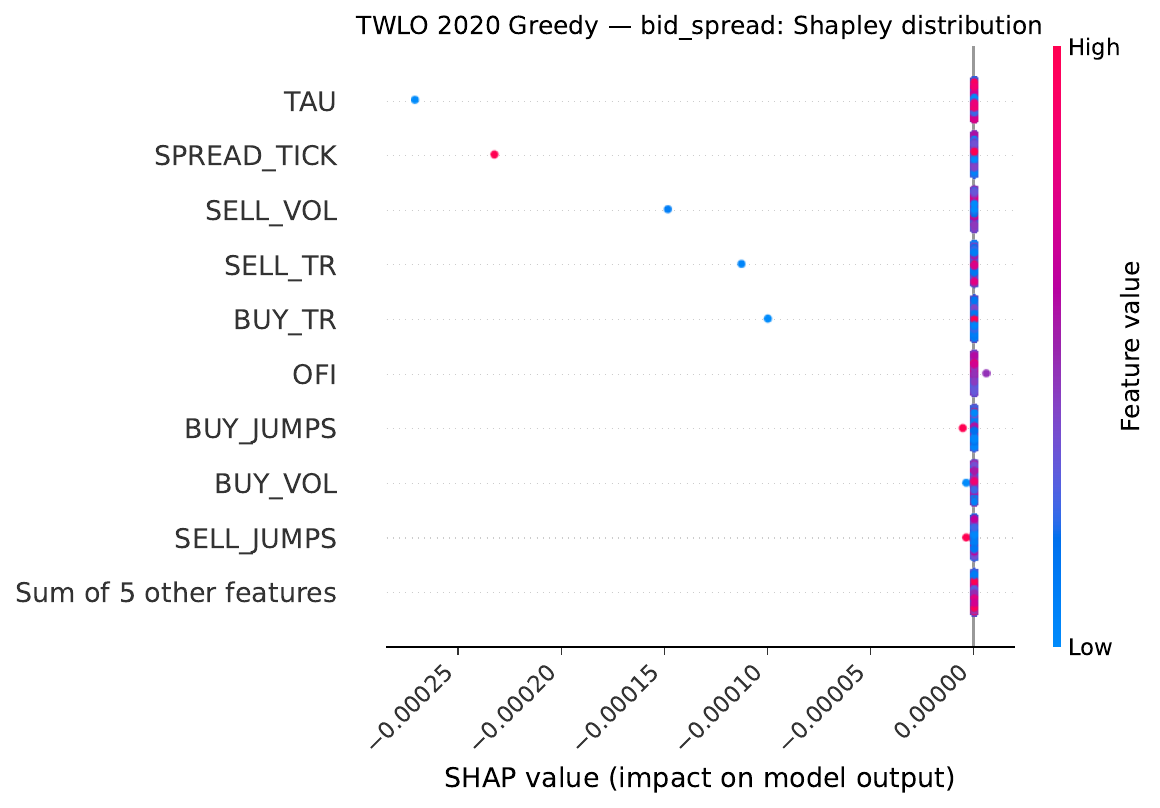}
            \caption{Greedy - Bid spread.}
        \end{subfigure}\hfill
        \begin{subfigure}{0.24\linewidth}
            \includegraphics[width=\linewidth]{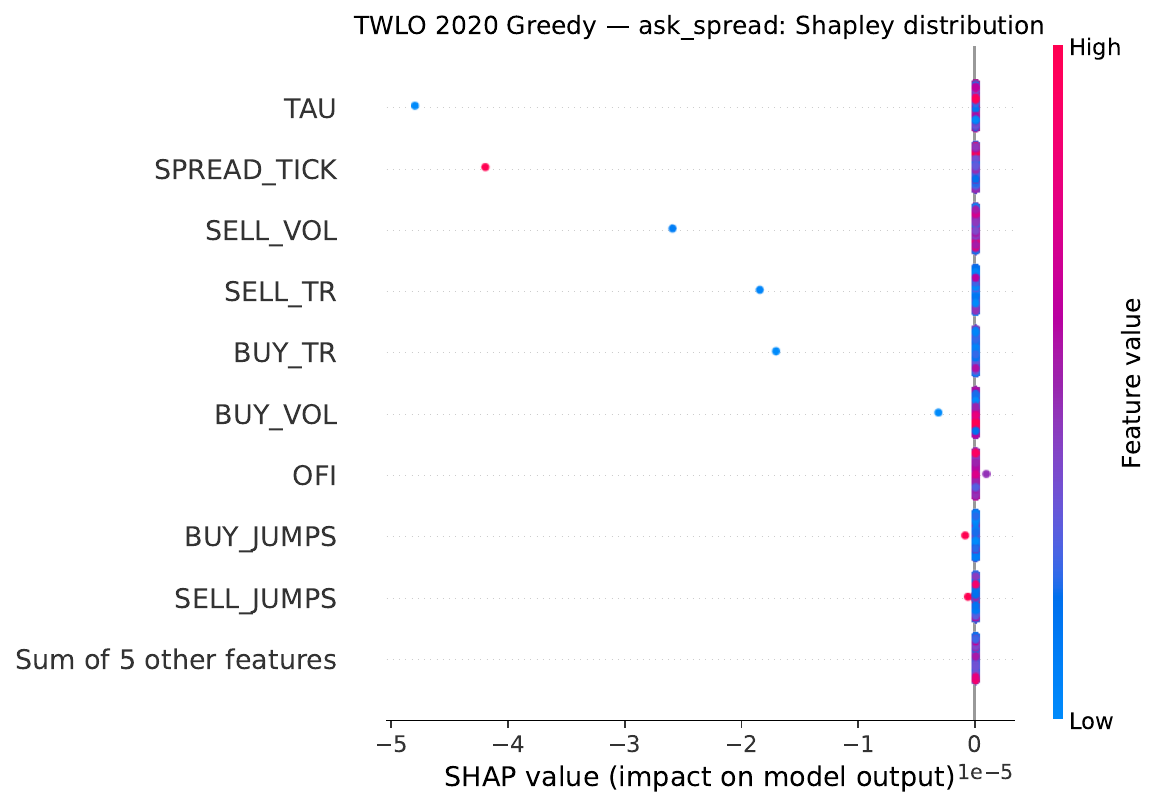}
            \caption{Greedy - Ask spread.}
        \end{subfigure}\hfill
        \begin{subfigure}{0.24\linewidth}
            \includegraphics[width=\linewidth]{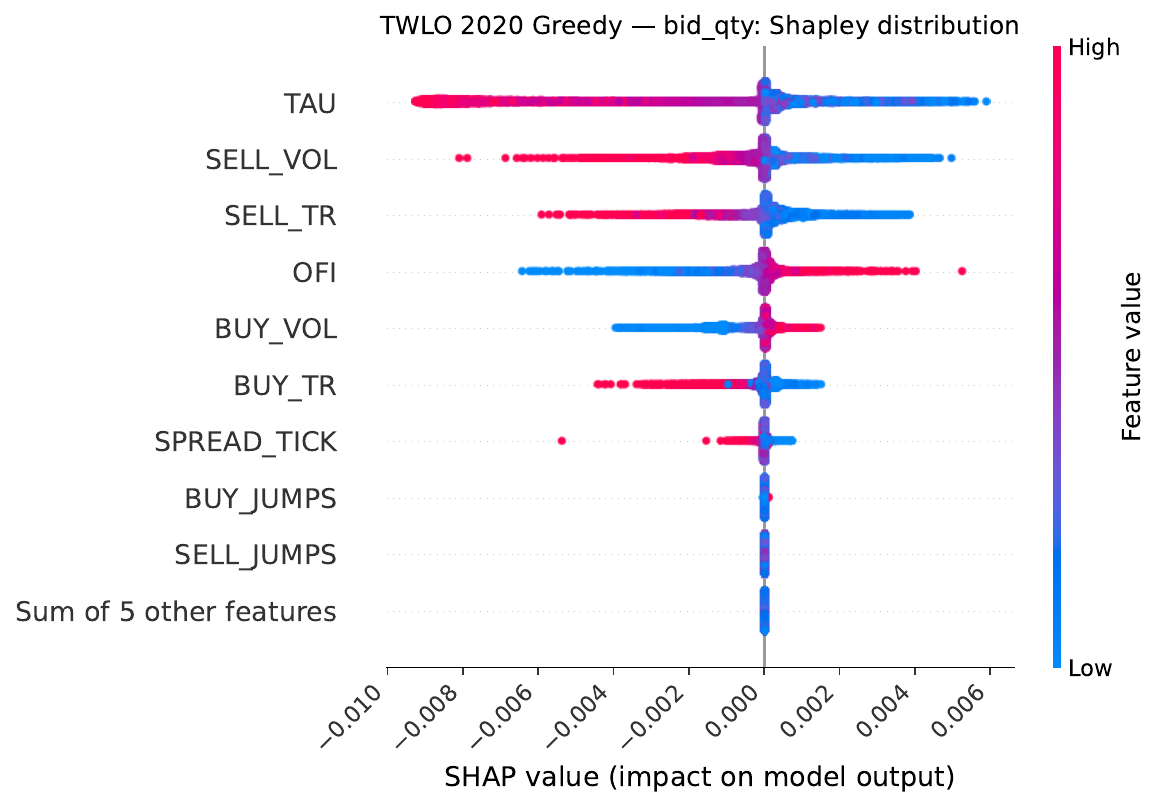}
            \caption{Greedy - Bid quantity.}
        \end{subfigure}\hfill
        \begin{subfigure}{0.24\linewidth}
            \includegraphics[width=\linewidth]{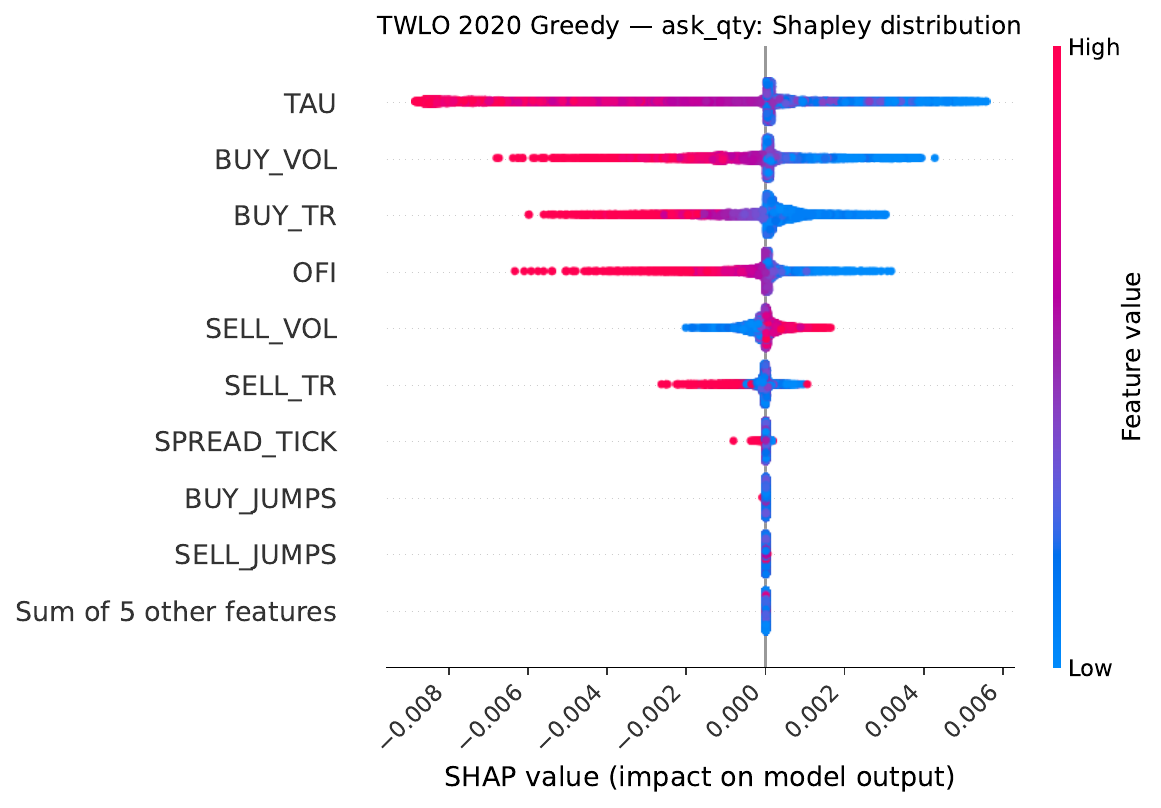}
            \caption{Greedy - Ask quantity.}
        \end{subfigure}
\vspace{0.6em}
        \begin{subfigure}{0.24\linewidth}
            \includegraphics[width=\linewidth]{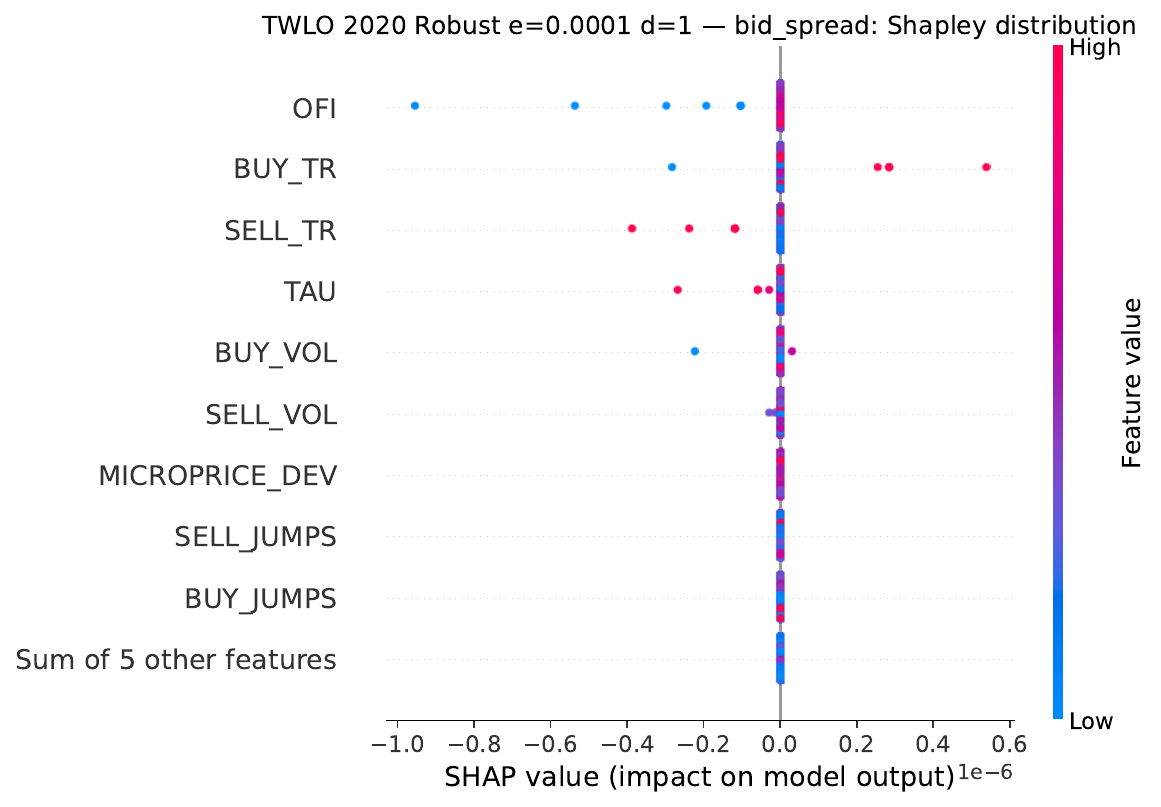}
            \caption{Robust - Bid spread.}
        \end{subfigure}\hfill
        \begin{subfigure}{0.24\linewidth}
            \includegraphics[width=\linewidth]{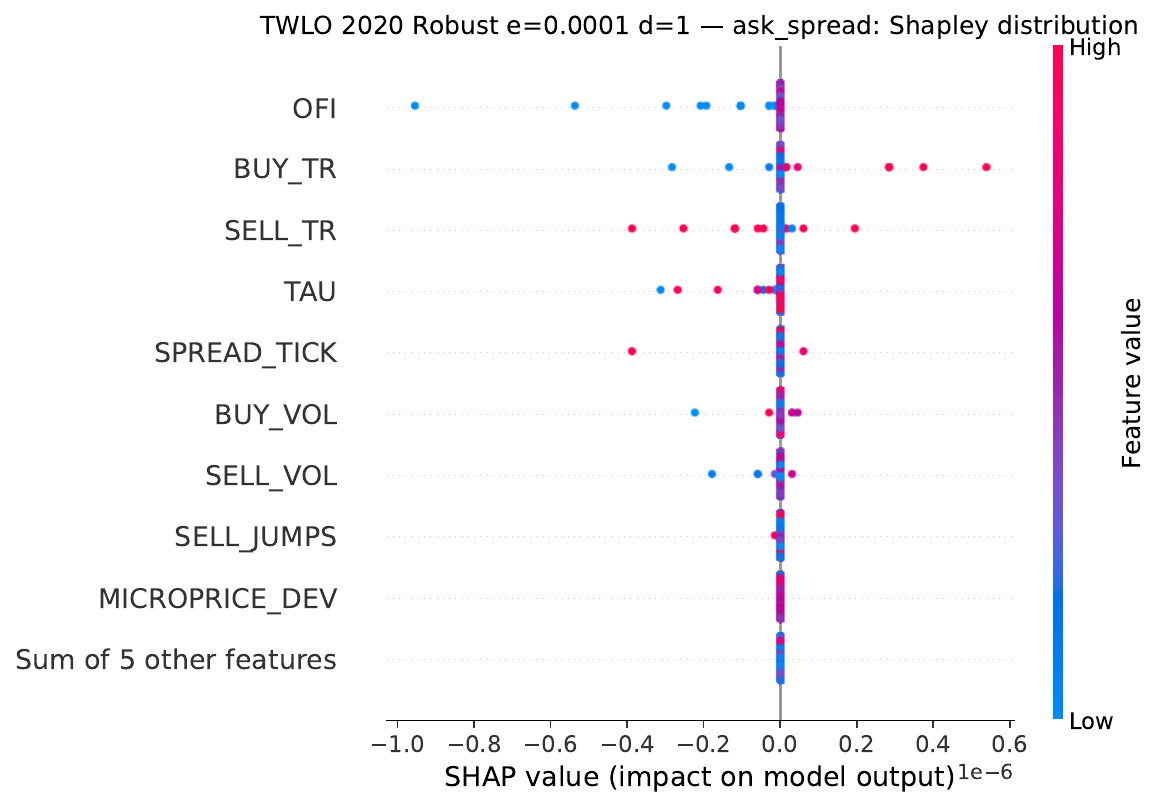}
            \caption{Robust - Ask spread.}
        \end{subfigure}\hfill
        \begin{subfigure}{0.24\linewidth}
            \includegraphics[width=\linewidth]{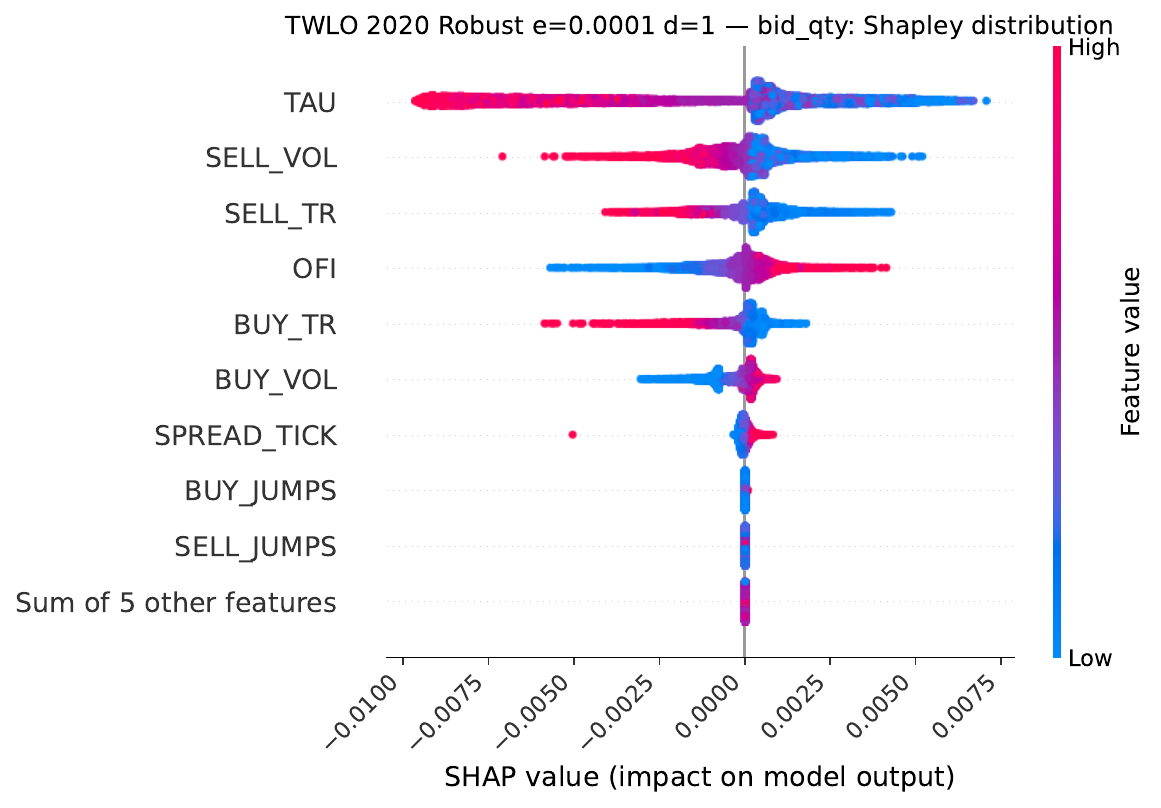}
            \caption{Robust - Bid quantity.}
        \end{subfigure}\hfill
        \begin{subfigure}{0.24\linewidth}
            \includegraphics[width=\linewidth]{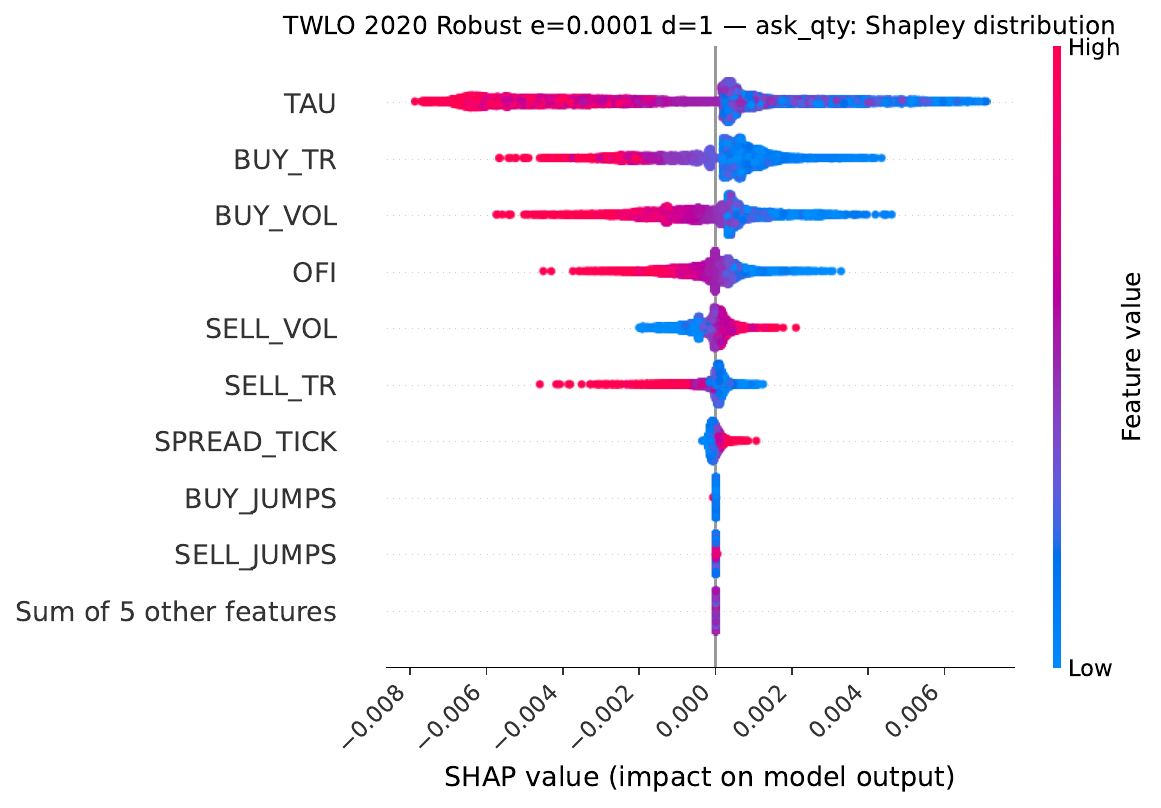}
            \caption{Robust - Ask quantity.}
        \end{subfigure}
    \end{minipage}}
    \caption{Distribution of absolute Shapley values for greedy and robust policies with $\bar{\varepsilon} = 0.0001$ and $\delta = 1$ on TWLO, 2020, where the robust policy is the best test Sharpe Pareto configuration. The decomposition remains concentrated on a small set of trade-flow and timing variables, with robustness mainly shifting their magnitude rather than changing which state features matter most.}
    \label{fig:shapley_twlo_2020}
\end{figure}

\subsection[Robust vs. Greedy Quoting Behaviour]{Robust vs.\ Greedy Quoting Behaviour}\label{sec:appendix_action_state}

Figures~\ref{fig:action_vs_state_real_aapl}--\ref{fig:action_vs_state_real_twlo} contrast the quoting policies of the robust and greedy agents as a function of each individual state feature, holding the remaining features at their mean and reporting the conditional mean (with $\pm 1$ standard deviation band) of each of the four quote components. For illustration, we use the best Sharpe robust configurations from the test Pareto frontier for each stock-year panel.

Across stocks and periods, the two agents share broadly the same conditional shape: their action curves move in the same direction with each state feature, reflecting a common market-making logic. The differences between robust and greedy are level shifts rather than shape changes\textemdash the robust agent posts marginally wider bid and ask half-spreads and smaller bid and ask quantities, with the gap most visible in volatile states. The level differences are small by construction: we cap the half-spread action at roughly the $95$th\textendash$99$th percentile of the empirical half-spread distribution ($0.03$ for AAPL and $0.5$ for TSLA) and cap the quantity at $0.01\%$ of total market size for both stocks. Because the empirical half-spread is small relative to these caps and the agent acts on a one-minute interval, both policies push the half-spread toward its upper bound and concentrate their optimisation on the quoted quantity. The robust agent's improvement over the greedy benchmark therefore comes mainly from a more cautious choice of quantity in adverse states, not from materially different spread shapes.

Reading the conditional curves feature by feature, several patterns emerge. 
As the number of sell market orders or the average sell trade size rises, both agents lower the bid quantity, and a symmetric pattern appears on the ask side as the agent's ask quantity decreases with the number of market buy orders and average market buy order size. 
Because quoted quantities are defined as percentages of opposite-side market flow. When opposite-side order arrival or trade size increases, the agent can quote a smaller fraction while maintaining a broadly similar expected absolute fill level. The order-flow-imbalance feature shows the canonical predictive-signal response of \citet{ContKukanovStoikov2014}: a positive imbalance, which forecasts upward price moves, induces both agents to enlarge the bid quantity, while a negative imbalance enlarges the ask quantity. As a function of time-to-close, both agents quote smaller quantities at the open and especially at the close and exhibit the well-known U-shape in intraday activity (\citealp{AndersenBollerslev1998}); the steeper drop near the close mirrors the end-of-day inventory-liquidation behaviour central to \citet{Cartea2015} and \citet{CarteaJaimungalPenalva2015}. 
In terms of inventory, the RL agents tend to lower bid quantity and increase ask quantity when inventory increases. This behavior is expected as the agents are supposed to keep their inventory low around 0.

\begin{figure}[H]
    \centering
    \makebox[\textwidth][c]{\begin{minipage}{1.14\textwidth}\centering
        \begin{subfigure}{\linewidth}
            \includegraphics[width=\linewidth]{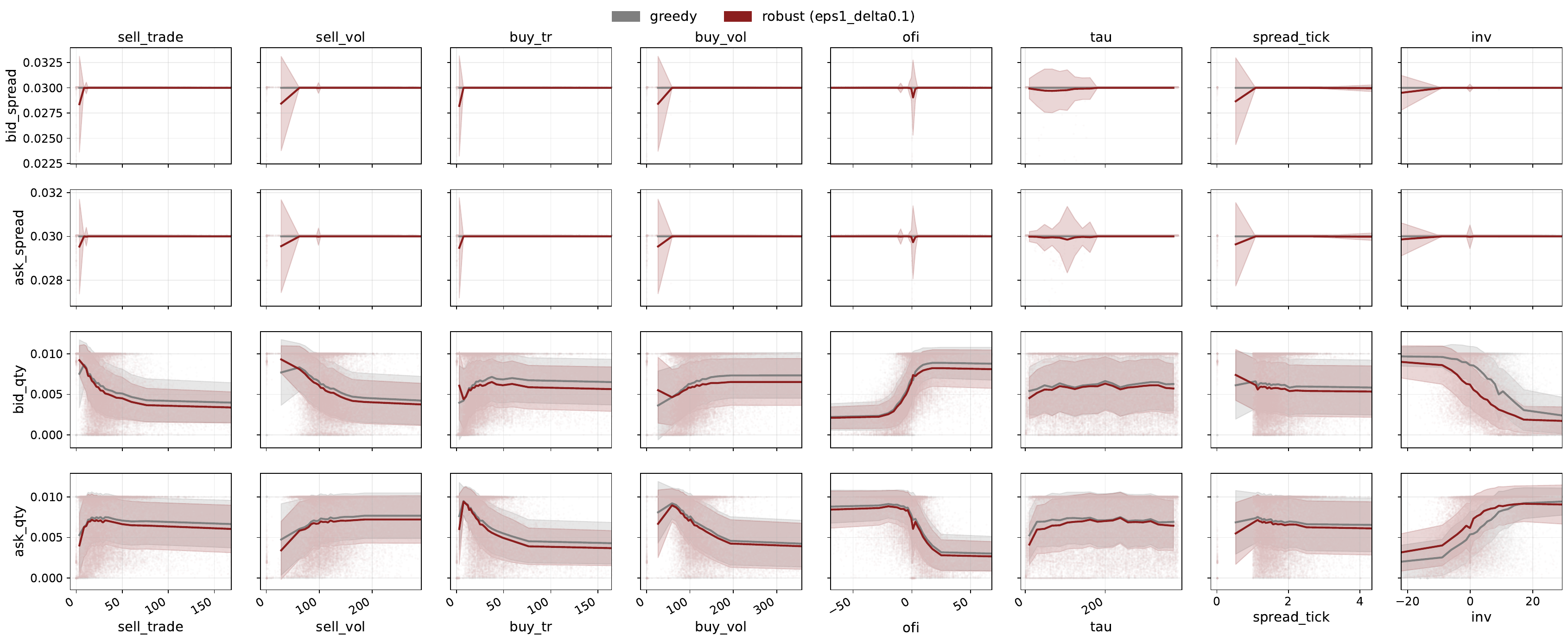}
            \caption{AAPL, 2019 ($\bar\varepsilon=1,\ \delta=0.1$).}
        \end{subfigure}

        \vspace{0.4em}

        \begin{subfigure}{\linewidth}
            \includegraphics[width=\linewidth]{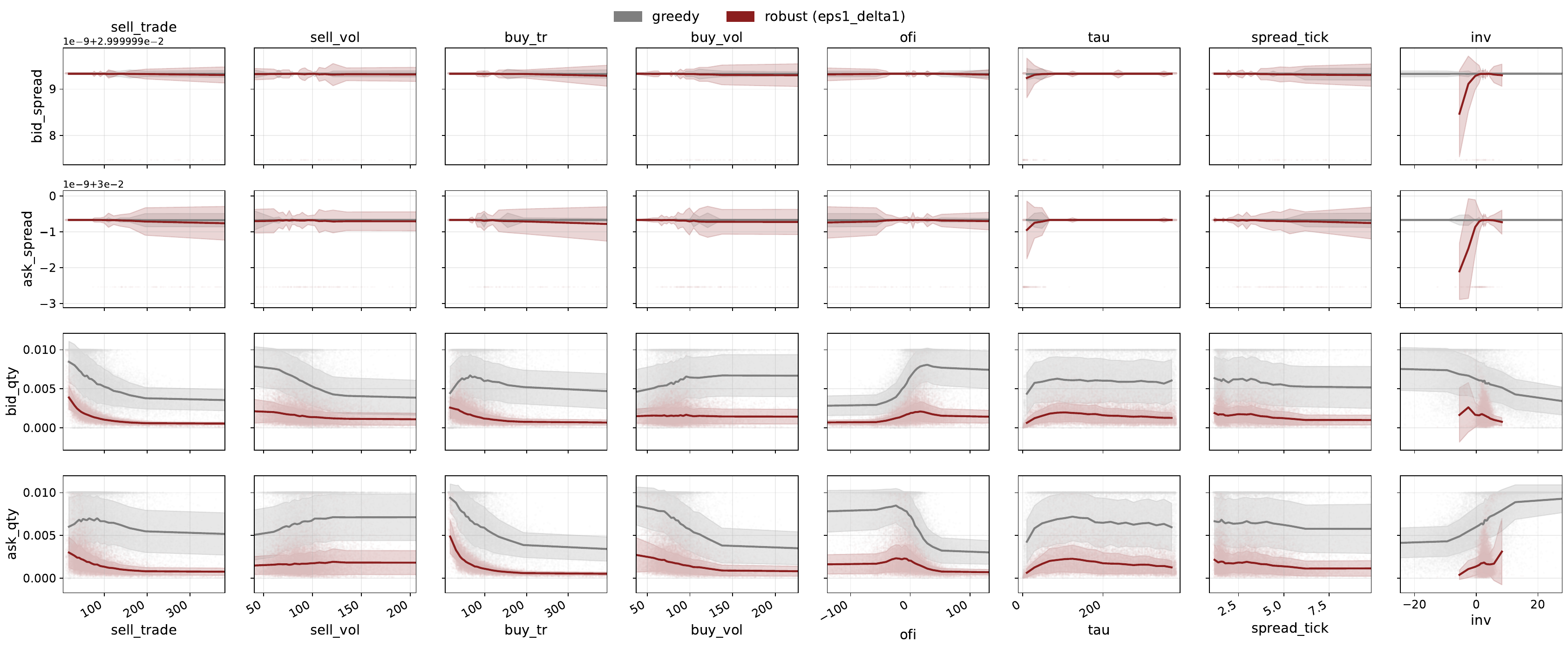}
            \caption{AAPL, 2020 ($\bar\varepsilon=1,\ \delta=1$).}
        \end{subfigure}
    \end{minipage}}
    \caption{Mean quote (bid quantity, ask quantity, bid spread, ask spread) of the robust and greedy agents conditional on each state feature for AAPL, with $\pm1$ s.d.\ bands. The robust policy corresponds to the best test Sharpe Pareto configuration and responds structurally to state features, mainly differing from the greedy benchmark through more conservative quantity adjustments; this difference is more pronounced in 2020.}
    \label{fig:action_vs_state_real_aapl}
\end{figure}

\begin{figure}[H]
    \centering
    \makebox[\textwidth][c]{\begin{minipage}{1.14\textwidth}\centering
        \begin{subfigure}{\linewidth}
            \includegraphics[width=\linewidth]{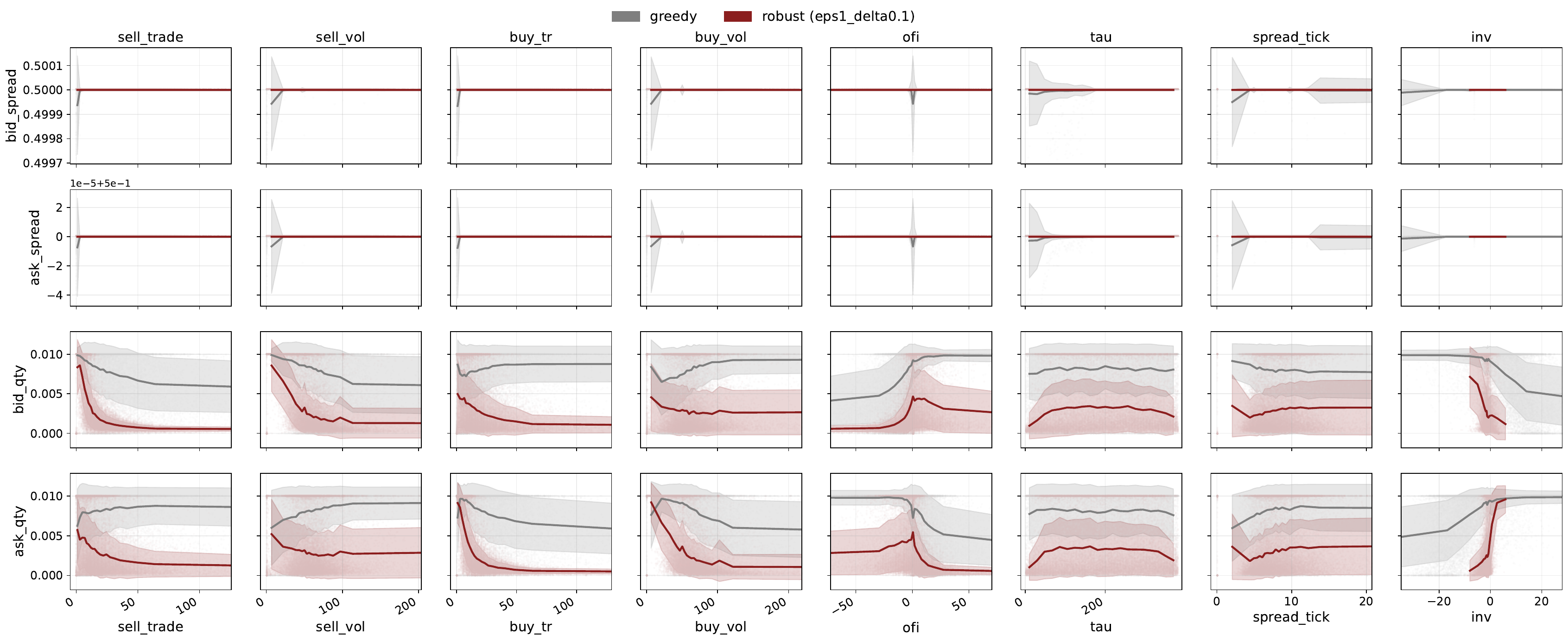}
            \caption{TSLA, 2019 ($\bar\varepsilon=1,\ \delta=0.1$).}
        \end{subfigure}

        \vspace{0.4em}

        \begin{subfigure}{\linewidth}
            \includegraphics[width=\linewidth]{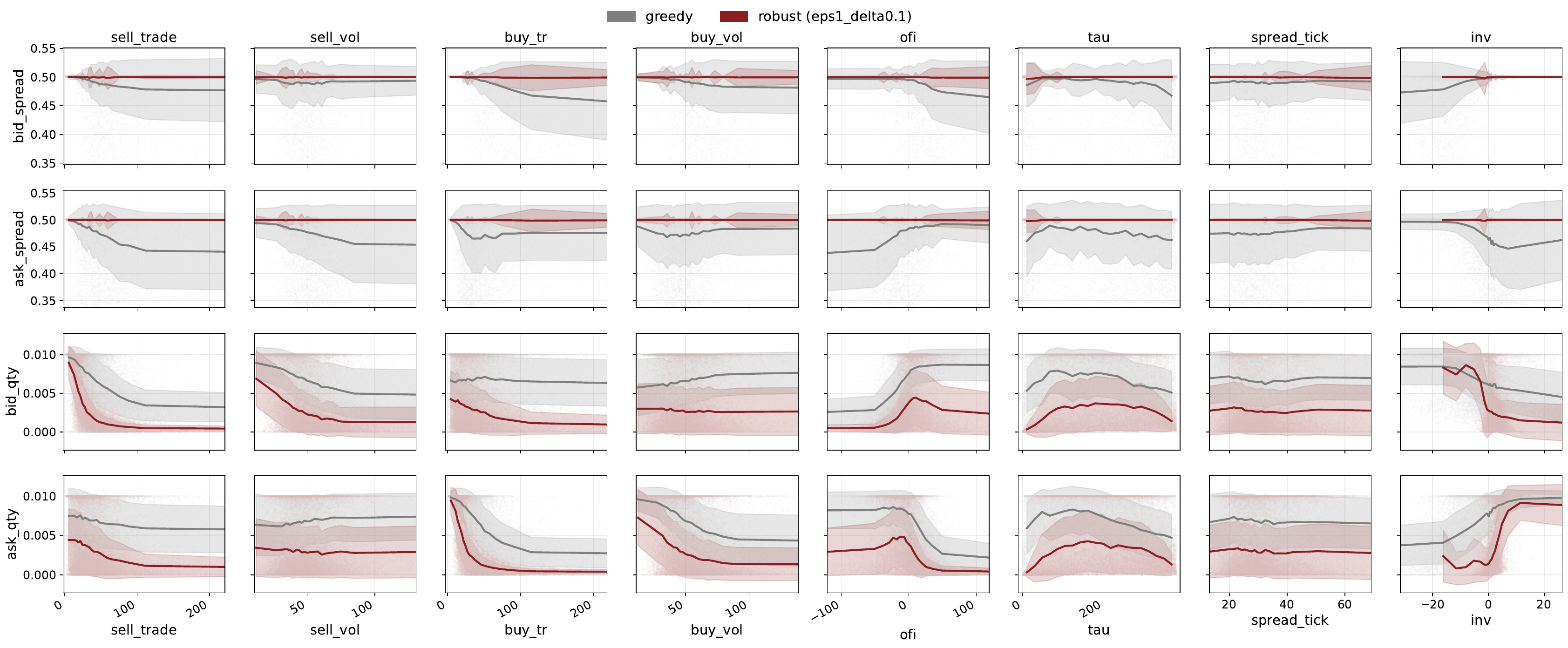}
            \caption{TSLA, 2020 ($\bar\varepsilon=1,\ \delta=0.1$).}
        \end{subfigure}
    \end{minipage}}
    \caption{Mean quote (bid quantity, ask quantity, bid spread, ask spread) of the robust and greedy agents conditional on each state feature for TSLA, with $\pm1$ s.d.\ bands. The robust policy corresponds to the best test Sharpe Pareto configuration and responds structurally to state features, mainly differing from the greedy benchmark through more defensive quantity and spread responses in riskier states.}
    \label{fig:action_vs_state_real_tsla}
\end{figure}

\begin{figure}[H]
    \centering
    \makebox[\textwidth][c]{\begin{minipage}{1.14\textwidth}\centering
        \begin{subfigure}{\linewidth}
            \includegraphics[width=\linewidth]{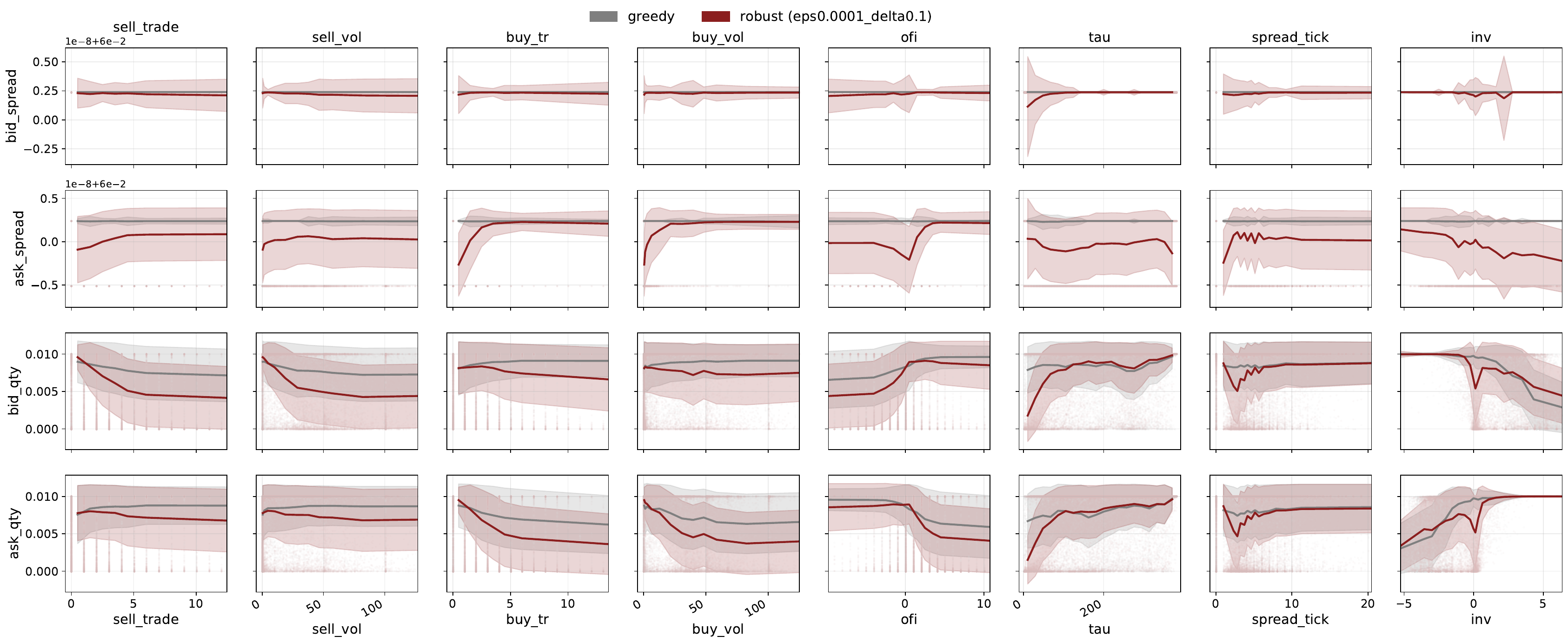}
            \caption{MKC, 2019 ($\bar\varepsilon=0.0001,\ \delta=0.1$).}
        \end{subfigure}

        \vspace{0.4em}

        \begin{subfigure}{\linewidth}
            \includegraphics[width=\linewidth]{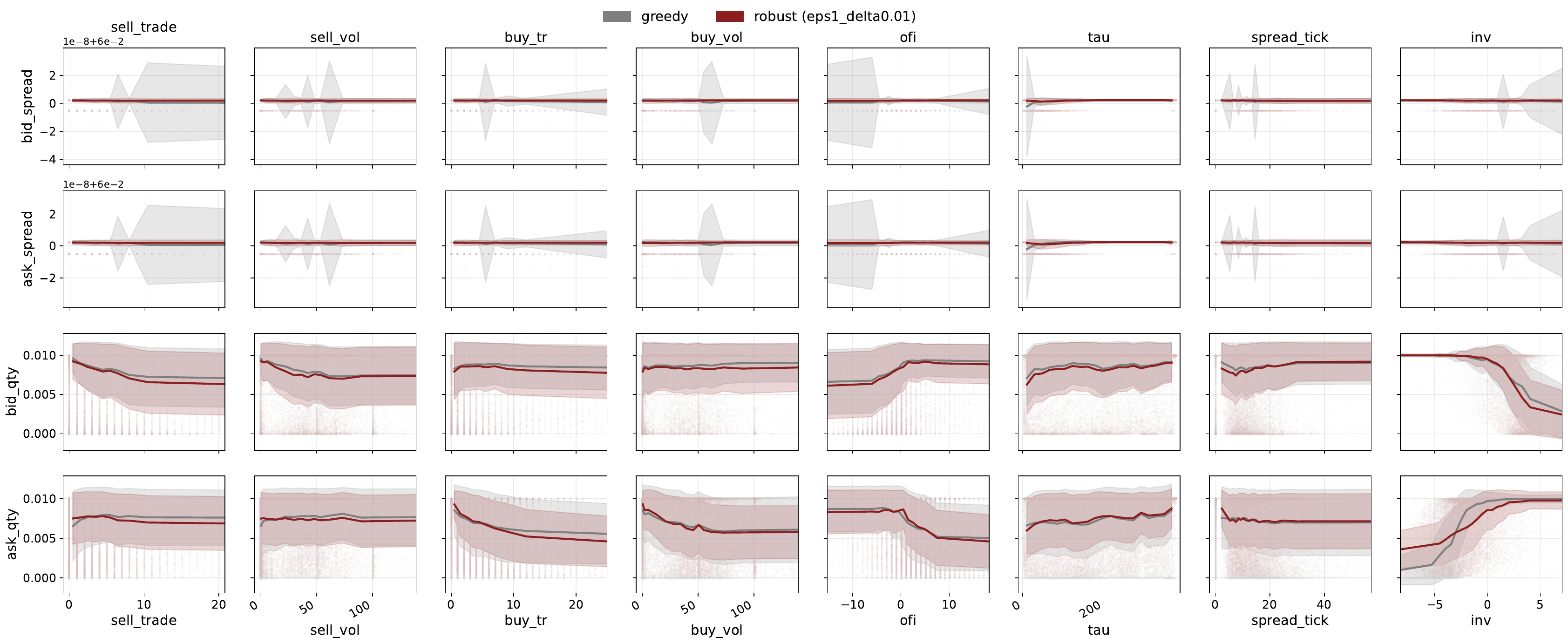}
            \caption{MKC, 2020 ($\bar\varepsilon=1,\ \delta=0.01$).}
        \end{subfigure}
    \end{minipage}}
    \caption{Mean quote (bid quantity, ask quantity, bid spread, ask spread) of the robust and greedy agents conditional on each state feature for MKC, with $\pm1$ s.d.\ bands. The robust policy corresponds to the best test Sharpe Pareto configuration and responds structurally to state features, though the two policies remain close overall, consistent with the comparatively weaker role of robustness in this stock.}
    \label{fig:action_vs_state_real_mkc}
\end{figure}

\begin{figure}[H]
    \centering
    \makebox[\textwidth][c]{\begin{minipage}{1.14\textwidth}\centering
        \begin{subfigure}{\linewidth}
            \includegraphics[width=\linewidth]{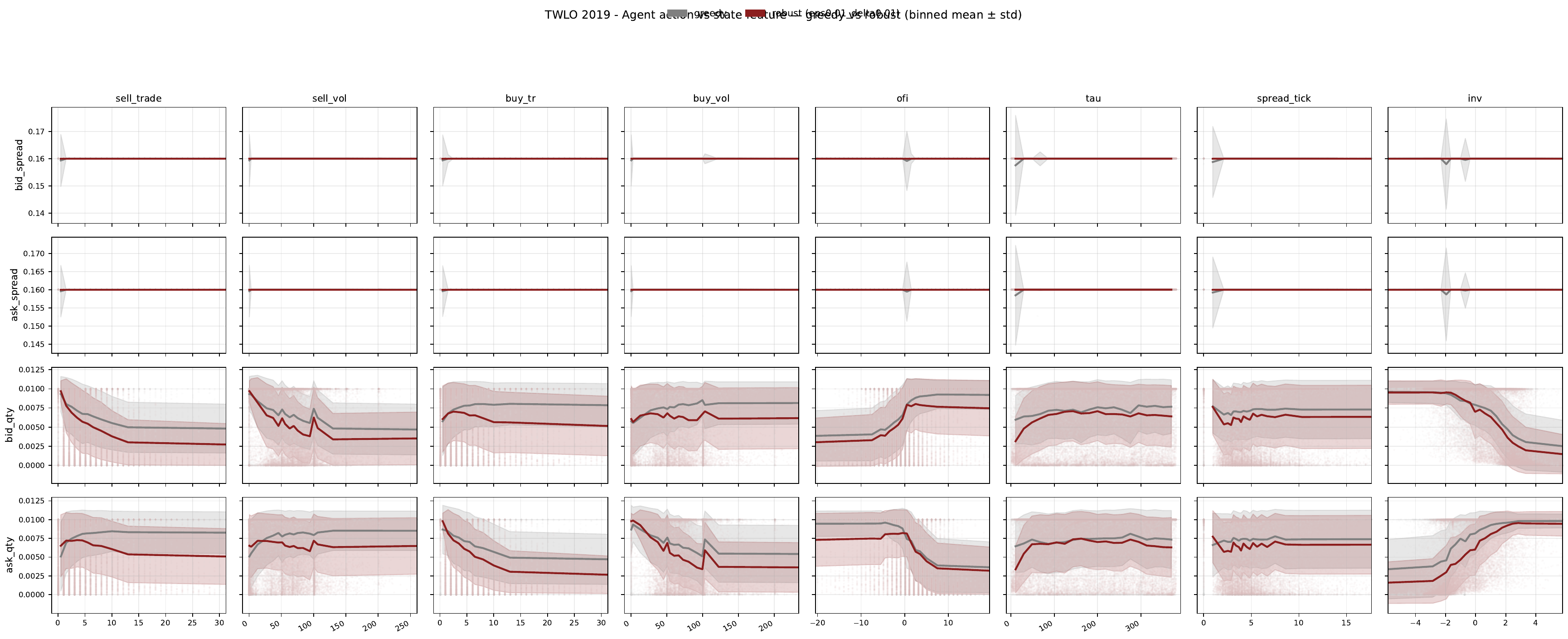}
            \caption{TWLO, 2019 ($\bar\varepsilon=0.01,\ \delta=0.01$).}
        \end{subfigure}

        \vspace{0.4em}

        \begin{subfigure}{\linewidth}
            \includegraphics[width=\linewidth]{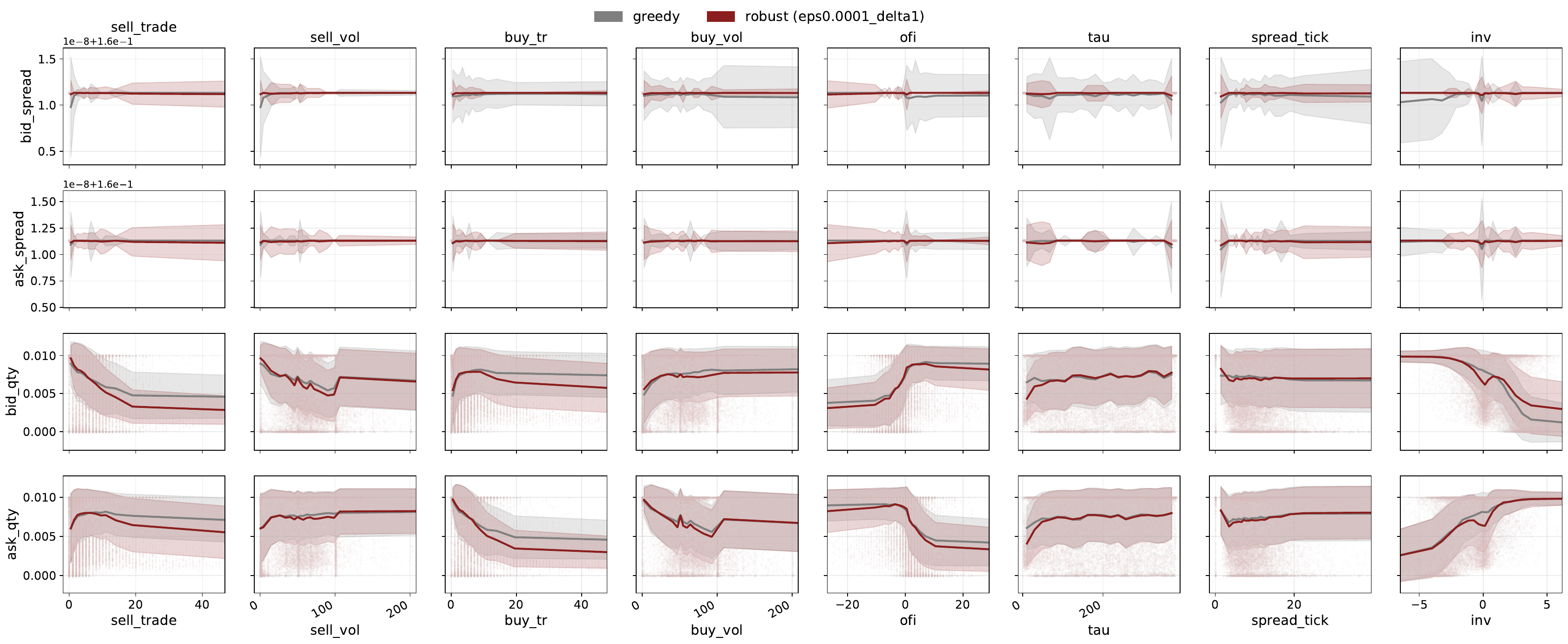}
            \caption{TWLO, 2020 ($\bar\varepsilon=0.0001,\ \delta=1$).}
        \end{subfigure}
    \end{minipage}}
    \caption{Mean quote (bid quantity, ask quantity, bid spread, ask spread) of the robust and greedy agents conditional on each state feature for TWLO, with $\pm1$ s.d.\ bands. The robust policy corresponds to the best test Sharpe Pareto configuration and responds structurally to state features, mainly differing from the greedy benchmark through more selective quantity provision in less favorable states.}
    \label{fig:action_vs_state_real_twlo}
\end{figure}

\subsection[Validation--Test Pareto Frontier]{Validation--Test Pareto Frontier}\label{sec:appendix_val_test_frontier}

To select $(\bar\varepsilon,\delta)$ in a way that does not overfit to the test sample, we choose configurations on the \emph{validation} Pareto frontier in (Sharpe, mean P\&L) space and evaluate them out-of-sample. Figure~\ref{fig:pareto_real} shows, for each stock\textendash period, the validation frontier (left subpanel) alongside the corresponding test performance (right subpanel); points are coloured by $\log_{10}(\delta)$. Configurations on the validation frontier dominate the greedy benchmark out-of-sample, indicating that the $(\bar\varepsilon,\delta)$ ranking transfers reliably from validation to test, including in the COVID-affected 2020 window.

The gap between validation and test performance, where it appears, is mostly driven by distributional shifts between the validation and test windows. Appendix~\ref{sec:full-summary} reports the train/validation/test distributions of the four flow features that the Shapley analysis identifies as the dominant inputs to the policy: the number of buy and sell trades, the average buy and sell trade sizes, the order-flow imbalance, and the spread (in ticks). The 2019 panels already show that good validation--test transfer does not require the two samples to match perfectly. For AAPL 2019 and TSLA 2019, the validation and test distributions differ visibly in trade counts and spread, yet the broad shape of the state distributions remains comparable enough for the validation ranking to remain informative out of sample. The 2020 panels display more pronounced shifts associated with the COVID-19 shock, but the nature of the shift differs across names. For AAPL 2020, the test window exhibits substantially higher trade counts, a wider spread distribution, and larger volatility-related measures than the validation window, while average trade sizes are not higher on test. For TSLA 2020, trade counts remain broadly similar between validation and test, whereas the clearest changes occur in spread, returns, and volatility-related variables. The lower-liquidity names exhibit a related but distinct pattern. For MKC 2019, the validation and test distributions remain broadly aligned in trade counts, trade sizes, and order-flow imbalance, which is consistent with the relatively small but stable transfer from the validation frontier to the test window. In MKC 2020, by contrast, the spread-tick distribution becomes more dispersed and the flow variables shift more noticeably between validation and test, while the stock's low-volatility, low-liquidity environment leaves only limited room for robust hyperparameters to separate themselves out of sample. For TWLO 2019, the validation and test distributions are again reasonably well aligned, and this is reflected in the comparatively clean transfer of the validation ranking to the test frontier. For TWLO 2020, the validation and test distributions remain broadly similar for trade counts, trade sizes, and order-flow imbalance, while the clearest difference appears in the spread distribution. Even so, the robust configurations continue to transfer comparatively well out of sample. Importantly, the \emph{ranking} of $(\bar\varepsilon,\delta)$ configurations on the validation frontier still produces points that dominate the greedy benchmark out-of-sample, which is the property robustness is meant to deliver: stable hyperparameter selection in the presence of distributional drift between the two evaluation windows. In short, hyperparameter selection performs best when the validation and test state distributions are broadly aligned, but it remains reasonably reliable even when some features shift materially, provided the validation frontier still captures the relevant risk-return trade-offs.

\begin{figure}[H]
    \centering
    \makebox[\textwidth][c]{\begin{minipage}{1.14\textwidth}\centering
        \begin{subfigure}{0.49\linewidth}
            \includegraphics[width=\linewidth]{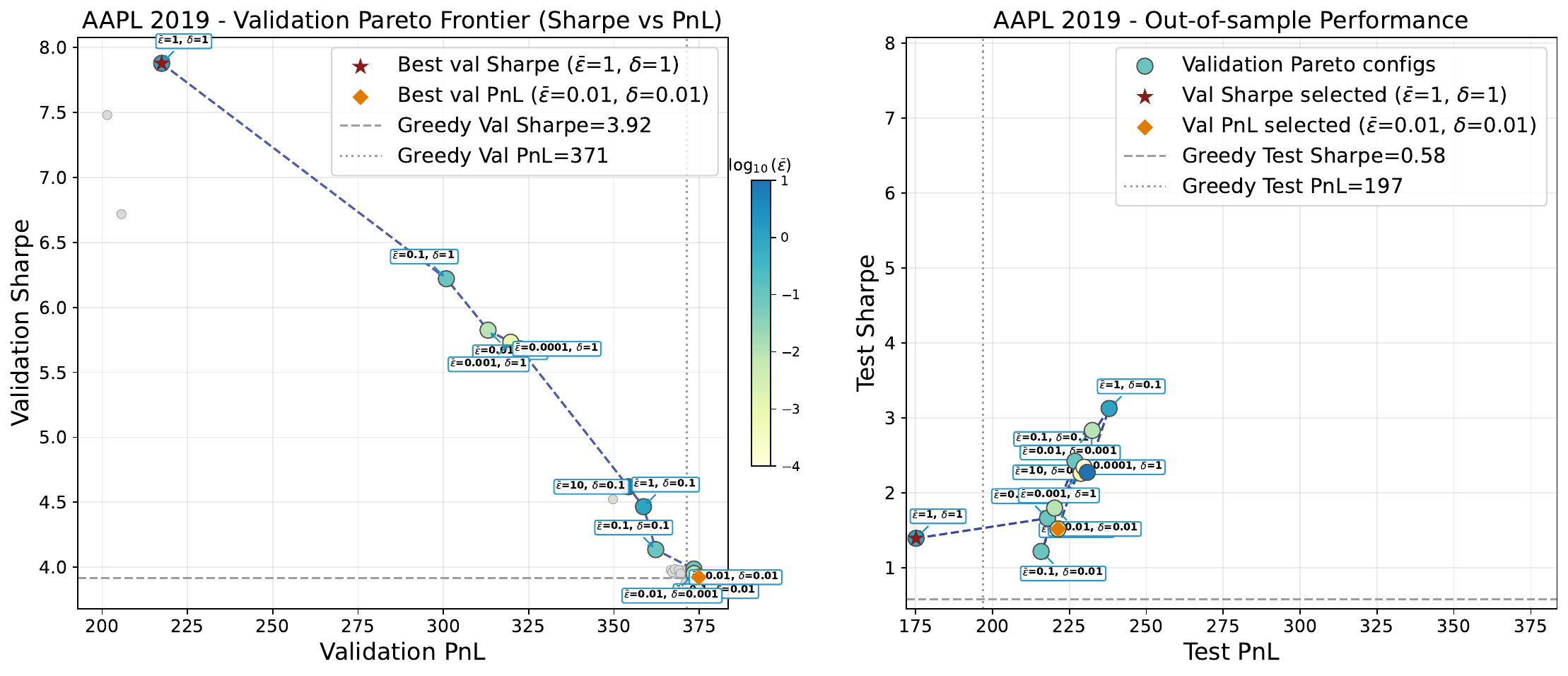}
            \caption{AAPL, 2019.}
        \end{subfigure}
        \begin{subfigure}{0.49\linewidth}
            \includegraphics[width=\linewidth]{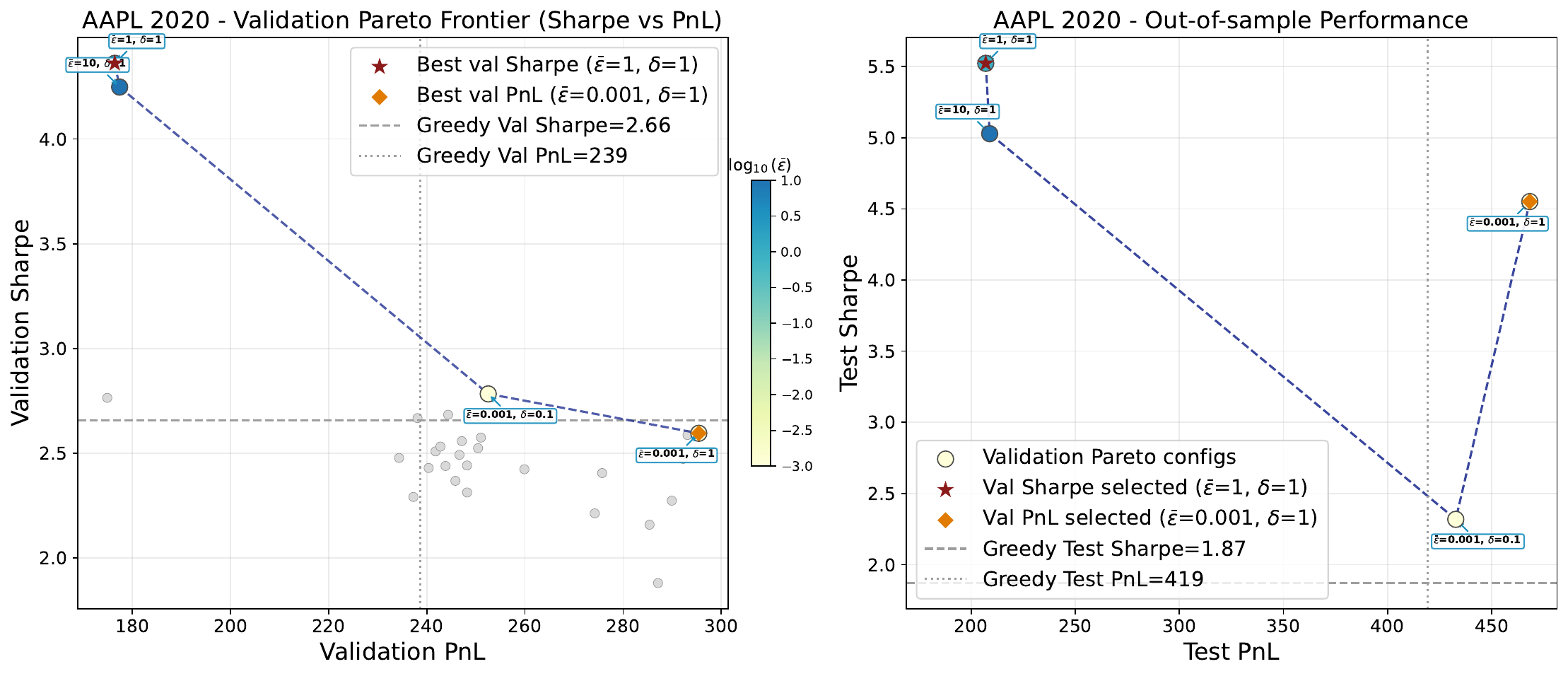}
            \caption{AAPL, 2020.}
        \end{subfigure}
        \begin{subfigure}{0.49\linewidth}
            \includegraphics[width=\linewidth]{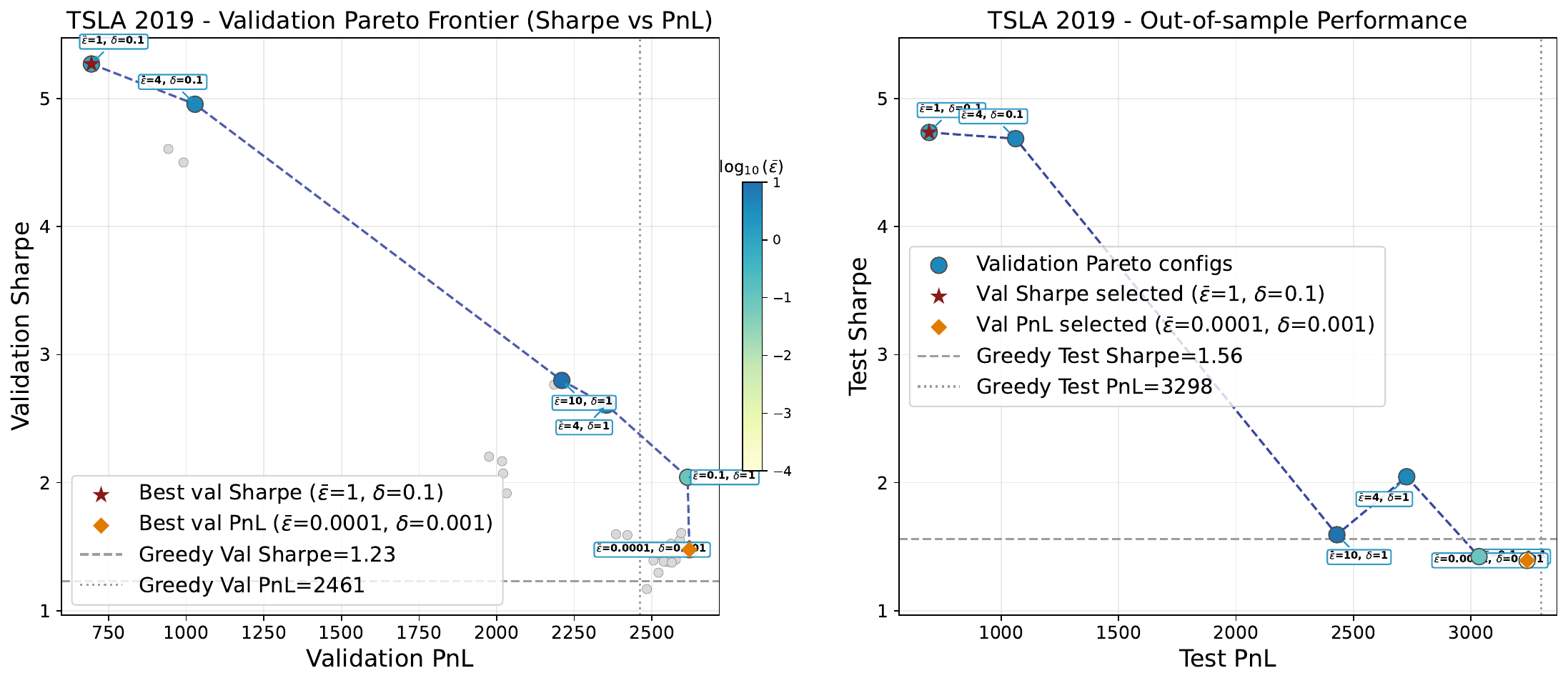}
            \caption{TSLA, 2019.}
        \end{subfigure}
        \begin{subfigure}{0.49\linewidth}
            \includegraphics[width=\linewidth]{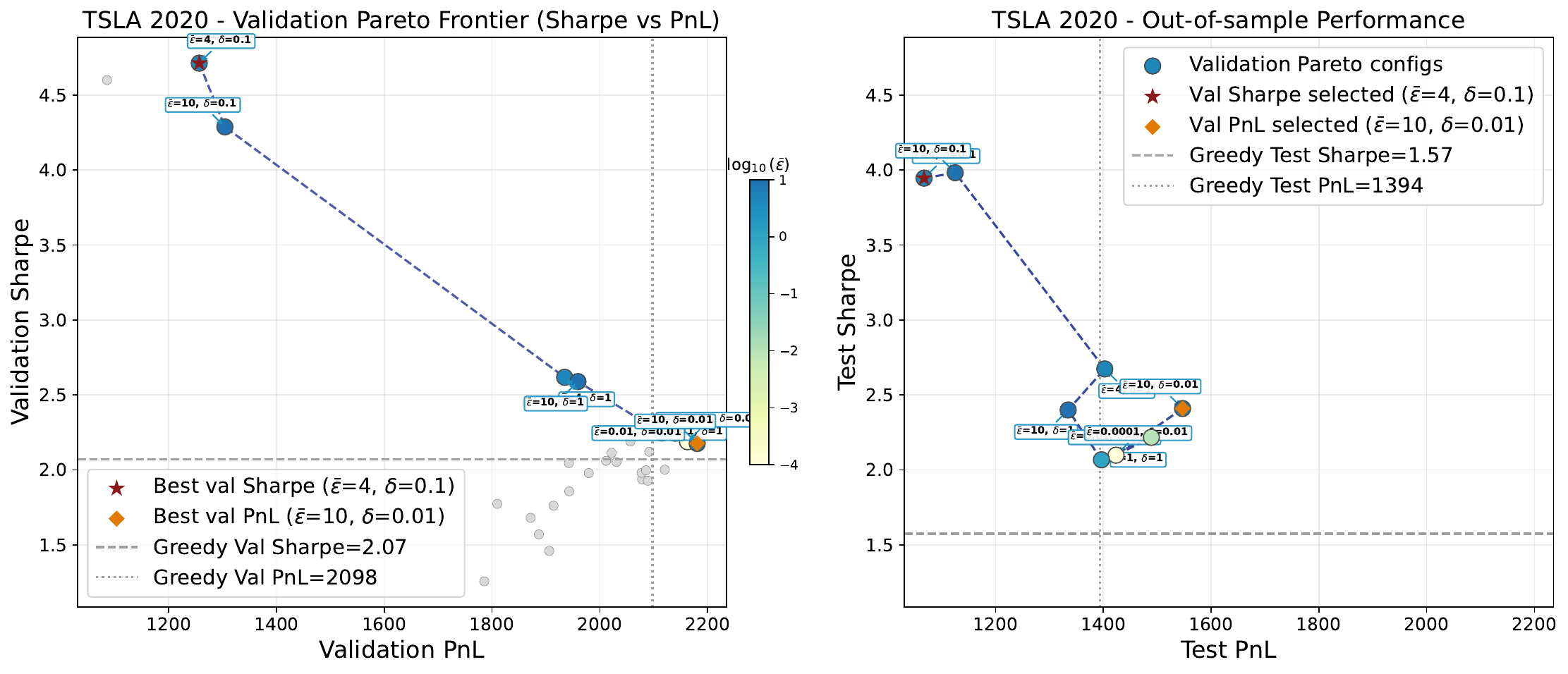}
            \caption{TSLA, 2020.}
        \end{subfigure}
        \begin{subfigure}{0.49\linewidth}
            \includegraphics[width=\linewidth]{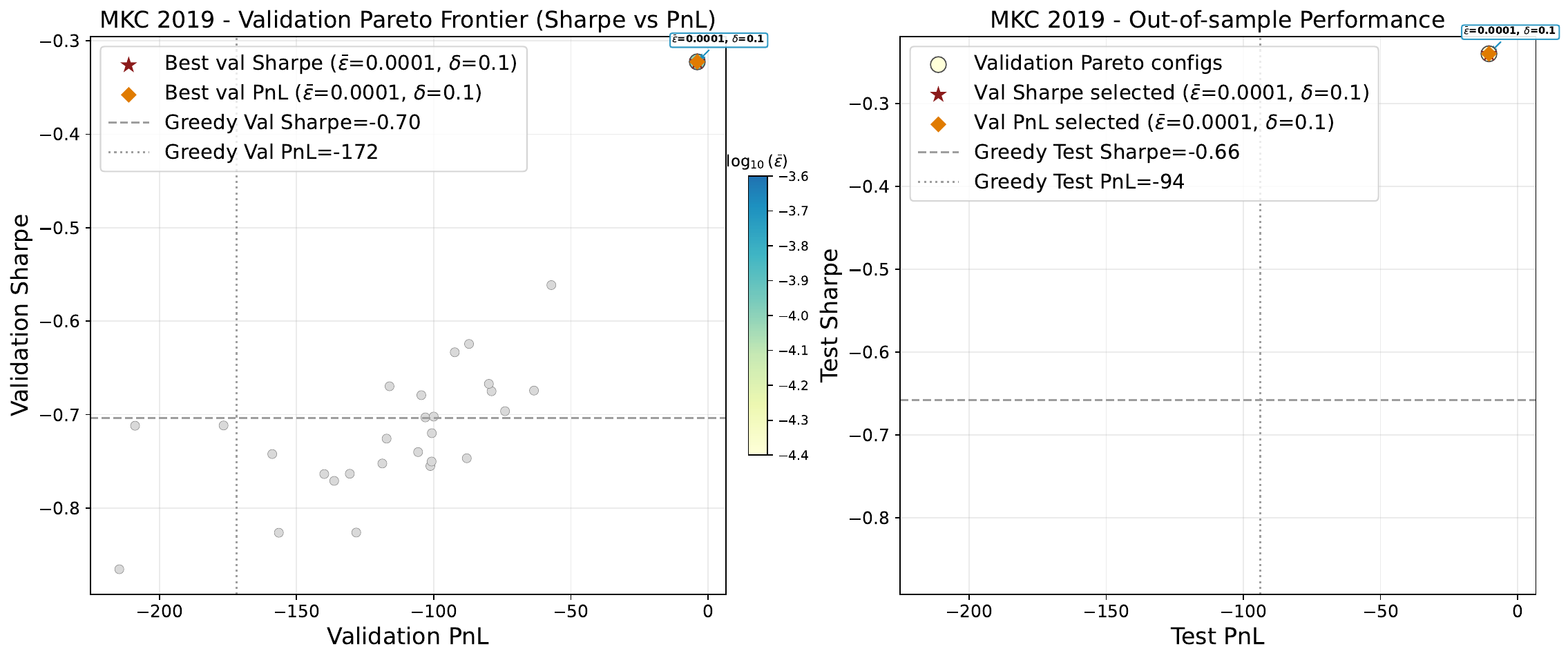}
            \caption{MKC, 2019.}
        \end{subfigure}
        \begin{subfigure}{0.49\linewidth}
            \includegraphics[width=\linewidth]{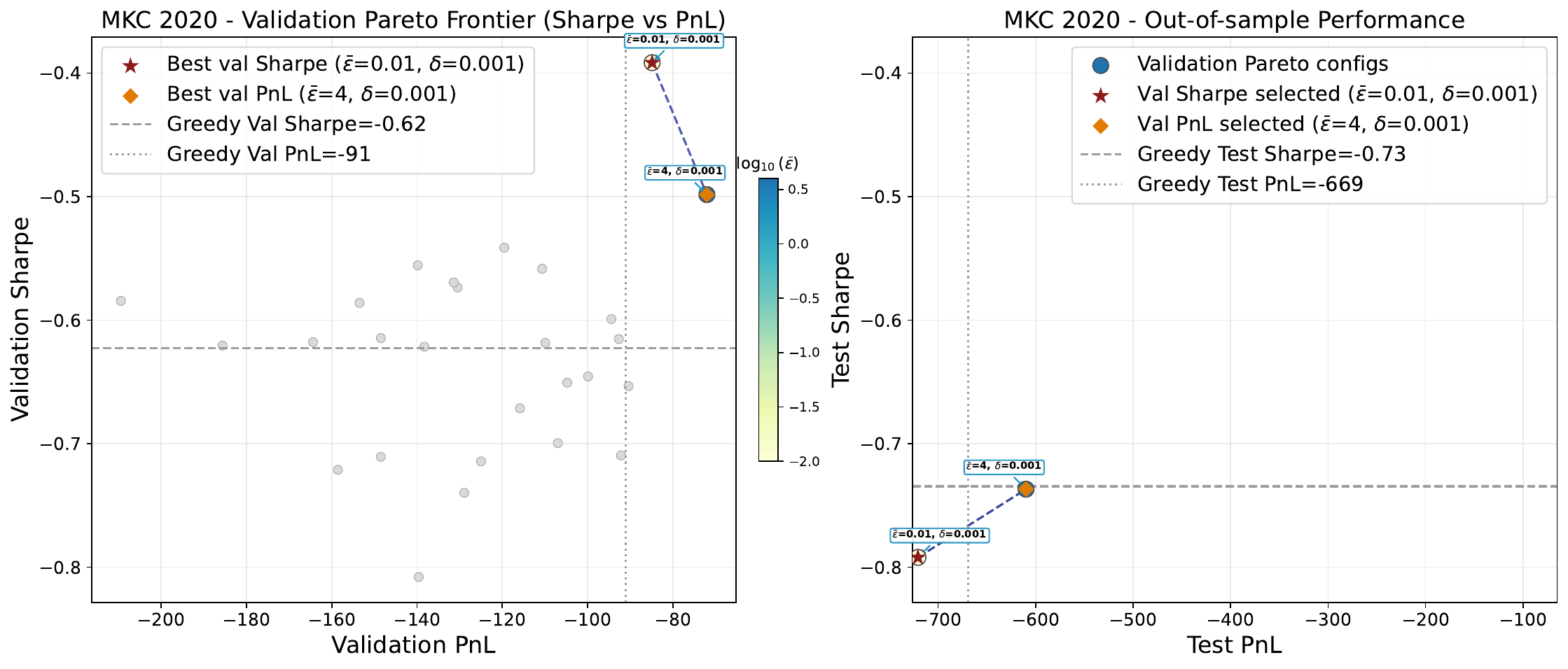}
            \caption{MKC, 2020.}
        \end{subfigure}
        \begin{subfigure}{0.49\linewidth}
            \includegraphics[width=\linewidth]{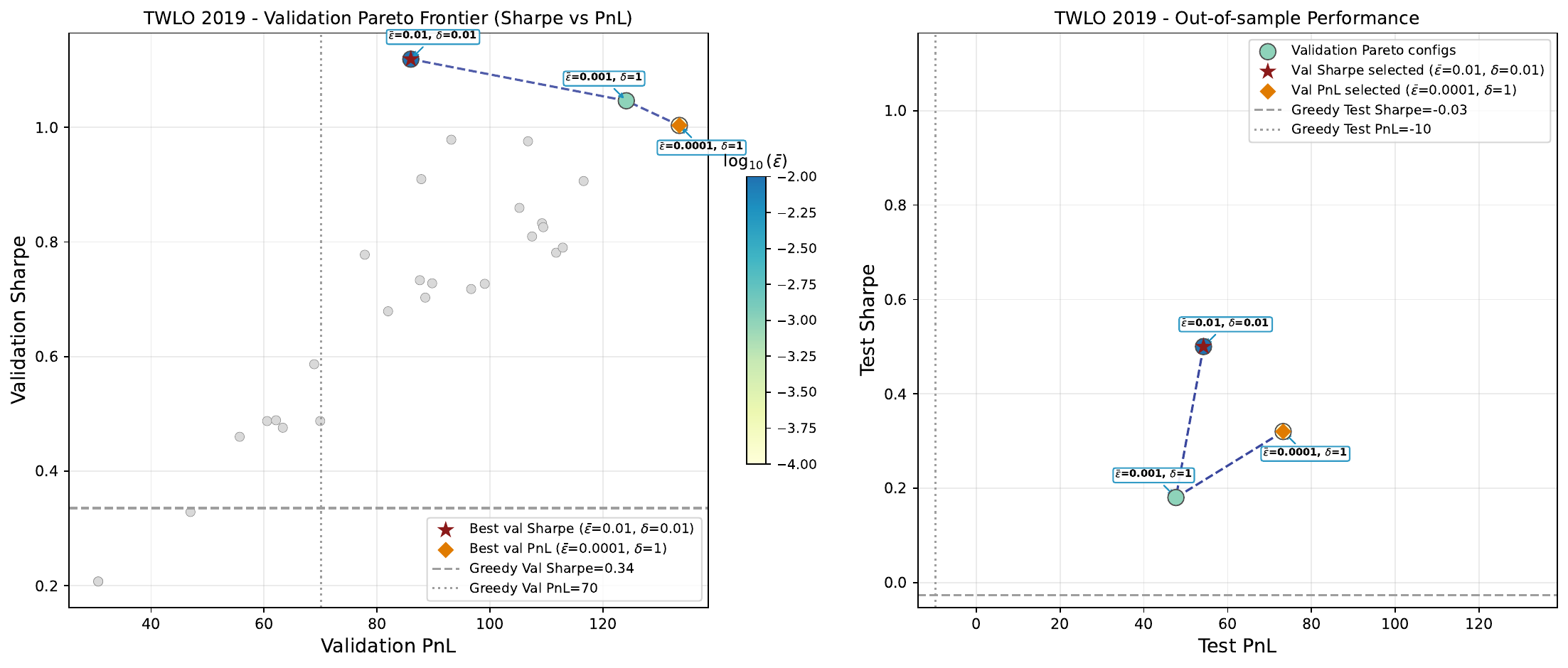}
            \caption{TWLO, 2019.}
        \end{subfigure}
        \begin{subfigure}{0.49\linewidth}
            \includegraphics[width=\linewidth]{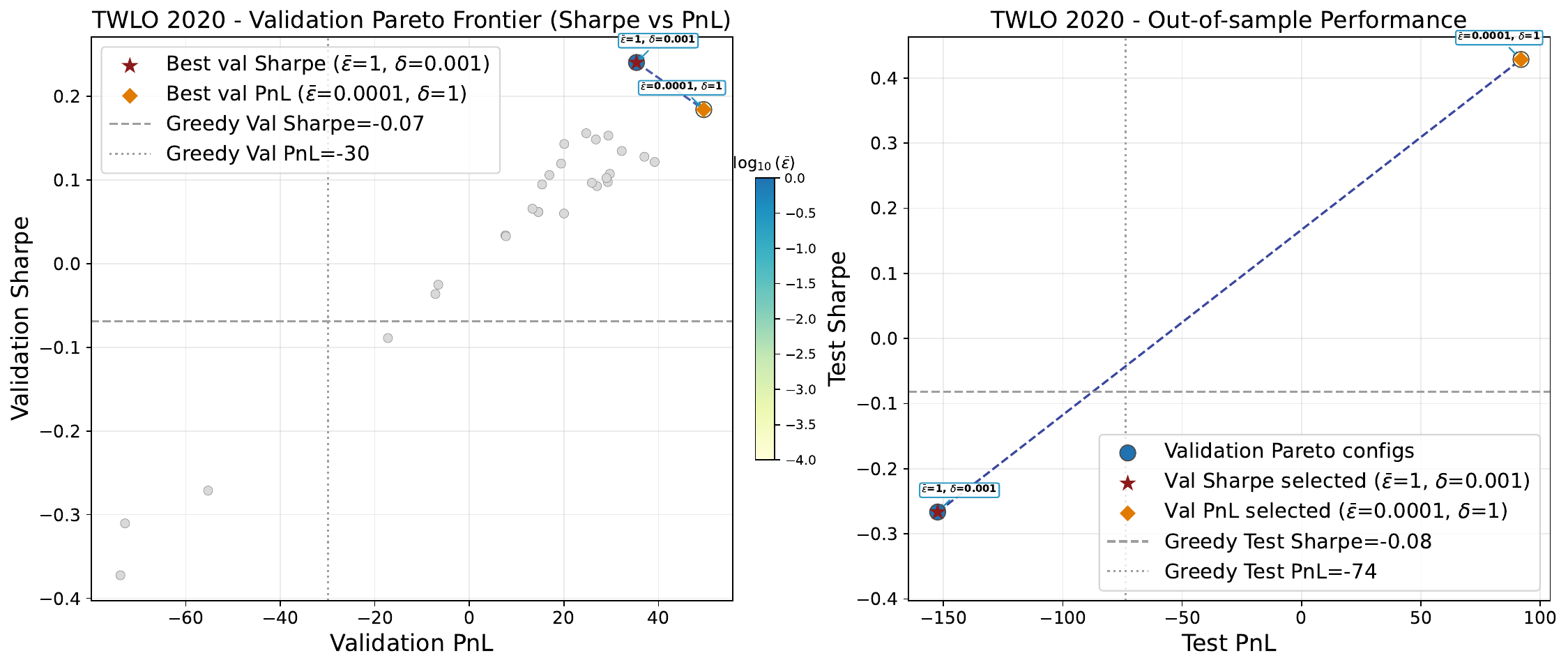}
            \caption{TWLO, 2020.}
        \end{subfigure}
    \end{minipage}}
    \caption{Validation Pareto frontier (left) and out-of-sample test performance (right) in Sharpe vs.\ mean-P\&L space. Colour encodes $\log_{10}(\delta)$; dashed line is the greedy benchmark. Robust configurations dominate the greedy benchmark out-of-sample across all stocks and periods, with the separation smaller for MKC and more pronounced in 2020.}
    \label{fig:pareto_real}
\end{figure}

\section[Proofs and Regularity Arguments]{Proofs and Regularity Arguments}\label{sec:proofs}
This appendix collects the auxiliary regularity results and the proofs of the market-making and robust dynamic programming statements from Sections~\ref{sec:setting} and~\ref{sec:method}.

\subsection[Regularity of the Sinkhorn Ambiguity Correspondence]{Regularity of the Sinkhorn Ambiguity Correspondence}\label{app:regularity-ambiguity}
Before reporting the main proofs, we establish the regularity results needed for the Sinkhorn ambiguity correspondence and the robust-MDP theorem.

We start by recalling the definition of weak convergence and its induced topology to which we refer as \emph{weak topology}. To this end for every $k \in \N$, $X \subseteq \R^k$, we define the set of continuous bounded functions $g: X\rightarrow \R$ via
\[
C_{b}(X, \R):=\left\{g \in C(X,\R) ~\middle|~ \sup_{x \in X }{|g(x)|}< \infty\right\},
\]
where $C(X,\R)$ denotes the set of continuous functions mapping from $X$ to  $\R$. Then, the weak topology on $\mathcal{M}_1(X)$ (denoted by $\tau_0$) is induced by the convergence
\begin{equation}\label{eq_convergence_topology_1}
\mu_n \xrightarrow{\tau_0} \mu \text{ for } n \rightarrow \infty ~\Leftrightarrow~ \lim_{n \rightarrow \infty} \int g \D \mu_n = \int g \D \mu \text{ for all } g \in C_{\operatorname{b}}(X, \R),
\end{equation}
compare also \cite[Theorem 6.9]{villani2009optimal}.

\begin{lem}[\cite{lu2025distributionally}~Lemma 6.4,~Lemma 6.6 and~Lemma 6.7]\label{prop_wasserstein}
Let $\Omloc \times A \ni (x,a) \mapsto \widehat{\PP}(x,a) \in  (\mathcal{M}_1(\Omega_{\operatorname{rnd}}),\tau_1)$ be continuous\footnote{We have $\PP_n \rightarrow \PP$ as $n \rightarrow \infty$ in the Wasserstein-1 topology denoted by $\tau_1$ if and only if $\PP_n \rightarrow \PP$ weakly as $n \rightarrow \infty$ and the first moments converge.} with finite first moments. Fix $\delta>0$ and assume that the primal radius is chosen so that the Sinkhorn ball \eqref{eq_definition_sinkhorn_ball_primal} is feasible, equivalently the shifted radius satisfies $\bar\varepsilon\ge0$ in the dual formulation. Then the set-valued map
\[
\Omloc \times A \ni (x,a) \twoheadrightarrow \mathcal{B}_{\varepsilon,\delta}(\widehat{\PP}(x,a))\subset\mathcal M_1(\Omega_{\operatorname{rnd}})
\]
is nonempty, compact-valued, and continuous w.r.t.\ the weak topology $\tau_0$.
\end{lem}

\begin{lem}\label{lem_assumptions_fulfilled_wasserstein}
Fix $\delta>0$ and a feasible Sinkhorn radius, i.e.\ a radius satisfying the shifted-radius condition $\bar\varepsilon\ge0$ in Section~\ref{subsec:sinkhorn-dual}. Then the lifted set-valued map
\begin{align*}
    \Omloc \times A &\twoheadrightarrow \mathcal{M}_1(\Omloc), \\
    (x,a) &\longmapsto \mathcal{P}^{\pi}_{\varepsilon,\delta}(x,a)
\end{align*}
as defined in \eqref{eq_definition_projected_local_ambiguity} is nonempty, compact-valued, and continuous w.r.t.\ the weak topology.
\end{lem}
\begin{proof}[Proof of Lemma~\ref{lem_assumptions_fulfilled_wasserstein}]
{
By Lemma~\ref{prop_wasserstein}, it is enough first to prove that the reference innovation kernel
\[
        \Omega_{\mathrm{loc}}\times A \ni (x,a)
        \longmapsto \widehat{\PP}(x,a)\in \mathcal M_1(\Omega_{\operatorname{rnd}})
\]
from Section~\ref{transition_prob} is continuous with respect to the Wasserstein--$1$ topology and has finite first moments. Let $y:=(x,a)\in\Omega_{\mathrm{loc}}\times A$. The reference innovation law is generated by drawing an ensemble index $\kappa\sim\mathrm{Unif}\{1,\ldots,K\}$ and then drawing the innovation vector
\[
        Z(y,\kappa)=\bigl(Z^{(c)}(y,\kappa)\bigr)_{c\in\mathcal C_{\ge0}\cup\mathcal C_{\mathbb R}}.
\]
Equivalently, on a common probability space, let $\xi=(\xi_c)_{c\in\mathcal C_{\ge0}\cup\mathcal C_{\mathbb R}}$ be independent standard normal coordinates and define
\[
        Z^{(c)}(y,k,\xi_c)
        :=
        \begin{cases}
        \exp\!\left(\mu_{\theta_k}^{(c)}(y)+\sigma_{\theta_k}^{(c)}(y)\xi_c\right), & c\in\mathcal C_{\ge0},\\[0.4em]
        \mu_{\theta_k}^{(c)}(y)+\sigma_{\theta_k}^{(c)}(y)\xi_c, & c\in\mathcal C_{\mathbb R}.
        \end{cases}
\]
Then $\widehat\PP(y)=\mathcal L(Z(y,\kappa,\xi))$ on $\Omega_{\operatorname{rnd}}$. If $y_n\to y$, the continuity of the neural-network maps $\mu_{\theta_k}^{(c)}$ and $\sigma_{\theta_k}^{(c)}$ implies with the continuous mapping theorem (\cite{billingsley2013convergence}) that $Z^{(c)}(y_n,k,\xi_c)\to Z^{(c)}(y,k,\xi_c)$ almost surely for every $k$ and every component $c$. Local boundedness of the network outputs on the convergent sequence, together with the Gaussian and lognormal moment bounds, gives an integrable dominating random variable. Hence dominated convergence yields
\[
        \mathbb E\bigl[\|Z(y_n,\kappa,\xi)-Z(y,\kappa,\xi)\|\bigr]\longrightarrow0.
\]
The coupled pair $(Z(y_n,\kappa,\xi),Z(y,\kappa,\xi))$ is a coupling of $\widehat\PP(y_n)$ and $\widehat\PP(y)$, so $$
W_1(\widehat\PP(y_n),\widehat\PP(y)) \leq \mathbb E\bigl[\|Z(y_n,\kappa,\xi)-Z(y,\kappa,\xi)\|\bigr] \to 0.$$
The same domination argument gives finite first moments $\int_{\Omega_{\rm rnd}} \|Z(y,x,z)\| \widehat{\PP}(dz) = \E[\|Z(y,x,\xi)\|<\infty$.

Therefore the unlifted Sinkhorn-ball correspondence $\Omloc \times A \ni (x,a)\twoheadrightarrow\mathcal B_{\varepsilon,\delta}(\widehat\PP(x,a))$ is nonempty, compact-valued, and weakly continuous on $\Omega_{\operatorname{rnd}}$ by Lemma~\ref{prop_wasserstein}. 
It remains only to lift this correspondence to full next-state laws. The maps $x\mapsto\pi(x)$ and $(x,a,z)\mapsto\Phi_\pi(x,a,z)$ are continuous under the product topology used for $\Omloc$, including the discrete topology on the time-to-maturity coordinate. For any bounded continuous test function $g\in C_b(\Omloc)$, the function $z\mapsto g(\Phi_\pi(x,a,z))$ is bounded and continuous. Hence push-forward by the reconstruction map preserves weak convergence. Non-emptiness and compactness are inherited from the unlifted Sinkhorn ball, and continuity of the lifted correspondence follows from the preceding push-forward argument. This proves the claim.
}
\end{proof}

\begin{cor}\label{cor_assumptions_fulfilled}
Under the feasible-radius condition in Lemma~\ref{lem_assumptions_fulfilled_wasserstein}, the projected ambiguity correspondence $\mathcal{P}^{\pi}_{\varepsilon,\delta}(x,a)$ fulfils the requirements of \cite[Assumption 2.2]{neufeld2022markov} with $p=0$, i.e.\ for the case where $\mathcal{M}_1(\Omloc)$ is equipped with the weak topology.
\end{cor}
\begin{proof}
The fulfilment of \cite[Assumption 2.2~(i)]{neufeld2022markov} follows from  Lemma~\ref{lem_assumptions_fulfilled_wasserstein}. Moreover, the requirement from \cite[Assumption 2.2~(ii)]{neufeld2022markov} is automatically fulfilled for $p=0$.
\end{proof}
\begin{lem}\label{lem_assumptions_reward_fulfilled}
{Let $\Pi_C$ denote the componentwise two-sided projection
\[
\Pi_C(u):=(-C_{\operatorname{global}})\vee (u\wedge C_{\operatorname{global}})
\]
on the real-valued state coordinates, and let it act as the identity on the time-to-maturity coordinate. Define the modified reward function
\begin{equation}\label{eq_reward_tilde}
\begin{aligned}
\widetilde{r}:\Omloc \times A \times \Omloc &\rightarrow \R,\\
(X_t,a_t,X_{t+1}) &\mapsto r(\Pi_C(X_t),a_t,\Pi_C(X_{t+1}))\mathbf 1_{\{\tau_t >0\}}.
\end{aligned}
\end{equation}
Then $\widetilde r$ fulfils the requirements of \cite[Assumption 2.4]{neufeld2022markov} with $p=0$.}
\end{lem}

\begin{proof}[Proof of Lemma~\ref{lem_assumptions_reward_fulfilled}]
{
The action set $A$ is compact and the image of $(X_t,X_{t+1})$ under $(\Pi_C,\Pi_C)$ is contained in a compact subset of the real-valued state coordinates; the time-to-maturity window $\tau^t\in\mathbb N_0^m$ lives in a discrete space with the product discrete topology, and the factor $\mathbf 1_{\{\tau_t>0\}}$ (which depends only on the last component $\tau_t$ of $\tau^t$) is continuous for that topology. Since $r$ is continuous in its arguments on this compact set, $\widetilde r$ is continuous and bounded.

As $r$ is locally Lipschitz on compact subsets, there is a constant $L<\infty$ such that for all $(X_t,a_t,X_{t+1})$ and $(X_t',a_t',X_{t+1}')$,
\begin{align*}
&\left|\widetilde{r}(X_t,a_t,X_{t+1})-\widetilde{r}(X_t',a_t',X_{t+1}') \right|\\
&\le L\left( \left\|\Pi_C(X_t)-\Pi_C(X_t')\right\|+\left\|a_t-a_t'\right\|+\left\|\Pi_C(X_{t+1})-\Pi_C(X_{t+1}')\right\|\right)\\
&\le L\left( \left\|X_t-X_t'\right\|+\left\|a_t-a_t'\right\|+\left\|X_{t+1}-X_{t+1}'\right\|\right),
\end{align*}
because the projection $\Pi_C$ is $1$-Lipschitz. This verifies the regularity condition required in \cite[Assumption 2.4]{neufeld2022markov}.}
\end{proof}

\subsection[Proofs of Market-Making and Dynamic Programming Results]{Proofs of Market-Making and Dynamic Programming Results}\label{app:proofs-main-results}

\begin{proof}[Proof of Lemma~\ref{lem_reward_approximation}]
For any $u\in \{0,\dots,T-1\}$ let
\[
M_u:=W_u+S_uI_u
\]
be the marked-to-market wealth. By the cash and inventory dynamics,
\[
\begin{aligned}
M_{u+1}-M_u
&=
(W_{u+1}-W_u)+S_u(I_{u+1}-I_u)
+I_{u+1}(S_{u+1}-S_u) \\
&=
Q_{u+1}^{\operatorname{ask}}\delta_u^{\operatorname{ask}}
+
Q_{u+1}^{\operatorname{bid}}\delta_u^{\operatorname{bid}}
-
c_{u+1}(Q_{u+1}^{\operatorname{bid}}+Q_{u+1}^{\operatorname{ask}})
+
I_{u+1}(S_{u+1}-S_u).
\end{aligned}
\]
Summing from \(u=t\) to \(T-1\) gives
\[
\begin{aligned}
W_T+S_TI_T =M_T&=  M_t +\sum_{u=t}^{T-1} (M_{u+1}-M_u)\\
&=
W_t+S_tI_t +
\sum_{u=t}^{T-1}
\left[
Q_{u+1}^{\operatorname{ask}}\delta_u^{\operatorname{ask}}
+
Q_{u+1}^{\operatorname{bid}}\delta_u^{\operatorname{bid}}
-
c_{u+1}(Q_{u+1}^{\operatorname{bid}}+Q_{u+1}^{\operatorname{ask}})
\right] +
\sum_{u=t}^{T-1}I_{u+1}(S_{u+1}-S_u).
\end{aligned}
\]
We now show that the last summation has zero conditional expectation. Since
\[
I_{u+1}
=
I_u+Q_{u+1}^{\operatorname{bid}}-Q_{u+1}^{\operatorname{ask}},
\]
we have
\[
\begin{aligned}
I_{u+1}(S_{u+1}-S_u)
=
I_u(S_{u+1}-S_u) +
\bigl(Q_{u+1}^{\operatorname{bid}}
-
Q_{u+1}^{\operatorname{ask}}\bigr)(S_{u+1}-S_u).
\end{aligned}
\]
By \eqref{eq:def_X} we have that \(I_u\) is \(\sigma(X_u,a_u)\)-measurable and, by assumption (1), the price increment has
conditional mean zero. Thus
\[
\mathbb E_{\mathbb P}
\left[
I_u(S_{u+1}-S_u)
\,\middle|\,
X_u,a_u
\right]=I_u \cdot \mathbb E_{\mathbb P}
\left[
S_{u+1}-S_u
\,\middle|\,
X_u,a_u
\right]
=
0.
\]
Moreover, by the conditional independence of the fills and the price innovation,
\[
\begin{aligned}
&\mathbb E_{\mathbb P}
\left[
\bigl(Q_{u+1}^{\operatorname{bid}}
-
Q_{u+1}^{\operatorname{ask}}\bigr)(S_{u+1}-S_u)
\,\middle|\,
X_u,a_u
\right] =
\mathbb E_{\mathbb P}
\left[
Q_{u+1}^{\operatorname{bid}}
-
Q_{u+1}^{\operatorname{ask}}
\,\middle|\,
X_u,a_u
\right]
\mathbb E_{\mathbb P}
\left[
S_{u+1}-S_u
\,\middle|\,
X_u,a_u
\right]
=
0.
\end{aligned}
\]
Therefore
\[
\mathbb E_{\mathbb P}
\left[
I_{u+1}(S_{u+1}-S_u)
\,\middle|\,
X_u,a_u
\right]
=
0,
\]
which implies, by the tower property,
\[
\mathbb E_{\mathbb P}
\left[
\sum_{u=t}^{T-1}I_{u+1}(S_{u+1}-S_u)
\,\middle|\,
X_t=x
\right]
=
0.
\]

Hence, for fixed \(\mathbf a\) and \(\mathbb P\),
\begin{equation}\label{eq:proof_lem_eq1}
\begin{aligned}
&\mathbb E_{\mathbb P}
\left[
W_T+S_TI_T
\,\middle|\,
X_t=x
\right]  =
W_t+S_tI_t
+
\mathbb E_{\mathbb P}
\left[
\sum_{u=t}^{T-1}
\left(
Q_{u+1}^{\operatorname{ask}}\delta_u^{\operatorname{ask}}
+
Q_{u+1}^{\operatorname{bid}}\delta_u^{\operatorname{bid}}
-
c_{u+1}(Q_{u+1}^{\operatorname{bid}}+Q_{u+1}^{\operatorname{ask}})
\right)
\,\middle|\,
X_t=x
\right].
\end{aligned}
\end{equation}
This implies with \eqref{eq:proof_lem_eq1} that
\[
\begin{aligned}
&\mathbb E_{\mathbb P}
\left[
W_T+S_TI_T
-
\sum_{u=t}^{T-1}
\frac{\gamma_u}{2}I_{u+1}^2\sigma^2\Delta_u\,\middle|\,
X_t=x
\right] \\
&\quad =
W_t+S_tI_t +
\mathbb E_{\mathbb P}
\left[
\sum_{u=t}^{T-1}
\left(
Q_{u+1}^{\operatorname{ask}}\delta_u^{\operatorname{ask}}
+
Q_{u+1}^{\operatorname{bid}}\delta_u^{\operatorname{bid}}
-
c_{u+1}(Q_{u+1}^{\operatorname{bid}}+Q_{u+1}^{\operatorname{ask}})
-
\frac{\gamma_u}{2}I_{u+1}^2\sigma^2\Delta_u
\right)
\,\middle|\,
X_t=x
\right]\\
&\quad=
W_t+S_tI_t
+
\mathbb E_{\mathbb P}
\left[
\sum_{u=t}^{T-1}
r(X_u,a_u,X_{u+1})
\,\middle|\,
X_t=x
\right].
\end{aligned}
\]
\end{proof}

\subsection[Infinite-Horizon Reformulation]{Infinite-Horizon Reformulation of the Finite-Horizon Problem}
Proposition~\ref{prop_2} is formulated in stationary discounted-MDP language. Since $r^C=0$ whenever the time-to-maturity equals $\tau=0$, the original finite-horizon problem is equivalent to an infinite-horizon discounted stationary MDP on the augmented time-including state space.

Let $\tau_{\operatorname{cur}}(x)$ denote the current, i.e. last,
time-to-maturity coordinate of the rolling-window state x.
To see this, we define first the restriction of the state space to finitely remaining time steps, and initial time, respectively:
\[
\Omega_{\le T}
:=
\bigl\{x\in\Omloc:\tau_{\operatorname{cur}}(x)\in\{0,\ldots,T\}\bigr\},
\qquad
\Omega_T^{\operatorname{init}}
:=
\bigl\{x\in\Omega_{\le T}:\tau_{\operatorname{cur}}(x)=T\bigr\},
\]
where $\tau_{\operatorname{cur}}(x) \in \N$ denotes the current time-to-maturity of the state $x\in \Omloc$.
For an admissible infinite-horizon Markov policy
\(\ab=(a_t)_{t\ge 0}\), define
\[
\mathfrak{P}_{x,\ab}^{\varepsilon,\delta,\infty}
:=
\left\{
\PP\in\mathcal{M}_1\bigl((\Omega_{\le T})^{\mathbb{N}_0}\bigr):
\begin{array}{l}
X_0=x \quad \PP\text{-a.s.},\\[0.15cm]
\PP(X_{t+1}\in\cdot\,|\,X_0,\ldots,X_t)
\in
\mathcal{P}^{\pi}_{\varepsilon,\delta}
\bigl(X_t,a_t(X_t)\bigr)
\quad \PP\text{-a.s. for all }t\ge 0
\end{array}
\right\}.
\]
The value function is then defined by
\begin{equation}\label{eq_robust_problem_1}
V(x)
:=
\sup_{\ab}
\inf_{\PP\in\mathfrak{P}_{x,\ab}^{\varepsilon,\delta,\infty}}
\E_{\PP}
\left[
\sum_{t=0}^{\infty}
\alpha^t r^C(X_t,a_t(X_t),X_{t+1})
\right],
\qquad x\in\Omega_{\le T}.
\end{equation}

\begin{proof}[Proof of Proposition~\ref{prop_2}]
By Corollary~\ref{cor_assumptions_fulfilled}, the projected local ambiguity
correspondence
\[
(x,a)\twoheadrightarrow \mathcal P^{\pi}_{\varepsilon,\delta}(x,a)
\]
satisfies the required non-emptiness, compactness, measurability, and
continuity assumptions with $p=0$. Since the stopped state space
$\Omega_{\leq T}$ is invariant under all transition kernels in the local
ambiguity sets, we may regard the restricted correspondence as a
correspondence from $\Omega_{\leq T}\times A$ into
$\mathcal M_1(\Omega_{\leq T})$. By
Lemma~\ref{lem_assumptions_reward_fulfilled}, the stopped and clipped reward
$r^C$ is bounded and continuous. Therefore the discounted robust
Markov-decision-process theorem of
\cite[Theorem~2.7]{neufeld2022markov} applies on $\Omega_{\leq T}$ to the
Bellman operator
\[
\T v(x)
=
\sup_{a\in A}
\inf_{\PP\in\mathcal P^{\pi}_{\varepsilon,\delta}(x,a)}
\E_\PP\bigl[
r^C(x,a,X_1)+\alpha v(X_1)
\bigr],
\qquad x\in\Omega_{\leq T}.
\]

The theorem yields a unique bounded continuous fixed point
$V\in C_b(\Omega_{\leq T})$, Borel-measurable optimal selectors
$a_{\operatorname{loc}}^*$ and $\PP_{\operatorname{loc}}^*$, and the
fixed-point identity
\[
V(x)=\T V(x)
=
\E_{\PP_{\operatorname{loc}}^*(x)}
\bigl[
r^C(x,a_{\operatorname{loc}}^*(x),X_1)+\alpha V(X_1)
\bigr],
\qquad x\in\Omega_{\leq T}.
\]
This proves Proposition~\ref{prop_2}~(i).

Define the stationary Markov policy
\[
\ab^*
:=
\bigl(a_{\operatorname{loc}}^*(X_t)\bigr)_{t\geq 0}.
\]
Moreover, for $x\in\Omega_{\leq T}$, define the infinite-horizon product
measure
\[
\PP_x^{*,\infty}(d\omega_0,d\omega_1,\ldots)
:=
\delta_x(d\omega_0)
\prod_{t=0}^{\infty}
\PP_{\operatorname{loc}}^*(\omega_t;d\omega_{t+1}).
\]
By the Ionescu--Tulcea theorem this defines a probability measure on
$(\Omega_{\leq T})^{\mathbb N_0}$. Since
\[
\PP_{\operatorname{loc}}^*(x)
\in
\mathcal P^{\pi}_{\varepsilon,\delta}
\bigl(x,a_{\operatorname{loc}}^*(x)\bigr),
\qquad x\in\Omega_{\leq T},
\]
we have
\[
\PP_x^{*,\infty}
\in
\mathfrak P_{x,\ab^*}^{\varepsilon,\delta,\infty}.
\]

The verification part of
\cite[Theorem~2.7~(iii)]{neufeld2022markov} gives, for every
$x\in\Omega_{\leq T}$,
\[
\begin{aligned}
V(x)
&=
\E_{\PP_x^{*,\infty}}
\left[
\sum_{t=0}^{\infty}
\alpha^t
r^C\bigl(X_t,a_{\operatorname{loc}}^*(X_t),X_{t+1}\bigr)
\right] =
\inf_{\PP\in\mathfrak P_{x,\ab^*}^{\varepsilon,\delta,\infty}}
\E_{\PP}
\left[
\sum_{t=0}^{\infty}
\alpha^t
r^C\bigl(X_t,a_{\operatorname{loc}}^*(X_t),X_{t+1}\bigr)
\right].
\end{aligned}
\]

It remains only to identify this infinite-horizon identity with the stopped
finite-horizon identity. To this end, let
$
k:=\tau_{\operatorname{cur}}(x).
$
By the deterministic time-to-maturity update,
\[
\tau_{\operatorname{cur}}(X_{t+1})
=
\bigl(\tau_{\operatorname{cur}}(X_t)-1\bigr)^+,
\]
we have $\tau_{\operatorname{cur}}(X_t)=0$ for all $t\geq k$. Since
$r^C(x,a,y)=0$ whenever $\tau_{\operatorname{cur}}(x)=0$, it follows that,
under every admissible measure,
\[
\sum_{t=0}^{\infty}
\alpha^t
r^C\bigl(X_t,a_{\operatorname{loc}}^*(X_t),X_{t+1}\bigr)
=
\sum_{t=0}^{k-1}
\alpha^t
r^C\bigl(X_t,a_{\operatorname{loc}}^*(X_t),X_{t+1}\bigr),
\]
with the convention that the sum is empty if $k=0$. Hence, for every
$x\in\Omega_{\leq T}$,
\[
\begin{aligned}
V(x)
=
\E_{\PP_x^{*,\infty}}
\left[
\sum_{t=0}^{k-1}
\alpha^t
r^C\bigl(X_t,a_{\operatorname{loc}}^*(X_t),X_{t+1}\bigr)
\right] &=
\inf_{\PP\in\mathfrak P_{x,\ab^*}^{\varepsilon,\delta,\infty}}
\E_{\PP}
\left[
\sum_{t=0}^{k-1}
\alpha^t
r^C\bigl(X_t,a_{\operatorname{loc}}^*(X_t),X_{t+1}\bigr)
\right]\\
&=\inf_{\PP\in\mathfrak P_{x,\ab^*}^{\varepsilon,\delta}}
\E_{\PP}
\left[
\sum_{t=0}^{k-1}
\alpha^t
r^C\bigl(X_t,a_{\operatorname{loc}}^*(X_t),X_{t+1}\bigr)
\right].
\end{aligned}
\]
In particular, if $x\in\Omega_T^{\operatorname{init}}$, then
$k=\tau_{\operatorname{cur}}(x)=T$, and therefore
\[
\begin{aligned}
V(x)
=
\E_{\PP_x^{*,\infty}}
\left[
\sum_{t=0}^{T-1}
\alpha^t
r^C\bigl(X_t,a_{\operatorname{loc}}^*(X_t),X_{t+1}\bigr)
\right] &=
\inf_{\PP\in\mathfrak P_{x,\ab^*}^{\varepsilon,\delta,\infty}}
\E_{\PP}
\left[
\sum_{t=0}^{T-1}
\alpha^t
r^C\bigl(X_t,a_{\operatorname{loc}}^*(X_t),X_{t+1}\bigr)
\right]\\
&=
\inf_{\PP\in\mathfrak P_{x,\ab^*}^{\varepsilon,\delta}}
\E_{\PP}
\left[
\sum_{t=0}^{T-1}
\alpha^t
r^C\bigl(X_t,a_{\operatorname{loc}}^*(X_t),X_{t+1}\bigr)
\right].
\end{aligned}
\]
This proves Proposition~\ref{prop_2}~(ii).
\end{proof}

\begin{proof}[Proof of Proposition~\ref{prop_VI}]
The Bellman operator is an $\alpha$-contraction on $C_b(\Omloc)$ equipped with the sup norm: for $v,w\in C_b(\Omloc)$,
\[
\|\T v-\T w\|_\infty\le \alpha\|v-w\|_\infty.
\]
Compare \cite[Theorem 2.7~(ii)]{neufeld2022markov}. Banach's fixed-point theorem therefore yields
\[
\lim_{n\to\infty}\|\T^n v_0-V\|_\infty=0,
\qquad v_0\in C_b(\Omloc),
\]
with geometric rate at most $\alpha$. This proves Proposition~\ref{prop_VI}.
\end{proof}


\end{document}